\documentclass[11pt]{article}

\usepackage{fullpage}
\usepackage{amsmath,amsfonts,amsthm,mathrsfs,xspace,hyperref,graphicx,bbm}
\usepackage{endnotes}
\usepackage{bm}
\usepackage[shortlabels]{enumitem}
\usepackage{amssymb,latexsym}
\usepackage{relsize}
\usepackage{cleveref}
\usepackage{stmaryrd}
\usepackage{xcolor}
\usepackage{comment}
\usepackage{tikz}
\usepackage{mdframed}
\usepackage{tabularx}
\usepackage{booktabs}
\usepackage{multirow}

\usepackage{hieroglf}

\usepackage{CJKutf8}

\usepackage[linesnumbered,boxed]{algorithm2e}

\crefname{algocf}{alg.}{algs.}
\Crefname{algocf}{Algorithm}{Algorithms}

\crefname{algocf}{pse.}{pseu.}
\Crefname{algocf}{Pseudocode}{Pseudocode}

\mdfdefinestyle{figstyle}{linecolor=black!7, backgroundcolor=black!7, innertopmargin=10pt, innerleftmargin=25pt, innerrightmargin=25pt, innerbottommargin=10pt }

\usepackage[utf8]{inputenc}
\usepackage[backend=bibtex,  
style=alphabetic,      
maxbibnames=99,      
minalphanames=3,      
backref=true
style=plain
sorting=nyt]{biblatex}

\definecolor{weborange}{rgb}{.8,.3,.3}
\definecolor{webblue}{rgb}{0,0,.8}
\definecolor{internallinkcolor}{rgb}{0,.5,0}
\definecolor{externallinkcolor}{rgb}{0,0,.5}

\definecolor{DarkBlue}{rgb}{0,0,0.8}  
\definecolor{DarkOrange}{rgb}{0.8,0.4,0}  
\def\mylinkcolor{DarkBlue}

\hypersetup{
	colorlinks=true, urlcolor=\mylinkcolor,
  linkcolor=\mylinkcolor, citecolor=\mylinkcolor}

\newtheorem{theorem}{Theorem}[section]
\newtheorem{proposition}[theorem]{Proposition}
\newtheorem{lemma}[theorem]{Lemma}
\newtheorem{claim}[theorem]{Claim}
\newtheorem{fact}[theorem]{Fact}
\newtheorem{corollary}[theorem]{Corollary}
\newtheorem{definition}[theorem]{Definition}
\theoremstyle{definition}
\newtheorem{problem}[theorem]{Problem}

\newtheorem{example}[theorem]{Example}

\DeclareMathOperator*{\Ex}{\mathbb{E}}

\newcommand{\ket}[1]{{|#1\rangle}}
\newcommand{\bra}[1]{{\langle#1|}}
\newcommand{\bigket}[1]{{\left |#1 \right \rangle}}
\newcommand{\bigbra}[1]{{\left \langle#1 \right|}}

\newcommand{\ketbra}[2]{|#1\rangle\! \langle #2|}
\newcommand{\Tr}{\mbox{\rm Tr}}
\newcommand{\Id}{\mathcal{I}}

\newcommand{\ceil}[1]{\left\lceil #1 \right\rceil}
\newcommand{\C}{\ensuremath{\mathbb{C}}}

\newcommand{\X}{\mathcal{X}}
\newcommand{\Y}{\mathcal{Y}}
\newcommand{\A}{\mathcal{A}}
\newcommand{\B}{\mathcal{B}}

\newcommand{\D}{\mathrm{D}}

\newcommand{\puretomixed}[1]{\llbracket #1 \rrbracket}

\newcommand{\eps}{\varepsilon}
\newcommand{\poly}{\mathrm{poly}}
\newcommand{\polylog}{\mathrm{polylog}}

\newcommand{\wt}[1]{\widetilde{#1}}
\newcommand{\what}[1]{\hat{#1}}

\newcommand{\strategy}{\mathscr{S}}

\renewcommand{\P}{\mathsf{P}}

\newcommand{\supp}{\mathrm{supp}}

\newcommand{\ba}{\mathbf{a}}
\newcommand{\bb}{\mathbf{b}}

\newcommand{\br}{\mathbf{r}}
\newcommand{\bX}{\mathbf{X}}
\newcommand{\bY}{\mathbf{Y}}
\newcommand{\bA}{\mathbf{A}}
\newcommand{\bB}{\mathbf{B}}
\newcommand{\bOmega}{{\boldsymbol{\Omega}}}
\newcommand{\bomega}{{\boldsymbol{\omega}}}
\newcommand{\bD}{\mathbf{D}}
\newcommand{\bM}{\mathbf{M}}

\newcommand{\bZ}{\mathbb{Z}}

\renewcommand{\cal}[1]{\mathcal{#1}}
\newcommand{\Qsf}{\mathsf{Q}}
\newcommand{\Ssf}{\mathsf{S}}
\newcommand{\Rsf}{\mathsf{R}}

\newcommand{\class}[1]{\ensuremath{\mathsf{#1}}\xspace}

\newcommand{\MIP}{\class{MIP}}
\newcommand{\CLMIP}{\class{CLMIP}}
\newcommand{\coclass}{\mathrm{co}}
\newcommand{\MIPco}{\class{MIP}^{\coclass}}

\newcommand{\RE}{\class{RE}}
\newcommand{\coRE}{\class{coRE}}
\newcommand{\NEXP}{\class{NEXP}}
\newcommand{\IP}{\class{IP}}
\newcommand{\PSPACE}{\class{PSPACE}}
\newcommand{\NP}{\class{NP}}

\newcommand{\Language}{\class{L}}
\newcommand{\Decisionproblem}{\class{D}}

\SetAlgorithmName{Pseudocode}{}{}

\newcommand{\TIME}{\mathsf{TIME}}

\newenvironment{gamespec}{
	\begin{mdframed}[style=figstyle]}{
\end{mdframed}}

\newcommand{\jqnote}[1]{}

\newcommand{\bofh}{\mathcal{B}(\mathcal{H})}

\newcommand{\alicealg}{\mathscr{A}}
\newcommand{\bobalg}{\mathscr{B}}
\newcommand{\cyclicsepvector}{\tau}
\newcommand{\positivecone}{\mathcal{H}_{\ket{\cyclicsepvector}}^{+}}
\newcommand{\tracialstate}{\ket{\tau}}

\newcommand{\MEstate}[1]{\ket{\text{ME}_{#1}}}

\newcommand{\gamesequence}{\mathscr{G}}
\newcommand{\verifiersequence}{\mathscr{V}}

\newcommand{\parameterTM}{\mathtt{P}}
\newcommand{\samplerTM}{\mathtt{Q}}
\newcommand{\deciderTM}{\mathtt{D}}

\newcommand{\decidertypedtype}{\mathtt{T}}
\newcommand{\decidertypedquestionpair}{\mathtt{E}}
\newcommand{\decidertypedfunction}{\mathtt{L}}
\newcommand{\proverTM}{\mathtt{P}}

\newcommand{\Canoline}[1]{\text{Can}(#1)}
\newcommand{\Nullline}[2]{\text{Null}_{#1}^{\text{LN}}(#2)}

\newcommand{\seedlength}{p}

\newcommand{\typelength}{l^{\text{t}}}
\newcommand{\outputlength}{l^{\text{out}}}
\newcommand{\detypesynchead}{\text{D}}

\newcommand{\Compressgame}{\text{Comp}}
\newcommand{\QuestionRed}{\text{QR}}
\newcommand{\Answerred}{\text{AR}}
\newcommand{\Pararep}{\text{Pararep}}

\newcommand{\binary}[1]{\textbf{bin}(#1)}
\newcommand{\binaryinv}[1]{\textbf{bininv}(#1)}

\usepackage{braket}
\usepackage{comment}

\newcommand{\tensor}{\otimes}

\newcommand\bC{\mathbb C}

\newcommand\bF{\mathbb F}

\newcommand\bI{\mathbb I}

\newcommand\bN{\mathbb N}

\newcommand\bR{\mathbb R}
\newcommand\bQ{\mathbb Q}

\newcommand\cA{\mathcal A}
\newcommand\cB{\mathcal B}
\newcommand\cC{\mathcal C}

\newcommand\cF{\mathcal F}
\newcommand\cG{\mathcal G}
\newcommand\cH{\mathcal H}
\newcommand\cI{\mathcal I}

\newcommand\cL{\mathcal L}
\newcommand\cM{\mathcal M}

\newcommand\cS{\mathcal S}

\newcommand\cU{\mathcal U}

\newcommand\cX{\mathcal X}
\newcommand\cY{\mathcal Y}
\newcommand\cZ{\mathcal Z}

\newcommand\rI{\mathrm I}

\newcommand{\prep}{\bot}

\newcommand{\BofH}{\mathcal{B}(\mathcal{H})}

\renewcommand{\cal}[1]{\mathcal{#1}}

\newcommand{\deltasync}{\delta_{\text{sync}}}

\newcommand{\eq}[1]{\hyperref[eq:#1]{(\ref*{eq:#1})}}
\renewcommand{\sec}[1]{\hyperref[sec:#1]{Section~\ref*{sec:#1}}}
\newcommand{\thm}[1]{\hyperref[thm:#1]{Theorem~\ref*{thm:#1}}}
\newcommand{\lem}[1]{\hyperref[lem:#1]{Lemma~\ref*{lem:#1}}}
\newcommand{\defn}[1]{\hyperref[def:#1]{Definition~\ref*{def:#1}}}
\newcommand{\prop}[1]{\hyperref[prop:#1]{Proposition~\ref*{prop:#1}}}
\newcommand{\cor}[1]{\hyperref[cor:#1]{Corollary~\ref*{cor:#1}}}
\newcommand{\fig}[1]{\hyperref[fig:#1]{Figure~\ref*{fig:#1}}}
\newcommand{\tab}[1]{\hyperref[tab:#1]{Table~\ref*{tab:#1}}}
\newcommand{\app}[1]{\hyperref[app:#1]{Appendix~\ref*{app:#1}}}
\newcommand{\chap}[1]{\hyperref[chap:#1]{Chapter~\ref*{chap:#1}}}
\newcommand{\factlink}[1]{\hyperref[fact:#1]{Fact~\ref*{fact:#1}}}
\newcommand{\proto}[1]{\hyperref[proto:#1]{Protocol~\ref*{fact:#1}}}
\newcommand{\examp}[1]{\hyperref[examp:#1]{Example~\ref*{fact:#1}}}

\newcommand{\Searchfromabove}{\texttt{\Searchfromabove}}

\newcommand{\indpoly}[1]{ind_{#1}}
\newcommand{\ldencoding}[1]{lde_{#1}}
\newcommand{\lddecoding}{dec}

\usepackage[postpunc={dot},
nostyles,
stylemods={tree},
automake
]{glossaries-extra}

\makeglossaries

\newglossaryentry{setsize}{
	name={$|S|$},
	description={: \quad Cardinality of a set},
	sort = {a1}
}

\newglossaryentry{matrixsize}{
	name={$\cM_n(T)$},
	description={: \quad $n$ by $n$ matrix of elements of $T$},
	sort = {a3}
}

\newglossaryentry{expectation}{
	name={$\Ex_{x \sim \mu}$},
	description={: \quad Expectation over the distribution $\mu$},
	sort = {a3}
}

\newglossaryentry{Conditionalexp}{
	name={$\P_{Y | X = x}(y)$},
	description={: \quad Probability of $Y$ conditioning on $X$},
	sort = {a4a}
}

\newglossaryentry{Vardis}{
	name={$\| \P_{X_0} - \P_{X_1} \|$},
	description={: \quad The variation distance between probability distribution $\P_{X_0}$ and $\P_{X_1}$, given in~\eqref{eq:vardis} },
	sort = {a4b}
}

\newglossaryentry{kronecker}{
	name={$\delta_{a,b}$},
	description={: \quad The delta kronecker product},
	sort = {a5}
}

\newglossaryentry{substring}{
	name={$s|_a$},
	description={: \quad Coordinate of $s$ index by $a$},
	sort = {a6}
}

\newglossaryentry{hamming}{
	name={$|s|$},
	description={: \quad Hamming weight for the string $s$},
	sort = {a6}
}

\newglossaryentry{dotproduct}{
	name={$s \cdot t$},
	description={: \quad Dot product between $s$ and $t$ for string $s$ and $t$},
	sort = {a6}
}

\newglossaryentry{zerooutmap}{
	name={$\pi_{> j}(s)$},
	description={: \quad The map which zeros out the first $j$ entries of the string $s$},
	sort = {a6}
}

\newglossaryentry{numberset}{
	name={$[n]$},
	description={: \quad The set $\{0, 1 \cdots, n-1 \}$},
	sort = {a7}
}

\newglossaryentry{numberset2}{
	name={$[n,m]$},
	description={: \quad The set $\{n, n+1 \cdots, m -1\}$},
	sort = {a7}
}

\newglossaryentry{binaryfun}{
	name={$\binary{n}$},
	description={: \quad Binary representation for the integer $n$},
	sort = {a8}
}
\newglossaryentry{binaryinv}{
	name={$\binaryinv{s}$},
	description={: \quad The inverse binary function, i.e. $\binaryinv{s} = m$ where $m$ is the unique integer such that $\binary{m} = s$ for $s \in \{0,1\}^n$},
	sort = {a8}
}

\newglossaryentry{descriptionTM}{
	name={$\langle \mathtt{A} \rangle$},
	description={: \quad The minimial description lenght of a Turing machine $|\mathtt{A}|$},
	sort = {b1a}
}

\newglossaryentry{descriptionTMinput}{
	name={$\langle \mathtt{A}(x) \rangle$},
	description={: \quad The description of the Turing machine $\mathtt{A}$ which is hardcoded to run $x \in \{0,1\}^*$ as input (in this case $\langle \mathtt{A}(x) \rangle$ takes the empty tape as input, and will return $\mathtt{A}(x)$) after the computation step.},
	sort = {b1}
}

\newglossaryentry{sizeTM}{
	name={$|\mathtt{A}|$},
	description={: \quad The minimial description lenght of a Turing machine $|\mathtt{A}|$},
	sort = {b2}
}

\newglossaryentry{RuntimeTM}{
	name={$\TIME_{\mathtt{A}}(n)$},
	description={: \quad The maximum of the runtime and decription size for the Turing machine $\mathtt{A}$},
	sort = {b3}
}

\newglossaryentry{complementdecision}{
	name={\class{coD}},
	description={: \quad The complement of the decision problem $\Decisionproblem$. },
	sort = {b4b}
}

\newglossaryentry{mappingreduction}{
	name={$\Decisionproblem_1 \leq \Decisionproblem_2$},
	description={: \quad $\Decisionproblem_1$ is reducible to $\Decisionproblem_2$. },
	sort = {b4c}
}

\newglossaryentry{mappingreductionpoly}{
	name={$\Decisionproblem_1 \leq_p \Decisionproblem_2$},
	description={: \quad $\Decisionproblem_1$ is polynomial-time reducible to $\Decisionproblem_2$. },
	sort = {b4c}
}

\newglossaryentry{RE}{
	name={\RE},
	description={: \quad The set of recursively enumerable languages, complete with respect to the halting problem. },
	sort = {b4c}
}

\newglossaryentry{coRE}{
	name={\coRE},
	description={: \quad The complement of $\RE$, complete with respect to the non-halting problem. },
	sort = {b4b}
}

\newglossaryentry{MIPstar}{
	name={$\MIP^*$},
	description={: \quad Multiprover interactive proof system with tensor product model of entanglement (two round, one prover with completness 1 and soundness $\frac{1}{2}$)},
	sort = {b5a}
}

\newglossaryentry{MIPco}{
	name={$\MIP^{co}$},
	description={: \quad Multiprover interactive proof system with commuting opereator model of entanglement (two round, one prover with completness 1 and soundness $\frac{1}{2}$)},
	sort = {b5b}
}

\newglossaryentry{Searchfrombelow}{
	name={$\texttt{searchfrombelow}_{\eps}$},
	description={: \quad The search from below algorithm for $\eps \in [0,1]$, Teminates whenever $\omega^*(\cG) > \eps$ (Runs forever otherwise)},
	sort = {b6a}
}

\newglossaryentry{Searchfromabove}{
	name={$\texttt{searchfromabove}_{\delta}$},
	description={: \quad The search from above algorithm for $\delta \in [0,1]$, Terminates whenever $\omega^{co}(\cG) < \delta$ (Runs forever otherwise)},
	sort = {b6b}
}

\newglossaryentry{TMD}{
	name={\small\textpmhg{F}\normalsize},
	description={: \quad The input for the Turing machine for the proof of $\RE$/$\coRE$ completeness. This is the only non-western character used in this paper},
	sort = {b7}
}
\newglossaryentry{Finitefield}{
	name={$\bF_{2^p}$},
	description={: \quad Finite fields over $2^p$, p is always assume to be odd in this paper},
	sort = {c1}
}

\newglossaryentry{Canobasis}{
	name={$\hat{e}_i$},
	description={: \quad Canoical basis for for the field $\bF_{2^p}$},
	sort = {c2}
}

\newglossaryentry{Finitetrace}{
	name={$\Tr(a)$},
	description={: \quad Finite field trace for the element $a \in \bF_{2^p}$},
	sort = {c3}
}

\newglossaryentry{Canobasismap}{
	name={$\kappa(a)$},
	description={: \quad The bijection map between $\bF_{2^p}$ to $\{0,1\}^p$ whenever $a \in \bF_{2^p}$. The bijection map between $\bF_{2^{p}}^m$ to $\{0,1\}^{pm}$ whenever $a \in \bF_{2^p}^m$  },
	sort = {c4}
}

\newglossaryentry{dimspace}{
	name={$\text{dim}(V)$},
	description={: \quad Dimension of the subspace $V \subseteq \bF_{2^p}^m$  },
	sort = {c5}
}

\newglossaryentry{Orthogonalsubspace}{
	name={$W^{\bot}$},
	description={: \quad The orthogonal subspace of $W$ for $W \subseteq V \subseteq \bF_{2^p}^m$. $W^{\bot}$ is the orthogonal subspace over $\bF_{2^p}^m$ if unspecified},
	sort = {c6}
}

\newglossaryentry{Canocomplement}{
	name={$W^{C}$},
	description={: \quad The canonical complement of $W$ for $W \subseteq W \subseteq \bF_{2^p}^m$ for a canoical basis subspace $V$. $W^{C}$ is the canoical complement over $\bF_{2^p}^m$ if unspecified},
	sort = {c7}
}

\newglossaryentry{Disjointpartsubspace}{
	name={$V_{< i}$},
	description={: \quad The union of the first $i$ subspace in a disjoint partition of $V$},
	sort = {c7}
}

\newglossaryentry{Zeromap}{
	name={$\pi_{> j}^m$},
	description={: \quad The map which zeros out the first $j$ entries for elements of $\bF_{2^p}^m$},
	sort = {c8}
}

\newglossaryentry{KerLinfun}{
	name={$\ker(\decidertypedfunction)$},
	description={: \quad Kernel subspace for a linear function $\decidertypedfunction$},
	sort = {d1}
}

\newglossaryentry{ProjCompsubspace}{
	name={$\decidertypedfunction^{\bot}$},
	description={: \quad The Linear map which projects onto $\left(\ker{(\decidertypedfunction)}^{\bot}\right)^{C}$ where $\decidertypedfunction: V \rightarrow V$ is a linear function over a canoical basis subspace $V$},
	sort = {d1a}
}

\newglossaryentry{Canolinerep}{
	name={$\Canoline{l}$},
	description={: \quad Canonical representation of an affine line, define as $\Canoline{l} := (v, \Nullline{v}{u}) \in \bF_{2^p}^{2m}$},
	sort = {d2}
}

\newglossaryentry{RMcode}{
	name={$\text{RM}_b$},
	description={: \quad Reed-Muller encoding for the string $b$ },
	sort = {d3}
}

\newglossaryentry{LDpolyset}{
	name={$\text{IdPoly}(p,m,d)$},
	description={: \quad The set of polynomails $\textbf{g}: \bF_{2^p}^m \rightarrow\bF_{2^p}$ with individual degree of at most $d$},
	sort = {d3}
}
\newglossaryentry{SoundnessLDT}{
	name={$\mathbf{\eta}_{\text{LD}}(p,m,d,\eps)$},
	description={: \quad The soundness parameter function for the quantum low-individual degree test, given in~\Cref{thm:soundnessQLDT}},
	sort = {d4a}
}

\newglossaryentry{SoundnessPLDT}{
	name={$\mathbf{\eta}_{\text{SLD}}(p,m,d,k,\eps)$},
	description={: \quad The soundness parameter function for the $(p,m,d,k)$-simultaneous quantum low-individual degree test, given in~\Cref{lem:simuQLID}},
	sort = {d4b}
}
\newglossaryentry{Hilbertspace}{
	name={$\cH$},
	description={: \quad Hilbert space},
	sort = {e1}
}

\newglossaryentry{BoundHS}{
	name={$\BofH$},
	description={: \quad Bounded operator acting on $\cH$},
	sort = {e2a}
}

\newglossaryentry{Vectornorm}{
	name={$|\ket{\psi}|$},
	description={: \quad Vector norm},
	sort = {e2b}
}

\newglossaryentry{VNA}{
	name={$\alicealg$},
	description={: \quad von Neumann algebras},
	sort = {e2}
}
\newglossaryentry{VNApositive}{
	name={$\alicealg^{+}$},
	description={: \quad Set of positive elements within $\alicealg$ },
	sort = {e3}
}

\newglossaryentry{VNAcommutant}{
	name={$\alicealg'$},
	description={: \quad Commutant of $\alicealg$ },
	sort = {e3}
}

\newglossaryentry{VNAtrace}{
	name={$\tau$},
	description={: \quad Trace function (often assocated with a von Neumann algebra $\alicealg$) },
	sort = {e4}
}

\newglossaryentry{VNAtraceFD}{
	name={$\Tr(\cdot)$},
	description={: \quad Normalized trace for finite dimensional matrix },
	sort = {e4}
}

\newglossaryentry{HSnorm}{
	name={$||A||_2$},
	description={: \quad Hilbert schmidt norm for $A \in \alicealg$, where $\alicealg$ is a tracial von Neumann algebra},
	sort = {e5}
}

\newglossaryentry{Statenorm}{
	name={$\| \psi\|$},
	description={: \quad State norm for the state $\psi$},
	sort = {e5}
}

\newglossaryentry{SDHilbertspace}{
	name={$\cL^2(\alicealg, \tau)$},
	description={: \quad Hilbert space for the standard form of $\alicealg$, where $(\alicealg, \tau)$ is a tracial von Neumann algebra  },
	sort = {e6}
}

\newglossaryentry{SDTracialstate}{
	name={$\ket{\tau}$},
	description={: \quad Vector state associated to the trace $\tau$ for the standard form of $\alicealg$, where $(\alicealg, \tau)$ is a tracial von Neumann algebra  },
	sort = {e6}
}

\newglossaryentry{Oppmap}{
	name={$\alicealg^{op}$},
	description={: \quad The opposite algebra of $\alicealg$  },
	sort = {e7}
}

\newglossaryentry{Oppmapele}{
	name={$a^{op}$}, 
	description={: \quad The bijection map between $\alicealg$ in standard form to $\alicealg'$  through the opposite algebra map  },
	sort = {e7}
}

\newglossaryentry{Postcone}{
	name={$\positivecone$},
	description={: \quad The canonical positive cone for the von Neumann algebra $\alicealg$ in standard form  },
	sort = {e7}
}

\newglossaryentry{Datameas}{
	name={$A_{[f|s]}$},
	description={: \quad Data processing measurement  },
	sort = {f1}
}

\newglossaryentry{GenPauli1}{
	name={$\rho^{W, p}_{a}$},
	description={: \quad The PVM with outcome $a$ associated with Generalized Pauli meausrement over $\bF_{2^{q}}$ for $W \in \{X,Z\}$ },
	sort = {f2}
}

\newglossaryentry{GenPauli2}{
	name={$\rho^{W, p}(a)$},
	description={: \quad The observable associated with Generalized Pauli meausrement over $\bF_{2^{p}}$ for $W \in \{X,Z\}$ over $a \in \bF_{2^p}$, also can be define using $s \in \bF^m_{2^p}$, given by \eqref{eq:genPaulimeasurement}},
	sort = {f2}
}

\newglossaryentry{GenPauli3}{
	name={$\{\rho^{X}_{s} \}_{s \in V}$},
	description={: \quad Generalized Pauli measurement with outcome over a register subspace $V \subseteq \bF_{2^p}^m$},
	sort = {f2}
}

\newglossaryentry{GenPauli4}{
	name={$U_{2 \rightarrow p}$},
	description={: \quad The Unitary which converts between the generalized (qubit) Pauli measurement over $2^p$ to the one qubit Pauli measurement given in~\Cref{lem:PaulitogeneralPauli}},
	sort = {f2}
}

\newglossaryentry{Consistency}{
	name={$\simeq_{\delta}$},
	description={: \quad $\delta$-consistant when the underlying state and distribution is clear. $\delta$-consistant is given by~\eqref{eq:Consistencydef}},
	sort = {f3a}
}

\newglossaryentry{Close}{
	name={$\approx_{\delta}$},
	description={: \quad $\delta$-close when the underlying state and distribution is clear. $\delta$-close is given by~\eqref{eq:Closedef}},
	sort = {f3a}
}

\newglossaryentry{Correlations}
{
	name = {$C_{x,y,a,b}$},
	description={: \quad Correlations, the subscipt $x,y,a,b$ are often omitted}, 
	sort= {g1}
}

\newglossaryentry{Correlationtensor}
{
	name = {$C_q$},
	description={: \quad The set of all quantum tensor product correlations. $C_q(\cX, \cA)$ if the questionset and the asnwer set are specifed, sometime written as $C_*$}, 
	sort= {g1}
}

\newglossaryentry{CorrelationtensorFD}
{
	name = {$C_q^{n}$},
	description={: \quad The set of all quantum tensor correlations realizable in $\cM_n(\cC) \otimes \cM_n(\cC)$}, 
	sort= {g1a}
}

\newglossaryentry{Correlationcommuting}
{
	name = {$C_{qc}$},
	description={: \quad The set of all quantum commuting operator correlations. $C_{qc}(\cX, \cA)$ if the questionset and the asnwer set are specifed}, 
	sort= {g1b}
}

\newglossaryentry{Correlationsyncronous}
{
	name = {$C_{t}^s$},
	description={: \quad The set of all quantum syncronous correlations under model $t \in \{*, co\}$}, 
	sort= {g1c}
}

\newglossaryentry{Synchronicity}
{
	name = {$\deltasync(\mu, C)$},
	description={: \quad The synchronicity for the correlation $C$ under the discribution $\mu$, sometimes written as $\deltasync(\mu, \strategy)$ for quantum strategy instead }, 
	sort= {g1c}
}

\newglossaryentry{Tracialstrategy}
{
	name = {$(\cL^2(\alicealg, \tau),\sigma \ket{\tau}, \{A_a^x \}, \{(B_b^y)^{op}\})$},
	description={: \quad Tracially embeddable strategy, definition given in~\Cref{def:Tracialemd}. Finite dimension translation chart given in~\Cref{fig:tracialemdtofd}}, 
	sort= {g2a}
}

\newglossaryentry{symstrategy}
{
	name = {$(\cL^2(\alicealg, \tau), \sigma \ket{\tau},\{ A_{a}^x \})$},
	description={: \quad Symmetric strategy, definition given in~\Cref{def:symstrategy}.}, 
	sort= {g2a}
}

\newglossaryentry{games}
{
	name = {$\cG$},
	description={: \quad Non-local games, often denoted with $\cG =(\cX, \cA, \mu, D)$ where $\cX$ is the question set, answer set $\cA$, question distribution $\mu$ and validation function $D$}, 
	sort= {h1}
}

\newglossaryentry{acceptgames}
{
	name = {$\cG^{\text{accept}}$},
	description={: \quad The accepting syncronous game, with $\omega^{*}(\cG^{\text{accept}}) = \omega^{co}(\cG^{\text{accept}}) =1$}, 
	sort= {h2}
}

\newglossaryentry{rejectgames}
{
	name = {$\cG^{\text{reject}}$},
	description={: \quad The rejecting syncronous game, with $\omega^{*}(\cG^{\text{reject}}) = \omega^{co}(\cG^{\text{reject}}) =\frac{1}{3}$}, 
	sort= {h3}
}

\newglossaryentry{zoracgame}
{
	name = {$\cG_{\bot}$},
	description={: \quad Oracularization transformation for the game $\cG$}, 
	sort= {h3}
}

\newglossaryentry{anchorgame}
{
	name = {$\cG_{\bot}$},
	description={: \quad Anchoring transformation for the game $\cG$}, 
	sort= {h4}
}

\newglossaryentry{parallelsquestion}{
	name={$(\vec{x}, \vec{y})$},
	description={: \quad Question pair for $r$-fold parallel repetition},
	sort = {h5}
}
\newglossaryentry{parallelszanswer}{
	name={$(\vec{a}, \vec{b})$},
	description={: \quad Answer pair for $r$-fold parallel repetition},
	sort = {h5}
}

\newglossaryentry{parallelrepetitiongame}
{
	name = {$\cG^{\otimes r}$},
	description={: \quad $r$-fold parallel repetition of a game $\cG$}, 
	sort= {h5}
}

\newglossaryentry{Corrgamevalue}{
	name={$\omega(\cG, C)$},
	description={: \quad The value of the game under correlation $C$ or strategy $\strategy$},
	sort = {h6}
}

\newglossaryentry{tensorvalue}{
	name={$\omega^*(\cG)$},
	description={: \quad Tensor product value of $\cG$},
	sort = {h7}
}

\newglossaryentry{commutingvalue}{
	name={$\omega^{co}(\cG)$},
	description={: \quad Commuting operator value of $\cG$},
	sort = {h8}
}

\newglossaryentry{syncvalue}{
	name={$\omega^{t}_s(\cG)$},
	description={: \quad Syncronous value},
	sort = {h9}
}

\newglossaryentry{CondLinearfun}{
	name={$\decidertypedfunction$},
	description={: \quad Conditionally linear function},
	sort = {i1}
}

\newglossaryentry{CondLinearfundetail}{
	name={$\decidertypedfunction_{j, s}$},
	description={: \quad The $j$th level CL function used to define $\decidertypedfunction$ when the previous function returns $s$. },
	sort = {i1a}
}

\newglossaryentry{CondLineardist}{
	name={$\decidertypedfunction^P$},
	description={: \quad The CL function which is used to define a CL distribution (\Cref{def:CLdistribution}) for $P \in \{A,B\}$~. Same notation ($\decidertypedfunction^V$) is used to define a typed CL distribution (\Cref{def:typedCLsampler}) for $v \in \decidertypedtype$ },
	sort = {i1b}
}

\newglossaryentry{Neightindvec}{
	name={$\text{neigh}_{\decidertypedquestionpair}(v)$},
	description={: \quad Indicator vector for $v \in \decidertypedtype$, where $(\decidertypedtype, \decidertypedquestionpair)$ is a graph },
	sort = {i1c}
}

\newglossaryentry{Verifiersequence}{
	name={$\verifiersequence$},
	description={: \quad A CL Verifier sequence for a $\MIP^*$/$\MIPco$ protocol, as specified in~\Cref{def:CLverifier} },
	sort = {i2a}
}

\newglossaryentry{Samplersequence}{
	name={$\samplerTM_{\verifiersequence}$},
	description={: \quad The sampler for $\verifiersequence$, as specified in~\Cref{def:CLverifier} },
	sort = {i2b}
}

\newglossaryentry{Decidersequence}{
	name={$\deciderTM_{\verifiersequence}$},
	description={: \quad The decider for $\verifiersequence$, as specified in~\Cref{def:CLverifier} },
	sort = {i2c}
}

\newglossaryentry{CLlevelfn}{
	name={$\mathbf{k}(n)$},
	description={: \quad The level function for a CL verifier },
	sort = {i2d}
}

\newglossaryentry{CLcardfn}{
	name={$\mathbf{m}(n)$},
	description={: \quad The cardinality function for a CL verfier  },
	sort = {i2e}
}

\newglossaryentry{Cfieldfn}{
	name={$\mathbf{p}(n)$},
	description={: \quad The field size function for a CL verfier  },
	sort = {i2f}
}

\newglossaryentry{tensorstate}{
	name={$\psi^{\alicealg_1 \alicealg_2}$},
	description={: \quad A state define within $\alicealg_1 \otimes \alicealg_2$, also notation for the restriction of $\alicealg_1 \otimes \alicealg_2$ if the state is define in a bigger Hilbert space  },
	sort = {w1a}
}

\newglossaryentry{CQstate}{
	name={$\phi^{\X\alicealg}$},
	description={: \quad Classical quantum state with the classical component being $\cX$ },
	sort = {w1b}
}

\newglossaryentry{CQstateshort}{
	name={$\puretomixed{x_1, \cdots, x_i}$},
	description={: \quad Shorthand for $\ketbra{x_1}{x_1} \otimes \ketbra{x_2}{x_2} \cdots \otimes \ketbra{x_i}{x_i}$, },
	sort = {w1c}
}

\newglossaryentry{Relativeent}{
	name={$\D(\psi_1 \| \psi_2)$},
	description={: \quad Relative entropy betweenm the state $\psi_1$ and $\psi_2$  },
	sort = {w2a}
}

\newglossaryentry{Mutualinfo}{
	name={$\rI(\alicealg_1 : \alicealg_2)_{\psi}$},
	description={: \quad Mutual information between algebra $\alicealg_1$ and $\alicealg_2$ for the state $\psi$ },
	sort = {w2b}
}

\newglossaryentry{Noiseparameter}{
	name={$\eta_{\text{Anchor}}$},
	description={: \quad The noise parameter },
	sort = {w3a}
}
\newglossaryentry{Criticalset}{
	name={$C$},
	description={: \quad The critical set given by~\Cref{prop:existcriticalset}  },
	sort = {w3b}
}

\newglossaryentry{PRparameter}{
	name={$\eta_{\text{PR}}$},
	description={: \quad The constant given in~\Cref{lem:classical_skew} },
	sort = {w3ba}
}

\newglossaryentry{Critcoordinate}{
	name={$\mathbf{R}_C$},
	description={: \quad The question/answer correlation for coordinates in the critical set $C$, consist of question distribution $\mathbf{Q}_C = (\bX_C, \bY_C)$ and answer distribution $\mathbf{S}_C = (\bA_C, \bB_C)$ },
	sort = {w3c}
}

\newglossaryentry{Depbreaking}{
	name={$\bOmega$},
	description={: \quad The dependence breaking probability distribution define for coordinates outside of the critical set, given in~\Cref{def:Depbreaking}. We further use $\bOmega_{-i}$ to denote the distribution without the coordiate $i$ },
	sort = {w3d}
}

\newglossaryentry{Parastatedependonomega}{
	name={$\ket{\Phi_{ (\omega_{-i}, \vec{r}_{C}), s,y}} $},
	description={: \quad The (unnormalized) post measurement state $\ket{\psi}$ with measurements condition on $\bOmega_\mi = \omega_\mi$, $\blR_C = \vr_C$, $s$ and $y$ are either the normal measurement, or the slanted measuremnt given in~\eqref{eq:states-operators-4} if $s$ or $y$ are $\dummyx$ },
	sort = {w4a}
}

\newglossaryentry{Parastatedependonomeganormalfactor}{
	name={$\gamma_{(\omega_{-i}, \vec{r}_C ),s,y}$},
	description={: \quad The normalized factor for the state $\ket{\Phi_{ (\omega_{-i}, \vec{r}_{C}), s,y}}$ },
	sort = {w4b}
}

\newglossaryentry{Parastatedependonomeganormal}{
	name={$\ket{\wt{\Phi}_{ (\omega_{-i}, \vec{r}_{C}), s,y}}$},
	description={: \quad Normalized version of $\ket{\Phi_{ (\omega_{-i}, \vec{r}_{C}), s,y}} $ },
	sort = {w4c}
}

\newglossaryentry{CQcond1}{
	name={$\Xi^{\bOmega \bX_C \bY \mathbf{Q}_C \alicealg}$},
	description={: \quad  The classical quantum state which packages all of the post-measurement state from the strategy $\strategy^{\otimes r}$, with outcome being restricted to $(\ba_C, \bb_C) \in \blS_C$. For each index, Alice's measurement is being condition on $\Omega = \omega$ and $\blQ_C = (\vec{x}_C, \vec{y}_C)$, Bob's mesurement only depends on $\vec{y}$ },
	sort = {w5a}
}

\newglossaryentry{CQcond3}{
	name={$\xi^{\bOmega \bX_C \bY \mathbf{Q}_C \alicealg}$},
	description={: \quad  $\Xi^{\bOmega \bX_C \bY \mathbf{Q}_C \alicealg}$ being additionally conditioning on $(\vx_C, \vy_C, \va_C, \vb_C)$ giving a winning question/answer on all coordinates in the critical set $C$  },
	sort = {w5b}
}

\newglossaryentry{CQcond2}{
	name={$\Lambda^{\bOmega \bX \bY_C  \mathbf{Q}_C \alicealg} $},
	description={: \quad The classical quantum state which packages all of the post-measurement state from the strategy $\strategy^{\otimes r}$, with outcome being restricted to $(\ba_C, \bb_C) \in \blS_C$. For each index, Bob's measurement is being condition on $\Omega = \omega$ and $\blQ_C = (\vec{x}_C, \vec{y}_C)$, Alice's mesurement only depends on $\vec{x}$ },
	sort = {w5c}
}

\newglossaryentry{CQcond4}{
	name={$\lambda^{\bOmega \bX \bY_C  \mathbf{Q}_C \alicealg} $},
	description={: \quad  $\lambda^{\bOmega \bX_C \bY \mathbf{Q}_C \alicealg}$ being additionally conditioning on $(\vx_C, \vy_C, \va_C, \vb_C)$ giving a winning question/answer on all coordinates in the critical set $C$  },
	sort = {w5d}
}

\newcommand{\Turingmachinehalt}{\small\textpmhg{F}\normalsize}

\newcommand{\mi}{{-i}}
\newcommand{\dummy}{{\mathsmaller{\perp}}}

\newcommand\dummyx{{\dummy\!/x}}

\newcommand{\ac}{\mathbf{Q}_C}

\newcommand{\va}{{\vec{a}}}
\newcommand{\vb}{{\vec{b}}}
\newcommand{\vx}{{\vec{x}}}
\newcommand{\vy}{{\vec{y}}}
\newcommand{\vr}{{\vec{r}}}

\newcommand{\xp}{{x,\dummy}}
\newcommand{\py}{{\dummy, y}}
\newcommand{\xy}{{x,y}}
\newcommand{\pxy}{{\dummy\!/x,y}}
\newcommand{\pxp}{{\dummy\!/x,\dummy}}

\newcommand{\pp}{{\dummy, \dummy}}

\newcommand{\sspx}{{(\bomega_{\mi}, \vr_C) , \dummy,x}}
\newcommand{\ssxp}{{(\bomega_{\mi}, \vr_C) , x, \dummy}}
\newcommand{\sspy}{{(\bomega_{\mi}, \vr_C) , \dummy,y}}
\newcommand{\sspp}{{(\bomega_{\mi}, \vr_C) , \dummy,\dummy}}
\newcommand{\ssxy}{{(\bomega_{\mi}, \vr_C) , x,y}}
\newcommand{\sspxy}{{(\bomega_{\mi}, \vr_C) , \dummyx,y}}
\newcommand{\sspxp}{{(\bomega_{\mi}, \vr_C) , \dummyx,\dummy}}

\newcommand{\Ganchor}{\cG_\dummy}
\newcommand{\ReedMcode}{\text{RM}}

\newcommand{\Mplayer}{M_{\text{player}}}
\newcommand{\Mvalue}{M_{\text{value}}}
\newcommand{\bMplayer}{\textbf{M}_{\text{player}}}
\newcommand{\bMvalue}{\textbf{M}_{\text{value}}}

\newcommand{\blQ}{\mathbf{Q}}
\newcommand{\blS}{\mathbf{S}}
\newcommand{\blR}{\mathbf{R}}

\newcommand{\Idpoly}{\text{IdPoly}}

\begin{document}

\title{$\MIPco=\coRE$ }

\author{
Junqiao Lin\thanks{Junqiao.Lin@cwi.nl} \\ CWI \& QuSoft, Amsterdam, The Netherlands
}

\clearpage\maketitle
\thispagestyle{empty}

\begin{abstract}
In 2020, a landmark result by Ji, Natarajan, Vidick, Wright, and Yuen showed that $\MIP^*$, the class of languages that can be decided by a classical verifier interacting with multiple computationally unbounded provers sharing entanglement in the tensor product model, is equal to $\RE$. We show that the class $\MIPco$, a complexity class defined similarly to $\MIP^*$ except with provers sharing the commuting operator model of entanglement, is equal to the class $\coRE$.\footnote{We remark that the ${co}$ modifiers on both sides of $\MIPco=\coRE$ refer to different things!}. This shows that giving the provers two different models of entanglement leads to two completely different computational powers for interactive proof systems. Our proof builds upon the compression theorem used in the proof of $\MIP^*=\RE$, and we use the tracially embeddable strategies framework to show that the same compression procedure in $\MIP^* =\RE$ also has the same desired property in the commuting operator setting. We also give a more streamlined proof of the compression theorem for non-local games by incorporating the synchronous framework used by Mousavi et al. [STOC 2022], as well as the improved Pauli basis test introduced by de la Salle [ArXiv:2204.07084]. 

We introduce a new equivalence condition for $\RE$/$\coRE$-complete problems, which we call the weakly compressible condition. We show that both $\MIP^*$ and $\MIPco$ satisfy this condition through the compression theorem, and thereby establish that the uncomputability for $\MIP^*$ and $\MIPco$ can be proved under a unified framework (despite these two complexity classes being different). Notably, this approach also gives an alternative proof of the $\MIP^*=\RE$ theorem, which does not rely on the preservation of the entanglement bound. In addition to non-local games, this new condition could also potentially be applicable to other decision problems. 
\end{abstract}


\newpage

\setcounter{page}{1}
\tableofcontents
\newpage
\begin{CJK*}{UTF8}{gbsn}

\section{Introduction}

A two-prover non-local game, $\cG$ is played between a polynomial-time verifier and two computationally unbounded, but non-communicating provers, which we name Alice and Bob. In this scenario, the verifier first samples a pair of questions $(x,y)$ from a predetermined distribution $\mu$ and sends $x$ to Alice (resp.~$y$ to Bob), who then responds with the answer $a$  (resp.~$b$). The verifier then computes a predicate function $D(x,y,a,b)$ that outputs either $0$ or $1$, with $1$ indicating the verifier accepts (meaning the provers win the game), and $0$ if the verifier rejects (meaning the provers lose the game). The provers know the initial question distribution and the predicate function and can strategize before the game, but cannot communicate during the game.

Non-local games are widely studied in the quantum information community. Famously,~\cite{bellEinsteinPodolskyRosen1964a} showed that if the two provers employ the laws of quantum mechanics in their strategy, certain non-local games can be won with a higher probability, and this has since led to several experimental setups refuting the local realist model in our universe. Additionally, non-local games play an important role in quantum cryptography, both in the study of device-independent cryptography~\cite{bowlesSelftestingPauliObservables2018,jainParallelDeviceIndependentQuantum2020}, and in the analysis of certain quantum key exchange protocols~\cite{ekertQuantumCryptographyBased1991}. 

When discussing models of entanglement in a non-local game, two natural models are considered. The first model, which is the more conventional model in the quantum information community, is the \textit{tensor product model}. In this model, the provers are assumed to share an entangled state defined by a unit vector in a tensor factor of two potentially infinite-dimensional Hilbert spaces, $\cH_A \otimes \cH_B$. Alice, in this model, can make local measurements on the Hilbert space $\cH_A$ and Bob similarly on the Hilbert space $\cH_B$ in order to sample an answer pair $(a,b)$. The other, more general model is the \textit{commuting operator model}, which is used in, e.g., the Haag-Kastler axioms of quantum mechanics~\cite{haagAlgebraicApproachQuantum1964}. This model defines the entangled state shared between the provers within a single Hilbert space $\cH$. In this model, the provers are allowed to make measurements on the same Hilbert space as long as their measurement operators commute. We call the optimal winning probability for a game $\cG$ over all tensor product strategies the \textit{tensor product value} of the game $\cG$, or $\omega^*(\cG) \in [0,1]$, and the optimal winning probability over all the commuting operator strategies as the \textit{commuting operator value} of the game $\cG$ or $\omega^{co}(\cG)\in [0,1]$. Since the commuting operator model includes the tensor product model (by considering the tensor product space as a single Hilbert space), we have $\omega^*(\cG) \leq \omega^{co}(\cG)$ for all non-local games $\cG$. 

\paragraph{The complexity of non-local games. }In computational complexity, non-local games are known as two-prover one-round \textit{Multiprover Interactive Proof systems}, and these games are used to model the complexity class $\MIP$. In this paper, when discussing a Multiprover Interactive Proof system, unless otherwise stated, we implicitly assume that it is the two-prover one-round variant. Roughly speaking, a language is in $\MIP = \MIP(\frac{2}{3}, \frac{1}{3})$ if every instance $x$ of the language can be translated into a non-local game such that if $x$ is in the language, then the provers have a classical strategy that wins the game with a high probability ($\geq \frac{2}{3}$). If $x$ is not in the language, then the provers cannot win the game with high probability ($< \frac{1}{3}$) given any classical strategy. Famously~\cite{babaiNondeterministicExponentialTime1991a} showed that $\MIP=\NEXP$, where $\NEXP$ is the set of languages decidable by a nondeterministic exponential-time Turing machine, and the technique used provided the foundation for showing the famous $\class{PCP}$ theorem~\cite{aroraProbabilisticCheckingProofs1998,aroraProofVerificationHardness1998}. 


The notion of a Multiprover Interactive proof system with ``entangled provers" was first introduced in~\cite{cleveConsequencesLimitsNonlocal2004}, where the computational model is defined similarly to $\MIP$, except the provers are now given access to shared quantum entanglement. In this paper, we use $\MIP^*$ (resp.~$\MIPco$) to denote the complexity class defined by multiprover interactive proof systems with the tensor product model of entanglement (resp.~multiprover interactive proof systems with the commuting operator model of entanglement). For $t \in \{*, co\}$, the complexity class $\MIP^{t}$ is complete under polynomial time reduction with respect to the $t$ non-local game value problem, which is defined as a decision problem over the following two sets of non-local games:
\begin{equation*}
	\Language_{\text{yes}}^{t} = \left\{ \cG: \omega^{t}(\cG) = 1 \right\} \qquad \text{ and } \qquad	\Language_{\text{no}}^{t} = \left\{ \cG: \omega^{t}(\cG) < \frac{1}{2} \right\}. 
\end{equation*} 
For clarity, we refer to the $*$ non-local game value problem as the \textit{tensor product value problem} and the $co$ non-local game value problem as the \textit{commuting operator value problem}, and we define them more formally in~\Cref{def:entangledgameproblem}. 

A recent breakthrough result by Ji et al. shows that $\MIP^*=\RE$~\cite{jiMIPRE2022a}, where $\RE$ is the complexity class that contains all decision problems in which the ``yes" case can be verified by a Turing machine in finite time. In other words, it is possible to reduce an instance of the halting problem to an instance of a non-local game $\cG$ for which $\omega^{*}(\cG) = 1$ if the Turing machine halts in a finite number of steps, and $\omega^*(\cG) < \frac{1}{2}$ if the Turing machine does not halt. Previously, it was known that if the tensor product and the commuting operator value for a non-local game coincide, then one can construct a terminating algorithm that estimates the quantum value of the game up to some constant error~\cite{navascuesConvergentHierarchySemidefinite2008b}. The existence of such an algorithm, in conjunction with the $\MIP^*=\RE$ theorem, implies the existence of a game that has a commuting operator value strictly larger than its tensor product value. This, in turn, provides a negative answer to both Tsirelson's problem in quantum information and Connes's Embedding problem, a long-standing open problem in operator algebra~\cite{connesClassificationInjectiveFactors1976, ozawaConnesEmbeddingConjecture2013a}. 

In contrast, another natural variant for $\MIP$ is $\MIPco$, which is defined similarly to $\MIP^*$, but the provers are given access to the commuting operator model of entanglement instead. $\MIPco$ is known to be in $\coRE$ by the algorithm known as the ``NPA hierarchy" proposed in~\cite{navascuesConvergentHierarchySemidefinite2008b}, and it has been conjectured to be $\coRE$-complete~\cite[Section 1.4]{jiMIPRE2022a}. 

As a main contribution of this paper, we give a positive answer to this conjecture. 
\begin{theorem} \label{thm:MIPcocoRE}
	$\MIPco = \coRE$. 
\end{theorem}
This is a nice complementary result to $\MIP^* = \RE$, as it implies that employing two different axioms of quantum entanglement gives two completely different uncomputable verification powers to a $\MIP$ protocol. The proof of the main theorem follows by showing that the key technique used in $\MIP^*=\RE$, the gap compression for non-local games, also holds in the commuting operator model. A key part of this adaptation relies on the recently discovered tracially embeddable strategies framework~\cite{linTracialEmbeddableStrategies2024}. Then, we combine the gap compression theorem with the compression proof approach from~\cite[Section 1.1]{mousaviNonlocalGamesCompression2022} to show that $\coRE \leq_p \MIPco$. In conjunction with the NPA hierarchy algorithm, this also shows that all $p$-prover $r$-round $\MIPco$ protocols are equivalent to the $2$-prover $1$-round $\MIPco$ protocol. 

We also streamline the proof for the gap compression theorem in this paper. Mainly, we incorporate some recent simplifications, such as the simplification to the Pauli basis test given in~\cite{delasalleSpectralGapStability2022} and the synchronous game framework used in~\cite{mousaviNonlocalGamesCompression2022} for zero-gap $\MIP^*$ in our proof. Since we make use of the synchronous game framework, our result also implies that $\MIPco_s$, the complexity for the commuting operator value problem for synchronous games, is also $\coRE$-complete. 

\paragraph{The compressible condition. }The proof of the $\MIP^*=\RE$ theorem relies on a key technique known as gap compression theorem for non-local games. On a high level, the gap compression theorem shows the existence of an algorithm that takes a sequence of non-local games $\{\cG_n\}_{n \in \bN}$ with a polynomial time verifier and outputs a ``compressed" sequence of non-local games $\{\cG_n^{\Compressgame}\}_{n \in \bN}$ with a polylog time verifier. Furthermore, the compression procedure preserves the value for the tensor product value problem. 

To be more precise, for all $n \in \bN$, if $\omega^*(\cG_n) = 1$, then $\omega^{*}(\cG_n^{\Compressgame})=1$ and if $\omega^*(\cG_n) < \frac{1}{2}$, then $\omega^{*}(\cG_n^{\Compressgame}) < \frac{1}{2}$. The compression theorem is known as the ``gap" compression theorem because it preserves the $\frac{1}{2}$ gap between the yes/no cases for the tensor product value problem. The gap compression given in~\cite{jiMIPRE2022a} also has an additional clause on entanglement lower bound on the gap compression theorem, where the provers need a higher entanglement dimension, the dimension of the Hilbert space in which their joint entangled state is defined on, in the compressed game compared to the original game to formulate a strategy which wins with a probability of at least $\frac{1}{2}$. If we intuitively view the entanglement dimension as the amount of ``resources" that the provers need for the game, then the compression theorem essentially states that it is possible for the verifier to perform ``less work" when playing a non-local game with two entangled provers in the tensor product model. However, the provers would potentially have to prepare ``more resources" in the form of an entangled state with a larger Hilbert space dimension in order to convince the verifier to accept the given game. 

Based on the compression theorem, ~\cite{jiMIPRE2022a} constructs a game $\cG$ for every Turing machine such that the provers needs an entangled strategy whose entanglement dimensions correlate with the runtime of the Turing machine to succeed on the game with probability greater than $\frac{1}{2}$. Hence, if the given Turing machine does not halt, then $\omega^*(\cG) < \frac{1}{2}$, showing that $\RE \leq_p \MIP^*$. A similar observation was made in the first version of~\cite{nezhadiRecursiveCompressionMethod2025}, where for certain $\RE$/$\coRE$-complete problems such as the Halting problem and some versions of the word problem, a ``resources-dependent" version of the compression theorem can be formulated. 

Interestingly, \cite[Theorem 6.10]{mousaviNonlocalGamesCompression2022} shows that $\MIP^*$ being $\RE$-hard implies the existence of the gap compression theorem, which does not require the entanglement lower bound condition stated earlier. Since the entanglement lower bound condition is a specialized condition which only seems to apply in the context of non-local games with finite dimensional entanglement, an interesting question is whether this condition is necessary for showing that $\RE \leq_p \MIP^*$. 

In this paper, we give a new condition for decision problems being $\RE$/$\coRE$-complete which we refer to as ''compressible". Intuitively, decision problems that are compressible admit a ``compression theorem", similarly to non-local games, but without the need for the preservation of resources like previous work. Using this new formulation, we give an alternative proof for $\RE \leq_p \MIP^*$, which only relies on the gap compression theorem and the existence of an algorithm that \textbf{halts in the``yes" case}, and \textbf{runs forever in the ``no" case} for the tensor product value problem (this condition is trivially satisfied by $\MIP^* \leq_p \RE$). This, in turn, shows that the entanglement lower bound condition is \textbf{not} needed for the proof of $\MIP^*=\RE$. The same formulation can be used to show $\MIPco=\coRE$, where a gap compression theorem in conjunction with the existence of an algorithm which \textbf{halts in the ``no" case}, and \textbf{runs forever in the ``yes" case} for the commuting operator value problem implies that $\coRE \leq_p \MIPco$. Since our formulation works for general decision problems, we believe our approach can be generalized to establish the conjectured $\RE$/$\coRE$-completeness for other decision problems. 




\paragraph{Explicit separation between the models of entanglement. } As a corollary of the $\MIP^*=\RE$ theorem, the tensor product model of entanglement is mathematically different from the commuting operator model of entanglement. A natural follow-up question is whether we can find an experimental setup similar to the Bell test scenario to determine which is the right way to model entanglement in our universe? \cite[Theorem 12.10]{jiMIPRE2022a} shows the existence of a game $\cG$ such that $\omega^*(\cG) < \frac{1}{2}$ and $\omega^{co}(\cG) = 1$, which, in theory, could serve as the Bell test mentioned earlier. However, the question and answer sets for the given game have a magnitude of $10^{20}$, making it impractical for experimental implementation. 

We show that given a non-local game $\cG$, the promise problem of deciding whether $\omega^*(\cG) = \omega^{co}(\cG)$ or $|\omega^*(\cG) - \omega^{co}(\cG)| > c$ for any fixed constant $c \in [0,1]$ is $\RE$-complete. In other words, there is no algorithm which can be used to find explicit separation between the tensor product model and the commuting operator model for general non-local games! On the brighter side, our proof also gives a reduction for any non-halting Turing machine to an instance of a game such that $|\omega^*(\cG) - \omega^{co}(\cG)| > c$; however, since the technique used is similar to the one given in\cite[Theorem 12.10]{jiMIPRE2022a}, we suspect that the question size and answer size would be as large as those in~\cite{jiMIPRE2022a}.

\paragraph{Parallel repetition for the commuting operator model. } In an effort to show the gapped compression theorem, we also give the first ``informational-theoretical" proof of parallel repetition theorem for the commuting operator model. Intuitively, a parallel repetition theorem states that if $r$ instances of a non-local game are played in parallel, then the value of the game decays exponentially. Whether a parallel repetition theorem exist in general for the tensor product value of a game is still open (as the best bound is given by~\cite{yuenParallelRepetitionTheorem2016} where the decay is polynomial). However, a parallel repetition theorem (for the tensor product value) is known to hold for many special classes of games, see~\cite{cleveStrongParallelRepetition2008, kempeParallelRepetitionEntangled2011, jainParallelRepetitionTheorem2014, chaillouxParallelRepetitionFree2015, chungParallelRepetitionEntangled2015, dinurParallelRepetitionTheorem2015, bavarianAnchoredParallelRepetition2021}. In particular, a key part for showing a gap compression theorem for $\MIP^*=\RE$ in~\cite{jiMIPRE2022a} is the parallel repetition of anchored games. In contrast, very little seems to be known about parallel repetition for commuting operator values. To our knowledge, the only class of games that are known to admit a parallel repetition theorem is the XOR games~\cite{cleveStrongParallelRepetition2008}, in which the tensor product value and the commuting operator value coincide~\cite{tsirelsonQuantumAnaloguesBell1987}. 

We extended the parallel repetition results for anchoring games from~\cite{bavarianAnchoredParallelRepetition2021} to the commuting operator framework. In particular, we show that the informational-theoretical tools used in~\cite{bavarianAnchoredParallelRepetition2021}, such as an analogue of mutual information and Ulhmann's theorem, have an appropriate analogue in the commuting operator model and hence the majority of the proof from~\cite{bavarianAnchoredParallelRepetition2021} can be shown for the commuting operator model (see~\Cref{sec:qinfo} for more details). We believe these techniques also could be useful to show a parallel repetition theorem for other classes of games, and could potentially be of interest to the quantum information community. 

The proof of the parallel repetition for anchored games in this paper emerged from an early collaboration with William Slofstra and Henry Yuen. The proof of the parallel repetition theorem uses vastly different techniques from proving the main contribution of this paper, and we choose to present them in~\Cref{sec:parallelrepappendix} for readability.

\subsection{Technical overview}



\subsubsection{The compressible condition. } \label{sec:introcompression}


To introduce the compressible condition, we first present a simplified version of the compressible condition for languages and a simplified argument for a reduction from $\RE$ to a language which is compressible. We assume the standard Turing machine as the model of computation in this paper. For a Turing machine $\mathtt{T}$, we use $|\mathtt{T}|$ to denote the description length of the Turing machine. To abuse notation slightly, we also use the same notation to denote both the description of the Turing machine and the function that the Turing machine implements. We also take the convention that every integer given as an input for a Turing machine is being represented under its binary representation. This means that any integer $n$ will be treated as an $O(\polylog(n))$-bit input for all Turing machines. We give the definition of a compressible language below. 



\begin{definition}[Compressible language] \label{def:compressiblelanguage}
Let $\Language \subseteq \{0,1\}^*$. We say that $\Language$ is \textit{compressible} if there exists a universal algorithm $\texttt{Compress}^{\Language}$ with the following properties:
\begin{enumerate}
	\item (Output) $\texttt{Compress}^{\Language}$ which takes as input a description of a Turing machine $\texttt{Seq}_{\Language}$, and outputs a description of a Turing machine $\texttt{Seq}_{\Language}^{\Compressgame}$, computes the function $\texttt{Seq}_{\Language}^{\Compressgame}: \bN \rightarrow \{0,1\}^*$ in $O(\polylog(n))$ time.
	\item (Runtime): $\texttt{Compress}^{\Language}$ runs in $O(|\langle \texttt{Seq}_{\Language}\rangle|)$ time. 
	\item If $\texttt{Seq}_{\Language}$ implements a function which maps $\bN \rightarrow \{0,1\}^*$ and runs in $O(\poly(n))$ time, then for all $n \in \bN$, the following holds:
	\begin{itemize}
		\item (Completeness): If $\texttt{Seq}_{\Language}(n) \in \Language$, then $ \texttt{Seq}_{\Language}^{\Compressgame}(n) \in \Language$, 
		\item (Soundness): If $\texttt{Seq}_{\Language}(n) \not\in \Language$, then $ \texttt{Seq}_{\Language}^{\Compressgame}(n) \not\in \Language$.
	\end{itemize}
\end{enumerate}
\end{definition}
We point out that based on the above definition, if $\Language$ is compressible, then the language $\texttt{co}\Language = \{0,1\}^* \setminus \Language$ is also compressible. An important remark about the compressible condition is that the $\texttt{Compress}^{\Language}$ algorithm is language dependent. In other words, $\texttt{Compress}^{\Language}$ takes \textbf{any} Turing machine which computes a sequence of strings in $O(\poly(n))$ time and map it to a Turing machine which computes a sequence of strings with a significantly smaller runtime while still preserving whether the $n$th string of the output is in $\Language$ or not in $\Language$.



We remark that this is an unnatural condition for a language $\Language$, as $\texttt{Seq}_{\Language}^{\Compressgame}$, when generating the $n$th instance, cannot generate the $n$th instance of $\texttt{Seq}_{\Language}$ due to the smaller runtime requirement and, hence, cannot even decide whether $\texttt{Seq}_{\Language}(n) \in \Language$ or $\texttt{Seq}_{\Language}(n) \not\in \Language$. The $\texttt{Compress}^{\Language}$ algorithm has to, in some way, manipulate the description of the Turing machine $\texttt{Seq}_{\Language}$ in a black-box way to reduce the runtime.

Given a compressible language $\Language$ that is also ``non-trivially" in $\RE$, our goal is to show that $\Language$ must be $\RE$-complete. We first give a more precise definition for $\Language$ being ``non-trivially" in $\RE$ by making the following assumption:
\begin{enumerate}
	\item  $\Language \in \RE$. In other words, there exists a Turing machine $\texttt{Algo}_{\Language}: \{0,1\}^* \rightarrow \{0,1\}$, such that $\texttt{Algo}_{\Language}$, when running on the input $x \in \Language$, halts in finite time and outputs $1$, and \textbf{runs forever} if given input $y \not\in \Language$  (this can be achieved by changing the termination condition for the ``no" case for $\texttt{Algo}_{\Language}$ to running an infinite loop).
	\item $\Language$ is not trivial, meaning $|\Language| = | \left(\{0,1\}^* \setminus \Language\right)| = \infty$. There exist $x_{\text{yes}} \in \Language$ and $x_{\text{no}} \in \{0,1\}^* \setminus \Language$ which are both trivially computable. 
	\item We can generate an instance $x_{\text{yes}} \in \Language$, and an instance $x_{\text{no}} \not\in \Language$. 
\end{enumerate}
We now show that $\RE \leq \Language$ (where $\Language_1 \leq \Language_2$ implies that there exists a mapping reduction from $\Language_1$ to $\Language_2$). Let \gls{TMD}~be a Turing machine (\Turingmachinehalt~is the only non-western character used in this paper, and we use it to emphasize the fact that it is an instance of the Halting problem instead of a subroutine defined within this paper); we wish to reduce \Turingmachinehalt into an instance $x_{\text\Turingmachinehalt} \in \{0,1\}^*$ such that $|x_{\text\Turingmachinehalt}| = \poly(|\text{\Turingmachinehalt}|)$ and 
\begin{itemize}
	\item If \Turingmachinehalt halts in a finite number of steps, then $x_{\text\Turingmachinehalt} \in \Language$,
	\item If \Turingmachinehalt does not halt, then $x_{\text\Turingmachinehalt} \not\in \Language$. 
\end{itemize}

Consider the following function $\texttt{Seq}_{\Language}: \bN \rightarrow \{0,1\}^*$ defined by~\Cref{pseu:seqintro}.\\
	\vspace{10pt}
\IncMargin{1em}
\begin{algorithm}[H]
	\DontPrintSemicolon
	
	\textbf{Input}: Integer $n$. 
	
	Run \Turingmachinehalt for $n$ steps. If \Turingmachinehalt halts in the given steps, \textbf{return} $x_{\text{yes}}$, the yes-instance guaranteed by the non-triviality condition above. 
	
	Compute $\langle\texttt{Seq}_{\Language}\rangle$, the description for $\texttt{Seq}_{\Language}$. 
	
	Compute $\texttt{Seq}_{\Language}(1)$. 
	
	Simulate $\texttt{Algo}_{\Language}$, the $\RE$ algorithm for $\Language$ guaranteed by assumption 1 above, on the input $\texttt{Seq}_{\Language}(1)$ for $n$ steps. If the algorithm halts in the given steps, \textbf{return} $x_{\text{no}}$, the no-instance guaranteed by the non-triviality condition. 
	
	Compute $\texttt{Compress}^{\Language}(\langle\texttt{Seq}_{\Language}\rangle)$ and obtain the description for $\langle\texttt{Seq}_{\Language}^{\Compressgame}\rangle$. 
	
	\textbf{Return} $\texttt{Seq}_{\Language}^{\Compressgame}(n+1)$. 
	
	\caption{The description of $\texttt{Seq}_{\Language}$ for demonstrating the generalized compression framework.}
	\label{pseu:seqintro}
\end{algorithm}\DecMargin{1em}
\vspace{10pt} 
In the source code above, we use a self-referential trick to make $\texttt{Seq}_{\Language}$ perform computation steps on its own source code. We remark that this step can be done in polynomial time with respect to the description length of $\texttt{Seq}_{\Language}$ using Kleene's recursion theorem, and we refer to~\cite[Chapter 14.2]{jonesComputabilityComplexityProgramming1997} and \cite[Chapter 6.1]{sipserIntroductionTheoryComputation2006} for more details. 

We first analyze the runtime of $\texttt{Seq}_{\Language}$. Lines 2 and 5 of~\Cref{pseu:seqintro} trivially take time $O(n)$. By the recursion theorem mentioned earlier, we see that lines 3, 4 and 6 take time based on the description length of $\langle\texttt{Seq}_{\Language}\rangle$, independent of $n$. By looking at~\Cref{pseu:seqintro}, we see that the description length of $\langle\texttt{Seq}_{\Language}\rangle$ depends \textbf{only} on the description length of $\text{\Turingmachinehalt}$. Line 7 takes time $O(\polylog(n))$ by the definition of $\texttt{Compress}^{\Language}$. Hence, the runtime for $\texttt{Seq}_{\Language}(n)$ is $O(n \cdot \log(n) + \poly(|\text{\Turingmachinehalt}|)) = O(\poly(n))$ (since~\Turingmachinehalt is a fixed Turing machine, its description length is independent of the variable $n$). Thus, by the definition of $\texttt{Compress}^{\Language}$, since $\texttt{Seq}_{\Language}$ runs in $O(\poly(n))$ time,  the compression step on line 6 outputs a Turing machine that is a ``compressed" version of $\texttt{Seq}_{\Language}$. 

We claim that $x_{\text{\Turingmachinehalt}} = \texttt{Seq}_{\Language}(1)$ is the desired instance for the reduction. We first want to argue that $\texttt{Seq}_{\Language}$ never halts on line $5$ of~\Cref{pseu:seqintro} given \textbf{any} input $n \in \bN$. 

Suppose, for a contradiction, that there exists some $C \in \bN$ such that the algorithm $\texttt{Seq}_{\Language}$ halts in line $3$ of~\Cref{pseu:seqintro} with $C$ as the input; We can also assume that $C$ is the smallest integer such that this holds without loss of generality. The goal is to argue that whenever $\texttt{Seq}_{\Language}(C)$ does halt in line 3, then $\texttt{Seq}_{\Language}(1)$ simultaneously belong in $\Language$ and does not belong in $\Language$, thus creating a contradiction. 

By first analysing~\Cref{pseu:seqintro}, we see that $\texttt{Seq}_{\Language}$ halting in step 2 implies that it cannot halt in step 5. This also means that \Turingmachinehalt cannot halt in $C$ steps, and hence $\texttt{Seq}_{\Language}$ cannot halt in step 2 of~\Cref{pseu:seqintro} for all input $n < C$. 

Now, $\texttt{Seq}_{\Language}$ halting in line $5$ on input $C$ implies that $\texttt{Seq}_{\Language}(C) = x_{\text{no}} \not\in \Language$.   $\texttt{Seq}_{\Language}$ halting in line $5$ also implies that $\texttt{Algo}_{\Language}(\texttt{Seq}_{\Language}(1))$ halts in a finite number of steps, which means $\texttt{Seq}_{\Language}(1) \in \Language$ by the definition of $\texttt{Algo}_{\Language}$. By the minimal assumption of $C$ and the argument above, $\texttt{Seq}_{\Language}$ does not terminate on line $2$ or $5$ for all inputs $1 \leq n < T$, and hence $\texttt{Seq}_{\Language}(C-1) =  \texttt{Seq}_{\Language}^{\text{comp}}(C)$, which is not in $\Language$ by the definition of $\texttt{Compress}^{\Language}$. By a simple inductive argument, one can deduce that $\texttt{Seq}_{\Language}(1) \not\in \Language$. This creates the contradiction needed to argue that~\Cref{pseu:seqintro} never terminates via the exit clause on line 5. 

Now, suppose \Turingmachinehalt halts in $T$ steps, by construction, we have  $\texttt{Seq}_{\Language}(C) \in \Language$, and hence, by a similar inductive argument as above, we see that $\texttt{Seq}_{\Language}(1) \in \Language$. If \Turingmachinehalt does not halt, then by the above argument, we see that $\texttt{Algo}_{\Language}$ also cannot halt given the input $\texttt{Seq}_{\Language}(1)$, which, by the assumption we made on $\texttt{Algo}_{\Language}$, implies that $\texttt{Seq}_{\Language}(1) \not\in \Language$. This shows that $\Language$ is $\RE$-complete. We remark that this reduction is poly-time, as $x_{\text{\Turingmachinehalt}} = \texttt{Seq}_{\Language}(1)$ which can be computed in $O(\poly(|\text{\Turingmachinehalt}|))$ time.

If a language $\Language$ is shown to be compressible and in $\coRE$, $\{0,1\}^* \setminus \Language$ is compressible and in $\RE$, and hence $\Language$ is $\coRE$. Intuitively, the above argument relies on the fact that we can embed the $\RE$ algorithm into a uniform Turing machine which generates a sequence of decision problems for the given language and use compression to ``infinitely" reduce the runtime of the algorithm. The above proof approach for showing $\Language \leq \RE$ is a generalization of the conjectured approach for showing $\coRE \subseteq \MIPco$ in~\cite[Pseudocode 4]{mousaviNonlocalGamesCompression2022}. In comparison to the first draft of~\cite{nezhadiRecursiveCompressionMethod2025}, there is no dependency on some ``resource" in the $\texttt{Compress}^{\Language}$ map. Interestingly, as pointed out from~\cite{nezhadiRecursiveCompressionMethod2025}, the Halting Problem is also compressible, which means that any $\RE$-complete language is also compressible, and we present their formulation in~\Cref{exam:haltcompressible} of this paper.

\paragraph{The compressible condition for decision problems.} In order to apply the compressible condition to non-local games, we have to first reformulate the compressible condition to hold for decision problems. For a decision problem $\Decisionproblem$ given by $\Decisionproblem_{\text{yes}} \subseteq \{0,1\}^*$ and $\Decisionproblem_{\text{no}} \subseteq \{0,1\}^*$, we define $\Decisionproblem$ to be compressible if it admits a similar $\texttt{Compress}^{\Decisionproblem}$ algorithm which compresses a \textit{uniform problem instance} for $\Decisionproblem$, or a Turing machine $\texttt{Seq}_{\Decisionproblem}: \bN \rightarrow \Decisionproblem_{\text{yes}} \cup \Decisionproblem_{\text{no}}$. To draw the parallel to the definition given in~\Cref{def:compressiblelanguage}, one can interpret $\texttt{Seq}_{\Language} \rightarrow \bN \rightarrow \{0,1\}$ given in the previous section as a function which  maps $n \in \bN$ to an element to either in $\Language$ or $\{0,1\}^* \setminus \Language$. In this case, the $\texttt{Compress}^{\Decisionproblem}$ generates a description of a ``compressed" uniform problem instance $\texttt{Seq}_{\Decisionproblem}^{\text{comp}}$ which has the same completeness/soundness property given in~\Cref{def:compressiblelanguage} (i.e. for $i \in \{\text{yes}, \text{no}\}$ If $\texttt{Seq}_{\Decisionproblem}(n) \in \Decisionproblem_{i}$, then $\texttt{Seq}_{\Decisionproblem}^{\Compressgame}(n) \in \Decisionproblem_{i}$ for all $n \in \bN$), assuming the given input is a uniform problem instance for $\Decisionproblem$ which runs in $O(\poly(n))$.

Unfortunately, this generalization in practice is very hard to show for any non-trivial language. One of the reasons is that the compressible condition requires one to show the existence of a universal compressible map which works for all $O(\poly(n))$ uniform problem instances. In an effort to make this condition more applicable for general decision problems, we define a weaker notion of the compressible condition known as \texttt{weakly compressible condition}. Intuitively, instead of requiring a single $\texttt{Compress}^{\Decisionproblem}$, a decision problem is weakly compressible if for every $\alpha \in \bN$, there exists a $\texttt{Compress}^{\Decisionproblem}_{\alpha}$ which only ``compresses" uniform instances with runtime $O(n^{\alpha})$. Clearly, if $\Decisionproblem$ is compressible, the compression algorithm guaranteed by the compressible condition can be used to satisfy the condition for weakly compressible. In~\Cref{thm:compressRE} and~\Cref{thm:compresscosRE}, we show that if $\Decisionproblem \leq \RE$ or $\Decisionproblem \leq \coRE$, then the following statements about $\Decisionproblem$ are equivalent:
\begin{itemize}
	\item $\Decisionproblem$ is $\RE$/$\coRE$-complete.
	\item $\Decisionproblem$ is compressible decision problem. 
	\item $\Decisionproblem$ is weakly compressible decision problem. 
\end{itemize}
Thus showing that this weaker notion of compressibility can also be used for showing $\RE$/$\coRE$-completeness. We describe the compressible condition and weakly compressible for decision problems in more detail in~\Cref{sec:Compression}.

\paragraph{Applying the compressible condition to non-local games.} Recall from the previous section that the tensor product value problem is complete with respect to the complexity class $\MIP^{*}$, and the commuting operator value problem is complete with respect to the complexity class $\MIPco$. A uniform problem instance for the tensor product/commuting operator value problem is defined as a Turing machine $\verifiersequence: \bN \rightarrow \cG$, where to abuse notation, $\cG$ in this case is the set of all possible descriptions for a non-local game. One could intuitively think of the uniform sequence as the ``inputted Turing machine"

We show that a specific subclass of game sequences, which we refer to as ``conditionally linear verifier" admits a gap compression theorem described earlier in the introduction. We give an informal description of the conditionally linear verifier in the next section and the more detailed version in~\Cref{sec:CLverifier}. We give an informal version of the gap compression theorem below. 
\begin{theorem}[Gap compression, informal] \label{thm:gapcompressioninformal}
	For every $\alpha \in \bN$, there exists a polynomial time algorithm $\mathtt{Gapcompress}_{\alpha}$, that takes the input $\verifiersequence: \bN \rightarrow \cG$ a conditionally linear verifier such that $\verifiersequence(n)$ runs in $O(n^{\alpha})$ time, each game in the sequence can be sampled and decided in $O(n^{\alpha})$ time. The algorithm runs in $\poly(|\langle\mathtt{Gapcompress} \rangle|, \alpha)$ time and outputs a conditionally linear verifier $\verifiersequence^{\Compressgame}: n \rightarrow \cG$ such that, for $t \in \{*, co\}$:
	\begin{enumerate}
		\item (Runtime) $\verifiersequence^{\Compressgame}(n)$ runs in $O(\polylog(n))$ time, and each game $\verifiersequence^{\Compressgame}(n)$ can be sampled in $O(\polylog(n))$ time, and the decider function $D_n$ runs in $O(\polylog(n))$ time.
		\item (Completeness) If $\omega^{t}(\verifiersequence(n)) = 1$, then $\omega^{t}(\verifiersequence^{\Compressgame}(n)) = 1$ 
		\item (Soundness) If $\omega^{t}(\verifiersequence(n)) < \frac{1}{2}$, then $\omega^{t}(\verifiersequence^{\Compressgame}(n)) < \frac{1}{2}$. 
	\end{enumerate}
\end{theorem}
Roughly speaking, by considering the $\mathtt{Gapcompress}_{\alpha}$ guaranteed by the theorem, this shows that $\MIP^*/\MIPco$, when restricted from games generated by a conditionally linear verifier, are weakly compressible. In practice, there are many more caveats to the above theorem which is needed to argue for weakly compressible which was not listed above. Instead, we refer to~\Cref{thm:gappedcompression} for more details. In conjunction with the NPA-hierarchy~\cite{navascuesConvergentHierarchySemidefinite2008b}, which is a $\coRE$ algorithm for the commuting operator value problem (i.e. it halts if the game is in the set $\Language_{\text{no}}$ and otherwise runs forever), we can conclude the $\coRE$-completeness of $\MIPco$, thus showing the main theorem in this paper. We refer to~\Cref{sec:MIPcocoREproof} for more details.


On the other hand, by combining the Gap compression with the well-known \texttt{Searchfrombelow} algorithm (\Cref{pseu:searchfrombelow}), an $\RE$ algorithm for the tensor product value problem, we give an alternative proof for the $\MIP^*= \RE$ theorem using the weakly compressible condition we described above, and we refer to~\Cref{sec:MIPREproof} for more details. 

We remark that the $\texttt{Gapcompress}$ algorithm defined in the formal version of the gap compression theorem in~\Cref{thm:gappedcompression} is a more streamlined version of~\cite[Theorem 12.1]{jiMIPRE2022a}, and the conditionally linear verifier is a more concise version of the normal form verifier defined in~\cite[Definition 5.31]{jiMIPRE2022a}. The primary challenge in this paper is to show that the same compression algorithm also has the desired completeness/soundness condition under the commuting operator model, which we summarize in~\Cref{sec:introMIPstartoMIPco}.


In this paper, we also model finding an explicit separation between the models of entanglement as the $\frac{1}{2}$-Bell test separation decision problem, which we model as a decision problem over the following two sets of non-local games:
\begin{equation*}
	\Language_{\text{yes}} = \left\{ \cG: \omega^{*}(\cG) =  \omega^{\mathrm{co}}(\cG) \right\} \qquad \text{ and } \qquad	\Language_{\text{no}} = \left\{ \cG: | \omega^{*}(\cG) - \omega^{\mathrm{co}}(\cG) | > \frac{1}{2} \right\}. 
\end{equation*} 
We remark that the constant above being $\frac{1}{2}$ can be changed to any arbitrary fixed constant $c \in (0,1)$ by increasing the number of parallel repetition in the proof of~\Cref{thm:gapcompressioninformal}. The problem is known to be in $\RE$ by combining \texttt{Searchfrombelow} together with the NPA hierarchy, since the $\texttt{Gapcompress}_{\alpha}$ algorithm given in~\Cref{thm:gapcompressioninformal} preserves the tensor product value and the commuting operator value. It is not hard to see that $\texttt{Gapcompress}_{\alpha}$ also preserves the yes/no case for any given conditionally linear verifier and thus can be used to argue that the above problem is weakly compressible. Although it is impossible to find a game that realizes the separation computationally, such a separation can still be found using mathematical techniques, and we refer to~\Cref{sec:compproofoutline} for further discussions on this topic. 


\subsubsection{Towards proving the gap compression theorem}
 In this subsection, we provide, from an algorithmic level, a rough outline of the subroutine used in $\texttt{Gapcompress}$ given in~\Cref{thm:gapcompressioninformal}. We first give a high-level description of a conditionally linear verifier and some intuition as to why the runtime can be compressed. In standard non-local game literature, a non-local game is defined in terms of two finite sets, $\cX$ for the question set, $\cA$ for the answer set, $\mu \sim \cX \times \cX$ as the question distribution for the two provers and $D: \cX^2 \times \cA^2 \rightarrow \{0,1\}$ as the validation function. A conditionally linear verifier $\verifiersequence$ is a game sequence that takes as input $n \in \bN$ and outputs a non-local game $\cG_n$ with the following properties:

\begin{itemize}
	\item Each game in the sequence has a question distribution that follows a specific class of distributions known as ``conditionally linear distributions" which we define formally in \Cref{def:CLdistribution}. Intuitively, a verifier can sample a question pair by first sampling a seed $s$ and then applying two special deterministic functions $\decidertypedfunction^A, \decidertypedfunction^B$ in order to compute the question pair $(\decidertypedfunction^A(s),  \decidertypedfunction^B(s))$. \cite{jiMIPRE2022a} realizes that a pair of honest provers can sample a question pair (in a way such that the question sampled for one prover is unknown to the other prover) from the conditionally linear distribution by taking the outcome from a specific set of Pauli $Z$ measurements on shared EPR pairs (where a single EPR pair corresponds to the state $\frac{\ket{00} + \ket{11}}{\sqrt{2}}$). Hence, the verifier can effectively shrink the question size for the non-local game by simply asking honest provers to perform the correct measurement on some pre-prepared EPR pairs between the provers. As seen later in this section, this fact is used with self-testing techniques to force the provers to make the correct $Z$ measurements.
	\item $\verifiersequence$ is defined by a pair of Turing machines $(\samplerTM, \deciderTM)$. The sampler, $\samplerTM$, takes as input a natural number $n$ and returns a description for the two functions $\decidertypedfunction^A_n, \decidertypedfunction^B_n$ which describes the conditionally linear distribution that samples a question pair for $\cG_n$. The decider,  $\deciderTM$, also takes a natural number and returns a description of a Turing machine, which computes $D_n$, the decision function for $\cG_n$. This definition is closer to a definition given in the interactive proof system literature, and this allows us to use techniques from the PCP theorem to shrink the answer size for the game. 
\end{itemize}

The runtime, or the complexity for the conditionally linear verifier is defined as the maximum runtime for the Turing machine $\samplerTM$ and $\deciderTM$. We remark that although the question set and the answer set are not explicitly specified by $\verifiersequence$ in a conditionally linear verifier, both $\samplerTM$ and $\deciderTM$ being time bounded effectively define a finite set of questions/answers for each of the games $\cG_n$\footnote{if $\samplerTM$ runs in time $T$, then $\samplerTM$ can samples a string with length at most $T$, which means that the question set can effectively be taken as $\{0,1\}^T$}. The $\texttt{Gapcompress}$ algorithm constructed for a conditionally linear verifier consists of three main subroutines:  \textbf{question reduction} (\Cref{sec:introspection}), \textbf{answer reduction} (\Cref{sec:PCPanswerreduction}), and \textbf{parallel repetition} (\Cref{sec:parallelrep}). Since many of the techniques used are similar to the ones used in~\cite{jiMIPRE2022a}, we only give a brief summary of each subroutine and highlight our improvements over~\cite{jiMIPRE2022a} below. For a more detailed summary, we instead refer the readers to~\cite[Chapter 2]{jiMIPRE2022a}.

\paragraph{Question Reduction. } The goal of the question reduction protocol is to force the provers to make an ``ideal" measurement in order for them to sample from a question pair following the conditionally linear verifier. The question reduction is a combination of two tests: the Pauli Basis test and the Introspection protocol. The Pauli Basis test uses self-testing techniques in order to force the dishonest provers to prepare enough EPR pairs required to sample the input distribution, as well as perform an $X$ or $Z$ measurement on all the prepared EPR pairs. The Introspection protocol utilizes the EPR pairs guaranteed by the Pauli Basis test and forces the provers to make the correct Pauli $Z$ measurement so that the provers can sample a pair of questions from the given conditionally linear distribution. The Pauli basis test has a question sampling complexity of $\polylog(n)$, and the Introspection protocol has a sampling complexity independent of $n$ (where both tests have a decision complexity of $\poly(n)$), thereby reducing the complexity of $\samplerTM$ from $\poly(n)$ to $\polylog(n)$. For more details on the question reduction protocol, we refer the readers to~\Cref{sec:introspection}. 

In this paper, we use the recent simplification due to~\cite{delasalleSpectralGapStability2022} for the Pauli Basis test, which circumvents the low-individual degree test subroutine used in the original Pauli Basis test~\cite[Section 7]{jiMIPRE2022a}. As a result, the increase in soundness error in the question reduction theorem proven in our paper is independent of the size of the game, an improvement over the polynomial dependency in~\cite{jiMIPRE2022a}.  


\paragraph{Answer Reduction. } The goal of the answer reduction protocol is to reduce the complexity of the decider $\deciderTM$ from $\poly(n)$ time to $\polylog(n)$ time while retaining a $\polylog(n)$ runtime for the sampler. Similarly to~\cite{jiMIPRE2022a, natarajanNEEXPMIP2019a}, a classical probabilistically checkable proof (PCP) is used on the verification Turing machine $\deciderTM_n$. In particular, we use the same tailor-made PCP procedure used within~\cite{jiMIPRE2022a}'s answer reduction protocol as a black box in this paper. The PCP procedure reduces checking computation steps of $\deciderTM_n$ returning accept given the corresponding question/answer pair into checking whether the two provers share the same collection of low individual degree polynomials with specific properties, which can then be checked using the ``quantum low-individual degree test" introduced in~\cite{jiQuantumSoundnessTesting2022}.

There is an immediate problem with this approach. Recall that for a classical PCP which verifies an $\NP$ instance, the provers is intuitively trying to prove the following statement to the verifier:  ``given $x \in \Language$ and a polynomial size circuit \texttt{C}, there exists a proof string $s$ such that $C(x,s)=1$. Where as in the non-local game setting, the provers is trying to prove: ``given a validation function $D_n$ for a non-local game, and a question pair $(x,y)$, there exists an answer pair $(a,b)$, which is generated by two entangled provers, such that $D_n(x,y,a,b) = 1$. In order to compute the second statement in a PCP instance, both provers need to somehow play the game using a predetermined entangled strategy and output their answer without communicating with each other. Then, the provers need to pass their answer so that they can encode the computation step of $D_n(x,y,a,b)$ as a PCP instance. In order to get around this issue, a transformation known as Oracularization is applied before computing the PCP instance (see~\Cref{sec:Oracularization}). To ensure completeness is preserved through the oracularization transformation, we need an additional property in the completeness statement in the gap compression theorem: The perfect strategies in the original game must use a special kind of strategy known as an \textit{oracularizable strategy}, which is defined in~\Cref{def:Oracularizablestrategies}. We remark that this transformation is also used in~\cite{jiMIPRE2022a}, and the oracularizable strategy in this paper is the same as the ``commuting and consistent strategy" used in~\cite{jiMIPRE2022a}.

\paragraph{Parallel repetition. } By applying the above two subroutines to a conditionally linear verifier $\verifiersequence: \bN \rightarrow \cG$ with complexity $O(\poly(n))$, the resulting conditionally linear verifier $\verifiersequence': \bN \rightarrow \cG$  which runs in $O(\polylog(n))$ time, such that for $t \in \{*, co \}$
\begin{enumerate}
	\item (Completeness) If $\omega^{t}(\verifiersequence(n)) = 1$, then $\omega^{t}(\verifiersequence'(n)) = 1$, 
	\item (Soundness) If $\omega^{t}(\verifiersequence(n)) \leq \frac{1}{2}$, then $\omega^{t}(\verifiersequence'(n)) \leq 1 - \polylog(n)$. 
\end{enumerate}

Thus, in order to show~\Cref{thm:gapcompressioninformal}, one would need to apply a logarithmic-fold parallel repetition transformation to the game in order to amplify the ``soundness" condition for $\verifiersequence'$ to $< \frac{1}{2}$ while retaining the ``completeness" condition. Where recall, $r$-fold parallel repetition is a transformation for the game $\cG$ in which $r$ question pairs are sampled independently and sent to the provers, and the provers must treat each of the $r$ question pairs as independent questions and reply with $r$ corresponding answer pairs. The provers only win the $r$-fold parallel repetition game if they win on all $r$ independent instances of the game. We remark that applying a logarithmic-fold parallel repetitions to $\verifiersequence'$ only increases the runtime by a logarithm factor.

Unfortunately, a strong parallel repetition theorem, i.e. a parallel repetition theorem that shows an exponential decay in the optimal success rate for entangled provers is an open problem. However,~\cite{bavarianAnchoredParallelRepetition2021} shows that by applying a simple transformation, the anchored transformation, the resulting game would have a strong parallel repetition theorem, and this version of parallel repetition has been used in~\cite{jiMIPRE2022a}. We use a slight modification for the anchoring transformation\footnote{The modification is designed to preserve synchronicity within the game.} in our paper, and we give a proof for the anchored parallel repetition theorem for the commuting operator model in~\Cref{sec:parallelrepappendix}. 

\subsubsection{From $\MIP^*$ to $\MIPco$. } \label{sec:introMIPstartoMIPco}

Finally, we discuss some of the challenges in extending the gap compression theorem to the commuting operator models. Since the commuting operator model of entanglement cannot be approximated by finite-dimensional strategies, unlike the tensor product model, many techniques used in~\cite{jiMIPRE2022a} for proving the gap compression theorem are not known to generalize to the commuting operator model. 

\cite{linTracialEmbeddableStrategies2024} recently introduced a subclass of (two prover) commuting operator strategies known as \textit{tracially embeddable strategies}. Tracially embeddable strategies, while being infinite dimensional, have many similar structures to a finite-dimensional tensor product strategy. Hence, many techniques from~\cite{jiMIPRE2022a}, which hold for the finite-dimensional tensor product model, can easily be translated to the setting where the provers are restricted to tracially embeddable strategies.

Furthermore,~\cite{linTracialEmbeddableStrategies2024} shows that the behaviour of provers with access to the commuting operator model of entanglement can be \textbf{well approximated} by provers restricted to using tracially embeddable strategies. To be a bit more precise, the set of \textit{correlations}, or the set of probability distributions for outputting a certain answer pair given a question pair when the provers are given access to tracially embeddable strategies is dense within the set of correlations for provers with commuting operator strategies. Hence, the complexity of $\MIPco$ is precisely the same as the complexity of an interactive proof system where the provers are restricted to using tracially embeddable strategies. By working within this class of strategy, the following techniques used in the proof of the gap compression theorem in~\cite{jiMIPRE2022a} become available in the analysis of $\MIPco$: 
\begin{enumerate}
	\item Tracially embeddable strategies provide a natural generalization to ``density matrices" and the ``observable switching trick" to finite-dimensional strategies.~\cite{linTracialEmbeddableStrategies2024} uses this to replicate the rigidity statement for the Pauli basis test for provers restricted to using tracially embedded strategies similarly to~\cite[Theorem 7.14]{jiMIPRE2022a} (which shows the rigidity for finite dimensional tensor product strategies). This is expressed in~\Cref{thm:soundnessPaulibasis} in our paper, and it is a key part of showing the ``soundness" properties of the question reduction protocol in the commuting operator model. 
	\item Tracially embeddable strategies also give a natural notion of relative entropy for quantum states in the infinite-dimensional setting~\cite{arakiRelativeEntropyStates1977}. Finite dimensional von Neumann entropy is a crucial tool for proving the anchored parallel repetition theorem~\cite{bavarianAnchoredParallelRepetition2021} (which itself is based on the informational theoretical parallel repetition theorem for classical $\MIP$ by~\cite{razParallelRepetitionTheorem1995a}). By assuming the underlying strategy is tracially embeddable, we gave the analogue for many components used in the proof of~\cite{bavarianAnchoredParallelRepetition2021}, and we refer to~\Cref{sec:parallelrepappendix} for more details. 
	\item Lastly, the answer reduction protocol relies on the ``soundness" of the quantum low-individual degree test, which is only shown to hold for a special class of strategies known as synchronous strategies in~\cite{jiQuantumSoundnessTesting2022}. In~\cite{vidickAlmostSynchronousQuantum2022}, a ``rounding" lemma, or a lemma translating results proven using synchronous strategies to regular strategies when restricted to finite dimensional strategies, was shown. By working with tracially embeddable strategies, the same lemma can be shown for tracially embeddable strategies in~\cite{linTracialEmbeddableStrategies2024}. We remark that the ``rounding" lemma can be proven without using the tracially embeddable strategies framework by the works of~\cite{delasalleAlmostSynchronousCorrelations2023}. For details about synchronous strategies and the rounding lemma, we refer the reader to~\Cref{sec:introcorrelations}.
\end{enumerate}

We express and prove all our results using tracially embeddable strategies. As seen in~\Cref{def:Tracialemd}, tracially embeddable strategies are defined using languages of tracial von Neumann algebra in standard form, which might be intimidating for readers with no prior background on von Neumann algebras. We provide a brief introduction to basic tracial von Neumann algebra in~\Cref{sec:introVNA}, and we give a translation chart which converts the notation for tracially embeddable strategies to what the intuitive finite dimensional counterpart is in~\Cref{fig:tracialemdtofd} for clarity.  For more intuition about tracially embeddable strategies in the finite-dimensional setting, we refer to~\cite[Example 3.3]{linTracialEmbeddableStrategies2024}.

\subsection{Consequences}
In this subsection, we discuss some of the additional consequences for our results.

\paragraph{Uncomputability of the commuting operator value. }Recall from the previous section, that the NPA hierarchy is an algorithm which generates a series of upper bounds that estimates the commuting operator value of a game. Although as shown in~\cite[Theorem 12.10]{jiMIPRE2022a}, there exists a game $\cG$ such that $\omega^{co}(\cG) = 1$ and  $\omega^{*}(\cG) \leq \frac{1}{2}$.  The best algorithm used in practice for estimating the tensor product value of the game is still the NPA hierarchy due to the inefficiency of the \texttt{Searchfrombelow} algorithm. Estimating the commuting operator value of a game is also important in the recently introduced compiled non-local game setting~\cite{kalaiQuantumAdvantageAny2022}, where the optimal bound of the game is known to be bounded by the commuting operator value of the game due to a recent result by~\cite{kulpeBoundQuantumValue2025}. 

As an obvious consequence from the main result of this paper is that \textbf{there is no algorithm which can estimate the commuting operator value of the game} up to any constant $c \in (0,1)$, or else one can use this algorithm to construct a halting algorithm for the $co$ non-local game value problem. Our result also implies that one cannot compute the convergence rate of the NPA hierarchy in general.

\paragraph{Connection to noncommutative polynomials. } Let $\cF_n^m$ be the free group consisting of $m$ elements of order $n$, and let $\bQ(\cF_n^m)$ be the finitely-generated $*$-algebra over $\bQ$. \cite[Theorem 1.1]{mehtaPositivityUndecidableTensor2023} showed the $\bQ(\cF_n^m)$ tensor product positivity problem is $\coRE$-hard in the case where $n,m \geq 2$, $(n,m) \neq (2,2)$, where the $\bQ(\cF_n^m)$ tensor product positivity problem is defined as follows: Given an element $g \in \bQ(\cF_n^m) \otimes \bQ(\cF_n^m)$, deciding whether the element $g$ is positive in $(\bQ(\cF_n^m) \otimes \bC) \otimes (\bQ(\cF_n^m)\otimes \bC)$. Since one can also view elements of $(\bQ(\cF_n^m)\otimes \bC)$ as a $m$ variate noncommutative polynomials over elements of order $n$, the above problem can also be formulated as determining the positivity for two noncommutative polynomials tensor-producted together. The $\bQ(\cF_n^m)$ tensor product positivity problem is conjectured to be $\Pi^2_0$-complete. 

By considering the game polynomial introduced in~\cite{wattsNoncommutativeNullstellensatzePerfect2023} and its connection to the commuting operator strategies, the complexity class $\MIPco$ is related to the ``gap" version of the $\bQ(\cF_n^m)$ tensor product positivity problem, where the gap $\bQ(\cF_n^m)$ tensor product positivity problem is defined as follows decision problem: given $g \in \bQ(\cF_n^m) \otimes \bQ(\cF_n^m)$, decide whether $g - \cI$ is positive or $g$ is not positive, where $\cI$ is the identity element in $\bQ(\cF_n^m) \otimes \bQ(\cF_n^m)$. \cite{mehtaPositivityUndecidableTensor2023} shows that our main theorem, $\MIPco=\coRE$ implies that the gap $\bQ(\cF_n^m)$ tensor product positivity problem is also $\RE$-complete. This, in turn, also implies that the $\bQ(\cF_n^m)$ tensor product positivity problem is at least $\RE$-hard, giving stronger evidence that this problem is $\Pi^2_0$-complete. 

\paragraph{Uncomputability results in operator algebra. } Famously, as a corollary of the $\MIP^*=\RE$ theorem, both Connes embedding's problem and Kirchberg's QWEP conjecture are shown to be false due to its relationship with Tsirelson's problem~\cite{fritzTsirelsonProblemKirchberg2012,ozawaQWEPConjecture2004}. In conjunction with recent work in operator algebra, our main result also gives a negative result to a stronger variant of these two famous conjectures. 

Roughly speaking, the disproof of the Connes embedding problem states that every tracial von Neumann algebra on a separable Hilbert space cannot be approximated by a limit of finite-dimensional matrices. A recent result by~\cite{arulseelanUniversalTheoryLocally2025} shows that our main result also implies that furthermore all tracial von Neumann algebras cannot be approximated by \emph{any} computable object, thus showing the class of tracial von Neumann algebra is ``uncomputable" in nature. 

The disproof of Kirchberg's QWEP conjecture states that the maximum tensor product (or the ``algebraic" tensor product) norm for $C^*(\cF_n  \times \cF_n)$ (where $\cF_n$ denotes the free group with $n$ generators) is different than the minimum tensor product (or the ``analytic" tensor product) norm. By using our main result,~\cite{goldbringDefinabilityCtensorNorms2025} shows that in general, there is no algorithm which can approximate the maximum tensor product norm of $C^*(\cF_n  \times \cF_n)$ for all $n \in \{2, 3, \cdots \}$. By taking the case where $n=2$, this also gives a negative answer to~\cite[Problem 4.2]{fritzCanYouCompute2014a}\footnote{Thus showing that ``no, you can not compute the operator norm"!}. This result also could potentially give additional insight into the Kirchberg's embedding problem\footnote{Not to be confused with Kirchberg's QWEP conjecture.}, another major open problem in $C$*-algebra and we refer to~\cite{arulseelanUniversalTheoryLocally2025, goldbringKirchbergsEmbeddingProblem2015} for more details. 

\paragraph{Estimation of quantum values for a game.} A variant of $\MIPco$ called the zero-gap $\MIPco$ ($\MIPco_0$) is introduced in~\cite{mousaviComplexityZeroGap2020}. $\MIPco_0$ is the complexity class which is complete with respect to the gapless commuting operator value problem, which is defined by the following two sets of non-local games
\begin{equation*}
	\Language_{\text{yes}}^{co} = \left\{ \cG: \omega^{t}(\cG) = 1 \right\} \qquad \text{ and } \qquad	\Language_{\text{no}}^{co} = \left\{ \cG: \omega^{t}(\cG) < 1 \right\}. 
\end{equation*} 
This class is also shown to be $\coRE$ complete due to a result by~\cite{slofstraTsirelsonProblemEmbedding2019}, which implies that $\MIPco= \MIPco_0$. Interestingly, the equivalency between the gap and zero-gap versions of $\MIPco$ does not extend to $\MIP^*$, as the complexity of the zero-gap $\MIP^*$ ($\MIP^*_0$) is shown to be $\Pi_2$-complete in~\cite{mousaviNonlocalGamesCompression2022}, where $\Pi_2$ is defined as $\coRE$ with access to an $\RE$ oracle.


\paragraph{Zero knowledge proof systems. } In~\cite{mastelTwoProverPerfect2024a}, it was shown that every $\MIP^*$ protocol has a perfect zero-knowledge proof system in the $\MIP^*$ model. In the same work, it was conjectured that the same would hold for the complexity class $\MIPco$ provided $\MIPco=\coRE$, and a parallel repetition theorem was proved for the commuting operator model. In conjunction with our result, this shows that one could similarly convert any $\MIPco$ protocol into a zero-knowledge $\MIPco$ protocol.


\subsection{Open problems}

\paragraph{The generalized compression framework. } In~\Cref{sec:introcompression}, we introduced the generalized compression framework for showing $\RE$-completeness/$\coRE$-completeness of a given decision problem, and we showed that the gap compression theorem for non-local games gives a way to use the generalized compression framework for showing that the tensor product/commuting operator value problem is $\RE$/$\coRE$-complete. Thus, an interesting open problem is whether this framework can be applied to other decision problems which are in $\RE$/$\coRE$, but conjectured to be $\RE$/$\coRE$-complete?

In the other direction, it was shown that assuming $\MIP^*=\RE$ or $\MIPco=\coRE$, there exist a ``compression theorem" for the non-local game value problem for the corresponding model of entanglement~\cite{mousaviNonlocalGamesCompression2022, mousaviComplexityZeroGap2020}. Thus, a natural follow-up question is whether \textbf{all} decision problems which are $\RE$/$\coRE$-hard admit a natural compression theorem for a subclass of uniform sequences of the said decision problem. 

\paragraph{Complexity for non-local games. } As mentioned earlier in the introduction, $\MIPco=\MIPco_0$, where $\MIPco_0$ is the zero-gap variant of $\MIPco$. This implies that there must exist a direct reduction from the commuting operator value problem to the gapless commuting operator value problem. Thus, an interesting open problem is whether there exists a natural reduction between these two problems without going through the Non-Halting problem?

Another open problem is whether there exists a more reasonable experimental realization between the tensor product and the commuting model of entanglement. Although we show that no algorithm can find such a separation, this does not eliminate other mathematical techniques which can be used to find a game that realizes said separation. This game would also provide a more reasonable way to construct a counterexample to Connes' embedding problem. 

Does there exist a more general uniform sequence of games that admit a gap compression theorem, or are conditionally linear verifiers the most general class of uniform games that can be compressed? As discussed in the technical overview, conditionally linear verifiers are tailor-made to take advantage of self-testing from non-local games literature and PCP constructions from the interactive proof systems literature, so we suspect any general uniform game sequence that admits a gap compression theorem must also be defined in a way that enables both techniques.

\paragraph{Estimating the commuting operator value for other classes of games. } In this paper, we show that the commuting operator value problem with parameters $(1, 1- \eps)$ for games with a conditional linear distribution as their input distribution, as well as synchronous games, is $\coRE$-complete for all constant $\eps > 0$. One natural question is whether this result holds for other classes of games. \cite{mancinskaGappreservingReductionsREcompleteness2025} shows that there exists a constant $c_{\text{Inde}} > 0$ such that the tensor product value problem for independent set games with parameters $(1, 1 - c_{\text{Inde}})$ is $\RE$-complete, and assuming the main result of our work,  the commuting operator value problem for independent set games with parameters $(1, 1 - c_{\text{Inde}})$ would be $\coRE$-complete. Similar results have been derived for the constraint satisfaction problem (CSP) games and 3-colouring games by~\cite{culfREcompletenessEntangledConstraint2025}. An interesting open problem is whether these type of results holds for other classes of games for both the tensor product value and the commuting operator value problems.

At the other extreme, are there parameters in which the commuting operator value problem is ``easy" (i.e. computable)? In~\cite{culfApproximationAlgorithmsNoncommutative2024}, it was shown that there exists a parameter $d_{\text{col}}$ such that for all $d' < d_{\text{col}}$, the tensor product value problem with parameters $(1, 1 - d')$ for 3-colouring games is decidable in polynomial time! Does the same phenomenon hold for commuting operator values? Does there exist a class of games such that the tensor product value problem with parameters $(1, 1-c)$ is computable, but uncomputable for the commuting operator value problem with the same parameter (or vice versa)? Does there exist a variant of the ``unique games conjecture" analogous to classical $\MIP$~\cite{khotPowerUnique2prover2002} or $\MIP^*$~\cite{kempeUniqueGamesEntangled2009, mousaviQuantumUniqueGames2024} for $\MIPco$.

\paragraph{Additional insight into the Connes embedding problem. } Our proof techniques for $\MIP^*=\RE$ and $\MIPco =\coRE$ can potentially give an alternative perspective into the counter-example for the Connes embedding problem. If we view non-local games as a functional which maps correlations into $(0,1)$, a norm on this set of functionals would correspond to its optimal success rate on the correlation set.~\Cref{thm:gapcompressioninformal} can intuitively be seen as a map that maps functionals acting on a correlation to one that maps on a smaller correlation set (in terms of input/output), while maintaining the norm to some degree. The fact that~\Cref{thm:gapcompressioninformal} preserves both tensor product and commuting operator values means that the difference between the tensor product value and the commutative operator value only lies in how the compression is being used. This intuition might be useful in constructing a counterexample for the Connes embedding problem using operator algebraic techniques.

Due to the characterization by~\cite{fritzTsirelsonProblemKirchberg2012}, the tensor product model corresponds to the ``max" tensor product between two free algebras, whereas the commuting operator model corresponds to the ``min" tensor product between two free algebras. Thus, $\MIPco=\coRE$ allows one to work with the ``max" tensor product when considering operator algebraic results which rely on the Connes embedding problem being false. Can this additional insight be helpful for the operator algebra community? 

\paragraph{Application to other operator algebra problems. } The complexity of approximating the values of non-local games has a natural connection to the study of operator algebra. As mentioned above, the $\MIP^*=\RE$ theorem gives a negative answer to the Connes embedding problem in the study of tracial von Neumann algebras.~\cite{bowenAldousLyonsConjectureSubgroup2024, bowenAldousLyonsConjectureII2024} gives a negative answer to the Aldous-Lyons problem~\cite{aldousProcessesUnimodularRandom2018} in probability theory by showing that $\class{TailoredMIP}^*$, $\MIP^*$ with a more restricted class of strategies, is $\RE$-complete. An important open problem in group theory is whether a non-hyperlinear group exists. It is shown in~\cite{paddockSatisfiabilityProblemsAlgebras2025} that $\class{LinMIP}^*$, or $\MIP^*$ protocols being restricted to Linear Constraint System game~\cite{cleveCharacterizationBinaryConstraint2014,kimSynchronousGameBinary2018}, being computable is equivalent to the existence of a non-hyperlinear group. The Linear Constraint System game does not fall into the conditionally linear verifier framework, and it would be interesting to see if a similar compression technique can be used to resolve this problem. 

Another interesting set of strategies is the set of invariant random subgroup (IRS) strategies introduced in~\cite{manzoorThereEquivalenceRelation2025}. Intuitively, this can be seen as the ``commuting operator variant" of the ``$Z$-aligned permutation strategies" introduced in~\cite{bowenAldousLyonsConjectureSubgroup2024}, used to disprove the Aldous-Lyons problem. It is conjectured that $\class{MIP}^{\text{IRS}}$, $\MIP$ with access to IRS strategies, is $\coRE$ complete. Showing this complexity theoretical result also has interesting implications for the Ergodic theory community, and we refer to~\cite{manzoorInvariantRandomSubgroups2025} for more details.  


\paragraph{Acknowledgements.}  I thank William Slofstra, Henry Yuen, and Hamoon Mousavi for their involvement in the parallel repetition project, which led to this work. I thank Henry Yuen and Thomas Vidick for many email exchanges for explaining the $\MIP^*=\RE$ result. I thank Seyed Sajjad Nezhadi for explaining the conjectured approach for showing $\MIPco=\coRE$ in~\cite{mousaviNonlocalGamesCompression2022}. I thank Andrew Marks for showing me~\Cref{exam:haltcompressible}, and the discussion related to the compressible condition,  I thank Michael Chapman for the discussion on Oracularization and Oracularizable strategies, and I thank Aareyan Manzoor for some of the discussion on the potential applications on operator algebras in this work. I thank Srijita Kundu for helpful discussion about the parallel repetition theorem. I thank Yuming Zhao for helpful discussion related to this work. I thank Nick Resch, William Slofstra, and Ronald de Wolf for going through the earlier draft of this paper. I thank the anonymous FOCS reviewer for the helpful comments. Finally, I would like to thank my supervisors, Stacey Jeffery and Jonas Helsen, for going through some of the earlier drafts of this paper and listening to me yapping about MIPco for the past year. This work is supported by ERC STG grant 101040624-ASC-Q.

\end{CJK*}
\section{Classical preliminaries}

\subsection{Finite sets and Turing Machines} \label{sec:preFinitesetTM}
In this paper, we use $\bN$ to denote the set of natural numbers. For a finite set $S$, we use \gls{setsize} to denote the number of elements in $S$, and for $a, b \in S$, we use \gls{kronecker} for the Kronecker delta between the two elements. Given a (potentially infinite) set $T$ and $n \in \bN$, we use\gls{matrixsize} to denote the set of $n$ by $n$ matrices over the set $T$. Given a distribution $\mu$, we use \gls{expectation} to denote the expectation over the distribution $\mu$ and for a set $S$, we use $\Ex_{x \in S}$ to denote the expectation over the set $S$. 

For a bit string $a, s, t \in \{0,1\}^n$, we use \gls{hamming} to denote the Hamming weight of $s$ and \gls{dotproduct} to denote the inner product between $s$ and $t$, or $\sum s_i t_i \text{ mod } 2$. We use $\text{\gls{substring}} \in \{0,1\}^{n}$ to denote the string
\begin{equation*}
 	(s|_a)_i = \begin{cases}
 		s_i &\text{if } a_i =1 \\
 		0 &\text{otherwise}
 	\end{cases}.
\end{equation*}
We use \gls{zerooutmap} to denote the function which zeros out the first $j$ entries of the string $s$ and we use $\pi_{\leq j}$ to denote the function which zeros out everything except for the first $j$ entries of the string. In other words 
\begin{align}
	\pi_{> j}(s_0, \cdots, s_{n-1}) &= (0, \cdots, 0, s_j, \cdots, s_{n-1})  \label{eq:defzeromapstring} \\
	\nonumber \pi_{< j}(s_0, \cdots, s_{n-1}) &= (s_0, \cdots, s_{j-1}, 0, \cdots, 0),
\end{align}
and we take the convention that $\pi_{> j} = \pi_{\geq  j+1}$.

In this paper, we assume all $\log$ are in base $2$. For integers $n \leq m$, we use \gls{numberset} to denote the set $\{0, \cdots, n-1\}$, and for $n \leq m$ we use \gls{numberset2} to denote the set $\{n, n+1, \cdots, m-1\}$. For $n \in \bN$, we use $\text{\gls{binaryfun}} \in \{0,1\}^{\lceil \log(n) \rceil}$ to denote the binary representation for the number $n$, and for $s \in \{0,1\}^n$, we use \gls{binaryinv} to denote the unique integer $i$ such that $\binary{i} = s$.  

We use the Turing machine as the model of computation in this paper. Let $\mathtt{A}(x_1, \cdots x_m)$ denote an $m$-input Turing machine. We assume that $\mathtt{A}$, in this case,  consists of $m$ input tapes, a single work tape, and a single output tape. When specifying the output of an $m$-input Turing machine, we might sometimes define it only for accepting fewer than $m$ inputs; in this case, we assume the Turing machine only reads the first $n$ of the input tapes during the computation step. We use \gls{descriptionTM} to denote the description of the Turing machine $\mathtt{A}$ (represented under $\{0,1\}^*$ string). To abuse notation, we use \gls{descriptionTMinput} to denote the description of the Turing machine $\mathtt{A}$ being hardcoded to run $x \in \{0,1\}^*$ as input. We use \gls{sizeTM} to denote the minimum description length of $\mathtt{A}$, and we note that the description of the Turing machine is a constant that is independent of the input size. We use $\TIME_{\mathtt{A}}(x_1, \cdots , x_m)$ to denote the maximum of $|\mathtt{A}|$ and the runtime of the Turing machine $\mathtt{A}$ running with input $(x_1, \cdots , x_m)$. $\TIME_{\mathtt{A}}(x_1, \cdots , x_m)$ could potentially be $\infty$ if the Turing machine $\mathtt{A}$ does not halt on input $(x_1, \cdots , x_m)$.

Let $\{f_n\}_{n \in \bN}$ be a sequence of functions that map $m$ finite sets $\{S_i^n\}_{i \in [m]}$ to $\{0,1\}^*$. We say the Turing machine $\mathtt{A}$ computes the sequence of functions $f_n$ if $\mathtt{A}$ is an $(m+1)$-input Turing machine in which for all $n \in \bN$ 
\begin{equation*}
	\mathtt{A}(\binary{n}, x_1, \cdots x_m) = f_n( x_1, \cdots x_m), 
\end{equation*}
where each element of $S_i^n$ is encoded using a binary representation with $x_i \in S_i^n$. If $\mathtt{A}$ is the Turing machine which computes the functions $\{f_n\}$, to abuse notation, we use $\mathtt{A}_n$ to denote the function $f_n$. For $n \in \bN$, we use the notation \gls{RuntimeTM} to denote the maximum of $\TIME_{\mathtt{A}}(n, x_1, \cdots, x_m)$ over all input $x_1 \cdots x_m \in \{0,1\}^*$. In this paper, if $f$ takes an integer as input, the integer will always be represented under the binary representation, and hence any integer $n$ is considered a $\log(n)$-bit input under this formulation. 

Let $\Decisionproblem = (\Language_{\text{yes}}^{\Decisionproblem}, \Language_{\text{no}}^{\Decisionproblem})$ be a decision problem for two disjoint non-empty subset $\Language_{\text{yes}}^{\Decisionproblem}, \Language_{\text{no}}^{\Decisionproblem} \subseteq \{0,1\}^*$. We use \gls{complementdecision} to denote the complement of $\Decisionproblem$ (i.e. $\Language_{\text{yes}}^{\class{co}\Decisionproblem} = \Language_{\text{no}}^{\Decisionproblem}$ and $\Language_{\text{no}}^{\class{co}\Decisionproblem} = \Language_{\text{yes}}^{\Decisionproblem}$). We remark that this is different than the notion $\Decisionproblem^{co}$, which we define later in this paper. To abuse notation, we write $x \in \Decisionproblem$ as $x \in  \Language_{\text{yes}}^{\Decisionproblem} \cup \Language_{\text{no}}^{\Decisionproblem}$ and, for a set $S$, we write $f: S\rightarrow \Decisionproblem$ as $f: S \rightarrow \Language_{\text{yes}}^{\Decisionproblem} \cup \Language_{\text{no}}^{\Decisionproblem}$. For two decision problems $\Decisionproblem_1$ and $\Decisionproblem_2$, we write \gls{mappingreduction} if there exists a mapping reduction from $\Decisionproblem_1$ to $\Decisionproblem_2$ and \gls{mappingreductionpoly} if furthermore the reduction is under polynomial time. We define a \textit{uniform problem instance} for $\Decisionproblem$ as a Turing machine $\texttt{Seq}: \bN \rightarrow \Decisionproblem$. Intuitively, this is a way to package a countable number of decision problems from $\Decisionproblem$ in a uniform manner.

In this paper, we consider the following two complexity classes. Recall, the complexity class recursively enumerable languages, \gls{RE}, corresponds to the class of decision problems in which there exists an algorithm that can decide all instances in $\Language_{\text{yes}}$ in finite time (but the same algorithm could potentially run forever for instances in $\Language_{\text{no}}$). We say that a language $\Language \subseteq \{0,1\}^*$ is in $\RE$ (or $\Language \in \RE$) if there exists an algorithm that can correctly decide whether $x \in \Language$ in a finite amount of time. $\RE$ is complete with respect to the halting problem. The halting problem is defined by $\Language_{\text{yes}}^{\RE} = \Language_{\text{halt}}$ and $\Language_{\text{no}}^{\RE} = \Language_{\text{nothalt}}$, with the definition of $\Language_{\text{halt}}, \Language_{\text{nothalt}}$ given below:
\begin{itemize}
	\item $\Language_{\text{halt}}$: The set of Turing machines (represented under the binary description) that halt on the empty input,
	\item $\Language_{\text{nothalt}}$: The set of Turing machines which does not halt on the empty input.
\end{itemize}
Similarly, the complexity class \gls{coRE}, or the complement of $\RE$, corresponds to the class of decision problems in which there exists an algorithm which can decide all instances in $\Language_{\text{no}}$ in finite time. We say that $\Language \in \coRE$ if there exists an algorithm that can correctly decide whether $x \not\in \Language$ in a finite amount of time (but could potentially run forever if $x \in \Language$). $\coRE$ is complete with respect to the non-halting problem. The non-halting problem is defined similarly as the halting problem but with the ``yes" and ``no" instances being swapped, or $\Language_{\text{yes}}^{\coRE} =  \Language_{\text{nothalt}}$, $\Language_{\text{no}}^{\coRE} =  \Language_{\text{halt}}$.

For a more comprehensive introduction on computability and complexity theory, we refer the reader to \cite{sipserIntroductionTheoryComputation2006}.




\subsection{Finite fields} \label{sec:fieldintro}
In this subsection, we recall some basic properties regarding finite fields of the form \gls{Finitefield}. In this paper, we always assume that $p$ is odd for finite fields of the form $\bF_{2^p}$. $\bF_{2^p}$ can always be viewed as a $p$-dimensional vector space over $\bF_2$. Unless otherwise specified, we always assume that $\bF_{2^p}$ has its basis defined over $\bF_2$ (although the basis specified might be different). Given an element $a \in \bF_{2^{p}}$ and a set of basis $\{ \hat{e}_i\}_{i \in [p]}$ for $\bF_{2^p}$, there exists a bijection map from $a$ to $\bF_2^{p}$ by 
\begin{equation} \label{eq:finitefieldbijection}
	\kappa_{\{ \hat{e}_i\}}: a \rightarrow (a_0, \cdots, a_{p-1})
\end{equation}
where $a = \sum_{i=0}^p a_i \hat{e}_i$. Since elements of $\bF_2$ can be represented as elements of $\{0,1\}$, any element of $\bF_{2^p}$ can be represented as a bit string in $\{0,1\}^p$ as long as the set of bases is specified. Recall from~\cite[Definition 2.1.80]{mullenHandbookFiniteFields2013}, every finite field $\bF_{2^{p}}$ admits a trace function over $\bF_2$, or $\text{\gls{Finitetrace}}:\bF_{2^p} \rightarrow \bF_2$ defined as
\begin{equation*}
	\Tr(a) := \sum_{i = 0}^{p} a^{2^i}.
\end{equation*}
The trace function has the properties that it is $\bF_{2}$-linear, meaning $\Tr(a + b) = \Tr(a) + \Tr(b)$ and $\Tr(cb) = c \cdot \Tr(b)$ for $a, b\in \bF_{2^p}$ and $c \in \bF_{2}$. The field $\bF_{2^p}$ is a linear space over $\bF_2$ of dimension $p$. We refer to a set of bases $\{ \hat{e}_i\}_{i \in [p]}$ for $\bF_{2^p}$ over $\bF_2$ to be \textit{self-dual} if for all $i,j \in [p]$
\begin{equation*}
	\Tr(\hat{e}_i \hat{e}_j) = \delta_{i,j}. 
\end{equation*}
Self-dual bases are known to exist for all finite fields of the form $\bF_{2^q}$~\cite[Theorem 1.9]{blakeApplicationsFiniteFields1993}. We refer to a set of bases $\{\hat{e}_i\}_{i \in [p]}$ as normal if there exists an element $a \in \bF_{2^{p}}$ such that $\hat{e}_i = a^{2^i}$. The following lemma shows that, for $p$ odd, there exists an efficient deterministic algorithm that computes a self-dual normal basis $\{\hat{e}_i\}_{i \in [p]}$ for $\bF_{2^p}$ given $p$. The deterministic algorithm also outputs a description of an efficient algorithm for computing finite field multiplication when elements of  $\bF_{2^p}$ are represented under the bijection specified by~\eqref{eq:finitefieldbijection} using the basis $\{\text{\gls{Canobasis}}\}_{i \in [p]}$.
\begin{lemma}[Computability of finite fields, Lemma 3.16 of \cite{jiMIPRE2022a}] \label{lem:compfinitefield}
	There exists a deterministic algorithm that, given an odd integer $p >0$, outputs a self-dual normal basis of $\bF_{2^p}$ over $\bF_2$ and the multiplication tables for the basis in $\poly(n)$ time.  
\end{lemma}
In the lemma above, a multiplication table for the set of basis $\{ \hat{e}_i\}_{i \in [p]}$ is the unique matrix representation $\{M_{\hat{e}_i}\}_{i \in [p]} \subseteq \cM_{k}(\bF_2)$ such that 
\begin{equation*}
	M_{\hat{e}_i} \kappa_{\{ \hat{e}_i\}}(a) = \kappa(\hat{e}_i a) 
\end{equation*}
for all $i \in [p]$ and $a \in \bF_{2^p}$. In this paper, we refer to the set of self-dual basis generated by~\Cref{lem:compfinitefield} as the \textit{canonical basis} of $\bF_{2^p}$. We use \gls{Canobasismap} to denote the bijection given in~\eqref{eq:finitefieldbijection} for the canonical basis. We represent elements $x \in \bF_{2^p}$ as $\kappa(x) \in \{0,1\}^p$ (where we identify elements of $\bF_2$ as $\{0,1\}$) in this paper and we refer to this as the \textit{canonical representation} for $\bF_{2^p}$. In this paper, we represent elements of any element $x \in \bF_{2^p}$ as elements of $\bF_{2^p}$ through the bijection map $\kappa$. Due to~\Cref{lem:compfinitefield}, for elements represented under the canonical representation, addition, multiplication, inversion and computing the trace can all be computed in $\poly(p)$ time (see~\cite[Lemma 3.18]{jiMIPRE2022a} for more details). 


Let $m \in \bN$, any elements in $\bF_{2^p}^m$ can be represented as elements of $\{0,1\}^{mp}$ through the canonical representation of $\bF_{2^p}$, and we refer to this as the Canonical representation for $\bF_{2^p}^m$. For clarity, we use $\{\hat{e}_i\}_{i \in [p\cdot m]}$ to denote the canonical basis for $\bF_{2^p}^m$, and $\{e_i\}_{i \in [m]}$ to denote the element $(0_0,\cdots, 1_i, \cdots, 0_m)$ in $\bF_{2^p}^m$ (where $1$ is the identity element in $\bF_{2^p}$). To abuse notation, we also use $\kappa$ to denote the map from $\bF_{2^p}^m \rightarrow \{0,1\}^{p \cdot m}$ where for $s= (s_0, \cdots, s_m) \in \bF_{2^p}^m$ 
\begin{equation} \label{eq:canorepfinitefield}
	\kappa(s) := (\kappa(s_0), \cdots, \kappa(s_m)),
\end{equation}
where $\kappa$ is the bijection given in~\eqref{eq:finitefieldbijection} for the canonical basis $\{\hat{e}_i\}_{i \in [p\cdot m]}$.  When describing elements of $\bF_{2^p}^m$, we always assume that the element is represented under the canonical representation. We refer to a subspace as a \textit{canonical basis subspace} if it is the span of some subsets of the canonical basis. We remark that this is the same definition as the ``register subspaces" in~\cite{jiMIPRE2022a}. We use \gls{dimspace} to denote the dimension of the subspace $V \subseteq \bF_{2^p}^m$. For any subspace $W \subseteq V \subseteq  \bF_{2^p}^m$, we define the orthogonal subspace of $W$ over $V$ as the space
\begin{equation*}
\text{\gls{Orthogonalsubspace}} :=  \{ u \in V: u \cdot w = 0\text{ for all } v \in W \}.
\end{equation*}
Unless otherwise stated, $W^{\prep}$ defaults to the orthogonal subspace over $\bF_{2^p}^m$. 

For subspaces $V_1, V_2 \subseteq V  \subseteq \bF_{2^p}^m$, we say that the two subspaces are disjoint if $V_1 \cap V_2 = \{0\}$. For any $k \leq l$, we refer to a set of pairwise disjoint partition of subspace $\{V_j\}_{j \in [k]}$ of $V$ as a \textit{disjoint partition} of $V$ if $\oplus_{j \in [k]} V_j = V$. For any disjoint partition of subspaces $\{V_i\}_{i \in [k]}$ of $V \subseteq \bF_{2^\seedlength}^m$ and $0 \leq i < k$ we write 
\begin{equation*}
	\text{\gls{Disjointpartsubspace}} := \bigoplus_{j \in [i]} V_{j},  \qquad	V_{> i} := \bigoplus_{i < j < k} V_{j}
\end{equation*}
and we use the convention that $V_{< i+1} = V_{\leq i}$,  $V_{> i} = V_{\geq i+1}$ and $V_{< 0 } = \{0\}$. Given $\{V_j\}_{j \in [k]}$, a disjoint partition of $V$ and $v \in V$, there exists a unique decomposition $s_j \in V_j$ for each $j \in[k]$ such that $s = \sum_{j \in [k]} s_j$. 

Given subspace $U,W \subseteq V \subseteq \bF_{2^p}^m$, we say $(U,W)$ forms a pair of complementary subspaces over $V$ if $U$ and $W$ form a disjoint partition of $V$ and $U + W = V$, and we say $(U,W)$ forms a pair of complementary subspace if it forms a complementary subspaces over $\bF_{2^p}^m$. Give $v \in V$ and a pair of complementary subspace over $V$, there exists a unique decomposition $v = u + w$ such that $u \in U$ and $v \in W$. There could potentially be multiple subspace of $V$ which can be used to form a pair complementary subspaces with $W$. For example, for $V  = \bF_{2}^2$ and $W = \text{span}{(1,1)}$, the subspace $\text{span}\{(1,0)\}$ and $ \text{span}\{(0,1)\}$ both forms a pair of complementary subspace of $V$ with $W$. 

We wish to define a notion of a unique ``canonical complement" in this paper. For a canonical basis subspace $V$ and $W \subseteq V$, we define the \textit{canonical complement} as the following: Let $\{\hat{e}_0, \cdots, \hat{e}_{\dim{V}-1} \}$ be the canonical basis element used to define $V$, and let $\{w_1, \cdots , w_{\dim(W)}\}$ be a set of linearly independent vectors in $W$.  Write each of the vector as $w_i = \sum_{j \in [\dim(V)]} a_{i,j} \hat{e}_j$ for some $ a_{i,j} \in \bF_{2^p}^m$ and run the Gaussian elimination on the $\dim(V)$ by $\dim(W)$ matrix defined by $(a_{i,j})$. Let $I$ be the set of $\dim(V)$ column with leading 1 entries in the resulting matrix, we define the \textit{canonical complement} over $V$, or \gls{Canocomplement} to be the subspace $\text{span}\{ \hat{e}_j| j \not\in I\}$. Unless otherwise stated, $W^{C}$ defaults to the canonical complement of $\bF_{2^p}^m$. The canonical complement is unique and can be computed efficiently in $\poly$ time. We remark that this is the same definition for canonical complement given in~\cite[Definition 3.6]{jiMIPRE2022a}. 

We recall the following lemma regarding the subspaces of $\bF_{2^p}^m$. 
\begin{lemma}[Lemma 3.14 of~\cite{jiMIPRE2022a}]\label{lem:finitefielddecomp}
	Let $\{\hat{e}_i\}$ be the canonical basis for $\bF_{2^p}$ over $\bF_2$ and let $V$ be a subspace of $\bF_{2^p}^m$ with linear independent basis $\{ b_1 \cdots , b_t\} \subseteq \bF_{q^k}^n$. Then the following holds:
	\begin{itemize}
		\item $\kappa(V)$ is a subspace of $\bF_2^{mp}$.
		\item $\{ \kappa(\hat{e}_i b_j)  \}_{(i,j) \in [p] \times [n] }$ forms a set of linearly independent basis of $\kappa(V)$ over $\bF_2$.
		\item Let $V,W$ be complementary subspaces of of $\bF_{2^p}^m$. Then $\kappa(V)$ and $\kappa(W)$ are complementary subspaces of $\bF_{2}^{pm}$. Furthermore, for all $a \in \bF_{2^p}^m$ with $a = a^V + a^W$, $a^V \in V$ and $a^W \in W$, we have $\kappa(a^V) \in \kappa(V)$ and $\kappa(a^W)\in \kappa(W)$. 
	\end{itemize}
\end{lemma}

Given $(v_0, \cdots, v_{n-1}) \in \bF_{2^{ \seedlength}}^m$ and $j \in [n]$, we use $\pi_{> j}^m$ to denote the function which zeros out the first $j$ entries of $\bF_{2^{\seedlength}}^m$ and we use $\pi_{\leq j}^m$ to denote the function which zeros out everything except for the first $j$ entries of $\bF_{2^{ \seedlength}}^m$. In other words 
\begin{align}
\text{\gls{Zeromap}}(v_0, \cdots, v_{m-1}) &= (0, \cdots, 0, v_j, \cdots, v_{m-1}) \label{eq:defzeromapfield} \\
\nonumber \pi_{\leq j}^m(v_0, \cdots, v_{m-1}) &= (v_0, \cdots, v_{j-1}, 0, \cdots, 0).
\end{align}

We remark that this is a different map define in~\eqref{eq:defzeromapstring}, as the entry specified here are over the finite field elements $\bF_{2^{ \seedlength}}$ instead of the $\{0,1\}$ string. We further remark that $\kappa^{-1} \circ \pi_{\leq j} \circ \kappa$ and $\bF_{2^\seedlength}^{m}$ are different maps, where the first map are usually used for treating an element from $\bF_{2^p}^{m}$ as a string and the second one are usually used for treating an element from $\bF_{2^p}^{m}$ as a vector space over $\bF_{2^p}$. 

In this paper, we work with linear functions over canonical basis subspace. For a linear function $\decidertypedfunction$ mapping from $V \rightarrow V$, we use \gls{KerLinfun} to denote the subspace of $V$ such that $\decidertypedfunction(a) = 0$ for all $a \in \ker(\decidertypedfunction)$. 

For a linear function $\decidertypedfunction: V \rightarrow V$ over a canonical basis subspace $V \subseteq \bF_{2^p}^m$, we define the linear function $\text{\gls{ProjCompsubspace}}: V \rightarrow V$ as the projection into the subspace of $\left(\ker{(\decidertypedfunction)}^{\bot}\right)^{C}$. To be more precise, for $v = v_1 + v_2 \in V$ with $v_1 \in \ker{(\decidertypedfunction)}^{\bot}$ and $v_2 \in \left(\ker{(\decidertypedfunction)}^{\bot}\right)^{C}$, $\decidertypedfunction(v) = v_2$. We remark that this is the same as the linear map defined in~\cite[Definition 3.11]{jiMIPRE2022a}. 

By a simple calculation, we see that
\begin{equation*}
	\ker{(L)}^{\bot}=  \ker{\left(\decidertypedfunction^{\bot}\right)}. 
\end{equation*}

\subsection{Affine lines and polynomials over a finite field}
In this subsection, we recall some properties related to affine lines and low-degree polynomials over a finite field. In this paper, we refer to an affine line $\textbf{l}$ over $\bF_{2^p}^m$ as the set of the form
\begin{equation*}
	\{ u + t \cdot v: t \in \bF_{2^p}\}
\end{equation*}
for $u, v \in \bF_{2^p}^m$. Given a line $\textbf{l}$, we wish to define a unique $u_\textbf{l}, v_\textbf{l} \in \bF_{2^p}^m$ which can be used to represent $\textbf{l}$. Given $u \in V$, we define the linear function $\text{Null}_{v}^{\text{LN}}: \bF_{2^p}^m \rightarrow \bF_{2^p}^m $ as 
\begin{equation*}
	\Nullline{v}{u} := u_v,
\end{equation*}

where $u = u_v + u_{v}^{C}$ is the unique decomposition such that $u_v \in \text{span}({v})$ and $u_{v^{C}} \in\text{span}({v})^{C}$. We remark that $\text{Null}_{v}^{\text{LN}}$ is the same as~\cite[Definition  7.3]{jiMIPRE2022a}. We define the canonical representation of an affine line as the following:
\begin{definition}[Canonical representation of an affine line]\label{def:canorepline}
	Let $p \in \bN$ be an odd integer and $m \in \bN$, and let $l = \{ u + tv: t \in \bF_{2^p}\}$ be an affine line passing through $\bF_{2^p}^m$. The canonical representation of $l$ is defined as 
	\begin{equation*}
		\text{\gls{Canolinerep}} := (v, \Nullline{v}{u}) \in \bF_{2^p}^{2m}.
	\end{equation*}
\end{definition}
In the definition above, we see that $\Nullline{v}{u} =  \Nullline{v}{u'}$ for $u, u' \in l$ as $u' = u + t \cdot v$ any scalar $t \in \bF_{2^p}$. Hence, the canonical representation is independent of the initial point chosen. We remark that this is the same definition given in~\cite[Definition 7.3]{jiMIPRE2022a}.

Given a function $\textbf{f}: \bF_{2^p}^m \rightarrow \bF_{2^p}$, we say $\textbf{f}$ is an \textit{$m-$variant polynomial} over $\bF_{2^p}$ if $\textbf{f}$ is of the form
\begin{equation*}
	\textbf{f}(x_1, \cdots, x_m) = \sum_{(i_1, \cdots, i_m) \in [2^p]^m} \alpha_{i_1, \cdots, i_m} x_1^{i_1} \cdots x_m^{i_m},
\end{equation*}
where each of the $\alpha_{i_1, \cdots, i_m}$ are some coefficients in $\bF_{2^p}$. Furthermore, we say that $\textbf{f}$ has an \textit{individual degree} of $d \in \bN$ if the sum above is defined over $[d]^m$ instead (in other words, $\alpha_{i_1, \cdots, i_m} =0$ if there exists a $j \in [m]$ such that $i_j > d$). We use \gls{LDpolyset} to denote the set of polynomials $\textbf{g}: \bF_{2^p}^m \rightarrow\bF_{2^p}$ with individual degree of at most $d$. We recall the following lemma regarding the distance between two distinct low-individual degree polynomials. 
\begin{lemma}[Schwartz-Zippel~\cite{schwartzFastProbabilisticAlgorithms1980,zippelProbabilisticAlgorithmsSparse1979}] \label{lem:Schwartz_Zipple}
	Let $\textbf{f},\textbf{g}\in \Idpoly(p,m,d)$ be two different $m$-variant polynomials with individual degree of at most $c$, then 
	\begin{equation*}
		\Pr_{u \sim \bF_{2^p}^m} [\textbf{f}(u) = \textbf{g}(u)] \leq \frac{md}{2^p}. 
	\end{equation*}
\end{lemma}

For a more comprehensive introduction for finite fields, we refer the readers to~\cite{mullenHandbookFiniteFields2013}. 

\subsection{Generalized Reed-Muller code } \label{sec:RMcode}
Finally, we recall the generalized Reed-Muller code in this subsection. Recall from~\cite{jiQuantumSoundnessTesting2022}, a linear $[n, c, d]_{\bF_{q}}$ code is a set $\mathfrak{C}$ of functions $g: [p] \rightarrow \bF_{q}$ with size $|\mathfrak{C}| = q^c$ that is closed under linear combination, such that for any two distinct $g \neq g'$, the number of coordinates $i \in [n]$ such that $g(i) \neq g'(i)$ is at least $d$. Given $\mathfrak{C}$, a linear $[n, c, d]_{\bF_{q}}$ code, the tensor code $\mathfrak{C}^{\otimes m}$ is the set of all functions $f:[n]^m \rightarrow \bF_{q}$ such that the restriction $f|_{\textbf{l}_j}$ to any axis-parallel line $\textbf{l}_j$ is a codeword in $\mathfrak{C}$, where for $j \in m$, an axis parallel line $\textbf{l}_j$ is defined as
\begin{equation*}
	\textbf{l}_j = \{(s_0, \cdots, s_{j-1}, x, s_{j+1}, \cdots, s_{m-1}): x \in \bF_{q}\}
\end{equation*}
Given constants $p, c \in \bN$, the set $\mathfrak{C}$ consists of all degree $c$ polynomials $\textbf{f}: \bF_{2^p} \rightarrow \bF_{2^p}$ is a $[2^p, c, c]_{\bF_{2^p}}$ code by the Schwartz-Zippel lemma. We further see that the set of low-individual degree polynomials $\textbf{f}: \bF_{2^p}^m \rightarrow \bF_{2^p}$ with individual degree at most $c$ is a tensor code $\mathfrak{C}^{\otimes}$ (since $\textbf{f}|_{\textbf{l}_j}$ always gives a 1-variant polynomial of degree at most $c$). 

Low-individual degree polynomials can also be used to define an error correction code with good distance properties in the context of the \textit{generalized Reed-Muller code}. Given a string $s \in \{0,1\}^m$, we define the \textit{indicator polynomial} for $m$ over the field $\bF_{2^p}$ as the $m-$variant polynomial with $\indpoly{y}: \bF_q^m \rightarrow \bF_q$ as
\begin{equation*}
	\indpoly{m,a}(x) := \prod_{i: a_i = 0} x_i \cdot \prod_{i: a_i = 1} (1 - x_i).
\end{equation*} 
where here, we identify $\{0,1\}^m$ as elements of $\bF_{2^p}^m$. The indicator polynomial has individual degree of 1 and has the properties that for all $s \in \{0,1\}^m \subseteq \bF_{2^p}^m$, $\indpoly{a}(s) = 0$ except when $a = s$. 

Let $M = 2^m$, for any elements $b \in \{0,1\}^M$, we define the generalized Reed-Muller encoding of $a$ to be the polynomial
\begin{equation} \label{eq:REM}
	\text{\gls{RMcode}}(x) := \sum_{y \in \{0,1\}^m} b_{\binaryinv{y}} \indpoly{m,y}(x),
\end{equation}
where recalled, $\binaryinv{\cdot}$ is the map which maps the corresponding binary representation back to an integer. Since each $\indpoly{m,y}(x)$ has an individual degree of $1$, $\ReedMcode_b$ also has an individual degree of $1$. For any $y \in \{0,1\}^m \subseteq \bF_q^m$, evaluating $\ReedMcode_b$ with $y$ returns the $\binaryinv{y}$th coordinate of the string $b$.

\section{Quantum preliminaries}

\subsection{Von Neumann algebras}  \label{sec:introVNA}

We use the language of tracial von Neumann algebras to discuss non-local games in this paper. We introduce some of the necessary background needed for the main body of this paper. This section follows a similar structure as~\cite[Section 2.3]{linTracialEmbeddableStrategies2024}.

Let \gls{Hilbertspace} be a Hilbert space, and \gls{BoundHS} denote the set of bounded operators on $\cH$. Given $\ket{\psi} \in \cH$, we use \gls{Vectornorm} to denote the vector norm. Recall, a (concrete, unital) $C$*-algebra $\alicealg \subseteq \bofh$ is a normed $*$-algebra with $\cI_{\cH} = \cI_{\alicealg}$ and closed in the norm topology. We use $\alicealg^{+}$ to denote the set of positive elements within $\alicealg$ (i.e.~elements of the form $s^*s$ for $s \in \alicealg$).

A \textit{state} on a $C$*-algebra is a linear function $\psi: \alicealg \rightarrow \bC$, which is \textit{positive}, meaning that $\psi(a) \geq 0$ for all $a \in \text{\gls{VNApositive}}$ and satisfies $\psi(\cI) = 1$. We use \gls{Statenorm} to denote the operator norm of $\psi$, or 
\begin{equation*}
	\| \psi \| := \sup\{ \psi(z) | z \in \alicealg^{+}\},
\end{equation*}
and we write $\| \psi\|_{\alicealg}$ in order to emphasize the underlying algebra that the norm is taken over. A state $\psi$ on $\alicealg$ is said to be \textit{faithful} if, for all $a \in \alicealg^{+}$, we have $\psi(a) = 0$ if and only if $a = 0$. Furthermore, we say that the  state is a \textit{tracial state} if $\psi(st) = \psi(ts)$ for all $s,t \in \alicealg$. The famed GNS representation theorem states that every state $\psi$ on a $C$*-algebra $\alicealg$ induces a \textit{representation} (a $*$-homomorphism to some $\BofH$) $\pi_{\psi}$ onto $\cB(\cH_{\psi})$, and a unit vector $\ket{\psi} \in \cH_{\psi}$ such that $\psi(z) = \braket{\psi|\pi_{\psi}(z)|\psi}$ for all $z \in \alicealg$, and $\overline{\alicealg \ket{\psi}} = \cH$ (we refer to \cite[Theorem 4.5.2]{kadisonFundamentalsTheoryOperator1997} for more details about the GNS representation). This representation is specified with the triplet $(\pi_{\psi}, \cH_{\psi}, \ket{\psi})$.

An element $P \in \alicealg$ is a projector if $P^2 = P$. An element $V \in \alicealg$ is a partial isometry if and only if $VV^*$ and $V^*V$ are both a projector, and unitary if furthermore $VV^* = V^*V = \cI_{\alicealg}$. For projector $P, Q \in \alicealg$, we say the two projectors are equivalent if there exist some partial isometry $V \in \alicealg$ such that $VV^* = P$ and $V^*V = Q$. We use $\cU(\alicealg)$ to denote the set of unitary elements ($A^* A = A^* A= \cI$) in $\alicealg$ in this paper. 

For $\alicealg \subseteq \BofH$, the \textit{commutant} \text{\gls{VNAcommutant}} of $\alicealg$ is defined to be the set of all elements which commute with $\alicealg$, or $\alicealg' := \{z \in \BofH : zw = wz \text{ for all } w \in \alicealg\}$. A $C$*-algebra $\text{\gls{VNA}} \subseteq \BofH$ is said to be a \textit{von Neumann algebra} if $\alicealg = \alicealg''$. By the von Neumann bicommutant theorem, an equivalent definition for von Neumann algebra  $\alicealg$ is for $\alicealg$ to be closed in the weak $*$-topology. Since the weak $*$-topology is more coarse than the norm topology, not every $C$*-algebra is a von Neumann algebra. Unless stated otherwise, $\alicealg$ is assumed to be a concrete von Neumann algebra for the remainder of this paper. 

A state $\psi$ on $\alicealg$ is said to be \textit{normal} if for all bounded increasing nets $\{A_{\lambda}\} \subseteq \alicealg^{+}$ with $A = \sup_{\lambda} \{ A_{\lambda} \}$, we have $\psi(A) = \lim \psi(A_{\lambda})$. 


\paragraph{Tracial von Neumann algebras.}  We refer to a von Neumann algebra to be \textit{tracial} if it admits a faithful normal tracial state \gls{VNAtrace}, and we use $(\alicealg, \tau)$ to emphasize the existence of $\tau$. Whenever $\alicealg$ is finite-dimensional, we use \gls{VNAtraceFD} to denote the trace function for clearly\footnote{We remark that the notation for the trace function for a finite-dimensional matrix is the same as the trace function for a finite field. In the context of the $\Tr$ function in this paper, we typically use lower case letter for finite field element and upper case letter for a matrix element.}. The faithful trace $\tau$ naturally gives the notion of a ``Hilbert-Schmidt'' norm on $\alicealg$, defined to be
\begin{equation*}
	\text{\gls{HSnorm}} := \sqrt{\tau(A^* A)}.
\end{equation*}
Recall that the \textit{standard form} for a tracial von Neumann algebra $(\alicealg, \tau)$ is the GNS representation triplet $(\chi_{\tau}, \text{\gls{SDHilbertspace}}, \text{\gls{SDTracialstate}})$ of $\alicealg$ for the tracial state $\tau$, where $\cL^2(\alicealg, \tau)$ denotes the Hilbert space for the representation. Note that the standard form for a tracial von Neumann algebra is unique up to canonical isomorphism \cite[Proposition 7.5.1]{anantharamanIntroductionII1Factors2010}. For simplicity of notation, if $(\alicealg, \tau)$ is in standard form, for each $a \in \alicealg$, we use $a$ to denote $\chi_{\tau}(a)$ as the operator defined within $\cB(\cL^2(\alicealg, \tau))$. In this representation, the vector $\ket{\tau}$ is \textit{cyclic}, meaning that $\overline{\chi_{\tau}(\alicealg) \ket{\tau}} = \cL^2(\alicealg, \tau)$, and \textit{separating}, meaning that for all $z \in \alicealg$, we have $z \ket{\tau} = 0$ if and only if $z = 0$. This means that each $\sigma \in \alicealg$ gives a unique vector $\sigma \ket{\tau} \in \cL^2(\alicealg, \tau)$, and we can specify the action of $a$ acting on the Hilbert space $\cL^2(\alicealg, \tau)$ by its \textit{left regular representation}:
\begin{equation*}
	a (\sigma \ket{\tau}) = (a \sigma) \ket{\tau}
\end{equation*}
for all $\sigma \in \alicealg$.

Recall, given a von Neumann algebra $\alicealg$, the \textit{opposite algebra} $\text{\gls{Oppmap}} := \{a^{op}: a \in \alicealg \}$ is a von Neumann algebra which has the same linearity as $\alicealg$, but has the opposite multiplication structure, or more precisely $(ab)^{op} = (b)^{op} (a)^{op}$. The algebra $\alicealg^{op}$ can also be faithfully embeddable onto $\cB(\cL^2(\alicealg, \tau))$ by
\begin{equation} \label{eq:oppemb}
	\chi_{\tau}^{op}(a^{op}) (\sigma  \ket{\tau}) = (\sigma a) \ket{\tau}.
\end{equation}
This is known as the right regular representation for $\alicealg$. Clearly, $\chi_{\tau}^{op}(\alicealg)  \subseteq \alicealg'$, and in fact, $\alicealg^{'} = \alicealg^{op}$ \cite[Theorem 7.1.1]{anantharamanIntroductionII1Factors2010}. For simplicity of notation, we use \gls{Oppmapele} to denote $\chi_{\tau}^{op}(a)$ in this paper. The map $op: a \rightarrow a^{op}$ forms a $*$-anti-isomorphism from $\alicealg \rightarrow \alicealg'$, meaning 
\begin{equation*}
	(\lambda a + b)^{op} = \lambda a^{op} + b^{op}, \; \quad (ab)^{op} = b^{op} a^{op}  \; \quad (a^*)^{op} = ((a)^{op})^*
\end{equation*}
for all $a,b \in \alicealg$, $\lambda \in \bC$. 

\paragraph{Finite dimension example.} To make the definition above more concrete, let $n \in \bN$ and $\alicealg = \cM_n(\bC)$, we define the maximally entangled state $\MEstate{n}$ as the vector state on $\bC^{n} \otimes \bC^{n}$ as 
\begin{equation*}
\MEstate{n} := \frac{1}{\sqrt{n}} \sum_{i \in n} \ket{i} \otimes \ket{i}
\end{equation*}
We remark that this is precisely the vector state which arises from applying the GNS theorem on the normalized matrix trace $\Tr_n(A) $ on the algebra $\cM_n(\bC)$, the resulting vector representation for $\Tr$ are $\MEstate{n}$. Under this representation, elements of $\cM_n(\bC)$ gets mapped to $\cM_n(\bC) \otimes \cI_n$, with the commutant being $\cI_n \otimes \cM_n(\bC)$. For all vectors $\ket{\psi} \in \bC^{n} \otimes \bC^{n}$, we can always find some positive element $\sigma \in \cM_n(\bC) \tensor \cI_n$ such that $\sigma \MEstate{n}$, mainly, the canonical square root (the unique square root which is also positive) of the reduced density on the first register. The opposite algebra map, in this case, is the map $op: \alicealg \otimes \cI_n \rightarrow \cI_n \otimes \alicealg$ defined by $op(A \otimes \cI_n) = \cI_n \otimes A^{T}$, where $A^{T}$ is the transpose of $A$. As a further sanity check, for all $A,B \in \cM_n(\bC)$ 
\begin{equation}
	op(A \otimes \cI_n) (B \otimes \cI_n) \MEstate{n} = (B \otimes A^T) \MEstate{n}  = (B A\otimes \cI_n)\MEstate{n}, 
\end{equation}
consistent with the definition given above. For a more comprehensive introduction on von Neumann algebra, we refer the reader to \cite{blackadarOperatorAlgebras2006}.

\subsection{Quantum measurements}

Let $\cH$ be a (potentially infinite-dimensional) Hilbert space, and let $\BofH$ denote the set of the bounded operators on $\cH$. If $\cA$ is a finite set, then a \textit{positive operator-valued measure (POVM)} on $\cH$ with outcome set $\cA$ is a collection of positive operators $\{ A_a \}_{a \in \cA} \subseteq \BofH$, such that $\sum_{a \in \cA} A_a = \Id_{\cH}$. A \textit{projection-valued measure (PVM)}, or \textit{projective measurement}, is a POVM where each of the operators $A_a$ is a projection operator (i.e.~$A_a^2 = A_a$). 



Let $\textbf{f}: \cS \rightarrow \cA$ be a function mapping a finite set $\cS$ to another finite set $\cA$, and let $\{A_t\}_{t \in \cA}$ be a POVM with measurement outcome in $T$. For all $s \in \cS$, we denote
\begin{equation}\label{eq:dataprocessmeasdef}
	\text{\gls{Datameas}} := \sum_{a: \textbf{f}(a) = s} A_a,
\end{equation}
and $A_{[f|s]} = 0$ if $s$ is not in the image for $f$. Intuitively, this corresponds to performing the measurement which first samples an element from $A_t$, and apply the map  $f$ through the measurement outcome. This is known as a ``data processed measurement" in the literature. Before ending this section, we recall the orthogonalization lemm, which is used to approximate a set of POVMs on a von Neumann algebra by a PVM being defined on the same algebra. 

\begin{lemma}[Orthogonalization lemma, Theorem 1.2 of~\cite{delasalleOrthogonalizationPositiveOperator2022}] \label{lem:orthogonalizationlemma}
	Let $\alicealg \subseteq \BofH$ be a von Neumann algebra and let $\ket{\psi} \in \cH$ be a unit vector. For any POVM $\{ A_a \} \subseteq \alicealg$ such that $\sum_a \braket{\psi|A_a^2|\psi} >1-  \epsilon$, there exists a PVM $\{P_a\} \subseteq \alicealg$ such that
	\begin{equation*}
		\sum_{a} \braket{\psi| (A_a - P_a)^2 |\psi}< 9 \eps.
	\end{equation*}
\end{lemma} 
If $(\alicealg, \tau)$ is a tracial von Neumann algebra in standard form, we can replace $\ket{\psi}$ by $\sigma \ket{\tau}$ for some $\sigma \in \alicealg$ for the lemma above in order to obtain a Hilbert-Schmidt norm approximation of the original POVM. 

\paragraph{Generalized Pauli measurements.} Let $p \in \bN$ be an odd integer. Recall, for $W \in \{X,Z\}$, the generalized Pauli measurement over $\bF_{2^p}$ are the sets of PVM  $\left\{\text{\gls{GenPauli1}} = \ketbra{a^{W, p}}{a^{W, p}} \right\}_{a \in \bF_{2^{p}}}$ where 
\begin{equation*}
	\ket{a^{X, p}} := \frac{1}{\sqrt{2^{p}}} \sum_{b \in \bF_{2^p}} (-1)^{\Tr(ab)}  \ket{b}, \qquad \ket{a^{Z, p}} := \ket{a},
\end{equation*}
for all $a \in \bF_{2^{p}}$. In the case where $p = 1$, the generalized Pauli measurements corresponds to the eigenspace of the Pauli $X$ and $Z$ matrices.

In this paper, we often associate PVMs with binary observables to better analyze the commutation properties of these measurements. A binary observable is a unitary matrix which squares to the identity. For an odd integer $p \in \bN$, $W \in \{X,Z\}$ and $a \in \bF_{2^p}$, we define the generalized Pauli matrices \gls{GenPauli2} as 
\begin{equation*}
	\rho^{X, p}(a) := \sum_{b \in \bF_{2^p}} \ketbra{b + a}{b},  \qquad \rho^{Z, p}(a) := \sum_{b \in \bF_{2^p}} (-1)^{\Tr(a \cdot b)} \ketbra{j}{j}
\end{equation*}
where the addition and multiplication above are over $\bF_{2^p}$. In the case where $p = 1$, we drop the superscript and simply write $\rho^{W}_i$ to denote the qubit Pauli $W$-measurement for $i \in \{0,1\}$ and $\rho^{W}$ for the Pauli $W$-matrix. We see that for $W \in \{X,Z\}$ the eigenspace for $\rho^{W, p}(a)$ is precisely the PVM measurement for the generalized Pauli measurement, and we can write 
\begin{equation}\label{eq:genpaulitoobser}
	\rho^{W, p}(a) := \sum_{b \in \bF_{2^p}} (-1)^{\Tr(ab)} \rho^{W, p}_{a}. 
\end{equation}
This also implies that $\rho^{W, p}(a)$ commutes with $\rho^{W, p}(b)$ for any $a,b \in \bF_{2^p}$. As shown in~\cite[equation 19]{jiMIPRE2022a}, the generalized Pauli measurements can also be written as
\begin{equation}\label{eq:genpaulitomeasure}
	\rho^{W, p}_{a} = \Ex_{b \in \bF_{2^p}} (-1)^{-\Tr(ab)}  	\rho^{W, p}(b).
\end{equation}
By a simple calculation, we see that for all $W \in \{X,Z\}$ and $a,b\in \bF_{2^p}$, the generalized Pauli observables obey the following relationships  
\begin{equation}\label{eq:genpauliaddmult}
	\rho^{W, p}(a) \cdot  \rho^{W, p}(b)  = \rho^{W, p}(a+b).
\end{equation}
The generalized Pauli observables also follow the ``twisted commutation" relations, whereby
\begin{equation}\label{eq:genpaulicomm}
	\rho^{X, p}(a) 	\cdot \rho^{Z, p}(b) = (-1)^{\Tr(ab)} \rho^{Z, p}(b) \cdot \rho^{X, p}(a). 
\end{equation}
For $W \in \{X,Z\}$ and $s \in \bF_{2^p}^m$, we define
\begin{equation} \label{eq:genPaulimeasurement}
	\rho^{W, p}(s) = \bigotimes_{i \in [m]} \rho^{W, p}(s_i) \qquad \text{and } \qquad \rho^{W, p}_{s} = \bigotimes_{i \in [m]} \rho^{X, p}_{s_i}
\end{equation}
where each $s_i \in \bF_{2^p}$. We recall the following lemma which shows the existence of a unitary which converts between the generalized Pauli measurement to the one qubit Pauli measurement.
\begin{lemma}[Lemma 3.26 of~\cite{jiMIPRE2022a}] \label{lem:PaulitogeneralPauli}
	Let $p, m \in \bN$ where $p$ is an odd integer, there exists a unitary $\text{\gls{GenPauli4}}:  (\bC^{2})^{\otimes p \cdot m} \rightarrow (\bC^{2^p})^{\otimes m}$ such that for all $W \in \{X,Z\}$ and for all $s\in \bF_{2^p}^{\otimes n}$, we have 
	\begin{align*}
		&\rho^{p, W}_s =  U_{2 \rightarrow p}^* \left( \bigotimes^{m}_{i = 0} \bigotimes^{p}_{j = 0} \rho^{W}_{\kappa(s_i)_j}  \right)U_{2 \rightarrow p} \\
		&\left(U_{2 \rightarrow p}^* \otimes U_{2 \rightarrow p}^* \right)\MEstate{2^p}^{\otimes m} =  \MEstate{2}^{\otimes p \cdot m}. 
	\end{align*}
\end{lemma}
\jqnote{You probably need to fix this later!}


To abuse notation, for register subspace $V \subseteq \bF_{2}^m$ and $W \in \{X,Z\}$, we use \gls{GenPauli3} to denote the measurement
\begin{equation*}
	\rho^{W}_{s} := \sum_{a| a_V = s} \rho^{W}_{s}
\end{equation*}
where for $a \in \bF_{2}^m$, $a = a_V + a_V^C$ is the unique decomposition such that $a_V \in V$ and $a_V^C \in V^C$. For register subspace $V \subseteq \bF_{2^p}^m$, we use $\{\rho^{W}_{s} \}_{s \in V}$ to denote the measurement  $\{\rho^{V}_{s} \}_{s \in \kappa(W)}$. We recall the following lemma from~\cite{jiMIPRE2022a}. 
\begin{lemma}[Lemma 8.5 of~\cite{jiMIPRE2022a}] \label{lem:linePaulicommutation}
	Let $\decidertypedfunction_1, \decidertypedfunction_2: \bF_{2^\seedlength}^m \rightarrow \bF_{2^\seedlength}^m$ be two linear map. Then
	\begin{equation*}
		\ker{\left(\decidertypedfunction_2\right)}^{\bot} \subseteq \ker{\left(\decidertypedfunction_1\right)},
	\end{equation*}
	implies that the measurement operators $\left\{\rho^{Z, \seedlength}_{[\decidertypedfunction_1|s]}\right\}_{s\in \bF_{2^\seedlength}^m}$ and $\left\{\rho^{X, \seedlength}_{[\decidertypedfunction_2|s]}\right\}_{s\in \bF_{2^\seedlength}^m}$ pairwise commute, as well as the measurement operators $\left\{\rho^{Z, \seedlength}_{[\decidertypedfunction_2|s]}\right\}_{s\in \bF_{2^\seedlength}^m}$ and $\left\{\rho^{X, \seedlength}_{[\decidertypedfunction_1|s]}\right\}_{s\in \bF_{2^\seedlength}^m}$ pairwise commute
\end{lemma}

\subsection{Distance between quantum measurements}

In this subsection, we introduce some distance between measurements which will be useful for the analysis of non-local games. Let $\cX$ be a finite set, $\mu$ be a probability measurement, $(\alicealg, \tau)$ be a tracial von Neumann algebra represented under the standard form $(\chi_{\tau}, \cL^2(\alicealg, \tau), \ket{\tau})$ and $\ket{\psi} \in \cL^2(\alicealg, \tau)$. We say that the two sets of POVM, $\{ A_a^x \}_{x \in \cX} \subseteq \alicealg$ and $\{ B_a^x \}_{x \in \cX} \subseteq \alicealg'$, are $\delta$-\textit{consistent} with each other with respect to $\ket{\psi}$ and $\mu$ if
\begin{equation} \label{eq:Consistencydef}
	\Ex_{x \sim \mu} \sum_{a \neq b} \braket{\psi|A_a^x B_b^x|\psi} \leq O(\delta),
\end{equation}
and we write 
\begin{equation*}
	A_a^x \text{ \gls{Consistency} }  B_a^x
\end{equation*}
if $\{ A_a^x \}$ and $\{ B_a^x \}$ are $\delta$-consistent with each other and $\ket{\psi}$ and $\mu$ are clear from context.

For two sets of POVM $\{ A_a^x \}_{x \in \cX}, \{ B_a^x \}_{x \in \cX} \subseteq \BofH$, we say that $\{ A_a^x \}$ and $\{ B_a^x \}$ are $\delta$-\text{close} with each other with respect to $\ket{\psi}$ and $\mu$ if
\begin{equation} \label{eq:Closedef}
	\Ex_{x \sim \mu} \sum_{a} \| \left(A_a^x -  B_a^x\right) \ket{\psi}\|^2 \leq O(\delta),
\end{equation}
and we write
\begin{equation*}
	A_a^x \text{ \gls{Close} }  B_a^x
\end{equation*}
if $\{ A_a^x \}$ and $\{ B_a^x \}$ are $\delta$-close with each other and $\ket{\psi}$ and $\mu$ are clear from context. We also use the same notation to denote distances between matrices if the superscript ``x" are omitted when describing $\simeq$ or $\approx$ distances. By definition, $A_a^x  \approx_{\eps} B_a^x $ is the same as writing $\left(A_a^x - B_a^x \right) \approx_{\eps} 0$. We remark that both measurements distance defined above are analogous to~\cite[Definition 5.15, 5.16]{jiMIPRE2022a}.

For three sets of POVM $\{ A_a^x \}_{x \in \cX}, \{ B_a^x \}_{x \in \cX}, \{ C_a^x \}_{x \in \cX} \subseteq \BofH$, if $ A_a^x \approx_{\delta} B_a^x$ and $B_a^x \approx_{\eps} C_a^x $ over $\ket{\psi} \in \cH$ and $\mu$ by the triangle inequality, this implies that $A_a^x \approx_{\delta+ \eps} C_a^x $ over $\mu$ and $\ket{\psi}$. Since applying a function to the measurement output cannot decrease the probability of two measurement outcome agree with each other, we get the following analogue of~\cite[Fact 4.26]{natarajanNEEXPMIP2019a}.
\begin{fact}[Data processing] 
	Let $\cX$ be a finite set,  $\{ A_{a}^x \}_{a \in \cA}$ and  $\{ B_{b}^y \}_{(a,b) \in \cA^2}$ be two sets of POVMs, and $\textbf{f}: \cA \rightarrow \cB$ be a function. Then $A_a^x \simeq_{\eps} B_a^x$ implies that $A_{[\textbf{f}|a]}^x \simeq_{\eps} B_{[\textbf{f}|a]}^x$ 
\end{fact}
We remark that the above fact does not work for the $\approx_{\eps}$ measurement outcome and we refer to~\cite[Fact 4.26]{natarajanNEEXPMIP2019a} for more details. We show the following trivial lemmas about distances of POVM measurements. The first lemma converts between closeness and distance for a pair of commuting measurement. We remark that this is an analogue of~\cite[Fact 4.13, 4.14]{natarajanNEEXPMIP2019a}
\begin{lemma}[Conversion between closeness and distance]\label{lem:closetodistance}
	Let $\cX$ be a finite set, $\mu$ be a distribution over $\cX$, $\ket{\psi} \in \cH$ and $\{ A_{a}^x \}_{a \in \cA}$ and  $\{ B_{b}^y \}_{(a,b) \in \cA^2}$ be two sets of POVMs in $\BofH$ such that $A_a^x B^y_b = B^y_b A_a^x$ for all $(x,y,a,b) \in \cX^2 \times \cA^2$. Then the following holds:
	\begin{itemize}
		\item If $A_a^x \simeq_{\delta} B^y_b$ over $\mu$ and $\ket{\psi}$, then $A_a^x \approx_{\delta} B^y_b$ over $\mu$ and $\ket{\psi}$.
		\item If $A_a^x \approx_{\delta} B^y_b$ over $\mu$ and $\ket{\psi}$ and additionally both $\{ A_{a}^x \}, \{ B_{b}^y \}$ are PVMs, then $A_a^x \simeq_{\delta} B^y_b$ over $\mu$ and $\ket{\psi}$.
		\item  If $A_a^x \approx_{\delta} B^y_b$ over $\mu$ and $\ket{\psi}$ and additionally either $\{ A_{a}^x \}$ or $\{ B_{b}^y \}$ are PVMs, then $A_a^x \simeq_{\sqrt{\delta}} B^y_b$ over $\mu$ and $\ket{\psi}$.
	\end{itemize}
\end{lemma} 
\begin{proof}
	By definition $A_a^x \simeq_{\delta} B^y_b$, we have 
	\begin{equation*}
		\Ex_{x \sim \mu} \sum_{a} \braket{\psi|A_a^x B_a^x|\psi} \geq 1-  O(\delta)
	\end{equation*}
	By expanding the definition of closeness, we have 
	\begin{align*}
		\Ex_{x \sim \mu} \sum_{a}\| \left(A_a^x -  B_a^x\right) \ket{\psi}\|^2 &= 	\Ex_{x \sim \mu} \sum_{a}  \braket{\psi|(A_a^x)^2 +(B_b^y)^2 - 2 A_a^x B_a^x  |\psi} \\
		&\leq \Ex_{x \sim \mu} \sum_{a}  \braket{\psi|A_a^x +B_b^y - 2 A_a^x B_a^x  |\psi} \\
		&\leq 2 - 2	\Ex_{x \sim \mu} \sum_{a} \braket{\psi|A_a^x B_a^x|\psi}
	\end{align*}
	and item 1 follows accordingly. For item 2, if both $\{ A_{a}^x \}, \{ B_{b}^y \}$ are PVMs, then the above inequality becomes an equality and the statement follows accordingly. For item 3, without lost of generality assume that $\{ A_{a}^x \}$ is projective, then 
	\allowdisplaybreaks{
	\begin{align*}
	 1 - \Ex_{x \sim \mu} \sum_{a} \braket{\tau| \sigma A^x_a B^x_a \sigma |\tau} &=  \Ex_{x \sim \mu} \sum_{a} \braket{\tau| \sigma (A^x_a)^2 \sigma |\tau} - \Ex_{x \sim \mu} \sum_{a} \braket{\tau| \sigma A^x_a B^x_a \sigma |\tau} \\
	 &=  \Ex_{x \sim \mu} \sum_{a} \braket{\tau| \sigma A^x_a \cdot (A^x_a - B^x_a) \sigma |\tau} \\
	 &\leq \Ex_{x \sim \mu} \sum_{a} \| A^x_a \sigma\ket{\tau} \| \cdot \| (A^x_a - B^x_a) \sigma \ket{\tau}  \| \\
	 &\leq \sqrt{ \Ex_{x \sim \mu} \sum_{a} \| A^x_a \sigma\ket{\tau} \|^2 } \sqrt{ \Ex_{x \sim \mu} \sum_{a} \|  (A^x_a - B^x_a)\sigma\ket{\tau} \|^2} 
	\end{align*}
}
	where in line 2, we use the fact that $\{ A_{a}^x \}$ is projective and a PVM, and the third line follows from Cauchy-Schwartz and the forth line follows from Jensen's inequality. Bounding the first term in line 4 by $1$ completes the claim for the lemma. 
\end{proof}
The second lemma gives a way to combine measurements while preserving distances between the measurements, we remark that this is an analogue of~\cite[Fact 4.20]{natarajanNEEXPMIP2019a}. 
\begin{lemma}[Combination of measurement preserves distance] \label{lem:combmeasure}
	Let $\cX, \cA, \cC$ be finite sets, $\mu$ be a distribution over $\cX^2$ with marginal distribution $\mu_X \sim \cX$ and $\mu_Y \sim \cX$ over the first and second coordinates respectively. For each $(x,y) \in \cX^2$, let $\{ A_{a,b}^x \}_{(a,b) \in \cA^2}$ and  $\{ B_{a,b}^x \}_{(a,b) \in \cA^2}$ be two sets of POVMs in $\BofH$, and let $\{ C_{a,c}^{x,y} \}_{(a,c) \in \cA  \times \cC}$ be a set of POVM in $\BofH$. If $A^{x}_{a,b} \approx_{\delta} B^{x}_{a,b}$ with respect to $\ket{\psi}$ and either $\mu_X$ or $\mu_Y$, then 
	\begin{equation*}
		C_{a,c}^{x,y} A^{x}_{a,b}  \approx_{\delta}C_{a,c}^{x,y} B^{x}_{a,b}, \qquad \text{and} \qquad  A^{x}_{a,b}  C_{a,c}^{x,y} \approx_{\delta} B^{x}_{a,b}C_{a,c}^{x,y}
	\end{equation*}
	where $\approx_{\delta}$ is over the state $\ket{\psi}$ and the distribution $(x,y) \sim \mu$. 
\end{lemma}
\begin{proof}
 	Since both implication follows a similar proof, we only show the first one below. Fix $(x,y) \in \cX^2$ and $(a,b) \in \cA^2$. By expanding the vector state, we see that
	\begin{align*}
		\sum_c  \| \left(C_{a,c}^{x,y} A^{x}_{a,b}  -  C_{a,c}^{x,y} B^{x}_{a,b}   \right) \ket{\psi}\|^2 &= \sum_c \braket{\psi| (A^{x}_{a,b} -  B^{x}_{a,b})^*  (C_{a,c}^{x,y})^*C_{a,c}^{x,y}  (A^{x}_{a,b} -  B^{x}_{a,b})|\psi } \\
		&\leq  \braket{\psi| (A^{x}_{a,b} -  B^{x}_{a,b})^* (A^{x}_{a,b} -  B^{x}_{a,b})|\psi } \\
		&=	 \| \left(	A^{x}_{a,b}  - B^{x}_{a,b}   \right) \ket{\psi}\|^2  
	\end{align*}
	where the second inequality follows from $C$ being a POVM. The lemma follows accordingly.
\end{proof}

\subsection{Quantum correlations } \label{sec:introcorrelations}

In this subsection, we introduce different notions of quantum information that will be used in this paper. Given two finite sets $\cX$ and $\cA$, a (bipartite) correlation set with question set $\cX$ and answer set $\cA$ is the set $\{\text{\gls{Correlations}}\}_{(x,y) \in \cX^2, (a,b) \in \cA^2} \subseteq [0,1]^{\cX^2 \times \cA^2}$ such that 
\begin{equation*}
	\sum_{(a,b) \in \cA^2} C_{x,y,a,b} = 1
\end{equation*}
for all $(x,y) \in \cX^2$. For fixed question pair $(x,y) \in \cX^2$, $\mu(a,b) = C_{x,y,a,b}$ forms a probability distribution over $\cA^2$. We remark that a correlation set could be defined with two different question set and two different answer set, but the formulation above is equivalent by setting some of the $ C_{x,y,a,b} = 0$. In this paper, we are primarily concerned with two sets of correlations, the quantum tensor correlations and the quantum commuting correlations, which we introduce below:

\paragraph{Quantum tensor correlations. }A correlation set $\{C_{x,y,a,b}\}_{(x,y) \in \cX^2, (a,b) \in \cA^2}$ is a \textit{quantum tensor correlation}, if there exist two collections of POVM, $\{ A^x_a \}_{a \in \cA} \subseteq \bM_m(\bC)$ and $\{ B^y_b \}_{b \in \cA} \subseteq \bM_n(\bC)$, along with an entangled state $\ket{\psi} \in \bC^m \otimes \bC^n$ such that 
\begin{equation*}
	C_{x,y,a,b} = \braket{\psi|A_a^x \otimes  B_b^y|\psi}
\end{equation*}
for all $(x,y) \in \cX^2$ and $(a,b) \in \cA^2$. In this case, we refer to the set $\strategy= ( \bC^m \otimes \bC^n, \{ A^x_a \}_{a \in \cA}, \{ B^y_b \}_{b \in \cA}, \ket{\psi})$ as the \textit{tensor product strategy} (or a $*$ strategy) which realizes the correlation $C_{x,y,a,b}$. We use $C_q(\cX, \cA)$ to denote the set of quantum tensor correlations with input set $\cX$ and output set $\cA$ in this paper or simply~\gls{Correlationtensor} if $\cX$ and $\cA$ is clear from context. To abuse notation, we write $C_q$ as $C_*$ so that it is consistent with the non-local games notation. For an integer $n$, we use \gls{CorrelationtensorFD} to denote the set of quantum correlations achievable by a tensor product strategies with dimension $n$
\begin{equation*}
	C_{q}^n:= \left\{ \{\braket{\psi|A_a^x \otimes B_b^y|\psi}\}_{(x,y,a,b)} | \ket{\psi} \in \bC^n \tensor \bC^n,  \{A_a^x\}, \{B^y_b\} \text{ are POVMs in } \bM_n(\bC)    \right\}.
\end{equation*}
For $m < n$, we have $C_q^m \subseteq C_q^n$ and furthermore
\begin{equation*}
	C_q = \bigcup_{n \in \bN^{+}}^{\infty} C_q^n.
\end{equation*}

\paragraph{Quantum commuting correlations. }A correlation $\{C_{x,y,a,b}\}_{(x,y) \in \cX^2, (a,b) \in \cA^2}$ is a \textit{quantum commuting correlation} if there exist a (potentially infinite-dimensional) Hilbert space $\cH$, two sets of POVM $\{ A^x_a \}_{a \in \cA}, \{ B^y_b \}_{b \in \cA} \subseteq \BofH$ such that $[A^x_a, B^y_b] = A^x_a \cdot B^y_b - B^y_b \cdot A^x_a = 0$ for all $(x,y) \in \cX^2$ and $(a,b) \in \cA^2$, and a vector state $\ket{\psi} \in \cH$ such that 
\begin{equation*}
	C_{x,y,a,b} = \braket{\psi|A_a^xC  B_b^y|\psi}
\end{equation*}
for all $(x,y) \in \cX^2$ and $(a,b) \in \cA^2$. In this case, we refer to the set $\strategy= ( \cH, \{ A^x_a \}_{a \in \cA}, \{ B^y_b \}_{b \in \cA}, \ket{\psi})$ as the \textit{commuting operator strategy} (or a $qc$ strategy) which realizes the correlation $C_{x,y,a,b}$. We refer to a strategy (for both tensor product and commuting operator) to be a projective strategy if both the measurement operator $\{A^x_a\}$ and $\{B^y_b\}$ are PVMs. 

We use $C_{qc}(\cX, \cA)$ to denote the set of quantum commuting correlations with input set $\cX$ and output set $\cA$ and \gls{Correlationcommuting} if $\cX$ and $\cA$ are clear from context. Since any tensor product strategy is a commuting operator strategy by definition (as $[A_a^x \otimes \cI_n,\cI_m \otimes B_b^y]$), we have $C_q \subseteq C_{qc}$. As a consequence of the $\MIP^*=\RE$ theorem, we know that this inclusion is strict, or $\overline{C_q} \subsetneq C_{qc}$~\cite{jiMIPRE2022a}. 

In this paper, we work with a specific class of commuting operator strategies known as tracially embeddable strategies introduced in~\cite{linTracialEmbeddableStrategies2024}. We define this class of strategies below.
\begin{definition}[Tracially embeddable strategy, Definition 3.1 of~\cite{linTracialEmbeddableStrategies2024}] \label{def:Tracialemd}
	Let $\cX$ and $\cA$ be a finite set. A commuting operator strategy $\strategy = (\cH, \ket{\psi},\{A_a^x \}_{a \in \cA} , \{B_b^y\}_{b \in \cA})$ is called tracially embeddable if there exists a tracial von Neumann algebra $(\alicealg, \tau)$ with standard form $(\chi_{\tau}, \cL^2(\alicealg, \tau), \ket{\tau})$ and $\sigma \in \alicealg^{+}$ such that $\cH = \cL^2(\alicealg, \tau)$, $\ket{\psi} = \sigma \ket{\tau}$, $\{A_a^x \}_{a \in \cA} \subseteq \alicealg$ and $ \{B_b^y\}_{b \in \cA} \subseteq \alicealg'$. 
\end{definition}
We represent a tracially embeddable strategy as $\strategy = \text{\gls{Tracialstrategy}}$ in this paper (we write $\sigma \ket{\tau}$ in the formulation as $\ket{\psi}$ when the density matrix is not used within the proof). We remark that $(B_b^y)^{op}$ in the above formulation is actually in $\alicealg$ instead, and this is similar to writing Bob's measurement as $(B^y_b \otimes \cI)$ in the finite dimension case (even though Bob's measurement is made on the second register). A correlation $C_{x,y,a,b}$ is tracially embeddable if there exists a tracially embeddable strategy which realizes said correlation, and we use $C^{\Tr}_{qc} \subseteq C_{qc} $ to denote the set of tracially embeddable correlations. As the main result of~\cite{linTracialEmbeddableStrategies2024}, the set of tracially embeddable correlations can be used to approximate the set of quantum commuting correlations. 
\begin{theorem}[Approximation of tracially embeddable correlations,Theorem 3.2 of~\cite{linTracialEmbeddableStrategies2024}] \label{thm:tracialembedding}
	Let $\cX$ and $\cA$ be two arbitrary finite sets, then
	\begin{equation*}
		\overline{\cC_{qc}^{Tr}(\cX, \cA)} = \cC_{qc}(\cX, \cA).
	\end{equation*}
	where the closure above is in the $l_1$ norm of $[0,1]^{|\cX|^2 \cdot |\cA|^2}$. 
\end{theorem}
Intuitively, a tracially embeddable strategy is a commuting operator strategy with similar structure as a finite-dimensional, tensor product strategy. We refer to~\cite[Example 3.3]{linTracialEmbeddableStrategies2024} for more intuition on these similarities. In this paper, we primarily work with tracially embeddable strategies when considering a correlation from the commuting operator model. Due to the similarities with finite-dimensional strategies, a large portion of the proofs given in this paper follow similarly to some of the proofs given in~\cite{jiMIPRE2022a}. Audiences with no prior background in operator algebra might find the reference chart given in~\Cref{fig:tracialemdtofd} to be helpful when reading some of the proofs in this paper. 

\begin{table}[!htb]
	\centering
	\small
	\vspace{1em}
	\begin{tabular}{|c |c| c|}
	\hline
			 & Tracially embeddable   & Finite-dimensional, tensor product   \\
			 & strategies &    strategy (over $\bC^n \otimes \bC^n$) \\
		\hline
		Algebra &  $\alicealg$  & $\cM_n (\bC) \otimes \cI_n$ \\
		\hline
		Commutant & $\alicealg'$ &  $\cI_n \otimes \cM_n (\bC)$ \\
		\hline
		Hilbert space ($\cH$) & $\cL^2(\alicealg, \tau)$ & $\bC^n \otimes \bC^n$ \\
		\hline
		Tracial state & $\ket{\tau}$ & $\MEstate{n} = \frac{1}{\sqrt{n}}\sum_{i=0}^n \ket{ii}$ \\
		\hline
		Reduced density matrices & $\sigma^2$ & $\Tr_B(\ketbra{\psi}{\psi})$ \\
		\hline
		Measurement operator for prover 1 (Alice) & $A^x_a$ & $A^x_a \otimes \cI_n$ \\
		\hline
		Measurement operator for prover 2 (Bob) & $B^y_b$ & $\cI_n \otimes B^y_b$ \\
		\hline
		Observable switching trick & $A \ket{\tau} = A^{op} \ket{\tau}$ & $A \otimes \cI_n \MEstate{n} = \cI_n \otimes A^T \MEstate{n}$ \\
		\hline
	\end{tabular}
	\caption{A diagram translating components of a tracially embeddable strategy to its finite-dimensional counterpart. We assume the finite-dimensional strategy is defined over registers $A$ and $B$. }
	\label{fig:tracialemdtofd}
\end{table}
Tracially embeddable strategies also give a notion of a symmetric strategy for the commuting operator model, given by the definition below:
\begin{definition}[Symmetric strategy] \label{def:symstrategy}
	Let $( \cL^2(\alicealg, \tau), \sigma \ket{\tau},\{ A_{a}^x \},  \{ (B_{b}^y)^{op} \} )$ be a tracially embeddable strategy. We call this strategy \textit{symmetric} if $A_{a}^x = B_{a}^x$ for all $x \in \cX$ and $a \in \cA$. 
\end{definition}
In the finite-dimensional setting, a symmetric strategy is equivalent to $A_{a}^x \otimes \cI = (B_{a}^x)^{T} \otimes \cI$ for all $(x,a) \in \cX \times \cA$. Symmetric strategies will be written as $\strategy^{\text{sym}} = (\cL^2(\alicealg, \tau), \sigma \ket{\tau},\{ A_{a}^x \})$ in this paper.

\paragraph{Synchronous correlations. } In this paper, we work with a set of correlations known as synchronous correlations. For $t \in \{ *, qc\}$ and finite set $\cX$ and $\cA$, a correlation $\{ C_{x,y,a,b}\} \in C^t(\cX, \cA)$ is synchronous iff for all $x \in \cX$ and $(a,b) \in \cA^2$
\begin{equation*}
	C_{x,x,a,b} = \delta_{a,b},
\end{equation*}
and we use \gls{Correlationsyncronous} to denote the set of synchronous correlations for models $t$. Synchronous correlations were first studied in~\cite{paulsenEstimatingQuantumChromatic2016}, and have been used in~\cite{mousaviNonlocalGamesCompression2022} to study the complexity of zero gap $\MIP^*$. We call a strategy that realizes a synchronous correlation to be a \textit{synchronous strategy}. The following theorem shows that all synchronous correlations can be realized by a symmetric strategy. 
\begin{lemma}[Synchronous correlations can be realized using a symmetric strategy] \label{lem:structsync}
	Let $C_{x,y,a,b} \in C_{qc}^s$, then there exists a projective, symmetric strategy $\strategy = (\cL^2(\alicealg, \tau), \ket{\tau}, \{A_a^x\})$ which realizes $C_{x,y,a,b}$. Furthermore, if $C_{x,y,a,b} \in C_q^s$, then $\strategy$ is finite-dimensional.
\end{lemma}
In the above lemma, the state used for the $\strategy$ is precisely the GNS to the tracial state to the algebra $\alicealg$, or in the finite dimension case, the maximally entangled state $\MEstate{n}$.The above statement can be proven by first taking the double commutant of $\{ A^x_a\}$ from~\cite[Theorem 5.5]{paulsenEstimatingQuantumChromatic2016}, then applying point iii) of~\cite[Theorem 5.5]{paulsenEstimatingQuantumChromatic2016} to get the desired result. In this paper, we assume all synchronous correlations are realized using synchronous strategies guaranteed by~\Cref{lem:structsync}. Since the synchronous strategy guaranteed in~\Cref{lem:structsync} can be represented by a symmetric strategy, we denote all synchronous strategies in this paper as the projective strategy $\strategy = (\cL^2(\alicealg, \tau), \ket{\tau}, \{A_a^x \})$. 

In this paper, we also consider correlations which are approximately synchronous. Given a distribution $\mu \sim \cX^2$, we denote the \textit{synchronicity} of the correlation set with respect to $\mu$ as
\begin{equation} \label{eq:syncgame}
	\text{\gls{Synchronicity}} := \max\{\Ex_{x \sim \mu_x} \sum_{a \neq b} C_{x,x,a,b} , \Ex_{y \sim \mu_y} \sum_{a \neq b} C_{y,y,a,b}\}, 
\end{equation}
where $\mu_x$ (resp. $\mu_y$) denotes the marginal distribution of $x$  (resp. $y$) over $\mu$. For a quantum stategy $\strategy$, we use $\deltasync(\mu, \strategy)$ to denote the synchronicity for the correlation generated by $\strategy$ with respect to $\mu$. For a tensor product/commuting operator strategy $\strategy = \strategy= ( \cH, \{ A^x_a \}_{a \in \cA}, \{ B^y_b \}_{b \in \cA}, \ket{\psi})$, by definition we have 
\begin{equation} \label{eq:distancemeasuresyncstrategy}
 A^x_a \simeq_{\deltasync(\mu, \strategy)} B^x_a
\end{equation}
where $\simeq_{\deltasync(\mu, \strategy)}$ is over the state $\ket{\psi}$ and both the distribution $\mu_x$ and $\mu_y$. We define a correlation $C$ to be \textit{$\delta$-synchronous} with respect to $\mu$ if $\deltasync(\mu, C)  \leq \delta$ and $\delta$-synchronous if the underlying distribution $\mu$ is clear from context. We recall the following lemma which states that all approximately synchronous correlations can be approximated by a correlation realized by a symmetric strategy. 

\begin{corollary}[Corollary A.7 of~\cite{linTracialEmbeddableStrategies2024}] \label{cor:orthogonalizationlemmasync}
	Let $C_{x,y,a,b}$ be a $\delta$-synchronous correlation with respect to some distribution $\mu$, and let $(\cL^2(\alicealg, \tau),\sigma \ket{\tau}, \{A_a^x \} , \{B_b^y\})$ be a strategy which realizes the correlation  $C_{x,y,a,b}$. Then there exists a symmetric, projective, and $\delta^{\frac{1}{4}}$-synchronous strategy $( \cL^2(\alicealg, \tau), \sigma \ket{\tau},\{ P_{a}^x \})$ with $A_a^x \approx_{ O(\delta)}  P^x_a$ over the distribution $\mu_x$, the marginal distribution of $\mu$ on the first variable, and over the state $\sigma \ket{\tau}$.Moreover,
	\begin{equation}\label{eq:orthogonalizationlemmasynceq}
		\Ex_{(x,y) \sim \mu} \sum_{a \in A} |\braket{\tau| \sigma A_a^x (B_b^y)^{op} \sigma |\tau} -\braket{\tau| \sigma P_a^x (P_b^y)^{op} \sigma |\tau}| \leq O(\delta^{\frac{1}{4}}),
	\end{equation} 
\end{corollary}
As shown in the theorem below, the set of $\delta$-synchronous can always be approximated by the set of synchronous correlations. 

\begin{theorem}[Rounding for synchronous correlations] \label{thm:Rounding}
	There exist a universal polynomial $s^{\text{Rd}}: [0,1] \rightarrow [0,1]$ such that $ \textbf{s}^{\text{Rounding}}(\delta) =  O(\delta^{\frac{1}{8}})$ such that that the following holds:
	Let $\mu \sim \cX^2$ be a distribution and $t \in \{q,qc\}$, and let $\{C_{x,y,a,b}\} \in \cC_{t}(\cX, \cA)$ be a $\delta$-synchronous correlation. Then there exist a collection of synchronous correlations $C_{x,y,a,b}^{s}  \subseteq \cC_{t}^{s} (\cX, \cA)$ such that
	\begin{equation*}
		\Ex_{(x,y) \sim \mu} \sum_{a,b}|C_{x,y,a,b} - C_{x,y,a,b}^{s} |  \leq \textbf{s}^{\text{Rounding}}(\delta).
	\end{equation*}
\end{theorem}
The rounding theorem is proven in the tensor model in~\cite[Corollary 3.3]{vidickAlmostSynchronousQuantum2022}, and in the commuting operator model independently in~\cite[Theorem 4.1]{linTracialEmbeddableStrategies2024} and~\cite[Theorem 2.1]{delasalleAlmostSynchronousCorrelations2023}. Note in the original formulations for all the reference above, $C_{x,y,a,b}^{s}$ is define as a convex combination of synchronous correlations. The theorem above follows because any convex combination of synchronous correlations are still synchronous by definition.

%

\subsection{Non-local games} \label{sec:prenonlocalgames}

A \textit{two-prover one-round (non-local) game} is described by a tuple \gls{games}$= (\cX^2, \cA^2, \mu, D)$, where $\cX$ is a finite set denoting the list of potential questions, $\cA$ is another finite set denoting the list of potential answers, $\mu$ is a distribution over $\cX^2$ which corresponds to the question distribution, and $D: \cX^2 \times \cA^2 \rightarrow \{0,1\}$ is the evaluation map. The game is played between two cooperating \textit{provers}, Alice and Bob, and a \textit{verifier}\footnote{We are adopting the notation from an interactive proof setting in this paper. In other non-local games literature, the provers might be referred to as ``players" and the verifier might be referred to as the ``referee" }. In this game, the verifier first samples a question pair $(x,y)$ according to the distribution $\mu$ and sends $x$ to Alice and $y$ to Bob. Upon receiving their questions, Alice (and resp.~Bob) must, without communicating with the other prover, respond with answers $a$ (resp.~$b$) in $\cA$ back to the referee, and the provers win if and only if $D(x,y,a,b) = 1$. Conventionally, non-local games are usually expressed with provers sharing different question and answer sets. However, by forcing the probability distribution $\mu$ to be zero on certain question pair, the formulations we give are equivalent to the conventional formulation. In this paper, we consider non-local games where the provers have access to the two models of entanglement described in the previous subsection, which gives two different \textit{values}, or the optimal success probability for the provers, for a given game $\cG$. We introduce these notions in the remainder of this subsection: 

\paragraph{Quantum value of a game. }  Under the quantum tensor product model, the provers first prepare a joint entangled quantum state $\ket{\psi} \in \bC^n \otimes \bC^m$ between them. After receiving the question from the verifier, the provers then perform localized measurements on their respective register based on the question they receive. The provers then respond to the verifier with the measurement output as their answers. In this case, the behaviour of the provers precisely describes a tensor product strategy defined in the previous subsection, and the probability of outputting the answer pair $(a,b)$ given the question pair $(x,y)$ is $C_{x,y,a,b}$, where $\{C_{x,y,a,b}\}_{x,y,a,b}$ is the correlation generated by the said strategy. 

Given a quantum tensor correlation $C = \{C_{x,y,a,b}\} \in C_q(\cX, \cA)$, we define the success rate for the correlation, or the \textit{value} of the correlation to be 
\begin{equation}\label{eq:quantumvaluecorrelation} 
	\text{\gls{Corrgamevalue}} := \sum_{(x,y) \in \cX^2} \mu(x,y) \sum_{(a,b) \in \cA} D(x,y,a,b) C_{x,y,a,b}. 
\end{equation}
Similarly, for a tensor product strategy $\strategy$, we use $\omega(\cG, \strategy)$ to denote the value of the correlation realized by the strategy $\strategy$. The optimal success rate given quantum tensor model of entanglement, or the \textit{tensor product value} of a game $\cG$ is
\begin{equation}
	\text{\gls{tensorvalue}} := \sup_{\{C_{x,y,a,b} \in C_q \}} \omega(\cG, C).
\end{equation}
We remark that since $C_q$ is not a closed set~\cite{slofstraSetQuantumCorrelation2019}, the supremum in the above equation might not be realizable by a tensor product correlation. 

Similarly, if the provers are allowed to use the commuting model of entanglement. The set up is similar as above, except the provers use commuting operator strategies instead. Given a quantum commuting operator correlation $C = \{C_{x,y,a,b}\} \in C_q(\cX, \cA)$, the value for the correlation is the same as~\eqref{eq:quantumvaluecorrelation}. The optimal success rate given quantum commuting model of entanglement, or the \textit{commuting operator value} of a game $\cG$ is
\begin{equation}
	\text{\gls{commutingvalue}} := \sup_{\{C_{x,y,a,b} \in C_{qc} \}} \omega(\cG, C).
\end{equation}
Due to~\Cref{thm:tracialembedding}, the $C_{qc}$ in the above equation can be replaced with $C_{qc}^{\Tr}$. Since $C_q \subsetneq C_{qc}$, $\omega^*(\cG) \leq \omega^{co}(\cG)$ for all games $\cG$. For model $t \in \{*,co\}$, We call a strategy $\strategy$ in model $t$ a perfect strategy for game $\cG$ if $\omega(\cG, \strategy) = 1$. 

While the value of a game can be defined in terms of either quantum correlations or quantum strategies, there is a distinction between correlations and strategies. From the verifier's point of view, he can only ``detect" the correlation by sampling a pair of questions and getting a pair of response from the provers. However, the provers can choose different strategies (which the verifier cannot detect) in order to realize the same correlation. 

\paragraph{Synchronous games. }Given a game $\cG$, we call a game \textit{synchronous} iff $D(x,x, a,b) = \delta_{a,b}$. In other words, the provers must provide the same answer pair when given the same question pair. For a synchronous game $\cG$, we call a question pair $(x,y)$ to be synchronous iff $x = y$. For model $t \in \{*, qc\}$ and a synchronous game $\cG$, we define the synchronous value for $\cG$ in model $t$ to be
\begin{equation*}
	\text{\gls{syncvalue}}:= \sup_{\{C_{x,y,a,b} \in C_t^s \}} \omega(\cG, C).
\end{equation*}
Intuitively, a synchronous strategy corresponds to the set of strategies in which the provers always give the same answer when given the same questions. Since $C_{q}^s \subseteq C_{q}$ (resp. $C_{qc}^s \subseteq C_{qc}$), we have $\omega^{*}_s(\cG) \leq \omega^{*} (\cG)$  (resp. $\omega^{co}_s(\cG) \leq \omega^{co} (\cG)$). However, as seen by the following theorem, these two values are equivalent whenever $\cG$ admits a perfect strategy. 
\begin{theorem}[Perfect quantum value implies perfect synchronous value, Theorem 3.2 of~\cite{mousaviNonlocalGamesCompression2022}] \label{thm:syncvaluepreserve}
	Let $\cG = (\cX, \cA, \mu, D)$ be a synchronous game such that $\mu(x,x) > 0$ for all $x \in \cX$. For model $t \in \{0,1\}$, $\omega^{t}(\cG) = 1 \rightarrow \omega^{t}_s(\cG) =1$.
\end{theorem}
We call a synchronous game $\cG$ $c$-balanced if there exists some constant $c \in [0,1]$ such that $c \cdot \mu_x(x) \leq \mu(x,x)$ and $c \cdot \mu_y(x) \leq \mu(x,x)$ for all $x \in \cX$. In other words, the synchronous question pair will always appear with at least probability $c$ on the marginal distribution. Based on the definition of $\delta$-synchronous correlations, we have the following lemma about any correlations which are near perfect for any balanced game. 
\begin{lemma}[Almost perfect correlation implies almost synchronous correlation ] \label{lem:balancedperfecttoalmostsync}
	Let $\cG = (\cX, \cA, \mu, D)$ be a $c$-balance synchronous game, and for model $t \in \{*, qc\}$,  let $\strategy= (\cL^2(\alicealg, \tau), \ket{\tau}, \\ \{A_a^x \}, \{(B_b^y)^{op}\})$ be a tracially embeddable strategy such that $\omega(\cG, \strategy) \geq 1 - \eps$. Then $\deltasync(\mu, C) \leq \frac{\eps}{c}$ and there exists a symmetric and projective strategy $\strategy^{\text{sym}} = ( \cL^2(\alicealg, \tau), \sigma \ket{\tau},\{ P_{a}^x \})$ defined on the same Hilbert space as $\strategy$ such that 
	\begin{equation*}
		\omega(\cG, \strategy^{\text{sym}}) \geq 1 - \eps - \left(\frac{\eps}{c} \right)^{\frac{1}{4}}. 
	\end{equation*}
\end{lemma}
\begin{proof}
	For any correlation $C$, $\cG$ being synchronous and $\omega(\cG, \strategy) \geq 1 - \eps$ implies
	\begin{equation*}
		\sum_{x \in \cX} \mu(x,x)  \sum_{a \neq b} \braket{\tau|\sigma A_a^x B_b^x \sigma| \tau} \leq  \eps.
	\end{equation*}
	By the $c$-balanced condition,
	\begin{align*}
		c\cdot \left(\Ex_{x \sim \mu_x} \sum_{a \neq b}\braket{\tau|\sigma A_a^x B_b^x \sigma| \tau}\right)	&\leq \sum_{x, a\neq b} \mu(x,x) \braket{\tau|\sigma A_a^x B_b^x \sigma| \tau} \leq \eps \\
		c\cdot \left(\Ex_{x \sim \mu_y} \sum_{a \neq b}\braket{\tau|\sigma A_a^x B_b^x \sigma| \tau}\right)	&\leq \sum_{x, a\neq b} \mu(x,x) \braket{\tau|\sigma A_a^x B_b^x \sigma| \tau} \leq \eps.
	\end{align*}
	This shows that $ \deltasync(\mu, \strategy)\leq \frac{\eps}{c}$. For the second part of the lemma, by~\Cref{cor:orthogonalizationlemmasync}, there exist a symmetric strategy $\strategy^{\text{sym}} = ( \cL^2(\alicealg, \tau), \sigma \ket{\tau},\{ P_{a}^x \})$ such that 
	\begin{equation}
		\Ex_{(x,y) \sim \mu} \sum_{a \in A} |\braket{\tau| \sigma A_a^x (B_b^y)^{op} \sigma |\tau} -\braket{\tau| \sigma P_a^x (P_b^y)^{op} \sigma |\tau}| \leq O(\left(\frac{\eps}{c} \right)^{\frac{1}{4}}).
	\end{equation} 	
	By the triangle inequality, $|\omega(\cG, \strategy^{\text{sym}}) -  \omega(\cG, \strategy)| \leq O( \left(\frac{\eps}{c} \right)^{\frac{1}{4}})$, and hence the lemma follows accordingly. 
\end{proof}

For a synchronous game $\cG$, we introduce the notion of an \textit{oracularizable strategy} for tracially embeddable strategy. This class of strategies plays an important part in the oracularization within the answer reduction procedure (see~\Cref{sec:PCPanswerreduction} for more details). We remark that this set of strategy follows an analogue of the ``commuting" property, a property for finite-dimensional strategy, given in~\cite[Definition 5.8]{jiMIPRE2022a}.

\begin{definition}[Oracularizable strategy] \label{def:Oracularizablestrategies}
	A tracially embeddable strategy $\strategy = ( \cL^2(\alicealg, \tau), \sigma \ket{\tau}, \\ \{ A_{a}^x \},  \{ (B_{b}^y)^{op} \} )$ for a synchronous game $\cG = (\cX, \cA, \mu, D)$ is oracularizable if whenever $\mu(x,y) >0$, then for all $(a,b) \in \cA^2$
	\begin{equation*}
		[A_{a}^x, B_{b}^y] = A_{a}^x B_{b}^y - B_{b}^y A_{a}^x = 0. 
	\end{equation*}
\end{definition}

We remark that in the above theorem, both $\{ A_{a}^x \}$ and $\{ (B_{b}^y) \}$ are defined in $\alicealg$. Hence the above condition does not follow immediately from the definition for a commuting operator strategy. We remark that the notion of an Oracularizable strategy used in this paper is different from the one defined in~\cite[Definition 2.14]{mousaviNonlocalGamesCompression2022}, and we discuss more about the difference between the two notions in~\Cref{sec:Oracularization}. In this paper, we also assume that the verifier is computationally bounded, and we give the formulation for a computationally bounded verifier in~\Cref{sec:MIPcompressibility}.

\paragraph{Parallel repetition. } In this paper, we also consider a transformation of a game known as parallel repetition. Given a non-local game $\cG = (\cX, \cA, \mu, D)$ and $r \in \bN$, we define the $r$-fold parallel repetition to be the game $\text{\gls{parallelrepetitiongame}} = (\cX^r, \cA^r, \mu^r, D^r)$ as the game with the following question distribution and validation function
\begin{itemize}
	\item $\mu^n((x_0, \cdots, x_{r-1}), (y_0, \cdots y_{r-1})) = \prod_{i = 0}^{r} \mu(x_i, y_i)$.
	\item  $\D^r\left((x_0, \cdots, x_{r-1}), (y_0, \cdots y_{r-1}), (a_0, \cdots , a_{r-1}),(b_0, \cdots , b_{r-1})\right) =\prod_{i = 0}^{r} D(x_i,. y_i, a_i,b_i)$
\end{itemize}
Intuitively, the above transformation corresponds to the verifier sampling $r$ pairs of questions $(x_i, y_i), i \in [r]$ from the distribution $\mu$, and sends them to the provers. The provers must respond with answers $(a_i,b_i) \in \cA^2$ for $i \in [r]$, and the provers win iff $D(x_i,y_i,a_i,b_i) = 1$ for all $i \in [r]$. If a game $\cG$ is synchronous, then its $r$-fold parallel repetition $\cG^{\otimes r}$ is also synchronous. For clarity of notation, we use \gls{parallelsquestion} (resp. \gls{parallelszanswer}) to emphasize that the question/answer pairs (resp. $(\vec{a}, \vec{b})$) come from the parallel-repeated game. Parallel repetition plays an important part in the final step in proving the compression theorem, and we refer the reader to~\Cref{sec:parallelrep} and~\Cref{sec:parallelrepappendix} for more details. 

\begin{CJK*}{UTF8}{gbsn}


\section{Compression condition and the compression theorem} \label{sec:Compression}

In this section, we introduce the compression condition for general decision problems, and show the equivalence between a compressible language and $\RE$/$\coRE$. Recall from the preliminaries that given a decision problem $\Decisionproblem = (\{\Language_{\text{yes}}^{\Decisionproblem}, \Language_{\text{no}}^{\Decisionproblem}\})$, a uniform instance of $\Decisionproblem$ is a Turing machine $\texttt{Seq}_{\Decisionproblem}: \bN \rightarrow \Decisionproblem$. We introduce the notion of a compressible decision problem below.
\begin{definition}[Compressible problems]\label{def:compressibleproblem}
	Let $\Decisionproblem = (\{\Language_{\text{yes}}^{\Decisionproblem}, \Language_{\text{no}}^{\Decisionproblem}\})$ be a decision problem. We say that the decision problem $\Decisionproblem$ is compressible if there exists an algorithm $\texttt{Compress}^{\Decisionproblem}$, which takes, as input, $\langle \texttt{Seq}_{\Decisionproblem}\rangle$, a description of a uniform instance of $\Decisionproblem$. $\texttt{Compress}^{\Decisionproblem}$ outputs a description of a uniform instance of $\Decisionproblem$, $\langle\texttt{Seq}^{\text{Comp}}_{\Decisionproblem}\rangle$, such that the following holds:
	\begin{itemize}
		\item (Runtime): $\TIME_{\texttt{Compress}^{\Decisionproblem}} = O(\poly(|\langle\texttt{Seq}_{\Decisionproblem}\rangle|))$. 
		\item (Consistency of the output) $\texttt{Seq}^{\text{Comp}}_{\Decisionproblem}(n) \in \Decisionproblem$ for all $n \in \bN$, even if the initial input for $\texttt{Compress}^{\Decisionproblem}$ is not a valid uniform instance of $\Decisionproblem$. 
		\item (Complexity bound for the output) $\TIME_{\texttt{Seq}^{\text{Comp}}_{\Decisionproblem}} = O(\polylog(n))$. 
	\end{itemize}
	Furthermore, if $\TIME_{\texttt{Seq}_{\Decisionproblem}} = O(\poly(n))$, then for all $n \in \bN$, the following holds:
	\begin{itemize}
		\item (Completeness) $\texttt{Seq}_{\Decisionproblem}(n) \in\Language_{\text{yes}}^{\Decisionproblem} \Longrightarrow \texttt{Seq}_{\Decisionproblem}^{\text{Comp}}(n) \in\Language_{\text{yes}}^{\Decisionproblem}$.
		\item (Soundness) $\texttt{Seq}_{\Decisionproblem}(n) \in\Language_{\text{no}}^{\Decisionproblem} \Longrightarrow \texttt{Seq}_{\Decisionproblem}^{\text{Comp}}(n) \in\Language_{\text{no}}^{\Decisionproblem}$.
	\end{itemize}
\end{definition}
This generalizes the compressible property introduced in~\Cref{sec:introcompression} to decision problems. Similar to the remark given in~\Cref{sec:introcompression}, one should interpret the algorithm $\texttt{Compress}^{\Decisionproblem}$ given in the above as a property associated with the decision problem rather than the uniform sequence itself. This means that $\texttt{Compress}^{\Decisionproblem}$ is a single algorithm which works \textbf{for all} uniform instances of $\texttt{Seq}_{\Decisionproblem}$. Intuitively, the compressible property allows one to generate a problem instance more efficiently (even though the new problem instance might be equally hard to decide). If $\Decisionproblem$ is compressible, then $\class{co}\Decisionproblem$ is also compressible by definition. 

We use the following clever example given in~\cite{nezhadiRecursiveCompressionMethod2025} to show that the halting problem is compressible. 

\begin{example}[The Halting problem is compressible] \label{exam:haltcompressible}
	For the Halting problem, consider the compression algorithm $\texttt{Compress}^{\texttt{HALT}}$ for the halting problem defined by the following: given the input  $\langle \texttt{Seq}_{\texttt{HALT}} \rangle$, $\texttt{Compress}^{\texttt{HALT}}$ returns a description of $\langle \texttt{Seq}_{\texttt{HALT}}^{\text{Comp}}\rangle$, described by~\Cref{pseu:compressionhalt1}
	
	\vspace{10pt}
	\IncMargin{1em}
	\begin{algorithm}[H]
		\DontPrintSemicolon
		\textbf{Input}: Integer $n$
		
		Compute and return the description of $\langle \texttt{Prog}_{n}^{\langle\texttt{Seq}_{\texttt{HALT}} \rangle } \rangle$, where the pseudocode for $\texttt{Prog}_{n}^{\langle\texttt{Seq}_{\texttt{HALT}} \rangle }$ is given in~\Cref{pseu:compressionhalt2}
	
		\caption{The description of $\texttt{Seq}_{\texttt{HALT}}^{\text{Comp}}$.}
		\label{pseu:compressionhalt1}
	\end{algorithm}\DecMargin{1em}
	\vspace{10pt} 
	
	\vspace{10pt}
	\IncMargin{1em}
	\begin{algorithm}[H]
		\DontPrintSemicolon
		Compute the description of $\langle \texttt{Seq}_{\texttt{HALT}}(n) \rangle$ using $\langle \texttt{Seq}_{\texttt{HALT}} \rangle$ and $n$, halt and return ERROR if $\langle \texttt{Seq}_{\texttt{HALT}} \rangle$ is not a valid description for a Turing machine. 
		
		Run the program $\langle \texttt{Seq}_{\texttt{HALT}}(n) \rangle$, halt and return ERROR if $\langle \texttt{Seq}_{\texttt{HALT}}(n) \rangle$ is not a valid description for a Turing machine. 
		
		\caption{The description of the output for $\texttt{Prog}_{n}^{\langle\texttt{Seq}_{\texttt{HALT}} \rangle }$. We remark that the above program depends on both $\texttt{Seq}_{\texttt{HALT}}$ and $n$}
		\label{pseu:compressionhalt2}
	\end{algorithm}\DecMargin{1em}
	\vspace{10pt} 
	
	In the above example, one should interpret the description of $\langle \texttt{Prog}_{n}^{\langle\texttt{Seq}_{\texttt{HALT}} \rangle } \rangle$ given in~\Cref{pseu:compressionhalt2} as an instance for the halting problem (and hence, the runtime of line 1 from~\Cref{pseu:compressionhalt2} is irrelevant). The program $\texttt{Seq}_{\texttt{HALT}}^{\text{Comp}}$ merely generates different instances of the halting problem by outputting the description of $\langle \texttt{Prog}_{n}^{\langle\texttt{Seq}_{\texttt{HALT}} \rangle } \rangle$. 
	
	Since the Turing machine $\texttt{Seq}_{\texttt{HALT}}^{\text{Comp}}(n)$ only returns the description for $\langle \texttt{Prog}_{n}^{\langle\texttt{Seq}_{\texttt{HALT}} \rangle } \rangle$ which depends on $n$, the description $\langle \texttt{Seq}_{\texttt{HALT}}^{\text{Comp}}(n) \rangle$ can be computed in $O(\poly(|\langle \texttt{Seq} \rangle|))$ time. Since any integer input is represented under the binary representation, $\TIME_{\texttt{Seq}^{\text{Comp}}_{\Decisionproblem}} = O(\polylog(n))$. Furthermore, for all integer $n \in \bN$ 
	\begin{itemize}
		\item If $\texttt{Seq}_{\texttt{HALT}}(n)$ returns a description of a halting program, then $\texttt{Seq}_{\texttt{HALT}}^{\text{comp}}(n) = \langle \texttt{Prog}_{n}^{\langle\texttt{Seq}_{\texttt{HALT}} \rangle } \rangle$ is a description of a halting program. 
		\item Otherwise, if $\texttt{Seq}_{\texttt{HALT}}(n)$ returns a description of a non-halting program, then $\texttt{Seq}_{\texttt{HALT}}^{\text{comp}}(n) = \langle \texttt{Prog}_{n}^{\langle\texttt{Seq}_{\texttt{HALT}} \rangle } \rangle$ is a description of a non-halting program.
	\end{itemize}
	This shows that the halting problem is an example of a compressible decision problem. 
\end{example}

The above example also shows that all $\RE$/$\coRE$-complete problems are also compressible. Although the compressible property offers a novel characterization for an $\RE$/$\coRE$-complete decision problem, it is often hard to construct the $\texttt{Compress}$ algorithm and make this characterization practical. To this end, we give a weaker notion of compressibility below. 

\begin{definition}[Weakly compressible problems]  \label{def:weaklycompressibleproblem}
	Let $\Decisionproblem = (\{\Language_{\text{yes}}^{\Decisionproblem}, \Language_{\text{no}}^{\Decisionproblem}\})$ be a decision problem. We say that the decision problem $\Decisionproblem$ is weakly compressible if for every $\alpha \in \bN $, there exists an algorithm $\texttt{Compress}_{\alpha}^{\Decisionproblem}$, which takes, as input, $\langle \texttt{Seq}_{\Decisionproblem} \rangle$, a description of a uniform instance of decision problems for $\Decisionproblem$. $\texttt{Compress}_{\alpha}^{\Decisionproblem}$ outputs a description for a uniform instance of decision problems for $\Decisionproblem$, $\langle\texttt{Seq}^{\text{Comp}}_{\Decisionproblem}\rangle$, such that the following holds:
	There exists an integer $\gamma= O(\poly(\alpha))$ such that
	\begin{enumerate}
		\item (Runtime): $\TIME_{\texttt{Compress}_{\alpha}^{\Decisionproblem}} = O(\poly(|\langle\texttt{Seq}_{\Decisionproblem, \alpha }\rangle|, \alpha))$. 
		\item (Consistency of the output) $\texttt{Seq}^{\text{Comp}}_{\Decisionproblem}(n) \in \Decisionproblem$ for all $n \in \bN$, even if the initial input for $\texttt{Compress}_{\alpha}^{\Decisionproblem}$ is not a valid uniform instance of $\Decisionproblem$. 
		\item (Complexity bound for the output) $\TIME_{\texttt{Seq}^{\text{Comp}}_{\Decisionproblem}} = O(\polylog(n)^{\gamma})$.
	\end{enumerate}
	Furthermore, if there exists some constant $n_0 \in \bN$ such that for all $n \geq n_0$
	\begin{equation} 
		\TIME_{\texttt{Seq}_{\Decisionproblem}} \leq n^\alpha. 
	\end{equation}
	Then there exist some constant $n_0^{\Compressgame} = \poly(\gamma, n_0)$ such that for all $n \geq n_0^{\Compressgame}$
	\begin{itemize}
		\item (Completeness) $\texttt{Seq}_{\Decisionproblem}(n) \in\Language_{\text{yes}}^{\Decisionproblem} \Longrightarrow \texttt{Seq}_{\Decisionproblem}^{\text{Comp}}(n) \in\Language_{\text{yes}}^{\Decisionproblem}$.
		\item (Soundness) $\texttt{Seq}_{\Decisionproblem}(n) \in\Language_{\text{no}}^{\Decisionproblem} \Longrightarrow \texttt{Seq}_{\Decisionproblem}^{\text{Comp}}(n) \in\Language_{\text{no}}^{\Decisionproblem}$.
	\end{itemize}
\end{definition}

Instead of requiring a single \texttt{Compress} algorithm which works for all uniform problem instances, the weakly compressible condition instead just requires a $\texttt{Compress}_{\alpha}$ algorithm that compresses all uniform problem instances that run in $O(n^{\alpha})$ time for all $\alpha \in \bN$ (i.e. these algorithms could be different depending on $\alpha$). If $\Decisionproblem$ is compressible, then it is trivially weakly compressible. We have the following two theorems relating the compressible condition to $\RE$/$\coRE$-complete languages.

\begin{theorem}[Compression criteria for $\RE$-complete problems] \label{thm:compressRE}
	Let $\Decisionproblem = (\{\Language_{\text{yes}}^\Decisionproblem, \Language_{\text{no}}^\Decisionproblem\})$ be a decision problem. If $\Language_{\text{no}}^\Decisionproblem \in \coRE$, then the following are equivalent:
	\begin{enumerate}
		\item $\Decisionproblem$ is $\RE$-complete. 
		\item $\Decisionproblem$ is a compressible decision problem.
		\item $\Decisionproblem$ is a weakly compressible decision problem.
	\end{enumerate}
\end{theorem}

\begin{theorem}[Compression criteria for $\coRE$-complete problems] \label{thm:compresscosRE}
		Let $\Decisionproblem = (\{\Language_{\text{yes}}^\Decisionproblem, \Language_{\text{no}}^\Decisionproblem\})$ be a decision problem. If $\Language_{\text{yes}}^\Decisionproblem \in \coRE$, then the following are equivalent:
		\begin{enumerate}
		\item $\Decisionproblem$ is $\coRE$-complete. 
		\item $\Decisionproblem$ is a compressible decision problem.
		\item $\Decisionproblem$ is a weakly compressible decision problem.
	\end{enumerate}
\end{theorem}

Since whenever $\Decisionproblem$ is compressible/weakly compressible, $\class{co}\Decisionproblem$ is also compressible/weakly compressible, this implies that~\Cref{thm:compressRE} implies~\Cref{thm:compresscosRE} being true. Hence, we provide a proof for~\Cref{thm:compressRE} below. We remark that this proof is inspired by the suggested approach for showing $\MIPco=\coRE$ in~\cite[Conjecture 1.4]{mousaviNonlocalGamesCompression2022}. 

\newcommand{\coREalgoproof}{\texttt{coREalgo}_{\Language_{\text{no}}}}
\newcommand{\UniformseqREproof}{\texttt{Seq}_{\Decisionproblem,\text{\Turingmachinehalt}}}

\begin{proof}
	Note that (2) trivially implies (3), and by~\Cref{exam:haltcompressible} (1) implies (2). Hence it remains to show (3) implies (1). This proof follows a similar structure as the one presented in~\Cref{sec:introcompression}. 
	
	Hence, fix a $\Decisionproblem = (\{\Language_{\text{yes}}^\Decisionproblem, \Language_{\text{no}}^\Decisionproblem\})$ as per~\Cref{thm:compressRE}, and assume that $\Decisionproblem$ is a weakly compressible problem. Since by assumption,  $\Language_{\text{no}}^\Decisionproblem \in \coRE$, let $\coREalgoproof$ denote the $\coRE$ algorithm that halts whenever $x \not\in \Language_{\text{no}}^\Decisionproblem$ and runs forever if $x \in \Language_{\text{no}}^\Decisionproblem$ (this can be assumed by appending an infinite loop whenever $\coREalgoproof$ halts and correctly decides $x \in \Language_{\text{no}}^\Decisionproblem$). The algorithm $\coREalgoproof$ already implies that $\Decisionproblem \in \RE$.  Hence the only thing we need to show is that $\Decisionproblem$ can be reduced to the halting problem. 
	
	Hence, fix a Turing machine \Turingmachinehalt. We wish to find an instance $x_{\text{\Turingmachinehalt}}$ such that whenever \Turingmachinehalt halts, then $x_{\text{\Turingmachinehalt}} \in \Language_{\text{yes}}^\Decisionproblem$ and $x_{\text{\Turingmachinehalt}} \in \Language_{\text{no}}^\Decisionproblem$ otherwise. For every $\beta \in \bN$, let $\texttt{Compress}_{\beta}^{\Decisionproblem}$ be the compression algorithm guaranteed by the weakly compressible condition given in~\Cref{def:weaklycompressibleproblem}. Since we assume $\Language_{\text{yes}}^\Decisionproblem$ and $\Language_{\text{no}}^\Decisionproblem$ are non-empty in the preliminary, let $x_{\text{yes}} \in \Language_{\text{yes}}^\Decisionproblem$ and $x_{\text{no}} \in \Language_{\text{no}}^\Decisionproblem$ be an arbitrary element in these two sets.
	
	Let $\alpha, C_0 \in \bN$ be two constants to be specified later in the proof. We construct the following uniform game sequence $\texttt{Seq}_{\Decisionproblem,\text{\Turingmachinehalt}}$ based on \Turingmachinehalt, as shown below. 
	
	\vspace{10pt}
	\IncMargin{1em}
	\begin{algorithm}[H]
		\DontPrintSemicolon
		
		\textbf{Input}: Integer $n$. 
		
		Run \Turingmachinehalt for $n$ steps. If \Turingmachinehalt halts in the given steps, \textbf{return} $x_{\text{yes}}$.
		
		Compute the description of $\langle \texttt{Seq}_{\Decisionproblem,\text{\Turingmachinehalt}} \rangle$.
		
		Compute $\langle \texttt{Seq}_{\Decisionproblem,\text{\Turingmachinehalt}}(C_0) \rangle$, the Turing machine which is hardcoded into computing $\texttt{Seq}_{\Decisionproblem,\text{\Turingmachinehalt}}(C_0)$.
		
		Simulate line 6-7 for $\max\{0, n - C_0\}$ steps, if line 6-7 halts in the given steps, \textbf{return} $x_{\text{no}}$.
		
		\Indp Run the Turing machine $\langle \texttt{Seq}_{\Decisionproblem,\text{\Turingmachinehalt}}(C_0) \rangle$ until  $\texttt{Seq}_{\Decisionproblem,\text{\Turingmachinehalt}}(C_0)$ is computed. 
		
		Run $\coREalgoproof$ with $\texttt{Seq}_{\Decisionproblem,\text{\Turingmachinehalt}}(C_0)$ as input. 
		
		\Indm Otherwise, apply $\texttt{Compress}_{\alpha}^{\Decisionproblem}$ with the input $\langle \texttt{Seq}_{\Decisionproblem,\text{\Turingmachinehalt}} \rangle$ to obtain the description for $\langle \texttt{Seq}_{\Decisionproblem,\text{\Turingmachinehalt}}^{\text{comp}} \rangle$.
		
		Compute and \textbf{return} $\texttt{Seq}_{\Decisionproblem,\text{\Turingmachinehalt}}^{\text{comp}}(n+1)$
		
		\caption{The description for $\texttt{Seq}_{\Decisionproblem,\text{\Turingmachinehalt}}$.}
		\label{pseu:REcompressibleproof}
	\end{algorithm}\DecMargin{1em}
	\vspace{10pt} 
		
	 We point out that the length of the source code $|\langle \UniformseqREproof \rangle|$ only depends on the Turing machine \Turingmachinehalt. Similar to the proof given in~\Cref{sec:introcompression}, in line 3, we use Kleene's Recursion Theorem to allow $\langle \UniformseqREproof \rangle$ to perform computation on its own source code. This step can be performed in $O(\poly(|\langle \UniformseqREproof \rangle|) = O(\poly(|\text{\Turingmachinehalt}|))$ time. We also point out that $\UniformseqREproof$ is a valid uniform problem instance (i.e. its output range is in $\Decisionproblem$), since it can only terminate in line 2 and 5 (or else the output is $x_{\text{yes}}$ and $x_{\text{no}}$ respectively, which are both in $\Decisionproblem$), and in line 9 (which is in $\Decisionproblem$ by point 2 of~\Cref{def:weaklycompressibleproblem}). We begin by deriving the runtime for $\UniformseqREproof$ on input $n \in \bN$ by examining~\Cref{pseu:REcompressibleproof} line by line:
	\begin{itemize}
		\item For line 2, simulating \Turingmachinehalt for n steps takes time 
		\begin{equation*}
			\poly(|\text{\Turingmachinehalt}|, n).
		\end{equation*}
		\item For line 3, by Kleene's Recursion theorem, computing the description of $\langle \UniformseqREproof \rangle$ takes time
		\begin{equation*}
			\poly(|\text{\Turingmachinehalt}|, |\langle \UniformseqREproof \rangle|, n) = \poly(|\text{\Turingmachinehalt}|, n).
		\end{equation*} 
		\item For line 4, the Turing machine $\langle \texttt{Seq}_{\Decisionproblem,\text{\Turingmachinehalt}}(C_0) \rangle$ is essentially $\langle \texttt{Seq}_{\Decisionproblem,\text{\Turingmachinehalt}} \rangle$, but replacing every instance of $n$ occurring after line 1 with $C_0$. Since $C_0$ is a constant, this takes time 
		\begin{equation*}
			O(\poly(|\langle \UniformseqREproof \rangle|, C_0)).
		\end{equation*}
		\item For line 5-7, simulating line 6-7 for $\max\{0, n - C_0\}$ steps takes time 
		\begin{equation*}
			O(\poly(|\cG_{C_0}|, C_0 ,n)) = O(\poly(n,C_0, |\text{\Turingmachinehalt}|)).
		\end{equation*} 
		\item For line 8 and 9, applying $\texttt{Compress}_{\alpha}^{\Decisionproblem}$ on the input $\langle \UniformseqREproof \rangle$ and computing $\texttt{Seq}_{\Decisionproblem,\text{\Turingmachinehalt}}^{\text{comp}}(n+1)$ takes time
		\begin{equation*}
			O(\poly(\alpha, \log(n)^{\poly{(\alpha)}}, |\langle \UniformseqREproof \rangle|))
		\end{equation*}
		by condition 1 and 3 in~\Cref{def:weaklycompressibleproblem}.
	\end{itemize}
	Although many parameters appear in the runtime analysis, the only variable which is not set to a constant is $n$! By combining the runtime analysis above, we have
	\begin{equation*}
		\TIME_{\UniformseqREproof}(n) = O(\poly(n, \log(n)^{\poly{(\alpha)}}, \alpha, |\text{\Turingmachinehalt}|, C_0 )).
	\end{equation*}
	Since $C_0$ does not depend exponentially on $\alpha$ in the above equation, there exists a choice for $\alpha, n_0 \in \bN$ such that by setting $C_0 = g(n_0, \textbf{f}(\alpha))$, where $g$ is the polynomial used to define $\gamma$ and $\textbf{f}$ is the polynomial used to define $n_0^{\Compressgame}$ in~\Cref{def:compressibleproblem}, we have
	\begin{equation*}
		\TIME_{\UniformseqREproof}(n) \leq n^{\alpha},
	\end{equation*}
	for all $n \geq n_0$. Fix $\alpha$, $C_0$ and $n_0$ as the constant which satisfies the property above. By the completeness/soundness condition of~\Cref{def:compressibleproblem}, for all $n \geq n_0$
	\begin{itemize}
		\item  $\UniformseqREproof \in\Language_{\text{yes}}^{\Decisionproblem} \Longrightarrow \UniformseqREproof^{\text{Comp}}(n) \in\Language_{\text{yes}}^{\Decisionproblem}$, 
		\item $\UniformseqREproof(n) \in\Language_{\text{no}}^{\Decisionproblem} \Longrightarrow \UniformseqREproof(n) \in\Language_{\text{no}}^{\Decisionproblem}$.
	\end{itemize}
	
	We argue that $x_{\text{\Turingmachinehalt}} = \UniformseqREproof(C_0) \in \Decisionproblem$ is the instance needed for the proof (i.e. $x_{\text{\Turingmachinehalt}} \in \Language_{\text{yes}}^\Decisionproblem$ if \Turingmachinehalt halts, and  $x_{\text{\Turingmachinehalt}} \in \Language_{\text{no}}^\Decisionproblem$ otherwise). 
	
	We first show that $\langle \UniformseqREproof \rangle$ can never terminate at line 5 of~\Cref{pseu:REcompressibleproof}. Consider $\UniformseqREproof$ with input $n < C_0$; since line 5 does not perform any operation by definition, it also can never terminate in this spot. Similarly, if \Turingmachinehalt halts in time $T$, for any input $n > T$,~\Cref{pseu:REcompressibleproof} terminates in line 2 before reaching line 5.  
	
	
	Hence, suppose for a contradiction that $\UniformseqREproof(n)$ terminates at line 5 for some input $n \geq C_0$, and suppose that \Turingmachinehalt does not halt on step $n$, and without loss of generality assume that $n$ is the smallest natural number such that $\UniformseqREproof(n)$ terminates at line 5. By definition, this means $\UniformseqREproof(n) = x_{\text{no}} \in \Language_{\text{no}}^\Decisionproblem$.
	
	Since $\langle \texttt{Seq}_{\Decisionproblem,\text{\Turingmachinehalt}}(C_0) \rangle$ is a terminating program, this implies that for $\UniformseqREproof(n)$ to terminate at line 5, $\coREalgoproof(\texttt{Seq}_{\Decisionproblem,\text{\Turingmachinehalt}}(C_0))$ also would terminate. By the definition of $\coREalgoproof$, we have $\texttt{Seq}_{\Decisionproblem,\text{\Turingmachinehalt}}(C_0) \not\in \Language_{\text{no}}^\Decisionproblem$. Since $\langle \UniformseqREproof \rangle$ is a valid uniform sequence, we have $\texttt{Seq}_{\Decisionproblem,\text{\Turingmachinehalt}}(C_0) \in \Language_{\text{yes}}^\Decisionproblem$.
	
	Now consider $\UniformseqREproof(n-1)$, since $n$ is the smallest natural number such that $\UniformseqREproof(n)$ terminates at line 5, and the Turing machine \Turingmachinehalt does not halt in $n-1$ steps. $\UniformseqREproof(n-1)$, by default terminates on line 9 of~\Cref{pseu:REcompressibleproof}. This means that $\UniformseqREproof(n-1) = \UniformseqREproof(n)^{\text{comp}} \in  \Language_{\text{no}}^\Decisionproblem$ by the weakly compressible condition. By an inductive argument, since $n > C_0$, this implies that $\UniformseqREproof(C_0) \in \Language_{\text{no}}^\Decisionproblem$, contradicting the fact that $\UniformseqREproof(C_0) \in \Language_{\text{yes}}^\Decisionproblem$. Thus, $\langle \UniformseqREproof \rangle$ cannot halt on line 5 of~\Cref{pseu:REcompressibleproof}.
	
	
	We first focus on the case where \Turingmachinehalt terminates in $T$ steps. If $T \leq C_0$, then $\UniformseqREproof(T) = x_{\text{yes}} \in \Language_{\text{yes}}^\Decisionproblem$ by line 1 of~\Cref{pseu:REcompressibleproof}. Hence, assume $C_0 < T$. By line 2 of~\Cref{pseu:REcompressibleproof}, $\UniformseqREproof(T) = x_{\text{yes}} \in \Language_{\text{yes}}^\Decisionproblem$. Now, consider $\UniformseqREproof(T-1)$, since it cannot terminate at line 2 and 5 of~\Cref{pseu:REcompressibleproof}, this implies that the Turing machine $\langle \UniformseqREproof(T-1) \rangle$ will terminate on line 9 of~\Cref{pseu:REcompressibleproof}. Hence we have $\UniformseqREproof(T-1) = \UniformseqREproof^{\text{comp}}(T) \in \Language_{\text{yes}}^\Decisionproblem$ by the weakly compressible condition. Again, by an inductive argument, we have $\UniformseqREproof(T-1) \in \Language_{\text{yes}}^\Decisionproblem$, which completes the proof for the case where \Turingmachinehalt halts in finite time.
	
	Now, suppose \Turingmachinehalt does not halt. By the above claim, $\coREalgoproof$ also does not terminate when given the input $\UniformseqREproof(C_0)$ (or else $\UniformseqREproof$ terminates at line 5 of~\Cref{pseu:REcompressibleproof}, deriving a contradiction). By the definition of $\coREalgoproof$, this implies that $\UniformseqREproof(C_0) \in \Language_{\text{no}}^\Decisionproblem$. Hence showing that $\Decisionproblem$ is $\RE$-complete. This shows (3) implies (1) in~\Cref{thm:compressRE}, which completes the proof for~\Cref{thm:compressRE}. 
\end{proof}

We remark that in the above proof, since $C_0$ is a constant, computing $\UniformseqREproof(C_0)$ takes $O\left(\poly(|\text{\Turingmachinehalt}|)\right)$ time. This shows that the reduction from the Halting problem to $\Decisionproblem$ is a polynomial-time reduction. 

Interestingly, for the argument given above, independent of whether \Turingmachinehalt halts or not, we can never observe $\texttt{Seq}_{\Decisionproblem,\text{\Turingmachinehalt}}$ halting at line 5. Thus, one might be tempted to remove line 5-7 from~\Cref{pseu:REcompressibleproof} on the proof above. However, without these three lines, we cannot argue that $x_{\text{\Turingmachinehalt}} \in \Language_{\text{no}}^\Decisionproblem$ whenever \Turingmachinehalt does not halt. Interestingly, although $x_{\text{no}}$ is never actually returned, it still plays a critical role in the above argument. We highlight this as an open problem: can one remove line 5, or remove the ``no instances" (and retain line 5 of~\Cref{pseu:REcompressibleproof}) and still repeat the same argument for the proof above. 

As pointed out before, every compressible decision problem is also a weakly compressible problem. In this paper, we show the converse assuming that $\Decisionproblem \in \RE$ or $\Decisionproblem \in \coRE$. However, it is an interesting problem if the converse is true in general. The issue is that given a description of a uniform instance $\langle \texttt{Seq}_{\Decisionproblem} \rangle$, there is no algorithm which can determine the smallest $\alpha \in \bN$ such that $\TIME_{\texttt{Seq}_{\Decisionproblem}} = O(n^{\alpha})$ (or whether $\TIME_{\texttt{Seq}_{\Decisionproblem}} = O(\poly(n))$ at all). Hence, there is no clear way to construct a universal $\texttt{Compress}^{\Decisionproblem}$ algorithm given the $\texttt{Compress}^{\Decisionproblem}_{\alpha}$ algorithm. 


In the remainder of this paper, we show a specific set of $\MIP^*$/$\MIPco$ protocols, \textit{conditional linear samplable games}, form a decision problem that is weakly compressible. On a high level, the conditional linear samplable games have a specifically tailored question distribution known as \textit{conditional linear distribution} which makes defining a $\texttt{Compress}$ map possible. We give the definition for the conditional linear distribution in~\Cref{sec:condlineardistribution}, then we give a definition for a conditional linear samplable game in~\Cref{sec:MIPcompressibility}. 


\end{CJK*}
\section{Conditionally linear distribution} \label{sec:condlineardistribution}

In order to define the set of $\MIP^*$/$\MIPco$ protocols which are weakly compressible, we need to define a specific question distribution known as the \text{conditionally linear distribution}, and prove some lemmas related to it. We remark that this is the same set of distributions used in~\cite{jiMIPRE2022a} to show the $\RE$-completeness of $\MIP^*$. 

\subsection{Conditionally linear functions and conditionally linear distribution}  \label{sec:condlinearfun}
We start this subsection by first introducing the notion of the conditionally linear function, a key building block for the conditionally linear distribution.
\begin{definition}[Conditionally linear function] \label{def:condlinfun}
	Let $\seedlength \in \bN$ be an odd integer and $m, k \in \bN$, and let $V$ be a canonical basis subspace of $\bF_{2^{\seedlength}}^m$. We say that the function $\text{\gls{CondLinearfun}}: V \rightarrow V$ is a $k$-th level conditionally linear function over $V$ if $\decidertypedfunction$ can be defined using the following construction: 
	\begin{itemize}
		\item There exists a disjoint partition of subspaces $\{V_h\}_{h \in [k]}$ of $V$, where each $V_h$ is a canonical basis subspace, which we refer to as ``registers" for the function $\decidertypedfunction$.
		\item There exists a single linear function $\decidertypedfunction_{0,0}: V_{0} \rightarrow V_0$, this is referred to as the zeroth-level linear function of $\decidertypedfunction$. 
		\item For $0 < j < k$, the function $\decidertypedfunction_{j}: V_{\leq j} \rightarrow V_{\leq j}$ is defined recursively as follows:
		\begin{itemize}
			\item There exists a collection of linear functions $\{\text{\gls{CondLinearfundetail}}: V_j \rightarrow V_j\}_{s \in V_{< j}}$, which is referred to as the $j$-th level linear functions. 
			\item For every input $s \in V_{\leq j}$, write $s$ as $s = s_{j} + s_{< j}$ for $s_j \in V_j$ and $s_{< j} \in V_{< j}$, then
			\begin{equation*}
				\decidertypedfunction_{j} := \decidertypedfunction_{j-1}(s_{< j}) + \decidertypedfunction_{j,  \decidertypedfunction_{j-1}(s_{< j})}(s_j)
			\end{equation*} 
		\end{itemize}
		\item Finally, we have $\decidertypedfunction = \decidertypedfunction_{k-1}$.
	\end{itemize}     
\end{definition}

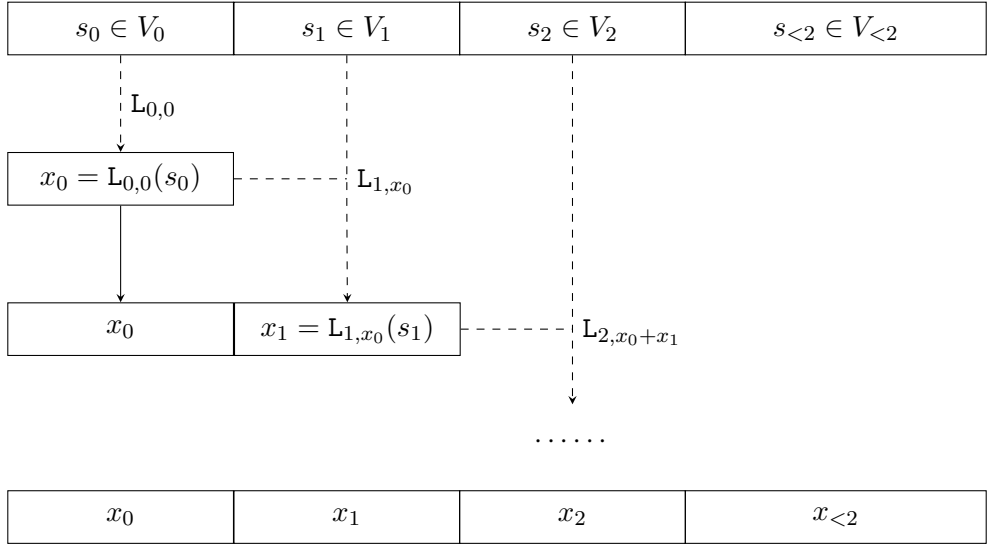
\begin{figure}[!b]
	\centering
	\begin{tikzpicture}
		\node[draw, minimum width=3cm, minimum height=0.7cm, anchor=east] at (0,0) (xv1) {$s_0 \in V_0$};
		\node[draw, minimum width=3cm, minimum height=0.7cm] at (1.5,0) (xv3) {$s_1 \in V_1$};
		\node[draw, minimum width=3cm, minimum height=0.7cm] at (4.5,0) (xv4) {$s_2 \in V_2$};
		\node[draw, minimum width=4cm, minimum height=0.7cm] at (8,0) {$s_{<2} \in V_{< 2}$};
		\node[draw, minimum width=3cm, minimum height=0.7cm, anchor=east] at (0,-2) (xl1) {$x_0 = \decidertypedfunction_{0,0}(s_0)$};

		\node[draw, minimum width=3cm, minimum height=0.7cm, anchor=east] at (0,-4) (xl4) {$x_0$};
		\node[draw, minimum width=3cm, minimum height=0.7cm, anchor=west] at (0,-4) (xv2) {$x_1 = \decidertypedfunction_{1,x_0}(s_1)$};
		
		\draw[dashed, -stealth, to path={|- (\tikztotarget)}] (xv3.south) -- ++(0,-2) node[above right] {$\decidertypedfunction_{1,x_0}$} -- (xv2.north);
		\draw[dashed,-stealth] (xv1.south) -- ++(0,-1) node[above right] {$\decidertypedfunction_{0,0}$} -- (xl1.north);
		\draw[dashed, to path={|- (\tikztotarget)}] (xl1.east) edge (1.5,-2);
		\draw[-stealth] (xl1.south) -- (xl4.north);

		\draw[dashed, -stealth, to path={|- (\tikztotarget)}] (xv4.south) -- ++(0,-4) node[above right] {$\decidertypedfunction_{2,x_0+x_1}$} -- (4.5,-5);
		\draw[dashed,to path={|- (\tikztotarget)}] (xv2.east)  -- (4.5,-4);
		
		\node[right] at (3.75,-5.5) {$\;\cdots\cdots$};
		
		\node[draw, minimum width=3cm, minimum height=0.7cm, anchor=east] at (0,-6.5)  {$x_0$};
		\node[draw, minimum width=3cm, minimum height=0.7cm] at (1.5,-6.5)  {$x_1$};
		\node[draw, minimum width=3cm, minimum height=0.7cm] at (4.5,-6.5)  {$x_2$};
		\node[draw, minimum width=4cm, minimum height=0.7cm] at (8,-6.5) {$x_{<2}$};
		
	\end{tikzpicture}
	\caption{A diagram representation of the conditionally linear function $\decidertypedfunction$. 
	}\label{fig:cl-functions}
\end{figure}

We refer to~\Cref{fig:cl-functions} for more intuition about conditionally linear functions. In this paper, we abbreviate conditionally linear as CL. In this paper, we also define CL functions $\decidertypedfunction$ with $\seedlength$ being an even integer. When this occurs, we assume that the range of $\decidertypedfunction$, $V$ , contains an additional canonical basis (i.e. $\decidertypedfunction: \bF_{2^{\seedlength+1}}^m \rightarrow \bF_{2^{\seedlength+1}}^m$), and the additional canonical basis created are merged into the register $V_0$ and lies in the kernel space of $L_{0,0}$. We give a simple example of a CL function below. 
\begin{example}
	Consider the finite field $\bF_{2}^3$ with basis $(e_0, e_1, e_2)$, and the function 
	\begin{equation*}
		\decidertypedfunction(x_0, x_1, x_2) = (x_0, x_0 \cdot x_1 + (1+x_0)\cdot x_2, 0) = (x_0, x_0 \cdot x_1 + x_1 \cdot x_2 + x_2, 0 ).
	\end{equation*}
	The function $\decidertypedfunction$ is a second level CL function with registers $V_0 = \text{span}((1,0,0))$ and $V_1 = \text{span}((0,1,0), \\ (0,0,1))$. $\decidertypedfunction$ can be defined as the following: the zeroth level linear function for $\decidertypedfunction$ can be taken as
	\begin{equation*}
		\decidertypedfunction_{0,0}(x_0,0,0) = (x_0,0,0),
	\end{equation*}
	and two first level linear functions for $\decidertypedfunction$ can be specified by
	\begin{equation*}
		\decidertypedfunction_{1,0}(0,x_1,x_2) = (0,x_1,0), \, \quad  \text{and} \quad \, 	\decidertypedfunction_{1,1}(0,x_1,x_2) = (0,x_1,0).
	\end{equation*}
\end{example}
 Since $\{0\} \subseteq V \subseteq \bF_{2^\seedlength}^m$, any $k$-th level CL function is also a $k'$-th level CL function for $k' \geq k$. Since the definition for a CL function is recursive, for each $j \in [k]$ and $s \in V_j$, one can define a $k -j$ level CL function from $V_{\geq j} \rightarrow L_{\geq j}$ by setting the initial linear function to be $\decidertypedfunction_{j, s}$. Similarly, we can combine a $k_1$-th level CL function and a collection of $k_2$-th level CL functions with the same register into a $(k_1 + k_2)$-level CL function through the series composition. We make this transformation more precise below. We remark that this is the composition defined in~\cite[Lemma 4.7]{jiMIPRE2022a}.

\begin{definition}[Series composition of CL functions] \label{defn:seriescomposCLfn}
		Let $V\subseteq \bF_{2^\seedlength}^m$ be a canonical basis subspace and let $V = V^1 \oplus V^2$ be a disjoint partition where both $V^1$ and $V^2$ are canonical basis subspaces. Let $\mathtt{M}: V^1 \rightarrow V^1$ be a $k_1$-th level CL function with registers $\{V_j^1\}_{j \in [k_1]}$, and let $\{ \mathtt{N}^x:  V^2 \rightarrow  V^2 \}_{x \in V^1}$ be a collection of $k_2$-th level CL functions which share the same registers $\{V^2_j\}_{j \in [k_2]}$.	Define the \textit{series composition} between $\mathtt{M}$ and $\{ \mathtt{N}^x \}_{x \in V^1}$ to be a  $(k_1 + k_2)$-level CL function $\decidertypedfunction: V \rightarrow V$ , defined as follows:
		\begin{itemize}
			\item The function has registers $\{V_0, \cdots V_{n+m} \}$, where $V_j = V^1_j$ for  $0 \leq j < k_1$ and $V_j = V^2_{j-k_1}$ for $k_1 \leq j < k_1+k_2$. 
			\item For $0 \leq j < k_1$ and for each $h \in V_{<j}$, we have 
			\begin{equation*}
				\decidertypedfunction_{j, h} = \mathtt{M}_{j, h}
			\end{equation*}
			\item For $k_1 \leq j < k_1+k_2$, we define 
			\begin{equation*}
				\decidertypedfunction_{j, h_v + h_w} = \mathtt{N}_{j-m, h_w}^{h_v},
			\end{equation*}
			where $h_v \in V$ and $h_w \in W$.
		\end{itemize}
\end{definition} 

In other words,  for every input $s \in V$, write $s = v_1+v_2$, where $v_1 \in V^1$ and $v_2 \in V^2$. $\decidertypedfunction$ first applies the $k_1$-th level CL function $\mathtt{M}$ to $v_1$, the $V^1$ component for $V$. Based on the output of $\mathtt{M}(v_1)$, $\decidertypedfunction$ will then apply the $k_2$-th level CL function $\mathtt{N}^{\mathtt{M}(v_1)}$ to $v_2$, the $V^2$ component within $V$. We can also combine two CL functions in parallel as described below, we remark that this is the CL function constructed in~\cite[Lemma 4.9]{jiMIPRE2022a}.
\begin{definition}[Parallel composition of CL functions] \label{defn:parallelcomposCLfn}
		Let $V \subseteq \bF_{2^{\seedlength}}^m$ be a canonical basis subspace and let $V = V^1 \oplus V^2$  be a disjoint partition where both $V^1$ and $V^2$ are canonical basis subspaces. Let $\mathtt{M}: V^1 \rightarrow V^1$ and $\mathtt{N}: V^2 \rightarrow V^2$ be a  $k$-th level CL function with registers $\{V_j^1\}_{j \in [k]} \subseteq V^1$ and  $\{V_j^2\}_{j \in [k]} \subseteq V^2$ respectively. Define the \textit{parallel composition} between $\mathtt{M}$ and $\mathtt{N}$ to be a level $k$ CL function $\decidertypedfunction: V \rightarrow V$, defined as follows:
	\begin{itemize}
		\item The function has registers $\{V_j = V_j^1 \oplus V_j^2\}$. 
		\item For every $j \in [k]$ and $h_{< j} \in V_{< j}$, write $h_{< j} = h_{< j}^1 + h_{<j}^2$ for $h_{< j}^1 \in V_{< j}^{1}$ and $h_{< j}^2 \in V_{< j}^2$ (resp. $h_j^2 \in V_j^2$). Define $\decidertypedfunction_{j, h_{< j}}: V_j \rightarrow V_j$ to be
		\begin{equation*}
			\decidertypedfunction_{j, h_{< j}} = \mathtt{M}_{j, h_{< j}^1} +  \mathtt{N}_{j, h_{< j}^2} 
		\end{equation*}
	\end{itemize}
\end{definition} 

We introduce the notion of a CL distribution below. In this paper, we consider mainly games with a CL distribution as the question distribution. 
\begin{definition}[Conditionally linear distribution] \label{def:CLdistribution}
	A distribution $\mu$ is a $(k, m, \seedlength)$ CL distribution if there exist two $k$-th level CL function $\text{\gls{CondLineardist}}: \bF_{2^{\seedlength}}^m \rightarrow \bF_{2^{\seedlength}}^m$, $P \in \{A,B\}$ with the same register $\{V_j\}_{j\in [k]}$ such that $\mu$ can be sampled in the following way.
	\begin{enumerate}
		\item Uniformly sample $s \in \bF_{2^{\seedlength}}^m$. 
		\item Give $\decidertypedfunction^A(s)$ to Alice, and $\decidertypedfunction^B(s)$ to Bob.
	\end{enumerate}
	Given a sample $(x,y) \sim \mu$,  we refer to the $s \in \bF_{2^{\seedlength}}^m$ as the ``seed" for the given sample if $(x,y) = (\decidertypedfunction^{A}(s), \decidertypedfunction^{B}(s))$.
\end{definition}
We refer to a non-local game $\cG= (\cX, \cA, \mu, D)$ as CL samplable if the sampling procedure is a CL distribution, and as a $k$-th level CL samplable game if furthermore the question distribution is a $k$-th level CL distribution. 

For any game $\cG$, if $\mu$ is a $(k, m, \seedlength)$ CL distribution, then we can write $\cX = \bF_{2^{\seedlength}}^m = \{0,1\}^{p \cdot k}$ by mapping elements of   $\bF_{2^{\seedlength}}^m$ into its canonical representation. We remark that in~\cite{jiMIPRE2022a} each of the $\decidertypedfunction$ does not necessarily need to share the same register. However, we choose to present it this way for simplicity of notation, and the introspection procedure in~\Cref{sec:introspection} would still work even if $\decidertypedfunction^A$ and $\decidertypedfunction^B$ are defined using different registers. Finally, we have the following lemma stating that any $r$-fold parallel repetition of a CL samplable game is still CL samplable, and the proof follows trivially from applying the parallel composition given in~\Cref{defn:parallelcomposCLfn}. 
\begin{lemma} \label{lem:parallelrefCLfunction}
	Let $r \in \bN$, and let $\cG = (\cX, \cA, \mu, D)$ be a CL samplable game via a $(k,m,p)$-CL distribution. Then, $\cG^{\otimes r} = (\cX^{r}, \cA^r, \mu^{\otimes r}, D^{\otimes r})$ is samplable via a $(k,r \cdot m,p)$ CL distribution. 
\end{lemma}

\subsection{Typed conditionally linear distribution}
In this subsection, we introduce a more general notion of CL distributions, which we call the \textit{typed} CL distribution in this paper. We also show that any game with a typed CL distribution as its sampling distribution can be modified into a game which is CL samplable with only a polynomial decay on the success rate. We remark that the typed CL distribution defined in this subsection is essentially the same as~\cite[Section 4, Section 6]{jiMIPRE2022a}, but with some minor tweaks to accommodate synchronicity condition in the context of non-local games. 

Let $(\decidertypedtype, \decidertypedquestionpair)$ be an undirected graph with at least one edge in $\decidertypedquestionpair$. We represent the elements of $\decidertypedtype$ as elements of $[|\decidertypedtype|]$ and we represent the edges of the graph as $(v^0, v^1) \in \decidertypedtype^2$ with $v^0 = v^1$ representing a self-edge in the graph. We introduce the notion of a \textit{Typed CL distribution below}. 


\begin{definition}[Typed conditionally linear distribution] \label{def:typedCLsampler}
	Let $(\decidertypedtype, \decidertypedquestionpair)$ be a typed graph with $\decidertypedtype = \{0,1\}^{\typelength}$ and let $\{\decidertypedfunction^{v}: V \rightarrow V\}_{v \in \decidertypedtype}$ be a collection of $k$-th level CL functions of size $\seedlength$ over a canonical basis subspace $V \subseteq \bF_{2^{\seedlength}}^m$, where all of the $\decidertypedfunction^{v}$ share the same registers $\{V_j\}_{j \in [k]}$. A distribution $\mu$ is a $(\decidertypedtype, \decidertypedquestionpair,\{\decidertypedfunction^{v}\})$-typed distribution if $\mu$ can be sampled in the following manner:
	\begin{enumerate}
		\item Uniformly sample $(v_0, v_1) \in \decidertypedtype^2$ and perform rejection sampling until $(v_0, v_1) \in \decidertypedquestionpair$. 
		\item Uniformly sample $s \in \bF_{2^{\seedlength}}$ and $b \in \{0,1\}$. 
		\item Compute $x_0 = \decidertypedfunction^{v_0}(s)$ and $x_1 = \decidertypedfunction^{v_1}(s)$
		\item Return the pair $\left((v_b,x_b),  (v_{1-b},x_{1-b})\right)$, as the sample outcome. 
	\end{enumerate}
\end{definition}


We remark in the above sampling procedure, the self-edges are sampled with twice the weight in comparison to other edges. Let $\cG = (\mu, V, \cX, \cA)$ be a synchronous game such that $\mu$ is a $(\decidertypedtype, \decidertypedquestionpair,\{\decidertypedfunction^{v}\})$-typed distribution. We refer to the set $\decidertypedtype$ as the ``question label" and the set $\decidertypedquestionpair$ as the ``question pair" for the game $\cG$. Since the question label is included as an output in the sampling procedure, the synchronous question pair for $\cG$ corresponds to the self-edges in the graph $(\decidertypedtype, \decidertypedquestionpair)$. For the remainder of the paper when discussing games with typed CL distribution as the sampling distribution, we also assume that $\frac{|\decidertypedtype|}{\seedlength}$ is always an integer in this paper by implicitly padding $\decidertypedtype$ with extra vertices which only contains self-loops. Since we usually associate $|\decidertypedtype| = O(\seedlength)$ in this paper, this assumption does not change the complexity of $\decidertypedtype$. 


Intuitively, typed CL distributions are a generalization of CL distributions, where instead of having two CL functions, there could potentially be $|\decidertypedtype|$ different CL functions used for sampling. We give a method of taking a synchronous game with typed CL distribution and simulating it with a synchronous game using a normal CL distribution. Given a graph $(\decidertypedtype, \decidertypedquestionpair)$ and a vertex $v \in \decidertypedtype$, we define the \textit{neighbour indicator} for $v$ to be the vector $\text{\gls{Neightindvec}}\in \{0,1\}^{|\decidertypedtype}|$ such that
\begin{equation*}
	\text{neigh}_{\decidertypedquestionpair}(v_0)_{v_1} := \begin{cases}
		1 &\text{if } \{v_0, v_1\} \in \decidertypedquestionpair \\
		0 &\text{otherwise}
	\end{cases}.
\end{equation*} 

The following transformation shows that games with a typed-CL distribution as the sampling procedure can always be simulated with a (normal) CL samplable game. We remark that this construction is similar to the one given in~\cite[Section 6.2]{jiMIPRE2022a}.
 

\begin{definition}[Detyped conditionally linear distribution] \label{def:detypeCLsampler}
	Let $\seedlength, m \in \bN$ with $\seedlength$ an odd positive integer. Let $(\decidertypedtype, \decidertypedquestionpair)$ be a graph, $\{\decidertypedfunction^v: V \rightarrow V\}_{v \in \decidertypedtype}$ be a collection of $k$-th CL function with registers $\{V_j\}_{j \in [k]}$ for some canonical basis subspace $V \subseteq \bF_{2^{\seedlength}}^{m}$ and let $\mu$ be a $(\decidertypedtype, \decidertypedquestionpair, \{\decidertypedfunction^v \}_{v \in \decidertypedtype})$-typed distribution. Let $b  = \frac{|\decidertypedtype|}{\seedlength}$ and $\typelength = \ceil{\frac{\log_2(|\decidertypedtype|)}{\seedlength}}$. 
	
	We define the detyped transformation for $\mu$, denoted as $\mu^{\detypesynchead}$, to be a $(k+2, m+ 4\cdot b + 2\cdot \typelength, \seedlength)$ CL distribution. Where the second level CL function $\decidertypedfunction^A, \decidertypedfunction^B: \bF_{2^\seedlength}^{4\cdot b + 2\cdot \typelength} \oplus V \subseteq \bF_{2^\seedlength}^{ 4\cdot b + 2\cdot \typelength + m}$. 
	\begin{itemize}
		\item $\decidertypedfunction^A$ is defined as a series composition between $\decidertypedfunction^{A,1}$ and $\{\decidertypedfunction^{A,2,s} \}_{s \in S_1}$ as per~\Cref{defn:seriescomposCLfn}. Where $\decidertypedfunction^{A,1}$ is a second level CL samplable function acting on $V^1 = \bF_{2^\seedlength}^{2\cdot \typelength + 4\cdot b}$, and $\{\decidertypedfunction^{A,2,s} \}_{s \in S_1}$ is a collection of $k$-th level CL functions acting on $V^2 = V$. We represent elements of $V^1$ under the canonical representation (i.e. as elements of $\{0,1\}^{2 \cdot \seedlength (\typelength+ 2 \cdot  b)}$) in this definition formulation. We give the details for each level of $\decidertypedfunction^A$ below:
		\begin{itemize}
			\item The first second level CL function $\decidertypedfunction^{A,1}$ acts on two registers, $V^1_0 = \{0,1\}^{ 2 \cdot \seedlength \cdot(\typelength + b)} = \{0,1\}^{2b\cdot \seedlength + 4 \cdot |\decidertypedtype| }$ and $V^1_1 = \{0,1\}^{ 2 \cdot \seedlength \cdot b} = \{0,1\}^{2 \cdot |\decidertypedtype|}$. 
			\item We define the zeroth level linear function for $\decidertypedfunction^{A,1}_{0.0}$ as follows: For all  $x \in V^1_0$ as $(x_0, x_1, x_2,x_3)$, where $x_0, x_2 \in \{0,1\}^{\seedlength \cdot \typelength}$ and $x_1, x_3 \in \{0,1\}^{|\decidertypedtype|}$, then $\decidertypedfunction^{A,1}_{0,0}$
			\begin{equation} \label{eq:detype1}
				\decidertypedfunction^{A,1}_{0,0} (x_0, x_1, x_2,x_3) = (x_0, x_1,0,0)
			\end{equation}
			for all elements $x \in V^1$. 
			\item Fix $s_0 \in \{0,1\}^{\seedlength \cdot \typelength}$ and $s_1 \in \{0,1\}^{|\decidertypedtype|}$, and let $v = \binaryinv{(\pi_{\leq \ceil{\log_2(|\decidertypedtype|)}}(s_0)}$ (i.e. extract the first $\ceil{\log_2(|\decidertypedtype|)}$ bits of $s_0$ and treat it as an integer). For all $x \in V^1_1$, parse $x$ as $(x_4, x_5)$, where $x_4, x_5 \in  \{0,1\}^{|\decidertypedtype|}$. Whenever $s_0 \in [|\decidertypedtype|] = \decidertypedtype$ and $s_1 = \text{neigh}_{\decidertypedquestionpair}(v)$, we define the first level linear function for $\decidertypedfunction^{A,1}$ as 
			\begin{equation} \label{eq:detype2}
				\decidertypedfunction^{A,1}_{1,(s_0, s_1, s_2,s_3)} (x_4,x_5) = (x_4, (x_5)_v),
			\end{equation}
			where $(x_5)_v$ zeros out all entries of $x_5$ except for the $v$th entry. Otherwise 
			\begin{equation*}
				\decidertypedfunction^{A,1}_{1,(s_0, s_1, s_2,s_3)} (x_4,x_5) = 0.
			\end{equation*}
			\item  The collection of $k$-th level CL functions $\{\decidertypedfunction^{A,2,s}\}_{s \in S_1}$ is defined to be the following: Parse $s \in V^1$ as $(s_0, s_1, s_2, s_3,s_4,s_5)$, and define $v$ the same way as the above clause. We define 
			\begin{equation} \label{eq:detype3}
				\decidertypedfunction^{A,2,s}(x) = \decidertypedfunction^v
			\end{equation}
			whenever $v \in \decidertypedtype$,  $s_1 = s_4 = \text{neigh}_{\decidertypedquestionpair}(v)$ and $(s_5)_v =1$. Otherwise we set $\decidertypedfunction^{A,2,s}(x) = 0$
		\end{itemize}
		\item $\decidertypedfunction^B$ is defined mostly similar to $\decidertypedfunction^A$, except we replace the equation \eqref{eq:detype1}, \eqref{eq:detype2} and \eqref{eq:detype3} by
		\begin{align*}
			\decidertypedfunction^{B,1}_{0,0} (x_0, x_1, x_2,x_3) &= (x_2, x_3,0,0), \\
			\decidertypedfunction^{B,1}_{1,(s_0, s_1, s_2,s_3)} (x_4,x_5) &= (x_5, (x_4)_v), \\
			\decidertypedfunction^{B,2,s}(x) &= \decidertypedfunction^v
		\end{align*}
	\end{itemize}
\end{definition}
We also give the notion of a ``non-trivial seed" for a detyped CL distribution below. Using the same notation as~\Cref{def:detypeCLsampler}, for $s = s^1 \oplus s^2 \in V^1 \oplus  V^2$. Parse $s^1 = (s_0,s_1, s_2, s_3, s_4, s_5)$, where $s_0, s_2 \in \{0,1\}^{\seedlength \cdot \typelength}$ and $s_1, s_3, s_4, s_5 \in \{0,1\}^{|\decidertypedtype|}$. Furthermore, parse $s_0$ into an element $v_0$ as per described in the definition of $\decidertypedfunction^{A,1}$, and parse $s_2$ into $v_1$ in the same way. We call $s$ a non-trivial seed for the CL distribution  $\mu^{\detypesynchead}$ if the following holds
\begin{enumerate}
	\item $v_0, v_1 \in \decidertypedtype$ and $(v_0, v_1) \in \decidertypedquestionpair$
	\item $s_1 = s_4 = \text{neigh}_{\decidertypedquestionpair}(v_0)$ and $s_3= s_5 = \text{neigh}_{\decidertypedquestionpair}(v_1)$
\end{enumerate}
otherwise, we refer to $s$ as a trivial seed. 

The detyping procedure might seem convoluted at first. Intuitively, $(s_0, s_1, s_4)$ in $s^1$ as given in the above definition dictates which vertices $\decidertypedfunction^{A,1}$ samples and $(s_2, s_3, s_5)$ dictates which vertices $\decidertypedfunction^{B,1}$ samples. If we perform a rejection sampling on the set of non-trivial seeds, we see that this is equivalent to the original typed CL distribution since both $s_1 = s_4 = \text{neigh}_{\decidertypedquestionpair}(v_0)$ and $s_3= s_5 = \text{neigh}_{\decidertypedquestionpair}(v_1)$ occurs with the same probability given a fixed $(v_0, v_1) \in \decidertypedquestionpair$. For a detyped sampler $\mu^{\detypesynchead}$, in the event that a non-trivial seed is being chosen, the resulting output $(x,y)$ has the following form:
\begin{align}
	\nonumber x &= (v, \text{neigh}_{\decidertypedquestionpair}(v), 0,0,\text{neigh}_{\decidertypedquestionpair}(v), (\text{neigh}_{\decidertypedquestionpair}(u))_{v}, \decidertypedfunction^v(s^2)) \\
	y &= (u, \text{neigh}_{\decidertypedquestionpair}(u), 0,0,\text{neigh}_{\decidertypedquestionpair}(u), (\text{neigh}_{\decidertypedquestionpair}(v))_{u}, \decidertypedfunction^u(s^2)) \label{eq:parsetypedgame}
\end{align}
for some $u,v \in \decidertypedtype^2$. Since the seed is non-trivial, $(u,v) \in \decidertypedquestionpair$ and hence $(\text{neigh}_{\decidertypedquestionpair}(u))_{v}$ and $(\text{neigh}_{\decidertypedquestionpair}(v))_{u}$ will always be $1$. By~\Cref{def:detypeCLsampler}, the output from $\mu^{\detypesynchead}$ can only be parsed in the above form iff $s$ is initially chosen to be a non-trivial seed. We have the following lemma lower bounding the set of non-trivial seeds in a detyped transformation. 
\begin{lemma}[Percentage of non-trivial seeds in a detyped CL distribution] \label{lem:nontrivialseedCLdistribution}
	Let $\mu$ be a $(\decidertypedtype, \decidertypedquestionpair, \{\decidertypedfunction^v \}_{v \in \decidertypedtype})$-typed distribution with $|\decidertypedtype| =  t$ and $\{\decidertypedfunction^v: V \rightarrow V\}_{v \in \decidertypedtype}$ for some subspace $V  = \bF_{2^{\seedlength}}^{m}$ and let $\mu^{\detypesynchead}$ be the corresponding $(k+2, m+ 4\cdot b + 2\cdot \typelength, \seedlength)$ detyped CL distribution as defined in~\Cref{def:detypeCLsampler}. For every $s$ sampled uniformly random from $\bF_{2^{\seedlength}}^{4\cdot b + 2\cdot \typelength} \oplus V$, $s$ is a non-trivial seed for $\mu^{\detypesynchead}$ with probability at least 
	\begin{equation*}
		\frac{1}{4t^2 \cdot 16^{t}}.
	\end{equation*}
\end{lemma}

\begin{proof}
	We use the same notation as the one given in~\Cref{def:detypeCLsampler}. We first point out that for $s = s^1 \oplus s^2 \in V^1 \oplus  V^2$, only $s^1$ dictates whether $s$ is a non-trivial seed. Hence consider $s^1$ uniformly sampled from $V^1$ and write $s^1 = (s_0,s_1, s_2, s_3, s_4, s_5)$ into the parsing as defined as in~\Cref{def:detypeCLsampler}. 
	
	We first lower bound the probability that $s^1$ satisfies the first clause of being a non-trivial seed. If we take the first $\ceil{\log_2(t)}$ bit $s_0 \in \{0,1\}^{\typelength \cdot \seedlength}$ and convert it back to integer through the map $\binaryinv{\cdot}$, then with probability at least $\frac{1}{2}$ we obtain some $v_0 \in \decidertypedtype$. Similarly, with probability $\frac{1}{2}$, we can parse $s_2$ into some $v_1 \in \decidertypedtype $. Since we assume there is at least one edge in $\decidertypedquestionpair$, the probability that $(v_0, v_1) \in \decidertypedquestionpair$ is at least $\frac{1}{n^2}$ given that $(v_0, v_1) \in [|t|] \times [|t|]$. Hence $s^1$ satisfies the first clause with probability at least $\frac{1}{4t^2}$. 
	
	For the second clause, for any given $v \in \decidertypedtype$, since there is only one unique string in $\{0,1\}^{|\decidertypedtype}|$ which is equal to $\text{neigh}_{\decidertypedquestionpair}(v)$. Given $v_0 \in \decidertypedtype$, $s_1$ and $s_3$ will have probability $\left(\frac{1}{2^t}\right)^2$ to be equal to $\text{neigh}_{\decidertypedquestionpair}(v_0)$. Hence $s^1$ satisfies the second clause with probability $\left(\frac{1}{2^t}\right)^4 = \frac{1}{16^t}$ given the first clause. Hence $s^1$ has a probability $\frac{1}{4t^2}\cdot \frac{1}{16^t} $ of being a non-trivial seed. 
\end{proof}

Given a synchronous game with a typed CL distribution as the input distribution, we define the following transformation which replaces the typed CL distribution with its detyped counterpart below. 

\begin{definition}[Detyped conditionally linear game]
	Let $\seedlength, m \in \bN$ with $\seedlength$ being an odd positive integer. Let $(\decidertypedtype, \decidertypedquestionpair)$ with $\decidertypedtype = \{0,1\}^{\typelength}$ and let $\{\decidertypedfunction^v: V \rightarrow V\}_{v \in \decidertypedtype}$ be a $k$-th level CL function for some canonical basis subspace $V \subseteq \bF_{2^{\seedlength}}^{m}$. Let $\cG = (\cX, \cA, \mu, V)$ be a synchronous game with $\mu$ being a $(\decidertypedtype, \decidertypedquestionpair, \{\decidertypedfunction^v \}_{v \in \decidertypedtype})$-typed distribution as defined in~\Cref{def:typedCLsampler}. We define the detyped and synchronization transformation $\cG^{\detypesynchead} = (\cX^{\detypesynchead}, \cA, \mu^{\detypesynchead}, V^{\detypesynchead}) $ for $\cG$ as follows:
	\begin{itemize}
		\item The distribution $\mu^{\detypesynchead}$ is the $(k+2, m+ 4\cdot b + 2\cdot \typelength, \seedlength)$  CL distribution given by~\Cref{def:detypeCLsampler}. 
		\item For the question pair $(x,y) \in \bF_{2^{\seedlength}}^{m+ 4\cdot b + 2\cdot \typelength}$, the verification $V^{\detypesynchead}$ is given as follows:
		\begin{itemize}
			\item If $s$ is a non-trivial seed, parse $(x,y)$ as per ~\Cref{eq:parsetypedgame}, then 
			\begin{equation*}
				V^{\detypesynchead}\left(x, y,a,b \right) = V\left((v^0, L^{v^0}(s)), (v^1, L^{v^1}(s)),a,b \right)
			\end{equation*}
			in other words, the same as the original game. 
			\item Otherwise, $V^{\detypesynchead}\left(x,y,a,b \right) = \delta_{a,b}$.
		\end{itemize}
	\end{itemize}
\end{definition}
We see that the above transformation also preserves synchronicity for a given game. One might wonder the reason for such a roundabout way to defining the detyping procedure, since instead, one can sample two arbitrary vertices from $\decidertypedtype$ and  perform rejection sampling on the case where these two vertices are connected in $\decidertypedquestionpair$.  As seen in the lemma below, the main reason for the roundabout way is that it allows the transformation to also preserves any perfect oracularizable strategies after the transformation. 

To be more precise, let $\cG$ be a non-local game which uses some typed CL sampler $\mu$. If we were to replace $\mu$ with the ``sampling two vertices, and perform rejection sampling" approach for $\cG$, any oracularizable strategy might no longer be oracularizable because the measurement operator used between ``two non-connected vertices" might not commute. Using the detyping procedure listed above, at least one of the provers can deduce whether $s$ is the trivial seed, and thus, can adjust their measurement operator accordingly. To this end, we have the following lemma which shows how much the completeness/soundness condition changes for a detyped and synchronization transformation. We remark that the proof of this lemma is similar to~\cite[Lemma 6.18]{jiMIPRE2022a}. 


\begin{lemma}[Preservation of completeness/soundness of the detyped and synchronization game] \label{lem:detypeandsync}
	Let $\cG = (\cX, \cA, \mu, V)$ be a synchronous game with $\mu$ being a $(\decidertypedtype, \decidertypedquestionpair, \{\decidertypedfunction^v \}_{v \in \decidertypedtype})$-typed distribution. For model $t \in \{*, co\}$, then 
	\begin{itemize}
		\item \textbf{(Completeness)} If there exists a perfect oracularizable synchronous strategies for $\cG$ in model $t$, then there exist a perfect oracularizable synchronous strategies for $\cG^{\detypesynchead}$ in model $t$. 
		\item  \textbf{(Soundness)} If $\omega^t(\cG) > 1 - \eps$, then 
		\begin{equation*}
			\omega^{t}(\cG^{\detypesynchead}) > 1 - \frac{\eps}{\left(4|\decidertypedtype|^2 \cdot 16^{|\decidertypedtype|}\right)} . 
		\end{equation*}
	\end{itemize}
\end{lemma}
\begin{proof}
	Fix $t \in \{*, co\}$, and assume that the quantum strategy performed below is defined using model $t$. In the proof below, for $v  \in \decidertypedtype$, we define
	\begin{equation} \label{eq:proofdetypeeq1}
		\text{view}(v) = (\binary{v}, \text{neigh}_{\decidertypedquestionpair}(v), 0,0,\text{neigh}_{\decidertypedquestionpair}(v),e_v),
	\end{equation}
	where $\binary{v}$ above is only taken over the first $\ceil{\log_2(|\decidertypedtype|)}$ bits. By definition, for $(x,y) \sim \mu^{\detypesynchead}$, $(x,y) = \left((\text{view}(v_0), \decidertypedfunction^{v_0}(s^2), (\text{view}(v_1),\decidertypedfunction^{v_1}(s^2))\right)$ for some $(v_0, v_1) \in \decidertypedquestionpair$ and $s^2 \in \bF_{2^{ \seedlength}}^m$ iff $\mu^{\detypesynchead}$ is sampled using a non-trivial seed in the sampling procedure. 
	
	\paragraph{Completeness. }Let $\strategy = (\alicealg, \{\tilde{A}^x_a\}, \cH ,\tracialstate)$ be a perfect oracularizable synchronous strategy for $\cG$. We define an oracularizable synchronous strategy for $\cG^{\detypesynchead}$ as follows: Given a question label $x \in \cX$, if there exists some $v \in \decidertypedtype$ and $s^2 \in \bF_{2^{ \seedlength}}^m$ such that $x= (\text{view}(v), \decidertypedfunction^{v}(s^2)) $, then set $A^x_a = \tilde{A}^{(v, \decidertypedfunction^v(s^2))}_a$. If $x$ cannot be parsed in the above format, the prover always returns some predetermined fixed element $\star \in \cA$. This ensures that $A^x_\star = \cI$ and $A^x_a = 0$ for all $a \in \cA \setminus \{\star\})$. 
	
	Restricted to the question set where $(x,y)$ are parsed correctly, we see that $A^x_a$ is the same as $\tilde{A}^{(v, \decidertypedfunction^v(s^2))}_a$ with the same decider function. This implies that $A^x_a$ is a perfect oracularizable synchronous strategy when restricted to this case. Otherwise, since $A^x_a$ returns the same answer, and the only non-trivial question pair in this scenario is the synchronicity question pair. This strategy passes with perfect accuracy. This means that $A^x_a$ remains a perfect oracularizable synchronous strategy in the case where at least one of $x,y$ cannot be parsed correctly, hence showing that the strategy given above is a perfect oracularizable strategy for $\cG^{\detypesynchead}$. 
	
	\paragraph{Soundness. }Let $\strategy = (\alicealg, \{A^x_a\}, \{(B^y_b)^{op}\} ,\tracialstate)$ be a tracially embeddable strategy for $\cG^{\detypesynchead}$ with a success rate of $1 - \eps$. By~\Cref{lem:detypeandsync}, with probability $\frac{1}{4|\decidertypedtype|^2 \cdot 16^{|\decidertypedtype|}}$, the output from $\mu^{\text{\detypesynchead}}$ would be from a non-trivial seed. Conditioning on this case, the distribution $\mu^{\detypesynchead}$ is exactly the same as $\mu$. We define the strategy for $\cG$ with the same measurement state/algebra as $\strategy$, but with the measurement operator replaced by $\hat{A}^{(v, \decidertypedfunction^{v}(s^2))}_{a} = A^{(\text{view}(v), \decidertypedfunction^{v}(s^2))}$. This strategy will succeed at $\cG$ with probability at least $1 - 4|\decidertypedtype|^2 \cdot 16^{|\decidertypedtype|} \cdot \eps$. This implies if $\omega^t(\cG) > 1 - \eps$, then  
	\begin{equation*}
		\omega^{t}(\cG^{\detypesynchead}) > 1 - \frac{\eps}{\left(4|\decidertypedtype|^2 \cdot 16^{|\decidertypedtype|}\right)},
	\end{equation*}
 which completes the claim. 
\end{proof}       

In this paper, $|\decidertypedtype|$ is typically taken to be the complexity bound of $O(\log(n))$ for some integer $n$ or some constant. Hence the increase in soundness shown in the above lemma can still be ``reset" using the parallel repetition presented in~\Cref{sec:parallelrep}.  

\subsection{The quantum low-individual degree test} \label{sec:lowinddegtest}

In this subsection, we recall the quantum low-individual degree test from~\cite{jiQuantumSoundnessTesting2022}, and show that it is CL samplable. This test is an important subroutine used within the answer reduction protocol presented in~\Cref{sec:PCPanswerreduction}. We start this subsection by giving some high-level intuition about this test. 

The quantum low-individual degree test is first introduced in~\cite{jiQuantumSoundnessTesting2022}, and it is based on the classical low-degree test given in~\cite{babaiNondeterministicExponentialTime1991a}. The classical low-individual degree test is used in some of the earlier work on the PCP theorem~\cite{aroraProbabilisticCheckingProofs1998,aroraProofVerificationHardness1998}, and the quantum low-individual degree is used in a similar way in the answer reduction protocol in this paper. 

The quantum low-individual degree test is parametrized by $(p,m,d) \in \bN^3$, where $m = 2^c$ for some constant $c$. Intuitively, the goal of the quantum low-individual degree test is to force two entangled provers to prove to the verifier that they both share the same global $m$-variant polynomial over $2^p$ with individual degree of at most $d$. To give a better intuition on how the quantum low-individual-degree test works, we first give a brief description of its classical counterpart. 
\begin{figure}[!htbp]
	\centering
	\begin{gamespec}
		\setlength{\tabcolsep}{1em}
		The quantum low-individual degree test is parametrize by $(p, m, d) \in \bN^3$ where $p,m$ are both an odd integer. Perform the following test with probability $\frac{1}{8}$ each:
		\begin{itemize}
			\item \textbf{Axis parallel line test.} The verifier uniformly samples $s = (s_0, \cdots, s_{m-1}) \in \bF_{2^p}^m$ and $j \in [m]$. Let $\textbf{l}_j$ be the $j$th axis-parallel line given in~\eqref{eq:axisparallelline}. Recall from~\Cref{def:canorepline} that $\Canoline{\textbf{l}_j}$ is the canonical representation of a line. 
			\begin{itemize}
				\item The verifier sends $(\text{Point},s)$ to one of the provers, and receive $a \in \bF_{m}$ as a response.
				\item The verifier sends $(\text{Aline},\Canoline{\textbf{l}_j})$ to the other prover, and receive a degree $d$ polynomial $\textbf{f}: \bF_{2^p} \rightarrow \bF_{2^p}$ as a response.  
			\end{itemize}
			The verifier accepts if $\textbf{f}(s) = a$.
			\item \textbf{Diagonal line test.} The verifier uniformly samples $s = (s_0, \cdots, s_{m-1}) \in \bF_{2^p}^m$, $j \in [m]$ and $v \sim \bF_{2^p}^j$. Extend $v$ as an element of $\bF_{2^p}^m$ by appending $0$ on the last $m-j$ coordinates. Define the line $\textbf{d}_{j,v} = \{s + x \cdot v: x \in \bF_{2^p}\}$ to be the line passing through $s$ in the direction of $v$.  
			\begin{itemize}
				\item The verifier sends $(\text{Point},s)$ to one of the provers, and receive $a \in \bF_{m}$ as a response.
				\item The verifier sends $(\text{Dline},\Canoline{\textbf{d}_{j,v}})$ to the other prover, and receive a degree $d \cdot m$ polynomial $\textbf{g}: \bF_{2^p} \rightarrow \bF_{2^p}$ as a response.  
			\end{itemize}
			The verifier accepts if $\textbf{g}(s) = a$.
		\end{itemize}
		Perform the following test with probability $\frac{1}{4}$ each:
		\begin{itemize}
			\item \textbf{Point consistency test.} The verifier uniformly samples $s \in \bF_{2^p}^m$. The verifier sends $(\text{Point},s)$ to both provers, and receive $(a,b) \in \bF_{2^p}^2$. The verifier accepts if $f(s) = a$.
			\item \textbf{Axis parallel line consistency test.} The verifier uniformly samples $s \in \bF_{2^p}^m$ and $j \in [m]$. Let $\textbf{l}_j$ be the axis parallel line define in the ``Axis parallel line test". The verifier sends $\Canoline{\textbf{l}_j}$ to both provers, and receive two degree $c$ polynomial $\textbf{f}^A, \textbf{f}^B:  \bF_{2^p} \rightarrow \bF_{2^p}$. The verifier accepts if $\textbf{f}^A = \textbf{f}^B$.
			\item \textbf{Diagonal line line consistency test.} The verifier uniformly samples $s \in \bF_{2^p}^m$, $j \in [m]$ and $v \sim \bF_{2^p}^j$. Let $\textbf{d}_{j,v}$ be the line define in the ``Diagonal line test". The verifier sends $\Canoline{\textbf{d}_{j,v}}$ to both provers, and receive two degree $d \cdot m$ polynomial $\textbf{g}^A, \textbf{g}^B:  \bF_{2^p} \rightarrow \bF_{2^p}$. The verifier accepts if $\textbf{g}^A = \textbf{g}^B$.
		\end{itemize} 
		\vspace{1em}
	\end{gamespec}
	\caption{The sampling/decision procedure for the $(p,m,d)$-quantum low-individual degree test, the only change is the distribution of the synchronicity test, which only changes the constant in~\Cref{thm:soundnessQLDT} (see~\cite[Section 4.1]{jiQuantumSoundnessTesting2022} for more details).}
	\label{fig:LDTdescription}
\end{figure}

Suppose the two provers agree on a global $m$-variant low-individual degree polynomial $\textbf{g}: \bF_{2^p}^m \rightarrow \bF_{2^p}$ with individual degree of at most $d$. To demonstrate to the verifier that they both share the same polynomial, both provers can simply send $g$ to the verifier. However, since $\textbf{g}$ can have at most $(d+1)^{m}$ monomials, this means that the prover needs to send a messages with potential length $O((d+1)^{m} \cdot p)$-bits which is inefficient if $m$ is large. Instead, the verifier can send one of the prover $u \in \bF_{2^p}^m$ and ask him to evaluate the polynomial $\textbf{g}$ on $u$ and return the outcome $\textbf{g}(u) \in \bF_{2^p}$. The verifier then gives the other prover a random ``parallel axis line" $\textbf{l}: \bF_{2^p} \rightarrow\bF_{2^p}$ that pass through $u$, where for $i \in [m]$, an axis parallel line passing through $u$ is defined as 
\begin{equation} \label{eq:axisparallelline}
	\textbf{l} = \{ u + x \cdot e_i | x \in \bF_{2^p} \}.
\end{equation}
The prover receiving the axis parallel line is expected to return the polynomial $\textbf{g}_\textbf{l}: \bF_{2^p} \rightarrow\bF_{2^p}$, where $\textbf{g}_\textbf{l}$ is the $m \cdot d$-th degree polynomial corresponding to $\textbf{g}$ restricted to the range of $\textbf{l}$, or
\begin{equation*}
	\textbf{g}_\textbf{l}(x) = \textbf{g}(u + x \cdot e_i). 
\end{equation*}

Intuitively, this is asking the prover to evaluate $\textbf{g}$ on \emph{all} of $\textbf{l}$. If the two provers share the same low-individual degree polynomial in the beginning of the protocol, then $\textbf{g}_\textbf{l}(0)$ would be consistent with $\textbf{g}(u)$. If, on the other hand, the two provers do not share the same low-individual degree polynomial, then by~\Cref{lem:Schwartz_Zipple}, for two different $m$-variant low-individual degree polynomials $\textbf{g}$ and $\textbf{g}'$ and $u \in \textbf{l}$, $\textbf{g}(u) = \textbf{g}'(u)$ occurs with probability at most $\frac{d}{q}$. This means that the restricted polynomial $\textbf{g}_\textbf{l}$ generated by the axis parallel line prover is unlikely to agree with the prover with who is expected to evaluate $\textbf{g}$ somewhere on $\textbf{l}$. This protocol allows the verifier to check the consistency of a global low-individual degree polynomial with message size $O(d \cdot k)$, significantly more efficient than the previous protocol. 

In order to make sure the above protocol remains quantum sound (i.e. the provers will have a low success rate if they do not share the same low-individual degree test polynomial)~\cite{jiQuantumSoundnessTesting2022} added two additional questions types. The first is the ``diagonal line test", where, one of the provers still receives a random point $u \in \bF_{2^p}^m$ and similarly is expected to return $\textbf{g}(u)$ in the ideal case. The other prover is given a random ``diagonal line" intersecting with $u$, where we define the notion of a diagonal line below, and is expected to return an $m \cdot d$-th degree polynomial similarly as to above. This addition is mostly used to enforce a commutation relationship for the proof of soundness for the quantum low-degree test. Secondly, to ensure synchronicity, a ``consistency test" is added, where the two provers are given the same question (which can be a line, a point, or a diagonal line) and they are expected to output the same answer. We formally give the definition for a quantum low-individual degree test on~\Cref{fig:LDTdescription}

 We remark that in our description, the verifier sends the canonical representation of the line. This is equivalent to the original formulation where the verifier sends a set containing all points in the line to the prover. Both formulations are designed to hide the point $u$ (for the ``point" player) when sending the line. As seen in the description from the classical low-degree test, if both provers share an $m$-variate polynomial $\textbf{h}$ over $\bF_{2^p}$ with individual degree of at most $d$. The provers can simply plug their respective input into $\textbf{h}$ to obtain a perfect (classical) strategy. We recall the following soundness result related to the quantum low-individual degree test:
\begin{theorem}[Quantum soundness for the quantum low-individual degree test, Theorem 4.1 of~\cite{jiMIPRE2022a}] \label{thm:soundnessQLDT}
	There exist a universal constant $1 \geq c_{\text{LD},1}$ and $0 < c_{\text{LD},2} \leq 1$  and a function 
	\begin{equation*}
		\text{\gls{SoundnessLDT}} =  c_{\text{LD},1} (dm)^{c_{\text{LD},1}} (\eps^{c_{\text{LD},2}} + 2^{-c_{\text{LD},2}p} + 2^{-c_{\text{LD},2}md})
	\end{equation*}
	such that the following holds. Let $\cG^{\text{LD}}$ be the $(p,m,d)$-low-individual degree test with $q = 2^p$, and let $\strategy =( \cL^2(\alicealg, \tau), \ket{\tau},\{ A_{a}^x \} )$ be a synchronous strategy for $\cG^{\text{LD}}$ which succeed with probability $1- \eps$. There exist a set of PVM $\{G_{\textbf{g}}\} \subseteq \alicealg'$ with outcome labelled by $m$-variant polynomials $\textbf{g} \in \Idpoly(p,m,d)$, such that 
	\begin{equation*} 
		\Ex_{s \sim \bF_{2^p}^m} \sum_{\textbf{g} \in \Idpoly(p,m,d)} \braket{\tau| A^{(\text{point}, s)}_{\textbf{g}(s)} \cdot  G_{\textbf{g}} |\tau} \geq 1 - \mathbf{\eta}_{\text{LD}}(p,m,d,\eps).
	\end{equation*}
\end{theorem}
In other words, if the provers succeed on the $(p,m,c)$-low-individual degree test with high probability. Then the provers, in essence, are secretly sampling a low-individual degree polynomial which are then used as a part of the strategy. Although~\cite[Theorem 4.1]{jiQuantumSoundnessTesting2022} is originally proven for tensor codes, as shown in~\Cref{sec:fieldintro}, the set of $m$-variant polynomial over $\bF_{2^p}^m$ with low-individual degree of at most $c$ is a tensor code $\mathfrak{C}^{\otimes m}$ for a $[2^p, c, c]_{\bF_{2^p}}$ linear code $\mathfrak{C}$, and hence the same statement can be directly applied to the quantum low-individual degree test. We remark that as shown in~\cite[Corollary 4.4]{linTracialEmbeddableStrategies2024}, the above theorem actually applies for general tracially embeddable strategies.


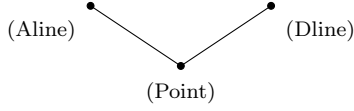
\begin{figure}[t!] 
	\centering
	\begin{tikzpicture}[scale=.8]
		
		\tikzset{type/.style args={[#1]#2}{
				draw,circle,fill,scale=0.25,
				label={[font=\scriptsize, label distance=1pt]#1:#2}
		}}

		\draw (0,1) coordinate (Axis-line) node[type={[225]$(\text{Aline})$}] {};
		\draw (1.5,0) coordinate (Point) node[type={[270]$(\text{Point})$}] {};
		\draw (3,1) coordinate (D-line) node[type={[315]$(\text{Dline})$}] {};
		
		\draw (Axis-line) -- (Point);
		\draw (Point) -- (D-line);
	\end{tikzpicture}
	
	\caption{The typed graph for the quantum low-individual degree test}
	\label{fig:LDtesttype}
\end{figure}

For the remainder of this subsection, we show that the quantum low-individual degree test can be sampled via a typed CL distribution, and hence can be converted to a game with a CL distribution as the input distribution via~\Cref{lem:detypeandsync}. 

\begin{lemma}[The quantum low-individual degree test can be sample via a CL distribution] \label{lem:CLsamplableLDT}
	Let $\cG = (\cX, \cA, \mu, V)$ be a $(p,m,d)$-quantum low-individual degree test, then there exists a game $\cG' = (\cX', \cA, \mu', V')$ which is $(5, 9+ m' + 2 \cdot m, p)$ CL samplable where $m' = \ceil{\frac{\log(m)}{p}}$, and $\cX \subseteq \cX'$ such that the following holds: For any $t \in \{ *, co\}$ and $\eps >  0$. Any synchronous, oracularizable strategy $\strategy$ in model $t$ such that $\omega(\cG',\strategy) \geq 1 - \eps$. $\strategy$, when restricted on the question pair from $\cG$, $\omega(\cG, \strategy) \geq 1 - c \eps$ for some constant $c$. 
\end{lemma}

\begin{proof}
We first show that the quantum low-individual degree test can be sampled via a typed CL distribution. Let $(p,m,d)$ be the parameter specified and let $q = 2^p$. For simplification, we first assume $\frac{\log(m)}{p} \in \bN$ (and hence $m' \in \bN$). The typed graph associated with the quantum low-individual degree test can be specified by the types stated in~\Cref{fig:LDtesttype}. The CL function maps $\bF_{2^p}^{m' + 2 \cdot m} \rightarrow \bF_{2^p}^{m' + 2 \cdot m}$, where we write $\bF_{2^p}^{m' + 2 \cdot m} = \bF_{2^p}^{m'} \oplus \bF_{2^p}^m \oplus \bF_{2^p}^m  =  V_0 \oplus V_1 \oplus V_2$. The CL functions is a level $3$-CL function with register $\{V_0, V_1, V_2\}$. For the description below, we assume the function have the input $s = (s_0, s_1, s_2) \in V_0 \oplus V_1 \oplus V_2$. We define the collection of CL functions $\{\decidertypedfunction^v\}$ as the following

\begin{itemize}
	\item  Define $\decidertypedfunction^{\text{Point}}$ to be the third level CL function
	\begin{equation*}
		\decidertypedfunction^{\text{Point}}(s_0, s_1, s_2) = (0,0,s_2).
	\end{equation*}
	In other words, $\decidertypedfunction^{\text{Point}}$ projects the input $s$ onto $V_2$, and $s_2$ corresponds to the point in the ``point question" for the quantum low-individual degree test. $\decidertypedfunction^{\text{Point}}$ can be realized as a third level CL function with registers $\{V_i\}_{i \in [3]}$ in the following way: We define
	\begin{equation*}
		\decidertypedfunction_{0,0}^{\text{Point}}(s_0,0,0) = (0,0,0), \qquad 	\decidertypedfunction_{1,x_0}^{\text{Point}}(0,s_1,0) = (0,0,0), \qquad \decidertypedfunction_{2,x_0 + x_1}^{\text{Point}}(0,0,s_2) = (0,0,s_2),
	\end{equation*}
	for all $x_0 \in  V_0$ and $x_1 \in V_1$. 
	\item Define $\decidertypedfunction^{\text{Dline}}$ to be the third level CL function. For any input $(s_0, s_1, s_2) \in \bF_{2^p}^{m' + 2 \cdot m}$, let $\hat{j} = \kappa(s_0)$. The function $\decidertypedfunction^{\text{Dline}}$ is defined as
	\begin{equation*}
		\decidertypedfunction^{\text{Dline}}(s_0, s_1, s_2) = (s_0,0,\Nullline{e_{\binaryinv{\hat{j}}}}{s_2} ), 
	\end{equation*}
	where $\text{Null}^{\text{LN}}$ is the function used in~\Cref{def:canorepline} to defined the canonical representation of a line. In the example above, $(e_{\binaryinv{\hat{j}}},\Nullline{e_{\binaryinv{\hat{j}}}}{s_2}$ defines an axis parallel line thought the $\binaryinv{\hat{j}} = j$th axis through the point $s_2$. $\decidertypedfunction^{\text{Dine}}$ can be realized as a third level CL function with registers $\{V_i\}_{i \in [3]}$ in the following way: We define
	\begin{equation*}
		\decidertypedfunction_{0,0}^{\text{Dline}}(s_0,0,0) = (s_0,0,0), \qquad \decidertypedfunction_{1,s_0}^{\text{Dline}}(0,s_1,0) = (0,0,0)
	\end{equation*}
	for all $x_0 \in V_0$. For the second level, for all $x_0 \in V_0$ , $x_1 \in V_1$, we define the second level linear function for $\decidertypedfunction^{\text{Dline}}$ as
	\begin{equation*}
		\decidertypedfunction_{2, x_0 + x_1}^{\text{Dline}}(0,0, s_2) = (0,0,\Nullline{e_{\binaryinv{x_0}}}{s_2}).
	\end{equation*}
	\item Let $\hat{j}$ be the same as the definition above. Recall that $\pi_{\leq j}^m$ refers to the zero-out map for the basis element $e_1 \cdots e_m$ in $\bF_{2^p}^m$ and let $v = \pi_{\leq \binaryinv{\hat{j}}}^m(s_0)$.  Define $\decidertypedfunction^{\text{Aline}}$ to be the third level CL function
	\begin{equation*}
		\decidertypedfunction^{\text{Aline}}(s_0, s_1, s_2) = (s_0 , v , \Nullline{v}{s_2}).
	\end{equation*}
	 In this case, $(v,\Nullline{v}{s_2})$ corresponds to the diagonal line $D_{\binaryinv{s_0}, v}$ which passes through $s_2$, and with the last $m - \binaryinv{\hat{j}}$ coordinates being zero. $\decidertypedfunction^{\text{Aline}}$ can be realized as a third level CL function with registers $\{V_i\}_{i \in [3]}$ in the following way: We define
	\begin{equation*}
		\decidertypedfunction_{0,0}^{\text{Aline}}(s_0,0,0) = (s_0,0,0),
	\end{equation*}
	to be the $0$th level linear function for $\decidertypedfunction^{\text{Aline}}$. Since each $\pi_{\leq j}^m$ is a linear function, we define the first level linear function as 
	\begin{equation*}
		\decidertypedfunction_{1,x_0}^{\text{Aline}}(0,s_1,0) = (0,\pi_{\binaryinv{x_0}}(s_1),0),
	\end{equation*}
	for all $x_0 \in V_0$. We define the second level linear function for $\decidertypedfunction$ as
	\begin{equation*}
		\decidertypedfunction_{2, x_0 + x_1}^{\text{Aline}}(0,0, s_2) = (0,0,\Nullline{x_1}{s_2})
	\end{equation*}
	for all $x_0 \in V_0$ and $x_1 \in V_1$.
\end{itemize}
This shows that $\cG$ is typed CL samplable. In the case where $\frac{\log(m)}{p} \not \in \bN$, we can simply treat the space $\bF_{2^p}^{m'}$ as a $\{0,1\}^{p \cdot m'}$ bit string using the canonical representation, and only use the first $\log(m)$ bits to define the CL function. Finally, we use the detyped transformation and apply~\Cref{lem:detypeandsync} to complete the proof of this lemma. 
\end{proof}

When discussing quantum low-individual degree test in this paper, we refers to the version which is CL samplable given by the above lemma. This version of the quantum low-individual degree test still retains the soundness property from~\Cref{thm:soundnessQLDT} (by changing the $a$ in the aforementioned theorem to the constant $ (4 * 4^2 + 16^4)^b \cdot a = (65600)^b \cdot a$). 

We remark that the only time we take advantage of the structure of $\bF_{2^p}^m$ (instead of treating it as a bit string of $\{0,1\}^{pm}$) when using the CL distribution is to sample a diagonal line intersecting the point $s$ in the quantum low-individual degree test. As mentioned previously, the diagonal line test was not used in the classical low-degree test, and the only purpose conceptually is to enforce a single commutation relationship within the proof of quantum soundness (see \cite[Lemma 6.1]{jiQuantumSoundnessTesting2022} for more details). It would be interesting to see if the quantum soundness of the quantum low-individual degree test still holds without the diagonal line test, as this would allow us to show the compression theorem for a simpler class of question distribution. This also shows that only Pauli $X$ and $Z$ measurements on fixed EPR pairs are sufficient for the gap compression theorem.

\begin{CJK*}{UTF8}{gbsn}
\section{Interactive proof systems and the gap compression theorem} \label{sec:MIPcompressibility}
In this section, we formally define the notion of a conditional linear verifier, which is the set of possible $\MIP^*$/$\MIPco$ protocols that we can show to be weakly compressible. We start this section by formally defining the notion of $\MIP^*$ and $\MIPco$ below.

\subsection{Interactive proof systems with entanglement}
In this section, we define the complexity classes $\MIP^*$ and $\MIPco$ more formally. Recall from the introduction that $\MIP$ stands for \textit{multi-prover interactive proof system}, the class of languages $\Language$ decidable by a probabilistic polynomial-time classical verifier when given (classical) interacting with computationally unbounded and non-communicating provers (i.e. the provers cannot talk to each other). In this model, the verifier can interact with multiple provers and may interact with the provers through multiple rounds of interactions. The verifier might adapt his questions based on the previous round of interactions and may leverage the lack of communication between the provers to ``cross-interrogate" them. If $z \in \Language$, the verifier can formulate an interaction such that the prover can provide enough evidence to convince the verifier to accept with probability $1$. On the other hand, if $z \not\in \Language$, the verifier can also formulate an interaction which ensures that the provers can only convince the verifier with probability at most $\frac{1}{2}$ to accept the given instance\footnote{The original formulation is $\geq \frac{2}{3}$ if $x \in \Language$ and $\leq \frac{1}{3}$ otherwise. However, we remark this is equivalent to the formulation given due to sequential repetition.}). As shown in~\cite{babaiNondeterministicExponentialTime1991a}, the computational power of $\MIP$ is equivalent to $\NEXP$, and this can be achieved with just a one-round interaction with two provers. 

In this paper, we consider two variants of $\MIP$ where the provers are still non-communicating, but share entanglement among them. We denote the variant where the provers share the tensor product model of entanglement as $\MIP^*$ and the commuting operator model of entanglement as $\MIPco$. In this paper, we focus on the variant of $\MIP^*$ and $\MIPco$ with two provers and one-round of interactions since this is sufficient for proving lower bounds. The two provers one-round $\MIP^*$ protocol (resp. $\MIPco)$) is denoted as $\MIP^*(2,1)$ (resp. $\MIPco(2,1)$) in the literature, and for simplicity of notation, unless otherwise specified, we drop $(2,1)$ when discussing $\MIP^*(2,1)$ (resp. $\MIPco(2,1)$). In this paper, we also work with $\MIP$ with completeness 1 and soundness $\frac{1}{2}$, meaning that there exists a verifier behaviour such that if $z \in \Language$, then the verifier accepts with probability 1, and if $z \not\in \Language$, then the verifier accepts with probability at most $\frac{1}{2}$. The completeness 1, soundness $\frac{1}{2}$ $\MIP^*$ protocol (resp. $\MIPco)$) is denoted as $\MIP^*_{1, \frac{1}{2}}$ (resp. $\MIPco_{1, \frac{1}{2}}$) in the literature, and similarly, we drop the subscript for the simplicity of notation. For a more general definition on $\MIP^*$, we refer the readers to~\cite[Section 6.1]{vidickQuantumProofs2016}. We formally give the definitions for \gls{MIPstar} and \gls{MIPco} used in this paper below. 
\begin{definition}[Multi-prover proof system with entanglement]  \label{def:MIPdef}
	Let $t \in \{*, co\}$. A language $\mathtt{L}$ is in $\MIP^t$ if there exist a pair of probabilistic (potentially multi-input) Turing machines $(\samplerTM, \deciderTM)$ such that $\TIME_{\samplerTM}(z) = \TIME_{\deciderTM}(z)  = O(\poly(|z|))$ for all $z \in \{0,1\}^*$. Furthermore, there exists an infinite sequence of games $\cG_{z} = (\cX_z, \cA_z, \mu_z, D_z)$ indexed by $z \in \{0,1\}^n$ and two increasing polynomial functions $x(n), a(n): \bN \rightarrow \bN$ with $\cX_z = \{0,1\}^{x(|z|)}$ and $\cA_z = \{0,1\}^{a(|z|)}$, such that
	\begin{itemize}
		\item \textbf{(Uniformity)} $\samplerTM(z, \text{sample})$ outputs a sample from the distribution $\mu_z$, and $\deciderTM(z, x, y, a, b) = D_z(x,y,a,b)$ for all $(x,y,a,b) \in \cX_z^2 \times \cA_z^2$.  
		\item \textbf{(Completeness)} If $z \in \cL$, then $\omega^t(\cG_z) = 1$. 
		\item \textbf{(Soundness)} If $z \not\in \cL$, then $\omega^t(\cG_z) \leq \frac{1}{2}$. 
	\end{itemize}
\end{definition}
The above definition is similar to the one given in~\cite[Definition 5.29]{jiMIPRE2022a}. Intuitively, the pair of Turing machines $(\samplerTM, \deciderTM)$ completely specifies the behaviour for the verifier. In comparison to the standard definition of an interactive proof system, we allow the sampler $\samplerTM$ to perform additional computation steps. This allows the verifier to extract additional information about the sampling distribution $\mu_z$. This would be useful in defining a \texttt{Compression} algorithm (as per~\Cref{def:weaklycompressibleproblem}). 

We remark that given the pair of Turing machines $(\samplerTM, \deciderTM)$, it is hard to extract the exact description of $\cG_{z} = (\cX_z, \cA_z, \mu_z, D_z)$ for a particular instance $z \in \{0,1\}^*$ by the definition above. However, as seen in the next subsection, we can hardcode $(\samplerTM, \deciderTM)$ to run other computational procedures in a way such that the description for $\cG_{z} = (\cX_z, \cA_z, \mu_z, D_z)$ can be properly extracted.  As observed in~\cite{cleveConsequencesLimitsNonlocal2004}, for $t \in \{*, co\}$, the complexity class $\MIP^{t}$ is complete with respect to the following decision problem. 
\begin{definition}[$(1,\frac{1}{2})$ non-local game value problem]  \label{def:entangledgameproblem}
	For $t \in  \{*, co\}$, the $(1,\frac{1}{2})$ $t$ non-local game value problem is a decision problem defined by the following two sets. 
	\begin{itemize}
		\item $\Language^{\MIP^t}_{\text{yes}} = \{\langle \cG\rangle | \, \omega^t(\cG) = 1\} $. 
		\item $\Language^{\MIP^t}_{\text{no}} = \{\langle \cG \rangle |\,  \omega^t(\cG) \leq \frac{1}{2}\} $. 
	\end{itemize}
\end{definition}
For clarity, we refer to the $(1,\frac{1}{2})$ $*$ non-local game value problem as the $(1,\frac{1}{2})$ tensor product value problem, and the $(1,\frac{1}{2})$ $co$ non-local game value problem as the $(1,\frac{1}{2})$ commuting operator value problem. Finally, we wish to give a notion of a ``uniform problem instance" for interactive proof systems. 

\begin{definition}[Uniform verifier sequence] \label{def:genseq}
	Let $t \in \{*, co\}$, and let $\cG_{n} = (\cX_n, \cA_n, \mu_n, D_n)$ be a sequence of games. A verifier sequence $\verifiersequence = (\samplerTM, \deciderTM)$ is a pair of Turing machines such that $\samplerTM(n, \text{sample})$ outputs a sample from the distribution $\mu_n$, and $\deciderTM(z, x, y, a, b) = D_n(x,y,a,b)$ for all $(x,y,a,b) \in \cX_n^2 \times \cA_n^2$.  Furthermore, we say that $\verifiersequence$ runs in $O(\mathbf{f}(n))$ time if
	\begin{equation*}
		\TIME_{\samplerTM} = \TIME_{\deciderTM} = O(\mathbf{f}(n)). 
	\end{equation*} 
\end{definition}

In the above definition, the runtime for $\verifiersequence$ might initially seem different from the runtime defined for a uniform problem sequence used in~\Cref{sec:Compression}. However, to see the similarity, one should intuitively think of $(\samplerTM, \deciderTM)$ as a Turing machine which can be used to generate a description of the sequence non-local games $\cG_n$ in the above definition. Since we do not make any assumption on the implementation on the Turing machine $ (\samplerTM, \deciderTM)$, any argument made in~\Cref{sec:Compression} still applies to the above definition. 

Audiences with no prior background in complexity might be confused about the reason for representing a sequence of non-local games as a uniform Turing Machine instead of the description of the game itself. By representing a sequence of games as uniform Turing Machines, one can convert the computation step of deciding whether the verifier accepts into an instance of a 3-SAT formula via the well-known Cook\text{-}Levin encoding. This is crucial for initiating the Answer reduction step of the compression procedure. 


\subsection{Conditionally Linear verifier} \label{sec:CLverifier}

Before introducing the gap compression theorem for non-local games, we first define the type of games that can be shown to be weakly compressible. Recall from~\Cref{sec:condlinearfun} that a CL samplable game is a non-local game with CL distribution as the sampling procedure for the game. 

For $t \in  \{*, co\}$ and constant $k \in \bN$, we define a synchronous $k$-th level CL samplable $\MIP^t$ ($k\text{-}\CLMIP^t$) as the complexity class $\MIP^t$ except restricted to synchronous games which are also $k$-th level CL samplable. This complexity class is complete with respect to the following decision problem.  
\begin{itemize}
	\item $\Language^{k\text{-}\CLMIP^t}_{\text{yes}} = \{\langle \cG\rangle | \, \omega^t(\cG) = 1, \, \cG\text{ is a synchronous $k$-th level CL samplable game}\} $, 
	\item $\Language^{k\text{-}\CLMIP^t}_{\text{no}} = \{\langle \cG \rangle | \, \omega^t(\cG) \leq \frac{1}{2}, \, \cG\text{ is a synchronous $k$-th level CL samplable game}\}$. 
\end{itemize}
Showing that $k\text{-}\CLMIP^{co}$ (resp. $k\text{-}\CLMIP^*$) being $\coRE$-complete (resp. $\RE$-complete) implies that $\coRE \subseteq \MIPco$ (resp. $\RE \subseteq \MIP^*$ ). 

We are now ready to describe a notion of a Conditionally Linear verifier. Intuitively, one can think of the CL verifier as a more structured version of a uniform verifier for synchronous $k$-th level CL samplable games. We formally introduce the notion of a CL verifier below.

\begin{definition}[Conditionally Linear verifier] \label{def:CLverifier}
	Let $\text{\gls{CLlevelfn}}, \text{\gls{CLcardfn}}, \text{\gls{Cfieldfn}}: \bN \rightarrow \bN$, where the range of $\mathbf{p}(n)$ maps integers to odd integers. Let $\gamesequence = \{\cG_n = (\cX_n, \cA_n, \mu_n, D_n)\}_{n \in \bN}$ be an infinite sequence of games indexed by $n \in \bN$. Each $\mu_n$ is a $(\mathbf{k}(n), \mathbf{m}(n), \mathbf{p}(n))$ CL distribution defined over two $\mathbf{k}(n)$-level CL functions $\decidertypedfunction^{0,n}, \decidertypedfunction^{1,n}$ with registers $\{ V^n_j \}_{j \in [\mathbf{k}(n)]}$ as defined in~\Cref{def:CLdistribution}, and $\cA_n = \{0,1\}^*$ (in this case, $\cX_n = \bF_{2^{\mathbf{p}(n)}}^{\mathbf{m}(n)}= \{0,1\}^{\mathbf{p}(n) \cdot \mathbf{m}(n)}$ by definition).
	
	A $\mathbf{k}(n)$ level CL verifier \gls{Verifiersequence} is a tuple $(\text{\gls{Samplersequence}}, \text{\gls{Decidersequence}})$, where $\samplerTM_{\verifiersequence}$ is a five-input Turing machine,  and $\deciderTM_{\verifiersequence}$ is a six-input Turing machine, such that, for all $n \in \bN$
	\begin{itemize}
		\item $\samplerTM_{\verifiersequence}(n, \text{Parameter}) = (\mathbf{k}(n), \mathbf{m}(n) , \mathbf{p}(n))$. 
		\item $\samplerTM_{\verifiersequence}(n, \text{Divide},s) = (s_0, \cdots s_{\mathbf{k}(n)-1})$, for all $s \in V^n$, where $s_j \in V_j^n$ and $\sum_{j \in [\mathbf{k}(n)]} s_j = s$.
		\item $\samplerTM_{\verifiersequence}(n, \text{Function},p,j,s,x) = \decidertypedfunction_{j, s}^{p,n}(x)$, for all $j \in [\mathbf{k}(n)]$, $s \in V_{< j}^n$,  $x \in V_j^n$, and $p\in \{0,1\}$. Where recall $\{\decidertypedfunction_{j, s}^{p,n}\}_{s \in V^n_{< j}}$ are the $j$th level linear function for $\decidertypedfunction^{p,n}$ (where we associate $\decidertypedfunction^{0,n} = \decidertypedfunction^{A,n}$ and $\decidertypedfunction^{1,n} = \decidertypedfunction^{B,n} $ ).
		\item $\deciderTM_\verifiersequence(n,x,y,a,b)= D_n(x,y,a,b)$ for all $(x,y) \in \cX_n^2$ and $(a,b) \in \cA_n^2$
	\end{itemize}
	If $\mathbf{k}(n) = k$ for some constant $k \in \bN$, then we simply call $\verifiersequence$ a $k$-th level CL verifier. We say that $\verifiersequence$ has a sampling complexity of $O(\mathbf{f}(n))$ if 
	\begin{equation*}
		\mathbf{k}(n) \cdot \mathbf{m}(n) \cdot \mathbf{p}(n) = \TIME_{\parameterTM_{\verifiersequence}}(n)= \TIME_{\samplerTM_{\verifiersequence}}(n) = O(\mathbf{f}(n)),
	\end{equation*}
	and we say that $\verifiersequence$ has a verification complexity of $O(\mathbf{g}(n))$ if $\TIME_{\deciderTM_{\verifiersequence}} = O(\mathbf{g}(n))$.
\end{definition}

We use the term CL verifier as a shorthand for any $\mathbf{k}(n)$ verifier. We refer to a CL verifier $\verifiersequence$ as a \textit{$k$-th level synchronous Conditionally Linear verifier} if every game $\cG_n = \{ \cX_n, \cA_n, \mu_n, D_n\}$ generated by $\verifiersequence$ is a $k$-th level CL samplable synchronous game. To abuse notation, we write $\samplerTM_{\verifiersequence}(n, \text{Parameter}) \leq O(\mathbf{f}(n))$ as a shorthand for $\mathbf{m}(n) \cdot \mathbf{p}(n) \leq O(\mathbf{f}(n))$, where $\mathbf{m}(n)$ and $\mathbf{p}(n)$ are the output for  $\samplerTM_{\verifiersequence}(n, \text{Parameter})$. In the above definition, we refer to $\samplerTM_{\verifiersequence}$ as a CL sampler and $\deciderTM_\verifiersequence$ as a CL decider, and we drop the subscript $\verifiersequence$ if the underlying game sequence is clear from context. We remark that although in the above definition, the answer set $\cA_n = \{0,1\}^*$ is an infinite set. However, since the decider $\deciderTM$ is always assumed to be time-bounded, $\deciderTM$ can read at most $\TIME_{\deciderTM}(n)$ bits; thus $\cA_n$ is, in practice, always a finite set. 

Although the above definition does not explicitly give the computational procedure $\samplerTM_{\verifiersequence}(n, \text{sample})$ akin to~\Cref{def:genseq}, the computation procedure can be implemented in the following manner:
\begin{enumerate}
	\item The verifier first runs $\samplerTM_{\verifiersequence}(n, \text{Parameter})$ to obtain $\mathbf{k}(n)$ and $\mathbf{m}(n) \cdot \mathbf{p}(n)$, then the verifier samples a random seed $s \in \{0,1\}^{\mathbf{m}(n) \cdot \mathbf{p}(n)}$. 
	\item The verifier then runs $\samplerTM_{\verifiersequence}(n, \text{Divide},s) = (s_1 \cdots , s_{\mathbf{k}(n)})$ to partition $s$ into linear components within $\{V_j^n\}_{j\in [\mathbf{k}(n)]}$.
	\item For $p \in \{0,1\}$, the verifier performs the following: 
	\begin{enumerate}
		\item The verifier first computes $x_0^p =  \samplerTM_\decidertypedfunction(n, \text{Function},p,0,0, s_0)$, the output for the $0$-th linear function for $\decidertypedfunction^{p,n}$.
		\item For $1 \leq j < \mathbf{k}(n)$, the verifier computes $x_j^p =  \samplerTM_\decidertypedfunction(n, \text{Function},p,j,x_{j-1}^p, s_j)$, the output for the $j$th linear function for $\decidertypedfunction^{p,n}$.
		\item Finally, the verifier computes $x^p = \sum_{j} x_j^i $ by adding up all the components together.
	\end{enumerate}
	\item The verifier returns $(x^0, x^1)$, as the question pair. 
\end{enumerate}
Step 1 takes $O(\log(\mathbf{p}(n)) + \log(\mathbf{m}(n))) = O(\mathbf{f}(n))$ time, Step 2 and 3 take $O(\mathbf{k}(n) \cdot \mathbf{m}(n) \cdot \mathbf{p}(n)) = O(\mathbf{f}(n))$ time by~\Cref{lem:compfinitefield}. This implies that $\samplerTM_{\verifiersequence}(n, \text{sample})$ runs in time $O(\mathbf{f}(n))$, consistent with the definition given in~\Cref{def:genseq}. This shows that a Conditionally Linear verifier is a specific instance of a uniform verifier sequence for CL samplable games.

\subsection{Compression theorem for interactive proof system}
We introduce the gap-compression theorem for non-local games and show how the theorem can be used to lower-bound the complexity of $\MIP^*$ and $\MIPco$ in this section. 

\begin{theorem}[Gap compression for non-local games] \label{thm:gappedcompression}
	For all constants $\alpha,k \in \bN$, there exists an algorithm $\texttt{Gapcompress}_{\alpha, k}$ that takes the input a pair of Turing machines $(\samplerTM, \deciderTM)$. $\texttt{Gapcompress}_{\alpha, k}$ outputs a tuple of Turing machines $\verifiersequence^{\Compressgame} = (\samplerTM^{\Compressgame}, \deciderTM^{\Compressgame})$ such that the following holds:
	
	There exists an integer $\gamma = O(\text{poly}(\alpha, k)) $ such that, for models $t \in \{*, co\}$
	\begin{enumerate}
		\item (Runtime): $\TIME_{\texttt{Gapcompress}_{\alpha, k}}(\samplerTM, \deciderTM) = O(\poly(\alpha), |\samplerTM|, |\deciderTM|)$. 
		\item (Independence of the sampler): The Turing machine $\samplerTM^{\Compressgame}$ only depends on the parameter $\alpha$ and $k$ and is a sampler for a synchronous 7-th level CL verifier. The Turing machine $\deciderTM^{\Compressgame}$ is a decider for a synchronous 7-th level CL verifier. $\verifiersequence^{\Compressgame}$ is a synchronous 7-th level CL verifier sequence for an infinite sequence of games $\gamesequence^{\Compressgame} = \{\cG_n^{\Compressgame} \}_{n \in \bN}$. 
		\item (Complexity bounds for the output) $\verifiersequence^{\Compressgame}$ has sample complexity and verification complexity of $O(\log(n)^\gamma)$.
	\end{enumerate}
	Furthermore, if the input $\verifiersequence = (\samplerTM, \deciderTM)$ is a synchronous $k'$th level CL verifier for the infinite sequence of synchronous games $\gamesequence = \{\cG_n\}_{n \in \bN}$ for some constant $k' < k$, and there exists a constant $n_0 \in \bN$ such that for all $n \geq n_0$
	\begin{equation} \label{eq:gapcompressinputcondition}
		\max\left\{ \TIME_{\samplerTM}, \TIME_{\deciderTM}\right\} \leq n^\alpha. 
	\end{equation}
	Then there exists a constant $n_0^{\Compressgame} = \poly(\gamma, n_0)$ such that for all $n \geq n_0^{\Compressgame}$
	\begin{itemize}
		\item (Completeness) If there exists a perfect oracularizable strategy for $\cG_n$ in model $t$, then there exists a perfect oracularizable strategy for $\cG^{\Compressgame}_n$ in model $t$. 
		\item (Soundness) 
		\begin{equation*}
			\omega^t(\cG_n) \leq \frac{1}{2} \Longrightarrow \omega^t(\cG^{\Compressgame}_n) \leq \frac{1}{2}
		\end{equation*}
	\end{itemize}
\end{theorem}
We give a proof for~\Cref{thm:gappedcompression} in~\Cref{sec:proofcompressiongame}. By clause 2 of the above theorem, for any $\alpha, k$, the output for the algorithm $\texttt{Gapcompress}_{\alpha, k}$ will always be a synchronous 7-th level CL verifier (even if the input $(\samplerTM, \deciderTM)$ might not be a valid CL verifier). Hence, as a corollary,~\Cref{thm:gappedcompression} shows the following
\begin{corollary} \label{cor:KMIPcomp}
	For $t \in \{*, co\}$ and constant $k \in \bN$ such that $k \geq 7$, the complexity class $k\text{-}\CLMIP^t$ is weakly compressible.
\end{corollary}
We remark that since the soundness gap in~\Cref{thm:gappedcompression} is controlled by the parallel repetition theorem, the soundness condition can actually be proven for any soundness value $c \in (0,1)$ (instead of $\frac{1}{2}$). Before we continue, we first give the definition for a trivial synchronous accepting/rejecting game for both $\MIP^*$ and $\MIPco$ below.  
\begin{definition}[Accepting/Rejecting game] \label{def:syncrejgame}
	We define the synchronous accepting game to be the game $\text{\gls{acceptgames}} = (\cX, \cA, \mu, D^{\text{accept}})$, where $\cX= \{\star\}$ and $\cA = \{0,1\}^*$, 
	\begin{equation*}
		D^{\text{accept}}(\star,\star, a,b) = \delta_{a,0} \delta_{b,0}
	\end{equation*}
	for all $a,b \in \{0,1\}^*$ (i.e. the prover automatically wins if both provers return 0). We define the synchronous \textit{rejecting} game to be the game $\text{\gls{rejectgames}} = (\cX, \cA, \mu, D^{\text{reject}})$, where $\cX= \{0,1\}$, $\mu(1,0) =\mu(0,1) =  \frac{1}{3}$, $\mu(0,0)  = \mu(1,1) = \frac{1}{6}$, and $\cA = \{0,1\}^*$, 
	\begin{equation*}
		D^{\text{reject}}(x,y, 0,0) = \delta_{x,y},
	\end{equation*}
	and $D^{\text{reject}}(x,y, 0,0) = 0$ for all other $a,b \in \{0,1\}^*$.
\end{definition}

For model $t \in \{*, co\}$, we have $\omega^{t}(D^{\text{accept}}) =1$ and $\omega^{t}(D^{\text{reject}}) = \frac{1}{3}$. Both games are trivially samplable via a CL distribution for any level and always computable in constant time. Intuitively, this is the trivial synchronous game which is in the yes/no case for the $(1,1/2)$ tensor product/commuting operator value problem defined in~\Cref{def:entangledgameproblem} for both $t \in \{*, co\}$. This shows that for $t \in  \{*, co\}$, both $\Language^{k\text{-}\CLMIP^t}_{\text{yes}}$ and $\Language^{k\text{-}\CLMIP^t}_{\text{no}}$ are non-empty. Based on~\Cref{cor:KMIPcomp} we show the main theorem of this paper. 

\subsubsection{$\MIPco=\coRE$} \label{sec:MIPcocoREproof}

In order to show the main theorem of this paper, we first show the following lemma. 
\begin{lemma} \label{lem:MIPcoincoRE}
	$\MIPco \in \coRE$
\end{lemma}
The above lemma also shows that $\Language^{k\text{-}\CLMIP^{co}}_{\text{yes}} \subseteq \Language^{\MIPco}_{\text{yes}} \in \coRE$. To show this lemma, we recall the well-known NPA Hierarchy algorithm developed in~\cite{navascuesConvergentHierarchySemidefinite2008b}. We summarize the functionality of the NPA Hierarchy below. 
\begin{theorem}[Output for the NPA Hierarchy~\cite{navascuesConvergentHierarchySemidefinite2008b}] \label{thm:NPA}
	For any integer $n \in \bN$ and $\cG = (\cX, \cA, \mu, D)$, there exists a terminating algorithm $\texttt{NPAHierarchy}$ such that $\texttt{NPAHierarchy}(\cG, n) = \eps_n \in \bR$ with 
	\begin{equation*}
		\lim_{n \rightarrow\infty} \eps_n = \omega^{co}(\cG)
	\end{equation*}
\end{theorem}
Based on~\Cref{thm:NPA}, for $\delta \in [0,1]$ we define the $\texttt{Searchfromabove}_{\delta}$ algorithm below. 

\vspace{10pt}
\IncMargin{1em}
\begin{algorithm}[H]
	\DontPrintSemicolon
	
	\textbf{Input}: $\cG = (\cX, \cA, \mu, D)$
	
	If the input $\cG$ is not a valid description of the game, terminate and \textbf{return} ``Error"; 
	
	Set $n = 1$
	
	\While{True}{
		
		Compute $\eps_n = \texttt{NPAHierarchy}(\cG, n) $

		\uIf{$\eps_n < \delta$}{
			\textbf{Return } False ; 
		}

		n = n+1 ;
		
	}	
	\caption{The description for \gls{Searchfromabove}.}
	\label{pseu:searchfromabove}
\end{algorithm}\DecMargin{1em}
\vspace{10pt} 
$\texttt{Searchfromabove}_{\delta}$ gives a systematic way to generate a sequence of upper bounds for $\omega^{co}(\cG)$ to check whether $\omega^{co}(\cG) \leq \delta$. By definition, the algorithm runs forever if $\omega^{co}(\cG) \geq \delta$. Hence the algorithm $\texttt{Searchfromabove}_{1}$ directly implies~\Cref{lem:MIPcoincoRE}.

By combining~\Cref{cor:KMIPcomp}, ~\Cref{thm:compresscosRE} and the fact that $\Language^{k\text{-}\CLMIP^{co}}_{\text{yes}} \in \coRE$  this shows that $k\text{-}\CLMIP^{co} = \coRE$ and hence $\coRE \in \MIPco$. This, along with~\Cref{lem:MIPcoincoRE}, shows the main theorem of this paper. 
\begin{corollary} \label{cor:MIPcocoRE}
	$\MIPco=\coRE$
\end{corollary}

We remark that if we specifically tailored~\Cref{pseu:REcompressibleproof} specifically for $ k\text{-}\CLMIP^{co}$, we have the following pseudocode, where in the pseudocode below, $\cG_{C_0}$ is represented as~\Cref{pseu:verifiersequencemipco} with $C_0$ being hard coded in. \Cref{pseu:verifiersequencemipco} is also a more refined version of~\cite[Pseudocode 4]{mousaviNonlocalGamesCompression2022}. 

\begin{figure}[!t]
	\vspace{10pt}
	\IncMargin{1em}
	\begin{algorithm}[H]
		\DontPrintSemicolon
		
		\textbf{Input}: Integer $n$. 
		
		Run \Turingmachinehalt for $n$ steps. If \Turingmachinehalt halts in the given steps, return $\cG^{\text{reject}}$.
		
		Compute the description of $\verifiersequence$.
		
		Compute the description of $\cG_{C_0} = \verifiersequence(C_0)$, the $C_0$th game of the CL verifier $\verifiersequence$.
		
		Simulate $\texttt{Searchfromabove}_1$ specified in~\Cref{pseu:searchfromabove} with $\cG_{C_0}$  as the input for $\max\{0, n - C_0\}$ steps. If $\texttt{Searchfromabove}_1$ halts in line 6 (of~\Cref{pseu:searchfromabove}) in the given steps, return $\cG^{\text{accept}}$. 
		
		Apply $\texttt{Gapcompress}_{\alpha, 7}$ on the verifier $\verifiersequence$ to obtain $\verifiersequence^{\text{comp}}$.
		
		Compute $\cG^{\text{comp}}_{n+1} = \verifiersequence^{\text{comp}}(n+1)$ and execute the game $\cG^{\text{comp}}_{n+1}$ with the two provers.
		
		\caption{The description for $\verifiersequence$ which can be used to show that $\coRE \subseteq k\text{-}\CLMIP^{co}$. \Turingmachinehalt is the instance of the halting problem for the reduction.}
		\label{pseu:verifiersequencemipco}
	\end{algorithm}\DecMargin{1em}
	\vspace{10pt} 
\end{figure}

\subsubsection{$\MIP^*=\RE$} \label{sec:MIPREproof}

Similarly to the previous section, we need to first show the following lemma. 
\begin{lemma} \label{lem:MIPstartbound}
	$\MIP^* \in \RE$
\end{lemma}
By a similar reason as above, the above lemma shows that $\Language^{k\text{-}\CLMIP^*}_{\text{no}} \in \coRE$. We show the above lemma by recalling a well-known fact about the quantum tensor correlations in the literature.
\begin{fact}[Discretization of quantum tensor correlations] \label{fact:disCorrelations}
	For any integer $n \in \bN$ and $\eps > 0$, there exists a terminating algorithm to search over $C_q^n$ to generate a finite subset $C_{\eps}^n \subseteq C_q^n$ such that for all correlations $C \in C_q^n$, there exists a correlation $C' \in C_{\eps}^n$ such that 
	\begin{equation*}
		\sum_{x,y,a,b} |C_{x,y,a,b} - C'_{x,y,a,b} | \leq \eps.
	\end{equation*}
\end{fact}

Based on~\Cref{fact:disCorrelations}, for $\delta \in [0,1]$, we define the algorithm $\texttt{searchfrombelow}_{\delta}$ below. 

\vspace{10pt}
\IncMargin{1em}
\begin{algorithm}[H]
	\DontPrintSemicolon
	
	\textbf{Input}: $\cG = (\cX, \cA, \mu, D)$
	
	If the input $\cG$ is not a valid description of the game, terminate and \textbf{return} ``Error" ; 
	
	Set $n = 1$
	
	\While{True}{
		
		Compute the finite subset $C_{1/n}^n$ defined by~\Cref{fact:disCorrelations}. 
		
		\For{$C' \in C_{\frac{1}{n}}^n$}{
			
			Compute $\eps_n = \Ex_{(x,y) \sim \mu} \sum_{(a,b) \in \cA^2} C'_{x,y,a,b} D(x,y,a,b)$
			
			\uIf{$\eps_n > \delta$}{
				\textbf{Return} True ; 
			}
			
		}
		
		n = n+1 ;
		
	}	
	\caption{The description for \gls{Searchfrombelow} for $\eps \in [0,1)$.}
	\label{pseu:searchfrombelow}
\end{algorithm}\DecMargin{1em}
\vspace{10pt} 

Intuitively, $\texttt{searchfrombelow}_{\delta}$ gives a systematic way to search through the correlation set $C_q$ to find a strategy which validates that $\omega^*(\cG) > \delta$ (which might be impossible depending on $\cG$).

We are now ready to show~\Cref{lem:MIPstartbound}.
\begin{proof}
	To show that $\MIP^*\subseteq \RE$, we need to show that the algorithm above halts in the Yes case for the non-local game value problem. Hence, let $\cG$ be a game such that $\omega^*(\cG) = 1$ and consider the algorithm $\texttt{searchfrombelow}_{0.5}$ running on $\cG$. Combining the definition $\omega^*$ and $C_q = \bigcup_{n \in \bN^{+}} C_q^n$, we see that there must exist some $n \in \bN$ and $C \in C_q^n$ such that 
	\begin{equation*}
		\Ex_{(x,y) \sim \mu} \sum_{(a,b) \in \cA^2} C_{x,y,a,b} D(x,y,a,b) \geq 0.75. 
	\end{equation*}
	Hence by the definition, there exists a constant $C' \in C_{\frac{1}{n}}^n$ such that 
	\begin{equation*}
		\Ex_{(x,y) \sim \mu} \sum_{(a,b) \in \cA^2} (C_{x,y,a,b} - C'_{x,y,a,b})  D(x,y,a,b) \leq \sum_{x,y,a,b} |C_{x,y,a,b} - C'_{x,y,a,b} | \leq \frac{1}{n},
	\end{equation*}
	where the inequality follows since $\mu(x,y),D(x,y,a,b) \leq 1$. By combining the two inequalities
	\begin{equation*}
		\Ex_{(x,y) \sim \mu} \sum_{(a,b) \in \cA^2} C'_{x,y,a,b} D(x,y,a,b) > 0.5.
	\end{equation*}
	This completes the proof of the lemma. 
\end{proof}
By a similar proof as the above lemma, we see that $\texttt{searchfrombelow}_{\eps}(\cG)$ terminates iff $\omega^*(\cG) > \eps$. We remark that $\texttt{searchfrombelow}_{\eps}$ does not work for the set of commuting operator correlations, as there is no way to discretize the set of commuting operator strategies based on the dimension of the Hilbert space. 

The algorithm $\texttt{searchfrombelow}_{0.5}$ is precisely the algorithm needed to show~\Cref{lem:MIPcoincoRE}. By combining~\Cref{cor:KMIPcomp}, ~\Cref{thm:compressRE} and the fact that $\Language^{k\text{-}\CLMIP^*}_{\text{no}} \in \coRE$, this shows that $k\text{-}\CLMIP^* = \RE$ and hence $\RE \in \MIP^*$. This, along with~\Cref{lem:MIPstartbound}, shows the main theorem of this paper. 
\begin{corollary} \label{cor:MIPstarRE}
	$\MIP^*=\RE$
\end{corollary}

In a similar vein as~\Cref{cor:MIPcocoRE}, we remark that~\Cref{cor:MIPstarRE} can be proven using~\Cref{pseu:verifiersequencemipstar} below. 

\vspace{10pt}
\IncMargin{1em}
\begin{algorithm}[H]
	\DontPrintSemicolon
	
	\textbf{Input}: Integer $n$. 
	
	Run \Turingmachinehalt for $n$ steps. If \Turingmachinehalt halts in the given steps, \textbf{return} $\cG^{\text{accept}}$.
	
	Compute the description of $\verifiersequence$.
	
	Compute the description of $\cG_{C_0} = \verifiersequence(C_0)$, the $C_0$th game of the CL verifier $\verifiersequence$.
	
	Simulate $\texttt{Searchfrombelow}_{0.5}$ specified in~\Cref{pseu:searchfrombelow} with $\cG_{C_0}$ as the input for $\max\{0, n - C_0\}$ steps. If $\texttt{Searchfrombelow}_{0.5}$ halts in line 6 (of~\Cref{pseu:searchfrombelow}) in the given steps, \textbf{return} $\cG^{\text{reject}}$. 
	
	Apply $\texttt{Gapcompress}_{\alpha, 7}$ on the verifier $\verifiersequence$ to obtain $\verifiersequence^{\text{comp}}$.
	
	Compute and \textbf{return}$\cG^{\text{comp}}_{n+1} = \verifiersequence^{\text{comp}}(n+1)$.
	
	\caption{The description for $\verifiersequence$ which generates the game sequences $\{\cG_n\}_{n \in \bN}$ for the proof of $\RE \subseteq \MIP^*$. \Turingmachinehalt is the instance of the halting problem for the reduction.}
	\label{pseu:verifiersequencemipstar}
\end{algorithm}\DecMargin{1em}
\vspace{10pt} 

\subsubsection{Finding an explicit separation between $C_{q}$ and $C_{qc}$ is $\RE$-complete} \label{sec:compproofoutline}

Finally, we give the last application of~\Cref{thm:gappedcompression}, which is to show that finding a Bell test between the tensor product model and commuting operator model is $\RE$-complete. In order to show this, we first give a formal definition of this problem using a decision problem. 
\begin{definition}[The $\delta$-Bell test separation decision problem]
	Given a constant $\delta \in (0,1)$, the $\delta$-Bell test separation decision problem $\Decisionproblem_{\delta\text{-bell}}$ is defined by the following two sets. 
	\begin{itemize}
		\item $\Language^{\Decisionproblem_{\delta\text{-bell}}}_{\text{yes}} = \{\langle \cG\rangle | \, \omega^{*} (\cG) = \omega^{co}(\cG) \} $. 
		\item $\Language^{\Decisionproblem_{\delta\text{-bell}}}_{\text{no}} = \{\langle \cG \rangle |\, |\omega^{*} (\cG) - \omega^{co}(\cG)| \geq \delta \} $. 
	\end{itemize}
\end{definition} 

The above problem is already known to be $\RE$-hard prior to this work using the algorithm $\texttt{Searchsamevalue}_{\delta}$, which we give in~\Cref{pseu:searchsamevalue} below for completeness. As shown in~\cite[Theorem 12.10]{jiMIPRE2022a} the set $\Language^{\Decisionproblem_{\frac{1}{2}\text{-bell}}}_{\text{no}}$ is non-empty, as there exists a (synchronous $12$-th level CL samplable game) such that $ |\omega^{*} (\cG) - \omega^{co}(\cG)| \geq \frac{1}{2}$. Given this, we have the following theorem for the complexity of the $\frac{1}{2}$-Bell test separation problem.
\begin{theorem} \label{thm:samevalueREcomplete}
	The $\frac{1}{2}$-Bell test separation problem is $\RE$-complete
\end{theorem}

The above theorem follows from realizing that~\Cref{thm:gappedcompression} also shows that the $\frac{1}{2}$-Bell test separation problem is weakly compressible for synchronous games which are $7$-th level CL samplable with the ``yes" case consisting of games $\cG$ such that $\omega^*(\cG) = \omega^{co}(\cG) =1$, and the ``no"  case consisting of games $\cG$ such that $\omega^{co}(\cG) =1$ and $\omega^*(\cG) \leq \frac{1}{2}$ (where the $\cG$ in both cases are synchronous $12$-th level CL sample games). Since the argument follows similarly as the one given in~\Cref{sec:MIPREproof}, we do not give the details here. We remark that the above proof can be easily changed to any $\delta \in (0,1)$ since the soundness condition from \Cref{thm:gappedcompression} can be changed to hold for any $\delta \in (0,1)$. 
 
 \vspace{10pt}
 \IncMargin{1em}
 \begin{algorithm}[H]
 	\DontPrintSemicolon
 	
 	\textbf{Input}: $\cG = (\cX, \cA, \mu, D)$
 	
 	If the input $\cG$ is not a valid description of the game, terminate and \textbf{return} ``Error"; 
 	
 	Set $n = 1, \eps_1^{\text{lower}}  = 0$
 	
 	\While{True}{
 		
 		Compute the finite subset $C_{1/n}^n$ defined by~\Cref{fact:disCorrelations}.

 		\For{$C' \in C_{\frac{1}{n}}^n$}{
 			
 			Compute $\eps' = \Ex_{(x,y) \sim \mu} \sum_{(a,b) \in \cA^2} C'_{x,y,a,b} D(x,y,a,b)$
 			
 			\uIf{$\eps' > \eps_n^{\text{lower}}$}{
 				$\eps_n^{\text{lower}}  = \eps'$  ; 
 			}
 		}

 		Compute $\eps_n^{\text{upper}} = \texttt{NPAHierarchy}(\cG, n) $

 		\uIf{$|\eps_n^{\text{upper}} - \eps_n^{\text{lower}}| \leq \min{(\delta - \frac{1}{n}, 0)}$}{
 			\textbf{Return} True ; 
 		}

 		$n = n+1$;
 		
 		$\eps_n^{\text{lower}} = \eps_{n-1}^{\text{lower}}$;
 		
 	}	
 	\caption{The description for $\texttt{Searchsamevalue}_{\delta}$.}
 	\label{pseu:searchsamevalue}
 \end{algorithm}\DecMargin{1em}
 \vspace{10pt}

We remark that instead of using~\Cref{thm:compressRE}, ~\Cref{thm:samevalueREcomplete} can also be proven by considering the following uniform verifier sequence defined in~\Cref{pseu:searchsamevalue}. By using a similar argument as the proof of~\Cref{thm:compressRE}, one can infer that $\omega^*(\cG_{C_0}) < 0.5$ by using line 5 of~\Cref{pseu:samevalueREcompletealt} and $\omega^*(\cG_{C_0}) = 1$ by using line 6 of~\Cref{pseu:samevalueREcompletealt}. 

 \begin{algorithm}[H]
	\DontPrintSemicolon
	
	\textbf{Input}: Integer $n$. 
	
	Run \Turingmachinehalt for $n$ steps. If \Turingmachinehalt halts in the given steps, \textbf{return} $\cG^{\text{accept}}$.
	
	Compute the description of $\verifiersequence$.
	
	Compute the description of $\cG_{C_0} = \verifiersequence(C_0)$, the $C_0$th game of $\verifiersequence$.
	
	Simulate $\texttt{Searchfrombelow}_{0.5}$ specified in~\Cref{pseu:searchfrombelow} with $\cG_{C_0}$  as the input for $\max\{0, n - C_0\}$ steps. If $\texttt{Searchfrombelow}_{0.5}$ halts in line 6 (of~\Cref{pseu:searchfrombelow}) in the given steps. Return $\cG^{\text{accept}}$. 
	
	Simulate $\texttt{Searchfromabove}_1$ specified in~\Cref{pseu:searchfromabove} with $\cG_{C_0}$  as the input for $\max\{0, n - C_0\}$ steps. If $\texttt{Searchfromabove}_1$ halts in line 6 (of~\Cref{pseu:searchfromabove}) in the given steps, \textbf{return} $\cG^{\text{reject}}$. 
	
	Apply $\texttt{Gapcompress}_{\alpha, 7}$ on the verifier $\verifiersequence$ to obtain $\verifiersequence^{\text{comp}}$.
	
	Compute and \textbf{return}$\cG^{\text{comp}}_{n+1} = \verifiersequence^{\text{comp}}(n+1)$ and $\cG^{\text{comp}}_{n+1}$ with the two provers.
	
	\caption{An alternative game sequence for the proof of~\Cref{thm:samevalueREcomplete}. \Turingmachinehalt is the instance of the halting problem for the reduction.}
	\label{pseu:samevalueREcompletealt}
\end{algorithm}\DecMargin{1em}
\vspace{10pt}

  \Cref{thm:samevalueREcomplete} implies that it is impossible for \textbf{any} computer program to systematically find a Bell test to separate the quantum tensor product model from the quantum commuting operator model! However, if we have prior knowledge about whether a Turing machine halts (for example, the Turing machine that arises from Pseudocode which contains an infinite loop), we could construct a bell experiment that realizes such a separation using~\Cref{thm:samevalueREcomplete}, giving infinitely many bell experiments to test the separation between the tensor product model and the commuting operator model. Unfortunately, since these constructions rely on complexity techniques, this also implies that any experimental setup generated by using~\Cref{thm:samevalueREcomplete} would be impractical for experimental usage~\footnote{For reference, the explicit separation proven by~\cite{jiMIPRE2022a} has an estimated question size and answer size of about $10^{20}$. This is completely impractical experimentally, as each question requires a different measurement configuration, and each answer requires a precise measurement setting. We expect any separation generated by~\Cref{thm:samevalueREcomplete} to have similar, if not higher, question size and answer size as techniques used are similar to the ones used by~\cite{jiMIPRE2022a}.}. Thus, it would be an interesting open question whether we can show, using techniques from the operator algebra community, an example of a bell experiment which separates between the tensor product model and the commuting operator model with a more reasonable question/answer size.

\subsection{Proof of~\Cref{thm:gappedcompression}} \label{sec:proofcompressiongame}

In this subsection, we give a proof for~\Cref{thm:gappedcompression} assuming some important propositions. Similar to~\cite[Theorem 12.1]{jiMIPRE2022a}, the proof of~\Cref{thm:gappedcompression} relies on three components: question reduction, answer reduction, and parallel repetition, which we state below. The first proposition is question reduction. This proposition states the existence of an algorithm which takes a $k$-th CL verifier and outputs a $3$rd level CL verifier with $O(\polylog(n))$ sample complexity without drastically increasing the verification complexity and the soundness condition.

\begin{proposition}[Question Reduction] \label{prop:QuestionReduction}
	For all constants $\alpha,k \in \bN$, there exists a polynomial time algorithm $\texttt{QuestionReduction}_{\alpha, k}$ that takes, as input, a pair of Turing machines $(\samplerTM, \deciderTM)$ and outputs a tuple of Turing machines $(\samplerTM^{\QuestionRed}, \deciderTM^{\QuestionRed})$ such that the following holds:
	
	For models $t \in \{*, co\}$, $\texttt{QuestionReduction}_{\alpha, k}$ outputs a pair of Turing machines $(\samplerTM^{\QuestionRed}, \deciderTM^{\QuestionRed})$ which is a fourth-level CL-verifier $\verifiersequence^{\QuestionRed}$ for an infinite sequence of games $\gamesequence^{\QuestionRed} = \{\cG_n^{\QuestionRed} \}_{n \in \bN}$ with the following properties: There exists an integer $\gamma^{\QuestionRed} = O(\text{poly}(\alpha)) $ such that
	\begin{enumerate}
		\item (Computation time): $\TIME_{\texttt{QuestionReduction}_{\alpha, k}} (\poly(\alpha), |\samplerTM|, |\deciderTM|)$. 
		\item (Synchronicity) The game sequence $\verifiersequence^{\QuestionRed}$ is a $3$-rd level synchronous CL verifier. 
		\item (Complexity bounds for the output):
		\begin{itemize}
			\item  $\samplerTM^{\QuestionRed}(n, \text{Parameter}) \leq \max\left\{\log^{\gamma^{\QuestionRed}}(n), C^{\text{trivial}}\right\}$, 
			\item  $\TIME_{ \samplerTM^{\QuestionRed} }(n) \leq \max\left\{\log^{\gamma^{\QuestionRed}}(n),C^{\text{trivial}}\right\} $,
			\item  $\TIME_{ \deciderTM^{\QuestionRed}} \leq \max\left\{n^{\gamma^{\QuestionRed}},C^{\text{trivial}}\right\}$,
		\end{itemize}
		for some universal constant $C^{\text{trivial}}$. 
		\item (Independence of the sampler) $\samplerTM^{\QuestionRed}$ is a $3$-th level CL sampler which only depends on the parameter $\alpha$ and $k$ and does not depend on $|\deciderTM|$, $|\deciderTM^{\QuestionRed}| =  O(\poly( \alpha, k))$. 
	\end{enumerate}
	Furthermore if the input $\verifiersequence = (\samplerTM, \deciderTM)$ is a $k'$th CL verifier for the infinite sequence of synchronous game $\gamesequence = \{\cG_n\}_{n \in \bN}$ for some constant $k' \in \bN$ such that $k' < k$, and there exists some constant $n_0 \in \bN$ such that for all $n \geq n_0$
	\begin{equation*}
		\max\left\{ \TIME_{\samplerTM}, \TIME_{\deciderTM}\right\} \leq n^\alpha. 
	\end{equation*}
	Then there exists some constant $n_0^{\QuestionRed} = \poly(\gamma^{\QuestionRed}, n_0)$ with $n_0 \leq n_0^{\QuestionRed}$ such that for all $n \geq n_0^{\QuestionRed}$
	\begin{enumerate}
		\item (Completeness) If there exists a perfect oracularizable synchronous strategy for $\cG_n$ in model $t$, then there exists a perfect oracularizable synchronous strategy for $\cG^{\QuestionRed}_n$ in model $t$. 
		\item (Soundness) There exists some universal function $\mathbf{s}^{\QuestionRed}_{\alpha}$ which depends on $k$ and $\eps$ with $\mathbf{s}^{\QuestionRed}_{\alpha} = O(\exp(k), \poly(\eps))$ such that, for any polynomial $\eps: \bN \rightarrow [0,1]$ 
		\begin{equation*}
			\omega^t(\cG_n) \leq 1 - \eps(n) \Longrightarrow \omega^t(\cG^{\QuestionRed}_n) \leq 1 - \mathbf{s}^{\QuestionRed}_{\alpha}(k, \eps(n)).
		\end{equation*}
	\end{enumerate}
\end{proposition}

Question Reduction relies on self-testing techniques used in~\cite{natarajanTwoplayerEntangledGames2018,griloSimpleProtocolVerifiable2020a}, which are unique to non-local games with entangled provers. We prove~\Cref{prop:QuestionReduction} in~\Cref{sec:introspection}. We remark that in comparison to the question procedure given in~\cite{jiMIPRE2022a}, there is no dependency on the parameter $n$ and $\lambda$ for soundness since the EPR tester does not use the low-degree test, and we refer to~\Cref{sec:Paulibasis} for more details. The second proposition is Answer reduction, which gives an algorithm which takes synchronous a CL verifier with $O(\poly(n))$ verification complexity and outputs a balanced synchronous CL verifier with $O(\polylog(n))$ verification complexity which does not increase the sample complexity and increases the soundness only by $\polylog(n)$.  

\begin{proposition}[Answer Reduction] \label{prop:AnswerReduction}
	For all constants $(\alpha,k) \in \bN$ there exists a polynomial time algorithm $\texttt{AnswerReduction}_{\alpha, k}$ that takes, as input, a pair of Turing machines $(\samplerTM, \deciderTM)$ and outputs a tuple of Turing machines $(\samplerTM^{\Answerred}, \deciderTM^{\Answerred})$ such that the following holds:
	
	For models $t \in \{*, co\}$, $\texttt{AnswerReduction}_{\alpha}$ outputs a pair of Turing machines $(\samplerTM^{\Answerred}, \deciderTM^{\Answerred})$,  which defines a synchronous CL-verifier $\verifiersequence^{\Answerred}$ for an infinite sequence of games $\gamesequence^{\Answerred} = \{\cG_n^{\Answerred} \}_{n \in \bN}$ with the following properties: There exists an integer $\gamma^{\Answerred} = O(\text{poly}(\alpha))$ such that
	
	\begin{enumerate}
		\item (Runtime): $\texttt{AnswerReduction}_{\alpha}$ has runtime
		\begin{equation*}
			\TIME_{\texttt{AnswerReduction}_{\alpha}}  (\poly(\alpha), |\samplerTM|, |\deciderTM|). 
		\end{equation*}
		\item (Dependency for $\samplerTM^{\Answerred}$) The Turing machine $\samplerTM^{\Answerred}$ only depends on the input $\samplerTM$ and $\alpha$. 
	\end{enumerate}
	
	Furthermore, if the input $\verifiersequence = (\samplerTM, \deciderTM)$ is a $k$-th level synchronous CL verifier for the infinite sequence of synchronous game $\gamesequence = \{\cG_n\}_{n \in \bN}$ for some constant $k \in \bN$, and there exists a constant $n_0 \in \bN$ such that for all $n \geq n_0$, we have
	\begin{itemize}
		\item  $|\langle\deciderTM\rangle| =  O(\poly( \alpha, k))$.
		\item  $\samplerTM(n, \text{Parameter}) \leq \log^{\alpha}(n)$,
		\item  $\TIME_{ \samplerTM }(n) \leq \log^{\alpha}(n)$,
		\item  $\TIME_{ \deciderTM} \leq n^{\alpha}$.
	\end{itemize}
	Then there exists an integer $n_0^{\Answerred} = \poly(\gamma^{\Answerred}, n_0)$ with $n_0 \leq n_0^{\Answerred}$ such that for all $n \geq n_0^{\Answerred}$
	\begin{enumerate}
		\item (Complexity bounds for the output): 
		\begin{itemize}
			\item  $\TIME_{ \samplerTM^{\Answerred} }(n) \leq \max\left\{\log^{\gamma^{\Answerred}}(n), C^{\text{trivial}}\right\}$,
			\item  $\TIME_{ \deciderTM^{\Answerred}}(n) \leq \max\left\{\log^{\gamma^{\Answerred}}(n), C^{\text{trivial}}\right\}$,
		\end{itemize}
		for some universal constant $C^{\text{trivial}}$. 
		\item (Level for the CL sampler) The Turing machine $\samplerTM^{\Answerred}$ is a $\max\{k+2, 6\}$th-level CL sampler.
		\item (Completeness) If there exists a perfect oracularizable strategy for $\cG_n$ in model $t$, then there exists a perfect oracularizable strategy for $\cG^{\Answerred}_n$ in model $t$. 
		\item (Soundness) There exists a universal function $\mathbf{s}^{\Answerred}_{\alpha}$ which depends on $n$ and $\eps$ with $O(\mathbf{s}^{\Answerred}_{\alpha}) = O(\polylog(n), \poly(\eps))$ such that, for any polynomial $\mathbf{\eps}: \bN \rightarrow [0,1]$ 
		\begin{equation} \label{eq:ARsoundnessclause}
			\omega^t(\cG_n) \leq 1 - \mathbf{\eps}(n) \Longrightarrow \omega^t(\cG^{\Answerred}_n) \leq 1 -\mathbf{s}^{\Answerred}_{\alpha}(\mathbf{\eps}(n), n)
		\end{equation}
	\end{enumerate}
\end{proposition}


We remark that the universal constant $C^{\text{trivial}}$ comes from $\cG^{\text{reject}}$. The answer Reduction  procedure we use is identical to the one used in~\cite[Chapter 10]{jiMIPRE2022a}, which is a modification of the PCP of proximity based on techniques from~\cite{babaiNondeterministicExponentialTime1991a}. We further remark that due to the technique used, $\texttt{AnswerReduction}_{\alpha}$ actually works for \textbf{all} $\MIP$, $\MIP^*$, and $\MIP^{co}$. We prove~\Cref{prop:AnswerReduction} in~\Cref{sec:PCPanswerreduction}. The third step is the parallel repetition theorem for anchored games. 

\begin{proposition}[Parallel repetition] \label{prop:Parallelrepetition}
	For all constants $\alpha \in \bN$, function $\mathbf{s}(n): \bN \rightarrow [0,1]$ with $O(\mathbf{s}(n)) = O(\polylog(n))$. Then there exists a function $\mathbf{r}(n): \bN\rightarrow \bN$ with $\mathbf{r}(n) = O(\alpha ,\mathbf{s}(n))$ and a polynomial time algorithm $\texttt{Parallelrep}_{\alpha, \textbf{s}(n)}$ that takes, as input, a pair of Turing machines $(\samplerTM, \deciderTM)$ and outputs a tuple of Turing machines $(\samplerTM^{\Pararep}, \deciderTM^{\Pararep})$ with $\TIME_{\texttt{Parallelrep}_{\alpha, \mathbf{s}(n)}} (|\samplerTM|, |\deciderTM|) = O(\polylog(n))$ such that the following holds
	\begin{enumerate}
		\item (Independence of the sampler) $\samplerTM^{\Pararep}$ only depends on $\samplerTM$ and the polynomial functions $\mathbf{s}(n)$. 
		\item (Complexity bounds for the output):
		\begin{itemize}
			\item $\TIME_{\samplerTM^{\Pararep} }(n) \leq O(\poly(\mathbf{r}(n), log^{\alpha}(n)))$,
			\item $\TIME_{\deciderTM^{\Pararep} }(n) \leq O(\poly(\mathbf{r}(n), log^{\alpha}(n)))$,
		\end{itemize} 
	\end{enumerate}
	Furthermore, if the input $\verifiersequence = (\samplerTM, \deciderTM)$ is a $k$-th level synchronous CL verifier for the infinite sequence of synchronous games $\gamesequence = \{\cG_n\}_{n \in \bN}$ for some constant $k \in \bN$, and there exists a constant $n_0 \in \bN$ such that for all $n \geq n_0$, we have
	\begin{equation*}
		\samplerTM(n, \text{Parameter}), \TIME_{ \samplerTM }(n),\TIME_{ \deciderTM} \leq \log^{\alpha}(n). 
	\end{equation*}
	Then for all $n \geq n_0$
	\begin{enumerate}
		\item (Parameter) $\samplerTM^{\Answerred}$ is a ($k+1$)-th level CL sampler.
		\item (Completeness) If there exists a perfect oracularizable strategy for $\cG_n$ in model $t$, then there exists a perfect oracularizable strategy for $\cG^{\Pararep}_n$ in model $t$. 
		\item (Soundness) 
		\begin{equation*}
			\omega^t(\cG_n) \leq 1 - \mathbf{s}(n) \Longrightarrow \omega^t(\cG^{\Pararep}_n) \leq \frac{1}{2}
		\end{equation*}
	\end{enumerate}
\end{proposition}

The anchored parallel repetition theorem is proven for the tensor product model in~\cite{bavarianAnchoredParallelRepetition2021} and the commuting operator model in~\Cref{sec:parallelrepappendix}. We remark that the theorem above actually works for all verifiers as well. We show that the anchor transformation and parallel repetition of a $k$-th level synchronous CL verifier becomes a $k+1$-th level synchronous CL verifier in~\Cref{sec:parallelrep}. We are now ready to give a proof for~\Cref{thm:gappedcompression} below, which is just applying the three propositions above in sequence. 

\begin{proof}
	Given constant $(\alpha, k) \in \bN$, we specify the pseudocode for $\texttt{Gapcompress}_{\alpha, k}$ as follows
	
	\vspace{10pt}
	\IncMargin{1em}
	\begin{algorithm}[H]
		\DontPrintSemicolon
		
		\textbf{Input}: Turing Machines $(\samplerTM, \deciderTM)$. 
		
		Compute $\verifiersequence^{\QuestionRed} = (\samplerTM^{\QuestionRed}, \deciderTM^{\QuestionRed}) = \texttt{QuestionReduction}_{\alpha, k}(\samplerTM, \deciderTM)$.
		
		Compute $\alpha^{\QuestionRed} = \mathbf{s}^{\QuestionRed}(\alpha)$, where $ \mathbf{s}^{\QuestionRed}$ is the function used to define $\gamma^{\QuestionRed}$ in~\Cref{prop:QuestionReduction}. 
		
		Compute $\verifiersequence^{\Answerred} =  (\samplerTM^{\Answerred}, \deciderTM^{\Answerred}) = \texttt{AnswerReduction}_{\alpha^{\QuestionRed}}(\samplerTM^{\QuestionRed}, \deciderTM^{\QuestionRed})$.
		
		Compute $\alpha^{\Answerred} = \mathbf{s}^{\Answerred} (\alpha^{\QuestionRed})$, where $ \mathbf{s}^{\Answerred}$ is the function used to define $\gamma^{\Answerred}$ in~\Cref{prop:AnswerReduction}. 
		
		Compute the description of the function $\mathbf{s}(n) = \mathbf{s}_{\alpha^{\Answerred} }^{\Answerred} ( \mathbf{s}_{\alpha}^{\QuestionRed}( \frac{1}{2},n), n)$. Where $\mathbf{s}_{\alpha^{\Answerred} }^{\Answerred}$ (resp. $\mathbf{s}_{\alpha}^{\QuestionRed}$) is the function used in the soundness condition in~\Cref{prop:AnswerReduction} (resp.~\Cref{prop:QuestionReduction}).
		
		\textbf{Return} $\verifiersequence^{\Compressgame} =  (\samplerTM^{\Compressgame}, \deciderTM^{\Compressgame}) =  \texttt{Parallelrep}_{\alpha^{\Answerred}, \mathbf{s}(n)} (\samplerTM^{\Answerred}, \deciderTM^{\Answerred})$.
		
		\caption{The description for $\texttt{Gapcompress}_{\alpha, k}$.}
		\label{pseu:Gapcompress}
	\end{algorithm}\DecMargin{1em}
	\vspace{10pt}

	\begin{table}[!htb]
		\centering
		\footnotesize
		\vspace{1em}
		\begin{tabularx}{\textwidth}{c c c c c}
			\hline
			& \multicolumn{2}{c}{Time Complexity} &      &  \\
			\cmidrule{2-4}
			Verifier	 & Sampler & Decider   &  Level & Soundness\\
			\hline
			$\verifiersequence^{\QuestionRed}$ &  $\leq \log^{O(\poly(\alpha))}(n)$  & $\leq n^{O(\poly(\alpha))}$ & 3  &  $\omega^{t} (\cG_n) \leq \frac{1}{2} \rightarrow \omega^{t}(\cG_n^{\QuestionRed}) \leq 1 - \mathbf{s}_{\alpha}^{\QuestionRed}( \frac{1}{2},n)$ \\
			$\verifiersequence^{\Answerred}$ &  $\leq \log^{O(\poly(\alpha))}(n)$  &  $\leq \log^{O(\poly(\alpha))}(n)$  & 6  &  $\omega^{t}(\cG_n) \leq \frac{1}{2} \rightarrow \omega^{t}(\cG_n^{\Answerred}) \leq 1 - \mathbf{s}_{\alpha^{\Answerred} }^{\Answerred} (  \mathbf{s}_{\alpha}^{\QuestionRed}( \frac{1}{2},n), n)$ \\
			$\verifiersequence^{\Compressgame}$ & $\leq \log^{O(\poly(\alpha))}(n)$  & $\leq \log^{O(\poly(\alpha))}(n)$ & 7  &  $\omega^{t}(\cG_n) \leq \frac{1}{2} \rightarrow \omega^{t}(\cG_n^{\Compressgame}) \leq \frac{1}{2}$ \\
			\hline
		\end{tabularx}
		\caption{The time complexity for the 3 CL verifiers listed in~\Cref{pseu:Gapcompress} and the soundness statement for the verifier sequences assuming the input $(\samplerTM, \deciderTM)$ is a synchronous $k'$th-level CL verifier for $k' \leq k$. We remark that only the last column is dependent upon the input $(\samplerTM, \deciderTM)$ being a synchronous CL verifier with the appropriate runtime condition.}
		\label{fig:paramscompress}
	\end{table}
	
	For simplicity, we listed how the parameter changes throughout~\Cref{pseu:Gapcompress} in~\Cref{fig:paramscompress}. We verify each of the clause in~\Cref{thm:gappedcompression} below:
	\begin{itemize}
		\item (Level of the CL verifier and the output being a synchronous game): Since the output for $\texttt{AnswerReduction}_{\alpha^{\QuestionRed}}$ is always a synchronous CL verifier, and $\texttt{Parallelrep}_{\alpha^{\Answerred},\mathbf{s}(n)}$ retains synchronicity for a CL verifier. The output for $\texttt{Gapcompress}_{\alpha, k}$ will always be a synchronous game sequence. The level for the output of $\texttt{Gapcompress}_{\alpha, k}$ are tracked in~\Cref{fig:paramscompress}.
		\item (Computation time) We first see that both $\alpha^{\Answerred}$ and $\mathbf{s}(n)$ in~\Cref{pseu:Gapcompress} are polynomial functions of $\alpha$, independent from the input $(\samplerTM, \deciderTM)$. Hence, both steps 2 and 4 can be computed in $O(\poly(\alpha))$ time (or hardcoded into the description of $\texttt{Gapcompress}_{\alpha, k}$). By the similar reasoning, we have $\TIME_{\texttt{QuestionReduction}_{\alpha, k}} = \TIME_{\texttt{AnswerReduction}_{\alpha^{\QuestionRed}}} = \TIME_{\texttt{Parallelrep}_{\alpha^{\Answerred}, \mathbf{s}(n)}} = O(\poly(\alpha), |\samplerTM|, |\deciderTM|)$. Hence, $\TIME_{\texttt{Gapcompress}_{\alpha, k}} (\poly(\alpha), |\samplerTM|, |\deciderTM|)$. 
		\item (Independence of the sampler) By~\Cref{prop:Parallelrepetition}, $\samplerTM^{\Compressgame}$ only depends on the polynomial $s(n)$ (which itself only depends on $\alpha$) and $\samplerTM^{\Answerred}$. $\samplerTM^{\Answerred}$ depends, in addition to the parameter $\alpha$, on both $\samplerTM^{\QuestionRed}$, and only on $|\deciderTM^{\QuestionRed}|$ by definition. Since $|\deciderTM^{\QuestionRed}| =O(\poly(\alpha, k))$ given any $\deciderTM$ as input for $\texttt{QuestionReduction}_{\alpha, k}$. $\samplerTM^{\Compressgame}$ only depends on the parameters $\alpha$ and $k$, as claimed. 
		\item (Complexity bounds for the output) This follows from the complexity parameter, which we kept track of in~\Cref{fig:paramscompress}, and the complexity bound for the outputs of $\texttt{AnswerReduction}_{\alpha^{\QuestionRed}}$, which does not depend on the input.
	\end{itemize}
	
	Now, assume the input $\verifiersequence = (\samplerTM, \deciderTM)$ is a synchronous $k'$th level CL verifier for the infinite sequence of synchronous game $\gamesequence = \{\cG_n\}_{n \in \bN}$ for some constant $k' \in \bN$ with $k' < k$, and some constant $n_0 \in \bN$ which satisfies~\eqref{eq:gapcompressinputcondition}. Take $n_0^{\QuestionRed}$ be the constant guaranteed by~\Cref{prop:QuestionReduction} and take $n_0^{\Compressgame}$ to be the constant $n_0^{\Answerred}$ guaranteed by~\Cref{prop:AnswerReduction} (where in this case, $n_0^{\Answerred} = \poly(n_0^{\QuestionRed}, \alpha^{\QuestionRed})$ ). Fix $n > n_0^{\Compressgame}$ and let $t \in \{*, co\}$.
	
	
	\begin{itemize}
		\item (Completeness) Since the game sequence $\verifiersequence$ is synchronous, any perfect strategy for $\cG_n$ is also a synchronous strategy. Since $\texttt{QuestionReduction}_{\alpha, k}$, $\texttt{AnswerReduction}_{\alpha^{\QuestionRed}}$ and $\texttt{Parallelrep}_{\alpha^{\Answerred}, \mathbf{s}(n)}$ preserve the existence of a perfect oracularizable strategy for $\cG_n$ in model $t$,  $\texttt{Gapcompress}_{\alpha, k}$ also preserves the existence of a perfect oracularizable strategy for $\cG_n$ in model $t$.
		\item (Soundness) Assume that $\omega^{t}(\cG_n) \leq \frac{1}{2}$; By the soundness property of $\texttt{QuestionReduction}_{\alpha, k}$, the previous condition implies that $\omega^{t}(\cG_n^{\QuestionRed}) \leq 1 -  \mathbf{s}_{\alpha}^{\QuestionRed}( \frac{1}{2},n)$, which, by the soundness property of $\texttt{AnswerReduction}_{\alpha^{\QuestionRed}}$, implies that $\omega^{t}(\cG_n^{\QuestionRed}) \leq 1 - \mathbf{s}_{\alpha^{\Answerred} }^{\Answerred} (  \mathbf{s}_{\alpha}^{\QuestionRed}( \frac{1}{2},n), n)$. The completeness condition follows from the soundness condition of $\texttt{Parallelrep}_{\alpha^{\Answerred}, \mathbf{s}(n)}$.
	\end{itemize}
\end{proof}

\end{CJK*}
\section{Question Reduction}~\label{sec:introspection}
In this section, we give a proof for~\Cref{prop:QuestionReduction} by showing an algorithm that takes as input a synchronous CL verifier and transforms it into another synchronous CL verifier with a lower sampling complexity. Intuitively, this is done by asking both provers to sample their own question pairs for the $n$th game, and play the game based on the question pair they sampled. This procedure might first seem counter-intuitive, since the provers can always pre-select a question pair before the game  rather than sampling it honestly during the interaction. 

Roughly speaking, the verifier takes advantage of the entanglement shared between the provers to force them to sample a ``fresh" question pair for the given game. By leveraging self-testing techniques, the verifier can force the provers to make certain measurements on their entangled resources, thereby generating a question pair for the original game. 

The question reduction protocol consists of two components. The first component is the $n$-Pauli basis test. This is a subroutine that forces the provers to perform either an all $X$ or all $Z$ Pauli measurements on $n^{\alpha}$ EPR pairs, and we present this protocol in~\Cref{sec:Paulibasis}. The second component is the introspection test, where the verifier forces the provers to perform a specific set of measurements on the $n^{\alpha}$ EPR pairs, whereby the measurement outcomes are precisely the input distribution for the original game. 

\subsection{The magic square game}

\begin{figure}[!b]
	\begin{center}
		\begin{tabular}{|c|c|c|}
			\hline
			$x_1$ & $x_2$ & $x_3$ \\
			\hline 
			$x_4$ & $x_5$ & $x_6$ \\
			\hline
			$x_7$ & $x_8$ & $x_9$ \\
			\hline
		\end{tabular} 		\qquad
		\begin{tabular}{|c|c|c|}
			\hline
			$\cI_2 \otimes \rho^{Z}$ & $\rho^{Z} \otimes \cI_2$ & $\rho^{Z} \otimes \rho^{Z}$ \\
			\hline 
			$\rho^{X} \otimes \cI_2$ & $\cI_2 \otimes \rho^{Z}$ & $\rho^{X} \otimes \rho^{X}$ \\
			\hline
			$\rho^{X} \otimes \rho^{Z}$ & 	$\rho^{Z} \otimes \rho^{X}$ & $ - \left(\rho^{X}  \rho^{Z} \right) \otimes  \left(\rho^{X}  \rho^{Z} \right)$ \\
			\hline
		\end{tabular}
		
	\end{center}
	\caption{Left: The description for the magic square game, where each row and column corresponds to an equation. Right: A oracularizable perfect strategy for the magic square game.}
	\label{fig:magicsquare}
\end{figure}
We first introduce a key subroutine for the Pauli basis test, the Mermin-Peres magic square game~\cite{merminSimpleUnifiedForm1990, peresIncompatibleResultsQuantum1990}, in this subsection. We use the BCS formulation of this game as presented in~\cite{cleveCharacterizationBinaryConstraint2014}, where the game is defined by six equations and nine variables over $\bF_2$, where the variables are arranged on a three-by-three grid as presented in~\Cref{fig:magicsquare}. Every row and column in~\Cref{fig:magicsquare} corresponds to a constraint that multiplies to $1$, except for the last column, where the constraint multiplies to $-1$ instead. In this game, the referee randomly samples a constraint and a variable in the constraint and sends the constraint to one of the provers and the variable to the other prover. The prover must then respond with an assignment for their given constraint or variable. The provers win the game if their assignments are consistent with each other. If one of the provers is given an equation as the question, their assignment must also satisfy the constraint for the given equation. We also modify the magic square to be synchronous, meaning the verifier additionally samples a constraint or a variable with constant probability and sends it to both provers and expects the same answer in return.

The magic square game admits a perfect synchronous oracularizable strategy for both quantum models by using the measurement operator defined in~\Cref{fig:magicsquare}. In this strategy, the provers initially prepare two copies $\MEstate{2}$. In the event that the prover receives a variable, he measures the observable as displayed in the grid and returns the resulting eigenvalue (which is either $1$ or $-1$) as the assignment to the variable. If the prover receives an equation, he measures the observables for all three variables, returning the eigenvalue for each of the observables as the corresponding assignment for each of the variables. Since each observable on the grid commutes with all other observables that share a row or column with it, the order of measurement does not matter for the constraint question, and their measurement commutes on all possible question pairs, which implies that this perfect strategy is oracularizable.

\subsection{The Pauli basis test} \label{sec:Paulibasis}
\jqnote{TODO: You need to change all the reference for~\cite{linTracialEmbeddableStrategies2024} to the correct ArXiv version once this is done.}

In this section, we recall the Pauli basis test. In this paper, we use the version of the Pauli basis test based on the elegant simplification given in~\cite{delasalleSpectralGapStability2022} and presented it similarly to~\cite[Section 6]{linTracialEmbeddableStrategies2024}. In comparison to the original Pauli basis test given in~\cite[Section 7]{jiMIPRE2022a}, this version does not rely on the low-individual degree test. Since this is a simple adaptation of the Pauli basis test, we present the protocol as is, and instead refer the reader to either~\cite{delasalleSpectralGapStability2022} or~\cite[Section 6]{linTracialEmbeddableStrategies2024} for more intuition.

\jqnote{TODO, rewrite this, or move this to ~\cite{linTracialEmbeddableStrategies2024}}
Intuitively, the goal of the Pauli basis test is to force two honest provers to prepare $n$ copies of EPR pairs between them and measure either $(\rho^X)^{\tensor n}$ or $(\rho^Z)^{\tensor n}$ on their half of the EPR pair. For $n \in \bN$, we define the $n$ qubit Pauli basis test as the $\mu$-dependent Pauli basis test defined in~\cite[Section 6.3]{linTracialEmbeddableStrategies2024}, where $\mu$ is the uniform distribution over a subset $S^{\text{Paulibasis}}_n \subseteq \{0,1\}^n$ such that the spectral gap of $\mu$ is a constant.

We remark that the subset suggested in~\cite[Theorem 1.3]{delasalleSpectralGapStability2022} cannot be used directly in this context since it cannot be uniformly generated by a single circuit. Instead, recall in~\cite[Example 1.2]{delasalleSpectralGapStability2022}, given any $[k, n, d]$ binary linear code $C$, the code space for $C$, $S_{C} \subset \bF_{2}^n$, is a subset such that the uniform distribution of $S_C$ has a spectral gap of $\frac{2k}{d}$. Intuitively, any good binary linear code would lead to an efficient EPR tester. For any $n \in \bN$, consider the Justesen code~\cite{justesenClassConstructiveAsymptotically1972} with $R = \frac{\log(n)}{n}$, by definition this is a code with dimension $k = \lfloor \log(n) \rfloor$, length $n$ and distance $d \geq  0.11(n -  \log(n))$. Let $S^{\text{PB}}_n \subseteq\{0,1\}^n = \bF_2^n$ be the code space of the Justesen code mentioned above. Since the encoding map for the Justesen code can be implemented in $O(\poly(n))$ time, there exists an encoding map $\mathbf{\pi}^{\text{PB}}: \bN \times \{0,1\}^* \rightarrow \{0,1\}^*$ such that $\mathbf{\pi}^{\text{PB}} (n, s)$ is a bijection map that maps elements from $\{0,1\}^{\lfloor \log(n) \rfloor}$ to $S^{\text{PB}}_n$, with $\TIME_{\mathbf{\pi}^{\text{PB}}} (n) = \poly(n)$. Furthermore, the uniform distribution on $S^{\text{PB}}_n$ has a spectral gap of $\frac{2k}{d} \leq \frac{\log(n)}{0.11(n -  \log(n))} \leq 30$ for all $n$. Hence, it is sufficient to take $S^{\text{Paulibasis}}_n = S^{\text{PB}}_n$ in this context.

We present the sample/decision procedure for the $n$-Pauli basis test in~\Cref{fig:PauliBasis}, and provide a diagram representation for the input distribution in~\Cref{fig:PauliBasisgraph}.  We recall the following rigidity theorem about the $n$ qubit Pauli basis test.

\begin{figure}[!b]
	\centering
	\begin{tikzpicture}[scale=.8]
		
		\tikzset{type/.style args={[#1]#2}{
				draw,circle,fill,scale=0.25,
				label={[font=\scriptsize, label distance=1pt]#1:#2}
		}}
		
		\foreach \i in {1,...,6} \draw (0,8-9/8*\i)
		coordinate (Constraint-\i)
		node[type={[180]$\text{Constrant}_\i$}] {};
		
		\foreach \i in {1,...,9} \draw (2.5,9-\i)
		coordinate (Variable-\i)
		node[type={[330]$\text{Variable}_\i$}] {};
		

		\draw (5.5,8) coordinate (Cord-X) node[type={[45]$(\text{Coordinate},X)$}] {};
		\draw (5.5,4) coordinate (Cord-Z) node[type={[315]$(\text{Coordinate},Z)$}] {};
		
		\foreach \i in {1,...,3} \foreach \j in {1,...,3}
		\pgfmathsetmacro{\k}{(\i-1)*3+\j}
		\draw[blue] (Constraint-\i) -- (Variable-\k);
		
		\foreach \i in {4,...,6} \foreach \j in {1,...,3}
		\pgfmathsetmacro{\k}{\i-3+(\j-1)*3}
		\draw[blue] (Constraint-\i) -- (Variable-\k);
		\draw[blue] (Variable-1) -- (Cord-X);
		\draw[blue] (Variable-5) -- (Cord-Z);
		
		\draw (8,8) coordinate (Pauli-X) node[type={[0]$(\text{Pauli}, X)$}] {};
		\draw (8,4) coordinate (Pauli-Z) node[type={[0]$(\text{Pauli},Z)$}] {};
		\draw (8,6) coordinate (Commutation) node[type={[0]$\text{Commutation}$}] {};
		\draw (7,7) coordinate (Commutation-X) node[type={[0]$(\text{Commutation}, X)$}] {};
		\draw (7,5) coordinate (Commutation-Z) node[type={[0]$(\text{Commutation}, Z)$}] {};
		
		\foreach \from \to in {Cord-X/Commutation-X, Cord-Z/Commutation-Z,Commutation-Z/Commutation, Commutation-X/Commutation}
		\draw[red] (\from) -- (\to);
		\foreach \from \to in {Cord-X/Pauli-X, Cord-Z/Pauli-Z}
		\draw (\from) -- (\to);
	\end{tikzpicture}
	
	\caption{The typed graph $(\decidertypedtype^{\text{PB}}, \decidertypedquestionpair^{\text{PB}})$ for the $n$ qubit Pauli basis test, where each of the vertices above also contains a self-loop (this is a black edge). Similar to the definition of the CL distribution, each vertices in the graph represents a potential question label and each edge represents a potential question pair. Each of the edges are colour coded in order to better explain the game procedure and we refer to~\Cref{fig:PauliBasis} for more details. }
	\label{fig:PauliBasisgraph}
\end{figure}
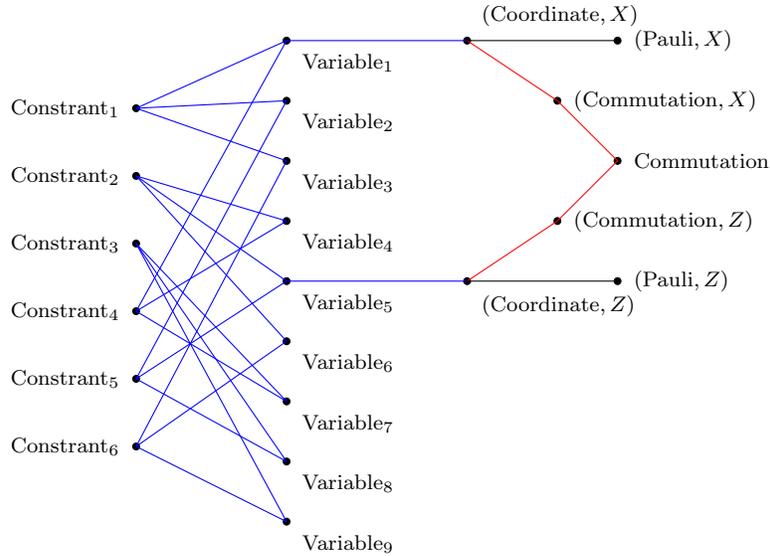

\begin{figure}[!htbp]
	\centering
	\begin{gamespec}
		\setlength{\tabcolsep}{1em}
		\begin{tabularx}{\textwidth}{ l   l   X   }
			\toprule
			Question label & Question content & Answer format \\
			\midrule
			(Pauli,$W$) &  & $ t_{W} \in \{0,1\}^n$ \\
			(Coordinate, $W$)& $u_W \in S^{\text{PB}}_n$& $ t_{\text{(Pauli, W)}} \in \{0,1\}^{n}$ \\
			(Commutation, $W$)& $u_W \in S^{\text{PB}}_n$& $t_W  \in \{0,1\}$ \\
			Commutation& $(u_X,u_Z)  \in S^{\text{PB}}_n \times S^{\text{PB}}_n$& $(t_{X}', t_{Z}') \in \{0,1\}^{2n}$ \\
			$\text{Variable}_i$& $(u_X,u_Z) \in S^{\text{PB}}_n \times S^{\text{PB}}_n$& $t_{\text{var}} \in \{0,1\}$ \\
			$\text{Constraint}_i$& $(u_X,u_Z)\in S^{\text{PB}}_n \times S^{\text{PB}}_n$& $t_{\text{cons}} \in \{0,1\}^{3}$ \\
			\bottomrule
			\multicolumn{3}{c}{Figure: Q and A format for the $n$ qubit Pauli basis test, where $W\in \{X,Z\}$. }
		\end{tabularx}
		\begin{center}
			\textbf{Sampling procedure}
		\end{center}
		\begin{enumerate}
			\item Sample $(u_X,u_Z) \in S^{\text{PB}}_n \times S^{\text{PB}}_n$ uniformly at random. 
			\item Uniformly samples $(n_0, n_1) \in \decidertypedtype^{\text{PB}} \times \decidertypedtype^{\text{Paulibasis}}$, where $(\decidertypedtype^{\text{PB}},  \decidertypedquestionpair^{\text{PB}})$ is the graph in~\Cref{fig:PauliBasisgraph}, and perform rejection sampling until $(n_0, n_1) \in \decidertypedquestionpair^{\text{PB}}$. 
			\item Send the question label and question content corresponding to $n_0$ to one of the provers, and send the question content corresponding to the $n_1$ to the other prover.
		\end{enumerate}
		\begin{center}
			\textbf{Verification procedure}
			
		\end{center}
		\begin{itemize}
			\item (Self-loop): The provers win iff they output the same answer. 
			\item (Pauli, $W$) \textbf{--}  (Coordinate, $W$): Alice and Bob win iff $t_{W}|_{u_W} = t_{\text{(Pauli, W)}}|_{u_W}$. 
			\item  If $(n_0, n_1)$ are a red edge and $u_x \cdot u_z = 1$, the provers win if for the question label (Commutation, $W$) and (Commutation), the prover answers $0$, otherwise
			\begin{itemize}
				\item  (Coordinate, $W$)  \color{red} \textbf{--}   \color{black}  (Commutation, $W$): The provers win iff $u_W \cdot t_{\text{Pauli},W} = t_W$. 
				\item  (Commutation)  \color{red} \textbf{--}   \color{black}  (Commutation, $W$): The provers win iff $t_W = t_W'$. 
			\end{itemize}
			\item If $(n_0, n_1)$ are a blue edge and $u_x \cdot u_z = 0$, the provers wins if for the question label (Constraint, $i$) and (Variable, $j$), the prover answers $0$, otherwise:
			\begin{itemize}
				\item (Variable $1$)  \color{blue} \textbf{--}   \color{black}  (Coordinate, $X$): The provers iff $u_X \cdot t_X = t_{\text{var}}$.
				\item (Variable $5$)  \color{blue} \textbf{--}   \color{black}  (Coordinate, $Z$): The provers iff $u_Z \cdot t_Z = t_{\text{var}}$.
				\item (Constraint)  \color{blue} \textbf{--}  \color{black}   (Variable): The provers win iff $t_{\text{cos}}$ is consistent with the constraint in the magic square game, and $t_{\text{var}}$ is consistent with the assignment of $v_i$ within $t_{\text{cos}}$. 
			\end{itemize}
		\end{itemize}
		\vspace{1em}
	\end{gamespec}
	\caption{The description for the $n$ qubit Pauli basis test. Where $W \in \{X, Z\}$ in the decision procedure. }
	\label{fig:PauliBasis}
\end{figure}

\jqnote{TODO: I think this proof can be simplify for synchronous strategies, also to projectors are not trivial}
\begin{theorem}[Rigidity for the $n$ qubit Pauli basis test] \label{thm:soundnessPaulibasis}
	Let $\cG^{\text{PB}}_n$ be the $n$ qubit Pauli basis test and let $\strategy = (\cL^2(\alicealg, \tau), \sigma \ket{\tau}, \{P_a^x \}_{x \in \cX})$ be a projective, tracially embeddable strategy such that $\omega(\cG^{\text{PB}}_n, \strategy) \geq 1 - \eps$. 
	There exist two isometries $V_{A}: \cL^2(\alicealg, \tau) \rightarrow  \cL^2(\alicealg, \tau) \tensor \bC^{2^{2n}}$ and $V_{B}: \cL^2(\alicealg, \tau) \rightarrow \cL^2(\alicealg, \tau) \tensor \bC^{2^{2n}}$ with $( V_B \tensor \cI_{2^{2n}}) V_{A} =  V_{A} (V_B \tensor \bI_{2^{2n}})$ and a state $\ket{\text{Aux}} \in \cL^2(\alicealg, \tau) \otimes \bC^{2^{2n}}$ such that
	\begin{equation*} 
		\left\| \left(  V_B \tensor \cI_{2^{2n}}  \right) V_{A} \left(\sigma\ket{\tau}\right)  -\ket{\text{Aux}}  \ket{\text{ME}_2}^{\tensor n} \right\|^2 \leq O\left(\poly(\eps)\right),
	\end{equation*}
	and for all $W \in \{X, Z\}$ and $u \in \bF_{2^{n}}$
	\begin{align*} 
		&|  ((V_{A} P_s^{(\text{Pauli}, W)}  V_{A}^*)_{\alicealg A_1 A_2} \tensor (\cI_{2^{2n}} )_{B_1 B_2} - \\
		&\quad (\cI_{\cH})_{\alicealg} \otimes (\rho^{W}_s)_{A_1} \tensor (\cI_{2^{n}})_{A_2} \otimes  (\cI_{2^{2n}} )_{B_1 B_2} ) \ket{\text{Aux}}_{\alicealg A_2 B_2}  \MEstate{2}^{\tensor n}_{A_1 B_1}  |^2 \leq O\left(\poly(\eps )\right). \\
		&|  ((V_{B} (P_s^{(\text{Pauli}, W)})^{\text{op}}  V_{B}^*)_{\alicealg B_1 B_2} \tensor (\cI_{2^{2n}} )_{A_1 A_2} - \\
		&\quad (\cI_{\cH})_{\alicealg} \otimes (\rho^{W}_s)_{B_1} \tensor (\cI_{2^{n}})_{B_2} \otimes  (\cI_{2^{2n}} )_{A_1 A_2} ) \ket{\text{Aux}}_{\alicealg A_2 B_2}  \MEstate{2}^{\tensor n}_{A_1 B_1}  |^2 \leq O\left(\poly(\eps )\right). 
	\end{align*}

	where the subscript $\alicealg$, $A_1$, $B_1$, $A_2$, $B_2$ and $\alicealg$ denotes the registers on which each operator acts (listed for clarity). 
\end{theorem}


The above rigidity follows from~\cite[Theorem 6.4]{linTracialEmbeddableStrategies2024} by defining $\mu$ as the uniform distribution of $S^{\text{PB}}_n$ as per the discussion above. We remark that in comparison to the $\mu$-dependent Pauli basis test defined in~\cite[Figure 3]{linTracialEmbeddableStrategies2024}, the sampling procedure for the $n$ qubit Pauli basis test is changed so that it is easier to show that the input distribution is samplable via a typed CL distribution. Although the anti-commutation test (the red vertices in~\Cref{fig:PauliBasisgraph}) within the $n$ qubit Pauli basis test is nine times more likely to occur than the commutation test (the blue vertices), the ratio between the likelihood of the two tests is still a constant. Furthermore, the question types (Variable 1), (Variable, 5), (Commutation, X), and (Commutation, Z) are added to ensure that there exists a perfect oracularizable strategy for the test. ~\Cref{thm:soundnessPaulibasis} still follows by modifying the inequality of the proof for~\cite[Lemma B.1]{linTracialEmbeddableStrategies2024} with a larger constant in the case where $u \cdot v = 0$, and changing equation (67) and (70) to incorporate the extra question labels added. 



In some sense, we can view~\Cref{thm:soundnessPaulibasis} as the ``soundness" condition about the Pauli basis test, since a $1 - \eps$ approximate strategy guarantees an approximate version of ``Pauli $X$ or Pauli $Z$ measurements on $n$-EPR pairs". In the following theorem, we show the ``completeness" and the ``runtime" condition related to the $n$ qubit Pauli basis test.

\begin{theorem}[Properties of the $n$ qubit Pauli basis test] \label{thm:completenessPaulibasis}
	Let $\cG^{\text{PB}}_n$ be the $n$ qubit Pauli basis test
	\begin{enumerate}
		\item (Computation time): $\cG^{\text{PB}}_n$ is samplable via a $(\decidertypedtype^{\text{PB}}, \decidertypedquestionpair^{\text{PB}}, \{\decidertypedfunction^{v}\}_{v \in \decidertypedtype^{\text{PB}}} )$ typed CL distribution, where each $\decidertypedtype^{\text{PB}}: \bF_2^{2 \cdot s^{\text{SB}}_n} \rightarrow  \bF_2^{2 \cdot s^{\text{SB}}_n}$ is a first level CL function, where $s^{\text{SB}}_n = \lfloor \log(n) \rfloor$.  $\cG^{\text{PB}}_n$ has a decision complexity of $\poly(n)$. 
		\item (Completeness): There exists a perfect finite-dimensional symmetric oracularizable strategy $\strategy^{\text{PB}} = (\bC^{2^n} \otimes \bC^{2^n}, \MEstate{2}^{\otimes n} \otimes \MEstate{2}, \{P_a^x \})$ such that for all $s \in \{0,1\}^n$ and $W \in \{X,Z\}$
		\begin{equation*}
			 P_{s}^{(\text{Pauli}, W)} =  \rho^W_{s} \otimes \cI_2
		\end{equation*}
	\end{enumerate}
\end{theorem}
\begin{proof}
	For the remainder of this proof, we denote $W \in \{X,Z\}$.	Fix an integer $n\in \bN$. We first show that the $n$ qubit Pauli basis test is samplable via a typed CL distribution. Identify $\bF_2$ with $\{0,1\}$, and define each $\decidertypedfunction^v$ as follows: For every input $(s_{X}, s_{Z}) \in \{0,1\}^{s^{\text{SB}}_n} \times \{0,1\}^{s^{\text{SB}}_n}$, we define each of the linear function accordingly:
	\begin{itemize}
		\item  For $v \in \{ (\text{Pauli}, W), (\text{Coordinate}, W), (\text{Commutation}, W)\}$, 
		\begin{align*}
			\decidertypedfunction^{(\text{Pauli}, W)}(s_{X}, s_{Z}) &=  (0,0), \\
			\decidertypedfunction^{ (\text{Coordinate}, X)}(s_{X}, s_{Z}) = \decidertypedfunction^{ (\text{Commutation}, X)}(s_{X}, s_{Z}) &= (s_{X}, 0), \\
			\decidertypedfunction^{ (\text{Coordinate}, Z)}(s_{X}, s_{Z})  =\decidertypedfunction^{ (\text{Commutation}, Z)}(s_{X}, s_{Z}) &= (0, s_{Z}).
		\end{align*}
		\item Otherwise, $\decidertypedfunction^{v}$ is the identity function, or 
		\begin{equation*}
			\decidertypedfunction^{v}(s_{X}, s_{Z}) =  (s_{X}, s_{Z}). 
		\end{equation*}
	\end{itemize}
	
	For the decision process, given input $(u_{X}, u_{Z}) \in  S^{\text{PB}}_n \times S^{\text{PB}}_n$ and the output listed in the ``Answer format" from~\Cref{fig:PauliBasis}, the decision process for all the edges can be decided in $O(\poly(n))$ time since it involves either computing an inner product between elements of $\{0,1\}^n$ or some form of consistency test. Furthermore, since the map $\mathbf{\pi}^{\text{PB}}$ is computable uniformly in $O(\poly(n))$ time, the decider can compute each of the $(u_{X}, u_{Z})$ from $(s_{X}, s_{Z})$ sampled above. 
	For the ``completeness" condition, we define the synchronous strategy $\strategy^{\text{PB}}$ over the Hilbert space $\bC^{2^n} \otimes \bC^{2^n}$ as follows: The joint state used between the two provers is $n$ EPR pairs, or $\MEstate{2}^{\otimes n} \otimes \MEstate{2}$. For $(u_X, u_Z) \in S^{\text{PB}}_n \times S^{\text{PB}}_n$, define
	\begin{equation*}
		\rho^{W}(u_W)_0 = \sum_{b \cdot u_W = 0} \rho^{W}_b, \qquad  \rho^{W}(u_W)_0 = \sum_{b \cdot u_W = 1}\rho^{W}_b,
	\end{equation*}
	and we see that
	\begin{equation*}
			\rho^{W}(u_W) = \rho^{W}(u_W)_0  - 	\rho^{W}(u_W)_1. 
	\end{equation*}
	We define the measurement operator $\{P^x_a\}$ for the symmetric strategy $\strategy^{\text{PB}}$ as
	\begin{align*}
		&P^{(\text{Pauli}, W)} = \rho^{W}_{t_W} \otimes \cI_2 &\text{ for all } t_W \in \{0,1\}^n, \\
		&P^{(\text{Coordinate}, W), u_{W}}_{ t_{\text{(Pauli, W)}}} = \sum_{ t_{\text{(Pauli, W)}}|_{u_{W}} = b}  \rho^{W}_b \otimes \cI_2  &\text{ for all }  t_{\text{(Pauli, W)}} \in \{0,1\}^n. \\
		&P^{(\text{Commutation, W}), (u_W) }_{t_W} = \rho^{W}(u_W)_{t_W} \otimes \cI_2 &\text{ for }  u_X \cdot u_Z = 0, t_W \in \{0,1\}\\
		&P^{(\text{Commutation}), (u_X, u_Z) }_{(t_W', t_Z')} = \left(\rho^{X}(u_X)_{t_X'}\right)\left( \rho^{Z}(u_Z)_{t_Z'}\right) \otimes \cI_2&\text{ for }  u_X \cdot u_Z = 0, t_W' \in \{0,1\}\\
		&P^{(\text{Commutation}), (u_X, u_Z) }_{a} = \begin{cases}
			\cI_{2^{n+1}} & \text{if } a= 0\\
			0 & \text{otherwise } \\
		\end{cases} & \text{ for } u_X \cdot u_Z = 1. \\
	& P^{(\text{Commutation, W}), (u_W) }_{a}  = \begin{cases}
		\cI_{2^{n+1}} & \text{if } a= 0\\
		0 & \text{otherwise } \\
	\end{cases} & \text{ for } u_X \cdot u_Z = 1. 
	\end{align*}
	For the magic square game, given $(u_X, u_Z)$ such that $u_x \cdot u_z = 1$, by \eqref{eq:genpaulicomm}, the observables $\rho^{X}(u_X)$ and  $\rho^{Z}(u_Z)$ anti-commutes. By applying~\cite[Theorem 7.11]{jiMIPRE2022a}, there exists a symmetric projective oracularizable strategy $\strategy^{\text{MS}} = (\bC^{2^{n+1}} \otimes \bC^{2^{n+1}}, \MEstate{2}^{\otimes n} \otimes \MEstate{2}, \{M_a^x \})$ such that for $b \in \{0,1\}$, 
	\begin{equation*}
		M_b^{(\text{Variable},1)} = \rho^{W}(u_W)_n \otimes \cI_2, \qquad M_b^{(\text{Variable},5)} = \rho^{W}(u_W)_n  \otimes \cI_2.
	\end{equation*}
	For $i \in \{2,3,4,6,7,8,9\}$, and $j \in \{1 ,\cdots, 6\}$,set 
	\begin{equation*}
		P_b^{(\text{Variable},i), (u_x, u_z)} = M_b^{(\text{Variable},i)}, \qquad P_b^{(\text{Constraint},j), (u_x, u_z)} = M_b^{(\text{Constraint},j)}. 
	\end{equation*}
	In the event that  $u_x \cdot u_z = 0$, set 
	\begin{equation*}
		P_b^{(\text{Variable},i), (u_x, u_z)} = P_b^{(\text{Constraint},j), (u_x, u_z)} = \begin{cases}
			\cI_{2^{n+1}} & \text{if } b= 0\\
			0 & \text{otherwise } \\
		\end{cases}. 
	\end{equation*}
	Since all measurement operators $\strategy^{\text{PB}}$ are defined within $\bigotimes_{i = 0}^{n+1} \cM_2(\bC)$, for all $a \in \cA$
	\begin{equation*}
		\left(P_{a} \otimes \cI_{\mathbf{B}} \right) \MEstate{2}^{\otimes n}_{\mathbf{AB}} \otimes \MEstate{2}_{\mathbf{AB}}  = \left( \cI \otimes P_{a}^{T} \right) \MEstate{2}^{\otimes n}_{\mathbf{AB}} \otimes \MEstate{2}_{\mathbf{AB}}
	\end{equation*}
	 for all measurements within $\strategy^{\text{PB}}$. We now verify that $\strategy^{\text{PB}}$ is indeed an oracularizable and perfect strategy by considering all possible question pairs in $\cG^{\text{PB}}_n$. 
	\begin{itemize}
		\item (Self-loop) Since all the measurements are projective, this is trivially true. 
		\item (Pauli, $W$) \textbf{--}  (Coordinate, $W$)  This follows since both question labels require a measurement in the $W$ basis. 
		\item  If $u_x \cdot u_z = 0$
		\begin{itemize}
			\item The red-edge question pairs are trivially perfect/oracularizable, since one of the measurement operators is always the identity. 
			\item  (Coordinate, $W$)  \color{red} \textbf{--}   \color{black} (Commutation, $W$): This follows since both question labels require a measurement in the $W$ basis. 
			\item  (Commutation)  \color{red} \textbf{--}   \color{black}  (Commutation, $W$): By~\eqref{eq:genpaulicomm}, $\rho^{X}(u_X)$ commutes with $\rho^{Z}(u_Z)$. Hence, $\{ \rho^{X}(u_X)_i\}_{i \in \{0,1\}}$ pairwise commute with $\{ \rho^{Z}(u_Z)_i\}_{i \in \{0,1\}}$, and the statement follows accordingly. 
		\end{itemize}
		\item If $u_x \cdot u_z = 1$
		\begin{itemize}
			\item The blue edge question pairs are trivially perfect/oracularizable, since one of the measurement operators is always the identity. 
			\item (Variable $1$)  \color{blue} \textbf{--}   \color{black}  (Coordinate, $X$):  This follows since both question labels require a measurement in the $X$ basis. 
			\item (Variable $5$)  \color{blue} \textbf{--}   \color{black}  (Coordinate, $Z$):  This follows since both question labels require a measurement in the $Z$ basis. 
			\item (Constraint)  \color{blue} \textbf{--}  \color{black}   (Variable): This follows by~\cite[Theorem 7.11]{jiMIPRE2022a}.
		\end{itemize}
		This shows that $\strategy^{\text{PB}}$ is a symmetric oracularizable strategy.
	\end{itemize}
\end{proof}


\subsection{The Introspection protocol} \label{sec:introspectionprotocol}

In this subsection, we present the introspection protocol for a game that is CL samplable, which provides the algorithm  $\texttt{QuestionReduction}_{\alpha, k}$ required in~\Cref{prop:QuestionReduction}. For $\alpha, n \in \bN$, and a CL samplable game $\cG = (\cX, \cA, \mu, D)$, where the question distribution $\mu$ is a $(k,m,p)$ CL distribution with $m \cdot p \leq n^{\alpha}$, we define the sampling procedure for the $(\cG,\alpha, k)$-introspection protocol in~\Cref{fig:Introspectionsample}, the verification procedure in~\Cref{fig:Introspectionverification}, and a typed graph for the input distribution in~\Cref{fig:Introspectiongraph}. We remark that the introspection protocol is almost the same as the one presented in~\cite[Figure 10]{jiMIPRE2022a}, with minor adjustments for clarity. 


\begin{figure}[!t]
	\centering
	\begin{tikzpicture}[scale=.8]
		
		\tikzset{type/.style args={[#1]#2}{
				draw,circle,fill,scale=0.25,
				label={[font=\scriptsize, label distance=1pt]#1:#2}
		}}
		
		\foreach \i in {1,...,6} \draw (0,8-9/8*\i)
		coordinate (Constraint-\i)
		node[type={[180]$\text{Constrant}_\i$}] {};
		
		\foreach \i in {1,...,9} \draw (1.5,9-\i)
		coordinate (Variable-\i)
		node[type={[330]$\text{Variable}_\i$}] {};
		

		\draw (4,8) coordinate (Cord-X) node[type={[45]$(\text{Coord},X)$}] {};
		\draw (4,4) coordinate (Cord-Z) node[type={[315]$(\text{Coord},Z)$}] {};

		\foreach \i in {4,...,6} \foreach \j in {1,...,3}
		\pgfmathsetmacro{\k}{\i-3+(\j-1)*3}
		\draw[blue] (Constraint-\i) -- (Variable-\k);
		\draw[blue] (Variable-1) -- (Cord-X);
		\draw[blue] (Variable-5) -- (Cord-Z);
		
		\draw (7,8) coordinate (Pauli-X) node[type={[90]$(\text{Pauli}, X)$}] {};
		\draw (7,4) coordinate (Pauli-Z) node[type={[270]$(\text{Pauli},Z)$}] {};
		\draw (5,6) coordinate (Commutation) node[type={[0]$\text{Comm}$}] {};
		\draw (4.5,7) coordinate (Commutation-X) node[type={[0]$(\text{Comm}, X)$}] {};
		\draw (4.5,5) coordinate (Commutation-Z) node[type={[0]$(\text{Comm}, Z)$}] {};
		
		\foreach \from \to in {Cord-X/Commutation-X, Cord-Z/Commutation-Z,Commutation-Z/Commutation, Commutation-X/Commutation}
		\draw[blue] (\from) -- (\to);
		\foreach \from \to in {Cord-X/Pauli-X, Cord-Z/Pauli-Z}
		\draw[blue] (\from) -- (\to);
		\foreach \i in {1,...,3} \foreach \j in {1,...,3}
		\pgfmathsetmacro{\k}{(\i-1)*3+\j}
		\draw[blue] (Constraint-\i) -- (Variable-\k);
		
		\draw (8,8) coordinate (GPauli-X) node[type={[0]$(\text{Gen Pauli}, X)$}] {};
		\draw (8,4) coordinate (GPauli-Z) node[type={[0]$(\text{Gen Pauli},Z)$}] {};
		\foreach \from \to in {GPauli-X/Pauli-X, GPauli-Z/Pauli-Z}
		\draw(\from) -- (\to);

		\draw (8.5,8.7) coordinate (Hide-1-A) node[type={[90]$(\text{Hide}_0, \decidertypedfunction^{A})$}] {};
		\draw (10.5,8.7) coordinate (Hide-2-A) node[type={[90]$(\text{Hide}_1, \decidertypedfunction^{A})$}] {};
		\draw (11.8,8.7) node {$\cdots$};
		\draw (13,8.7) coordinate (Hide-k-A) node[type={[90]$(\text{Hide}_{k-1}, \decidertypedfunction^{A})$}] {};
		\draw (15,8.7) coordinate (Read-A) node[type={[80]$(\text{Read},  \decidertypedfunction^{A})$}] {};
		
		\draw (8.5,7.3) coordinate (Hide-1-B) node[type={[270]$(\text{Hide}_0, \decidertypedfunction^{B})$}] {};
		\draw (10.5,7.3) coordinate (Hide-2-B) node[type={[270]$(\text{Hide}_1, \decidertypedfunction^{B})$}] {};
		\draw (11.8,7.3) node {$\cdots$};
		\draw (13,7.3) coordinate (Hide-k-B) node[type={[270]$(\text{Hide}_{k-1}, \decidertypedfunction^{B})$}] {};
		\draw (15,7.3) coordinate (Read-B) node[type={[80]$(\text{Read},  \decidertypedfunction^{B})$}] {};

		\draw (16,6) coordinate (Intro-B) node[type={[180]$(\text{Intro},  \decidertypedfunction^{B})$}] {};
		\draw (18,6) coordinate (Intro-A) node[type={[215]$(\text{Intro},  \decidertypedfunction^{A})$}] {};
		
		\draw (11.8,3) coordinate (Sample-A) node[type={[270]$(\text{Sample},  \decidertypedfunction^{A})$}] {};
		\draw (11.8,5) coordinate (Sample-B) node[type={[90]$(\text{Sample},  \decidertypedfunction^{B})$}] {};
		
		\foreach \from/\to in {
			Hide-2-A/Hide-1-A,
			Hide-k-A/Read-A}
		\draw[green] (\from) -- (\to);

		\foreach \from/\to in {
			Hide-2-B/Hide-1-B,
			Hide-k-B/Read-B}
		\draw[orange] (\from) -- (\to);
		
		\draw[purple] (Intro-A) -- (Intro-B);
		
		\draw[green] (GPauli-X) to [out=60,in=180] (Hide-1-A);
		\draw[orange] (GPauli-X) to [out=300,in=180] (Hide-1-B);

		\draw[orange] (GPauli-Z) to [out=45,in=180] (Sample-B);
		\draw[green] (GPauli-Z) to [out=315,in=180] (Sample-A);
		
		\draw[green] (Sample-A) to [out=0,in=270] (Intro-A);
		\draw[orange] (Sample-B) to [out=0,in=225] (Intro-B);

		\draw[green] (Read-A) to [out=0,in=90] (Intro-A);
		\draw[orange] (Read-B) to [out=0,in=135] (Intro-B);
		
	\end{tikzpicture}
	
	\caption{The typed graph $(\decidertypedtype^{\text{Intro}}_k, \decidertypedquestionpair^{\text{Intro}}_k)$ for a $(n^\alpha, k, \cG)$-introspection protocol, where each of the vertices above also contains a self-loop (which is a black edge). The purple edges are intuitively the question pair in which the provers are asked to sample honestly from the question distribution, and play the original game. The orange and green edges are questions which are designed to make sure that the provers perform the correct measurement such that $ \decidertypedfunction^{A}(s)$ and $ \decidertypedfunction^{B}(s)$ can be sampled correctly. The blue edges correspond to question pairs from the $n^{\alpha}$-Pauli basis test. We remark that the typed graph above only depends on the parameter $k$, and functions $\decidertypedfunction^{A}, \decidertypedfunction^{B}$ within the question label are presented for clarity.  }
	\label{fig:Introspectiongraph}
\end{figure}

\begin{figure}[t!]
	\centering
	\footnotesize
	\begin{gamespec}
		\setlength{\tabcolsep}{1em}
		\begin{tabularx}{\textwidth}{ l   l   X   }
			\toprule
			Question label & Question content & Answer format \\
			\midrule
			$n^{\alpha}$-PB question labels&  \multicolumn{2}{l}{See~\Cref{fig:PauliBasis}} \\
			$(\text{Gen Pauli}, W)$&  & $s_{W}  \in  \bF_{2^p}^m$\\
			$(\text{Hide}_0, \decidertypedfunction^{P})$&  & $(t_{\leq 0}^{\bot}, r_{> 0})  \in  V_{0} \times V_{> 0}$\\
			$(\text{Hide}_i, \decidertypedfunction^{P})$, $i \in [k] \setminus 0$&  & $(t_{< i}^{\text{Line}},t_{\leq i}^{\bot}, r_{> i})  \in  V_{< i} \times V_{\leq i} \times V_{> i}$\\
			$(\text{Read},  \decidertypedfunction^{P})$&  & $(t_{\text{Read}, P}^{\bot},t_{\text{Read}, P}^{\text{Line}}, a_{\text{Read}, P})  \in \bF_{2^p}^m \times \bF_{2^p}^m \times \cA$\\
			$(\text{Sample},  \decidertypedfunction^{P})$&  & $(s_{\text{Sample}}, a_{\text{Sample}, P})  \in \bF_{2^p}^m \times \cA$\\
			$(\text{Intro},  \decidertypedfunction^{P})$&  & $(x_P, a_P)  \in \bF_{2^p}^m \times \cA$\\
			\bottomrule
			\multicolumn{3}{c}{Figure: Q and A format for the $(\cG,n^\alpha, k)$-introspection protocol, with $W \in \{X,Z\}$ and $P \in\{A,B\}$. }
		\end{tabularx}
		
		\begin{center}
			\textbf{Sampling procedure}
		\end{center}
		\begin{enumerate}
			\item Sample $(u_X,u_Z) \in S^{\text{PB}}_{n^{\alpha}}\times S^{\text{PB}}_{n^{\alpha}}$,  $(n_0, n_1) \in \decidertypedtype^{\text{Intro}}_k \times \decidertypedtype^{\text{Intro}}_k$,  where $(\decidertypedtype^{\text{Intro}}_k, \decidertypedquestionpair^{\text{Intro}}_k)$ is defined in~\Cref{fig:Introspectiongraph}, and perform rejection sampling until $(n_0, n_1) \in \decidertypedquestionpair^{\text{Intro}}_k$. 
			\item Send the question label and question content corresponding to $n_0$ to one of the provers, and send the question content corresponding to the $n_1$ to the other prover.
		\end{enumerate}
	\end{gamespec}
		\caption{\footnotesize The description for the sampling procedure for the $(\cG,n^\alpha, k)$-introspection protocol. Where $\cG = (\cX, \cA, \mu, D)$, such that the distribution $\mu$ is a $(k,m,p)$ CL distribution with $m \cdot p \leq n^{\alpha}$. The CL distribution is define by two CL functions $\decidertypedfunction^A$ and $\decidertypedfunction^{B}$ with registers $\{V_i\}_{i \in [k]}$ with $\bF_{2^p}^m = \bigcup_{i \in [k]}V_i$. }
	\label{fig:Introspectionsample}
\end{figure}

\begin{figure}[b!]
	\centering
	\footnotesize
	\begin{gamespec}
		\setlength{\tabcolsep}{1em}
		\begin{center}
			\textbf{Verification procedure}
		\end{center}
		\textbf{Synchronicity/EPR test.}
		\begin{itemize}
			\item (Self-loop): The provers win iff they output the same answer. 
			\item ($n^{\alpha}$-PB question labels): If $(n_0, n_1)$ is a blue edge, refer to~\Cref{fig:PauliBasis}. 
			\item (Pauli, $W$) \textbf{--}  (Gen Pauli, $W$): Let $\pi_{p \cdot m}(t_W)$ be the first $p \cdot m$ bits of $t_W$, the provers win iff the canonical representation of $s_W$ is equal to $\pi_{p \cdot m}(t_W)$.
		\end{itemize}
		\textbf{Hiding test. } We identify $t_{< 0}^{\text{Line}} = 0 \in \bF_{2^p}$. 
		\begin{itemize}
			\item (Gen Pauli, $X$) \textbf{--} $(\text{Hide}_0, \decidertypedfunction^{P})$: Write $s_{X} =  (s_{X})_{0} +(s_{X})_{0}^{C} + (s_{X})_{> 0} \in \ker{\decidertypedfunction^{P}_{0,0}} \oplus \ker{\decidertypedfunction^{P}_{0,0}}^{C} \oplus V_{> 0}$ (where the canonical complement is defined over $V_0$), the prover wins iff $(s_{X})_{0}^{C} = t_{\leq 0}^{\bot}$ and $(s_{X})_{> 0} =  r_{> 0}$.  
			\item $(\text{Hide}_{i}, \decidertypedfunction^{P})$ \textbf{--}  $(\text{Hide}_{i+1}, \decidertypedfunction^{P})$ for $i \in [k]$: Write
			\begin{align*}
				t_{\leq i+1 }^{\bot} &= 	\tilde{t}_{\leq i }^{\bot} + \tilde{t}_{ i +1}^{\bot} \in V_{\leq i} \oplus V_{i+1}, \quad t_{< i+1 }^{\text{Line}} = 	\tilde{t}_{< i}^{\text{Line}} + \tilde{t}_{ i}^{\text{Line}} \in V_{< i} \oplus V_{i}, \\
				r_{> i } &= 	\bar{r}_{i} + \bar{r}_{i}^{C} + \bar{r}_{>i+1} \in \ker{\left(\decidertypedfunction^{P}_{i+1,t_{< i +1}^{\text{Line}} }\right)} \oplus  \ker{\left(\decidertypedfunction^{P}_{i+1,t_{< i +1}^{\text{Line}} }\right)}^{C}  \oplus V_{> i+1}, 
			\end{align*}
			In the notation above, the bar above the variable indicates that the element is decomposed from $(\text{Hide}_{i})$ and the tilde above the variable refers to elements from $(\text{Hide}_{i+1})$. The complement for $\ker{\left(\decidertypedfunction^{P}_{i+1,t_{< i +1}^{\text{Line}} }\right)}^{C}$ is over the subspace $V_i$.
			
			The provers win iff
			\begin{equation*}
				\tilde{t}_{\leq i }^{\bot} = t_{\leq i }^{\bot},\quad  \bar{r}_{i}^{C} = \tilde{t}_{i +1 }^{\bot}, \quad \tilde{t}_{< i }^{\text{Line}} = t_{< i }^{\text{Line}}, \quad \bar{r}_{>i+1} = r_{>i+1}.
			\end{equation*} 
			\item $(\text{Hide}_{k-1}, \decidertypedfunction^{P})$ \textbf{--}   $(\text{Read},  \decidertypedfunction^{P})$: The provers win iff $t_{\text{Read}, P} = t_{k-1}$, and  $t_{\text{Read}, P}^{\bot} = t_{k-1}^{\bot}$.
			\item  $(\text{Read},  \decidertypedfunction^{P})$ \textbf{--}  $(\text{Intro},  \decidertypedfunction^{P})$: The provers win iff $t_{\text{Read}, P}^{\text{Line}} = x_P$, and $a_{\text{Read}, P} = a_P$.
		\end{itemize}
		\textbf{Sampling test. }
		\begin{itemize}
			\item (Gen Pauli, $Z$) \textbf{--} $(\text{Sample}, \decidertypedfunction^{P})$: The provers win iff $s_Z = s_{\text{Sample}}$. 
			\item $(\text{Sample}, \decidertypedfunction^{P})$ \textbf{--} $(\text{Intro},  \decidertypedfunction^{P})$: The provers win iff $\decidertypedfunction^{P}(s_{\text{Sample}}) = x_P$ and  $a_{\text{Sample}, P} = a_P$.
		\end{itemize}
		\textbf{Introspection of $\cG$ }
		\begin{itemize}
			\item $(\text{Intro},  \decidertypedfunction^{A})$   \color{purple} \textbf{--}   \color{black} $(\text{Intro},  \decidertypedfunction^{B})$: The provers win iff $D(x_A,x_B,a_A, a_B) = 1$. 
		\end{itemize}
		\vspace{1em}
	\end{gamespec}
	\caption{ The description for the verification procedure for the $(\cG,n^\alpha, k)$-introspection protocol. }
	\label{fig:Introspectionverification}
\end{figure}

We give a simple example to illustrate the introspection protocol. Suppose $\cG = (\cX, \cA, \mu, D)$ is a synchronous game, where $\mu$ is a $(1, m, 3)$ CL distribution defined by the CL functions
\begin{equation}\label{eq:introexample}
	\decidertypedfunction^{A}(s_0, s_1, s_2) = (s_0+ s_1,0,0) \, \quad \text{and} \quad \, 	\decidertypedfunction^{B}(s_0, s_1, s_2)  = (s_1 + s_2, 0,0), 
\end{equation}
for all $(s_0, s_1, s_2) \in \bF_2^3$. The introspection protocol first forces the provers to prepare three copies of the $\MEstate{2}$ by using the $3$-Pauli basis test as a subroutine. In the ideal scenario, the verifier wants the prover (Alice) who receives the question arising from $\decidertypedfunction^{A}$ to perform a Pauli $Z$ measurement on the first 2 copies of $\MEstate{2}^{\otimes 3}$ in order to sample two random bits $(s_0^A, s_1^A) \sim \{0,1\}^2$, compute $\decidertypedfunction^{A}(s_0^A, s_1^A,0)$ to obtain her question for $\cG$ and play the game accordingly. Intuitively, this ``samples" the first two bits, $s_0$ and $s_1$, which is the minimum amount of information Alice needs to compute $\decidertypedfunction^{A}$. The other prover, Bob, should perform a Pauli $Z$ measurement on the last two copies $\MEstate{2}^{\otimes 3}$ in order to sample $(s_1^B, s_2^B) \sim \{0,1\}^2$, and calculate $\decidertypedfunction^{B}(0, s_1^B, s_2^B)$ to obtain his half of the question pair and output his answer accordingly. By the properties of entanglement, if Alice and Bob have performed the procedure properly, $s_1^A = s_1^V$, and hence the question distribution sampled by the two provers is precisely the same as $\mu$. In the introspection game, the verifier wants the provers to perform this ``ideal scenario" when given the $(\text{Intro}, \decidertypedfunction^{P})$ question label in~\Cref{fig:Introspectiongraph}.

To enforce honesty from the provers, the verifier cross-references the measurements made by the provers with those made in the $(\text{Pauli}, W)$ question pair for $W \in \{X,Z\}$ (which, recall, forces the provers to perform an all $X$ or all $Z$ measurement on all $\MEstate{2}$ states by the properties of the $n^{\alpha}$ Pauli basis test). In particular, the verifier wants to make sure the provers perform the following task correctly:
\begin{enumerate}
	\item The provers should only measure the register they need in order to compute the function $\decidertypedfunction^{P}$. On the example given in~\Cref{eq:introexample}, the prover receiving the question which arises from $\decidertypedfunction^{A}$ should \textbf{only} perform the Pauli $Z$ on the first $2$ copies of $\MEstate{2}^{\otimes 3}$, and not measure the last copy (i.e. the kernel of $\decidertypedfunction^{A}$). 
	\item After sampling the bits required to compute the function $\decidertypedfunction^{P}$, the provers have to correctly apply the function (instead of using some pre-prepared question pair). 
\end{enumerate}
To ensure the first task is performed correctly, the verifier cross-references the $(\text{Intro}, \decidertypedfunction^{P})$ question with the question label $(\text{Read},  \decidertypedfunction^{P})$, in which the provers, in addition to performing the Pauli $Z$ measurement, also require the provers to make a Pauli $X$ measurement on the kernel space of $\decidertypedfunction^{A}$, and are expected to output the same answer as the $(\text{Intro}, \decidertypedfunction^{P})$ question. On the example above, since Alice, given the $(\text{Intro}, \decidertypedfunction^{A})$ question label can only perform an $X$ or $Z$ measurement on the third qubit, her answer must not depend on the measurement outcome for the third qubit. The $(\text{Read},  \decidertypedfunction^{P})$ question label is then cross-referenced with the $(\text{Pauli}, X)$ question from the $3$-qubit Pauli basis test to ensure consistency for the $X$ measurement on the third qubit.
To ensure the second task, the verifier cross-references the $(\text{Intro}, \decidertypedfunction^{P})$ question with the question label $(\text{Sample}, \decidertypedfunction^{P})$, in which the provers are expected to sample the entirety of the seed $s$ by performing Pauli $Z$ measurements on \textbf{all} of their $\MEstate{2}$ bits, compute the corresponding question $\decidertypedfunction^{P}(s)$ and generate the corresponding answer. The $(\text{Sample}, \decidertypedfunction^{P})$  question label is cross-referenced with the $(\text{Pauli}, Z)$ question to ensure consistency. 

In general, there are two additional problems. If the CL function $\decidertypedfunction^{P}$ is a level $k$ CL function, then the ``Read" question cannot be cross-checked with the $(\text{Pauli}, X)$ question, since the kernel space of the linear function for each level depends on the computation step from the previous level. Intuitively, the behaviour of the prover for the ``Read" question is enforced by a series of ``Hide" questions, each designed to enforce the ``honest measurement" for the Read question for one level. Since a CL distribution is defined using two CL functions which map subspaces of $\bF_{2^{\seedlength}}^m$ rather than of $\bF_2^m$, generalized Pauli measurements are needed for the introspection protocol. In combination with~\Cref{lem:PaulitogeneralPauli}, we see that for any $\seedlength \in \bN$, the $\seedlength \cdot n$ qubit Pauli basis test can also serve as a rigidity test for generalized Pauli measurement over $\MEstate{\seedlength}$, and we use the $(\text{Gen Pauli}, W)$ to convert between these two types of self-test.

We have the following theorem regarding the $(\cG,n^\alpha, k)$-introspection protocol. We remark that in comparison to~\cite[Theorem 8.3]{jiMIPRE2022a}, there is no longer dependency on $\alpha$ (the variable $R$ or $n \alpha$ in~\cite{jiMIPRE2022a}). This is due to our EPR tester (the $n$ qubit Pauli basis test) not using the low-degree test as a part of the subroutine. 

\begin{theorem}[Properties of the $(\cG,n^\alpha, k)$-introspection protocol] \label{thm:Introspectgame}
	Let $n, \alpha \in \bN$, and let $\cG = (\cX, \cA, \mu, D)$ be a $k$-th CL samplable game where the question distribution $\mu$ is a $(k,m,p)$ CL distribution with $m \cdot p \leq n^{\alpha}$. Let $\cG^{\text{intro}} = (\cX^{\text{intro}}, \cA^{\text{intro}}, \mu^{\text{intro}}, D^{\text{intro}})$ be the (typed) $(\cG,\alpha, k)$-introspection protocol specified in~\Cref{fig:Introspectionsample} and~\Cref{fig:Introspectionverification}. For $t \in \{*, co\}$, the following holds:
	\begin{itemize}
		\item (Sample complexity): $\cG^{\text{intro}}$ is samplable via a $(\decidertypedtype^{\text{Intro}}_k, \decidertypedquestionpair^{\text{Intro}}_k, \{\decidertypedfunction^{v}\}_{v \in \decidertypedtype^{\text{Intro}}} )$ typed CL distribution, where each $\decidertypedtype^{\text{Intro}}: \bF_2^{2 \cdot \alpha \cdot \lfloor \log(n) \rfloor}$ is a first level CL function. Furthermore, the question distribution only depends on the parameter $n^{\alpha}$ and k.
		\item (Completeness): If there exists a perfect oracularizable strategy for $\cG$ in model $t$, then there exists a perfect oracularizable strategy for $\cG^{\text{intro}}$ in model $t$. 
		\item (Soundness): There exists a polynomial $\mathbf{P}^{\text{Intro}} (\eps, k, \alpha)$ such that
		\begin{equation*}
			 \omega^t(\cG) \leq 1- \eps \Longrightarrow 	\omega^t(\cG^{\text{Intro}}) \leq 1 - \mathbf{P}^{\text{Intro}}(\eps, \exp(k)). 
		\end{equation*}
	\end{itemize}
\end{theorem}
\begin{proof}
	Fix $n, \alpha, k  \in \bN$, model  $t \in \{*, co\}$. Let $\cG = (\cX, \cA, \mu,D)$ be a game that satisfies the description for~\Cref{thm:Introspectgame}, and let $\cG^{\text{Intro}} = (\cX^{\text{Intro}} , \cA^{\text{Intro}} , \mu^{\text{Intro}} ,D^{\text{Intro}} )$ be the (typed)-introspection protocol. Let $(\decidertypedtype^{\text{Intro}}_k, \decidertypedquestionpair^{\text{Intro}}_k)$ be the typed graph as given in~\Cref{fig:Introspectiongraph}. Since $\cG$ is a $k$-th CL samplable game with the input distribution $\mu$ being a $(k,m,p)$ CL distribution, by~\Cref{def:CLdistribution}, there exist two CL functions $\decidertypedfunction^{A}, \decidertypedfunction^{B}: \bF_{2^p}^m \rightarrow \bF_{2^p}^m$ over registers $\{V_j\}_{j \in [k]}$ which can be used to sample from $\mu$. Furthermore, recall from the preliminaries that $U_{2 \rightarrow p}$ is the unitary map given in~\Cref{lem:PaulitogeneralPauli}, and we write $U_{2 \rightarrow p}^m$ as a unitary acting on $\bC^{2^{n^{\alpha}}}$ defined by $U_{2 \rightarrow p}^{\otimes m} \otimes \cI_{2^{n^{\alpha} - m \cdot p}}$.
	
	
	For the ``sample complexity" clause in the theorem, since the ``question content" specified in~\Cref{fig:Introspectionsample} are empty except for the question labels from the $n^{\alpha}$-Pauli basis game, in which the corresponding CL functions are already specified in the proof of~\Cref{thm:completenessPaulibasis}; the distribution $\mu^{\text{intro}}$ is samplable via a $(\decidertypedtype^{\text{Intro}}_k, \decidertypedquestionpair^{\text{Intro}}_k, \{\decidertypedfunction^{v}\}_{v \in \decidertypedtype^{\text{Intro}}} )$ typed CL distribution as specified by the theorem statement. The ``furthermore" part follows from~\Cref{fig:Introspectionsample} depends only on $n^{\alpha}$ for the Pauli basis test and $k$ for the number of ``Hide" question labels. This concludes the proof for the ``Sample complexity" clause of the theorem. 
	
	\jqnote{You need to divide completeness and the soundness proof into it's own subsection.}
	
	For the ``completeness" clause in the theorem, let $\strategy$ be a perfect oracularizable strategy in model $t$ for $\cG$. Since $\cG$ is synchronous, $\strategy$ is synchronous and hence by~\Cref{lem:structsync}, we can write $\strategy = (\cL^2(\alicealg, \tau), \ket{\tau}, \{A_a^x \})$ as a projective synchronous strategy.
	
	
	Before defining the perfect strategy for $\cG^{\text{Intro}}$, we first introduce some notations for some data processed Pauli measurements which is used as a part of the perfect strategy. For $P \in \{A,B\}$, we define the data processing measurement for the CL function $\decidertypedfunction^{P}$ as 
	\begin{equation*}
		\rho^{p, Z}_{[\decidertypedfunction^{P}|x]} = \sum_{a \in V| \decidertypedfunction^{P}(a) = x} \rho^{p, Z}_{a}, 
	\end{equation*}
	for all $x \in V$. By the definition of a CL function given in~\Cref{def:CLverifier}, the measurement above is equivalent to the following: The prover first performs the data processing measurement $\{\rho^{p, Z}_{[\decidertypedfunction^{P}_{0,0}|x_0]}\}_{x_0 \in V_0}$ to sample some $x_0 \in V_0$. Then, for $1 \leq i < k$, the prover performs the measurement 
	\begin{equation}\label{eq:introcompdata}
		\left\{\rho^{p, Z}_{[\decidertypedfunction^{P}_{i,x_{< i}}|x_i]} \right\}_{x_i \in V_i}
	\end{equation}
	to obtain measurement outcome $x_i$ and compute $x_{< i+1} = x_i + x_{< i}$. The final measurement outcome is $x =  x_{< k}$. For $j \in [k]$, we define the measurement operator 
	\begin{equation} \label{eq:introcomp1}
		\left\{\rho^{p, Z}_{[\decidertypedfunction^{P}_{< j-1}|x_{< j}]}\right\}_{x_{< j} \in V_{\leq j}}
	\end{equation}
	similarly to $\rho^{p, Z}_{[\decidertypedfunction^{P}|x]}$, except that only the first $j$ measurements from~\eqref{eq:introcompdata} are performed. 
	
	Recall from~\Cref{sec:fieldintro} that, for a linear map $\decidertypedfunction$, we use $\decidertypedfunction^{\bot}$ to denote the linear map that projects onto $\left(\ker{(\decidertypedfunction)}^{\bot}\right)^{C}$. By~\Cref{lem:linePaulicommutation}, for a fixed $x_{< j} \in V_{< j}$, the measurement operator $\left\{\rho^{p, Z}_{[\decidertypedfunction^{P}_{i,x_{< i}}|x_i]} \right\}_{x_i \in V_i}$ pairwise commute with the measurement operator $\left\{\rho^{p, X}_{[\left(\decidertypedfunction^{P}\right)^{\bot}_{i,x_{< i}}|x_i^{\bot}]} \right\}_{x_i^{\bot} \in V_i}$. 

	We define a perfect symmetric oracularizable strategy $\strategy^{\text{Intro}} = (\bC^{2^{n^\alpha +1} } \otimes \bC^{2^{n^{\alpha} +1}} \otimes \cH, \MEstate{2}^{\otimes (n^\alpha +1)}\otimes \ket{\tau}, \{ M^x_a\} )$ for $\cG^{\text{Intro}}$ as follows: For the question labels $v \in  \decidertypedtype^{\text{Intro}}_k$ which intersects a blue edge as specified in~\Cref{fig:Introspectiongraph}. We define the measurement operator $M^x_a$
	\begin{align*}
		M^v_a = P_a^b \otimes \cI_{A}
	\end{align*}
	where $\strategy^{\text{PB}} = (\bC^{2^{n^{\alpha} +1}} \otimes \bC^{2^{n^{\alpha} +1}} , \MEstate{2}^{\otimes (n^{\alpha} +1)}, \{P_a^x \})$ is the perfect oracularizable strategy for the $n^{\alpha}$-Pauli basis test guaranteed by~\Cref{thm:completenessPaulibasis}. Notably, for $t_W \in \{0,1\}^{n^{\alpha}}$
	\begin{equation*}
		M^{(\text{Pauli}, W)}_{t_W} = \rho_{t_W}^{W} \otimes \cI_2 \otimes \cI_{\alicealg}.
	\end{equation*}
	
	Let $\cI_R =  \cI_{2^{n^{\alpha} - p \cdot m +1}}$. We define the measurement for the rest of the question label as follows:
	\begin{align*}
			M^{(\text{Gen Pauli}, W)}_{t_W} &= U_{2 \rightarrow p}^m(\rho^{p, W}_{s_W}) (U_{2 \rightarrow p}^m)^* \otimes \cI_{R} \otimes \cI_{\alicealg} &\text{ for all } s_W \in \bF_{2^p}^m, \\
			M^{(\text{Sample}, \decidertypedfunction^{P})}_{(s, a)} &= U_{2 \rightarrow p}^m(\rho^{p, Z}_{s})(U_{2 \rightarrow p}^m)^* \otimes  \cI_{R} \otimes A^{\decidertypedfunction^{P}(s)}_{a } &\text{ for all } s \in \bF_{2^p}^m, a \in \cA, \\
			M^{(\text{Intro}, \decidertypedfunction^{P})}_{(x, a)} &= U_{2 \rightarrow p}^m(\rho^{p, Z}_{[\decidertypedfunction^{P}|x]})(U_{2 \rightarrow p}^m)^* \otimes  \cI_{R} \otimes A^{x}_{a } &\text{ for all } s \in \bF_{2^p}^m, a \in \cA, 
	\end{align*}
	We define the measurement operator for $M^{(\text{Read}, \decidertypedfunction^{P})}_{t, t^{\bot},a}$ as follows: The prover first performs the measurement $U_{2 \rightarrow p}^m \left(\rho^{p, W}_{[\decidertypedfunction^{P}|t]}\right)(U_{2 \rightarrow p}^m)^*\otimes  \cI_{R} \otimes A^{x}_{a}$ (where the procedure for performing $\rho^{p, W}_{[\decidertypedfunction^{P}|t]}$ is defined in~\eqref{eq:introcompdata}) and samples $(t,a) \in \cX \times \cA$. Then the prover performs the measurement 
	\begin{equation*}
		\left\{U_{2 \rightarrow p}^m \left(\rho^{p, X}_{[\left(\decidertypedfunction^{P}_x\right)|x^{\bot}]}\right)(U_{2 \rightarrow p}^m)^* \otimes\cI_{R} \otimes \cI_{\alicealg} \right\}_{t^{\bot} \in V}.
	\end{equation*}
	Since these measurements commute, the measurement $M^{(\text{Read}, \decidertypedfunction^{P})}_{t, t^{\bot},a}$, defined as the product of the two measurements described, is a well-defined measurement. 
	
	For $i \in [k]$, the measurement operator for $M^{(\text{Hide}_{i}, \decidertypedfunction^{P})}_{t_{< i}, t^{\bot}_{\leq i},r_{> i}}$ is defined in a similar way. The prover first performs the measurement 
	\begin{equation*} 
		\left\{ U_{2 \rightarrow p}^m  \left(\rho^{p, Z}_{[\decidertypedfunction^{P}_{< j-1}|t_{< i}]} \right) (U_{2 \rightarrow p}^m)^* \otimes\cI_{R} \otimes \cI_{\alicealg} \right\}_{t_{< i} \in V_{\leq j}}
	\end{equation*}
	to sample $t_{< i} \in V_{< j}$. Then performs the measurement 
	\begin{equation*} 
		\left\{U_{2 \rightarrow p}^m  \left(\rho^{p, X}_{r_{> i}} \right) (U_{2 \rightarrow p}^m)^* \otimes\cI_{R} \otimes \cI_{\alicealg} \right\}_{r_{> i} \in V_{> j}},
	\end{equation*}
 	 to sample $r_{> i} \in V_{> j}$. We remark that by the comment after~\Cref{lem:PaulitogeneralPauli}, these two measurements commute. Lastly, the prover performs the measurement
 	 \begin{equation} 
 	 	\left\{  U_{2 \rightarrow p}^m \left(\rho^{p, X}_{\left[ \left(\decidertypedfunction^{P}_{\leq i, x_{< j}}\right)^{\bot}|x_{\leq i}^{\bot}\right]} \right)(U_{2 \rightarrow p}^m)^* \otimes\cI_{R} \otimes \cI_{\alicealg}   \right\}_{x_{\leq i}^{\bot} \in V_{\leq i}}
 	 \end{equation}
   	to sample $t^{\bot}_{\leq i} \in V_{\leq j}$. This measurement commutes with the first measurement as proven above and commutes with the second measurement because both are generalized Pauli $X$ measurements. Hence $M^{(\text{Hide}_{i}, \decidertypedfunction^{P})}_{t_{< i}, t^{\bot}_{\leq i},r_{> i}}$ defined as the product of the above three measurements is a well-defined measurement. For clarity, we write all the measurement operator on the table below. 
	\begin{table}[!htb]
		\centering
		\footnotesize
      \renewcommand{\arraystretch}{1.8}
	\begin{tabularx}{\textwidth}{c c c c c c c}
			\hline
		Label	& $V_0$ &  $V_1$ &$V_2$ & $\cdots$ &  $V_{k-1}$ & $\alicealg$   \\
			\hline
		$(\text{Gen Pauli}, X)$ & $(s_{W})_0 \sim \sigma^{X}$ & $(s_{W})_1 \sim\sigma^{X}$ & $(s_{W})_2 \sim\sigma^{X}$ &  &$(s_{W})_{k-1} \sim\sigma^{X}$ & $\cI_{\alicealg}$ \\
		\hline
		$(\text{Hide}_{0}, \decidertypedfunction^{P})$ & $t_0^{\bot} \sim \sigma^{X}_{(\decidertypedfunction^P_{0})^{\bot}}$ & $r_1 \sim\sigma^{X}$ & $r_2 \sim\sigma^{X}$ &  &$r_{k-1} \sim\sigma^{X}$ & $\cI_{\alicealg}$ \\
		 \hline
		\multirow{2}{*}{$(\text{Hide}_{1}, \decidertypedfunction^{P})$} &$t_0^{\text{Line}} \sim \sigma^{Z}_{\decidertypedfunction^P_0}$ & \multirow{2}{*}{$t_1^{\bot} \sim \sigma^{X}_{(\decidertypedfunction^P_1)^{\bot}}$} & \multirow{2}{*}{$r_2 \sim \sigma^{X}$} &   $\cdots$  & \multirow{2}{*}{$r_{k-1} \sim \sigma^{X}$} & \multirow{2}{*}{$\cI_{\alicealg}$} \\
		&$t_0^{\bot} \sim \sigma^{X}_{(\decidertypedfunction^P_{0})^{\bot}}$ &  &  &   $\cdots$  & & \\
		\hline
		\multirow{2}{*}{$(\text{Hide}_{2}, \decidertypedfunction^{P})$} &$t_0^{\text{Line}} \sim \sigma^{Z}_{\decidertypedfunction^P_0}$ & $t_1^{\text{Line}} \sim\sigma^{Z}_{\decidertypedfunction^P_1}$ & \multirow{2}{*}{$t_2^{\bot} \sim \sigma^{X}_{(\decidertypedfunction^P_{2})^{\bot}}$} &   $\cdots$  & \multirow{2}{*}{$r_{k-1} \sim \sigma^{X}$} & \multirow{2}{*}{$\cI_{\alicealg}$} \\
		&$t_0^{\bot} \sim \sigma^{X}_{(\decidertypedfunction^P_{0})^{\bot}}$ & $t_1^{\bot} \sim \sigma^{X}_{(\decidertypedfunction^P_{1})^{\bot}}$ & $ $ &   $\cdots$  & & \\
		\hline
		\multicolumn{7}{c}{$\cdots$} \\
	 \hline
		\multirow{2}{*}{$(\text{Hide}_{k-1}, \decidertypedfunction^{P})$} &$t_0^{\text{Line}} \sim \sigma^{Z}_{\decidertypedfunction^P_0}$ & $t_1^{\text{Line}} \sim\sigma^{Z}_{\decidertypedfunction^P_1}$ & $t_2^{\text{Line}} \sim\sigma^{Z}_{\decidertypedfunction^P_2}$ &   $\cdots$  & \multirow{2}{*}{$t_{k-1}^{\bot} \sim \sigma^{X}_{(\decidertypedfunction^P_{k-1})^{\bot}}$} & \multirow{2}{*}{$\cI_{\alicealg}$} \\
		&$t_0^{\bot} \sim \sigma^{X}_{(\decidertypedfunction^P_{0})^{\bot}}$ & $t_1^{\bot} \sim \sigma^{X}_{(\decidertypedfunction^P_{1})^{\bot}}$ & $t_2^{\bot} \sim\sigma^{X}_{(\decidertypedfunction^P_{2})^{\bot}}$ &   $\cdots$  & & \\
		\hline
		\multirow{2}{*}{$(\text{Read}, \decidertypedfunction^{P})$} &$(t^{\text{Line}})_0 \sim \sigma^{Z}_{\decidertypedfunction^P_0}$ & $(t^{\text{Line}})_1 \sim\sigma^{Z}_{\decidertypedfunction^P_1}$ & $(t^{\text{Line}})_2 \sim\sigma^{Z}_{\decidertypedfunction^P_2}$ &   $\cdots$  &$(t^{\text{Line}})_{k-1} \sim\sigma^Z_{\decidertypedfunction^P_{k-1}}$ & \multirow{2}{*}{$a \sim A^{t^{\text{Line}}}_{a}$} \\
		 &$(t^{\bot})_0 \sim \sigma^{X}_{(\decidertypedfunction^P_{0})^{\bot}}$ & $(t^{\bot})_1 \sim \sigma^{X}_{(\decidertypedfunction^P_{1})^{\bot}}$ & $(t^{\bot})_2 \sim\sigma^{X}_{(\decidertypedfunction^P_{2})^{\bot}}$ &   $\cdots$  &$(t^{\bot})_{k-1} \sim \sigma^{X}_{(\decidertypedfunction^P_{k-1})^{\bot}}$ & \\
		 
		\hline
		$(\text{Intro}, \decidertypedfunction^{P})$  & $t_0 \sim \sigma^{Z}_{\decidertypedfunction^P_0}$ & $t_1 \sim\sigma^{Z}_{\decidertypedfunction^P_1}$ & $t_2 \sim\sigma^{Z}_{\decidertypedfunction^P_2}$ &   $\cdots$  &$t_{k-1} \sim\sigma^Z_{\decidertypedfunction^P_{k-1}}$ & $a \sim A^{t}_{a}$ \\
		\hline
		$(\text{Sample}, \decidertypedfunction^{P})$ & $s_0 \sim \sigma^{Z}$ & $s_1 \sim\sigma^{Z}$ & $s_2 \sim\sigma^{Z}$ &  $\cdots$ &$s_{k-1} \sim\sigma^Z$ & $a \sim A^{\decidertypedfunction^{P}(s)}_{a}$ \\
		\hline
		$(\text{Gen Pauli}, Z)$ & $(s_{Z})_0 \sim \sigma^{Z}$ & $(s_{Z})_1 \sim\sigma^{Z}$ & $(s_{Z})_2 \sim\sigma^{Z}$ & $\cdots$  &$(s_Z)_{k-1} \sim\sigma^Z$ & $\cI_{\alicealg}$ \\
		\hline
		\end{tabularx}
		\caption{Summary of the measurement operator $M^v_a$, and as well as the output being sampled from each measurement operator. The notation $x \sim M$ are the variable $x$ sampled from the measurement operator. For $i \in [k]$, the measurement $\sigma^{Z}_{\decidertypedfunction^P_i}$ (resp.$\sigma^{X}_{(\decidertypedfunction^P_{i})^{\bot}}$) above are shorthand for $\sigma^{Z}_{[\decidertypedfunction^P_{i, t_{<i}}|x]}$ (resp. $\sigma^{X}_{[(\decidertypedfunction^P_{i, t_{<i}})^{\bot}| x]}$) (where the $v_i$ depends on the previous measurement outcome). We also omit the conjugation by $ U_{2 \rightarrow p}^m$ for clarity. }
		\label{fig:PerfectstrategyIntro}
	\end{table}

	First, the measurement $M^v_a$ are projective, as given any $v \in \decidertypedtype^{\text{Intro}}_k$, the measurements $M^v_a$ are defined by products of projective measurements which all commute with each other. Since $M^v_a \in \bigotimes_{i \in [n^{\alpha}+1]} \cM_2(\bC) \otimes \alicealg$, we have
	\begin{equation*}
 		\left(M^x_a \otimes \cI_{2^{n^{\alpha} +1}} \right) \MEstate{2}^{n^{\alpha} +1} \ket{\tau} = 	\left(\cI_{2^{n^{\alpha} +1}} \otimes (M^x_a)^{op} \right) \MEstate{2}^{n^{\alpha} +1} \ket{\tau} 
	\end{equation*} 
	where $M^x_a$ above is acting on one registers of the entangled state $\MEstate{2}^{n^{\alpha} +1}$ and the state $\ket{\tau}$. Thus, the strategy $\decidertypedtype^{\text{Intro}}_k$ succeeds with probability $1$ on the consistency equations. By~\Cref{thm:soundnessPaulibasis}, the strategy $\strategy^{\text{Intro}}_k$ is perfect and oracularizable when restricted to the question pair restricted to the blue edge (i.e. the $n^{\alpha}$-Pauli basis test) within~\Cref{fig:Introspectiongraph}. By~\Cref{lem:PaulitogeneralPauli}, the strategy $\strategy^{\text{Intro}}_k$ is perfect and oracularizable when restricted to the question pair $(\text{Pauli}, W)$ -- $(\text{Gen Pauli}, W)$. When restricted to the question pair $(\text{Intro}, \decidertypedfunction^{A})$ -- $(\text{Intro}, \decidertypedfunction^{B})$, by construction, the question pair $(t_A,t_B)$ sampled by the measurement operator of $\strategy^{\text{Intro}}$ precisely corresponds to the question distribution $\mu$, as $(t_A,t_B) = (\decidertypedfunction^{A}(s), \decidertypedfunction^{B}(s))$ for some $s \in V$. Since $\strategy$ is a perfect oracularizable strategy for the game $\cG$, $\strategy^{\text{Intro}}$ is also a perfect oracularizable strategy when restricted to the ``Intro" question pair. It is straightforward to verify that $\strategy^{\text{Intro}}$ remains a perfect and oracularizable strategy for the remainder question pairs by the table above, concluding the proof for ``completeness" part of the theorem. 
	
	For ``soundness", suppose that $\omega^t(\cG^{\text{Intro}}) > 1 - \eps$, we wish to show that $\omega^t(\cG) > 1 - O(\poly(\exp{(k)}, \eps))$. Let $\strategy = (\cL^2(\alicealg, \tau), \sigma \ket{\tau}, \{ A_{a}^x \}, \{ (B_{a}^x)^{op} \})$	be a tracially embeddable strategy in model $t$ with $\omega(\cG^{\text{Intro}}, \strategy) > 1 -\eps$. We start the proof by first showing the following claim, this is an analogue of~\cite[Lemma 8.19]{jiMIPRE2022a}
	\begin{lemma}
		There is a symmetric strategy
		\begin{equation*}
			\strategy' =  (\bC^{2^{n^{\alpha}}}_{A_1} \otimes \bC^{2^{n^{\alpha}}}_{B_1} \otimes \bC^{2^{n^{\alpha}}}_{A_2} \otimes \bC^{2^{n^{\alpha}}}_{B_2} \otimes  \cL^2(\alicealg, \tau), \ket{\text{ME}_2}^{\otimes n^{\alpha}}_{A_1 B_1} \otimes \ket{\text{Aux}}_{A_2 B_2 \alicealg},\{\hat{P^x_a} \})
		\end{equation*}
		for some vector state $\ket{\text{aux}}_{A_2 B_2 \alicealg} \in \bC^{2^{n^{\alpha}}}_{A_2} \otimes \bC^{2^{n^{\alpha}}}_{B_2} \otimes  \cL^2(\alicealg, \tau)$. Furthermore, $\omega(\cG^{\text{Intro}}, \strategy') > 1 - \poly(n,\eps)$ , and for $W \in \{X,Z\}$. 
		\begin{equation}
			P^{(\text{Pauli, W})}_{u} = \rho^{W}_u
		\end{equation}
	\end{lemma}
	
	\begin{proof}
		Consider $\strategy$ when restricted to the $n^{\alpha}$-Pauli basis test (the blue vertices in~\Cref{fig:Introspectiongraph}). Since by sampling a random question pair, there is a $O(k)$ probability that a question from the Pauli basis test is selected. This implies that $\strategy$ succeeds on the $n^{\alpha}$-Pauli basis test with probability at least $1 - O(k \cdot\eps)$. By~\Cref{thm:soundnessPaulibasis}, there exist two isometries $V_{A}, V_{B}$ with $(V_B \tensor \cI_{2^{2 \cdot n^{\alpha}}}) V_{A} =  V_{A} (V_B \tensor \bI_{2^{2 \cdot n^{\alpha}}})$; a state $\ket{\text{Aux}}_{A_2 B_2 \alicealg}  \in \bC^{2^{2n^\alpha}} \otimes \cL^2(\alicealg, \tau)$ such that	\begin{equation} \label{eq:introsoundness1}
			\left\| \left(  V_B \tensor \cI_{2^{2n^{\alpha}}}  \right) V_{A} (\sigma \ket{\tau})  -\ket{\text{ME}_2}^{\tensor n^{\alpha}} \ket{\text{Aux}}   \right\|^2 \leq O\left(\poly(k, \eps)\right),
		\end{equation}
		and for all $W \in \{X, Z\}$ and $u \in \bF_{2^{n^\alpha}}$
		\begin{align} 
			\nonumber &\|  ((V_{A} A^{\text{Pauli, W}}_u  V_{A}^*)_{ A_1 A_2 \alicealg} \tensor (\cI_{2^{2n^{\alpha}}} )_{B_1 B_2} - \\
			\label{eq:introsoundnessrig1} &\quad  (\rho^{W}_u)_{A_1} \tensor (\cI_{2^{n^{\alpha}}})_{A_2} \otimes  (\cI_{2^{2n^{\alpha}}} )_{B_1 B_2}  \otimes (\cI_{\cH})_{\alicealg})  \ket{\text{Aux}}_{\alicealg A_2 B_2}  \MEstate{2}^{\tensor n}_{A_1 B_1}  \|^2 \leq O\left(\poly(\eps )\right), \\
			\nonumber &\|  ( \left(V_{B} ( B^{(\text{Pauli, W})}_u)^{\text{op}}  V_{B}^*\right)_{\alicealg B_1 B_2} \tensor (\cI_{2^{2n^{\alpha}}} )_{A_1 A_2} - \\
			\label{eq:introsoundnessrig2} &\quad  (\rho^{W}_u)_{B_1} \tensor (\cI_{2^{n^\alpha}})_{B_2} \otimes  (\cI_{2^{2n^\alpha}} )_{A_1 A_2} \otimes (\cI_{\cH})_{\alicealg} ) \ket{\text{Aux}}_{\alicealg A_2 B_2}  \MEstate{2}^{\tensor n^\alpha}_{A_1 B_1}  \|^2 \leq O\left(\poly(\eps )\right).
		\end{align}
		For each $(x, a) \in \cX^{\text{Intro}} \times \cA^{\text{Intro}}$, we define $\hat{A_a^x} = V_{A} A^{x}_a V_{A}^*$, and likewise $\hat{B_a^x} = V_B B^{x}_a V_B^*$. Define $\strategy_1$ as a strategy which uses the state $\ket{\text{EPR}}^{\tensor n^{\alpha}}_{A_1 B_1} \ket{\text{Aux}}_{A_2 B_2 \alicealg}$ and the measurement operator $\hat{A_a^x}$ and $\hat{B_a^x}$ for all questions instead. By~\eqref{eq:introsoundness1}, $\strategy_1$ succeeds in $\cG^{\text{Intro}}$ with probability $1 - O(\poly(k, \eps))$. By the description given in~\Cref{fig:PerfectstrategyIntro}.  Since $\cG^{\text{Intro}}$ is $O(k)$-balance, this implies that $\strategy_1$ is $O(\poly(k,\eps))$-synchronous; hence by~\Cref{cor:orthogonalizationlemmasync}, there exists a projective, symmetric strategy
		\begin{equation*}
			\strategy_2 = (\bC^{2^{n^{\alpha}}}_{A_1} \otimes \bC^{2^{n^{\alpha}}}_{B_1} \otimes \bC^{2^{n^{\alpha}}}_{A_2} \otimes \bC^{2^{n^{\alpha}}}_{B_2} \otimes  \cL^2(\alicealg, \tau), \ket{\text{EPR}}^{\tensor n^{\alpha}}_{A_1 B_1} \ket{\text{Aux}}_{A_2 B_2 \alicealg},  \{\hat{P^x_a}\})
		\end{equation*}
		such that $\hat{A_a^x} \approx_{O(\delta)} \hat{P^x_a}$ with $\omega(\strategy_2, \cG^{\text{Intro}}) > 1 - O(\poly(n))$. 
		
		Define $\strategy'$ as the same measurement operator as $\strategy_2 $, except for the question label (Pauli, W) where instead the Pauli measurements $(\rho^{W}_u)_{A_1}$ (resp. $(\rho^{W}_u)_{B_1}$) are used instead. The lemma then follows by combining~\Cref{eq:introsoundnessrig1} and $\hat{A_a^x} \approx_{O(\delta)} \hat{P^a_x}$. 
	\end{proof}

	We wish to transform the underlying state of $\strategy'$ in a way that the underlying entangled state is $\ket{\text{ME}_{2^p}}^{\otimes m}$ instead, consider the strategy $\strategy''$, which is defined on the same Hilbert space as $\strategy'$ except that the under lying state is 
	\begin{equation*}
	 \left((U_{2 \rightarrow p}^m)_{A_1} \otimes (U_{2 \rightarrow p}^m)_{B_1} \otimes \cI_{A_2 B_2 \alicealg} \right)\ket{\text{ME}_2}^{\otimes n^{\alpha}}_{A_1 B_1} \otimes \ket{\text{Aux}}_{A_2 B_2 \alicealg}
	\end{equation*}
	and the measurement operator is defined as $P_a^x = (U_{2 \rightarrow p}^m \otimes \cI_{2^{n^{\alpha}}} \otimes \cI_{\alicealg})^* \hat{P_a^x} (U_{2 \rightarrow p}^m \otimes \cI_{2^{n^{\alpha}}} \otimes \cI_{\alicealg})$. Since $U_{2 \rightarrow p}^m$ is a unitary, $\omega(\strategy', \cG^{\text{Intro}}) = \omega(\strategy'', \cG^{\text{Intro}})$. By~\Cref{lem:PaulitogeneralPauli}, we can rewrite the state in $\strategy''$ as 
	\begin{equation*}
		\left(\ket{\text{ME}_{2^p}}^{\otimes m} \otimes \ket{\text{ME}_{2}}^{\otimes n^{\alpha} - p \cdot m} \right)_{A_1 B_1}  \otimes \ket{\text{Aux}}_{A_2 B_2 \alicealg}
	\end{equation*}
	and $P^{\text{\text{Pauli}, W}}_{a}  = (\rho^{p, W}_{\kappa^{-1}(\pi_{\leq p \cdot k}(a))})_{A_1}\otimes (\rho^{W}_{\pi_{> p \cdot k}(a)} \otimes \cI)_{\alicealg}$ with $(P^{\text{Pauli, W}}_{a})^{op} = (\rho^{p, W}_{\kappa^{-1}(\pi_{\leq p \cdot k}(a))})_{B_1}\otimes (\rho^{W}_{\pi_{> p \cdot k}(a)} )_{\alicealg}$. Since the question pair (Pauli, W) \textbf{--} (Gen Pauli, W) occurs with probability $O(\frac{1}{k})$, 
	\begin{equation*}
		(\rho^{p, W}_{s}\otimes \rho^{W}_{\pi_{> p \cdot k}(a)})  \simeq_{O(\poly(k, \eps))} P^{(\text{Gen Pauli, W})}_s, 
	\end{equation*}
	and since $\rho^{W}_{\pi_{> p \cdot k}(a)})$ is a set of PVM, by summing over $\rho^{W}$ and apply~\Cref{lem:closetodistance}, 
	\begin{equation*}
		\rho^{p, W}_{s}  \approx_{O(\poly(k, \eps))} P^{(\text{Gen Pauli, W})}_s. 
	\end{equation*}
	For simplicity of notation, we rewrite $\strategy'$ as follows. We shrink the registers $A_1$ and $B_1$ to include only the first $\ket{\text{ME}_{2^p}}^{\otimes m}$ pair, and combine the remaining parts of $A_1$ and $B_1$, as well as the registers $A_2$ and $B_2$, into the ``$\alicealg$" infinite-dimensional register. Hence, we can write
	\begin{equation*}
		\strategy'' =  \left(\bC^{2^{p\cdot m}}_{A_1} \otimes \bC^{2^{p \cdot m}}_{B_1} \otimes  \cL^2(\bC^{2^{4n^{\alpha} - 2p \cdot m}} \otimes  \alicealg, \Tr\otimes \tau), \ket{\text{ME}_{2^p}}^{\otimes m}_{A_1 B_1} \otimes \ket{\text{Aux}}_{\alicealg}, \{ P_{a}^x \} \right),
	\end{equation*}

	The remainder of the proof proceeds similarly as~\cite[Section 8.4.3]{jiMIPRE2022a}, except we use the notation from~\Cref{fig:tracialemdtofd} to translate the proof from the finite-dimensional setting to the tracially embeddable strategies setting. The full proof is provided in~\Cref{sec:QRsoundness} for completeness. 
\end{proof}

\subsection{Proof of~\Cref{prop:QuestionReduction}} \label{sec:QRproof}

In this subsection, we give a proof for~\Cref{prop:QuestionReduction}.
\begin{proof}
	Fix the constant $\alpha,k \in \bN$. We define the algorithm $\texttt{QuestionReduction}_{\alpha, k}$ as follows. Given a pair of Turing machine $(\samplerTM, \deciderTM)$, we first describe a sequence of typed samplable games $\cG_{n}^{\text{Intro}}$, then we use~\Cref{lem:detypeandsync} to convert $\cG_{n}^{\text{Intro}}$ into a CL samplable game as desired. Hence, fix some input $(\samplerTM, \deciderTM)$ and integer $n$, we define $\cG_{n}^{\text{Intro}}$ as the following:
	
	The game $\cG_{n}^{\text{Intro}}$ has the sampling procedure for the $(\cG, n^{\alpha}, k)$-introspection game as pre given~\Cref{fig:Introspectiongraph} for any arbitrary game $\cG$. By~\Cref{thm:Introspectgame}, the input distribution is independent of the game $\cG$. For the decision process, given $(v_0, v_1) \in \decidertypedquestionpair^{\text{Intro}}_k$, $u_x, u_z \in \{0,1\}^{\alpha \lceil\log(n) \rceil}$, the question label that the verifier sends to the two provers. Let $a,b \in \{0,1\}^*$ be the answer that the verifier receives the answers based on the question label. The verifier computes the following: If at any point in the computation process, $|a|, |b| \geq 3 \cdot n^{\alpha}$ (since each input have at most 3 item of length at most $n^{\alpha}$), or the computation step for running $\samplerTM$, $\deciderTM$ either returns an invalid output or runs for time more than $n^{\alpha}$ steps, the verifier terminates and returns $0$ (i.e. the verifier rejects). The verifier first computes $(k_n, m_n , p_n) = \samplerTM(n, \text{parameter})$, and rejects if $k_n > k$ and $m_n \cdot p_n > n^{\alpha}$. 
	
	Based on the vertices $(v_0, v_1) \in \decidertypedquestionpair^{\text{Intro}}_k$, the verifier first divides the answer associate with each question label into the format given in~\Cref{fig:Introspectionsample}. Then the verifier does the following based on $(v_0, v_1)$:
	\begin{itemize}
		\item ($n^{\alpha}$-Pauli Basis): The verifier accepts according to the rules described in~\Cref{fig:PauliBasis}, this can be done uniformly in time $O(\poly(n))$ by~\Cref{thm:completenessPaulibasis}. 
		\item (Self-loop): The verifier accepts iff $a = b$; otherwise reject. This can be done in $O(n^{\alpha})$ time by the terminating assumption above. 
		\item (Pauli, $W$) \textbf{--}  (Gen Pauli, $W$): The verifier accepts iff $|b| = p_n \cdot m_n$ and the first $p_n \cdot m_n$ bits of $a$ are equal to $b$ (where recall, the elements $b \in \bF_{2^p}^n$ is represented using the canonical representation in this paper). 
		\item (Gen Pauli, $X$) \textbf{--} $(\text{Hide}_0, \decidertypedfunction^{P})$: The verifier first uses $\samplerTM(n, \text{Function}, \vec{1})$ to compute the canonical basis which spans $V_0^n$ (where $\vec{1} \in \{0,1\}^{m(n) \cdot p_n}$ is the all $1$ vector). For each $\hat{e}_j$, the canonical basis which spans $V_0$, compute $\samplerTM(n, \text{Function}, P, 0, 0, \hat{e}_i)$, and run the standard Gaussian elimination to find the description of the subspace $\ker{\decidertypedfunction^{P, n}_{0,0}}$ and $\ker{\decidertypedfunction^{P, n}_{0,0}}^{\bot}$. Finally, parse $s_{X} =  (s_{X})_{0} +(s_{X})_{0}^{\bot} + (s_{X})_{> 0}^n \in \ker{\decidertypedfunction^{P, n}_{0,0}} \oplus\ker{\decidertypedfunction^{P, n}_{0,0}}^{\bot} \oplus V_{> 0}$. The verifier accepts iff the answers from the provers are in the correct subspace according to~\Cref{fig:Introspectionverification} (i.e. $ t_{\leq 0}^{\bot} \in V_0^n$) and $(s_{X})_{0}^{\bot} = t_{\leq 0}^{\bot}$ and $(s_{X})_{> 0} =  r_{> 0}$. 
		\item $(\text{Hide}_{i}, \decidertypedfunction^{P})$ \textbf{--}  $(\text{Hide}_{i+1}, \decidertypedfunction^{P})$ for $i \in [k]$: If $i \geq k_n$, treat this as a consistency check. Otherwise using the same technique as above, compute the description for $V_{< i}^n$, $V_i^n$ and $V_{> i+1}$. Parse 
		\begin{align*}
			t_{\leq i+1 }^{\bot} &= 	\tilde{t}_{\leq i }^{\bot} + \tilde{t}_{ i +1}^{\bot} \in V_{\leq i}^n \oplus V_{i+1}^n, \quad t_{< i+1 }^{\text{Line}} = 	\tilde{t}_{< i}^{\text{Line}} + \tilde{t}_{ i}^{\text{Line}} \in V_{< i}^n \oplus V_{i}^n,
		\end{align*}
		and use a similar computation step to compute the description for $\ker{\left(\decidertypedfunction^{P, n}_{i+1,t_{< i +1}^{\text{Line}} }\right)}$ and $\ker{\left(\decidertypedfunction^{P, n}_{i+1,t_{< i +1}^{\text{Line}} }\right)}^{\bot}$. Then the verifier parse
		\begin{align*}
			r_{> i } &= 	\bar{r}_{i} + \bar{r}_{i}^{\bot} + \bar{r}_{>i+1} \in \ker{\left(\decidertypedfunction^{P, n}_{i+1,t_{< i +1}^{\text{Line}} }\right)} \oplus  \ker{\left(\decidertypedfunction^{P,n}_{i+1,t_{< i +1}^{\text{Line}} }\right)}^{\bot}  \oplus V_{> i+1}^n, 
		\end{align*}
		The verifier accepts iff the answers from the provers are in the correct subspace 
		\begin{align*}
			\tilde{t}_{\leq i }^{\bot} = t_{\leq i }^{\bot},\quad  \bar{r}_{i}^{\bot} = \tilde{t}_{i +1 }^{\bot}, \quad \tilde{t}_{< i }^{\text{Line}} = \tilde{t}_{< i }^{\text{Line}}, \quad \bar{r}_{>i+1} = r_{>i+1}.
		\end{align*} 
		\item $(\text{Hide}_{k-1}, \decidertypedfunction^{P})$ \textbf{--}   $(\text{Read},  \decidertypedfunction^{P})$: The verifier accepts iff the answers from the provers are in the correct subspace, $t_{\text{Read}, P} = t_{k-1}$, and $t_{\text{Read}, P}^{\bot} = t_{k-1}^{\bot}$.
		\item  $(\text{Read},  \decidertypedfunction^{P})$ \textbf{--}  $(\text{Intro},  \decidertypedfunction^{P})$: The verifier accepts iff $t_{\text{Read}, P}^{\text{Line}} = x_P$, and $a_{\text{Read}, P} = a_P$.
		\item (Gen Pauli, $Z$) \textbf{--} $(\text{Sample}, \decidertypedfunction^{P})$: The verifier accepts iff $s_Z = s_{\text{Sample}}$. 
		\item $(\text{Sample}, \decidertypedfunction^{P})$ \textbf{--} $(\text{Intro},  \decidertypedfunction^{P})$: The verifier computes $\decidertypedfunction^{P, n}(s_{\text{Sample}})$ by using the algorithm provided in~\Cref{sec:CLverifier} with $\samplerTM$. The verifier accepts iff the output provided by $\decidertypedfunction^{P}(s_{\text{Sample}})$ is equal to $x_P$ and  $a_{\text{Sample}, P} = a_P$.
		\item $(\text{Intro},  \decidertypedfunction^{A})$   \textbf{--}  $(\text{Intro},  \decidertypedfunction^{B})$: The verifier accepts iff $\deciderTM(n, x_A,x_B,a_A, a_B) = 1$. 
	\end{itemize}
	The above procedure uniformly defines a $(\cG_n, k, p)$-introspection game for $n \geq n_0$ assuming $(\samplerTM, \deciderTM)$ are valid Turing machines as given in the second part of~\Cref{prop:QuestionReduction}.  By~\Cref{lem:compfinitefield} and the hardcoded computation bound on $\samplerTM$ and $\deciderTM$, the above procedure can be computed in $O(\poly(n,k))$ time. 
	
	Since $\cG_{n}^{\text{Intro}}$ is typed samplable and $|\decidertypedquestionpair^{\text{Intro}}_k| = 30 + 2k$, by applying~\Cref{lem:CLsamplableLDT} to each $\cG_n^{\text{Intro}}$, we obtain a sequence of $(3, \alpha \lceil\log(n) \rceil + C^{\text{detype}} ,2)$ CL samplable games $\cG^{\QuestionRed}_n$, where $C$ is some constant that depends linearly on $k$. Since the detyping procedure given in~\Cref{def:detypeCLsampler} only adds extra string parsing and synchronization checks to the sampling and decision procedure. Each $\cG^{\QuestionRed}_n$ can still be sampled in $O(\poly(\log(n), k))$ time and verified in $O(\poly(n))$ time. Pick $\gamma^{\QuestionRed} \in \bN$ to be sufficiently large so that $\cG^{\QuestionRed}_n$ can be sampled in $O(\log^{\gamma^{\QuestionRed}}(n))$ time, $k \cdot \alpha \cdot \lceil\log(n) \rceil + C^{\text{detype}} = O(\log^{\gamma^{\QuestionRed}}(n))$ and $\cG^{\QuestionRed}_n$ can be decided in $O(\poly(n))$ time. 
	
	Define $(\samplerTM^{\QuestionRed}, \deciderTM^{\QuestionRed})$ in the following way: let $n_0^{\text{Run}}$ be the constant such that for all $n > n_0^{\text{Run}}$, $\cG^{\QuestionRed}_n$ can be sampled in fewer than $\log^{\gamma^{\QuestionRed}}(n)$ steps (no big-O notation here!) and decided in time fewer than $n^{\gamma^{\QuestionRed}}(n)$ steps, and furthermore $k \cdot \alpha \cdot \lceil\log(n) \rceil + C^{\text{detype}} \leq \log^{\gamma^{\QuestionRed}}(n)$. For all $n < n_0^{\text{Run}}$, $(\samplerTM^{\QuestionRed}(n), \deciderTM^{\QuestionRed}(n))$ returns an encoding of the rejecting game $\cG^{\text{reject}}$ defined in~\Cref{def:syncrejgame}, otherwise return an encoding for $\cG^{\QuestionRed}_n$. 
	
	We now verify all the properties listed in~\Cref{prop:QuestionReduction}.
		\begin{enumerate}
		\item (Computation time): This follows since the description of $(\samplerTM^{\QuestionRed}, \deciderTM^{\QuestionRed})$ only depends on the description of $(\samplerTM, \deciderTM)$ and some fixed constant. 
		\item (Synchronicity): This follows from the fact that $\cG^{\text{Intro}}_n$ is always synchronous. 
		\item (Complexity bounds for the output): Let $C^{\text{trivial}}$ be the constant such that $\cG^{\text{reject}}$ can be both sampled and decided in time $C^{\text{trivial}}$. This follows from the definition of $(\samplerTM^{\QuestionRed}, \deciderTM^{\QuestionRed})$. 
	\item (Independency) This follows since both  $\cG_{n}^{\text{Intro}}$ and the detyping procedure can be defined uniformly for all verifier sequences $(\samplerTM, \deciderTM)$, and the sampling procedure for $\cG_{n}^{\text{Intro}}$ does not depend on $\deciderTM$
\end{enumerate}
Let $\verifiersequence = (\samplerTM, \deciderTM)$ and $n_0 \in \bN$ be as pre described in the theorem. Let $n_0^{\QuestionRed} = \max\{ n_0^{\text{Run}}, n_0\}$, which both depend on $n^{\lambda}$, and $n_0^{\text{Run}}$ depends on the constant $k$. For all $n \geq n_0^{\QuestionRed}$:
\begin{enumerate}
\item (Completeness): This follows from~\Cref{thm:Introspectgame} and~\Cref{lem:detypeandsync}. 
\item (Soundness): By~\Cref{thm:Introspectgame}: There exists a polynomial $\mathbf{s}^{\text{Intro}}_{\alpha}$ such that 
\begin{equation*}
	\omega^*(\cG_n) \leq  1- \eps(n) \Longrightarrow \omega^{t}_s(\cG_n^{\text{Intro}}) \leq 1- \mathbf{s}^{\text{Intro}}_{\alpha}(k,  \eps(n))
\end{equation*}
Let $\mathbf{s}^{\QuestionRed}_{\alpha}(k,  \eps(n)) = \frac{\mathbf{s}^{\text{Intro}}_{\alpha}(k,  \eps(n))}{4(30+2k)^2 \cdot 16^{(30+2k)}}$. By~\Cref{lem:detypeandsync}
\begin{equation*}
	\omega^*(\cG_n) \leq  1- \mathbf{s}^{\text{Intro}}_{\alpha}(k) \Longrightarrow \omega^{t}_s(\cG_n^{\text{Intro}}) \leq  1- \mathbf{s}^{\QuestionRed}_{\alpha}(\exp(k),  \eps(n)).
\end{equation*}
\end{enumerate}
This concludes the proof for~\Cref{prop:QuestionReduction}.
\end{proof}


\section{Answer reduction} \label{sec:PCPanswerreduction}

In this section, we give a proof for~\Cref{prop:AnswerReduction}. The goal of the answer reduction protocol is to transform a synchronous CL verifier into another synchronous CL verifier with a more efficient verification complexity. We remark that the transformation used in this section is the same as the one given in~\cite[Section 10]{jiMIPRE2022a}. We give some intuition for the answer reduction transformation below.

Recall that from the previous section that, after applying the question reduction transformation, for the $n$th game of the CL verifier, the verifier only has to sample a logarithmic-size question pair $(x,y)$ in $O(\polylog(n))$ time. However, the verifier still has to receive polynomial-size answers $(a,b)$ from Alice and Bob, and then computes $\deciderTM(n,x,y,a,b)$ in $O(\poly(n))$ time to decide whether to accept the given instance. 

On a high level, the goal of the answer reduction protocol is to let the verifier delegate the task of computing $\deciderTM(n,x,y,a,b)$ to the provers. Of course, since the provers are by definition dishonest, the verifier cannot simply give this task to the provers. One important observation about an interactive protocol which makes the answer reduction protocol possible is that the verifier actually does not care how the computation step is being performed, \textit{he only cares whether $\deciderTM(n,x,y,a,b)$ outputs 1 at the end of the computation step}! Hence, the goal for the verifier is to design a protocol in which the provers can somehow output ``sufficient evidence" to show that they have, indeed, run the computation step of  $\deciderTM(n,x,y,a,b)$ honestly. 

Fortunately, the verifier can already use a probabilistically checkable proof (PCP), a common tool in the computer science literature~\cite{aroraProofVerificationHardness1998}. Roughly speaking, let $\texttt{TM}$ be a two-input Turing machine which runs in $O(\exp(n))$ time and $x \in \{0,1\}^*$ be a string with $|x| = O(\polylog(n))$. Existing PCP in the computer science literature allows a polylogarithmic-timed verifier, with the help of two (computationally unbounded) provers, to verify that there exists a string $a \in \{0,1\}^*$ with $|a| = O(\poly(n))$ such that $\texttt{TM}(x,a) = 1$. This construction can be easily modified to hold for a pair of strings $(x,y)$, both polylogarithmic-sized as the initial input, and a pair of polynomial-size strings $(a,b)$. We remark in this case, since $|a|$ is exponential in size, a polynomial-time verifier cannot process the entire string $a$ even if he receives it from the prover! However, there are several challenges with directly using a PCP construction within the answer reduction-procedure, which we list below: 
\begin{enumerate}
	\item For a prover to compute the given PCP instance, it needs both question labels $(x,y)$. This is a problem in the non-local game setting, since each prover is expected to receive only its own question label. 
	\item The PCP construction only checks whether there exists an answer pair $(a,b)$ which causes the verifier to accept. In this case, the verifier also needs to check that each answer within the answer pair $(a,b)$ depends only on its corresponding question label ($x$ or $y$); i.e. the prover cannot generate the answer $a$ based on both the question labels $(x,y)$. This is a bigger problem for the verifier than it might at first appear, since, as previously mentioned, the verifier does not have the runtime to even process the answer labels $a$ and $b$. 
	\item Lastly, the PCP construction must also be a CL sampleable game in order to be used as a part of the proof for the gap compression theorem. 
\end{enumerate} 

In order to address the first problem, before applying the PCP procedure, a transformation known as \emph{oracularization} is first applied. This transformation was first introduced in~\cite[Section 9]{jiMIPRE2022a}, and is also part of the answer reduction procedure in the ``gapless compression" introduced in~\cite{mousaviNonlocalGamesCompression2022}. The goal of this transformation is to give both provers the two question labels $(x,y)$ and force them to generate the same answer pair $(a,b)$ in such a way that the answer label $a$ (resp. $b$) only depends on the question label $x$ (resp. $y$). To ensure consistency, the provers will sometimes give one prover only one of the two question labels in order to perform a consistency check with the other prover who receives both question labels. 

Unfortunately, as pointed out by point two above, since the verifier cannot process the entire question label $a$ and $b$, the consistency check mentioned is also not as straightforward as it might have initially seemed. To keep the verification complexity low, the provers are expected to encode the answers $a$ and $b$ as a low-individual degree polynomial using the Generalized Reed-Muller code introduced in~\Cref{sec:RMcode}. In this case, the consistency test for the verifier becomes verifying that the two provers share the same low-individual degree polynomial, which can be done through the quantum low-individual degree test given in~\Cref{sec:lowinddegtest}. 

To tackle the third problem, ~\cite{natarajanNEEXPMIP2019a} uses a special type of PCP known as a probabilistically checkable proof of proximity (PCPP), which allows one to check whether a specific string $a$ satisfies $\texttt{TM}(x,a)=1$ (rather than merely asserting the existence of such a string using a standard PCP). In this paper, we use the tailor-made PCPP protocol constructed in~\cite[Section 10]{jiMIPRE2022a}, which reduces the proof checking task to an instance of the \emph{simultaneous quantum individual low-degree test}, where the simultaneous quantum individual low-degree test (SLDT), in essence, is a parallel repeated version of the quantum individual low-degree test, which is designed to test if the provers share multiple low-individual degree polynomials.  As shown later in this section, since the SLDT has the same sampling procedure as a regular quantum individual low-degree test, this solves the last problem listed above by~\Cref{lem:CLsamplableLDT}. 

We organize this section as follows. In~\Cref{sec:Oracularization}, we recall the oracularization transformation mentioned above from~\cite[Section 9]{jiMIPRE2022a} and show that the completeness/soundness properties from the tensor product model also hold for the commuting operator model. In~\Cref{sec:simuQILD}, we formally define the notion of a simultaneous quantum low-individual degree test, and show that a similar soundness  property also holds for the commuting operator model. In~\Cref{sec:Answerreductionproof}, we give a summary of result of the PCPP construction from~\cite[Section 10]{jiMIPRE2022a}, and state and prove the protocol that shows~\Cref{prop:AnswerReduction}.



\begin{figure}[!b]
	\centering
	\small
	\begin{gamespec}
		\setlength{\tabcolsep}{1em}
		\begin{tabularx}{\textwidth}{ l   l   X   }
			\toprule
			Question label & Question content & Answer format \\
			\midrule
			(Prover, A)&  $x \in \cX$& $a_A \in \cA$\\
			(Prover, B)&  $y \in \cX$& $b_B \in \cA$\\
			(Oracularization) & $(x,y) \in \cX^2$ & $(a,b) \in \cA$ \\
			\bottomrule
			\multicolumn{3}{c}{Figure: Q and A format for the oracularization transformation for $\cG = (\cX, \cA, \mu, D)$ }
		\end{tabularx}
		
		\begin{center}
			\textbf{Sampling procedure}
		\end{center}
		\begin{enumerate}
			\item Sample $(x,y) \sim \mu$, and $(n_0, n_1) \in \{ \text{(Prover, A)}, \text{(Prover, B)},  \text{	(Oracularization)} \}^2$. 
			\item Send the question label and question content corresponding to $n_0$ to one of the provers, and send the question content corresponding to the $n_1$ to the other prover.
		\end{enumerate}
		\begin{center}
			\textbf{Verification procedure}
		\end{center}
		\begin{enumerate}
			\item (Oracularization) \textbf{--} (Oracularization): The provers win iff they output the same answer and $D(x,y,a,b) = 1$. 
			\item (Prover, P) \textbf{--} (Prover, P) for $\text{P} \in \{\text{A},\text{B}\}$: The provers win iff they output the same answer. 
			\item (Prover, A) \textbf{--} (Oracularization): The provers win iff $D(x,y,a,b) = 1$ and $a = a_A$. 
			\item (Prover, B) \textbf{--} (Oracularization): The provers win iff $D(x,y,a,b) = 1$ and $a = b_A$. 
			\item (Prover, A) \textbf{--} (Prover, B): The provers win automatically. 
		\end{enumerate}
		On any other input, the prover win automatically. 
		\vspace{1em}
	\end{gamespec}
	\caption{ The description for the oracularization transformation \gls{zoracgame} for the game $\cG = (\cX, \cA,\mu, D)$. }
	\label{fig:Oracularization}
\end{figure}

\subsection{Oracularization} \label{sec:Oracularization}
\jqnote{Skeleton draft comes first}
In this subsection, we recall the oracularization transformation used in~\cite[Section 9]{jiMIPRE2022a}, and show that the appropriate completeness/soundness conditions also hold for the commuting operator model. Given a non-local game $\cG = (\cX, \cA, \mu , D)$, we define the corresponding oracularization transformation in~\Cref{fig:Oracularization}.

We have the following lemma regarding the oracularization transformation for the game $\cG$. We remark that the ``soundness" condition for the below lemma also preserves the normal (non-synchronous) value of the game. 

\jqnote{Remark about~\cite{mousaviNonlocalGamesCompression2022}, remember, the math works, but you cannot computationally perform that because it is hard to get trivial questions from $\samplerTM$. Update, I think keep this as optional, as a ``to investigate if you still have the sanity left" for the thesis.}

\begin{lemma}[Properties related to the oracularization transformation] \label{lem:oratransformation}
	Let $\cG = (\cX, \cA, \mu, D)$ be a non-local game, and let $\cG^{\text{Ora}}$ be the oracularization transformation for the game $\cG$, then for $t \in \{*, co\}$, the following holds
	\begin{itemize}
		\item (Completeness): If there exists a perfect oracularizable strategy for $\cG$ in model $t$, then there exists a perfect oracularizable strategy for $\cG^{\text{Ora}}$ in model $t$. 
		\item (Soundness): There exists a polynomial $\mathbf{P}^{\text{Ora}} (\eps)$ such that
		\begin{equation*}
			\omega^t(\cG^{\text{Ora}}) \geq 1 - \eps \Longrightarrow \omega^t(\cG) \geq 1 - \mathbf{P}^{\text{Ora}}(\eps). 
		\end{equation*}
		\item (Sample complexity) If $\cG$ is samplable via a $(k,m,p)$ CL distribution, then $\cG^{\text{Ora}}$ is samplable via a $(k+1, m+4,p)$ CL distribution. 
	\end{itemize}
\end{lemma}

\begin{proof}
Let $\cG$ and $\cG^{\text{Ora}}$ be the non-local game as specified in the lemma and fix $t \in \{*, co\}$. We also shorten the question label ``(Oracularization)" to ``(Ora)" in this proof (and in the remainder of this paper) for convenience. 

For the ``completeness" property in the lemma statement, let $\strategy$ be a perfect oracularizable strategy for the game $\cG$. By~\Cref{thm:syncvaluepreserve}, $\strategy$ is synchronous and hence by~\Cref{lem:structsync} can also be assumed to be a projective strategy defined by $\strategy = (\cL^2(\alicealg, \tau), \ket{\tau}, \{P_a^x \})$. 

We construct a perfect (synchronous) oracularizable strategy on the Hilbert space $\cL^2(\alicealg, \tau)$ as follows: for all $(x,y,a,b) \in \cX^2 \times \cA^2$, define the synchronous strategy $\strategy^{\text{Ora}} = (\cL^2(\alicealg, \tau), \ket{\tau}, \{M_a^x \})$ for the game $\cG^{\text{Ora}}$ as follows:
\begin{align*}
	M_{a_A}^{\text{(Prover, A)}, x} = P^x_{a_A}, \qquad 	M_{b_B}^{\text{(Prover, B)}, y} = P^y_{b_B}, \qquad 	M_{(a,b)}^{\text{(Ora)}, (x,y)} = P^x_a P^y_b. 
\end{align*}

We first show that $\strategy^{\text{Ora}}$ is projective. Since $\strategy$ is a projective strategy, both $M_{a_A}^{\text{(Prover, A)}, x}$ and $M_{b_B}^{\text{(Prover, B)}, y}$ are projective. By the definition of an oracularizable strategy~\Cref{def:Oracularizablestrategies}, for $(x,y)$ with $\mu(x,y) > 0$, $[P^y_b,  P^x_a] = 0$. Since the question label $\text{(Oracularization)}, (x,y)$ only occurs whenever $\mu(x,y) > 0$, this shows that the measurement operator $M_{(a,b)}^{\text{(Oracularization)}, (x,y)}$ is projective; hence $\strategy^{\text{Ora}}$ is a projective strategy. 

The fact that $[P^y_b,  P^x_a] = 0$ whenever $\mu(x,y) > 0$ also implies that the set of measurement operators $\{M_{a}^{\text{(Prover, A)}, x}\}_{a \in \cA} \cup  \{M_{b}^{\text{(Prover, B)}, y}\}_{b \in \cA} \cup \{ M_{(a,b)}^{\text{(Oracularization)}, (x,y)} \}_{(a,b) \in \cA^2} \subseteq \alicealg$ pairwise commute with each other whenever $\mu(x,y) > 0$. This shows that $\strategy^{\text{Ora}}$ is oracularizable.

To show that $\strategy^{\text{Ora}}$ is perfect, we verify each of the possible question pairs below:
\begin{itemize}
	\item (Same label): This follows because the strategy $\strategy^{\text{Ora}}$ is projective and uses the tracial state $\ket{\tau}$ as a part of the strategy. 
	\item (Prover, A) \textbf{--} (Ora): Given $(x,y) \in \cX$ with $\mu(x,y) > 0$, and $a, a_A, b \in \cA$, the probability of outputting $((a,b), a_A)$ given the question pair $((\text{(Prover, A)}, x))$ is
	\begin{align*}
		\braket{\tau|M_{(a,b)}^{\text{(Ora)}, (x,y)} (M_{a_A}^{\text{(Prover, A)}, x})^{op}|\tau} = \braket{\tau|( P^x_a P^y_b) (P^x_{a_A})^{op}|\tau} = \braket{\tau|(  P^y_b P^x_a) P^x_{a_A}|\tau}
	\end{align*}
	where the third equality follows from $A\ket{\tau} = A^{op} \ket{\tau}$. Since $\strategy$ is a projective, perfect strategy, the resulting answer $((a,b), a_A)$ must satisfy $D(x,y,a,b) = 1$ and $a = a_A$. 
	\item (Prover, B) \textbf{--} (Ora): This follows from the same argument as above. 
\end{itemize}
This shows the completeness clause in the lemma.

For the ``soundness" property in the lemma statement, suppose that $\omega^t(\cG^{\text{Ora}}) > 1 - \eps$. Let $\strategy' =  (\cL^2(\alicealg, \tau), \sigma \ket{\tau},\{ A_{a}^{v, x} \},\{  \left(B_{a}^{V, y}\right) \})$ be a tracially embeddable strategy for $\cG^{\text{Ora}}$ in model $t$ such that $\omega(\cG^{\text{Ora}}, \strategy) > 1 - \eps$. By the sampling procedure as specified in~\Cref{fig:Oracularization}, the game $\cG^{\text{Ora}}$ is $\frac{1}{3}$-balanced. Hence, by~\Cref{lem:balancedperfecttoalmostsync}, there exist a symmetric strategy $\strategy^{\text{sym}} = ( \cL^2(\alicealg, \tau), \sigma \ket{\tau},\{ P_{a}^{V, y} \})$ such that 
\begin{equation} \label{eq:oracproofeq1}
	\omega(\cG^{\text{Ora}}, \strategy^{\text{sym}}) > 1 - \eps - \left(3 \eps \right)^{\frac{1}{4}} = 1- O(\poly(\eps)). 
\end{equation} 
We wish to argue that the strategy $\strategy = ( \cL^2(\alicealg, \tau), \sigma \ket{\tau},\{ P^{\text{(Prover, A)},x}_a \},\{  P^{\text{(Prover, B)},y}_b \} )$ for the game $\cG$ satisfies $\omega(\cG, \strategy) > 1 - O(\poly(\eps))$. For simplicity of notation, we write $P^{\text{(Prover, Q)},x}_a$ as $P^{\text{Q},x}_a$ for $\text{Q} \in \{\text{A},\text{B}\}$. Since the question label pair (Ora) \text{--} (Prover A) and  (Ora) \text{--} (Prover B) are selected with probability $1/9$,  we have 
\begin{align*}
	P^{\text{(Ora)},(x,y)}_{a,b} \simeq_{O(\poly(\eps))} ( P^{\text{A},x}_a)^{op}, \qquad P^{\text{(Ora)},(x,y)}_{a,b} \simeq_{O(\poly(\eps))} ( P^{\text{B},y}_b)^{op}, 
\end{align*}
over the distribution $(x,y) \sim \mu$. Since $P$ is projective, by~\Cref{lem:closetodistance}
\begin{equation}
	P^{\text{(Ora)},(x,y)}_{a,b} \approx_{O(\poly(\eps))} ( P^{\text{A},x}_a)^{op}, \qquad P^{\text{(Ora)},(x,y)}_{a,b} \approx_{O(\poly(\eps))} ( P^{\text{B},y}_b)^{op}.  \label{eq:oracproofeq2}
\end{equation}
Since the question label pair (Prover A) \text{--} (Prover A) is also selected with probability $1/9$, 
\begin{equation} 
	P^{\text{A},x}_a \approx_{O(\poly(\eps))}( P^{\text{A},x}_a)^{op}.  \label{eq:oracproofeq3}
\end{equation}
Since $\{ P \}$ is a projective strategy, $P \leq \cI$. Hence
\begin{align}
	\notag P^{\text{(Ora)},(x,y)}_{a,b} = (P^{\text{(Ora)},(x,y)}_{a,b})^2  &\approx_{O(\poly(\eps))} 	P^{\text{(Ora)},(x,y)}_{a,b} ( P^{\text{B},y}_b)^{op}\\
	\notag	&\approx_{O(\poly(\eps))} ( P^{\text{A},x}_a)^{op} ( P^{\text{B},y}_b)^{op}, \\
	&\approx_{O(\poly(\eps))} P^{\text{A},x}_a( P^{\text{B},y}_b)^{op}, \label{eq:oracproofeq4}
\end{align}
where the first approximation follows from~\Cref{lem:combmeasure} with $C$ being $\{ P^{\text{(Ora)},(x,y)}_{a,b} \}$, and treating the set $\cC$ as the singleton set, and the second and third approximation follows from~\Cref{lem:combmeasure} with $C$ being $\{ ( P^{\text{B},y}_b)^{op}\}$ in conjunction with \eqref{eq:oracproofeq2} and \eqref{eq:oracproofeq3} respectively. Since $P$ is projective and $P^{\text{A},x}_a$ commutes with $(P^{\text{B},y}_b)^{op}$, we have
\begin{equation*}
	P^{\text{A},x}_a( P^{\text{B},y}_b)^{op}= \left(P^{\text{A},x}_a( P^{\text{B},y}_b \right)^{op})^2.
\end{equation*}
Hence
\allowdisplaybreaks{
	\begin{align}
		\notag&\Ex_{(x,y) \sim \mu} \sum_{a,b} |\braket{\tau| \sigma P^{\text{(Ora)},(x,y)}_{a,b} \sigma |\tau} - \braket{\tau|\sigma \left(P^{\text{A},x}_a( P^{\text{B},y}_b)^{op}\right)  \sigma|\tau} | \\
		\notag&\leq	\Ex_{(x,y) \sim \mu} \sum_{a,b} |\braket{\tau| \sigma P^{\text{(Ora)},(x,y)}_{a,b} \left( P^{\text{(Ora)},(x,y)}_{a,b} - P^{\text{A},x}_a( P^{\text{B},y}_b)^{op}\right)\sigma|\tau} | + \\
		\notag&\qquad \Ex_{(x,y) \sim \mu} \sum_{a,b} |\braket{\tau| \sigma  \left( P^{\text{(Ora)},(x,y)}_{a,b} - P^{\text{A},x}_a( P^{\text{B},y}_b)^{op}\right) P^{\text{A},x}_a( P^{\text{B},y}_b)^{op}\sigma|\tau} | \\
		\notag&\leq\sqrt{\Ex_{(x,y) \sim \mu}\sum_{a,b} \braket{\tau| \sigma P^{\text{(Ora)},(x,y)}_{a,b} \sigma|\tau} } \sqrt{\Ex_{(x,y) \sim \mu}\sum_{a,b} |\braket{\tau| \sigma \left( P^{\text{(Ora)},(x,y)}_{a,b} - P^{\text{A},x}_a( P^{\text{B},y}_b)^{op}\right)^2 \sigma|\tau} | }  + \\
		\notag	&\qquad \sqrt{\Ex_{(x,y) \sim \mu}\sum_{a,b} |\braket{\tau| \sigma \left( P^{\text{(Ora)},(x,y)}_{a,b} - P^{\text{A},x}_a( P^{\text{B},y}_b)^{op}\right)^2 \sigma|\tau} | } \cdot \sqrt{\Ex_{(x,y) \sim \mu}\sum_{a,b} \braket{\tau|\sigma  P^{\text{A},x}_a( P^{\text{B},y}_b)^{op} \sigma|\tau} }\\
		\label{eq:oracproofeq6}&= 2  \sqrt{\Ex_{(x,y) \sim \mu} \sum_{a,b}| \left( P^{\text{(Ora)},(x,y)}_{a,b} -  P^{\text{A},x}_a( P^{\text{B},y}_b)^{op} \right) \sigma \ket{\tau}|^2} = O(\poly(\eps)),
	\end{align}
}
where the second line follows from the triangle inequality and the third line follows from Cauchy-Schwartz. The last line follows from $\{ P^{\text{(Ora)},(x,y)}_{a,b}\}$, $\{ P^{\text{A},x}_a\}$ and $\{( P^{\text{B},y}_b)^{op}\}$ being three sets of POVMs and~\eqref{eq:oracproofeq4}. 

Finally, since the question label (Ora) is selected with probability $1/3$ by the first prover, by~\eqref{eq:oracproofeq1}, with expectation over the distribution $\mu$, the measurement $\braket{\tau|\sigma \sigma P^{\text{(Ora)},(x,y)}_{a,b} \sigma \tau}$ will produce an answer $(a,b)$ such that $D(x,y,a,b)$ with probability at least $1 - O(\poly(\eps))$ (or else the players will lose $\cG^{\text{Ora}}$). This, in conjunction with~\eqref{eq:oracproofeq6}, shows the ``soundness" clause in the lemma.

For the ``sample complexity" property in the lemma statement, assume $\mu$, the input distribution for $\cG$, is samplable via a $(k,m,p)$ CL distribution. Let $\decidertypedfunction^A$ and $\decidertypedfunction^B$ be the two $(k,m,p)$ conditional linear functions used to define $\mu$. We show that $\cG^{\text{Ora}}$ is $(k+1, m+4,p)$ CL samplable by using the series composition of two sets of CL functions (given in~\Cref{defn:seriescomposCLfn}). Let $V_0 = \bF_{2^p}^{4}$ and write $V$, and $V_{>0} = \bF_{2^p}^{m}$, we define the two $(k+1, m+4,p)$ conditional linear functions $\decidertypedfunction^{A, \text{Ora}}$ and $\decidertypedfunction^{A, \text{Ora}}$ as follows:

For simplicity of notation, we assume elements in $V_0$ are represented under the canonical representation (i.e. as elements of $\{0,1\}^{4 \cdot p}$). Write every elements $s \in V_0$ as $s = (s_0, s_1, s_2, s_3, s_4)$, where $s_i \in \{0,1\}$ for $i \in [4]$ and $s_4 \in \{0,1\}^{4 \cdot (p-1)}$. Define 

\begin{equation*}
	\decidertypedfunction^{A, \text{Ora}}_{0,0}(s_0, s_1, s_2, s_3, s_4) = (s_0, s_1, \vec{0}),  \, \quad \,	\decidertypedfunction^{B, \text{Ora}}_{0,0}(s_0, s_1, s_2, s_3, s_4) = (s_2, s_3, \vec{0}),
\end{equation*}
where $\vec{0}$ is the all $0$ string in $\{0,1\}^{4 \cdot p -2}$, we define
\begin{align*}
		\decidertypedfunction^{A, \text{Ora}}_{>0, (0,0, \vec{0})} = \decidertypedfunction^{B, \text{Ora}}_{>0, (0,0, \vec{0})} = 	\decidertypedfunction^{A} \, \quad \,	\decidertypedfunction^{A, \text{Ora}}_{>0, (0,1, \vec{0})} = \decidertypedfunction^{B, \text{Ora}}_{>0, (0,1, \vec{0})} = 	\decidertypedfunction^{B} \, \quad \, 	\decidertypedfunction^{A, \text{Ora}}_{>0, (1,0, \vec{0})} = \decidertypedfunction^{B, \text{Ora}}_{>0, (1,0, \vec{0})} = \Id,
\end{align*}
where $\Id$ is the identity function on $V_{>0}$. Intuitively, $(0,0)$ label corresponds to the question label ``Prover A"; the $(0,1)$ label corresponds to ``Prover B"; and the $(1,0)$ label corresponds to ``Oracularization". If the prover receives the label  ``Oracularization", we give the entire seed to that prover in order for them to compute the question label $(x,y)$ themselves. We remark that we can treat the label $(1,1)$ as a free win for the provers, which occurs only with a constant probability; this raises the soundness condition by only a constant factor. This completes the argument for the ``sample complexity" claim. 
\end{proof}

We remark that throughout the paper, the ``completeness" condition for almost all transformations of games always has a requirement that preserves perfect \emph{oracularizable} strategies. The only use for this requirement in this paper is to show the ``completeness" condition for the oracularization transformation to hold. 
 
\subsection{The simultaneous quantum low-individual degree test} \label{sec:simuQILD}

\begin{figure}[!b]
	\centering
	\begin{gamespec}
		\setlength{\tabcolsep}{1em}
		\footnotesize
		\begin{tabularx}{\textwidth}{ l   l   X   }
			\toprule
			Question label & Question content & Answer format \\
			\midrule
			(Point)&  $s \in \bF_{2^p}^m$& $(a_0, \cdots, a_{k-1}) \in  \bF_{2^p}^k$\\
			(Dline)&  $(j, s_{\text{Dline}}) \in [m] \times \bF_{2^p}^m$& $k$ degree $d$ polynomial, $ \textbf{f}_i: \bF_{2^p} \rightarrow  \bF_{2^p}$, $i \in [k]$, where each $\textbf{f}_i$ is encoded as $\bF_{2^p}^d$ \\
			(Aline)&  $(j,v, s_{\text{Aline}}) \in [m] \times \bF_{2^p}^{2m}$ &  $k$ degree $dm$ polynomial, $\textbf{g}_i: \bF_{2^p} \rightarrow  \bF_{2^p}$, $i \in [k]$, where each $\textbf{g}_i$ is encoded as $\bF_{2^p}^{dm}$ \\
			\bottomrule
			\multicolumn{3}{c}{Figure: Q and A format for the $(p,m,d,k)$-simultaneous quantum low-individual degree test.}
		\end{tabularx}
		\normalsize
		\begin{center}
			\textbf{Sampling procedure}
		\end{center}
		The sampling procedure for the $(p,m,d,k)$-simultaneous quantum low-individual degree test is the same as the sampling procedure for the$(p,m,d)$ quantum low-individual degree test given in the proof for~\Cref{lem:CLsamplableLDT}
		\begin{center}
			\textbf{Verification procedure}
		\end{center}
		\begin{enumerate}
			\item (Same question label): The provers win iff they output the same answer. 
			\item (Point) \textbf{--} (DLine): The provers win iff $\textbf{f}_i(s) = a_i$ for all $i \in [k]$. 
			\item (Point) \textbf{--} (ALine): The provers win iff $\textbf{g}_i(s) = a_i$ for all $i \in [k]$. 
		\end{enumerate}
		On any other input, the prover wins automatically. 
		\vspace{1em}
	\end{gamespec}
	\caption{The description for the $(p,m,d,k)$-simultaneous quantum low-degree test. }
	\label{fig:simuLDT}
\end{figure}

We describe the simultaneous quantum low-individual degree test~\cite[Figure 3]{jiMIPRE2022a} below, which as we will see in the next section, is the key subroutine for the PCPP protocol. Intuitively, the $(p,m,d,k)$-simultaneous quantum low-individual degree test is a generalization of the $(p,m,d)$ quantum low-individual degree test defined in~\Cref{sec:lowinddegtest}, where the goal is to test whether the two provers agree on $k$ global $m$-variant low-individual degree polynomial $\textbf{g}: \bF_q^m \rightarrow \bF_q$ with individual degree of at most $d$. We define the $(p,m,d,k)$-simultaneous quantum low-individual degree test in~\Cref{fig:simuLDT}. We have the following lemma regarding the $(p,m,d,k)$-simultaneous quantum low-individual degree test.

\begin{lemma}[Properties of the $(p,m,d,k)$-simultaneous quantum low-individual degree test]  \label{lem:simuQLID}
	Let $p,m,d,k \in \bN$, and let $\cG^{\text{SLD}} = (\cX^{\text{SLD}}, \cA^{\text{SLD}}, \mu^{\text{SLD}}, D^{\text{SLD}})$ be the $(p,m,d,k)$-simultaneous quantum low-individual degree test specified in~\Cref{fig:simuLDT}, then the following holds:
	\begin{itemize}
		\item (Sample complexity): $\cG^{\text{SLD}}$ is samplable via a $(5, 9+ m' + 2\cdot m, p)$ typed CL distribution, where $m' = \ceil{\frac{\log(m)}{p}}$. 
		\item (Verification complexity): There exists a polynomial time Turing machine $\deciderTM^{\text{SLD}}$ which implements $D^{\text{SLD}}$ and runs in $O(\poly(p,m,d,k))$ time.
		\item (Soundness): There exists a universal constant $1 \geq c_{\text{SLD},1}$ and $0 < c_{\text{SLD},2} \leq 1$  and a function 
		\begin{equation*}
			\text{\gls{SoundnessPLDT}} =  c_{\text{SLD},1} (kdm)^{c_{\text{SLD},1}} (\eps^{c_{\text{SLD},2}} + 2^{-c_{\text{SLD},2}p} + 2^{-c_{\text{SLD},2}md})
		\end{equation*}
		such that the following holds. Let $\strategy =( \cL^2(\alicealg, \tau), \sigma \ket{\tau},\{ A_{a}^x \} )$ be a synchronous strategy for $\cG^{\text{SLD}}$ which succeed with probability $1- \eps$. There exist a set of PVM $\{G_{(\textbf{g}_0, \cdots \textbf{g}_{k-1})}\} \subseteq \alicealg'$ with outcome labelled by $k$ (potentially the same) $\textbf{g}_i\in \Idpoly(p,m,d)$ such that 
		\begin{align*} 
			&\Ex_{s \sim \bF_{q}^m} \sum_{\textbf{g}_0, \cdots \textbf{g}_{k-1} \in \Idpoly(p,m,d)} \braket{\tau| A^{(\text{point}, s)}_{(\textbf{g}_0(s), \cdots, \textbf{g}_{k-1}(s))} G_{(\textbf{g}_0, \cdots \textbf{g}_{k-1})} |\tau} \geq 1 - \mathbf{\eta}_{\text{SLD}}(p,m,d,k, \eps). 
		\end{align*}
	\end{itemize}
\end{lemma}

Although the simultaneous quantum low-degree test is a more sophisticated version of the quantum low-individual degree test. Its sampling procedure is exactly the same as the quantum low-individual degree test, and its verification procedure is essentially repeating the verification procedure for the quantum low-individual degree test $k$ times. Hence, the ``sample complexity" and the ``decision complexity" follows trivially from~\Cref{lem:CLsamplableLDT} and~\Cref{lem:compfinitefield} respectively. 

To show the ``soundness" clause for the above lemma, we recall the following condition from~\cite{natarajanNEEXPMIP2019a}.
\begin{definition}[Exactly linear functions, Definition 3.17 of~\cite{natarajanNEEXPMIP2019a}]
	Let $m, p \geq 0$. A function $\textbf{f}: \bF_{2^p}^m \times \bF_{2^p}^t \rightarrow \bF_{2^p}$ is exactly linear in $y$ if it can be written as 
	\begin{equation*}
		 \textbf{f}(x,y) = y_1 \cdot \textbf{f}_1(x)   +  y_2 \cdot \textbf{f}_2(x)   \cdots y_{k-1} \cdot \textbf{f}_{k-1}(x), 
	\end{equation*} 
	for a set of functions $\{ \textbf{f}_i \}_{i \in [k]}$, and we call a function $\textbf{g}: \bF_{2^p}^k \rightarrow \bF_{2^p}$ to be exactly linear if 
	\begin{equation*}
		\textbf{g}(x) = \sum_{i \in [k]} c_i \cdot x_i, 
	\end{equation*} 
	for some $c_i \in \bF_{2^p}$, $i \in [k]$. 
\end{definition}
We recall the following proposition about almost exactly linear individual degree $d$ polynomials.  
\begin{proposition}[Almost exactly linear individual degree polynomials, Proposition 3.18 of~\cite{natarajanNEEXPMIP2019a}] \label{prop:almostexactlinear}
	Suppose that $\textbf{f }(x,y):  \bF_{2^p}^m \times \bF_{2^p}^k \rightarrow \bF_{2^p}$ is a polynomial with individual degree of at most $d$ which is not exactly linear in $y$. Then the probability that, given a uniformly random $z \sim \bF_{2^p}^m$, the probability that the polynomial $\textbf{f}_{z}(y): \bF_{2^p}^k \rightarrow \bF_{2^p}$, $\textbf{f}_{z}(y) = \textbf{f}(z,y)$, being exactly linear is at most $ \frac{md}{2^p}$. 
\end{proposition}
We are now ready to give a proof for the ``soundness clause" for~\Cref{lem:simuQLID}, we remark that the proof below is a slightly modified version of~\cite[Theorem 4.43]{natarajanNEEXPMIP2019a} to account for the difference between the quantum low-degree test and the quantum low-individual degree test. 
\newcommand{\combinefun}{\text{combine}}
\begin{proof}

	Fix constant $p, m ,d, k \in \bN$. Let $\cG^{\text{SLD}}$ be the $(p,m,d,k)$-simultaneous quantum low-individual degree test, and let $\strategy = (\cL^2(\alicealg, \tau), \ket{\tau}, \{A_a^x \})$ be a tracially embeddable strategy for $\cG^{\text{SLD}}$ with $\omega(\cG^{\text{SLD}}, \strategy) \geq 1 - \eps$. We wish to show that $\strategy$ can be used as a part of a strategy for the $(p,m+k,d)$ quantum low-individual degree test. For simplicity of notation, for a set of functions $\{\textbf{f}_i: \bF_{2^p}^m \rightarrow \bF_{2^p}\}_{i \in [k]}$,  we write $\combinefun_{\textbf{f}}(x,y): \bF_{2^p}^m \times \bF_{2^p}^k \rightarrow  \bF_{2^p}$ as the function 
	\begin{equation}
		\combinefun_{\textbf{f}} (x,y) = x_0 \textbf{f}_0(y) + \cdots  x_{k-1} \textbf{f}_{k-1}(y).
	\end{equation}
	In the definition above $\textbf{f}_i$ could be potentially set as a constant (i.e. $\textbf{f}_i (x) = c_i$ for some $c_i \in \bF_{2_p}$), and for $a \in \bF_{2^p}^k$, we write $\combinefun_{a} (y): \bF_{2^p}^k \rightarrow  \bF_{2^p}$ as the function $\combinefun_{a} (y) = \sum_{i \in [k]} a_i \cdot y_i$. For $o \in [k]$ define $\vec{1}_0 \in \bF_{2^p}^o$ to be the vector with all coordinates being $1 \in  \bF_{2^p}$. We define a strategy for an instance of $(p,m+k,d)$-quantum low-individual degree test depending on $\strategy$ as follows: We specify the provers' behaviour based on the question label below
	\begin{enumerate}
		\item (Point) Given the question content $(s^1,s^2) \in \bF_{2^p}^m \times \bF_{2^p}^k$, the prover first performs the strategy $\strategy$ using the question label ``(Point)" and the question content $s^1$ and obtain points $ a = (a_0, \cdots, a_{k-1})$. The prover then returns $\combinefun_{a} (s^2)$ as the answer. 
		\item (DLine) Given the question content $(j,s_{\text{Dline}})$, write $j = (j^1, j^2) \in  \bF_{2^p}^m \times \bF_{2^p}^k$, and let $l$ be the axis on which the axis-parallel line is defined. The prover does the following depending on $l$
		\begin{itemize}
			\item If $l \in [m]$, then the prover performs the strategy $\strategy$ using the question label ``(DLine)" and the question content $(j^1,s_{\text{Dline}})$ and obtain $k$ degree-$d$ polynomials $(\textbf{f}_0, \cdots, (\textbf{f}_{k-1})$. The prover then returns the degree $d$ polynomial $\combinefun_{\textbf{f}} (j^2)$ as the answer. 
			\item  Otherwise, the prover first performs the strategy $\strategy$ using the question label ``(Point)" and the question content $j^1$ and obtain points $ a = (a_0, \cdots, a_{k-1})$. The prover then returns the degree $1$ polynomial $\combinefun_{a} (j^2_0, \cdots j^2_{l-m-1}, x,  j^2_{l-m+1}, j^2_{m+k -1} )$ as their answer. 
		\end{itemize}
		\item (ALine)  Given the question content $(j,s_{\text{Aline}})$, write $j = (j^1, j^2) \in  \bF_{2^p}^m \times \bF_{2^p}^k$, and let $l$ be the coordinates in which the diagonal line is defined. The prover does the following depending on $j$
		\begin{itemize}
			\item If $l \in [m]$, then the prover performs the strategy $\strategy$ using the question label ``(ALine)" and the question content $(j^1,s_{\text{Aline}})$ and obtain $k$ degree $d \cdot m$ polynomials $(\textbf{g}_0, \cdots, (\textbf{g}_{k-1})$. The prover then returns the degree $d \cdot m+1$ polynomial $\combinefun_{\textbf{g}} (x \cdot \vec{1}_k))$ as the answer. 
			\item  Otherwise, the prover performs the strategy $\strategy$ using the question label ``(Point)" and the question content $j^1$ and obtain points $ a = (a_0, \cdots, a_{k-1})$. The prover then returns the degree $1$ polynomial $\combinefun_{a} (j^2_0, \cdots, j^2_{l-m-1}, x \cdot \vec{1}_{m+k - l} )$ as their answer. 
		\end{itemize}
	\end{enumerate}
	Since $\strategy$ succeeds in $\cG^{\text{SLD}}$ with probability at least $1 - \poly(\eps)$. This implies that, given the same question label and content, the probability that the provers give the same answer is at least $1 - \poly(\eps)$, as well as given the label pair ``(point)" and ``(DLine)" (resp. ``(ALine)"), the answer pair satisfies $\textbf{f}_i = a_i$ (resp.  $\textbf{g}_i = a_i$). Using this, one can see that the above strategy for the $(p,m+k,d)$-quantum low-individual degree test succeeds with probability at least  $1 - \poly(\eps)$. For $(s^1,s^2) \in \bF_{2^p}^m \times \bF_{2^p}^k$ and $\nu \in \bF_{2^p}$, we define the measurement $A^{(\text{point}, s^1)}_{[\combinefun_{a}(s^2)|\nu]}$ to be the data processing measurement in which the prover applies the function $\textbf{f}(a_0, \cdots, a_{k-1}) = \combinefun_{a}(s^2)$ to the measurement outcome of $P^{(\text{point}, s^1)}_{a_0, \cdots, a_{k-1}}$. By~\Cref{thm:soundnessQLDT}, there exists a measurement $\{H_{\textbf{f}}\} \subseteq \alicealg'$ with outcomes $\textbf{f}$ being $m+k$-variate polynomial with individual degree of at most $d$ such that
	\begin{equation} \label{eq:sldt-1}
		\Ex_{(s^1, s^2) \sim \bF_{2^p}^m \times \bF_{2^p}^k} \sum_{\textbf{f} \in \text{IdPoly}(p,m+1,d)} \braket{\tau| A^{(\text{point}, s^1)}_{[\combinefun_{a}(s^2)| \textbf{f}(s^1, s^2)]} \cdot  H_{\textbf{f}} \,  |\tau} \geq 1 - \mathbf{\eta}_{\text{LD}}(p,m+k,d,\eps),
	\end{equation}
	Now we wish to show that the measurement outcome $\textbf{f}(x,y)$ from $\{H_{\textbf{f}}\}$ is exactly linear in $y$ with high probability. Fix a $\textit{g}$ such that the measurement outcome from $\textbf{f}$ is not exactly linear. For a fixed $s^{1} \in \bF_{2^p}^m$ and a fixed measurement outcome $a = (a_0, \cdots, a_{k-1})$ from $P^{(\text{point}, s^1)}_{(a_0, \cdots, a_{k-1})}$, by construction, the function $\combinefun_{a}(y):\bF_{2^p}^k \rightarrow \bF_{2^p}$ is exactly linear. However, by~\Cref{prop:almostexactlinear} the probability that $\textbf{f}_{s^{1}}(y) = \textbf{f}(s^{1},y)$ is exactly linear is at most $\frac{md}{2^p}$. As a result, for a uniformly chosen $s_1 \sim \bF_{2^p}$, the probability that $\combinefun_{a} = \textbf{f}_{s^{1}}$ is at most $\frac{md}{2^p}$. Hence by~\Cref{lem:Schwartz_Zipple}, for $(s^1, s^2) \sim \bF_{2^p}^m \times \bF_{2^p}^k$, the probability that $\combinefun_{a}(s^2) = \textbf{f}(s^1, s^2)$ is at most $\frac{md}{2^p}\cdot \frac{kd}{2^p}$ regardless of the measurement outcome for $a$. Combining the above fact with~\Cref{eq:sldt-1}, we see that the output of $\{H_{\textbf{f}}\}$ is exactly linear with probability at least $1 - \mathbf{\eta}_{\text{LD}}(p,m+k,d,\eps) - \left(\frac{md}{2^p}\right)^2 - \frac{m \cdot k d}{2^{2p}}$. 
	
	Define the measurement $\{G_{\textbf{g}_0, \cdots, \textbf{g}_{k-1}} \} \subseteq \alicealg'$ with the outcome set the same as $\{H\}$ as follows: The prover first measures according to $\{H\}$ to receive a polynomial $\textbf{f}(x,y)$. If $\textbf{f}(x,y)$ is exactly linear in $y$, it can be written as $\textbf{f}(x,y) = \sum_{k} y_i \textbf{g}_i(x)$, such that each $\textbf{g}_i$ is an $m$-variant polynomial with individual degree at most $d$. In this case, $G$ outputs the polynomials $\{\textbf{g}_i\}_{i \in [k]}$ as the output. If the measurement output from $\{H\}$ is not exactly linear, $G$ simply outputs $k$ random $m$-variate polynomials with individual degree $d$,  $\{\textbf{g}_i\}_{i \in [k]}$. Since $\combinefun_{\textbf{g}}$ is equal to $\textbf{g}$ whenever $\textbf{g}$ is exactly linear by definition, by replacing $H$ with $G$ on~\Cref{eq:sldt-1}, we see that 
	\begin{equation} \label{eq:sldt-2}
	 A^{(\text{point}, s^1)}_{[\combinefun_{a}(s^2)| \textbf{f}(s^1, s^2)]} \simeq_{\mathbf{\eta}_{\text{LD}}(p,m+k,d,\eps) +  \frac{m \cdot k d}{2^{2p}}}  G_{\combinefun_{\textbf{g}} =  \textbf{f}},
	\end{equation}
	where $\simeq$ is with respect to the distribution $\Ex_{(s^1, s^2) \sim \bF_{2^p}^m \times \bF_{2^p}^k}$ and the state $\ket{\tau}$. Now, for any fixed $s^1$ and $a = (a_i)_{i \in [k]}$, if $\textbf{g}_i(s^1) \neq a_i$ for any $i$, then the polynomial $\combinefun_{\textbf{g}}(s^1,y)$ and $\combinefun_{a}(y)$ are not equal and hence again by~\Cref{lem:Schwartz_Zipple}, the probability is at most $\frac{1}{2^p}$ (since both $\combinefun_{\textbf{g}}(s^1,y)$ and $\combinefun_{a}(y)$ are a multi-linear function). This implies that 
	\begin{equation*}
		A^{(\text{point}, s^1)}_{a_0, \cdots, a_{k-1}} \simeq_{\mathbf{\eta}_{\text{LD}}(p,m+k,d,\eps) +  \frac{m \cdot k d}{2^{2p}} + \frac{1}{2^p} }  G_{\combinefun_{\textbf{g}} =  \textbf{f}}.
	\end{equation*}
	Hence, the ``soundness" condition follows by setting 
	\begin{equation*}
		\mathbf{\eta}_{\text{SLD}}(p,m,d,k, \eps) = \eta_{\text{LD}}(p,m+k,d,\eps) +  \frac{m \cdot k d}{2^{2p}} + \frac{1}{2^p}.
	\end{equation*}
\end{proof}
We remark that since the above proof does not rely on the strategy $\strategy$ being synchronous, and thus the ``commuting operator soundness" condition in~\Cref{lem:simuQLID} also holds for general strategy (by replacing~\Cref{thm:soundnessQLDT} above with~\cite[Corollary 4.4]{linTracialEmbeddableStrategies2024}).


\subsection{Time bounded classical PCPP}


We recall the following PCPP protocol given in~\cite[Section 10]{jiMIPRE2022a} (which is a modification of the work by~\cite{harshaRobustPCPsProximity2004}). Since we do not modify the construction from~\cite{jiMIPRE2022a}, we do not go through the details of this construction in this paper, and instead refer to the original reference for more details. Recall from the preliminary that given a string $a \in \{0,1\}^n$, one can encode $a$ into a low-individual-degree $\log(n)$-variate polynomial with a low-individual degree of $2$ using the generalize Reed-Muller encoding given in~\eqref{eq:REM}. The following theorem is the main result of~\cite{jiMIPRE2022a} section 10.1-10.5.

\begin{theorem}[Time bounded decider PCPP] \label{thm:classicalPCPP}
Let $\deciderTM$ be a decider for a CL verifier $\verifiersequence$, and  $\alpha \in \bN$. There exist two Turing machines $(\mathtt{PCPParameter}_{\alpha}, \mathtt{ComputePCP}_{\alpha})$. $\mathtt{PCPParameter}_{\alpha}$ takes, as input $n \in \bN$ and outputs a tuple of parameters $(m^{\text{ans}},m^{\text{PCPP}}, g, p)$ such that the following holds:
\jqnote{Since x,y has padded length up to $\log^{\alpha}(n)$, the parameters are actually independent of $(x,y)$}
\begin{itemize}
	\item $\TIME_{\mathtt{PCPParameter}_{\alpha}} = O(\poly(\alpha, \log(n)))$
	\item $m^{\text{ans}}, g = O(\poly(\alpha, \log(n)))$ 
	\item $m^{\text{PCPP}} = 5 \cdot m^{\text{ans}} + 5 + g$, where g is padded such that $m^{\text{PCPP}}$ is of the form $2^i$ for some integer $i$. 
	\item Let $1 \geq c_{\text{SLD},1}$ and $0 < c_{\text{SLD},2} \leq 1$ be the universal constant defined within~\Cref{lem:simuQLID}. The field size $p$ is chosen to be the smallest integer which satisfies the following:
	\begin{enumerate}
		\item $p \geq \frac{\alpha c_{\text{SLD},2} + 3 c_{\text{SLD},1} ) \cdot \log(g)}{c_{\text{SLD},2}}$,
		\item $\frac{(2 + 5p) \cdot m^{\text{ans}}}{2^p} < \frac{1}{2}$,
		\item $\frac{p m^{\text{PCPP}}}{2^p} \leq s^{-  c_{\text{SLD},2} \alpha}$,
		\item $2^p$ is divisible by $m^{\text{PCPP}}$. 
	\end{enumerate}
	By the choice of parameters, one can check that $p = O(\polylog(\alpha, \log(n)))$ 
\end{itemize}
The Turing machine $\mathtt{ComputePCPP}_{\alpha}$ takes, as input, another Turing machine $\langle \deciderTM \rangle$, a natural number $n \in \bN$ and a pair of strings $(x,y)$ such that $|x|, |y| \leq \log^{\alpha}(n)$, and outputs a description of a $m^{\text{PCPP}}$-variate polynomial $\textbf{g}_{\deciderTM} \in \Idpoly(p, m^{\text{PCPP}},p)$. $\mathtt{ComputePCPP}_{\alpha}$ has the following properties:
\begin{itemize}
	\item (Time complexity): $\mathtt{ComputePCPP}_{\alpha}$ takes time $O(\poly(\alpha, \log^{\alpha}(n), | \langle \deciderTM \rangle |))$. 
	\item (The complexity of the PCPP formula): The description of $\textbf{g}_{\deciderTM}$ can be represented using \\ $O(\poly(\alpha, \log(n))$ bits, and evaluating $\textbf{g}_{\deciderTM}$ at a single point takes time $O(\poly(\alpha, \log(n))$. 
\end{itemize}
Furthermore, if $|\deciderTM| = \poly(\alpha)$ and there exist two strings $(a,b) \in \{0,1\}^{n^{\alpha}}$ such that $D(n,x,y,a,b) = 1$ with
\begin{equation*}
	\TIME_{ \deciderTM}(n,x,y,a,b) \leq n^{\alpha}. 
\end{equation*}
Then there exist five polynomials $\overline{\textbf{g}}_a, \overline{\textbf{g}}_b,\overline{\textbf{g}}_{w_0}, \overline{\textbf{g}}_{w_1},\overline{\textbf{g}}_{w_2} \in \Idpoly(p, m^{\text{ans}},p)$ such that the following holds:
\begin{itemize}
	\item $\overline{\textbf{g}}_a$ is the Reed-Muller encoding of $\text{enc}_{\Gamma}(a)$, where $\text{enc}_{\Gamma}$ is the encoding map for the ``padded  version" of $a$ which maps any string with length at most $n^{\alpha}$ to a string with $2^{ m^{\text{ans}}}$. The padded encoding map is given in~\cite[Proposition 10.19]{jiMIPRE2022a}. The Reed-Muller encoding is defined over $\bF_{2^p}^{m^{\text{PCPP}}}$. 
	\item $\overline{\textbf{g}}_b$ is the Reed-Muller encoding of $\text{enc}_{\Gamma}(b)$, defined similarly as $\overline{\textbf{g}}_a$. 
	\item  For every $s \in \bF_{2^p}^{m^{\text{PCPP}}}$, partition $s =(s_0, \cdots, s_4, b_0, \cdots, b_4, z)$, where for $i \in [5]$, $s_i \in \bF_{2^p}^{m^{\text{ans}}}$,  $b_i \in \bF_{2^p}$ and $z \in \bF_{2^p}^s$. Define the polynomial $\textbf{g}_{\deciderTM}^{\text{Full}}  \in \Idpoly(p,m^{\text{PCPP}},d)$ as the polynomial with individual degrees at most $32$ as 
	\begin{equation} \label{eq:PCPPpolycondition}
		\textbf{g}_{\deciderTM}^{\text{Full}}(s) = \textbf{g}_{\deciderTM}(s) \cdot (\overline{\textbf{g}}_a(s_0) - b_0) \cdot (\overline{\textbf{g}}_b(s_1) - b_1)  \cdot (\overline{\textbf{g}}_{w_0}(s_2) - b_2) \cdot (\overline{\textbf{g}}_{w_1}(s_3) - b_3) (\overline{\textbf{g}}_{w_2}(s_4) - b_4). 
	\end{equation}
	Moreover, for every $s \in \{0,1\}^{m^{\text{PCPP}}} \subseteq \bF_{2^p}^{m^{\text{PCPP}}}$, $\textbf{g}_{\deciderTM}^{\text{Full}}(s)= 0$. 
\end{itemize}
\end{theorem}
We remark that the parameter requirement for $(m^{\text{ans}}, g, p)$ given in the above theorem follows according to~\cite[Definition 10.22]{jiMIPRE2022a}, and indeed the individual degree for each polynomials in the above definition is at most $d$. Intuitively, $\textbf{g}_a$ and $\textbf{g}_b$ in the above theorem encodes the answers given by the provers for the game $\cG$, and $\overline{\textbf{g}}_{w_0}$, $\overline{\textbf{g}}_{w_1}$, $\overline{\textbf{g}}_{w_2}$ is an encoding for the 3-SAT instances from the computation steps of $\deciderTM$ via the well known Cook-Levin encoding. In the theorem above, a $\polylog$ time bounded verifier can only compute a description of the polynomial $\textbf{g}_{\deciderTM}$ given $\deciderTM$ and the question pair $(x,y)$, and can only evaluate $O(\polylog(n))$ points from  $\textbf{g}_{\deciderTM}$. The verifier must somehow query the potentially dishonest provers for the existence of $\textbf{g}_a, \textbf{g}_b,\textbf{g}_{w_0}, \textbf{g}_{w_1},\textbf{g}_{w_2}$ to confirm the existence of such answer pair $(a,b)$ such that $\deciderTM(n,x,y,a,b)=1$. 

On a high level, in the PCPP protocol, the verifier can compute the low-individual degree polynomial $\textbf{g}_{\deciderTM}$ since he has access to the description of the decider $\deciderTM$. Then, the verifier can, with the help of the prover, verify the existence of the five low-individual degree polynomials guaranteed by the above theorem, and as well as the resulting $\textbf{g}_{\deciderTM}^{\text{Full}}$ is $0$ on all the points within the ``0/1 subcube:. Verifying the existence of the five low-individual degree is easy via the simultaneous quantum low-individual degree test given on the last subsection. The verifier can use the following lemma to check that $\textbf{g}_{\deciderTM}^{\text{Full}}$ is zero on the subcube $\{0,1\}^{m^{\text{PCPP}}}$. 
\begin{lemma}[Polynomial basis of zero functions, Proposition 10.21 of~\cite{jiMIPRE2022a}] \label{lem:polyonzero}
	Let $m,p \in \bN$ and let $\mathbf{f} \in \Idpoly(p,m,d)$. Suppose $\mathbf{f}(s) = 0$ for all $s \in \{0,1\}^{m}\subseteq \bF_{2^p}^{m}$. Then there exist $m$ polynomials $\{\textbf{c}_i \}_{i \in [m]}$, each $\textbf{c}_i \in \Idpoly(p,m,d)$, and for all $(x_0, \cdots, x_m) \in \bF_{2^p}^m$
	\begin{equation*}
		\mathbf{f}(x_0, \cdots, x_m) = \sum_{i \in [m]} \textbf{c}_i(x_0, \cdots, x_m) \cdot \mathbf{zero}(x_i)
	\end{equation*}
	where $\mathbf{zero}: \bF_{2^p} \rightarrow \bF_{2^p}$ is the polynomial $x(1-x)$. 
\end{lemma}
Hence, instead of checking $\textbf{g}_{\deciderTM}^{\text{Full}}$ directly, the verifier can ask the provers to compute each $\textbf{c}_i$ guaranteed by the above theorem, and check their existence again via the simultaneous quantum low-individual degree test. 

We now give a summary for the verification procedure for a verifier given a decider $\deciderTM$ and the integer $n$ assuming the provers are honest. We first define $\textbf{g}_a, \textbf{g}_b,\textbf{g}_{w_0}, \textbf{g}_{w_1},\textbf{g}_{w_2}: \bF_{2^{p}}^{ m^{\text{PCPP}}} \rightarrow \bF_{2^{p}}$ as follows: For every $s \in \bF_{2^p}^{m^{\text{PCPP}}}$, partition $s =(s_0, \cdots, s_4, b_0, \cdots, b_4, z)$, where for $i \in [5]$, $s_i \in  \bF_{2^p}^{m^{\text{ans}}}$,  $b_i \in  \bF_{2^p}$ and $z \in  \bF_{2^p}^s$.
\begin{align}
	\nonumber &\textbf{g}_a(s_0, \cdots, s_4,b, w) = \overline{\textbf{g}}_a(s_0), \qquad \textbf{g}_{b}(s_0, \cdots, s_4,b, w) = \overline{\textbf{g}}_b(s_1), \qquad \textbf{g}_{w_0}(s_0, \cdots, s_4,b, w) = \overline{\textbf{g}}_{w_0}(s_2)\\
	&\textbf{g}_{w_1}(s_0, \cdots, s_4,b, w) = \overline{\textbf{g}}_{w_1}(s_3),\qquad \textbf{g}_{w_2}(s_0, \cdots, s_4,b, w) = \overline{\textbf{g}}_{w_2}(s_4), \label{eq:PCPpolypartition} 
\end{align}
Note for $v \in \{a,b,  w_0, w_1, w_2\}$, $\textbf{g}_v$ and $\overline{\textbf{g}}_v$ are almost identical, except that $\textbf{g}_v$ is a low-degree polynomial over $m^{\text{ans}}$ and $\overline{\textbf{g}}_v$ is a low-degree polynomial over $m^{\text{PCPP}}$. We can also rewrite~\eqref{eq:PCPPpolycondition} without the need to partition the input $s$, as follows: 
\begin{equation} \label{eq:PCPPpolyconditionalt}
	\textbf{g}_{\deciderTM}^{\text{Full}}(s) = \textbf{g}_{\deciderTM}(s) \cdot (\textbf{g}_a(s) - b_0) \cdot (\textbf{g}_b(s) - b_1)  \cdot (\textbf{g}_{w_0}(s) - b_2) \cdot (\textbf{g}_{w_1}(s) - b_3) (\textbf{g}_{w_2}(s) -  b_4). 
\end{equation}
The answer reduction is a combination between the oracularization transformation, however, instead of computing the decision procedure given in~\Cref{fig:Oracularization}, the encoding from~\Cref{thm:classicalPCPP} are used instead to verify that the provers indeed generate the correct answers for the given question pair. To be more precise, assuming that both provers are honest, the verifier performs the following on the answer reduction protocol with the provers.


\begin{enumerate}
	\item  The verifier first samples the question pair $(x,y)$ and $(n_A, n_B)$ according to the sampling procedure given in~\Cref{thm:classicalPCPP}, and send the question pair normally to the two provers. 
	\item Upon receiving the question label and the question pair, the prover computes the following:
	\begin{itemize}
		\item If the question label is (Prover, A), the prover generate the  answer $a$ for the question $x$ for $\cG$, and then generate the polynomial $\overline{\textbf{g}}_a$ by using the encoding map $\text{enc}_{\Gamma}$~\Cref{thm:classicalPCPP}. 
		\item Similarly, if the question label is (Prover, B), the prover generates the answer $b$ for the question $y$ for $\cG$, and then generate the polynomial $\overline{\textbf{g}}_b$. 
		\item If the question label is (Oracularization), the prover generates the answer pair $(a,b)$ for the question pair $(x,y)$ for $\cG$ (from which $a$ only depends on $x$ and $b$ only depends on $y$). Then, the prover computes the following 
		\begin{itemize}
			\item The polynomial $\overline{\textbf{g}}_a$, $\overline{\textbf{g}}_b$ similarly as above.
			\item $\textbf{g}_{\deciderTM}(x)$ base on $\deciderTM$, and the question pair $(x,y)$ given by the verifier. 
			\item The $m^{\text{ans}}$-variate polynomial $\overline{\textbf{g}}_{w_0}, \overline{\textbf{g}}_{w_1},\overline{\textbf{g}}_{w_2}$ and the $m^{\text{PCPP}}$-variate polynomial $\textbf{g}_{\deciderTM}^{\text{Full}}(x)$ guarantee by~\Cref{thm:classicalPCPP}. 
		\end{itemize}
	\end{itemize}
	\item Then, the verifier computes $(m^{\text{ans}},m^{\text{PCPP}}, g, p)= \mathtt{PCPParameter}_{\alpha}(n)$. The verifier also computes the description of $\textbf{g}_{\deciderTM}$, by running $\mathtt{ComputePCPP}_{\alpha}(\deciderTM, n, x, y)$. 
	\item If the question pair is (Prover, A) \textbf{--} (Oracularization) or  (Prover, A) \textbf{--} (Prover, A), the verifier uses the $(p, m^{\text{ans}}, p)$-quantum low-individual degree test to verify that the prover shares the same low degree polynomial $\overline{\textbf{g}}_a$. 
	\item If the question pair is (Prover, B) \textbf{--} (Oracularization) or  (Prover, B) \textbf{--} (Prover, B), the verifier uses the $(p, m^{\text{ans}}, p)$-quantum low-individual degree test to verify that the prover shares the same low degree polynomial $\overline{\textbf{g}}_b$. 
	\item If the question pair for both provers is (Oracularization), the verifier performs the following with some constant probability. 
	\begin{enumerate}
		\item (Low-individual degree test on assignments) The verifier arbitrarily picks $w \in \{w_0, w_1, w_2\}$, and uses the $(p, m^{\text{ans}}, p)$-quantum low-individual degree test to verify that the prover shares the same low degree polynomial $\overline{\textbf{g}}_w$. 
		\item (Simultaneous low-individual degree test) The verifier performs the $(p, m^{\text{ans}}, p, 6 + m^{\text{PCPP}})$-simultaneous low-individual degree test on the polynomials
		\begin{equation*}
			\textbf{g}_a, \textbf{g}_b,\textbf{g}_{w_0}, \textbf{g}_{w_1},\textbf{g}_{w_2}, \textbf{g}_{\deciderTM}^{\text{Full}},  \textbf{c}_0, \cdots,  \textbf{c}_{m^{PCPP} -1}
		\end{equation*} 
		where $c_0 \cdots c_{m^{PCPP} -1}$ are the polynomials guaranteed by~\Cref{lem:polyonzero} when applied to $\textbf{g}_{\deciderTM}^{\text{Full}}(x)$ to verify that the prover actually shares these polynomials.  
		\item (Evaluation test) The verifier samples $s \in \bF_{2^p}^{m^{PCPP}}$, and ask both provers to compute
		\begin{align*}
			(u_0, \cdots, u_4) = (\textbf{g}_a(s), \textbf{g}_b(s),\textbf{g}_{w_0}(s), \textbf{g}_{w_1}(s),\textbf{g}_{w_2}(s)), \, \qquad \, \gamma = \textbf{g}_{\deciderTM}^{\text{Full}}(s),  \\
			(\beta_0, \cdots, \beta_{m^{PCPP} -1}) =( \textbf{c}_0(s), \cdots,  \textbf{c}_{m^{PCPP} -1}(s)). 
		\end{align*}
		The verifier rejects under the following conditions.
		\begin{itemize}
			\item (Consistency check)  If the two provers output different values. 
			\item (Formula check)  Parse $s =(s_0, \cdots, s_4, b_0, \cdots, b_4, z)$, where for $i \in [5]$, $s_i \in \{0,1\}^{m^{\text{ans}}}$,  $b_i \in \{0,1\}$ and $z \in \{0,1\}^s$. The verifier rejects if
			\begin{equation*}
				\gamma \neq \textbf{g}_{\deciderTM}(s) \cdot (u_1 - b_0) \cdots (u_4 - b_4). 
			\end{equation*}
			\item (Zero on subcube test check) Parse $s = (s_0, \cdots, s_{m^{PCPP} -1})$, where each $s_i \in \bF_{2^p}$ for $i \in [m^{PCPP}]$. The verifier rejects if
			\begin{equation*}
				\gamma \neq \sum_{i \in[m^{PCPP}]} \beta_i \cdot \textbf{zero}(s_i). 
			\end{equation*}
		where $\textbf{zero}$ is the polynomial defined in~\Cref{lem:polyonzero}.
		\end{itemize}
	\end{enumerate}
\end{enumerate}

Since the input/output for the (evaluation test) on the prover side is the same as a ``point" question on the (simultaneous low-individual degree test), hence the evaluation test is attached as part of the consistency test when running the simultaneous low-individual degree test in the answer reduction protocol below. As mentioned in the beginning of the section, the ``oracularization" question pair can also be combined with the ``simultaneous low-individual degree test" question to make the above procedure as a one round interaction. 

To analyze the ``soundness" condition of the protocol, we have to ensure that if there is no valid answer pair $(a,b)$ with the appropriate length such that $\deciderTM(n,x,y,a,b) =1$, then the provers cannot generate valid polynomials $\textbf{g}_v$ for $v \in \{a,b,  w_0, w_1, w_2\}$ which can be used to trick the verifier in the above procedure. In order to state this more precisely, we need to first define the notion of a low-degree PCPP proof below. 

\begin{definition}[Low-degree PCPP proof, definition 10.23 of~\cite{jiMIPRE2022a}]
	Given $m^{\text{ans}},g, p \in \bN$, let $m^{\text{PCPP}}$ be as of the requirement given in~\Cref{thm:classicalPCPP}. A low-degree PCPP proof is a tuple $\mathbf{\Pi}_{m^{\text{ans}},g, p}$ of evaluation tables over polynomials 
	\begin{equation} \label{eq:PCPPallpolys}
		(\overline{\textbf{g}}_a, \overline{\textbf{g}}_b, \overline{\textbf{g}}_{w_0}, \overline{\textbf{g}}_{w_1},\overline{\textbf{g}}_{w_2}, \textbf{g}_{\deciderTM}^{\text{Full}},  \textbf{c}_0, \cdots,  \textbf{c}_{m^{PCPP} -1}),
	\end{equation}
	where $\overline{\textbf{g}}_a, \overline{\textbf{g}}_b, \overline{\textbf{g}}_{w_0} , \overline{\textbf{g}}_{w_1}, \overline{\textbf{g}}_{w_2} \in \Idpoly(p,m^{\text{ans}},p)$, and $\textbf{g}_{\deciderTM}^{\text{Full}},\textbf{c}_0, \cdots,  \textbf{c}_{m^{PCPP} -1} \in \Idpoly(p,m^{\text{PCPP}},p)$. The evaluation of $\mathbf{\Pi}_{m^{\text{ans}},g, p}$ at $s \in \{\bF_{2^p}^{m^{PCPP}}\}$ is given by 
	\begin{equation*}
		\textbf{eval}_s(\mathbf{\Pi}_{m^{\text{ans}},g, p}) = (\overline{\textbf{g}}_a(s_0), \overline{\textbf{g}}_b(s_1), \overline{\textbf{g}}_{w_0}(s_2), \overline{\textbf{g}}_{w_1}(s_3),\overline{\textbf{g}}_{w_2}(s_4), \textbf{g}_{\deciderTM}^{\text{Full}}(s),  \textbf{c}_0(s), \cdots,  \textbf{c}_{m^{PCPP} -1}(s)). 
	\end{equation*}
	where $s$ is parsed as  $s =(s_0, \cdots, s_4, b_0, \cdots, b_4, z)$. 
\end{definition}
Intuitively, a low-degree PCPP proof is a classical proof in which the provers can generate a proof for the above verification procedure. We now state the (classical) soundness result for the PCPP procedure given in~\Cref{thm:classicalPCPP}. 

\begin{theorem}[Classical soundness of the time bounded decider PCPP, theorem 10.25 of~\cite{jiMIPRE2022a}] \label{thm:classicalsoundnessPCPP}
	Let $n, \alpha \in \bN$, let $(\mathtt{ComputePCPP}_{\alpha}, \mathtt{PCPParameter}_{\alpha})$ be the two Turing machines given in~\Cref{thm:classicalPCPP}, and let
	\begin{itemize}
		\item $\deciderTM$ be a decider for a CL verifier $\verifiersequence$,
		\item $x,y$ be two strings of length at most $\log^{\alpha}(n)$,
		\item $(m^{\text{ans}},m^{\text{PCPP}}, g,p)$ be the outputs of $\mathtt{PCPParameter}_{\alpha}(n)$,
		\item $\textbf{g}_{\deciderTM}$ be the outputs from $\mathtt{ComputePCPP}_{\alpha}(\deciderTM,n, x,y)$. 
	\end{itemize}
	Then the following holds:
	\begin{itemize}
		\item (Completeness): If there exist two strings $a,b \in \{0,1\}^{n^{\alpha}}$ such that $D(n,x,y,a,b) = 1$ with $\TIME_{\deciderTM}(n,x,y,a,b) \leq n^{\alpha}$, then there exists a low-degree PCPP proof $\mathbf{\Pi}_{m^{\text{ans}},g, p}$ such that for all $s \in \{\bF_{2^p}^{m^{PCPP}}\}$
		\begin{equation*}
			\mathtt{ValidatePCPP}(\textbf{g}_{\deciderTM}, m^{\text{ans}},g, p, s, \textbf{eval}_s(\mathbf{\Pi}_{m^{\text{ans}},g, p})) = 1
		\end{equation*}
		where the function $\mathtt{ValidatePCPP}$ is as specified in~\Cref{pseu:validatePCPP}. 
		\item (Soundness): If there exists a low-degree PCPP proof $\mathbf{\Pi}_{m^{\text{ans}},g, p}$ such that 
		\begin{equation}
			\Pr_{s \sim \{\bF_{2^p}^{m^{PCPP}} \}}[	\mathtt{ValidatePCPP}(\textbf{g}_{\deciderTM}, m^{\text{ans}},g, p, s, \textbf{eval}_s(\mathbf{\Pi}_{m^{\text{ans}},g, p})) = 1] > \frac{1}{2}.
		\end{equation}
		Then there exist two strings $a,b \in \{0,1\}^{n^{\alpha}}$ such that $\overline{\textbf{g}_a} =  \text{enc}_{\Gamma}(a)$ and $\overline{\textbf{g}_b} =  \text{enc}_{\Gamma}(b)$, where $\overline{\textbf{g}_a}, \overline{\textbf{g}_b}$ are two low-individual degree polynomials used to define $\mathbf{\Pi}_{m^{\text{ans}},g, p}$ and $D(n,x,y,a,b) = 1$ with 
		\begin{equation*}
			\TIME_{\deciderTM}(n,x,y,a,b) \leq n^{\alpha}. 
		\end{equation*}
		
	\end{itemize}
\end{theorem}

 \begin{algorithm}[H]
	\DontPrintSemicolon
	
	\textbf{Input}: Polynomial $(\textbf{g}_{\deciderTM})$, parameter $m^{\text{ans}},g, p$, PCPP view $s, \Xi$
	
	Compute $m^{\text{PCPP}} = 5\cdot m^{\text{ans}} + 5 +g$.
	
	Parse $\Xi = (u_0, \cdots, u_n, \gamma, \beta_0, \cdots, \beta_{m^{\text{PCPP}}-1})$ where each variable is in $\{0,1\}^{m^{\text{PCPP}}}$, return 0 if this cannot be done.
	
	Parse $s =(s_0, \cdots, s_4, b_0, \cdots, b_4, z)$, where for $i \in [5]$, $s_i \in \{0,1\}^{m^{\text{ans}}}$,  $b_i \in \{0,1\}$ and $z \in \{0,1\}^s$, return 0 if this cannot be done.
	
	If $\gamma \neq \textbf{g}_{\deciderTM}(s) \cdot (u_1 - b_0) \cdots (u_4 - b_4)$, return 0.
	
	If $\gamma \neq \sum_{i \in[m^{PCPP}]} \beta_i \cdot \textbf{zero}(s)$, return 0
	
	Return 1 if all the clause above fails. 
	
	\caption{$\mathtt{ValidatePCPP}(\textbf{g}_{\deciderTM}, m^{\text{ans}},g, p, s, \Xi)$, the validation algorithm for the PCPP procedure.}
	\label{pseu:validatePCPP}
\end{algorithm}

In the above theorem, the completeness clause follows directly from~\Cref{thm:classicalPCPP}. However, the soundness clause is stated differently from the quantum soundness used in this paper because the definition of the PCP differs from that of the classical $\MIP$ (see~\cite{aroraProofVerificationHardness1998}). The soundness theorem essentially states the following: Imagine that the verifier can give a point $s$ to one of the provers, Alice, and ask her to evaluate $s$ on all the polynomials given in~\eqref{eq:PCPPallpolys} and output the answer. Alice has to evaluate the polynomials honestly, but she does not necessarily have to generate the given polynomials according to the procedure given by the PCPP. She can instead fix \emph{any} low-individual degree polynomials and evaluate the point $s$ on them. What the above theorem essentially states is that, unless $\TIME_{\deciderTM}$ can be satisfied (by some $a$, $b$), Alice cannot construct \textbf{any} sets of polynomials that can cause the verifier to accept using the procedure given in~\Cref{pseu:validatePCPP} with high probability. 


\subsection{The answer reduction transformation } \label{sec:Answerreductionproof}

Using the PCPP procedure given in the last subsection, we give a proof of~\Cref{prop:AnswerReduction} below. Since the transformation is essentially the same as the transformation given in~\cite[Section 10]{jiMIPRE2022a} (with a few lemmas related to the tensor product model being swapped for lemmas related to the commuting operator model instead).

\begin{figure}[!htbp]
	\centering
	\scriptsize
	\begin{gamespec}
		\setlength{\tabcolsep}{1em}
		\begin{tabularx}{\textwidth}{ l   l  l   X   }
			\toprule
			Orac Q  & SLDT Q  & Question content & Answer format \\
							\hline
								\rule{0pt}{3ex}   
			\multirow{3}{*}{(Prover, A)} & (Point) & $x \in \cX$,  $s^A \in \bF_{2^p}^{m^{\text{ans}}}$ & $u_A =\bF_q$\\
			& (Dline) &  $x \in \cX$,  $(j, s^{A}_{D}) \in [m^{\text{ans}}] \times \bF_{2^p}^{m^{\text{ans}}}$ &  $\textbf{f}_A^{\text{D}}: \bF_{2^p} \rightarrow  \bF_{2^p}$\\
			 & (Aline) &  $x \in \cX$,  $(j,v, s^{A}_{A}) \in  [m^{\text{ans}}] \times \bF_{2^p}^{2m^{\text{ans}}}$ & $\textbf{f}_A^{\text{A}}: \bF_{2^p} \rightarrow  \bF_{2^p}$\\
			 \hline
			 	 \rule{0pt}{3ex}   
			\multirow{3}{*}{(Prover, B)} & (Point)& $y \in \cX$,  $s^B \in \bF_{2^p}^{m^{\text{ans}}}$  & $u_B =\bF_q$\\
			 & (Dline) &  $y \in \cX$,  $(j,  s^{B}_{D}) \in [m^{\text{ans}}] \times \bF_{2^p}^{m^{\text{ans}}}$ &  $\textbf{f}_B^{\text{D}}: \bF_{2^p} \rightarrow  \bF_{2^p}$ \\
			 & (Aline) &  $y \in \cX$,  $(j,v, s^{B}_A) \in [m^{\text{ans}}] \times \bF_{2^p}^{2m^{\text{ans}}}$ & $\textbf{f}_B^{\text{D}}: \bF_{2^p} \rightarrow  \bF_{2^p}$ \\
			 \hline
			 \rule{0pt}{3ex}   
			  & (Point) & $(x,y) \in \cX^2$, $s^{w_o} \in \bF_{2^p}^{m^{\text{ans}}}$ & $u_{w_o} \in \bF_{2^p}$\\
			 ($\text{Ora}_o$) & (Dline)  &   $(x,y) \in \cX^2$,  $(j, s_{\text{D}}^{w_o}) \in [m^{\text{ans}}] \times \bF_{2^p}^{m}$ & $\textbf{f}_{w_o}^{\text{D}}: \bF_{2^p} \rightarrow \bF_{2^p}$\\
			 \quad $o \in [3]$ & (Aline)   &  $(x,y) \in \cX^2$,  $(j,v, s_{\text{A}}^{w_o}) \in [m^{\text{ans}}] \times \bF_{2^p}^{2m}$ & $\textbf{f}_{w_o}^{\text{D}}: \bF_{2^p} \rightarrow \bF_{2^p}$\\
			 \hline 
			 \rule{0pt}{3ex}   
			  \multirow{5}{*}{(\text{Ora})}& (Point) & $(x,y) \in \cX^2$,  $s \in \bF_{2^p}^{m}$ & $(u_0, \cdots, u_4, \gamma, \beta_0 \cdots, \beta_{m-1}) \in \bF_{2^p}$\\
			& (Dline) &   $(x,y) \in \cX^2$,  $(j, s^{\text{D}}) \in [m] \times \bF_{2^p}^{m}$ & $(\textbf{f}_{U_0}^{\text{D}}, \cdots, \textbf{f}_{U_4}^{\text{D}}, \textbf{f}_{\Gamma}^{\text{D}}, \textbf{f}_{B_0}^{\text{D}} \cdots,\textbf{f}_{B_{m-1}}^{\text{D}} ) $\\
			&  &   & ${\textbf{f}_{v}^{\text{D}}:\bF_{2^p} \rightarrow \bF_{2^p}}$\\
			 & (Aline)  &  $(x,y) \in \cX^2$,  $(j,v, s^{\text{A}}) \in [m] \times \bF_{2^p}^{2m}$ & $(\textbf{f}_{U_0}^{\text{A}}, \cdots, \textbf{f}_{U_4}^{\text{A}}, \textbf{f}_{\Gamma}^{\text{A}}, \textbf{f}_{B_0}^{\text{A}} \cdots,\textbf{f}_{B_{m-1}}^{\text{A}} )$ \\
			 &  &   & ${\textbf{f}_{v}^{\text{D}}:\bF_{2^p} \rightarrow \bF_{2^p}}$\\
			\bottomrule
			\multicolumn{4}{c}{Figure: Q and A format for the answer reduction protocol $\texttt{AnswerReduction}_{\alpha}$ applied to a CL verifer $(\samplerTM, \deciderTM)$. }
		\end{tabularx}
		\footnotesize
		\begin{center}
			\textbf{Sampling procedure}
		\end{center}
		If at any point during the sampling procedure given below, the following happens:
		\begin{itemize}
			\item There is an error in running the subroutine $\samplerTM$ or $\mathtt{PCPParameter}_{\alpha}(n)$. 
			\item Or if the runtime exceed $\log^{\gamma^{\Answerred}}(n)$ (for some $\gamma^{\Answerred}$ picked later in the proof of~\Cref{prop:AnswerReduction}). 
		\end{itemize}
		Then halt the sampling procedure and instead sample a question pair from $\cG_{\text{reject}}$, the rejecting game from~\Cref{def:syncrejgame}. 
		\begin{enumerate}
			\item Given the input $n$, sample a question pair $(x,y) \sim \mu_n$ by using $\samplerTM(n, \cdot)$. Then compute the parameter $(m^{\text{ans}}, m^{\text{PCPP}}, g,p) = \mathtt{PCPParameter}_{\alpha}(n)$ and set $m = m^{\text{PCPP}}$. 
			\item Uniformly $(n_{O,0}), n_{O, 1} \in \{ \text{(Prover, A)}, \text{(Prover, B)},  \text{(Ora)} \}^2$, if $\text{(Ora)}$ is sampled for $n_{O,i}$ for $i \in \{0,1\}$, uniformly sample another element $\{\text{(Ora)}_0, \text{(Ora)}_1,\text{(Ora)}_2, \text{(Ora)}\}$ and set $n_{O,i}$ to be that element. Sample $n_{L,0}, n_{L, 1} \sim \{(\text{Point}, \text{DLine}, \text{ALine})\}$ and sample $s \sim \{0,1\}^m$.
			\item Parse $s =(s_0, \cdots, s_4, b_0, \cdots, b_4, z)$, where for $i \in [5]$, $s_i \in \{0,1\}^{m^{\text{ans}}}$,  $b_i \in \{0,1\}$ and $z \in \{0,1\}^s$. 
			\item For $i \in \{0,1\}$ corresponds to the two provers
			\begin{enumerate}
				\item If $n_{O,i} =  \text{(Prover, A)}$, sample a question for the $(p, m^{\text{ans}},p)$-quantum low-individual degree test according to the label of $n_{L,i}$, where the point question uses $s_A = s_0$. Send $(\text{(Prover, A)}, x)$, as well as the question label for the quantum low-individual degree test to prover $i$.
				\item If $n_{O,i} =  \text{(Prover, B)}$, sample a question for the $(p, m^{\text{ans}},p)$-quantum low-individual degree test according to the label of $n_{L,i}$, where the point question uses $s_B = s_1$. Send $(\text{(Prover, B)}, y)$, as well as the question label for the quantum low-individual degree test to prover $i$.
				\item If $n_{O,i} =  \text{(Ora)}_o$ for $o \in [3]$, sample a question for the $(p, m^{\text{ans}},p)$-quantum low-individual degree test according to the label of $n_{L,i}$, where the point question uses $s_{w_o} = s_{o+2}$. Send $(\text{(Ora)}, (x,y))$, as well as the question label for the quantum low-individual degree test to prover $i$.
				\item Otherwise, sample a question for the $(p, m,p, 6+m)$-simultaneous quantum low-individual degree test according to the label of $n_{L,i}$, where the point question uses $s$. Send $(\text{(Ora)}, (x,y))$, as well as the question label for the simultaneous quantum low-individual degree test to prover $i$.
			\end{enumerate}
		\end{enumerate}
		
	\end{gamespec}
	\caption{\footnotesize The description for $\samplerTM^{\Answerred}$ for the answer reduction protocol $\texttt{AnswerReduction}_{\alpha}$ applied to a CL verifer $(\samplerTM, \deciderTM)$.  }
	\label{fig:Answerredsample}
\end{figure}

\begin{figure}[!htbp]
	\centering
	\small
		\begin{gamespec}

			\begin{center}
				\textbf{Verification procedure}
			\end{center}
			If at any point during the sampling procedure given below, the following happens:
		\begin{itemize}
			\item There is an error in running the subroutine $\deciderTM$, $\mathtt{PCPParameter}_{\alpha}$ or $\mathtt{ComputePCPP}_{\alpha}$.
			\item Or if the runtime exceed $\log^{\gamma^{\Answerred}}(n)$ (for some $\gamma^{\Answerred}$ picked later in the proof of~\Cref{prop:AnswerReduction}). 
		\end{itemize}
		Halt the decision procedure and immediately outputs $0$. 
		
	\textbf{Preprocessing steps. }
	\begin{enumerate}
		\item Compute the description of the polynomial $\textbf{g}_{\deciderTM} = \mathtt{ComputePCPP}_{\alpha}(\deciderTM, n,x,y)$ from~\Cref{thm:classicalPCPP}. 
		\item Compute the parameter $(m^{\text{ans}}, m^{\text{PCPP}}, g,p) = \mathtt{PCPParameter}_{\alpha}(n)$ and set $m = m^{\text{PCPP}}$. 
	\end{enumerate}
	The verifier then perform the following series of check in sequences: \\
	\textbf{Low-individual degree check. }
	\begin{enumerate}
		\item If $n_{O,0} = n_{O,1} = \text{(Prover, P)}$ for $\text{P} \in \{\text{A}, \text{B}\}$ or $n_{O,0} = n_{O,1} = (\text{Ora}_o)$ for $o \in [3]$, return $0$ if they loses on the $(p, m^{\text{ans}},p)$-quantum low-individual degree test. 
		\item If $n_{O,0} = n_{O,1} = \text{(Ora)}$, return $0$ if they loses on the $(p, m,p, 6+m)$-simultaneous quantum low-individual degree test instances. 
	\end{enumerate}
	\textbf{Prover consistency check. } For $i \in \{0,1\}$
	\begin{enumerate}
			\item If $(n_{O,i}, n_{L,i}) =(n_{O,1-i}, n_{L,1-i}) $, return $0$ if the answer given by both provers are inconsistent with each other. 
		\item If $(n_{O,i}, n_{L,i})= (\text{(Prover, A)},\text{(Point)}) $ and  $(n_{O,1-i}, n_{L,1-i})= (\text{(Ora)},\text{(Point)}) $, return $0$ if $u_A \neq u_0$.
		\item  If $(n_{O,i}, n_{L,i})= (\text{(Prover, B)},\text{(Point)}) $ and  $(n_{O,1-i}, n_{L,1-i})= (\text{(Ora)},\text{(Point)}) $, return $0$ if $u_B \neq u_1$.
		\item  If $(n_{O,i}, n_{L,i})= (\text{Ora}_o,\text{(Point)})$ for $o \in [3]$ and  $(n_{O,1-i}, n_{L,1-i})= (\text{(Ora)},\text{(Point)}) $, return $0$ if $u_{w_o} \neq u_{o+2}$.
	\end{enumerate}
	\textbf{PCPP proof check. } For $i \in \{0,1\}$
	\begin{enumerate}
		\item If $(n_{O,i}, n_{L,i})= (\text{(Ora)},\text{(Point)}) $, write $s =(s_0, \cdots, s_4, b_0, \cdots, b_4, z)$. Return $0$ if either
		\begin{enumerate}
			\item $\gamma \neq \textbf{g}_{\deciderTM}(s) \cdot (u_1 - b_0) \cdots (u_4 - b_4)$. 
			\item $\gamma \neq \sum_{i \in[m^{PCPP}]} \beta_i \cdot \textbf{zero}(s_i)$
		\end{enumerate}
	\end{enumerate}
	Return $1$ if none of the test above fails.  
	\vspace{1em}
			
	\end{gamespec}
		\caption{The description for $\deciderTM^{\Answerred}$ for the answer reduction protocol $\texttt{AnswerReduction}_{\alpha}$ applied to a CL verifer $(\samplerTM, \deciderTM)$. }
	\label{fig:Answerredver}
\end{figure}
\begin{proof} 

Fix the constants $\alpha,k \in \bN$. We define the algorithm $\texttt{AnswerReduction}_{\alpha, k}$ as follows: Given a pair of Turing machine $(\samplerTM, \deciderTM)$, we specify the sampling procedure $\samplerTM^{\Answerred}$ in~\Cref{fig:Answerredsample} and the decision procedure $\deciderTM^{\Answerred}$ by~\Cref{fig:Answerredver}. We denote $\cG_n$ to be the $n$-th game generated by the original game sequence (from the input $(\samplerTM, \deciderTM)$) and $\cG_n^{\Answerred} = (\cX_n^{\Answerred}, \cA_n^{\Answerred}, \mu_n^{\Answerred}, D_n^{\Answerred})$ be the $n$-th game from the answer reduction transformation.

The ``Runtime"  clause of~\Cref{prop:AnswerReduction} follows from the description of~\Cref{fig:Answerredsample} and~\Cref{fig:Answerredver}, where computing the description of $\langle\samplerTM^{\Answerred} \rangle$ and  $\langle\deciderTM^{\Answerred} \rangle$ does not involve computing specific instances of $(\samplerTM, \deciderTM)$ (and hence have no dependency on $n$). The ``Dependency for $\samplerTM^{\Answerred}$" follows from the description of~\Cref{fig:Answerredsample} (i.e. it only depends on $\samplerTM$ and the constant $\alpha$). 

Now, assume the input for $\texttt{AnswerReduction}_{\alpha, k}$ is a game sequence $\verifiersequence = (\samplerTM, \deciderTM)$ with the property given in the ``furthermore" part of~\Cref{prop:AnswerReduction}. To show the ``complexity bound for the output" part, we analyze the complexity of $\samplerTM^{\Answerred}$ and $\deciderTM^{\Answerred}$ without the ``runtime exceeding $\log^{\gamma^{\Answerred}}(n)$ clause from~\Cref{fig:Answerredsample} and~\Cref{fig:Answerredver}. We start with $\samplerTM^{\Answerred}$ where, by analyzing each step of the description given in~\Cref{fig:Answerredsample}, we incur the following:
\begin{enumerate}
	\item Sampling the question pair $(x,y)$ from the Turing machine $\samplerTM$ takes time $\log^{\alpha}(n)$ by definition. By~\Cref{thm:classicalPCPP} computing $\mathtt{PCPParameter}_{\alpha}(n)$ takes time $O(\poly(\alpha, \log(n)))$. This implies that step 1 takes step $O(\poly(\alpha, \log^{\alpha}(n)))$. 
	\item Sampling the ``oracularization" label takes constant time. By the parameter choice guaranteed by $m^{\text{ans}}$ and $\textbf{g}$, $m = O(\poly(\alpha, \log(n)))$ which means sampling $s$ takes time $O(\poly(\alpha, \log(n)))$, giving step 2 of a runtime of $O(\poly(\alpha, \log^{\alpha}(n)))$ steps. 
	\item Step 3 takes time $O(\poly(\alpha, \log^{\alpha}(n)))$ time due to the size of $m$.
	\item Sampling the question label for the quantum low-individual degree test/simultaneous quantum low-individual degree test takes time $O(m) = O(\poly(\alpha, \log^{\alpha}(n)))$ time. 
\end{enumerate}
Hence, since $\alpha$ is chosen to be a constant in the beginning, the runtime for $\samplerTM^{\Answerred}$, assuming it does not stop by the termination clause, takes time $O(\polylog(n))$ time. Now we analyze the runtime for $\deciderTM^{\Answerred}$. By studying each line of the description given in~\Cref{fig:Answerredver}, we incur the following:
\begin{itemize}
	\item\textbf{Preprocessing steps. } Computing $\mathtt{ComputePCPP}_{\alpha}(\deciderTM, n,x,y)$ takes time $O(\poly(\alpha, \log(n)))$ and computing $\mathtt{PCPParameter}_{\alpha}(n)$ takes time $O(\poly(\alpha, \log(n), |\langle \deciderTM \rangle |)) = O(\poly(\alpha, \log(n), k))$. 
	\item \textbf{Low-individual degree check. } Verifying the quantum low-individual degree test/simultaneous quantum low-individual degree test in this instance takes time $O(\poly(\alpha, \log(n)))$ by the choices of parameters and~\Cref{lem:compfinitefield}. 
	\item \textbf{Prover consistency check. } Comparing equality of outputs takes time $O(\poly(\alpha, \log(n)))$ again by the choices of parameters. 
	\item \textbf{PCPP proof check. } Evaluating the low-individual degree polynomial and comparing the corresponding output takes time $O(\poly(\alpha, \log(n)))$ again by the choices of parameters. 
\end{itemize}
Again, since $\alpha$ and $k$ are chosen to be constants, the runtime for $\deciderTM^{\Answerred}$, assuming it does not stop by the termination clause, takes time $O(\polylog(n))$. Hence, pick $\gamma^{\Answerred} \in \bN$ to be the minimum constant such that the following holds: 
\begin{equation}
	\TIME_{\samplerTM^{\Answerred}} (n) \leq O(\log^{\gamma^{\Answerred}}(n)), \, \qquad \, \TIME_{\deciderTM^{\Answerred}} (n) \leq O(\log^{\gamma^{\Answerred}}(n)). 
\end{equation}
and let $C^{\text{trivial}}$ be the constant such that $\cG^{\text{reject}}$ can be both sampled and decided in time $C^{\text{trivial}}$. Pick $n_0^{\text{AR}} \in \bN$ be the smallest integer such that 
\begin{equation}
		\TIME_{\samplerTM^{\Answerred}}(n_0^{\text{AR}}) \leq \log^{\gamma^{\Answerred}}(n_0^{\text{AR}}), \, \qquad \, \TIME_{\deciderTM^{\Answerred}} (n_0^{\text{AR}}) \leq O(\log^{\gamma^{\Answerred}}(n_0^{\text{AR}})), 
\end{equation}
and the ``complexity bound" clause follows from the definition of the big O notation. 

For the ``level clause", we see that the input distribution $\mu_n^{\Answerred}$ is essentially an instance of the oracularization transformation of $\cG_n$ and an instance of the quantum low-individual degree test/simultaneous quantum low-individual degree test depending on the oracularization label. By combining the ``sample complexity" clause on both~\Cref{lem:oratransformation} and~\Cref{lem:simuQLID}, and using the series composition of CL functions given by~\Cref{defn:seriescomposCLfn} (in this case, only the first label for oracularization is being used to control which of the low-individual degree/simultaneous quantum low-individual degree test is being performed), we obtain a $\max\{ k+2, 5+1\}$-CL distribution that samples $\mu_n^{\Answerred}$. This shows the `level clause" of the proposition. 

For the ``completeness clause", we describe the perfect oracularizable strategy for $\cG^{\Answerred}$ as follows:
\begin{enumerate}
	\item Upon receiving the questions, the provers first perform the perfect strategy for $\cG_n^{\text{Ora}}$, the oracularized version of $\cG_n$, guaranteed by~\Cref{lem:oratransformation} using the oracularization label and the question pair for the original game as the question label for $\cG_n^{\text{Ora}}$.
	\item For $P \in \{A,B\}$, if the oracularization label is ``(Label P)", the prover computes the following:
	\begin{enumerate}
		\item Generate the corresponding low-individual degree polynomial $\textbf{g}_{p}$ according to~\Cref{thm:classicalPCPP}. 
		\item Perform the quantum low-individual degree test using the SLDT Q label and the low-degree test question content as the question label, and using $\textbf{g}_{p}$ as the ``shared polynomial" between the two provers. We remark that this step is classical according to the procedure described in~\Cref{sec:lowinddegtest}. 
	\end{enumerate}
	\item For $o \in [3]$, if the oracularization label is ``$\text{Ora}_o$", the prover computes the following:
	\begin{enumerate}
		\item Compute the polynomials described in~\Cref{thm:classicalPCPP} using the question labels $(x,y)$, and the answer obtained $(a,b)$. 
		\item Perform the quantum low-individual degree test using the SLDT Q label and the low-degree test question content as the question label, and using $\textbf{g}_{w_o}$ as the ``shared polynomial" between the provers.
	\end{enumerate}
	\item If the oracularization label is ``$\text{Ora}$", the prover computes the following:
	\begin{enumerate}
		\item Compute the polynomials described in~\Cref{thm:classicalPCPP} using the question labels $(x,y)$, and the answer obtained $(a,b)$, also computes the polynomials described in~\Cref{thm:classicalPCPP}, as well as the ``zero polynomials" $\textbf{c}_i$, $i \in [m]$ generated by~\Cref{lem:polyonzero} applied to the polynomial $\textbf{g}_{\deciderTM}^{\text{Full}}$. 
		\item Perform the simultaneous quantum low-individual degree test using the SLDT Q label and the low-degree test question content as the question label, and using all the polynomials computed in the previous step as the ``shared polynomials" between the provers. 
	\end{enumerate}
\end{enumerate}
From a measurement perspective, this is essentially just an instance of $\cG_n^{\text{Ora}}$ for the provers. Hence, the completeness clause follows from the completeness clause of~\Cref{lem:oratransformation}. 

The remainder of the proof proceeds similarly to~\cite[Section 10.7]{jiMIPRE2022a}, except we use the notation from~\Cref{fig:tracialemdtofd} to translate the proof from the finite-dimensional setting to the tracially embeddable strategies setting. The proof is provided in~\Cref{sec:ARsoundness} for completeness. 

\end{proof}

\section{Parallel repetition} \label{sec:parallelrep}
In this section, we give a proof for~\Cref{prop:Parallelrepetition} by showing that both the anchoring transformation and the parallel repetitions map a $k$-th level synchronous CL verifier to a $k+1$-th level synchronous CL verifier. Recall from~\Cref{sec:prenonlocalgames}, given a non-local $\cG$ and $r \in \bN$, we use $\cG^{\otimes r}$ to denote the r-fold parallel repetition of $\cG$. We define the anchoring transformation for a game as follows:
\begin{definition}[Anchoring transformation] \label{def:anchortransformation}
	Given a game  $\cG = (\cX, \cA, \mu, D)$, we define anchored transformation $\text{\gls{anchorgame}} = (\cX_{\bot} = \cX \cup \{\bot\}, \cA_{\bot} =\cA \cup \{\bot\}, \mu^{\bot}, D^{\bot})$ for the game $\cG$, where the question distribution $\mu_{\bot}$ is defined as 
	\begin{equation*}
		\mu^{\bot}(x,y) = \begin{cases}
			\frac{1}{4} \mu(x,y) & \text{if } (x, y) \in \cX^2 \\
			\frac{1}{4} \mu_x(x) & \text{if } y = \bot \\
			\frac{1}{4} \mu_y(y) & \text{if } x = \bot \\
			\frac{1}{4} & \text{if } x = \bot, y = \bot \\
		\end{cases},
	\end{equation*}
	where recall $\mu_x$ and $\mu_y$ are the marginal distribution for $\mu$ on both provers' sides respectively. The evaluation $D^{\bot}$ is defined as 
\begin{equation*}
	D^{\bot}(x,y,a,b) = \begin{cases}
		D(x,y,a,b) & \text{if } (x, y) \in \cX^2 \\
		1 & \text{if } (x = a = \bot, y \in \cX) \lor (y = b = \bot, x \in \cX) \lor (x = y = a = b = \bot)  \\
		0 & \text{otherwise}
	\end{cases}.
\end{equation*} 
\end{definition}
We remark that the above definition is a modification for the anchoring transformation from~\cite{bavarianAnchoredParallelRepetition2021} with $\alpha$ taken to be $\frac{1}{2}$, and instead of giving the provers a free win, we now expect that the provers both need to output an anchoring symbol in order to win the game. This ensures that every synchronous game remains synchronous after the transformation. As shown in~\cite{jiMIPRE2022a}, this does not change the strong parallel repetition guarantee by~\cite{bavarianAnchoredParallelRepetition2021}. We see that the anchoring transformation only changes the value of a game by a constant factor through the following lemma. 

\begin{lemma}[Preservation of value of the anchoring transformation]
	Let $t \in \{*, co\}$, and let $\cG$ be a non-local game. If $\omega^{t} (\cG) \geq 1- \epsilon$ for some $\eps \in [0,1]$, then $\omega^{t} (\Ganchor) = 1- \frac{1 - \epsilon}{4}$.  
\end{lemma}
The above lemma follows trivially by formulating a strategy that answers the anchoring symbol $\dummy$ when $\dummy$ is given as the question. It is not hard to see that, for $t \in \{*,co\}$, if there exists a perfect oracularizable synchronous strategy for $\cG$ in model $t$, then there exists a perfect oracularizable synchronous strategy for $\cG_{\bot}$ in model $t$ by adding the measurement operator $A^{\bot}_{0} = \cI_{\alicealg}$ to the existing perfect strategy for $\cG$. Furthermore, if $\omega^{t}(\cG) = \eps$, then $\omega^{t}(\cG_{\bot}) = \frac{3}{4} + \frac{\eps}{4}$. We recall the parallel repetition theorem for the anchoring transformation.
\begin{theorem}[Anchored Parallel repetition theorem] \label{thm:anchoringparallelrep}
	There exists a universal constant $c^{\text{para}}$ such that, for any model $t \in \{*,co\}$, non-local games $\cG = (\cX, \cA, \mu, D)$ with $\omega^{t} \leq 1 - \eps$, and $r \in \bN$. Then
	\begin{equation} \label{eq:mainparallel}
		\omega^{t}(\cG_{\bot}^{\otimes r}) \leq \frac{16}{\eps} \cdot \exp{\left(\frac{- c^{\text{para}} \eps^{17} r }{\log(|\cA | + 1)}\right)}.
	\end{equation}
\end{theorem}
We give a proof for the above theorem in~\Cref{sec:parallelrepappendix}. Before giving a proof for~\Cref{prop:Parallelrepetition}, we first show the following lemma. 

\begin{lemma}[Properties related to the parallel repeated anchoring transformation ] \label{lem:anchortrans}
	Let $r \in \bN$, and let $\cG = (\cX, \cA, \mu, D)$ be a CL-samplable game where the question distribution $\mu$ is a $(k,m,p)$. Then, $\cG^{\otimes r}_{\bot}$, the r-fold parallel repetition of the anchor transformation of $\cG$, is samplable via a $(k+1,r \cdot (m+2),p)$ CL distribution. 
\end{lemma}

\begin{proof}
	Let $\decidertypedfunction^{A}$ and $\decidertypedfunction^{B}$ be the CL functions acting on $V \subseteq \bF_{2^p}^m$, which define the distribution $\mu$. We first show that $\cG_{\bot}$ is samplable via a $(k+1,m+2,p)$ distribution. Define $\decidertypedfunction^{A}_{\bot}$, $\decidertypedfunction^{B}_{\bot}$ as the following:
	$\decidertypedfunction^A_{\bot}$ is defined as a series composition (\Cref{defn:seriescomposCLfn}) of two CL functions, where $V^1 = \bF_{2^p}^2$, and $V_2 = V$. For $(s_0, s_1) \in \bF_{2^p}^2$, define the zeroth level CL function for $\decidertypedfunction^A_{\bot}$ is defined as
	\begin{equation*}
		(\decidertypedfunction^A_{\bot})_{0,0} (s_0, s_1) = (s_0, 0)
	\end{equation*}
	and the collection of level $k$ CL function to be $\{(\decidertypedfunction^A_{\bot})^{(s_0, s_1)}\}_{(s_0, s_1) \in \bF_{2^p}^2}$ to be
	\begin{equation*}
		(\decidertypedfunction^A_{\bot})^{(s_0, s_1)} = \decidertypedfunction^A
	\end{equation*}
	if the first bit of $\kappa(s_0)$ is 0, where recall $\kappa(s_0) \in \{0,1\}^p$ is the canonical representation of $s_0$ and $(\decidertypedfunction^A_{\bot})^{(s_0, s_1)} = 0$, i.e. the linear function which maps all elements in $\bF^{m}_{2^p}$ to 0, otherwise. Intuitively, the first bit of $\kappa(s_0)$ being 1 indicates that the anchoring question is sampled as Alice's question. $\decidertypedfunction^B_{\bot}$ is defined similarly to $\decidertypedfunction^A_{\bot}$ except the zeroth level CL function is defined as
	\begin{equation*}
		(\decidertypedfunction^B_{\bot})_{0,0} (s_0, s_1) = (s_1, 0). 
	\end{equation*}
	We see that the CL distribution defined by $\decidertypedfunction^A_{\bot}$ and $\decidertypedfunction^B_{\bot}$ precisely samples $\mu_{\bot}$. The proof of the lemma then follows from applying \Cref{lem:parallelrefCLfunction} on $\cG_{\bot}$. 
\end{proof}

We are now ready to show~\Cref{prop:Parallelrepetition}. 

\begin{proof}
	Fix the constant $\alpha$ and a function $\mathbf{s}(n): \bN \rightarrow [0,1]$, with $O(\mathbf{s}(n)) = O(\polylog(n))$. Let $\mathbf{r}(n)$ be a function such that 
	\begin{equation} \label{eq:parallelrepparameter}
		 \frac{16}{\mathbf{s}(n)} \cdot \exp{\left(\frac{- c^{\text{para}} \mathbf{s}(n)^{17} \mathbf{r}(n) }{\alpha \log(n)}\right)} \leq \frac{1}{2}
	\end{equation}
	for all $n \in \bN$. Since $O(\mathbf{s}(n)) = O(\polylog(n))$, $\mathbf{r}(n)$ can be taken to be $ O(\polylog(n))$. 
	
	Fix an input $(\samplerTM, \deciderTM)$ and integer $n$. We describe the corresponding output $(\samplerTM^{\Pararep}, \deciderTM^{\Pararep})$ for $\texttt{Parallelrep}_{\alpha, \textbf{s}(n)}$ below: If at any point in the computation process, the computation step for running $\samplerTM$ and $\deciderTM$ either returns an invalid output or runs for time more than $n^{\alpha}$, return $0$ (i.e. an invalid input). The Turing machine $\samplerTM^{\Pararep}$ is defined as the following: $\samplerTM^{\Pararep}$ reads the first input $n$ and computes $(\mathbf{k}(n), \mathbf{m}(n) , \mathbf{p}(n)) = \samplerTM(n, \text{parameter})$. 
	
	\begin{itemize}
		\item $\samplerTM^{\Pararep}(n, \text{parameter}) = (\mathbf{k}(n),  \mathbf{r}(n) \cdot (\mathbf{m}(n) +2), \mathbf{p}(n))$. 
		\item On input $(n, \text{Divide}, s)$,  $\samplerTM^{\Pararep}$ does the following: 
		\begin{enumerate} 
			\item Parses $s = (s_0, \cdots, s_{\mathbf{r}(n)-1})$, where each $s_i \in \{0,1\}^{ (\mathbf{m}(n) +2) \cdot \mathbf{p}(n)}$ (automatically returns 0 if the input does not match). 
			\item For each $i \in [\mathbf{r}(n)]$, parse $s_i = (s_{i,0}, s_{i, > 0})$, where $s_{i,0} \in \{0,1\}^{2 \cdot \mathbf{p}(n)}$ and $s_{i, > 0} \in \{0,1\}^{\mathbf{p}(n) \cdot \mathbf{m}(n)}$. Furthermore, parse
			\begin{equation*}
			 	(s_{i,1}, \cdots, s_{i,\mathbf{k}(n)}) = \samplerTM(n, \text{Divide},s_{i, > 0}). 
			\end{equation*}
			\item For $j \in [\mathbf{k}(n) +1]$, let $t_j = (s_{0,j}, \cdots, s_{\mathbf{r}(n)-1,j})$, and return $(t_0, \cdots, t_\mathbf{k}(n))$ as the output.
		\end{enumerate}
		\item On input $(n, \text{Function}, p,j,s,x)$, $\samplerTM^{\Pararep}$ does the following:
		\begin{enumerate} 
			\item Parses $s = (s_0, \cdots s_{\mathbf{r}(n)-1})$, and $x =  (x_0, \cdots x_{\mathbf{r}(n)-1})$., where each $s_i$ (resp. $x_i$) have the same number of bits with each other (automatically returns 0 if the input does not match). 
			\item If $j = 0$, apply the zeroth level linear function as specified in the proof of~\Cref{lem:anchortrans} to each of the $x_i$ to each $i \in [\mathbf{r}(n)]$ and return the concatenated output. 
			\item Otherwise, for each $i \in [\mathbf{r}(n)]$, parse $s_i = (s_{i,0}, s_{i, > 0})$ where $s_{i, 0} \in \{0,1\}^{2 \cdot \mathbf{p}(n)}$
			\begin{itemize}
				\item If the first bit of $s_{i, 0}$ is $0$ (i.e. the normal question), compute $t_i = \samplerTM(n, \text{Function}, p,j-1,s_{i, > 0},x_i)$.
				\item Otherwise (i.e. anchoring question), let $t_i = 0$ be the zero string that has the same length as $x_i$. 
			\end{itemize}
			Return the concatenated output of all the $t_i$. 
		\end{enumerate}
	\end{itemize}
	
	The Turing machine $\deciderTM^{\Pararep}$ is defined as the following:  $\deciderTM^{\Pararep}$ reads the first input n, and computes $(\mathbf{k}(n), \mathbf{m}(n) , \mathbf{p}(n)) = \samplerTM(n, \text{parameter})$. 
	\begin{enumerate}
		\item  On input $(n,x,y,a,b)$, parse
		\begin{align*}
			&x = (x_0, \cdots, x_{\mathbf{r}(n) -1 } ),  &y = (y_0, \cdots, y_{\mathbf{r}(n)-1} ), \quad &a = (a_0, \cdots, a_{\mathbf{r}(n) -1 } )  &b = (b_0, \cdots, b_{\mathbf{r}(n) -1 } ),
		\end{align*}
		where $x_i, y_i \in \{0,1\}^{ (\mathbf{m}(n) +2), \mathbf{p}(n)}$ for all $i \in [\mathbf{r}(n)]$ and each of the $a_i$, $b_i$ have the same number of bits. Output $0$ (i.e. automatically reject) if this cannot be done. 
		\item For each $i \in [\mathbf{r}(n)]$, parse $x_i = (x_{i,0}, x_{i,> 0})$,  $y_i = (y_{i,0}, y_{i,> 0})$ where $x_{i,0}, y_{i,0} \in \bF_{2^{p}}^2$. 
		\begin{itemize}
			\item If the first bit of $x_{i,0}$ and $y_{i,0}$ are both $0$ (i.e. the normal question), compute $t_i = \deciderTM(n, x_{i,> 0},  y_{i,> 0}, a_i, b_i)$.
			\item Otherwise, $t_i = 1$ iff whenever $x_{i}$ (resp. $y_{i}$) is the anchoring question (i.e. the first bit of $x_{i,0}$ (resp. $y_{i}$) is zero), then $a_i$ (resp. $b_i$) is the zero string. $t_i = 0$ if the above condition is not met.  
		\end{itemize}
		\item Return $\bigwedge_{i \in [\mathbf{r}(n)]} t_i$. 
	\end{enumerate}
	As seen from the description above, both $(\samplerTM^{\Pararep}, \deciderTM^{\Pararep})$ can be described by using  $(\samplerTM, \deciderTM)$  as a black box, and $\samplerTM^{\Pararep}$ depends only on $\samplerTM$ and 
	\begin{itemize}
		\item $\TIME_{\samplerTM^{\Pararep} }(n) \leq O(\poly(\mathbf{r}(n), log^{\alpha}(n)))$,
		\item $\TIME_{\deciderTM^{\Pararep} }(n) \leq O(\poly(\mathbf{r}(n), log^{\alpha}(n)))$.
	\end{itemize} 
	If $(\samplerTM, \deciderTM)$ is a synchronous $k$-th level CL sampler for an infinite sequence of synchronous games $\{\cG_n = (\cX_n,\cA_n, \mu_n, D_n )\}_{n \in \bN}$ for some constant $k \in \bN$,with some constant $n_0 \in \bN$ such that for all $n \geq n_0$,
	\begin{equation*}
		\samplerTM(n, \text{Parameter}), \TIME_{ \samplerTM }(n),\TIME_{ \deciderTM} \leq \log^{\alpha}(n). 
	\end{equation*}
	Then $(\samplerTM^{\Pararep}, \deciderTM^{\Pararep})$ is a $(k+1)$-th level CL sampler for an infinite sequence of synchronous games $\{\cG_{n, \bot}^{\otimes \mathbf{r}(n)}\}_{n \in \bN}$, where $\cG_{n, \bot}^{\otimes \mathbf{r}(n)}$ is the $\mathbf{r}(n)$-fold parallel repetition of the anchoring transformation for $\cG_n$.
	
	For completeness, note if $\cG$ admits a perfect (synchronous) oracularizable strategy $\strategy = ( \cL^2(\alicealg, \tau), \\ \ket{\tau}, \{ A_{a}^x \} )$. Then the $r$-fold parallel repetition game $\cG^{\otimes r}$ admits a perfect oracularizable strategy defined on $\cL^2(\alicealg, \tau)^{\otimes r}$ because the provers can simply performing r independent instances of $\strategy$ for each question labels on the parallel repeated game. Hence, combining with the remark after~\Cref{def:anchortransformation}, if there exists a perfect oracularizable strategy for $\cG_n$, then there exists a perfect oracularizable strategy for $\cG_{n, \bot}^{\otimes \mathbf{r}(n)}$.
	
	For soundness, since $\TIME_{ \deciderTM} \leq \log^{\alpha}(n)$ for $n > n_0$, this implies that $|\cA| = \log^{\alpha}(n)$, and hence combining~\eqref{eq:parallelrepparameter} and~\Cref{thm:anchoringparallelrep} shows the soundness condition. This completes the proof of~\Cref{prop:Parallelrepetition}.

\end{proof}

\begin{CJK*}{UTF8}{gbsn}
\end{CJK*}

\printbibliography

@misc{aldousProcessesUnimodularRandom2018,
  title = {Processes on {{Unimodular Random Networks}}},
  author = {Aldous, David and Lyons, Russell},
  year = {2018},
  month = nov,
  number = {arXiv:math/0603062},
  eprint = {math/0603062},
  publisher = {arXiv},
  doi = {10.48550/arXiv.math/0603062},
  archiveprefix = {arXiv}
}

@article{anantharamanIntroductionII1Factors2010,
  title = {An Introduction to {{II1}} Factors},
  author = {Anantharaman, Claire and Popa, Sorin},
  year = {2010},
  langid = {english}
}

@article{arakiPropertiesModularConjugation1974b,
  title = {Some Properties of Modular Conjugation Operator of von {{Neumann}} Algebras and a Non-Commutative {{Radon-Nikodym}} Theorem with a Chain Rule},
  author = {Araki, Huzihiro},
  year = {1974},
  month = feb,
  journal = {Pacific Journal of Mathematics},
  volume = {50},
  number = {2},
  pages = {309--354},
  publisher = {Mathematical Sciences Publishers},
  issn = {0030-8730}
}

@article{arakiRelativeEntropyStates1977,
  title = {Relative Entropy for States of von {{Neumann}} Algebras. {{II}}},
  author = {Araki, Huzihiro},
  year = {1977},
  journal = {Publications of the Research Institute for Mathematical Sciences},
  volume = {13},
  number = {1},
  pages = {173--192},
  issn = {0034-5318},
  doi = {10.2977/prims/1195190105},
  langid = {english}
}

@article{aroraProbabilisticCheckingProofs1998,
  title = {Probabilistic Checking of Proofs: A New Characterization of {{NP}}},
  shorttitle = {Probabilistic Checking of Proofs},
  author = {Arora, Sanjeev and Safra, Shmuel},
  year = {1998},
  month = jan,
  journal = {Journal of the ACM},
  volume = {45},
  number = {1},
  pages = {70--122},
  issn = {0004-5411},
  doi = {10.1145/273865.273901}
}

@article{aroraProofVerificationHardness1998,
  title = {Proof Verification and the Hardness of Approximation Problems},
  author = {Arora, Sanjeev and Lund, Carsten and Motwani, Rajeev and Sudan, Madhu and Szegedy, Mario},
  year = {1998},
  month = may,
  journal = {Journal of the ACM},
  volume = {45},
  number = {3},
  pages = {501--555},
  issn = {0004-5411, 1557-735X},
  doi = {10.1145/278298.278306},
  langid = {english}
}

@article{babaiNondeterministicExponentialTime1991a,
  title = {Non-Deterministic Exponential Time Has Two-Prover Interactive Protocols},
  author = {Babai, L{\'a}szl{\'o} and Fortnow, Lance and Lund, Carsten},
  year = {1991},
  month = mar,
  journal = {computational complexity},
  volume = {1},
  number = {1},
  pages = {3--40},
  issn = {1420-8954},
  doi = {10.1007/BF01200056},
  langid = {english}
}

@inproceedings{barakStrongParallelRepetition2009,
  title = {Strong {{Parallel Repetition Theorem}} for {{Free Projection Games}}},
  booktitle = {Approximation, {{Randomization}}, and {{Combinatorial Optimization}}. {{Algorithms}} and {{Techniques}}},
  author = {Barak, Boaz and Rao, Anup and Raz, Ran and Rosen, Ricky and Shaltiel, Ronen},
  editor = {Dinur, Irit and Jansen, Klaus and Naor, Joseph and Rolim, Jos{\'e}},
  year = {2009},
  series = {Lecture {{Notes}} in {{Computer Science}}},
  pages = {352--365},
  publisher = {Springer},
  address = {Berlin, Heidelberg},
  doi = {10.1007/978-3-642-03685-9_27},
  isbn = {978-3-642-03685-9},
  langid = {english}
}

@article{bavarianAnchoredParallelRepetition2021,
  title = {Anchored Parallel Repetition for Nonlocal Games},
  author = {Bavarian, Mohammad and Vidick, Thomas and Yuen, Henry},
  year = {2021},
  month = mar,
  journal = {arXiv:1509.07466 [quant-ph]},
  eprint = {1509.07466},
  primaryclass = {quant-ph},
  archiveprefix = {arXiv}
}

@article{bellEinsteinPodolskyRosen1964a,
  title = {On the {{Einstein Podolsky Rosen}} Paradox},
  author = {Bell, J. S.},
  year = {1964},
  month = nov,
  journal = {Physics Physique Fizika},
  volume = {1},
  number = {3},
  pages = {195--200},
  publisher = {American Physical Society},
  doi = {10.1103/PhysicsPhysiqueFizika.1.195}
}

@book{blackadarOperatorAlgebras2006,
  title = {Operator {{Algebras}}},
  author = {Blackadar, Bruce},
  editor = {Cuntz, Joachim and F.R. Jones, Vaughan},
  year = {2006},
  series = {Encyclopaedia of {{Mathematical Sciences}}},
  volume = {122},
  publisher = {Springer},
  address = {Berlin, Heidelberg},
  doi = {10.1007/3-540-28517-2},
  isbn = {978-3-540-28486-4 978-3-540-28517-5}
}

@book{blakeApplicationsFiniteFields1993,
  title = {Applications of {{Finite Fields}}},
  author = {Blake, Ian F. and Gao, XuHong and Mullin, Ronald C. and Vanstone, Scott A. and Yaghoobian, Tomik},
  editor = {Menezes, Alfred J.},
  year = {1993},
  publisher = {Springer US},
  address = {Boston, MA},
  doi = {10.1007/978-1-4757-2226-0},
  copyright = {http://www.springer.com/tdm},
  isbn = {978-1-4419-5130-4 978-1-4757-2226-0}
}

@misc{bowenAldousLyonsConjectureII2024,
  title = {The {{Aldous--Lyons Conjecture II}}: {{Undecidability}}},
  shorttitle = {The {{Aldous--Lyons Conjecture II}}},
  author = {Bowen, Lewis and Chapman, Michael and Vidick, Thomas},
  year = {2024},
  month = dec,
  number = {arXiv:2501.00173},
  eprint = {2501.00173},
  primaryclass = {quant-ph},
  publisher = {arXiv},
  doi = {10.48550/arXiv.2501.00173},
  archiveprefix = {arXiv}
}

@misc{bowenAldousLyonsConjectureSubgroup2024,
  title = {The {{Aldous--Lyons Conjecture I}}: {{Subgroup Tests}}},
  shorttitle = {The {{Aldous--Lyons Conjecture I}}},
  author = {Bowen, Lewis and Chapman, Michael and Lubotzky, Alexander and Vidick, Thomas},
  year = {2024},
  month = jul,
  number = {arXiv:2408.00110},
  eprint = {2408.00110},
  primaryclass = {math},
  publisher = {arXiv},
  doi = {10.48550/arXiv.2408.00110},
  archiveprefix = {arXiv}
}

@article{bowlesSelftestingPauliObservables2018,
  title = {Self-Testing of {{Pauli}} Observables for Device-Independent Entanglement Certification},
  author = {Bowles, Joseph and {\v S}upi{\'c}, Ivan and Cavalcanti, Daniel and Ac{\'i}n, Antonio},
  year = {2018},
  month = oct,
  journal = {Physical Review A},
  volume = {98},
  number = {4},
  eprint = {1801.10446},
  primaryclass = {quant-ph},
  pages = {042336},
  issn = {2469-9926, 2469-9934},
  doi = {10.1103/PhysRevA.98.042336},
  archiveprefix = {arXiv}
}

@misc{chaillouxParallelRepetitionFree2015,
  title = {Parallel {{Repetition}} of {{Free Entangled Games}}: {{Simplification}} and {{Improvements}}},
  shorttitle = {Parallel {{Repetition}} of {{Free Entangled Games}}},
  author = {Chailloux, Andr{\'e} and Scarpa, Giannicola},
  year = {2015},
  month = mar,
  number = {arXiv:1410.4397},
  eprint = {1410.4397},
  primaryclass = {quant-ph},
  publisher = {arXiv},
  doi = {10.48550/arXiv.1410.4397},
  archiveprefix = {arXiv}
}

@misc{chungParallelRepetitionEntangled2015,
  title = {Parallel Repetition for Entangled K-Player Games via Fast Quantum Search},
  author = {Chung, Kai-Min and Wu, Xiaodi and Yuen, Henry},
  year = {2015},
  month = apr,
  number = {arXiv:1501.00033},
  eprint = {1501.00033},
  primaryclass = {quant-ph},
  publisher = {arXiv},
  doi = {10.48550/arXiv.1501.00033},
  archiveprefix = {arXiv}
}

@inproceedings{cleveCharacterizationBinaryConstraint2014,
  title = {Characterization of {{Binary Constraint System Games}}},
  booktitle = {Automata, {{Languages}}, and {{Programming}}},
  author = {Cleve, Richard and Mittal, Rajat},
  editor = {Esparza, Javier and Fraigniaud, Pierre and Husfeldt, Thore and Koutsoupias, Elias},
  year = {2014},
  pages = {320--331},
  publisher = {Springer},
  address = {Berlin, Heidelberg},
  doi = {10.1007/978-3-662-43948-7_27},
  isbn = {978-3-662-43948-7},
  langid = {english}
}

@inproceedings{cleveConsequencesLimitsNonlocal2004,
  title = {Consequences and Limits of Nonlocal Strategies},
  booktitle = {Proceedings. 19th {{IEEE Annual Conference}} on {{Computational Complexity}}, 2004.},
  author = {Cleve, R. and Hoyer, P. and Toner, B. and Watrous, J.},
  year = {2004},
  month = jun,
  pages = {236--249},
  issn = {1093-0159},
  doi = {10.1109/CCC.2004.1313847}
}

@misc{cleveStrongParallelRepetition2008,
  title = {Strong {{Parallel Repetition Theorem}} for {{Quantum XOR Proof Systems}}},
  author = {Cleve, Richard and Slofstra, William and Unger, Falk and Upadhyay, Sarvagya},
  year = {2008},
  month = apr,
  number = {arXiv:quant-ph/0608146},
  eprint = {quant-ph/0608146},
  publisher = {arXiv},
  doi = {10.48550/arXiv.quant-ph/0608146},
  archiveprefix = {arXiv}
}

@article{connesClassificationInjectiveFactors1976,
  title = {Classification of {{Injective Factors Cases}}},
  author = {Connes, A.},
  year = {1976},
  month = jul,
  journal = {The Annals of Mathematics},
  volume = {104},
  number = {1},
  eprint = {1971057},
  eprinttype = {jstor},
  pages = {73},
  issn = {0003486X},
  doi = {10.2307/1971057},
  langid = {english}
}

@unpublished{delasalleAlmostSynchronousCorrelations2023,
  title = {Almost Synchronous Correlations and {{Tomita-Takesaki}} Theory},
  shorthand = {\small \normalsize dlSM23},
  author = {{de La Salle}, Mikael and Marrakchi, Amine},
  year = {2023},
  month = oct,
  doi = {10.48550/arXiv.2307.08129}
}

@misc{delasalleOrthogonalizationPositiveOperator2022,
  title = {Orthogonalization of {{Positive Operator Valued Measures}}},
  shorthand = {\small \normalsize dlS22a},
  author = {{de la Salle}, Mikael},
  year = {2022},
  month = jan,
  number = {arXiv:2103.14126},
  eprint = {2103.14126},
  primaryclass = {quant-ph},
  institution = {arXiv},
  doi = {10.48550/arXiv.2103.14126},
  archiveprefix = {arXiv}
}

@misc{delasalleSpectralGapStability2022,
  title = {Spectral Gap and Stability for Groups and Non-Local Games},
  shorthand = {\small \normalsize dlS22b},
  author = {{de la Salle}, Mikael},
  year = {2022},
  month = apr,
  number = {arXiv:2204.07084},
  eprint = {2204.07084},
  primaryclass = {math},
  publisher = {arXiv},
  archiveprefix = {arXiv}
}

@misc{dinurParallelRepetitionTheorem2015,
  title = {A Parallel Repetition Theorem for Entangled Projection Games},
  author = {Dinur, Irit and Steurer, David and Vidick, Thomas},
  year = {2015},
  month = mar,
  number = {arXiv:1310.4113},
  eprint = {1310.4113},
  primaryclass = {quant-ph},
  publisher = {arXiv},
  doi = {10.48550/arXiv.1310.4113},
  archiveprefix = {arXiv}
}

@article{ekertQuantumCryptographyBased1991,
  title = {Quantum Cryptography Based on {{Bell}}'s Theorem},
  author = {Ekert, Artur K.},
  year = {1991},
  month = aug,
  journal = {Physical Review Letters},
  volume = {67},
  number = {6},
  pages = {661--663},
  issn = {0031-9007},
  doi = {10.1103/PhysRevLett.67.661},
  copyright = {http://link.aps.org/licenses/aps-default-license},
  langid = {english}
}

@article{fritzTsirelsonProblemKirchberg2012,
  title = {Tsirelson's Problem and {{Kirchberg}}'s Conjecture},
  author = {Fritz, Tobias},
  year = {2012},
  month = jun,
  journal = {Reviews in Mathematical Physics},
  volume = {24},
  number = {05},
  eprint = {1008.1168},
  primaryclass = {math-ph, physics:quant-ph},
  pages = {1250012},
  issn = {0129-055X, 1793-6659},
  doi = {10.1142/S0129055X12500122},
  archiveprefix = {arXiv}
}

@misc{goldbringConnesEmbeddingProblem2021,
  title = {The {{Connes Embedding Problem}}: {{A}} Guided Tour},
  shorttitle = {The {{Connes Embedding Problem}}},
  author = {Goldbring, Isaac},
  year = {2021},
  month = sep,
  number = {arXiv:2109.12682},
  eprint = {2109.12682},
  primaryclass = {quant-ph},
  publisher = {arXiv},
  archiveprefix = {arXiv}
}

@misc{griloSimpleProtocolVerifiable2020a,
  title = {A Simple Protocol for Verifiable Delegation of Quantum Computation in One Round},
  author = {Grilo, Alex B.},
  year = {2020},
  month = jun,
  number = {arXiv:1711.09585},
  eprint = {1711.09585},
  primaryclass = {quant-ph},
  publisher = {arXiv},
  doi = {10.48550/arXiv.1711.09585},
  archiveprefix = {arXiv}
}

@article{haagAlgebraicApproachQuantum1964,
  title = {An {{Algebraic Approach}} to {{Quantum Field Theory}}},
  author = {Haag, Rudolf and Kastler, Daniel},
  year = {1964},
  month = jul,
  journal = {Journal of Mathematical Physics},
  volume = {5},
  number = {7},
  pages = {848--861},
  publisher = {American Institute of Physics},
  issn = {0022-2488},
  doi = {10.1063/1.1704187}
}

@phdthesis{harshaRobustPCPsProximity2004,
  type = {Thesis},
  title = {Robust {{PCPs}} of Proximity and Shorter {{PCPs}}},
  author = {Harsha, Prahladh},
  year = {2004},
  copyright = {M.I.T. theses are protected by copyright. They may be viewed from this source for any purpose, but reproduction or distribution in any format is prohibited without written permission. See provided URL for inquiries about permission.},
  langid = {american},
  school = {Massachusetts Institute of Technology}
}

@book{hiaiQuantumFDivergencesNeumann2021a,
  title = {Quantum F-{{Divergences}} in von {{Neumann Algebras}}: {{Reversibility}} of {{Quantum Operations}}},
  shorttitle = {Quantum F-{{Divergences}} in von {{Neumann Algebras}}},
  author = {Hiai, Fumio},
  year = {2021},
  series = {Mathematical {{Physics Studies}}},
  publisher = {Springer},
  address = {Singapore},
  doi = {10.1007/978-981-33-4199-9},
  isbn = {978-981-334-198-2 978-981-334-199-9},
  langid = {english}
}

@article{holensteinParallelRepetitionSimplifications2009,
  title = {Parallel Repetition: Simplifications and the No-Signaling Case},
  shorttitle = {Parallel Repetition},
  author = {Holenstein, Thomas},
  year = {2009},
  journal = {Theory of Computing},
  volume = {5},
  number = {1},
  eprint = {cs/0607139},
  pages = {141--172},
  issn = {1557-2862},
  doi = {10.4086/toc.2009.v005a008},
  archiveprefix = {arXiv},
  langid = {english}
}

@article{jainParallelDeviceIndependentQuantum2020,
  title = {Parallel {{Device-Independent Quantum Key Distribution}}},
  author = {Jain, Rahul and Miller, Carl A. and Shi, Yaoyun},
  year = {2020},
  month = sep,
  journal = {IEEE Transactions on Information Theory},
  volume = {66},
  number = {9},
  eprint = {1703.05426},
  primaryclass = {quant-ph},
  pages = {5567--5584},
  issn = {0018-9448, 1557-9654},
  doi = {10.1109/TIT.2020.2986740},
  archiveprefix = {arXiv}
}

@inproceedings{jainParallelRepetitionTheorem2014,
  title = {A {{Parallel Repetition Theorem}} for {{Entangled Two-Player One-Round Games}} under {{Product Distributions}}},
  booktitle = {2014 {{IEEE}} 29th {{Conference}} on {{Computational Complexity}} ({{CCC}})},
  author = {Jain, Rahul and Pereszl{\'e}nyi, Attila and Yao, Penghui},
  year = {2014},
  month = jun,
  pages = {209--216},
  issn = {1093-0159},
  doi = {10.1109/CCC.2014.29}
}

@misc{jiMIPRE2022a,
  title = {{{MIP}}*={{RE}}},
  author = {Ji, Zhengfeng and Natarajan, Anand and Vidick, Thomas and Wright, John and Yuen, Henry},
  year = {2022},
  month = nov,
  number = {arXiv:2001.04383},
  eprint = {2001.04383},
  primaryclass = {quant-ph},
  publisher = {arXiv},
  doi = {10.48550/arXiv.2001.04383},
  archiveprefix = {arXiv}
}

@misc{jiQuantumSoundnessTesting2022,
  title = {Quantum Soundness of Testing Tensor Codes},
  author = {Ji, Zhengfeng and Natarajan, Anand and Vidick, Thomas and Wright, John and Yuen, Henry},
  year = {2022},
  month = feb,
  number = {arXiv:2111.08131},
  eprint = {2111.08131},
  primaryclass = {quant-ph},
  publisher = {arXiv},
  archiveprefix = {arXiv}
}

@book{jonesComputabilityComplexityProgramming1997,
  title = {Computability and Complexity: From a Programming Perspective},
  shorttitle = {Computability and Complexity},
  author = {Jones, Neil D.},
  year = {1997},
  series = {Foundations of Computing},
  publisher = {MIT press},
  address = {Cambridge (mass.)},
  isbn = {978-0-262-10064-9},
  langid = {english},
  lccn = {511.352}
}

@article{justesenClassConstructiveAsymptotically1972,
  title = {Class of Constructive Asymptotically Good Algebraic Codes},
  author = {Justesen, J.},
  year = {1972},
  month = sep,
  journal = {IEEE Transactions on Information Theory},
  volume = {18},
  number = {5},
  pages = {652--656},
  issn = {1557-9654},
  doi = {10.1109/TIT.1972.1054893}
}

@book{kadisonFundamentalsTheoryOperator1997,
  title = {Fundamentals of the {{Theory}} of {{Operator Algebras}}. {{Volume I}}},
  author = {Kadison, Richard V. and Ringrose, John R.},
  year = {1997},
  publisher = {American Mathematical Soc.},
  googlebooks = {Q3J6TV6euVYC},
  isbn = {978-0-8218-0819-1},
  langid = {english}
}

@misc{kalaiQuantumAdvantageAny2022,
  title = {Quantum {{Advantage}} from {{Any Non-Local Game}}},
  author = {Kalai, Yael and Lombardi, Alex and Vaikuntanathan, Vinod and Yang, Lisa},
  year = {2022},
  month = mar,
  number = {arXiv:2203.15877},
  eprint = {2203.15877},
  primaryclass = {quant-ph},
  publisher = {arXiv},
  doi = {10.48550/arXiv.2203.15877},
  archiveprefix = {arXiv}
}

@misc{kempeParallelRepetitionEntangled2011,
  title = {Parallel {{Repetition}} of {{Entangled Games}}},
  author = {Kempe, Julia and Vidick, Thomas},
  year = {2011},
  month = may,
  number = {arXiv:1012.4728},
  eprint = {1012.4728},
  primaryclass = {quant-ph},
  publisher = {arXiv},
  doi = {10.48550/arXiv.1012.4728},
  archiveprefix = {arXiv}
}

@article{kimSynchronousGameBinary2018,
  title = {A Synchronous Game for Binary Constraint Systems},
  author = {Kim, Se-Jin and Paulsen, Vern I. and Schafhauser, Christopher},
  year = {2018},
  month = mar,
  journal = {Journal of Mathematical Physics},
  volume = {59},
  number = {3},
  eprint = {1707.01016},
  primaryclass = {quant-ph},
  pages = {032201},
  issn = {0022-2488, 1089-7658},
  doi = {10.1063/1.4996867},
  archiveprefix = {arXiv}
}

@misc{kulpeBoundQuantumValue2025,
  title = {A Bound on the Quantum Value of All Compiled Nonlocal Games},
  author = {Kulpe, Alexander and Malavolta, Giulio and Paddock, Connor and Schmidt, Simon and Walter, Michael},
  year = {2025},
  month = jul,
  number = {arXiv:2408.06711},
  eprint = {2408.06711},
  primaryclass = {quant-ph},
  publisher = {arXiv},
  doi = {10.48550/arXiv.2408.06711},
  archiveprefix = {arXiv}
}

@misc{linTracialEmbeddableStrategies2024,
  title = {Tracial Embeddable Strategies: {{Lifting MIP}}* Tricks to {{MIPco}}},
  shorttitle = {Tracial Embeddable Strategies},
  author = {Lin, Junqiao},
  year = {2024},
  month = jan,
  number = {arXiv:2304.01940},
  eprint = {2304.01940},
  primaryclass = {quant-ph},
  publisher = {arXiv},
  doi = {10.48550/arXiv.2304.01940},
  archiveprefix = {arXiv}
}

@inproceedings{mastelTwoProverPerfect2024a,
  title = {Two {{Prover Perfect Zero Knowledge}} for {{MIP}}*},
  booktitle = {Proceedings of the 56th {{Annual ACM Symposium}} on {{Theory}} of {{Computing}}},
  author = {Mastel, Kieran and Slofstra, William},
  year = {2024},
  month = jun,
  series = {{{STOC}} 2024},
  pages = {991--1002},
  publisher = {Association for Computing Machinery},
  address = {New York, NY, USA},
  doi = {10.1145/3618260.3649702},
  isbn = {9798400703836}
}

@misc{mehtaPositivityUndecidableTensor2023,
  title = {Positivity Is Undecidable in Tensor Products of Free Algebras},
  author = {Mehta, Arthur and Slofstra, William and Zhao, Yuming},
  year = {2023},
  month = dec,
  number = {arXiv:2312.05617},
  eprint = {2312.05617},
  primaryclass = {math},
  publisher = {arXiv},
  doi = {10.48550/arXiv.2312.05617},
  archiveprefix = {arXiv}
}

@article{merminSimpleUnifiedForm1990,
  title = {Simple Unified Form for the Major No-Hidden-Variables Theorems},
  author = {Mermin, N. David},
  year = {1990},
  month = dec,
  journal = {Physical Review Letters},
  volume = {65},
  number = {27},
  pages = {3373--3376},
  publisher = {American Physical Society},
  doi = {10.1103/PhysRevLett.65.3373}
}

@misc{mousaviComplexityZeroGap2020,
  title = {On the Complexity of Zero Gap {{MIP}}*},
  author = {Mousavi, Hamoon and Nezhadi, Seyed Sajjad and Yuen, Henry},
  year = {2020},
  month = apr,
  number = {arXiv:2002.10490},
  eprint = {2002.10490},
  primaryclass = {quant-ph},
  publisher = {arXiv},
  doi = {10.48550/arXiv.2002.10490},
  archiveprefix = {arXiv}
}

@inproceedings{mousaviNonlocalGamesCompression2022,
  title = {Nonlocal Games, Compression Theorems, and the Arithmetical Hierarchy},
  booktitle = {Proceedings of the 54th {{Annual ACM SIGACT Symposium}} on {{Theory}} of {{Computing}}},
  author = {Mousavi, Hamoon and Nezhadi, Seyed Sajjad and Yuen, Henry},
  year = {2022},
  month = jun,
  series = {{{STOC}} 2022},
  pages = {1--11},
  publisher = {Association for Computing Machinery},
  address = {New York, NY, USA},
  doi = {10.1145/3519935.3519949},
  isbn = {978-1-4503-9264-8}
}

@book{mullenHandbookFiniteFields2013,
  title = {Handbook of {{Finite Fields}}},
  author = {Mullen, Gary L. and Panario, Daniel},
  year = {2013},
  month = jun,
  edition = {0},
  publisher = {{Chapman and Hall/CRC}},
  doi = {10.1201/b15006},
  isbn = {978-0-429-10519-7},
  langid = {english}
}

@misc{natarajanNEEXPMIP2019a,
  title = {{{NEEXP}} in {{MIP}}*},
  author = {Natarajan, Anand and Wright, John},
  year = {2019},
  month = sep,
  number = {arXiv:1904.05870},
  eprint = {1904.05870},
  primaryclass = {quant-ph},
  publisher = {arXiv},
  archiveprefix = {arXiv}
}

@article{natarajanTwoplayerEntangledGames2018,
  title = {Two-Player Entangled Games Are {{NP-hard}}},
  author = {Natarajan, Anand and Vidick, Thomas},
  year = {2018},
  eprint = {1710.03062},
  primaryclass = {quant-ph},
  pages = {18 pages},
  doi = {10.4230/LIPIcs.CCC.2018.20},
  archiveprefix = {arXiv}
}

@article{navascuesConvergentHierarchySemidefinite2008b,
  title = {A Convergent Hierarchy of Semidefinite Programs Characterizing the Set of Quantum Correlations},
  author = {Navascues, Miguel and Pironio, Stefano and Acin, Antonio},
  year = {2008},
  month = jul,
  journal = {New Journal of Physics},
  volume = {10},
  number = {7},
  eprint = {0803.4290},
  primaryclass = {quant-ph},
  pages = {073013},
  issn = {1367-2630},
  doi = {10.1088/1367-2630/10/7/073013},
  archiveprefix = {arXiv}
}

@article{nezhadiRecursiveCompressionMethod2025,
  title = {The Recursive Compression Method for Proving Undecidability Results},
  author = {Nezhadi, Seyed Sajjad and Marks, Andrew and Yuen, Henry},
  year = {2025},
  journal = {In preparation}
}

@book{ohyaQuantumEntropyIts2004,
  title = {Quantum {{Entropy}} and {{Its Use}}},
  author = {Ohya, M. and Petz, Denes},
  year = {2004},
  month = mar,
  publisher = {Springer Science \& Business Media},
  googlebooks = {r2ullNVyESQC},
  isbn = {978-3-540-20806-8},
  langid = {english}
}

@misc{ozawaConnesEmbeddingConjecture2013a,
  title = {About the {{Connes Embedding Conjecture---Algebraic}} Approaches---},
  author = {Ozawa, Narutaka},
  year = {2013},
  month = feb,
  number = {arXiv:1212.1700},
  eprint = {1212.1700},
  primaryclass = {math},
  publisher = {arXiv},
  doi = {10.48550/arXiv.1212.1700},
  archiveprefix = {arXiv}
}

@misc{ozawaQWEPConjecture2004,
  title = {About the {{QWEP}} Conjecture},
  author = {Ozawa, Narutaka},
  year = {2004},
  month = may,
  number = {arXiv:math/0306067},
  eprint = {math/0306067},
  institution = {arXiv},
  archiveprefix = {arXiv}
}

@misc{paddockSatisfiabilityProblemsAlgebras2025,
  title = {Satisfiability Problems and Algebras of Boolean Constraint System Games},
  author = {Paddock, Connor and Slofstra, William},
  year = {2025},
  month = jan,
  number = {arXiv:2310.07901},
  eprint = {2310.07901},
  primaryclass = {quant-ph},
  publisher = {arXiv},
  doi = {10.48550/arXiv.2310.07901},
  archiveprefix = {arXiv}
}

@article{paulsenEstimatingQuantumChromatic2016,
  title = {Estimating Quantum Chromatic Numbers},
  author = {Paulsen, Vern I. and Severini, Simone and Stahlke, Daniel and Todorov, Ivan G. and Winter, Andreas},
  year = {2016},
  month = mar,
  journal = {Journal of Functional Analysis},
  volume = {270},
  number = {6},
  pages = {2188--2222},
  issn = {00221236},
  doi = {10.1016/j.jfa.2016.01.010},
  langid = {english}
}

@article{peresIncompatibleResultsQuantum1990,
  title = {Incompatible Results of Quantum Measurements},
  author = {Peres, Asher},
  year = {1990},
  journal = {Phys. Lett. A},
  volume = {151},
  pages = {107--108},
  doi = {10.1016/0375-9601(90)90172-K}
}

@inproceedings{razParallelRepetitionTheorem1995a,
  title = {A Parallel Repetition Theorem},
  booktitle = {Proceedings of the Twenty-Seventh Annual {{ACM}} Symposium on {{Theory}} of Computing},
  author = {Raz, Ran},
  year = {1995},
  month = may,
  series = {{{STOC}} '95},
  pages = {447--456},
  publisher = {Association for Computing Machinery},
  address = {New York, NY, USA},
  doi = {10.1145/225058.225181},
  isbn = {978-0-89791-718-6}
}

@article{schwartzFastProbabilisticAlgorithms1980,
  title = {Fast {{Probabilistic Algorithms}} for {{Verification}} of {{Polynomial Identities}}},
  author = {Schwartz, J. T.},
  year = {1980},
  month = oct,
  journal = {J. ACM},
  volume = {27},
  number = {4},
  pages = {701--717},
  issn = {0004-5411},
  doi = {10.1145/322217.322225}
}

@book{sipserIntroductionTheoryComputation2006,
  title = {Introduction to the {{Theory}} of {{Computation}}},
  author = {Sipser, Michael},
  year = {2006},
  publisher = {Thomson Course Technology},
  googlebooks = {VJ1mQgAACAAJ},
  isbn = {978-0-619-21764-8},
  langid = {english}
}

@article{slofstraSetQuantumCorrelation2019,
  title = {The Set of Quantum Correlation Is Not Closed},
  author = {Slofstra, William},
  year = {2019/ed},
  journal = {Forum of Mathematics, Pi},
  volume = {7},
  pages = {e1},
  publisher = {Cambridge University Press},
  issn = {2050-5086},
  doi = {10.1017/fmp.2018.3},
  langid = {english}
}

@article{slofstraTsirelsonProblemEmbedding2019,
  title = {Tsirelson's Problem and an Embedding Theorem for Groups Arising from Non-Local Games},
  author = {Slofstra, William},
  year = {2019},
  month = sep,
  journal = {Journal of the American Mathematical Society},
  volume = {33},
  number = {1},
  eprint = {1606.03140},
  primaryclass = {math-ph, physics:quant-ph},
  pages = {1--56},
  issn = {0894-0347, 1088-6834},
  doi = {10.1090/jams/929},
  archiveprefix = {arXiv}
}

@book{takesakiTheoryOperatorAlgebras2001,
  title = {Theory of {{Operator Algebras I}}},
  author = {Takesaki, M.},
  year = {2001},
  month = nov,
  publisher = {Springer Science \& Business Media},
  googlebooks = {dTnq4hjjtgMC},
  isbn = {978-3-540-42248-8},
  langid = {english}
}

@book{takesakiTomitaTheoryModular1970a,
  title = {Tomita's {{Theory}} of {{Modular Hilbert Algebras}} and Its {{Applications}}},
  author = {Takesaki, M.},
  year = {1970},
  series = {Lecture {{Notes}} in {{Mathematics}}},
  volume = {128},
  publisher = {Springer},
  address = {Berlin, Heidelberg},
  doi = {10.1007/BFb0065832},
  isbn = {978-3-540-04917-3 978-3-540-36267-8}
}

@article{tsirelsonQuantumAnaloguesBell1987,
  title = {Quantum Analogues of the of the Bell Inequality},
  author = {Tsirelson, Boris},
  year = {1987},
  month = feb,
  journal = {Journal of Soviet Mathematics},
  volume = {36},
  pages = {557--570}
}

@article{vidickAlmostSynchronousQuantum2022,
  title = {Almost Synchronous Quantum Correlations},
  author = {Vidick, Thomas},
  year = {2022},
  month = feb,
  journal = {Journal of Mathematical Physics},
  volume = {63},
  number = {2},
  eprint = {2103.02468},
  pages = {022201},
  issn = {0022-2488, 1089-7658},
  doi = {10.1063/5.0056512},
  archiveprefix = {arXiv},
  langid = {english}
}

@article{vidickQuantumProofs2016,
  title = {Quantum {{Proofs}}},
  author = {Vidick, Thomas and Watrous, John},
  year = {2016},
  journal = {Foundations and Trends{\textregistered} in Theoretical Computer Science},
  volume = {11},
  number = {1-2},
  eprint = {1610.01664},
  primaryclass = {quant-ph},
  pages = {1--215},
  issn = {1551-305X, 1551-3068},
  doi = {10.1561/0400000068},
  archiveprefix = {arXiv}
}

@article{wattsNoncommutativeNullstellensatzePerfect2023,
  title = {Noncommutative {{Nullstellens{\"a}tze}} and {{Perfect Games}}},
  author = {Watts, Adam Bene and Helton, John William and Klep, Igor},
  year = {2023},
  month = jul,
  journal = {Annales Henri Poincar{\'e}},
  volume = {24},
  number = {7},
  eprint = {2111.14928},
  primaryclass = {quant-ph},
  pages = {2183--2239},
  issn = {1424-0637, 1424-0661},
  doi = {10.1007/s00023-022-01262-1},
  archiveprefix = {arXiv}
}

@book{wildeQuantumInformationTheory2013,
  title = {Quantum {{Information Theory}}},
  author = {Wilde, Mark M.},
  year = {2013},
  publisher = {Cambridge University Press},
  address = {Cambridge},
  doi = {10.1017/CBO9781139525343}
}

@article{yuenParallelRepetitionTheorem2016,
  title = {A {{Parallel Repetition Theorem}} for {{All Entangled Games}}},
  author = {Yuen, Henry},
  editor = {Chatzigiannakis, Ioannis and Mitzenmacher, Michael and Rabani, Yuval and Sangiorgi, Davide},
  year = {2016},
  journal = {LIPIcs, Volume 55, ICALP 2016},
  volume = {55},
  pages = {77:1-77:13},
  publisher = {Schloss Dagstuhl -- Leibniz-Zentrum f{\"u}r Informatik},
  issn = {1868-8969},
  doi = {10.4230/LIPICS.ICALP.2016.77},
  copyright = {Creative Commons Attribution 3.0 Unported license, info:eu-repo/semantics/openAccess},
  isbn = {9783959770132},
  langid = {english}
}

@inproceedings{zippelProbabilisticAlgorithmsSparse1979,
  title = {Probabilistic Algorithms for Sparse Polynomials},
  booktitle = {Symbolic and {{Algebraic Computation}}},
  author = {Zippel, Richard},
  editor = {Ng, Edward W.},
  year = {1979},
  pages = {216--226},
  publisher = {Springer},
  address = {Berlin, Heidelberg},
  doi = {10.1007/3-540-09519-5_73},
  isbn = {978-3-540-35128-3},
  langid = {english}
}

@misc{arulseelanUniversalTheoryLocally2025,
  title = {The {{Universal Theory}} of {{Locally Universal Tracial}} von {{Neumann Algebras}} Is Not {{Computable}}},
  author = {Arulseelan, Jananan and Manzoor, Aareyan},
  year = {2025},
  month = aug,
  number = {arXiv:2508.21709},
  eprint = {2508.21709},
  primaryclass = {math},
  publisher = {arXiv},
  doi = {10.48550/arXiv.2508.21709},
  archiveprefix = {arXiv}
}

@misc{goldbringDefinabilityCtensorNorms2025,
  title = {On Definability of {{C}}*-Tensor Norms},
  author = {Goldbring, Isaac and Sinclair, Thomas},
  year = {2025},
  month = sep,
  number = {arXiv:2509.15086},
  eprint = {2509.15086},
  primaryclass = {math},
  publisher = {arXiv},
  doi = {10.48550/arXiv.2509.15086},
  archiveprefix = {arXiv}
}

@misc{goldbringKirchbergsEmbeddingProblem2015,
  title = {On {{Kirchberg}}'s {{Embedding Problem}}},
  author = {Goldbring, Isaac and Sinclair, Thomas},
  year = {2015},
  month = feb,
  number = {arXiv:1404.1861},
  eprint = {1404.1861},
  primaryclass = {math},
  publisher = {arXiv},
  doi = {10.48550/arXiv.1404.1861},
  archiveprefix = {arXiv}
}

@misc{manzoorInvariantRandomSubgroups2025,
  title = {Invariant {{Random Subgroups}}, {{Soficity}}, and {{L{\"u}ck}}'s Determinant Conjecture},
  author = {Manzoor, Aareyan},
  year = {2025},
  month = sep,
  number = {arXiv:2508.15154},
  eprint = {2508.15154},
  primaryclass = {math},
  publisher = {arXiv},
  doi = {10.48550/arXiv.2508.15154},
  archiveprefix = {arXiv}
}

@misc{manzoorThereEquivalenceRelation2025,
  title = {There {{Is An Equivalence Relation Whose}} von {{Neumann Algebra Is Not Connes Embeddable}}},
  author = {Manzoor, Aareyan},
  year = {2025},
  month = feb,
  number = {arXiv:2502.06697},
  eprint = {2502.06697},
  primaryclass = {math},
  publisher = {arXiv},
  doi = {10.48550/arXiv.2502.06697},
  archiveprefix = {arXiv}
}

@article{fritzCanYouCompute2014a,
  title = {Can You Compute the Operator Norm?},
  author = {Fritz, Tobias and Netzer, Tim and Thom, Andreas},
  year = {2014},
  month = aug,
  journal = {Proceedings of the American Mathematical Society},
  volume = {142},
  number = {12},
  eprint = {1207.0975},
  primaryclass = {math},
  pages = {4265--4276},
  issn = {0002-9939, 1088-6826},
  doi = {10.1090/S0002-9939-2014-12170-8},
  archiveprefix = {arXiv}
}

@misc{mancinskaGappreservingReductionsREcompleteness2025,
  title = {Gap-Preserving Reductions and {{RE-completeness}} of Independent Set Games},
  author = {Man{\v c}inska, Laura and Spaas, Pieter and Spirig, Taro and Vernooij, Matthijs},
  year = {2025},
  month = aug,
  number = {arXiv:2505.05253},
  eprint = {2505.05253},
  primaryclass = {quant-ph},
  publisher = {arXiv},
  doi = {10.48550/arXiv.2505.05253},
  archiveprefix = {arXiv}
}

@misc{culfApproximationAlgorithmsNoncommutative2024,
  title = {Approximation Algorithms for Noncommutative {{CSPs}}},
  author = {Culf, Eric and Mousavi, Hamoon and Spirig, Taro},
  year = {2024},
  month = sep,
  number = {arXiv:2312.16765},
  eprint = {2312.16765},
  primaryclass = {quant-ph},
  publisher = {arXiv},
  doi = {10.48550/arXiv.2312.16765},
  archiveprefix = {arXiv}
}

@misc{culfREcompletenessEntangledConstraint2025,
  title = {{{RE-completeness}} of Entangled Constraint Satisfaction Problems},
  author = {Culf, Eric and Mastel, Kieran},
  year = {2025},
  month = feb,
  number = {arXiv:2410.21223},
  eprint = {2410.21223},
  primaryclass = {quant-ph},
  publisher = {arXiv},
  doi = {10.48550/arXiv.2410.21223},
  archiveprefix = {arXiv}
}

@misc{mousaviQuantumUniqueGames2024,
  title = {A {{Quantum Unique Games Conjecture}}},
  author = {Mousavi, Hamoon and Spirig, Taro},
  year = {2024},
  month = sep,
  number = {arXiv:2409.20028},
  eprint = {2409.20028},
  primaryclass = {quant-ph},
  publisher = {arXiv},
  doi = {10.48550/arXiv.2409.20028},
  archiveprefix = {arXiv}
}

@inproceedings{khotPowerUnique2prover2002,
  title = {On the Power of Unique 2-Prover 1-Round Games},
  booktitle = {Proceedings 17th {{IEEE Annual Conference}} on {{Computational Complexity}}},
  author = {Khot, S.},
  year = {2002},
  month = may,
  pages = {25-},
  issn = {1093-0159},
  doi = {10.1109/CCC.2002.1004334}
}

@misc{kempeUniqueGamesEntangled2009,
  title = {Unique {{Games}} with {{Entangled Provers}} Are {{Easy}}},
  author = {Kempe, Julia and Regev, Oded and Toner, Ben},
  year = {2009},
  month = oct,
  number = {arXiv:0710.0655},
  eprint = {0710.0655},
  primaryclass = {quant-ph},
  publisher = {arXiv},
  doi = {10.48550/arXiv.0710.0655},
  archiveprefix = {arXiv}
}
\newpage
\appendix

\section{A parallel repetition theorem for the commuting operator model} \label{sec:parallelrepappendix}
The result from this appendix originates from a earlier joint collaboration between William Slofstra and Henry Yuen. The goal of this appendix is to show~\Cref{thm:anchoringparallelrep}. The anchored parallel repetition was originally shown in~\cite{bavarianAnchoredParallelRepetition2021} in the tensor product model. The goal of this appendix is to define some quantum informatics tools for Tracially embeddable strategies, and show how this can be used to show a parallel repetition for the commuting operator model. We remark that outside of the quantum informatics tools, the proof for the anchored parallel repetition theorem in the commuting operator model is near identical to the tensor product case, and we choose to include a version of this proof for completeness.  

Recall from~\Cref{sec:prenonlocalgames}, given a non-local games $\cG$, we define the $r$-fold parallel repetition of a game $\cG^{\otimes r} = (\cX^{r}, \cA^r, \mu^{r}, D^r)$ as the game with the following question distribution and validation function
\begin{itemize}
	\item $\mu^n((x_0, \cdots, x_{r-1}), (y_0, \cdots y_{r-1})) = \prod_{i = 0}^{r} \mu(x_i, y_i)$.
	\item  $\D^r\left((x_0, \cdots, x_{r-1}), (y_0, \cdots y_{r-1}), (a_0, \cdots , a_{r-1}),(b_0, \cdots , b_{r-1})\right) =\prod_{i = 0}^{r} D(x_i,. y_i, a_i,b_i)$
\end{itemize}
Furthermore, recall from~\Cref{def:anchortransformation} that given a game $\cG = (\cX, \cA, \mu, D)$, we use $\cG_{\bot} = (\cX_{\bot}, \cA_{\bot} , \mu^{\bot}, D^{\bot})$ to denote the anchoring transformation. We restate\Cref{thm:anchoringparallelrep} below for convenience.
\begin{theorem}[Anchored Parallel repetition theorem] 
	There exist a universal constant $c^{\text{para}}$ such that, for any models $t \in \{*,co\}$, non-local games $\cG = (\cX, \cA, \mu, D)$ with $\omega^{t} \leq 1 - \eps$, and $r \in \bN$, then
	\begin{equation*}
		\omega^{t}(\cG_{\bot}^{\otimes r}) \leq \frac{16}{\eps} \cdot \exp{\left(\frac{- c^{\text{para}} \eps^{17} r }{\log(|\cA | + 1)}\right)}.
	\end{equation*}
\end{theorem}
As mentioned above, the main bottleneck for extending a parallel repetition theorem to the commuting operator model is the lack for quantum information-theoretic tools. The informational-theoretic tools used in many parallel repetition theorem (see e.g.~\cite{jainParallelRepetitionTheorem2014,yuenParallelRepetitionTheorem2016}) are define for finite-dimensional strategies, and does not necessarily translate to the infinite-dimensional setting. In this appendix, we also assume that all $\cG = (\cX, \cA, \mu, V)$ is an anchored game, or a game which arises from the anchor transformation given in~\Cref{def:anchortransformation}. 

We organize the appendix as follows: in the first subsection, we introduce the formulation of relative entropy base on the work of~\cite{arakiRelativeEntropyStates1977} between two normal state acting on tracial von Neuman algebras. In the second subsection, we introduce the parallel repetition theorem and the main step for showing the parallel repetition theorem. In the third subsection, we show an analogue of~\cite[Proposition 5.1]{bavarianAnchoredParallelRepetition2021} for the commuting operator model. We remark that although the proof is based on~\cite{bavarianAnchoredParallelRepetition2021}, some notation differs from the original parallel repetition paper for the sake of clarity or consistency with the rest of the paper.

\subsection{Quantum information theory for tracial von Neumann algebras}\label{sec:qinfo}

In this subsection, we give a brief introduction for relative entropy defined on (tracial) von Neumann under the standard representation given in~\cite{arakiRelativeEntropyStates1977}. 

\subsubsection{Additional background on von Neumann algebra}

We start this subsection by first recalling some additional von Neumann algebra needed for this appendix. Let $\alicealg_1 \subseteq B(\cH_1)$ and $\alicealg_2 \subseteq B(\cH_2)$ be two von Neumann algebras. The von Neumann algebra tensor product $\alicealg_1 \otimes \alicealg_2$ is the weak operator closure of the span of $\{a \otimes b : a \in \alicealg_1, b \in \alicealg_2\}$ in $B(\cH_1 \otimes \cH_2)$. For two tracial von Neumann algebra $(\alicealg_1, \tau_1)$ and $(\alicealg_2, \tau_2)$, the von Neumann algebra $\alicealg_1 \otimes \alicealg_2$ remains a tracial von Neumann algebra with the trace being $\tau_1 \otimes \tau_2$. For a state \gls{tensorstate} on $\alicealg_1 \tensor \alicealg_2$, we let $\psi^{\alicealg_1}$ denote the restriction of $\psi^{\alicealg_1 \alicealg_2}$ to $\alicealg_1 = \alicealg_1 \otimes \cI$, so $\psi^{\alicealg_1}(a) = \psi^{\alicealg_1 \alicealg_2}(a \otimes \cI)$ for all $a \in \alicealg_1$. 

We recall the following lemma from~\cite{linTracialEmbeddableStrategies2024} about the existence of left and right inverse. 
\begin{lemma}[Existence of a left and right inverse]\ \label{lem:inversetrick1}
	Let $A, B \in \alicealg^{+}$ such that $B \leq A$, then there exists some element $R \in \alicealg$ with $R^* R \leq \cI$ and $A^{\frac{1}{2}} R  = B$. Furthermore, there exists some element $L \in \alicealg$ with $L^* L \leq \cI$ and $L A^{\frac{1}{2}} = B$.
\end{lemma}

We recall the following theorem about projectors in a tracial von Neumann algebra. 
\begin{proposition}[Corollary 2.8 of~\cite{takesakiTheoryOperatorAlgebras2001}]  \label{prop:projectorequ}
	Let $(\alicealg, \tau)$ be a tracial von Neumann algebra and $P, Q \in \alicealg$ be two projectors, then the following two conditions are equivalent: 
	\begin{itemize}
		\item $\tau(P) = \tau(Q)$.
		\item $P$ is equivalent to $Q$.
	\end{itemize}
\end{proposition}

\subsubsection{Tomita-Takesaki construction and the positive cone}

In order to discuss the relative entropy construction, we need to first give a brief summary to the Tomita-Takesaki construction~\cite{takesakiTomitaTheoryModular1970a}. Let $(\alicealg, \ket{\tau})$ be a tracial von Neumann algebra in standard form, and let $\mathfrak{S}_{\ket{\tau}}: \cH \rightarrow \cH$ be the antilinear (and potentially unbounded) map defined by 
\begin{equation} \label{eq:TT_map}
	\mathfrak{S}_{\ket{\tau}}: a \ket{\tau} \rightarrow a^* \ket{\tau}, a \in \alicealg.
\end{equation}
Since $\ket{\tau}$ is cyclic, the above map is well-defined on the dense subset $\alicealg \ket{\tau}$ of $\cH$. By taking the polar decomposition, we can write $\mathfrak{S}_{\ket{\tau}}$ as
\begin{equation} \label{eq:TT_map_2}
	\mathfrak{S}_{\ket{\tau}} =  \mathfrak{J}_{\ket{\tau}}\mathbf{\Delta}_{\ket{\tau}}^{\frac{1}{2}},
\end{equation}
for some antilinear isometry $\mathfrak{J}_{\ket{\tau}}$, known as the \textit{modular conjugation}, and some positive operator $\mathbf{\Delta}_{\ket{\tau}}$, known as the \textit{modular operator}. Remarkably the modular operator can be used to construct the commutant of $\alicealg$ algebraically, as $\mathfrak{J} \alicealg \mathfrak{J} = \alicealg'$. Using the modular operator, define the \textit{canonical positive cone} associated with $(\alicealg, \ket{\tau})$ as
\begin{equation} \label{eq:positive_cone}
	\text{\gls{Postcone}} = \overline{\{\mathbf{\Delta}_{\ket{\cyclicsepvector}}^{\frac{1}{4}} A \ket{\cyclicsepvector}, A \in \alicealg^{+} \}},
\end{equation}
where the closure is in the weak topology~\cite{arakiPropertiesModularConjugation1974b}. This is a pointed closed self-dual convex cone~\cite[Theorem 4 part 1]{arakiPropertiesModularConjugation1974b}. The following proposition gives a bijection between states on the von Neumann algebra $\alicealg$ and vectors on the canonical positive cone. 
\begin{proposition}[Theorem 6 of~\cite{arakiPropertiesModularConjugation1974b}] \label{prop:state_on_positive_cone}
	For every positive normal linear functional $\psi$ on a von Neumann algebra $\alicealg \subseteq \bofh$, there exists a unique $\ket{\psi} \in \positivecone$ such that $\psi(a) = \braket{\psi|a|\psi}$.
\end{proposition}

For a state $\psi$, we will use $\ket{\psi}$ to denote the unique vector in the positive cone above.  For a state $\psi$ acting on $\alicealg$, we use $\supp^{\alicealg}{\psi}$  to denote the complement of the minimal projector $P \in \alicealg$ such that $\psi(P) = 0$. For any vector $\psi$, we have $\supp(\psi) \ket{\psi} = \ket{\psi}$. Intuitively, we can think of the corresponding vector in the positive cone as the ``purification" of the linear functional on $\alicealg$. Indeed, in the finite-dimensional case this vector is the familiar purification from quantum information theory:

\begin{example} \label{exam:StandardformFD}
	Let $\alicealg = \bM_n(\bC) \tensor \Id_n \subseteq \bM_{n^2}(\bC)$, and let $\ket{i}, i=0,\ldots,n-1$ denote the standard basis vectors for $\C^n$. Recall, the GNS representation using the trace $\Tr(\cdot)$ maps $\bM_n(\bC)$ to $\bM_n(\bC) \tensor \Id_n \in \bM_{n^2}(\bC)$, with the linear functional $\Tr(\cdot)$ gets mapped to the maximally entangled state $\ket{\tau} = \frac{1}{\sqrt{n}} \sum_{i=0}^n \ket{ii}$. In this case, we see that the vector state $\ket{\tau}$ is a cyclic and separating vector for $\alicealg$. As per~\Cref{eq:TT_map}, 
	\begin{equation*}
		S_{\ket{\tau}} \ket{i}\ket{j}  = S_{\ket{\tau}} (\ket{i}\bra{j} \tensor \cI) \ket{\tau} =  (\ket{j}\bra{i} \tensor \cI) \ket{\tau}= \ket{j}\ket{i},
	\end{equation*}
	and hence $S_{\ket{\tau}} = C \circ \sum_i \sum_j \frac{\lambda_i}{\lambda_j} \ket{j,i}\bra{i,j}$, where $C$ is conjugation in the standard basis. Taking the polar decomposition, the modular operator is
	\begin{equation*}
		\mathbf{\Delta}_{\ket{\tau}} = (S_{\ket{\tau}})^* S_{\ket{\tau}} = \sum_i \sum_j  \ket{i,j}\bra{i,j} = \cI_4,
	\end{equation*}
	and the positive cone associated with $\ket{\tau}$ is
	\begin{equation*}
		\positivecone =\left\{  (A \tensor \cI) \ket{\tau},  A \in \bM_{n}(\bC)^{+} \right\}.
	\end{equation*}
	If $\phi$ is a linear functional on $\bM_n(\bC)$, then there is a unique positive matrix $\sigma$ such that $\phi(A) = \Tr(A\sigma) = \Tr(\sigma^{1/2} A \sigma^{1/2})$, and hence
	\begin{equation}
		\phi(A) = \braket{\tau|\sigma^{\frac{1}{2}}A\sigma^{\frac{1}{2}} \tensor \cI |\tau}.
	\end{equation}
	Thus the vector in $\positivecone$ associated with $\phi$ is $(\sigma^{\frac{1}{2}} \tensor \cI)\ket{\tau}$, which is commonly used in quantum information as a purification of the density matrix $\sigma$.
\end{example}

We end this subsection by reviewing some properties related to the positive cone $\positivecone$. The following proposition shows that changing the cyclic and separating vector $\ket{\tau}$ only changes the positive cone by a unitary in the commutant:
\begin{proposition}[Theorem 7 part 6 of~\cite{arakiPropertiesModularConjugation1974b}] \label{prop:positive_cone_unitary}
	Let $\alicealg \subseteq \bofh$ be a tracial von Neumann algebra in standard form and let $\ket{\tau_1}$ and $\ket{\tau_2}$ be two cyclic and separating vectors for $\alicealg$. Then there is a unitary $U \in \alicealg'$ such that for all normal states 
	$\psi$ on $\alicealg$, if $\ket{\psi_1}$ and $\ket{\psi_2}$ are the vectors associated to $\psi$ in the positive cones of $\ket{\tau_1}$ and $\ket{\tau_2}$ respectively, then $\ket{\psi_1} = U \ket{\psi_2}$.
\end{proposition}
The following proposition shows that every vector on $\cH$ can be related to some vector within $\positivecone$ via a partial isometry. 
\begin{proposition}\label{prop:positive_cone_polar_decomposition}
	Let $(\alicealg, \ket{\cyclicsepvector}) \subseteq \bofh$ be a tracial von Neumann algebra in standard form. For any $\ket{\psi} \in \cal{H}$, there exist a unitary $U \in \alicealg'$ and a unique $\ket{\psi^{+}} \in \positivecone$ such that $\ket{\psi} = U \ket{\psi^{+}}$.
\end{proposition}
\begin{proof}
	By~\cite[Theorem 7 part 5]{arakiPropertiesModularConjugation1974b}, there exist a partial isometry $V \in \alicealg'$ and a unique $\ket{\psi^{+}} \in \positivecone$ such that $\ket{\psi} = V \ket{\psi^{+}}$, and
	\begin{equation*}
		VV^* = \supp^{\alicealg'}(\psi), \quad V^{*}V = \supp^{\alicealg'}(\psi^{+}).
	\end{equation*}
	We wish to extend $V$ into a Unitary. We first see that $\tau(VV^*) = \tau(V^*V)$, and hence $\tau(\cI - VV^*) = \tau(\I - V^*V)$. By~\Cref{prop:projectorequ}, there exist a partial isometry $W$ such that 
	\begin{equation*}
		WW^* = \cI - \supp^{\alicealg'}(\psi), \quad W^{*}W =  \cI - \supp^{\alicealg'}(\psi^{+}).
	\end{equation*}	
	Take $U = V + W$, we see that
	\begin{align*}
		U^* U &= (V^* + W^*)(V + W) =  \cI  + W^*V + V^* W \\
		&= \cI  + W^*  \left(\cI - \supp^{\alicealg'}(\psi) \right)\supp^{\alicealg'}(\psi)  V + V^*  \supp^{\alicealg'}(\psi^{+}) \left(\cI - \supp^{\alicealg'}(\psi^{+}) \right)  W \\
		&= \cI,
	\end{align*}
	and by a similar calculation, we have $U U^* = \cI$ showing that $U$ is indeed a unitary. Furthermore
	\begin{equation*}
		U \ket{\psi^{+}} = V \ket{\psi^{+}} + W \ket{\psi^{+}} = \ket{\psi}  + W\left( \cI- \cI - \supp^{\alicealg'}(\psi^{+}) \right) \ket{\psi^{+}} = \ket{\psi},
	\end{equation*}
	and hence the proposition follows. 
\end{proof}

The Araki-Powers-Stormer inequality relates the norm distance between states to the distance between vectors in the positive cone:
\begin{proposition}[Araki-Powers-Stormer inequality, Theorem 4 part 8 of~\cite{arakiPropertiesModularConjugation1974b}] \label{prop:Araki_Powers_Stormer}
	Let $\alicealg \subseteq \bofh$ be a tracial von Neumann algebra in standard form, and let $\ket{\psi_1}, \ket{\psi_2} \in \positivecone$. Then
	\begin{equation*}
		\|\psi_1 - \psi_2\|_{\alicealg} \geq \| \ket{\psi_1} - \ket{\psi_2}\|^2 . 
	\end{equation*}
\end{proposition}
We'll use the Araki-Powers-Stormer inequality in the following form:
\begin{proposition}	\label{prop:uhlmann}
	Let $(\alicealg, \ket{\cyclicsepvector})\subseteq \bofh$ be a von Neumann algebra in the standard form. For any unit vectors $\ket{\psi_1}, \ket{\psi_2} \in \cal{H}$, there is a unitary operator $U \in \alicealg'$ such that
	\begin{equation} \label{eq:uhlmann}
		\bra{\psi_1} U \ket{\psi_2} \geq 1 - \frac{1}{2} \| \psi_1 - \psi_2 \|_{\alicealg}.
	\end{equation}
\end{proposition}
\begin{proof}
	By~\Cref{prop:positive_cone_polar_decomposition}, there are vectors $\ket{\psi_1^{+}}, \ket{\psi_2^{+}} \in \positivecone$ and unitaries $U_{\psi_2}, U_{\psi_1} \in \alicealg'$ such that $U_{\psi_i} \ket{\psi_i^{+}} = \ket{\psi_i}$, $i=1,2$. Since $\positivecone$ is self-dual, $\braket{\psi_1^+|\psi_2^+} \geq 0$, so 
	\begin{align*}
		\| \ket{\psi_1^+} - \ket{\psi_2^+} \| &= 2 - 2 \braket{\psi_1^+|\psi_2^+} = 2 - 2 \braket{\psi_1|U_{\psi_1}^* U_{\psi_2}|\psi_2},
	\end{align*}
	and hence the proposition follows from~\Cref{prop:Araki_Powers_Stormer} with $U = U_{\psi_1}^* U_{\psi_2}$.
\end{proof}
In the finite-dimensional case, we can take $V$ to be unitary in~\Cref{prop:uhlmann}. Indeed, let $\alicealg = \bM_n(\C) \otimes \cI \subseteq \bM_{n^2}(\C)$, and let $\ket{\psi_1}^{AB}$ and $\ket{\psi_2}^{AB}$ be two unit vectors in $\bC^{n^2}$ with reduced density matrices $\psi_i = \Tr_B(\ketbra{\psi_i}{\psi_i})$, $i=1,2$. By Uhlmann's theorem (see, e.g.~\cite[Theorem 9.2.1]{wildeQuantumInformationTheory2013}) there exist a unitary $U \in \bM_n(\bC)$ such that 
\begin{equation} \label{eq:ulhmannFD1}
	\braket{\psi_1| \cI \tensor U |\psi_2} = \| \sqrt{\psi_1}\sqrt{\psi_2} \|_1,
\end{equation}
where $\|\cdot\|_1$ is the matrix $1$-norm. By the 
Fuchs-van de Graaf inequality (see, e.g.~\cite[Theorem 9.3.1]{wildeQuantumInformationTheory2013}),
\begin{equation}\label{eq:ulhmannFD2}
	1 -  \| \sqrt{\psi_1}\sqrt{\psi_2} \|_1 \leq \frac{1}{2} ||\psi_1 - \psi_2||_1,
\end{equation}
and since 
\begin{equation*}
	\| \psi \|_1 = \sup\{ |\Tr(A^* \psi)\| : A \in  \bM_n(\bC) ,||A|| \leq 1 \} = \|\psi\|_{\bM_n(\C)},
\end{equation*}
\Cref{eq:uhlmann} follows from Equations \eqref{eq:ulhmannFD1} and \eqref{eq:ulhmannFD2}, with $V = \cI \otimes U \in \alicealg'$. 

\subsubsection{Relative Entropy}
In this subsection, we review the relative entropy between two positive normal linear functionals $\psi_1$ and $\psi_2$ on a tracial von Neumann algebra $(\alicealg, \ket{\tau})$ in standard form. We follow the construction in~\cite[Section 2]{arakiRelativeEntropyStates1977}. Let $\ket{\psi_1}$ and $\ket{\psi_2}$ be the vectors corresponding to $\psi_1$ and $\psi_2$ in the positive cone $\cH_{\ket{\tau}}^{+}$ of $\alicealg$. Similarly to~\Cref{eq:TT_map}, we define an anti-linear map $\mathfrak{S}_{\psi_1, \psi_2}^{\ket{\tau}} : \alicealg \ket{\psi_2}  \rightarrow \supp(\psi_1) \cH$ as
\begin{equation} \label{eq:Relative_modular_map}
	\mathfrak{S}_{\psi_1, \psi_2}^{\ket{\tau}}(a \ket{\psi_2}) =\supp(\psi_2) a^{*}\ket{\psi_1} \text{ for all } a \in \alicealg.
\end{equation}
Extend the above (potentially unbounded) map to $\alicealg \ket{\psi_2} \oplus (\alicealg \ket{\psi_2})^{\perp}$ by mapping $(\alicealg \ket{\psi_1})^{\perp}$ to zero. The extension is a well-defined antilinear operator whose domain of definition is dense in $\cH$. Hence, $\mathfrak{S}_{\psi_1, \psi_2}^{\ket{\tau}}$ has a polar decomposition
\begin{equation} 
	\mathfrak{S}_{\psi_1, \psi_2}^{\ket{\tau}} =  \mathfrak{S}_{\psi_1, \psi_2}^{\ket{\tau}}(\mathbf{\Delta}_{\psi_1, \psi_2}^{\ket{\tau}})^{\frac{1}{2}}
\end{equation}
where $\mathfrak{J}_{\psi_1, \psi_2}^{\ket{\tau}}$ is an antilinear isometry and $\mathbf{\Delta}_{\psi_1, \psi_2}^{\ket{\tau}} \in \alicealg$ is a positive operator. The relative entropy between $\psi_1$ and $\psi_2$ is defined to be 
\begin{equation} \label{eq:relative_ent}
	\D^{\ket{\tau}}(\psi_1 \| \psi_2) = \begin{cases}
		\int_{0}^{\infty} \log_2{\lambda} d \braket{\psi_1|E_{\lambda}|\psi_1} & \text{if } supp(\psi_1) \leq supp(\psi_2) \\
		+ \infty & \text{otherwise}
	\end{cases},
\end{equation}
where $\mathbf{\Delta}_{\psi_1, \psi_2}^{\ket{\tau}} = \int_{0}^{\infty} \lambda d E_{\lambda}^{\ket{\tau}}$ is the spectral decomposition of $\mathbf{\Delta}_{\psi_1, \psi_2}^{\ket{\tau}}$. 
By~\Cref{prop:positive_cone_unitary}, $\D^{\ket{\tau}}(\psi_1 \| \psi_2)$ is independent of $\ket{\tau}$, and thus we refer to this relative entropy as \gls{Relativeent}. To match the standard convention in quantum information, the order of arguments for $D$ is flipped from~\cite[Definition 3.1]{arakiRelativeEntropyStates1977}. We also use $\log = \log_2$ rather than $\log = \ln$, which changes our relative entropy by a factor of $\ln(2)$. Most propositions below are unchanged, but we do get a factor of $\ln(2)$ in~\Cref{prop:pinsker} below.

\begin{example} \label{exam:relativeentropy}
	Let $ \psi_1 (A) = \Tr(A \rho_1) $ and $\psi_2(A) = \Tr(A \rho_2)$ be two normal linear functionals on $\alicealg = \bM_n (\bC)$ such that 
	$supp(\psi_1) \leq supp(\psi_2)$. Recall from~\Cref{exam:StandardformFD} from that $\alicealg$ has standard form $\bM_n (\bC) \otimes \cI \subseteq \bM_{n^2} (\bC)$, and that $\ket{\tau} =  \frac{1}{\sqrt{n}} \sum_{i=1}^n\ket{ii}$ is a cyclic and separating vector for $\alicealg$. The vector corresponding to $\psi_i$ in the positive cone for $\ket{\tau}$ is $\ket{\psi_i} = 
	\rho_i^{\frac{1}{2}} \otimes \cI \ket{\tau}$, $i=1,2$. Hence
	\begin{equation*}
		S^{\ket{\tau}}_{\psi_1, \psi_2} (a \rho_2^{\frac{1}{2}}\tensor I) \ket{\tau} = ((\supp(\psi_2) a^{*}\rho_1^{\frac{1}{2}})\tensor  \cI) \ket{\tau} \text{ for all } a \in \alicealg,
	\end{equation*}
	and this implies that 
	\begin{equation*}
		S^{\ket{\tau}}_{\psi_1, \psi_2} (a \tensor I)  \ket{\tau} = ( \rho_2^{-\frac{1}{2}} a^* \rho_1^{\frac{1}{2}}\tensor \cI)   \ket{\tau}
		= S_{\ket{\tau}} \rho_1^{1/2} a \rho_2^{-1/2} \otimes \cI \ket{\tau}
		= S_{\ket{\tau}} (\rho_1^{1/2} \otimes (\rho_2^{-1/2})^T) (a \otimes \cI) \ket{\tau}
	\end{equation*}
	for all $a \in \alicealg$, where $S_{\ket{\tau}}$ is the modular operator for the vector $\ket{\tau}$, and $\rho_2^{-1/2}$ is the pseudoinverse of $\rho_2^{1/2}$. Since $S_{\ket{\tau}}^* S_{\ket{\tau}} = \cI$, the relative modular operator between $\psi_1$ and $\psi_2$ is
	\begin{equation*}
		\mathbf{\Delta}^{\ket{\tau}}_{\psi_1, \psi_2} = \rho_1 \tensor (\rho_2^{-1})^{T}.
	\end{equation*}
	Since both $\rho_2$ and $\rho_1$ are finite-dimensional matrices, the relative entropy is 
	\begin{align*}
		\D^{\ket{\tau}}(\psi_1 \| \psi_2) &= \braket{\psi_1|\log \mathbf{\Delta}_{\psi_1,\psi_2}^{\ket{\tau}}|\psi_1} 
		= \braket{\tau| (\rho_1^{\frac{1}{2}} \tensor \cI) \log(\rho_1 \otimes (\rho_2^{-1})^T) (\rho_1^{\frac{1}{2}} \tensor \cI)|\tau}
		= \Tr( \rho_1  \log(\rho_1) - \rho_1 \log(\rho_2)),
	\end{align*}
	which is the standard definition of von Neumann entropy. 
\end{example}

We remark that the definition of relative entropy holds for general von Neumann algebras in standard form. However, we will only focus on the case where $\alicealg$ is tracial, as we primarily work with tracially embeddable strategies in this paper. We refer to~\cite{arakiRelativeEntropyStates1977, ohyaQuantumEntropyIts2004} for more details. We use the following properties of relative entropy:
\begin{proposition}[Theorem 3.6 of~\cite{arakiRelativeEntropyStates1977}] \label{prop:properties_entropy}
	If $\psi_1, \psi_2$, and $\psi_3$ are positive normal linear functionals on a von Neumann algebra $\alicealg$, then:
	\begin{enumerate}
		\item If $\psi_1(\cI) = \psi_2(\cI)$, then $\D(\psi_1 \| \psi_2) \geq 0$, and $\D(\psi_1 \| \psi_2) = 0$ if and only if $\psi_1 = \psi_2$,
		\item $\D(\alpha \psi_1 \| \beta \psi_2) = \alpha \D(\psi_1 \| \psi_2) - \alpha \cdot \psi_1(\cI) \cdot  \log{(\frac{\beta}{\alpha})}$ for all $\alpha, \beta > 0$, and
		\item if $\psi_2 \leq \psi_3$ (meaning that $\psi_2(a) \leq \psi_3(a)$ for all $a \in \alicealg^{+}$),  then $\D(\psi_1\| \psi_3) \leq \D(\psi_1 \| \psi_2)$.
	\end{enumerate}
\end{proposition}
A linear map $\alpha : \alicealg_1 \to \alicealg_2$ between $C^*$-algebras is said to satisfy the Schwarz inequality if \begin{equation*}
	\alpha(a^* a) \geq \alpha(a)^* \alpha(a)
\end{equation*}
for all $a \in \alicealg_1$.

\begin{proposition}[Uhlmann monotonicity theorem, Theorem 5.3 of~\cite{ohyaQuantumEntropyIts2004}]\label{prop:Uhlmann_monotonicity_Schwarz}
	Let $\alpha: \alicealg_1 \rightarrow \alicealg_2$ be a unital map (meaning $\alpha(\cI_{\alicealg_1}) = \cI_{\alicealg_2}$) which satisfies the Schwarz inequality, and let $\psi_1^{\alicealg_i}, \psi_2^{\alicealg_i}$ be positive normal linear functionals on $\alicealg_i$, $i=1,2$, such that $\psi_1^{\alicealg_2} \circ \alpha \leq \psi_1^{\alicealg_1}$ and  $\psi_2^{\alicealg_2} \circ \alpha \leq \psi_2^{\alicealg_1}$. Then 
	\begin{equation*}
		\D(\psi_1^{\alicealg_1} \| \psi_2^{\alicealg_1}) \leq \D(\psi_1^{\alicealg_2} \| \psi_2^{\alicealg_2}).
	\end{equation*} 
\end{proposition}

\begin{proposition}[Additivity under direct sums]\label{prop:direct_sum_entropy}
	Let $\alicealg_1 \subseteq \bB(\cH_1)$ and $\alicealg_2\subseteq \bB(\cH_2)$ be two von Neumann algebras in standard form. Let $\psi_1^{\alicealg_i}, \psi_2^{\alicealg_i}$ be  positive normal linear functionals on $\alicealg_i$, $i=1,2$. Then
	\begin{equation*}
		\D(\psi_1^{\alicealg_1} \oplus \psi_1^{\alicealg_2} \| \psi_2^{\alicealg_1} \oplus \psi_2^{\alicealg_2}) =  \D(\psi_1^{\alicealg_1} \| \psi_1^{\alicealg_2} ) + \D(\psi_2^{\alicealg_1} \| \psi_2^{\alicealg_2})
	\end{equation*}
	on $\alicealg_1 \oplus \alicealg_2 \subseteq B(H_1 \oplus H_2)$. 
\end{proposition}
\begin{proof}
	See~\cite[Proposition 2.3]{hiaiQuantumFDivergencesNeumann2021a}  with $f(t) = t \log(t)$.
\end{proof}
The following proposition is an analogue of Pinsker's inequality:
\begin{proposition}[Theorem 5.5 of~\cite{ohyaQuantumEntropyIts2004}] \label{prop:pinsker} 
	If $\psi_1$ and $\psi_2$ are normal states on a von Neumann algebra $\alicealg$, then
	\begin{equation*}
		\|\psi_1 - \psi_2\|_{\alicealg}^2 \leq 2 \ln{(2)} \, \D(\psi_1 \| \psi_2).
	\end{equation*}
\end{proposition}

\begin{proposition}
	\label{prop:relative_min_entropy_chain_rule2}
	Let $\phi, \psi$ and $\upsilon$ be normal positive linear functionals on a von Neumann algebra $\alicealg$, such that $\D(\phi \| \psi) \leq \lambda_1$ and $\psi \leq 2^{\lambda_2} \upsilon$ for two scalars $\lambda_1, \lambda_2 \geq 0$. Then $\D(\phi \| \upsilon) \leq \lambda_1 + \phi(\cI) \lambda_2$.
\end{proposition}
\begin{proof} 
	By~\Cref{prop:properties_entropy}, $\D(\phi \| \upsilon) \leq \D(\phi \| 2^{-\lambda_2} \psi) = \D(\phi \| \psi) - \phi(\cI)\log{2^{- \lambda_2}} \leq \lambda_1 + \phi(\cI) \lambda_2$.
\end{proof}

\subsubsection{Mutual information}

Let $\alicealg_1$ and $\alicealg_2$ be two von Neumann algebras in standard form. If $\psi_i$ is a normal state on $\alicealg_i$, $i=1,2$, then there is a unique state $\psi_1 \otimes \psi_2$ on $\alicealg_1 \otimes \alicealg_2$ such that $\psi_1 \otimes \psi_2 (a \otimes b) = \psi_1(a) \psi_2(b)$ for all $a \in \alicealg_1$, $b \in \alicealg_2$ (indeed, if $\ket{\psi_i}$ is the vector in the
positive cone corresponding to $\psi_i$, then $\psi_1 \otimes \psi_2$ is the state corresponding to $\ket{\psi_1}\ket{\psi_2}$). This 
leads to a definition of the mutual information between two algebras:
\begin{definition}[Mutual information] \label{def:mutual_info_min}
	Suppose $\psi^{\alicealg_1 \alicealg_2}$ is a normal state on $\alicealg_1 \otimes \alicealg_2$. The \textit{mutual information}
	between $\alicealg_1$ and $\alicealg_2$ for the state $\psi^{\alicealg_1 \alicealg_2}$ is
	\begin{equation*}
		\text{\gls{Mutualinfo}} := \D(\psi^{\alicealg_1 \alicealg_2} \| \psi^{\alicealg_1}\otimes \psi^{\alicealg_2}).
	\end{equation*}
\end{definition}
If $\psi^{\alicealg_1 \alicealg_2 \alicealg_3}$ is a normal state on $\alicealg_1 \otimes \alicealg_2 \otimes \alicealg_3$, then we use the convention that $I(\alicealg_1 : \alicealg_2)_{\psi}$ denotes the relative entropy between $\alicealg_1$ and $\alicealg_2$ for the state $\psi^{\alicealg_1 \alicealg_2}$. 
\begin{example}
	Let $\alicealg_1 = \alicealg_2 = \bM_n (\bC)$ and let $\psi^{\alicealg_1 \alicealg_2}(a) = \Tr(\sigma^{\alicealg_1 \alicealg_2} a)$ be a state on $\alicealg_1 \tensor \alicealg_2$, where $\sigma^{\alicealg_1 \alicealg_2}$ is a density matrix in $\bM_{n^2}(\bC) $. For all $a \in \alicealg_1$, 
	\begin{equation*}
		\psi^{\alicealg_1}(a) = \Tr(\sigma^{\alicealg_1 \alicealg_2} (a \tensor \cI)) = \Tr( \sigma_1 a), 
	\end{equation*}
	where $\sigma^{\alicealg_1} := \Tr_{\alicealg_2}(\sigma^{\alicealg_1 \alicealg_2})$ is the partial trace. Similarly, $\psi^{\alicealg_1}(b) = \Tr((\sigma^{\alicealg_2} b)$ for all $b \in \alicealg_2$, where $\sigma^{\alicealg_2} := \Tr_{\alicealg_1} (\sigma^{\alicealg_1 \alicealg_2})$. If $a,b \in \bM_n(\bC)$, then 
	\begin{equation*}
		\psi_1 \otimes \psi_2 (a \otimes b) = \psi_1(a) \psi_2(b) = \Tr(\sigma^{\alicealg_1} a) \Tr(\sigma^{\alicealg_2} b) = \Tr((\sigma^{\alicealg_1} \otimes \sigma^{\alicealg_2})(a \otimes b)).
	\end{equation*}
	Hence $\rI(\alicealg_1 : \alicealg_2)_{\psi}$ is the relative entropy between the state with density matrix $\sigma^{\alicealg_1 \alicealg_2}$, and the state with density matrix $\sigma^{\alicealg_1} \otimes \sigma^{\alicealg_2}$. By~\Cref{exam:relativeentropy}, this is the usual mutual information between Alice and Bob's registers with state $\sigma^{\alicealg_1 \alicealg_2}$. 
\end{example}
Many of the properties of mutual information in finite dimensions extend to mutual information between von Neumann algebras. 
\begin{proposition}[Corollary 5.20 of~\cite{ohyaQuantumEntropyIts2004}]  \label{prop:entropy_tensor}
	Let $\alicealg_1$ and $\alicealg_2$ be two von Neumann algebras and let $\phi^{\alicealg_1 \alicealg_2}$ and $\psi^{\alicealg_1}\otimes \psi^{\alicealg_2}$ be two normal states on $\alicealg_1 \otimes \alicealg_2$. Then
	\begin{equation*}
		\D(\phi^{\alicealg_1 \alicealg_2} \| \psi^{\alicealg_1}\otimes \psi^{\alicealg_2}) = \D(\phi^{\alicealg_1} \| \psi^{\alicealg_1}) + \D( \phi^{\alicealg_1 \alicealg_2} \| \phi^{\alicealg_1} \otimes \psi^{\alicealg_2}).
	\end{equation*}
\end{proposition}
\begin{proposition}[Monotonicity]\label{prop:divergence_data_processing}
	If $\phi^{\alicealg_1 \alicealg_2}$ and  $\psi^{\alicealg_1 \alicealg_2}$ are two positive normal linear functionals on the von Neumann algebra $\alicealg_1 \otimes \alicealg_2$, then $\D(\phi^{\alicealg_1} \| \psi^{\alicealg_1} ) \leq \D(\phi^{\alicealg_1 \alicealg_2} \| \psi^{\alicealg_1 \alicealg_2} )$. As a result, if $\psi$ is a state on $\alicealg_1 \otimes \alicealg_2 \otimes \alicealg_3$ then $\rI(\alicealg_1 \otimes \alicealg_2 : \alicealg_3)_{\psi}
	\geq \rI(\alicealg_2 : \alicealg_3)_{\psi}$. 
\end{proposition}
\begin{proof}
	The map $\alpha:\alicealg_1 \rightarrow \alicealg_1 \otimes \alicealg_2 : a \mapsto a \otimes \Id_{\alicealg_2}$ is a unital homomorphism, and hence satisfies Schwarz's inequality. Thus
	\begin{equation*}
		\D (\phi^{\alicealg_1} \| \psi^{\alicealg_1} ) = \D (\phi^{\alicealg_1 \alicealg_2} \circ \alpha \| \psi^{\alicealg_1 \alicealg_2} \circ \alpha ) \leq \D (\phi^{\alicealg_1 \alicealg_2} \| \psi^{\alicealg_1 \alicealg_2})
	\end{equation*}
	by~\Cref{prop:Uhlmann_monotonicity_Schwarz}. 
\end{proof}

\begin{proposition}[Quantum Gibb's inequality]\label{prop:divergence_gibbs_inequality}
	Let $\phi^{\alicealg_1 \alicealg_2}$ be a normal state on the von Neumann algebra $\alicealg_1 \otimes \alicealg_2$, and $\psi^{\alicealg_1}$, $\psi^{\alicealg_2}$ be two normal state on $\alicealg_1$ and $\alicealg_2$ respectively. Then
	\begin{equation*}
		\rI(\alicealg_1 :\alicealg_2)_{\phi} \leq \D (\phi^{\alicealg_1 \alicealg_2} \| \psi^{\alicealg_1} \otimes  \psi^{\alicealg_2}).
	\end{equation*}
\end{proposition}
\begin{proof}
	Using~\Cref{prop:entropy_tensor} first on $\alicealg_1$ and then on $\alicealg_2$, we see that
	\begin{align*}
		\D & (\phi^{\alicealg_1 \alicealg_2}  \| \psi^{\alicealg_1} \otimes  \psi^{\alicealg_2} )	- \rI(\alicealg_1 :\alicealg_2)_{\phi} \\ 
		&= \D (\phi^{\alicealg_1} \| \psi^{\alicealg_1}) + \D (\phi^{\alicealg_1 \alicealg_2}\| \phi^{\alicealg_1} \otimes  \psi^{\alicealg_2} ) - \D (\phi^{\alicealg_1 \alicealg_2} \| \phi^{\alicealg_1} \otimes  \phi^{\alicealg_2} )	\\
		&= \D (\phi^{\alicealg_1} \| \psi^{\alicealg_1}) + \D (\phi^{\alicealg_2} \| \psi^{\alicealg_2}), 
	\end{align*}
	which is non-negative by~\Cref{prop:properties_entropy}, part (1). 
\end{proof}

In this appendix, we only use mutual information in the very restricted context of classical-quantum states discussed below. However, 
we have not seen~\Cref{def:mutual_info_min} in the literature previously, and it's interesting to discuss other possible definitions. 
For instance, the double dual $\alicealg^{**}$ of a $C^*$-algebra $\alicealg$ is a von Neumann algebra containing $\alicealg$, such that any state $\psi$ on $\alicealg$ extends to a normal state $\widehat{\psi}$ on $\alicealg^{**}$. Thus we can define the relative entropy between two states $\phi$ and $\psi$ on a $C^*$-algebra as the relative entropy $\D(\widehat{\phi}, \widehat{\psi})$ between the normal states $\widehat{\phi}$ and $\widehat{\psi}$ on $\alicealg^{**}$. If $\alicealg$ happens to be a von Neumann algebra and $\phi$ and $\psi$ are normal states, then $\D(\widehat{\phi},\widehat{\psi}) = \D(\phi,\psi)$, so this does not lead to a new notion of relative entropy. 

If $\psi_1$ and $\psi_2$ are two states on $C^*$-algebras $\alicealg_1$ and $\alicealg_2$ respectively, then there is a unique state $\psi_1 \otimes_{min} \psi_2$ on the min-tensor product $\alicealg_1 \otimes_{min} \alicealg_2$ such that $\psi_1 \otimes \psi_2(a \otimes b) = \psi_1(a) \psi_2(b)$ for all $a \in \alicealg_1$, $b \in \alicealg_2$, and this pulls back to a unique state $\psi_1 \otimes_{max} \psi_2$ on $\alicealg_1 \otimes_{max} \alicealg_2$ with the same property (the definition for the min/max tensor product are the standard definition used for $C^*$ algebra theory, and can be found in, e.g.~\cite[Section 3.8]{goldbringConnesEmbeddingProblem2021}). Hence if $\psi$ is a state on the max tensor product $\alicealg_1 \otimes_{max} \alicealg_2$, then we can define the mutual information between $\alicealg_1$ and $\alicealg_2$ for the state $\psi$ to be 
\begin{equation*}
	\rI(\alicealg_1 : \alicealg_2)_{\psi}^{max} := \D(\psi, \psi^{\alicealg_1} \otimes_{max} \psi^{\alicealg_2}),
\end{equation*}
where $\psi^{\alicealg_i}$ is the restriction of $\psi$ to $\alicealg_i$ inside of $\alicealg_1 \otimes_{max} \alicealg_2$. 
If $\psi$ is a state on $\alicealg_1 \otimes_{min} \alicealg_2$, then we can define $\rI(\alicealg_1  : \alicealg_2)_{\psi}^{min} := \D(\psi, \psi^{\alicealg_1} \otimes_{min} \psi^{\alicealg_2})$ similarly. Any state $\psi$ on $\alicealg_1 \otimes_{min} \alicealg_2$ pulls back to a state $\widetilde{\psi}$ on $\alicealg_1 \otimes_{max} \alicealg_2$, so there are seemingly two different choices for the mutual information between $\alicealg_1$ and $\alicealg_2$ in this case, $\rI(\alicealg_1 : \alicealg_2)_{\psi}^{min}$ and $\rI(\alicealg_1 : \alicealg_2)_{\widetilde{\psi}}^{max}$. However, there is a surjective homomorphism from $\alicealg_1 \otimes_{max} \alicealg_2$ to $\alicealg_1 \otimes_{min} \alicealg_2$, so the following lemma shows that $\rI(\alicealg_1 : \alicealg_2)_{\psi}^{min} = \rI(\alicealg_1 : \alicealg_2)_{\widetilde{\psi}}^{max}$. 
\begin{lemma}\label{lem:naturalrelentropy}
	If $\alpha : \alicealg \to \bobalg$ is a surjective $*$-homomorphism between $C^*$-algebras, and $\phi$ and $\psi$ are states on $\bobalg$, then $D(\phi \circ \alpha, \psi \circ \alpha) = \D(\phi, \psi)$.
\end{lemma}
The proof of~\Cref{lem:naturalrelentropy} follows from~\cite[Theorem 6.19]{hiaiQuantumFDivergencesNeumann2021a}; since we do not make further use of this lemma, we leave the complete proof as an exercise for the reader. Similarly, if $\psi$ is a normal state on the tensor product $\alicealg_1 \otimes \alicealg_2$ of two von Neumann algebras, then $\rI(\alicealg_1 : \alicealg_2)_{\psi} = \rI(\alicealg_1 : \alicealg_2)_{\psi}^{min} = \rI(\alicealg_1 : \alicealg_2)_{\widetilde{\psi}}^{max}$. 

\subsubsection{Classical-quantum states}
For a discrete finite set $\X$, let $\C^{\X}$ denote the von Neumann algebra of functions from $\X$ to $\C$. To match the standard notation from quantum information, let $\bra{x}$ denote the indicator function for $x \in X$. These functions span $\C^{\X}$, and give an isomorphism between $\C^{\X}$ and the algebra of $|\X| \times |\X|$ diagonal matrices. If $\alicealg$ is another von Neumann algebra, then $\C^{\X} \otimes \alicealg = \bigoplus_{x \in X} \bra{x} \otimes \alicealg$ is the von Neumann algebra of $|\X| \times |\X|$ diagonal matrices with coefficients from $\alicealg$, and every normal state on $\C^{\X} \otimes \alicealg$ is of the form 
\begin{equation*}
	\text{\gls{CQstate}}= \sum_{x \in \X} \P(x)\bra{x} \otimes \phi_x 
\end{equation*}
for some collection of normal states $\{ \phi_x \}_{x \in \X}$ on $\alicealg$ and probability measure $\P(x)$ on $\X$. Hence such states are called \emph{classical-quantum states on $\alicealg$ with classical part $\X$}. When dealing with multiple classical subsystems, we denote $\phi_x$ by $\phi^{\X\A}_{\X=x}$. Also, note that if $\phi^{\X\A}$ is a classical-quantum state, then $\phi^{\X}$ is the classical distribution $P$ on $\X$. We use the following example to connect our definition with the standard definition for classical-quantum state used in quantum information. 
\begin{example}
	Let $\X$ be a discrete finite set, $\P(x)$ be a probability distribution over $\X$ and let $(\alicealg, \tau)$ be a tracial von Neumann algebra. Define the collection of normal state $\{ \phi_x \}_{x \in \X}$ acting on $\alicealg$ as $\phi_x(A) = \tau(\sigma_x A)$ for some positive element $\sigma_x \in \alicealg^{+}$ with $\tau(\sigma_x) = 1$. We define the classical-quantum state $\phi^{\cX \cA}= \sum_{x \in \X} \P(x)\bra{x} \otimes \phi_x$, and we see that for all $A \in \cM_{|X|}(\bC) \otimes A$
	\begin{equation*}
		\phi^{\cX \cA}(A) = \Tr \otimes \tau\left(A \cdot \left(  \sum_{x \in \X} \P(x)\ketbra{x}{x} \otimes \sigma_x \right)\right)
	\end{equation*}
	where $\Tr$ in the above equation is defined over $\cM_{|X|}(\bC)$. In this case, one can intuitively think of $\sum_{x \in \X} \P(x)\ketbra{x}{x} \otimes \sigma_x$ as the ``density matrix" for the state $\phi^{\cX \cA}$, and we see that this is consistent with the standard definition for classical-quantum state in quantum information literatures (e.g.~\cite[Definition 4.3.5]{wildeQuantumInformationTheory2013}). 
\end{example}

We now prove some properties of classical-quantum states used in~\cite[Section 5.1]{bavarianAnchoredParallelRepetition2021}.

\begin{proposition}[Chain rule for relative entropy]\label{prop:divergence_chain_rule}
	Let $\phi^{\X \alicealg}= \sum_{x} \P(x)\bra{x} \otimes \phi_x^\alicealg$ and $\psi^{\X \alicealg}= \sum_{x} \Qsf(x) \bra{x} \otimes \psi^\alicealg$ be two classical-quantum states on $\alicealg$ with classical part $\X$. Then
	\begin{equation*}
		\D(\phi^{\X \alicealg} \| \psi^{\X \alicealg}) = \D(\P \| \Qsf) + \Ex_{x \sim \P} \D(\phi_x^{\alicealg} \| \psi^{\alicealg}_x ),
	\end{equation*}
	where $\D(\P \| \Qsf) = \sum_x \P(x) \log \frac{\P(x)}{\Qsf(x)}$ denotes the relative entropy between two classical distributions $\P$ and $\Qsf$. As a result, $\D(\phi \| \psi) \geq \Ex_{x \sim \P} \D(\phi_x^{\alicealg} \| \psi_x^{\alicealg})$.
\end{proposition}
\begin{proof}
	Using~\Cref{prop:direct_sum_entropy} and part (3) of~\Cref{prop:properties_entropy},
	\begin{align*}
		\D(\phi^{\X \alicealg} \| \psi^{\X \alicealg}) &= \sum_{x} \D( \P(x)\phi_x^\alicealg \| \Qsf(x) \psi_x^\alicealg)  \\
		&= \sum_{x} \P(x) \D( \phi_x^\alicealg \| \psi_x^\alicealg)  + \P(x) \log(\frac{\P(x)}{\Qsf(x)})\\
		&= \Ex_{x \sim \P} \D(\phi_x^{\alicealg} \| \psi_x^{\alicealg} ) + \D(\P \| \Qsf).
	\end{align*}
	Since relative entropy between classical distributions is non-negative, $\D(\phi \| \psi) \geq \Ex_{x \sim \P} \D(\phi_x^{\alicealg} \| \psi_x^{\alicealg})$.
\end{proof}
\begin{proposition}[Conditional mutual information]\label{prop:conditional_mutual_info}
	Let $\phi^{\X \alicealg_1 \alicealg_2} = \sum_{x} P(x) \ket{x} \otimes \phi_x$ be a classical-quantum state on $\alicealg_1 \tensor \alicealg_2$ with classical part $\X$. Then
	\begin{equation*}
		\rI(\X \alicealg_1 : \alicealg_2)_{\phi} - \rI(\X : \alicealg_2)_{\phi}  =  \Ex_{x \sim \P} \rI(\alicealg_1  : \alicealg_2)_{\phi_x}.
	\end{equation*}
\end{proposition}
\begin{proof}
	The restriction of $\phi^{\X \alicealg}$ to $\C^{\X}$ is $\phi^{\X} = \sum_x \P(x) \bra{x}$. Since $\D(P,P) = 0$,~\Cref{prop:divergence_chain_rule} implies that 
	\begin{equation*}
		\rI(\X: \alicealg_2)_{\phi} = \D \left(\sum_x \P(x)  \ket{x} \otimes \phi_x^{\alicealg_2} \| \sum_x  \P(x) \ket{x} \otimes \sum_{y}\P(y)  \phi_y^{\alicealg_2} \right) = \Ex_{x \sim \P} \D (\phi_x^{\alicealg_2} \| \phi^{\alicealg_2})~,
	\end{equation*}
	where $\phi^{\alicealg_2} = \sum_y P(y) \phi_y^{\alicealg_2}$. Similarly, 
	\begin{align*}
		\rI(\X \alicealg_1 : \alicealg_2)_{\phi} = \D(\phi^{X \alicealg_1 \alicealg_2} \| \phi^{X \alicealg_1} \otimes \phi^{\alicealg_2}) &= \Ex_{x \sim \P} \D(\phi_x \| \phi_x^{\alicealg_1} \otimes \phi^{\alicealg_2}).
	\end{align*}
	Applying~\Cref{prop:entropy_tensor} to $\alicealg_2$, we see that
	\begin{align*} 
		\rI(\X \alicealg_1 : \alicealg_2)_{\phi} - \rI(\X : \alicealg_2)_{\phi} &=\Ex_{x \sim \P} \D(\phi_x^{\alicealg_1 \alicealg_2} \| \phi_x^{\alicealg_1} \otimes \phi^{\alicealg_2} ) - \D(\phi_x^{\alicealg_2}  \| \phi^{\alicealg_2})
		\\&= \Ex_{x \sim \P} \D(\phi_x^{\alicealg_1 \alicealg_2} \| \phi_x^{\alicealg_1} \otimes \phi_x^{\alicealg_2} ) = \Ex_{x \sim \P} \rI(\alicealg_1  : \alicealg_2)_{\phi_x}.
	\end{align*}
\end{proof}
We can now prove a von Neumann algebraic version of \emph{quantum Raz's Lemma}, which is a central tool in~\cite{bavarianAnchoredParallelRepetition2021}. This is a quantum analogue of Raz's lemma, which is a key part of many proofs of the classical parallel repetition theorem~\cite{razParallelRepetitionTheorem1995a, holensteinParallelRepetitionSimplifications2009,barakStrongParallelRepetition2009}. 
\begin{lemma}[Quantum Raz's Lemma]
	\label{lem:quantum_raz}
	Let $\phi^{\X \alicealg}= \phi^{\X_1 \X_2 \ldots \X_n \alicealg}$ and $\psi^{\X \alicealg}=\psi^{\X_1}\otimes \psi^{\X_2}\otimes \ldots \otimes \psi^{\X_n} \otimes \psi^\alicealg$ be two classical-quantum states with classical component $\X= \X_1 \times \X_2 \times \ldots \X_n$. Then
	\begin{equation}
		\sum_{i=1}^n \, \rI (\X_i : \alicealg)_{\phi} \,\leq\, \D (\phi^{\X \alicealg} \,  \| \psi^{\X \alicealg})\;. 
	\end{equation}
\end{lemma}

\begin{proof}
	Let $\X_{\leq i} := \X_{1} \X_{2} \cdots \X_{i}$ and $\X_{\geq i} := \X_{i} \X_{i+1} \cdots \X_{n}$. For each $2 \leq i \leq n$, ~\Cref{prop:conditional_mutual_info} implies that 
	\begin{align}
		\nonumber \rI(\X_{\leq i-1} \alicealg: \X_i)_{\phi^{\X \alicealg}} - \rI(\X_{ \leq i-1}: \X_i)_{\phi^{\X \alicealg}} &=  \Ex_{x_{<i} \sim \phi^{\X_{\leq i-1}}}  \rI( \X_i :\alicealg)_{\phi_{x_{<i}}^{\X_{\leq i}\alicealg}} \\
		\label{eq:raz_1a} &= \rI(\X_{\leq i} : \alicealg)_{\phi^{\X \alicealg}} - \rI(\X_{\leq i-1} : \alicealg)_{\phi^{\X \alicealg}}. 
	\end{align} 
	Repeatedly applying~\Cref{prop:entropy_tensor}, we get that
	\begin{align*}
		\D(\phi^{\X \alicealg} \| \psi^{\X \alicealg}) &= \D(\phi^{\X_1}|\psi^{\X_1}) + \D(\phi^{\X \alicealg} |\phi^{\X_1}  \tensor \psi^{ \X_{\geq 2} \alicealg}) \\
		&=\D(\phi^{\X_1}|\psi^{\X_1}) + \D(\phi^{\X_1 \X_2}|\phi^{\X_1} \tensor \psi^{\X_2} ) 
		+ \D(\phi^{\X \alicealg} |\phi^{\X_{\leq 2}}  \tensor \psi^{ \X_{\geq 2} \alicealg}) \\
		&= \sum_{i=1}^n \D(\phi^{\X_{\leq i}}|\phi^{\X_{\leq i-1}} \tensor \psi^{\X_i}) + \D(\phi^{\X \alicealg} |\phi^{\X} \tensor \psi^{\alicealg}). 
	\end{align*}
	Hence by the quantum Gibb's inequality in~\Cref{prop:divergence_gibbs_inequality}, 
	\begin{align*}
		\D(\phi^{\X \alicealg} \| \psi^{\X \alicealg}) &\geq \sum_{i=2}^{n} \rI(\X_{\leq i -1 } : \X_{i})_{\phi^{\X \A}} +  \rI(\X: \alicealg)_{\phi^{\X \A}}.       \end{align*}
	Solving for $\rI(\X_{\leq i -1 } : \X_{i})_{\phi^{\X \A}}$ in~\Cref{eq:raz_1a}, we get the telescoping sum
	\begin{align*}
		\sum_{i=2}^{n} \rI(\X_{\leq i -1 } : \X_{i})_{\phi^{\X \A}} & =  
		\sum_{i=2}^{n} \left(\rI(\X_{\leq i-1} \alicealg: \X_i)_{\phi^{\X \alicealg}} - \rI(\X_{\leq i} : \alicealg)_{\phi^{\X \alicealg}} + \rI(\X_{\leq i-1} : \alicealg)_{\phi^{\X \alicealg}} \right)  \\
		&= \sum_{i=2}^{n} \rI(\X_{\leq i-1} \alicealg: \X_i)_{\phi^{\X \alicealg}} + \rI(\X_1 : \alicealg)_{\phi^{\X \alicealg}} - \rI(\X : \alicealg)_{\phi^{\X \alicealg}}.
	\end{align*}
	By~\Cref{prop:divergence_data_processing}, $\rI(\X_{\leq i -1} \alicealg : \X_{ i})_{\phi^{\X \A}} \geq \rI(\alicealg : \X_{i})$, so 
	\begin{align*}
		\D(\phi^{\X \alicealg} \| \psi^{\X \alicealg}) \geq \sum_{i=2}^{n} \rI(\X_{\leq i -1 } : \X_{i})_{\phi^{\X \A}} + \rI(\X : \alicealg)_{\phi^{\X \alicealg}} & = \sum_{i=2}^{n} \rI(\X_{\leq i -1} \alicealg : \X_{ i})_{\phi^{\X \A}} + \rI(\X_1 : \A) \\ 
		& \geq \sum_{i=1}^{n} \rI(\X_i : \alicealg)_{\phi^{\X \alicealg}}.
	\end{align*}
\end{proof}

\subsection{Proof of~\Cref{thm:anchoringparallelrep}}

Having presented the von Neumann algebra framework for nonlocal games and some quantum information theory tools within it, we prove~\Cref{thm:anchoringparallelrep} in this subsection. Before we begin, we first give some additional background on classical probability which is necessarily for the proof.  

\subsubsection{Probability distributions, random variables, and expectations.} 
In this appendix, we use $\P$, $\Qsf$, $\Ssf$ and $\Rsf$ to denote probability distributions. Given a probability distribution $P$ on a discrete finite set $\X$ and a random variable $f$ on $X$, we let $\Ex_{x \sim P} f(x)$ denote the expected value of $f$ with distribution $P$. We use $\P_X$ to denote the distribution of random variable $X$ and $\P_X(x)$ to denote the probability that $X = x$ for some value $x$. For multiple random variables, e.g.\ $X, Y, Z$, $\P_{XYZ}(x,y,z)$  denotes their joint distribution. All random variables are assumed to operate on the same probability space, which is usually implicit and clear from context. 

We use \gls{Conditionalexp} to denote the conditional distribution $\P_{YX}(y,x)/\P_X(x)$, which is defined when $\P_X(x) > 0$. We use the shorthand $\P_{X | y,z}$ to denote the distribution $\P_{X | Y =y,Z=z}$. For example, we may write $\P_{V | \omega_\mi, \vx_i, \vy_i}$ to denote $\P_{V | \Omega_\mi = \omega_\mi, \vx_i = \vx_i, \vy_i = \vy_i}$. For an event $W$ we let $\P_{X Y | W}$ denote the distribution conditioned on $W$. We use the notation $\Ex_{X} f(x)$ and $\Ex_{\P_X} f(x)$ to denote the expectation $\sum_{x} \P_X(x) f(x)$. Let $\P_{XY}$ be a joint distribution on $\X \times \Y$ and let $W$ denote an event. Then we define the distribution $\P_{X|W} \P_{Y|X}$ over $\X \times \Y$ as
\[
(\P_{X|W} \P_{Y|X})(x,y) = \P_{X|W}(x) \cdot \P_{Y | X = x}(y)~.
\]

For distributions $P_X$ and $P_Y$ over the same set $\X$ we use $\| \P_{X_0} - \P_{X_1} \|$ to denote their total variation distance,
\begin{equation} \label{eq:vardis}
	\text{\gls{Vardis}} \,=\, \frac{1}{2}\sum_{x \in \X} |\P_{X_0}(x) - \P_{X_1} (x)|\;.
\end{equation} 

We recall the following lemmas from~\cite[Section 3.2]{bavarianAnchoredParallelRepetition2021}, and we refer to Section 3.2 of the aforementioned paper for the proof.

\begin{lemma}\label{lem:trivial}
	Let $\Qsf_F$ and $\Ssf_F$ be two probability distributions for  random variable $F$ and let $\Rsf_{G | F}$ be a conditional probability distribution for random variable $\Ganchor$, conditioned on $F$. Then
	\[ \big\| \Qsf_{F} \Rsf_{G|F} - \Ssf_{F} \Rsf_{G|F}\big\|\,=\, \big\|\Qsf_{F}-\Ssf_{F}\big\|\;. \]
	Similarly, for two conditional probability distributions $\Qsf_{G|F}, \Ssf_{G|F}$ and a distribution $\Rsf_F$, 
	\[
	\big\| \Rsf_F \Qsf_{G|F} - \Rsf_F \Ssf_{G|F} \big\| \,=\, \Ex_F \big\| \Qsf_{G|F = f} - \Ssf_{G|F=f} \big\|\;,
	\]
	where $\Ex_F$ denotes the expectation over sampling $f$ from $\Rsf_F$.
\end{lemma}

\begin{lemma}[Data processing inequality]\label{lem:data-processing}
	Let $\Qsf_{FG}$ and $\Ssf_{FG}$ denote two probability distributions for random variables $F,G$. Then 
	\[
	\big \| \Qsf_{F} - \Ssf_F \big \| \leq \big \| \Qsf_{FG} - \Ssf_{FG} \big \|~.
	\]
\end{lemma}

\subsubsection{Overview of the proof to the parallel repetition theorem} \label{sec:paralleloverview}

We give a brief overview of the proof of~\Cref{thm:anchoringparallelrep} in this subsection. This proof follows a similar structure as~\cite{bavarianAnchoredParallelRepetition2021}, which itself follows a similar structure as the proof for parallel repetition theorem for \textit{classical} values for non-local games (see, e.g.~\cite{razParallelRepetitionTheorem1995a}). These approaches argue that if a ``too good to be true" strategy exists for the parallel repeated game, then, using information theoretical tools, a strategy which violates the optimal success probability for the original game can be constructed. 

To be more precise, suppose for a $r$-fold parallel repeated game $\cG^{\otimes r}$; there exists some strategy $\strategy^{\otimes r} = \{ \cH, \ket{\psi}, \{A_\vx^\va \}, \{B_\vy^\vb \} \}$ which has a success rate higher than the value indicated in~\Cref{thm:anchoringparallelrep}. Then, the strategy $\strategy$ must be performed in some correlated manner (i.e. the answer for the question pair $(\va_i,\vb_i)$ must rely on some other question pairs which are not $(\vx_i,\vy_i)$). More precisely, since the strategies are correlated between the different folds of the parallel repeated game,  there must exist some \textit{critical subset} $C \subseteq [r]$ such that conditioning on the provers winning on the subset $C$ using this strategy, the provers can win the overall parallel repeated game with high probability. If the provers were to somehow obtain an entangled state which mimics a ``post-measurement state" in which they had already won on those $C$ coordinates, then this gives them an advantage with the remaining $[r] \setminus C$ coordinates. Notably, this also gives a comparable advantage for a single instances of the game $\cG$ by running the strategy above and embedding the coordinate $(\vx_j, \vy_j)$ onto one of the $\{(\vx_i,\vy_i)\}_{i \in [r] \setminus C}$.

Unfortunately, it is not clear how to sample such a post-measurement state locally. Since conditioning on the provers winning on coordinate $C$ might change the input distribution on coordinate $j$ (as an example, this could occur when the answer given to coordinate $j$ is entirely dependent on the question on some coordinates on $C$), and in some cases, winning on coordinate $C$ might depend on one of the provers getting a certain input on the $j$th coordinate~! This means that creating such a ``post-measurement state" would often require the full question pair $\{(\vx_i,\vy_i)\}_{i \in [r]}$ (which includes coordinates $j$). This is one of the main challenges for showing the parallel repetition theorem using this approach.

One notable example in which the above approach works is the case where the classical input distribution to the non-local game is a product distribution between both provers~\cite{jainParallelDeviceIndependentQuantum2020}; in this case, it can be shown that this ``post-measurement state" is, on average, relatively uncorrelated to the prover's question pair on coordinates outside of $C$, and hence this state can be created locally by the provers. The anchored transformation used in~\cite{bavarianAnchoredParallelRepetition2021} and this paper intuitively destroys the correlations between the provers, and by using certain conditioning (on \textit{dependency breaking variable}, which we introduce next section), the prover's distribution can be made uncorrelated, and hence a similar argument can be made.

\subsubsection{Existences of a critical subset $C$ and dependency breaking variable} 
\label{sec:p_setup}
Recall from the beginning of this appendix that $\cG = (\cX, \cA_, \mu, D)$ denotes an anchoring game which arises from~\Cref{def:anchortransformation} in this appendix. Let $\cG^{\otimes r}$ denote the $r$-fold parallel repetition game and let $\strategy^{\otimes r} = \{ \cL^2(\alicealg, \tau), \ket{\psi} = \sigma \ket{\tau}, \{A_
\vx^\va \}, \{B_\vy^\vb \} \}$ be a tracially embeddable strategy which violates the bound given in~\Cref{thm:anchoringparallelrep}. We start with the following proposition which introduce the notion of the critical subset, since the proof is the same as in~\cite{bavarianAnchoredParallelRepetition2021}, we refer to the aforementioned paper for the proof. 

\begin{proposition}[Proposition 6.5 of~\cite{bavarianAnchoredParallelRepetition2021}] \label{prop:existcriticalset}
	\label{prop:subset}
	Let $t \in \{*, co\}$ and let $\cG$ be the game indicated by \Cref{thm:anchoringparallelrep} with $\omega^{t}(\cG) < 1- \epsilon$.	Let $W$ denote the indicator for winning all $r$ coordinates for the parallel repeated game $\cG^{\otimes r}$. Suppose that $n \geq \frac{16}{\eps} \log \frac{4}{\eps \cdot \P(W)}$. Then there exists a set $C \subseteq [r]$ of size at most $t = \frac{8}{\eps} \log \frac{4}{\eps \cdot \P(W)}$ such that
	
	$$
	\Ex_I\, \P (W_i | W_C) \,\geq \, 1 - \eps/2\;,
	$$
	where $\Ex_I$ denotes the expectation over a uniformly random $i$ chosen from $[r] \setminus C$ and $\P(W_i | W_C)$ denotes the probability, using the strategy $\strategy^{\otimes r}$, of winning the $i$-th instance of $\cG$ conditioned on winning all instances indexed by $C$.
\end{proposition}
\jqnote{$C$-critical, $R = [r] \setminus (\{i\}\cup C)$-irreverential coordinate, $i$-target coordinates.}

We fix a subset \gls{Criticalset} promised by \Cref{prop:subset}, which we call the \textit{critical subset} for the parallel repeated game $\cG^{\otimes r}$. We further assume without loss of generality that $C = \{r-|C|,\ldots,n-1\}$. 

For the remainder of this section, we reintroduce \emph{dependency-breaking variables} from~\cite{bavarianAnchoredParallelRepetition2021}. These are crucial tools for controlling the correlations between the input distributions between the provers. Since most of the propositions in this section are classical in nature and are proven in~\cite{bavarianAnchoredParallelRepetition2021}, we refer to Section 4 of the original paper for the proofs.

Recall from~\Cref{sec:introcorrelations}, the distributions $\mu_X$ and $\mu_Y$ denote the marginals of the game distribution $\mu$ for the anchored game $\cG$ on the first and second coordinates respectively. We first define a ``single copy'' distribution $\hat{P}$ as the law of random variables $(\Mplayer,\Mvalue,X,Y)$, where each random variable may depend on previously defined ones. Following the same set up as~\cite[Section 4.1]{bavarianAnchoredParallelRepetition2021},  we fix a ``noise'' parameter 
\begin{equation}
	\label{eq:def-eta}
	\text{\gls{Noiseparameter}} \,=\, \frac{1}{4}\;,
\end{equation}
We define the random variable $\Mplayer$ to be a uniform distribution over the finite set $\{A,B\}$, and we define $\Mvalue$ to have the following distribution over $\X$: for all $(x,y) \in \X^2$:
\[
\P_{\Mvalue | \Mplayer = A}(x) = \left\{
\begin{array}{ll}
	\frac{\mu_x(x)}{1 - \eta_{\text{Anchor}}}  & \mbox{if $x \neq \dummy$} \\[4mm]
	\frac{\frac{1}{2} - \eta_{\text{Anchor}}}{1 - \eta_{\text{Anchor}}} & \mbox{if $x = \dummy$ }
\end{array}
\right. \quad \text{and} \quad
\P_{\Mvalue | \Mplayer = B}(y) = \left\{
\begin{array}{ll}
	\frac{\mu_y(y)}{1 - \eta_{\text{Anchor}}}  & \mbox{if $y \neq \dummy$} \\[4mm]
	\frac{\frac{1}{2} - \eta_{\text{Anchor}}}{1 - \eta_{\text{Anchor}}} & \mbox{if $y = \dummy$ }
\end{array}
\right. \;.
\]
\jqnote{Remind the reader that $\sum_x \mu_{X \setminus \bot}(x) = \frac{1}{2}$. }

In other words, conditioned on $\Mplayer=A$ (resp. $\Mplayer = B$), the variable $\Mvalue$ takes on a value in $\X$ from a rescaled version of the distribution $\mu_X$ (resp. $\mu_Y$) where less weight is given to the dummy question $\dummy$.  Finally, define the random variables $(X,Y)$ as follows. 
\begin{itemize}
	\item If $\Mplayer = A$ then $X$ is chosen to be an ``$\eta_{\text{Anchor}}$-noisy'' copy of $\Mvalue$. Precisely, $X = \Mvalue$ with probability $1 - \eta_{\text{Anchor}}$ and $X = \dummy$ with probability $\eta_{\text{Anchor}}$. Define $Y$ to equal $y$ with probability $\mu_{Y|X}(y | m)$, where $m$ is the value of $\Mvalue$. 
	\item If $\Mplayer = B$ then $Y$ is chosen to be an ``$\eta_{\text{Anchor}}$-noisy'' copy of $\Mvalue$ and $X$ equals $x$ with probability $\mu_{X|Y}(x | m)$, where $m$ is the value of $\Mvalue$. 
\end{itemize}
This specifies the distribution $\hat{P}$. We recall the following properties about $\hat{P}$ from~\cite{bavarianAnchoredParallelRepetition2021}, which will be important for the parallel repetition theorem. 

\begin{claim}[Claim 4.1 of~\cite{bavarianAnchoredParallelRepetition2021}]
	\label{clm:dependency-breaking-1}
	Conditioned on $(\Mplayer,\Mvalue)$ the random variables $X$ and $Y$ are independent. 
\end{claim}

\begin{claim}[Claim 4.2 of~\cite{bavarianAnchoredParallelRepetition2021}]
	\label{clm:dependency-breaking-2}
	$\hat{\P}_{XY|\Mplayer = A}(x,y) = \hat{\P}_{XY|\Mplayer=B}(x,y) = \mu_{XY}(x,y)$ for all $(x,y) \in \X^2$. In particular, the marginal distribution $\hat{\P}_{XY}$ is identical to the game distribution $\mu$. 
\end{claim}

We remark that \Cref{clm:dependency-breaking-1} and \Cref{clm:dependency-breaking-2} states that, if $(x,y)$ is sampled according to $\hat{\P}$, then depending on whether $\bMplayer$ and $\bMvalue$ are conditioned, the distribution $X$ and $Y$ either follows the original distribution for the game, or becomes uncorrelated. We define the distribution $\P = (\bMplayer, \bMvalue, \bX, \bY, \bA, \bB)$ for the entirety of the input set as follows. Let $\bMplayer = ((\bMplayer)_0,\ldots,(\bMplayer)_{r-1})$, $\bMvalue= ((\bMvalue)_0,\ldots,(\bMvalue)_{r-1})$, $\bX = (\bX_0,\ldots,\bX_{r-1})$, and $\bY = (\bY_0,\ldots,\bY_{r-1})$ be vectors of random variables, define
\[
\P_{\bMplayer \bMvalue \bX \bY} \,=\, \prod_{i=0}^r \hat{\P}_{(\bMplayer)_i (\bMvalue)_i \bX_i \bY_i}~.
\]
Finally, we define random variables $\bA = (\bA_i)_{i \in [r]}$ and $\bB = (\bB_i)_{i \in [r]}$ as follows. When conditioned on $\bX$ and $\bY$, the random variables $\bA,\bB$ are independent of $\bD$ and $\bM$, and for all realizations $\vx, \vy$ of $\vx$ and $\vy$, define
\[
\P_{\bA \bB | \bX = \vx, \bY = \vy}(\va,\vb) \,=\, \bra{\psi} A_\vx^\va \, B_\vy^\vb  \ket{\psi}
\]
for all $\va = (\va_1,\ldots,\va_n) \in \A^n$ and $\vb = (\vb_1,\ldots,\vb_n) \in \B^n$. The following claim connects the distribution $P$ to the correlation set for the strategy. 

\begin{claim}[Claim 4.3 of~\cite{bavarianAnchoredParallelRepetition2021}]\label{claim:abxy}
	The marginal distribution of $\bX \bY \bA \bB$ is identical to the correlation set obtained in the repeated game $\cG^{\otimes r}$ when the provers use the strategy $\strategy^{\otimes r}$.
\end{claim}

We use the probability $\P$ as the distribution which models the behavior of $\strategy^{\otimes r}$ for the proof of \Cref{thm:anchoringparallelrep}. Having define the tuple of distribution $\P$, we are ready to define dependency breaking variables. These are crucial for controlling the correlations that arise when conditioning the distribution $\P$ on different events. Furthermore, we use $\blQ_C = (\bX_C,\bY_C)$ and $\blS_C = (\bA_C,\bB_C)$ to denote random variables associated with the provers' questions and answers in the coordinates indexed by the critical set $C$, and we use $\text{\gls{Critcoordinate}}=(\blQ_C, \blS_C)$ to denote the correlation for $\strategy^{\otimes r}$ over the subset of coordinate $C$. We remark that the event $\blS_C$ could potentially depend on the question pair from the coordinates outside of $C$. For $i \in [r]$ let $W_i$ denote the indicator variable for the event that the provers win round $i$ of the game for probability distribution $\P$. Let $W_C = \prod_{i \in C} W_i$, we define the \textit{dependency breaking variable} as:

\begin{definition}[Dependency breaking variable] \label{def:Depbreaking}
	Let $C \subseteq [n]$. For all $i \in [r] \setminus C$ define the \emph{$i$-th dependency-breaking variable $\bOmega_i$} as 
	\[\bOmega_i \,=\, ((\bMplayer)_i,(\bMvalue)_i)\;.\]
	Furthermore we define 
	\begin{equation}
	 \text{\gls{Depbreaking}} \,=\, (\bOmega_i)_{i \in [r] \setminus C}\qquad \text{and}\qquad \bOmega_\mi\,=\,(\bOmega_j)_{j \in [r] \setminus (C \cup \{i\})} \;,\quad \forall i \in [r] \setminus C\;.
	\end{equation}
\end{definition}

We remark that when $\eta_{\text{Anchor}} = 0$ the definition of the production distribution $\bOmega_i \blQ_C$ coincides with the one used by Holenstein~\cite{holensteinParallelRepetitionSimplifications2009}; in that case, the variable $(\bMplayer)_i$ is coupled to either $\bX_i$ or $\bY_i$ exactly. Here we set $\eta_{\text{Anchor}}$ to be a nonzero value that is related to the anchoring probability, as in~\eqref{eq:def-eta}. This ``noisy coupling'' between $\bOmega_i$ and the inputs $(\vx_i,\vy_i)$ is important for our analysis. Base on the above definition, we have the following claim. 

\begin{claim}[Claim 4.5 of~\cite{bavarianAnchoredParallelRepetition2021}]\label{claim:ufacts}
	The following properties hold for the distribution $\P$.
	\begin{enumerate}
		\item The joint distribution $(\bX,\bY,\bOmega)$ is product across its $r$ triples of coordinates. Furthermore, for any $i$, $\bX_i$ and $\bY_i$ are independent conditioned on $\bOmega_i$. In particular,
		$\P_{\bOmega_i \bX_i \bY_i} =\P_{\bOmega_i} \P_{\bX_i | \bOmega_i} \P_{\bY_i | \bOmega_i}$. 
		\item $\P_{\blS_C | \bX \bY} = \P_{\blS_C | \bOmega \bX \bY}$.
	\end{enumerate}
\end{claim}

Intuitively, this claim shows that the distribution of the answer is independent of the choices of $\Omega$. By pre-sampling on $\Omega$, the provers can sample the question indexed on $[r] \setminus C$ independently. This is the crucial property which allows to estimate the post measurement state of the provers conditioning on winning coordinates $W_C$, where recall $W_C$ is the event where the provers win on all the coordinates $C$. 

The following lemma shows that if the event $W_C$ occurs with significant probability then conditioning on $W_C$ only has a moderate effect on the distribution of $(\bX_i, \bY_i)$, on average over a uniformly random choice of $i \in [r] \setminus C$. Furthermore, the distribution of $\bOmega_\mi \blQ_C\blS_C$ is close to being independent from $(\bX_i,\bY_i)$. 

\begin{lemma}[Lemma 4.6 of~\cite{bavarianAnchoredParallelRepetition2021}]
	\label{lem:classical_skew}
	Let $C\subseteq [n]$ be the critical subset which arises from \Cref{prop:subset}, we define the constant
	\begin{equation}\label{eq:delta}
		\text{\gls{PRparameter}} = \frac{1}{r-|C|} \left ( \log \frac{1}{\P(W_C)} + |C| \log |\A|^2 \right)\;.
	\end{equation}
	
	Then following inequalities hold:
	\begin{enumerate}
		\item $ \frac{1}{r-|C|}\sum_{i=1}^m \| \P_{\bOmega_i \bX_i \bY_i | W_C} -  \P_{\bOmega_i \bX_i \bY_i} \| \leq \sqrt{\eta_{\text{PR}}}$.
		\item $ \frac{1}{r-|C|}\sum_{i=1}^m\|  \P_{\bOmega \blQ_C \bX_i \bY_i | W_C} -  \P_{\bOmega \blQ_C |W_C} \P_{\bX_i \bY_i | \bOmega_i } \| \leq \sqrt{\eta_{\text{PR}}}$.
		\item $ \frac{1}{r-|C|}\sum_{i=1}^m\| \P_{\bOmega_i |W_C} \P_{\bOmega_\mi \blQ_C | \bX_i = \dummy, \bY_i = \dummy, W_C} -\P_{\bOmega_i | W_C} \P_{\bOmega_\mi \blQ_C |\bOmega_i  W_C}  \big \| = O(\sqrt{\eta_{\text{PR}}})$.
		\item $ \frac{1}{r-|C|}\sum_{i=1}^m\big\|  \P_{\bX_i \bY_i}  \P_{\bOmega_\mi \blQ_C | \bX_i = \dummy, \bY_i = \dummy, W_C} -  \P_{\bX_i \bY_i}  \P_{\bOmega_\mi \blQ_C | \bX_i, \bY_i, W_C} \big\| = O(\sqrt{\eta_{\text{PR}}})$.
	\end{enumerate}
\end{lemma}

We use $\eta_{\text{PR}}$ to denote the constant given in \eqref{eq:delta} for the critical set $C$ for the remainder of this appendix.

\subsubsection{Notation for the proof of \Cref{thm:anchoringparallelrep}} \label{sec:states-operators}

In this subsection, we set up several notations used for the proof of \Cref{thm:anchoringparallelrep}. We note many notation originates from this section are similar to notations from~\cite[Section 4.4]{bavarianAnchoredParallelRepetition2021}. Recall from the previous subsection that $\strategy^{\otimes r} = \{ \cL^2(\alicealg, \tau), \ket{\psi}, \{A_
\vx^\va \}, \{B_\vy^\vb \} \}$ is a tracially embeddable strategy for $\cG^{\otimes r}$ for some anchoring game $\cG$ which violates the bound given by~\Cref{thm:anchoringparallelrep}. Let $C$ be the subset promise by~\Cref{prop:subset}. For each $(\va_C,\vb_C) \in (\A^{2C}), (\vx, \vy) \in \X^{2n}$, define
\begin{equation}
	\label{eq:states-operators-1}
	A^{\vx}_{\va_C} = \sum_{\va | \va_C} A^{\vx}_{\va} \qquad\text{and} \qquad B^{\vy}_{\vb_C} = \sum_{\vb | \vb_C} B^{\vy}_{\vb} \;,
\end{equation}
where $\va | \va_C$ (resp. $\vb | \vb_C$) indicates summing over all tuples $\va$ consistent with $\va_C$ (resp. $\vb$ consistent with $\vb_C$). We see that the set $\{ A^\vx_{\va_C} \}$ (resp. $\{ B^\vy_{\vb_C} \}$) denotes a POVM with outcomes in the set $\A^C$ (resp. $\B^C$), for all $\vx$ (resp. $\vy$). For all $i \in [r-|C|]$, $\omega_\mi \in \Omega_{\mi}$, $(\vx_C, \vy_C) \in \blQ_C$, and $(x,y) \in \X^2$, define
\begin{equation}
	\label{eq:states-operators-2}
	A^{(\omega_\mi,\vx_C), x}_{\va_C}  = \Ex_{\vx | \omega_\mi,\vx_C , x} A^{\vx}_{\va_C} \qquad \text{and} \qquad  B^{(\omega_\mi, \vy_C), y}_{\vb_C} = \Ex_{\vy | \omega_\mi,\vy_C, y}B^{\vy}_{\vb_C}\;,
\end{equation}
where $\Ex_{\vx | \omega_\mi,\vx_C , x}$ is shorthand for $\Ex_{\vx | \bOmega_\mi = \omega_\mi, \vx_C = \vx_C, \vx_i = x}$ and similarly for $\Ex_{\vy | \omega_\mi, \vy_C, y}$. Let $\X_{/\dummy} = \{ \dummyx : x\in \X\}$ be a disjoint copy of $\X$. Here, for each $x\in \X$, ``$\dummyx$'' is a new symbol that is used to distinguish elements in $\X$ from elements in $\X_{/\dummy}$. 
For all $\dummyx \in \X_{/\perp}$ define 
\begin{equation}
	\label{eq:states-operators-3}
	A^{(\omega_\mi,\vx_C),\dummyx}_{\va_C} \,=\, \eta_{\text{Anchor}} \, A^{(\omega_\mi,\vx_C),\dummy}_{\va_C} + (1 - \eta_{\text{Anchor}}) \, A^{(\omega_\mi,\vx_C),x}_{\va_C}\;.
\end{equation}
We remark that $A^{(\omega_\mi,\vx_C),\dummyx}_{\va_C}$ can be equivalently defined as $\Ex_{\vx | \bOmega_\mi = \omega_\mi, \vx_C = \vx_C, ((\bMplayer)_i, (\bMvalue)_i) = (A, x)}$, and further remark that 
\begin{equation}	\label{eq:states-operators-4}
	A^{(\omega_\mi,\vx_C),\dummy / \dummy}_{\va_C} = 	A^{(\omega_\mi,\vx_C),\dummy}_{\va_C}.
\end{equation}

Using that all operators are positive semidefinite, we observe for later use that
\begin{align}
	A^{(\omega_\mi,\vx_C),\dummy}_{\va_C} &\leq\, \frac{1}{\eta_{\text{Anchor}}}\, A^{(\omega_\mi,\vx_C),\dummyx}_{\va_C}\;,\label{eq:states-operators-3a}\\
	A^{(\omega_\mi,\vx_C),x}_{\va_C} &\leq\, \frac{1}{ 1 - \eta_{\text{Anchor}}}\, A^{(\omega_\mi,\vx_C),\dummyx}_{\va_C}\;.\label{eq:states-operators-3b}
\end{align}
For all $i \in [n] \setminus C$, $\omega_\mi \in \bOmega_\mi$, $\vr_C = (\vx_C, \vy_C, \va_C, \vb_C) \in \bR_C$, $x \in \X$, for all $s \in \{x,\dummyx\}$, and $y \in \X$, define the (unnormalized) state
\begin{equation}\label{eq:def-state-y}
	\text{\gls{Parastatedependonomega}} \,=\, \left(A^{(\omega_\mi, \vr_C),s}_{\va_C}\right)^{1/2} \, \left(B^{(\omega_\mi, \vr_C),y}_{\vb_C}\right)^{1/2} \, \ket{\psi}
\end{equation}
and the normalization factor
\begin{equation} \label{eq:def-gamma}
	\text{\gls{Parastatedependonomeganormalfactor}} \,=\, \big \| \, \ket{\Phi_{(\omega_\mi, \vr_C ),s,y}} \, \big \|~.
\end{equation}
Finally for $\gamma_{(\omega_\mi, \vr_C ),s,y} \neq 0$, we let
\begin{equation}\label{eq:def-psitt}
	\text{\gls{Parastatedependonomeganormal}} = \gamma_{(\omega_\mi, \vr_C ),s,y}^{-1} \, \ket{\Phi_{(\omega_\mi, \vr_C ),s,y}},
\end{equation}
and $\ket{\wt{\Phi}_{(\omega_\mi, \vr_C ),s,y}}  = 0$ otherwise, this denotes the normalized version of the state $\ket{\Phi_{(\omega_{\mi}, \vr_C),s,y}}$. Intuitively, for $s \in \cX$, the state $\ket{\wt{\Phi}_{(\omega_\mi, \vr_C ),s,y}}$ corresponds to Alice and Bob first perform the following pre-sampling procedure for the strategy $\strategy^{\otimes r}$: 
\begin{itemize}
	\item For the coordinates in $[n] \setminus C$ , Alice and Bob have pre-sampled $\bOmega_{\mi} = \omega_\mi$.
	\item For the coordinates in $C \cup \{i\}$, the questions are sample normally.
\end{itemize}
In this case, $\gamma_{(\omega_\mi, \vr_C ),s,y}$ corresponds to the expected probability in which Alice and Bob obtain the answer pair $(\va_C, \vb_C)$ given the pre-sampling procedure above using the strategy $\strategy^{\otimes r}$. The state $\ket{\wt{\Phi}_{(\omega_\mi, \vr_C ),s,y}}$ correspond to the average post-measurement state for obtaining the answer pair $(\va_C, \vb_C)$. 

Similarly, for $s \in \cX_{\dummy}$, the scenario is similar as above, except for the coordinate $i$, Alice and Bob have pre-sampled $\bOmega_i = (A, s)$ as the outcome instead of sampling normally. By \eqref{eq:states-operators-4}
\begin{equation} \label{eq:randomremarkParallel}
	\ket{\wt{\Phi}_{(\omega_\mi, \vr_C ),\dummy / \dummy,y}} = \ket{\wt{\Phi}_{(\omega_\mi, \vr_C ), \dummy,y}}
\end{equation}
\noindent for all $y \in \cX$.

For notational convenience, since many of the operator will be used in the context in which is over an expectation over the variable $i, \omega_{\mi}$ and $\vr_C$. For the clarity of notation, we often omit these subscripts if it is clear from context. For example, for some fix $i$ and the operator $A_{\omega_\mi,\dummy\!/\vx_i}(\va_C)$ expected over $(\omega_{\mi}, \vr_C) \sim (\bOmega_\mi \times \blR_C)$, we will instead write $\Ex_{(\omega_{\mi}, \vr_C)} A_\dummyx$. As another example, for the state $\ket{\wt{\Phi}_{x,y}}$ expected over $(\omega_{\mi}, \vr_C) \sim \bOmega_\mi \times \blR_C$, this is use to represent the state $\ket{\wt{\Phi}_{\sspx}}$. To make the above intuition more concrete, we have the following proposition. 


\begin{proposition}
	\label{prop:gamma}
	For all $s \in \X$,
	\begin{equation*}
		\gamma_{\ssxy} = \Big( \P_{\bA_C \bB_C | \bOmega_\mi = \omega_\mi, \bX_i = x, \bY_i = y}(\va_C,\vb_C) \Big)^{1/2}\;,
	\end{equation*}
	and for all $s = \dummyx \in \X_{/\perp}$,
	\begin{align}
		\nonumber \gamma_{(\omega_{\mi}, \vr_C),s,y} &= \Big( \eta_{\text{Anchor}} \, \P_{\bA_C \bB_C | \bOmega_\mi = \omega_\mi, \bX_i = \dummy, \bY_i = y}(\va_C,\vb_C) + (1 - \eta_{\text{Anchor}}) \, \P_{\bA_C \bB_C  | \bOmega_\mi = \omega_\mi, \bX_i = x, \bY_i = y}(\va_C,\vb_C)\Big)^{1/2} \\
		&= \P_{\bA_C \bB_C | \bOmega_\mi = \omega_\mi, \bOmega_i = (A,x), \bY_i = y}(\va_C,\vb_C)^{1/2}~.\label{eq:gamma-def-3}
	\end{align}
\end{proposition}

The proof of this proposition is similar to~\cite[Proposition 4.9]{bavarianAnchoredParallelRepetition2021} by expanding the definition of $A^{(\omega_\mi, \vr_C),x}_{\va_C}$ and  $B^{(\omega_\mi, \vr_C),y}_{\vb_C}$. 

\subsubsection{Proof of \Cref{thm:anchoringparallelrep}}\label{sec:analysis}

The proof for~\Cref{thm:anchoringparallelrep} requires the following proposition:
\begin{proposition}
	\label{prop:local_operators}
	For every $C\subseteq [r]$, $i\in[r] \backslash C$, $(\omega_{\mi}, \vr_C)$, $x,y\in \X$, there exist two unitary operator $U_{(\omega_{\mi}, \vr_C),x} \in \alicealg$ and $V_{(\omega_{\mi}, \vr_C),y} \in \alicealg'$ such that
	$$
	\Ex_I \,\, \Ex_{(\omega_{\mi}, \vr_C) | W_C} \,\, \Ex_{XY}  \big \| U_{(\omega_{\mi}, \vr_C),x} \, V_{(\omega_{\mi}, \vr_C),y} \ket{\wt{\Phi}_{\sspp}} - \ket{\wt{\Phi}_{\ssxy}} \big\| \,=\, O\big(\eta_{\text{PR}}^{1/16} \big)\;,
	$$
	where $\Ex_I$ denotes the expectation over a uniformly random $i \in [r] \setminus C$, $\Ex_{(\omega_{\mi}, \vr_C) | W_C}$ denotes the expectation over $(\omega_{\mi}, \vr_C)$ sampled from $\P_{(\omega_{\mi}, \vr_C) | W_C}$, and $\Ex_{XY}$ denotes the expectation over $(x,y)$ sampled from $\mu$. 
\end{proposition}

We remark that this is an analogue of~\cite[Proposition 5.1]{bavarianAnchoredParallelRepetition2021} for tracially embeddable strategies, and we give the proof in \Cref{sec:operators}. Based on the above proposition, we give a proof for \Cref{thm:anchoringparallelrep} below.
\begin{proof}
	Let $\cG$ be an anchored game, and supposed that there exist a tracially embeddable strategy$\strategy^{\otimes r} = \{ \cL^2(\alicealg, \tau),\ket{\psi}, \{A^\vx_\va \}, \{B^\vy_\vb \} \}$ which violates \eqref{eq:mainparallel}. Let $C$ be the critical subset $C$ as promised by \Cref{prop:subset}, and recall, without the lost of generality, we assume $C$ is the first $|C|$ coordinates of $[r]$. Fix  $i \in [r] \setminus C$, $(\omega_\mi,\vr_C = (\vx_C, \vy_C, \va_C, \vb_C)) \in (\bOmega_\mi, \bR_C)$ and $(x,y) \in \cX^2$. For each $\va$ such that $\pi_{\leq |C|} (\va) =\va_C$, by \eqref{eq:states-operators-2}, we have $A^{(\omega_{\mi}, \vx_C),x}_\va \leq A^{(\omega_{\mi}, \vx_C),x}_{\va_C}$, and hence by \Cref{lem:inversetrick1}, there exist a positive element $\widehat{A}^{(\omega_{\mi}, \vx_C),x}_{\va}$ such that
	\begin{equation*}
		A^{(\omega_{\mi}, \vx_C),x}_\va  = \left(A^{(\omega_{\mi}, \vx_C),x}_{\va_C}\right)^{1/2} \cdot \widehat{A}^{(\omega_{\mi}, \vx_C),x}_\va \cdot \left(A^{(\omega_{\mi}, \vx_C),x}_{\va_C}\right)^{1/2}.
	\end{equation*}
	Likewise, for each $\vb$ such that $\pi_{\leq |C|} (\vb) =\vb_C$, there exist a positive element $\widehat{B}^{(\omega_{\mi}, \vy_C),y}_\vb$ such that 
	\begin{equation*}
		B^{(\omega_{\mi}, \vy_C),y}_\vb  = \left(B^{(\omega_{\mi}, \vy_C),y}_{\vb_C}\right)^{1/2} \cdot \widehat{B}^{(\omega_{\mi}, \vy_C),y}_\vb \cdot\left(B^{(\omega_{\mi}, \vy_C),y}_{\vb_C}\right)^{1/2} .
	\end{equation*}
	 For fixed $(\va_C, \vb_C)$, both $\{\widehat{A}_{(\omega_{\mi}, \vx_C),x}^\va \}_{\va| \va_C}$ and $\{ \widehat{B}_{(\omega_{\mi}, \vy_C),y}^\vb \}_{\vb| \vb_C}$ forms a POVM, where $\va| \va_C$ (resp. $\vb| \vb_C$) denotes summing over $\va$ such that $\pi_{\leq |C|} (\va) =\va_C$  (resp. $\pi_{\leq |C|} (\vb) =\vb_C$ ). For each $i \in [r] \setminus C$ and $(\omega_\mi,\vr_C = (\vx_C, \vy_C, \va_C, \vb_C)) \in (\bOmega_\mi, \blR_C)$, we define a tracially embeddable strategy $\strategy_{(\omega_\mi,\vr_C)} = \{ \cL^2(\alicealg, \tau),\ket{\wt{\Phi}_\sspp},\{\widetilde{A}^{(\omega_\mi,\vr_C),x}_a\}_{(x,a) \in \cX\times \cA} , \widetilde{B}^{(\omega_\mi,\vr_C),y}_b\}_{(y,b) \in \cX \times \cA} \}$ for the game $\cG$ as the following: The joint entanglement between the prover is $\ket{\wt{\Phi}_\sspp}$ given in~\eqref{eq:def-psitt}. The measurement operators $\{\widetilde{A}^{(\omega_\mi,\vr_C),x}_a  \}$, $\{	\widetilde{B}^{(\omega_\mi,\vr_C),y}_b \}$ for the strategy $\strategy_{(\omega_\mi,\vr_C)}$ is defined as
	\begin{align*}
		\widetilde{A}^{(\omega_\mi,\vr_C),x}_a  = U_{(\omega_{\mi}, \vr_C),x}^\dagger \left(\sum_{\va | \va_i = a, \va_C}  \hat{A}^{(\omega_{\mi}, \vx_C),x}_\va  \right) U_{(\omega_{\mi}, \vr_C),x}\;\\
		\widetilde{B}^{(\omega_\mi,\vr_C),y}_b = V_{(\omega_{\mi}, \vr_C),y}^\dagger\left(\sum_{\vb | \vb_i = b, \vb_C}  \hat{B}^{(\omega_{\mi}, \vy_C),y}_\vb  \right) V_{(\omega_{\mi}, \vr_C),y}\;,
	\end{align*}
	where $\va | \va_i = a, \va_C$ (resp. $\vb | \vb_i = b, \vb_C$) denotes summing over tuples $\va$ that are consistent with $\va_C$ and the ith coordinate is equal to $a$ (resp. $\vb$ that are consistent with $\vb_C$  and the ith coordinate is equal to $b$), and $U_{(\omega_{\mi}, \vr_C),x}$, $V_{(\omega_{\mi}, \vr_C),y}$ are the unitary operators from Proposition~\ref{prop:local_operators}. We remark that we choose to represent the measurement as $\{\widetilde{A}^{(\omega_\mi,\vr_C),x}_a\}_{(x,a) \in \cX \times \cA}$ to emphasize that $\strategy_{(\omega_\mi,\vr_C)}$ is a strategy for a single copy of $\cG$ (i.e. given the question $x \in \cX$, the first prover will measure using $\{\widetilde{A}^{(\omega_\mi,\vr_C),x}_a\}_{a \in \cA}$ on her half of the state $\ket{\wt{\Phi}_\sspp}$). 
	
	We wish to show that on average over $(\omega_\mi,\vr_C)$ there exist a coordinate $i$ such that the strategy $\strategy_{(\omega_\mi,\vr_C)}$ violates the optimal success for the original game. Let $\Qsf_{A B | (\omega_{\mi}, \vr_C), x,y}$ denotes the correlation $C_{x,y,a,b}$ given the strategy $\strategy_{(\omega_{\mi}, \vr_C)}$, or
	\begin{align} 
		\nonumber \Qsf_{A B | (\omega_{\mi}, \vr_C), x,y}(a,b) &=   \braket{\wt{\Phi}_{\bot, \bot}| \tilde{A}^{(\omega_\mi,\vr_C),x}_a \cdot \tilde{B}^{(\omega_\mi,\vr_C),y}_b|\wt{\Phi}_{\bot, \bot}} \\
		&= \sum_{\va | \va_i = a, \va_C}\sum_{\vb | \vb_i = b, \vb_C} \left(  \braket{\wt{\Phi}_{\bot, \bot}|  U_{x}^\dagger  V_{y}^\dagger\left( \widehat{A}^{(\omega_{\mi}, \vx_C),x}_\va \cdot  \widehat{B}^{(\omega_{\mi}, \vx_C),y}_\vb \right) U_{x}  V_{y}|\wt{\Phi}_{\bot, \bot}}\right).
		\label{eq:main-lemma-0}
	\end{align}
	
	The following lemma shows that conditioning on $\vr_C$ selected base on the question/answer pairs which wins on the critical set $C$, the strategy on coordinate $i$ is closed to being independent from the other coordinates. 

	\begin{lemma}
		\label{lem:anchorpr_main_lemma}
		There exists a universal constant $\beta_{PR} \geq 1$ such that,
		\[
		\Ex_I \big \| \P_{(\bOmega_{\mi}, \blR_C) | W_C} \cdot \P_{XY} \cdot \Qsf_{AB | (\omega_{\mi}, \vr_C), x,y} - \P_{(\bOmega_{\mi}, \blR_C) \bX_i \bY_i  \bA_i \bB_i | W_C} \big \| \,\leq \, \beta_{PR}  \,\eta_{\text{PR}}^{1/16} \;,
		\]
		where $\eta_{\text{PR}}$ is defined in~\eqref{eq:delta} and we identify $(x,y,a,b)$ with $(\vx_i,\vy_i,\va_i,\vb_i)$.
		
	\end{lemma}
	\begin{proof}
		We remark that this is essentially the same as~\cite[Lemma 6.2]{bavarianAnchoredParallelRepetition2021}. We start with two claims, from which the proof of the lemma follows. 
		\begin{claim}\label{claim:main-1}
			For all $(\omega_{\mi}, \vr_C), x,y$ and $a,b$,
			\begin{equation*}\label{eq:anchorpr_main-2}
				\sum_{\va | \va_i = a, \va_C}\sum_{\vb | \vb_i = b, \vb_C}  \left(\braket{\wt{\Phi}_{\ssxy} | \big ( \widehat{A}^{(\omega_{\mi}, \vr_C), x}_\va \, \cdot \widehat{B}^{(\omega_{\mi}, \vr_C), y}_\vb \big ) |\wt{\Phi}_{\ssxy}  }\right) \,=\,\P_{\va_i \vb_i | (\omega_{\mi}, \vr_C),x,y}(a,b)\;.
			\end{equation*}
		\end{claim}
	\end{proof}
	
	The proof of this proposition is similar to~\cite[Claim 6.3]{bavarianAnchoredParallelRepetition2021} by expanding the definition of $\widehat{A}^{(\omega_{\mi}, \vr_C), x}_\va$ and  $\widehat{B}^{(\omega_{\mi}, \vr_C), y}_\vb$. 
	
	\begin{claim}\label{claim:main-2}
		The following holds:
		\begin{equation*}\label{eq:main-2-0}
			\Ex_I \big\| \P_{(\omega_{\mi}, \vr_C) |W_C} \cdot \P_{\vx_i \vy_i} \cdot \Qsf_{\va_i \vb_i |(\omega_{\mi}, \vr_C) \vx_i \vy_i } 
			-  \P_{(\omega_{\mi}, \vr_C) |W_C} \cdot \P_{\vx_i \vy_i} \cdot \P_{\va_i \vb_i | (\omega_{\mi}, \vr_C) \vx_i \vy_i} \big \|  
			\,=\,O(\eta_{\text{PR}}^{1/16})\;.
		\end{equation*}
	\end{claim}
	
	\begin{proof}
		Fix $(\omega_{\mi}, \vr_C), x,y$. We bound the total variation distance between the distribution $\Qsf$ and $\P$ below. For the ease of notation, we define
		\begin{equation*}
			\what{A}^{x}_a =  \sum_{\va | \va_i = a, \va_C} \widehat{A}^{(\omega_{\mi}, \vr_C), x}_\va \qquad  	\what{B}^{y}_b =  \sum_{\vb | \vb_i = b, \vb_C} \widehat{A}^{(\omega_{\mi}, \vr_C), y}_\vb,
		\end{equation*}
	for this proof. We see that $\{	\what{A}^{x}_a \}_{a \in \cA}$ and $\{	\what{B}^{y}_b \}_{b \in \cA}$ forms two sets of POVM. By \Cref{eq:main-lemma-0} and \Cref{claim:main-1} 
		\begin{align}
			&\big \| \Qsf_{\va_i \vb_i | (\omega_{\mi}, \vr_C), x,y} - \P_{\va_i \vb_i | (\omega_{\mi}, \vr_C),x,y}\big \|	\notag\\
			&= \frac{1}{2} \sum_{a,b}\big | \braket{\wt{\Phi}_{\bot, \bot}|  U_{x}^\dagger  V_{y}^\dagger\big ( \what{A}^{x}_a \cdot\what{B}^{y}_b \big ) U_{x}  V_{y}|\wt{\Phi}_{\bot, \bot}} -  \braket{\wt{\Phi}_{x,y} | \big ( \what{A}^{x}_a \cdot\what{B}^{y}_b \big ) |\wt{\Phi}_{x,y}  }	\big |
			\notag\\
			&= \frac{1}{2} \sum_{a,b}\big | \braket{\wt{\Phi}_{\bot, \bot}|  U_{x}^\dagger  V_{y}^\dagger \big( \what{A}^{x}_a \cdot\what{B}^{y}_b \big ) U_{x}  V_{y}|\wt{\Phi}_{\bot, \bot}} - \braket{\wt{\Phi}_{\bot, \bot}|  U_{x}^\dagger  V_{y}^\dagger\big ( \what{A}^{x}_a \cdot\what{B}^{y}_b \big) |\wt{\Phi}_{x, y}} \notag \\
			&\quad + \braket{\wt{\Phi}_{\bot, \bot}|  U_{x}^\dagger  V_{y}^\dagger \big( \what{A}^{x}_a \cdot\what{B}^{y}_b \big) |\wt{\Phi}_{x, y}} -	\braket{\wt{\Phi}_{x,y} | \big( \what{A}^{x}_a \cdot\what{B}^{y}_b \big) |\wt{\Phi}_{x,y}  }	\big | \notag \\
			&\leq \frac{1}{2}  \sum_{a,b} \left( \big \| \big( \what{A}^{x}_a \cdot\what{B}^{y}_b \big) U_{x}  V_{y}\ket{\wt{\Phi}_{\bot, \bot}}  \big \|  + \big \| \big( \what{A}^{x}_a \cdot\what{B}^{y}_b \big)\ket{\wt{\Phi}_{x, y}}  \big \| \right) \cdot  \big\| U_{x}  V_{y} \ket{\wt{\Phi}_{\bot,\bot}} - \ket{\wt{\Phi}_{x,y}} \big\| \notag \\
			&=  \frac{1}{2}  \left(  \sum_{a,b} \big \| \big( \what{A}^{x}_a \cdot\what{B}^{y}_b \big) U_{x}  V_{y}\ket{\wt{\Phi}_{\bot, \bot}}  \big \|  +  \sum_{a,b} \big \| \big( \what{A}^{x}_a \cdot\what{B}^{y}_b \big)\ket{\wt{\Phi}_{x, y}}  \big \| \right) \cdot  \big\| U_{x}  V_{y} \ket{\wt{\Phi}_{\bot,\bot}} - \ket{\wt{\Phi}_{x,y}} \big\| \notag \\
			&\leq \big\| U_{x}  V_{y} \ket{\wt{\Phi}_{\bot,\bot}} - \ket{\wt{\Phi}_{x,y}} \big\| 	\noindent \,\;.\label{eq:main-2-1}
		\end{align}
		Where the last line follows from Jensen's inequality and  $\{	\what{A}^{x}_a \}_{a \in \cA}$ and $\{	\what{B}^{y}_b \}_{b \in \cA}$ being both POVM. Thus, returning to \eqref{eq:main-2-0}
		\begin{align*}
			&\Ex_I \,\, \big\| \P_{(\omega_{\mi}, \vr_C) |W_C} \cdot \P_{\vx_i \vy_i} \cdot \Qsf_{\va_i \vb_i |(\omega_{\mi}, \vr_C) \vx_i \vy_i } 
			-  \P_{(\omega_{\mi}, \vr_C) |W_C} \cdot \P_{\vx_i \vy_i} \cdot \P_{\va_i \vb_i | (\omega_{\mi}, \vr_C) \vx_i \vy_i} \big \|  \\
			&\quad = \Ex_I \, \, \Ex_{(\omega_{\mi}, \vr_C) |W_C} \,\, \Ex_{\vx_i \vy_i} \,\, \big \| \Qsf_{\va_i \vb_i |(\omega_{\mi}, \vr_C) \vx_i \vy_i} - \P_{\va_i \vb_i | (\omega_{\mi}, \vr_C) \vx_i \vy_i} \big \|    \\
			&\quad \leq \sqrt{2} \, \Ex_I \, \, \Ex_{(\omega_{\mi}, \vr_C) |W_C} \,\, \Ex_{XY} \big\| U_{(\omega_{\mi}, \vr_C),x} \otimes V_{(\omega_{\mi}, \vr_C),y} \ket{\wt{\Phi}_\sspp} - \ket{\wt{\Phi}_{\ssxy}} \big\|\\
			&\quad \leq O(\eta_{\text{PR}}^{1/16})\;,
		\end{align*}
		where the first inequality is by~\eqref{eq:main-2-1} and the last inequality follows from \Cref{prop:local_operators}.
	\end{proof}
	
	
	Based on \Cref{lem:anchorpr_main_lemma}, we are ready to construct a strategy for the single instance of the non-local game $\cG$ which creates the contradiction. We remark that the remainder of this proof follows similarly as~\cite[Section 6.2]{bavarianAnchoredParallelRepetition2021}. Recall by the definition of the critical set $C$ in~\Cref{prop:subset}, we have
	\begin{equation*}
		\Ex_I\, \P (W_i | W_C) \,\geq \, 1 - \eps/2\;,
	\end{equation*}
	and recall from \Cref{claim:ufacts}, since sampling each $\bOmega_i$ are independent from each coordinates, and the answers are independent from the distribution $\bOmega_i$ for $i \in [r] \setminus (C \cup \{i\})$. This implies by sampling a uniformly random $i \in [r] \setminus C$ and then sampling from the distribution $\P_{\vx_i \vy_i (\omega_{\mi}, \vr_C)  \va_i \vb_i| W_C}$ yields a tuple $(i,\vx_i,\vy_i,(\omega_{\mi}, \vr_C),\va_i,\vb_i)$ such that $V(\vx_i,\vy_i,\va_i,\vb_i) = 1$ (i.e $W_i = 1$) with probability at least $1 - \eps/2$. By \Cref{lem:anchorpr_main_lemma}, the distribution $\P_{\bX_i \bY_i (\omega_{\mi}, \vr_C)  \bA_i \bB_i| W_C}$, is $\beta_{PR}  \,\eta_{\text{PR}}^{1/16}$ close to Alice and Bob performing the following strategy for the non-local game $\cG$

	\begin{enumerate}
		\item Post-select the question and answer pair $r_{C} \sim \blR_C$ for the critical set $C$ in which they win on all $C$ coordinates for the strategy $\strategy^n$, uniformly a coordinates $i \in [r] \setminus c$ and $(\omega_\mi) \sim \Omega_{\mi}$.  
		\item Up on receiving $(x,y)$, perform the strategy $\strategy_{(\omega_{\mi}, \vr_C)}$.
	\end{enumerate}
	By convexity, there exist some coordinate $i$ in which we can fix on step 1 of the above procedure which succeed with probability at least $1 - \eps/2 - \beta_{PR}  \eta_{\text{PR}}^{1/16}$, where $\beta_{PR} $ is the universal constant from \Cref{lem:anchorpr_main_lemma}. Let the strategy define above be $\strategy^{\text{contra}}$. Set
	\begin{equation*}
		c = \frac{1}{32 \, \log(e) \, (4\beta_{PR} )^{16}}. 
	\end{equation*}
	By the initial contradiction assumption, we have
	\begin{equation*}
		\omega(\cG^{\otimes r}, \strategy^{\tensor r}) \geq \frac{4}{\epsilon} \exp \left ( - \frac{c \, \eps^{17} \, r}{ \log (|\cA|+1) } \right),
	\end{equation*}
	and by rearranging the equation for $r$, we have
	\begin{equation*}
		\frac{\log (|\cA|+1)}{c \cdot \eps^{17}} \ln{\left(\frac{4}{\eps \cdot \omega(\cG^{\otimes r}, \strategy^{\tensor r} )} \right)} \leq n. 
	\end{equation*}
	Hence, we have 
	\[
	r \geq \frac{r}{\log (|\cA|+1)} \geq \frac{1}{c \cdot \, \eps^{17}} \,  \ln{\left(\frac{4}{\eps \cdot \omega(\cG^{\otimes r}, \strategy^{\tensor r} )} \right)} \geq \frac{16}{\eps} \log{\left(\frac{4}{\eps \cdot \omega(\cG^{\otimes r}, \strategy^{\tensor r} )} \right)}\;,
	\]
	where we use that $0 < \eps \leq 1$, and $0 < c \leq \frac{\eps^{16}}{16 \cdot \log(e)}$. Hence, if we consider $\eta_{\text{PR}}$, using the fact that $|C| \leq \frac{8}{\eps} \log \frac{4}{\eps \cdot \P(W)} \leq r/2$ from \Cref{prop:subset},
	\begin{align*}
		\eta_{\text{PR}} &= \frac{1}{r-|C|} \left ( \log \frac{1}{\P(W_C)} + |C| \log |\A|^2 \right)\;\\
		& \leq \frac{2}{n} \Big ( \frac{16 \cdot  \log |\A|^2}{\eps} \log \frac{4}{\eps \cdot \P(W)} \Big) \\
		&\leq \frac{2}{n} \cdot \frac{16 \cdot s}{\eps} \cdot \frac{c \log (e) \, \eps^{17} \, n}{s} \\
		&= 32 \, c \, \log(e) \, \eps^{16}~.
	\end{align*}
	This implies that 
	\begin{equation}
		\omega(\cG, \strategy^{\text{contra}}) \geq 1 - \eps/2 - \beta_{PR} \cdot  \eta_{\text{PR}}^{1/16} \geq   1 - \eps/2 - \beta_{PR}  \eta_{\text{PR}} > 1 - \eps, 
	\end{equation}
	\noindent giving a strategy which wins $\cG$ is strictly greater than $1 - \eps$, a contradiction. This concludes the proof for~\Cref{thm:anchoringparallelrep}.
\end{proof}

\subsection{Existence of the local unitary}\label{sec:operators}

In this section, we give a proof for~\Cref{prop:local_operators}. Similar to the proof of~\cite[Proposition 5.1]{bavarianAnchoredParallelRepetition2021}, the proof relies on two main lemmas. The first lemma, given below, guarantees the existence of a local unitary operator which allows the provers to make the local adjustment as per described in~\Cref{sec:paralleloverview}. We remark that this is an commuting operator model variant of~\cite[Proposition 5.1]{bavarianAnchoredParallelRepetition2021}, and the proof follows a similar structure. 
\begin{lemma}
	\label{lem:unitary_bounds}
	For all $i$, $(\omega_{\mi}, \vr_C)$, $x$ and $y$ there exists two unitary operator $U_{(\omega_{\mi}, \vr_C), x} \in \alicealg$ and $V_{(\omega_{\mi}, \vr_C), y} \in  \alicealg'$ such that with probability at least $1 - O(\eta_{\text{PR}}^{1/16})$ over the choice of a uniformly random $i \in [n] \setminus C$, 
	\begin{align}
	\Ex_{\bOmega_{-i} \blR_C |  W_C} \, \, \, \Ex_{\bX} \, \, \, \big \|U_{(\omega_{\mi}, \vr_C), x}  \ket{\wt{\Phi}_\sspp} - \ket{\wt{\Phi}_\ssxp }   \big \| &=  O(\eta_{\text{PR}}^{1/16})\;,\label{eq:ux_bound}\\
 	\Ex_{\bOmega_{-i} \blR_C |  W_C} \, \, \, \Ex_{Y} \, \, \, \big \| V_{(\omega_{\mi}, \vr_C),y}  \ket{\wt{\Phi}_\sspp}  - \ket{\wt{\Phi}_\sspy }  \big \| &=  O(\eta_{\text{PR}}^{1/16})\;,\label{eq:vy_bound}\\
	\Ex_{\bOmega_{-i} \blR_C |  W_C} \, \, \, \Ex_{\bX \bY}  \,\, \, \big \| V_{\ssxy}  \ket{\wt{\Phi}_\sspxy}  - \ket{\wt{\Phi}_\sspxp }  \big \| &=  O(\eta_{\text{PR}}^{1/16})\;.\label{eq:vxy_bound}
	\end{align}
	where $\Ex_{\bX}$, $\Ex_{\bY}$, and $\Ex_{\bX \bY}$ denote expectations under $\mu_X(x)$, $\mu_Y(y)$, and $\mu(x,y)$ respectively.
\end{lemma}
The second lemma, given below, relates the normalization factors $\gamma_\xy$ and $\gamma_{\pxy}$ defined in~\eqref{eq:def-gamma}. Since the second lemma is identical to~\cite[Lemma 5.17]{bavarianAnchoredParallelRepetition2021}, we instead refer the reader to the original reference for the proof. 

\begin{lemma}\label{lem:norms-are-close}
	With probability at least $1 - O(\eta_{\text{PR}}^{1/4})$ over the choice of $i \in [n] \setminus C$, 
	\begin{equation}
		\label{eq:norms-are-close-s0}
		\Ex_{\bX \bY} \, \Ex_{\bOmega_{\mi} \bR_C | \bX_i = \dummy,\bY_i = \dummy, W_C} \, \Big | 1 - \frac{\gamma_{(\omega_{\mi}, \vr_C), x,y}}{\gamma_{(\omega_{\mi}, \vr_C),\dummy,\dummy}} \Big |^2 \leq O(\eta_{\text{PR}}^{1/4})\;,
	\end{equation}
	and
	\begin{equation}
		\label{eq:norms-are-close-s1}
		\Ex_{\bX \bY} \, \Ex_{\bOmega_{\mi} \bR_C | \bX_i = \dummy, \bY_i = \dummy, W_C} \, \Big | 1 - \frac{\gamma_\sspxy}{\gamma_{(\omega_{\mi}, \vr_C),\dummy,\dummy}} \Big |^2 \leq O(\eta_{\text{PR}}^{1/4})\;.
	\end{equation}
\end{lemma}

\subsubsection{Local operator lemma}
In this subsection, we give a proof for~\Cref{lem:unitary_bounds}. Recall from the previous section that the strategy $\strategy^{\otimes r} = \{ \cL^2(\alicealg, \tau), \ket{\psi} = \sigma \ket{\tau}, \{A_\vx^\va \}, \{B_\vy^\vb \} \}$ is a tracially embeddable strategy which realizes a contradiction in~\Cref{lem:anchortrans}. For all $\omega$, $(\vx_C, \vy_C)$, $\va_C$, and $\bB_C$, we define the measurement operator
\begin{equation}
	\label{eq:a-bomega-def}
	A^{\omega, (\vx_C, \vy_C)}_{\va_C} = \Ex_{\bX | \bOmega = \omega, \blQ_C = (\vx_C, \vy_C)} \, A^{\vx}_{\va_C} \qquad \text{and} \qquad B^{\omega, (\vx_C, \vy_C)}_{\vb_C} =  \Ex_{\bY | \bOmega = \omega, \blQ_C = (\vx_C, \vy_C)} \, B^\vy_{\vb_C}
\end{equation}
where $A^{\vx}_{\va_C},B^\vy_{\vb_C}$ are defined in~\eqref{eq:states-operators-1}. For all $\omega,\vx,\vy,\va_C,\vb_C$, we define the vector state
\begin{align}
	\label{eq:Xidef}\ket{\Xi_{\omega,\vx_C, \vy,\va_C,\vb_C}} &= \left(A^{\omega, (\vx_C, \vy_C)}_{\va_C} \right)^{1/2} \, \left(B^\vy_{\vb_C} \right )^{1/2} \ket{\psi} \\
	\ket{\Xi_{\omega,\vx_C, \vy,\va_C,\vb_C}} &= \left(A^{\omega, (\vx_C, \vy_C)}_{\va_C} \right)^{1/2} \, \left(B^\vy_{\vb_C} \right )^{1/2} \ket{\psi}, 
\end{align}
where $\ket{\Xi_{\omega,\vx_C, \vy,\va_C,\vb_C}}, \ket{\Lambda_{\omega,\vx,\vy_C, \va_C,\vb_C}} \in \cL^2(\alicealg, \tau)$. We remark that in the above definition, $A^{\omega, (\vx_C, \vy_C)}_{\va_C}$ and $B^\vy_{\vb_C}$ uses the same value of $\vy$. Let $\Xi_{\omega,\vx_C, \vy,\va_C,\vb_C}$ and $\Lambda_{\omega,\vx,\vy_C, \va_C,\vb_C}$ denote the abstract normal states acting on $\cL^2(\alicealg, \tau)$ specified by the vector $\ket{\Xi_{\omega,\vx_C, \vy,\va_C,\vb_C}}$ and $\ket{\Lambda_{\omega,\vx,\vy_C, \va_C,\vb_C}}$, respectively, i.e. 
\begin{align*}
\Xi_{\omega,\vx_C, \vy,\va_C,\vb_C} (A) &= \braket{\Xi_{\omega,\vx_C, \vy,\va_C,\vb_C}| A |\Xi_{\omega,\vx_C, \vy,\va_C,\vb_C}} \\
\Lambda_{\omega,\vx,\vy_C, \va_C,\vb_C} (A) &= \braket{\Lambda_{\omega,\vx,\vy_C, \va_C,\vb_C}| A |\Lambda_{\omega,\vx,\vy_C, \va_C,\vb_C}}
\end{align*}
Now define the classical-quantum states 
\begin{align}
	\label{eq:Xi-def} \text{\gls{CQcond1}} &= \sum_{\omega, \vx_C,\vy, \va_C, \vb_C} \P_{\bOmega \bX_C \bY} (\omega,\vx_C, \vy) \, \bra{\omega, \vx_C, \vy, \va_C, \vb_C}\otimes \Xi_{\omega,\vx_C, \vy,\va_C,\vb_C} ~,   \\
	\label{eq:Lambda-def} \text{\gls{CQcond2}} &= \sum_{\omega, \vx, \vy_C, \va_C, \vb_C} \P_{\bOmega \bX \bY_C} (\omega, \vx, \vy_C) \, \bra{\omega, \vx, \vy_C, \va_C, \vb_C} \otimes \Lambda_{\omega,\vx,\vy_C, \va_C,\vb_C}~,
\end{align}
Both states are classical on the space $\bOmega$, $\bX$, $\bX_C$, $\bY$, $\bY_C$ and $\ac$ and quantum on the space $\alicealg$.We remark that all classical register listed above are multiple classical registers, as each of them are probability distribution over multiple coordinates (i.e. $\bX = (\bX_0, \cdots, \bX_{r-1})$). Observe that this state looks very similar -- but not quite -- to the one that occurs in an actual execution of the strategy $\strategy^{\otimes r}$. There are several important differences: one is that the measurements only produce answers for the coordinates indexed by $C$. Another difference is that the measurement operators $A^{\omega, (\vx_C, \vy_C)}_{\va_C}$ are not part of the strategy but instead are derived from the measurements operator. Also, note that there is no explicit register for question vector $\bX$ (except for the $\bX_C$ questions, which are included in the register $\bOmega$); instead these questions are implicitly averaged over within the $A^{\omega, (\vx_C, \vy_C)}$ measurement for a fix value of $(\omega, \vx_C, \vy_C)$. 

The state $\Xi^{\bOmega \bX_C \bY \ac \alicealg}$ is defined such that when restricted to the classical register $\Xi^{\bOmega \bX_C \bY \ac}$, the resulting state representing the probability distribution $\P_{\bOmega \bX_C \bY \bA_C \bB_C}$. To see this, observe that for any $(\omega,\vx_C, \vy,  \va_C, \vb_C)$, we have
\begin{align*}
	\Xi^{\bOmega \bX_C \bY \ac} &= \sum_{\omega, \vx_C,\vy, \va_C, \vb_C} \P_{\bOmega \bX_C \bY}   (\omega,\vx_C, \vy) \, \cdot \bra{\omega, \vx_C, \vy, \va_C, \vb_C} \otimes \Xi_{\omega,\vx_C, \vy,\va_C,\vb_C}(\cI_{\alicealg}) ~,   \\ \\
	&= \sum_{\omega, \vx_C,\vy, \va_C, \vb_C} \left(\P_{\bOmega \bX_C \bY}  (\omega,\vx_C, \vy)  \braket{\psi| A^{\omega, (\vx_C, \vy_C)}_{\va_C} B^\vy_{\vb_C}|\psi  } \right) \cdot  \, \bra{\omega, \vx_C, \vy, \va_C, \vb_C}  ~,    \\
	&= \sum_{\omega, \vx_C,\vy, \va_C, \vb_C}\P_{\bOmega \bX_C \bY}  (\omega,\vx_C, \vy)   \left( \Ex_{\bX | \bOmega = \omega, \bX_C = \vx_C,\bY= \vy} \braket{\psi|  A^{\vx}_{\va_C}  B^\vy_{\vb_C}|\psi} \right)  \cdot  \, \bra{\omega, \vx_C, \vy, \va_C, \vb_C}  ~,    \\
	&= \sum_{\omega, \vx_C,\vy, \va_C, \vb_C}\P_{\bOmega \bX_C \bY}  (\omega,\vx_C, \vy)  \left(  \Ex_{\bX | \bX_C = \vx_C,\bY= \vy} \P_{\bA \bB | \bX = \vx, \bY = \vy}(\va,\vb) \right)  \cdot  \, \bra{\omega, \vx_C, \vy, \va_C, \vb_C}  ~,    \\
	&= \sum_{\omega, \vx_C,\vy, \va_C, \vb_C}  \left( \P_{\bOmega \bX_C \bY \bA_C \bB_C}  (\omega,\vx_C, \vy, \va_C, \vb_C) \right)  \cdot  \bra{\omega, \vx_C, \vy, \va_C, \vb_C}  ~,    
\end{align*}
where in the third line we used Item 1 from Claim~\ref{claim:ufacts} and in the fourth line we used Item 2 from Claim~\ref{claim:ufacts} and each $\bY_i$ being independent of each other. By a similar calculation, the state $\Lambda^{\bOmega \bX \bA_C \bB_C}$ represents the probability distribution $\P_{\bOmega \bX \bY_C \bA_C \bB_C}$. Since both $\P_{\bOmega \bX_C \bY \bA_C \bB_C}$ and $\P_{\bOmega \bX \bY_C \bA_C \bB_C}$ are probability distributions, the state $\Xi^{\bOmega \bX_C \bY \ac} $ and $\Lambda^{\bOmega \bX \bY_C \ac}$ are indeed classical quantum states. 


Recall, the event $W_C$ corresponds to the event where the provers produces an winning answer given a winning question pair on all coordinates on the critical set $C$. Since the event $W_C$ is determined by the random variables $(\bOmega, \blR_C)$ we can condition the states $\Xi^{\bOmega \bX_C \bY \ac \alicealg},  \Lambda^{\bOmega \bX \bY_C  \ac \alicealg}$ on the event $W_C$ to obtain states
\begin{align*}
	\text{\gls{CQcond3}} &= \frac{1}{\P(W_C)} \sum_{\substack{\omega, \vx_C, \vy, \va_C, \vb_C: \\ (\vx_C, \vy_C,\va_C,\vb_C) \in W_C}} \P_{\bOmega \bX_C \bY} (\omega,\vx_C, \vy) \,  \cdot \bra{\omega, \vx_C, \vy, \va_C, \vb_C} \otimes \Xi_{\omega,\vx_C, \vy,\va_C,\vb_C} ~,   \\
	\text{\gls{CQcond4}} &= \frac{1}{\P(W_C)} \sum_{\substack{\omega, \vx_C, \vy, \va_C, \vb_C: \\ (\vx_C, \vy_C,\va_C,\vb_C) \in W_C}}  \, \P_{\bOmega \bX \bY_C} (\omega, \vx, \vy_C) \,\cdot \bra{\omega, \vx, \vy_C, \va_C, \vb_C} \otimes \Lambda_{\omega,\vx,\vy_C, \va_C,\vb_C}~,
\end{align*}
Since the event $W_C$ is a subset of all possible coordinates in $\cX^{2|C|} \times \cA^{2|C|}$, by definition, we have
\begin{equation}
	\P(W_C) \cdot \xi^{\bOmega \bX_C \bY \ac \alicealg}  \leq  \Xi^{\bOmega \bX_C \bY \ac \alicealg} \, \qquad 	\P(W_C) \cdot  \lambda^{\bOmega \bX \bY_C  \ac \alicealg}  \leq \Lambda^{\bOmega \bX \bY_C  \ac \alicealg}. \label{eq:Xitoxiinequality}
\end{equation} 
For a fix $\omega$ and $\vr_C = (\vx_C, \vy_C, \va_C, \vb_C) \in \cX^{2|C| \cA^{2|C|}}$, we write $\xi^{\bOmega \bX_C \bY \ac \alicealg}_{(\omega, \vr_C)}$ as the (normalize) state 
\begin{equation*}
	\Xi^{\bOmega \bX_C \bY \ac \alicealg}_{(\omega, \vr_C)} =  \sum_{\vy: \vy|_{C} = \vy_C} \P_{\bOmega \bX_C \bY} (\omega,\vx, \vy_C) \,  \cdot \bra{\omega, \vx_C, \vy, \va_C, \vb_C} \otimes \Xi_{\omega,\vx, \vy_C,\va_C,\vb_C} ~,  \\
\end{equation*}
In other words, the state $\Xi^{\bOmega \bX_C \bY \ac \alicealg}_{(\omega, \vr_C)}$ is equivalent to the quantum-classical $\xi^{\bOmega \bX_C \bY \ac \alicealg}$ restricted to the component where $\Omega = \omega$ and $\blR_C = \vr_C$ for all coordinates in $C$. We define the state $\Xi^{\bOmega \bX_C \bY \ac \alicealg}_{\omega, \vx_C, \vy_C}$ as the same conditioning above, but only for fixed $\vx_C, \vy_C$ value, and we define the state $\Lambda^{\bOmega \bX_C \bY \ac \alicealg}_{(\omega, \vr_C)}$ and $\Lambda^{\bOmega \bX_C \bY \ac \alicealg}_{\omega, \vx_C, \vy_C}$ in a similar manner as above. 

Likewise, for $\vr_C = (\vx_C, \vy_C, \va_C, \vb_C) \in W_C$, we define the (normalized) state $\xi^{\bOmega \bX_C \bY \ac \alicealg}_{(\omega, \vr_C)}$
\begin{equation*}
	\xi^{\bOmega \bX_C \bY \ac \alicealg}_{(\omega, \vr_C)} = \frac{1}{\P(W_C)} \sum_{\vy: \vy|_{C} = \vy_C} \P_{\bOmega \bX_C \bY} (\omega,\vx, \vy_C) \,  \cdot \bra{\omega, \vx, \vy_C, \va_C, \vb_C} \otimes \Xi_{\omega,\vx_C, \vy,\va_C,\vb_C} ~,  \\
\end{equation*}
and $\xi^{\bOmega \bX_C \bY \ac \alicealg}_{(\omega, \vr_C)} = 0$  if $\vr_C \not\in W_C$. We define the state $\lambda^{\bOmega \bX_C \bY \ac \alicealg}_{(\omega, \vr_C)} $ in a similar manner as above. By definition, for a fixed $\omega$ and $\vr_C = (\vx_C, \vy_C, \va_C, \vb_C) \in W_C$, we have $\xi^{\bOmega \bX_C \bY \ac \alicealg}_{(\omega, \vr_C)} = \Xi^{\bOmega \bX_C \bY \ac \alicealg}_{(\omega, \vr_C)}$. Similarly to the proof of~\cite[Lemma 5.12]{bavarianAnchoredParallelRepetition2021}, the main step of proving~\Cref{lem:unitary_bounds} is given by two claims which build on top of each other. The following claim is an analogue of~\cite[Claim 5.13]{bavarianAnchoredParallelRepetition2021} for the commuting operator model. We remark that the proof is rewritten for clarity.  

\begin{claim}\label{claim:xi-change-x}
	\begin{align}
		\label{eq:xi-change-1}\Ex_{i \sim [r] \setminus C} \,  \Ex_{\bOmega \blR_C| W_C}\, \I \big (\bY_i ; \alicealg \big )_{	\xi^{\bOmega \bX_C \bY \ac \alicealg}_{(\omega, \vr_C)}} \,&= \, O(\eta_{\text{PR}})~,  \\
		\label{eq:xi-change-2} \Ex_{i \sim [r] \setminus C} \,  \Ex_{\bOmega \blR_C| W_C} \,\I \big (\bX_i ; \bobalg \big )_{	\lambda^{\bOmega \bX \bY_C  \ac \alicealg}_{(\omega, \vr_C)} } \,&= \, O(\eta_{\text{PR}})~. 
	\end{align}
\end{claim}

\begin{proof}
	We present the proof for~\eqref{eq:xi-change-1}; the proof for~\eqref{eq:xi-change-2} follows from a similar calculation. First, for a fixed $\omega$ and $\vr_C = (\vx_C, \vy_C, \va_C, \vb_C) \in W_C$,  $\xi^{\bOmega \bX_C \bY \ac \alicealg}_{(\omega, \vr_C)} = \Xi^{\bOmega \bX_C \bY \ac \alicealg}_{(\omega, \vr_C)}$. Hence, by rearranging~\eqref{eq:xi-change-1},
	\begin{align}
		\Ex_{i \sim [r] \setminus C} \,  \Ex_{\bOmega \blR_C| W_C} \I \big (\bY_i ; \alicealg \big )_{\xi^{\bOmega \bX_C \bY \ac \alicealg}_{(\omega, \vr_C)}} &= \Ex_{i \sim [r] \setminus C} \,  \Ex_{\bOmega \blR_C| W_C} \D \big (\xi^{\bY_i \alicealg}_{(\omega, \vr_C)} \big \| \Xi_{(\omega, \vr_C)} ^{\bY_i} \otimes \Xi_{(\omega, \vr_C)} ^{\alicealg} \big )\notag\\
		\notag &\leq \Ex_{i \sim [r] \setminus C} \,  \Ex_{\bOmega \blR_C| W_C} \D \big (\xi^{\bY \alicealg}_{(\omega, \vr_C)} \big \| \Xi_{(\omega, \vr_C)} ^{\bY} \otimes \Xi_{(\omega, \vr_C)} ^{\alicealg} \big ) \\
		\label{eq:claim1a} &\leq \Ex_{i \sim [r] \setminus C} \,  \Ex_{\bOmega \blQ_C| W_C} \D \big (\xi^{\bY \alicealg}_{\omega, \vx_C, \vy_C} \big \| \Xi_{(\omega, \vr_C)} ^{\bY} \otimes \Xi_{\omega, \vx_C, \vy_C}^{\alicealg} \big ) 
	\end{align}
	Where the second line follows from~\Cref{prop:divergence_data_processing} and the third line follows from~\Cref{prop:divergence_chain_rule} on the distribution $\blS_C$. We wish to use~\Cref{prop:relative_min_entropy_chain_rule2} to bound the inequality. For all $\omega$ and $(\vx_C, \vy_C) \in \cX^{|C|^2}$
	\begin{align}
		\Xi^{\bY \ac \alicealg}_{\omega, \vx_C, \vy_C}  &= \sum_{\vy, \va_C, \vb_C} \P_{\bY | \omega, (\vx_C, \vy_C)}(\vy) \,\, \bra{\vy} \otimes \bra{\va_C \vb_C} \otimes \Xi_{\omega,\vx_C, \vy,\va_C,\vb_C}^{\alicealg}  \notag \\
		&\leq  \sum_{\vy} \P_{\bY | \omega}(\vy) \,\, \bra{\vy} \otimes \Tr_{|\cA|^{2|C|}} \otimes\left( \sum_{\va_C,\vb_C} \Xi_{\omega,\vx_C, \vy,\va_C,\vb_C}^{\alicealg} \right) \notag \\
		&= \Xi_{\omega , \vx_C, \vy_C}^{\bY} \otimes \Tr_{|\cA|^{2|C|}} \otimes \Xi^{\alicealg}_{\omega, \vx_C, \vy_C}   \;,\label{eq:claim3-2} 
	\end{align}
	where the second line follows since the state $\bra{\vy} \leq \Tr_{|\cA|^{2|C|}}$, and the third line follows because $\sum_{\va_C,\vb_C} \Xi_{\omega,\vx_C, \vy,\va_C,\vb_C}^{\alicealg}(\cI_{\alicealg}) = 1$ (by~\eqref{eq:Xidef} where both of the measurement operator $A$ and $B$ are POVMs). By considering the above state on the restriction of $\bM_{|\Y|^{r}} \otimes \cI_{|\cA|^{2|C|}} \otimes \alicealg$,
	\begin{equation*}
		\Xi^{\bY  \alicealg}_{\omega, \vx_C, \vy_C} \leq (|\cA|)^{2|C|} \cdot \left(\Xi_{\omega , \vx_C, \vy_C}^{\bY} \otimes \Xi^{\alicealg}_{\omega, \vx_C, \vy_C}\right).
	\end{equation*}
	Hence, by combining~\Cref{prop:relative_min_entropy_chain_rule2} and~\eqref{eq:claim1a}
	\begin{align*}
			\Ex_{i \sim [r] \setminus C} \,  \Ex_{\bOmega \blR_C| W_C} \I \big (\bY_i ; \alicealg \big )_{\xi^{\bOmega \bX_C \bY \ac \alicealg}_{(\omega, \vr_C)}}	&\leq \frac{1}{r-|C|} \left(  \Ex_{\bOmega \blQ_C| W_C} \D \big (\xi_{\omega, \vx_C, \vy_C}^{\bY  \alicealg} \big \| \Xi_{\omega, \vx_C, \vy_C}^{  \bY  \alicealg} \big) +  |C| \cdot \log{|\cA|^2} \, \right).\notag\\
			&\leq \frac{1}{r-|C|} \left( \Ex_{\bOmega \blQ_C| W_C} \D \big (\xi_{\omega, \vx_C, \vy_C}^{\bOmega \bX_C \bY \ac \alicealg} \big \| \Xi_{\omega, \vx_C, \vy_C}^{\bOmega \bX_C \bY \ac \alicealg} \big) +  |C| \cdot \log{|\cA|^2} \, \right),\notag\\
				&\leq \frac{1}{r-|C|} \left( \Ex_{\bOmega \blQ_C} \D \big (\xi_{\omega, \vx_C, \vy_C}^{\bOmega \bX_C \bY \ac \alicealg} \big \| \Xi_{\omega, \vx_C, \vy_C}^{\bOmega \bX_C \bY \ac \alicealg} \big) +  |C| \cdot \log{|\cA|^2} \, \right),\notag\\
			&\leq \frac{1}{r-|C|} \left( \D \big (\xi^{\bOmega \bX_C \bY \ac \alicealg} \big \| \Xi^{\bOmega \bX_C \bY \ac \alicealg} \big) +  |C| \cdot \log{|\cA|^2} \, \right),\notag\\
			&\leq \frac{1}{r-|C|} \left( \log\left(\frac{1}{P(W_C)}\right) +  |C| \cdot \log{|\cA|^2} \, \right) = \eta_{\text{PR}},\notag\\
	\end{align*}
	where the first line follows from~\Cref{prop:divergence_data_processing}, the second line follows from~\Cref{prop:divergence_chain_rule} and $\xi_{\omega, \vx_C, \vy_C}^{\bOmega \bX_C \bY \ac \alicealg} = 0$ whenever $(\vx_C, \vy_C, \va_C, \vb_C) \not\in W_C$. The third line follows from conditioning a probability distribution will never increase the relative entropy (i.e.~\Cref{prop:relative_min_entropy_chain_rule2}). The fifth line follows by combining~\Cref{prop:relative_min_entropy_chain_rule2} and~\eqref{eq:Xitoxiinequality}. Thus showing~\eqref{eq:xi-change-1}.

\end{proof}

Given a fix $\omega_{\mi}$ sampled from $\bOmega_{\mi}$, $\vr_C \in W_C$ and $(x,y) \in \cX^2$, let $\omega_{\mi,x}^{A} = (\omega_{\mi}, \omega_{i} = (A,x) )$ in $\bOmega$. We define the (normalize) state
\begin{equation} \label{eq:xidefcondoni}
	\xi^{\bOmega \bX_C \bY \ac \alicealg}_{(\omega_\mi, \vr_C), x,y} = \sum_{\vy: \vy|_{C} = \vy_C,  \vy_i = y} \P_{\bOmega \bX_C \bY} ( \omega_{\mi,x}^{A},\vx, \vy_C) \,  \cdot \bra{\omega_{\mi,x}^{A}, \vx_C, \vy, \va_C, \vb_C} \otimes \Xi_{\omega_{\mi,x}^{A},\vx_C, \vy,\va_C,\vb_C} ~.   \\
\end{equation}
Similarly, let $\omega_{\mi,y}^{B} = (\omega_{\mi}, \omega_{i} = (B,y))$ and we define the (normalize) state
\begin{equation*}
	\lambda^{\bOmega \bX \bY_C \ac \alicealg}_{(\omega_\mi, \vr_C), x,y} = \sum_{\vx: \vx|_{C} = \vx_C,  \vx_i = x} \P_{\bOmega \bX_C \bY} ( \omega_{\mi,y}^{B},\vx_C, \vy) \,  \cdot \bra{\omega_{\mi,y}^{B}, \vx, \vy_C, \va_C, \vb_C} \otimes \lambda_{\omega_{\mi,y}^{B},\vx, \vy_C,\va_C,\vb_C} ~.  \\
\end{equation*}
The second claim relates the states $\xi$ and $\lambda$ associated with different choices of $i,(\omega_{\mi}, \vr_C),x,y$. 
\begin{claim}\label{claim:xi-change-y}
	The following hold:
	\begin{align}
		\Ex_{i \sim [r] \setminus C} \, \Ex_{\bOmega_{-i} \blR_C |  W_C} \, \, \, \Ex_{\bX \bY} \,\, \big \| \xi^{\alicealg}_{(\omega_\mi, \vr_C), x,y} - \xi^{\alicealg}_{(\omega_\mi, \vr_C), x,\bot} \big \|_\alicealg^2\; &= \; O\big(\sqrt{\eta_{\text{PR}}} \big)~, \label{eq:E_A_1} \\
		\Ex_{i \sim [r] \setminus C} \, \Ex_{\bOmega_{-i} \blR_C |  W_C}   \, \, \, \Ex_{\bX \bY} \,\, \big \| \lambda^{\alicealg}_{(\omega_\mi, \vr_C), x,y} - \lambda^{\alicealg}_{(\omega_\mi, \vr_C), \bot,y} \big \|_\alicealg^2\; &= \; O\big(\sqrt{\eta_{\text{PR}}} \big)~, \label{eq:E_A_2}
	\end{align}
	where the expectation over $XY$ is with respect to the distribution $\mu_{XY}$. 
\end{claim}

\begin{proof}
	We show~\eqref{eq:E_A_1}; the proof of~\eqref{eq:E_A_2} is similar. Fix $i \in [r] \setminus C$ and $\vr_C \in W_C$, the state
	\[
		\xi_{(\omega, \vr_C)}^{\bY_i \alicealg} = \Ex_{\bY_i | \blR_C= \vr_C, W_C} \, \bra{\vy_i}^{\bY_i} \otimes \xi_{(\omega, \vr_C), \vy_i}^\alicealg = \xi_{(\omega, \vr_C)}^{\bY_i} \otimes 	\xi_{(\omega, \vr_C)}^{\alicealg}
	\]
	where the state $\xi_{(\omega, \vr_C), \vy_i}^\alicealg$ is $\xi^{\bOmega \bX_C \bY \ac \alicealg}_{(\omega, \vr_C)}$ conditioning on the ith coordinates of $\vy$ being $\vy_I$. By applying~\Cref{prop:pinsker}
	\begin{align}
		\Ex_{i \sim [r] \setminus C} \, \Ex_{\bOmega_{-i} \blR_C |  W_C} \, \, \, \big \| \xi_{(\omega, \vr_C), \vy_i}^\alicealg -  \xi_{(\omega, \vr_C)}^\alicealg \big \|_\alicealg^2 &\leq 2  \ln 2 \, \Ex_{i \sim [r] \setminus C} \, \Ex_{\bOmega_{-i} \blR_C |  W_C}  \,\,  \D \big(   \xi_{(\omega, \vr_C), \vy_i}^\alicealg  \, \big \|\,  \xi_{(\omega, \vr_C)}^\alicealg \big )\notag\\
		&= 2\ln 2\, \Ex_I \, \Ex_{\blR | W_C} \,\,  \I (\bY_i ; \alicealg)_{\xi^{\bOmega \bX_C \bY \ac \alicealg}_{(\omega, \vr_C)}} \notag\\
		& = O(\eta_{\text{PR}})~,\label{eq:claim24-1a}
	\end{align} 
	where the last line follows from~\Cref{claim:xi-change-x}. The rest of the proof follows in an identical manner to that of~\cite[Claim 5.14]{bavarianAnchoredParallelRepetition2021} .
	
\end{proof}

We are now ready to give the proof of \Cref{lem:unitary_bounds}.

\begin{proof}[Proof of \Cref{lem:unitary_bounds}]
	We start by showing the existence of operators $V_{(\omega_{\mi}, \vr_C),y}$ that satisfy~\eqref{eq:vy_bound}.  Let $i \in [r] \setminus C$, $(\omega_{\mi}) \in \bOmega_{\mi}$, $\vr_C \in W_C$ and $(x,y) \in \cX^2$. We start the proof by showing the following claim. 
	
	\begin{claim}
		\label{clm:xi-purification}
		For all $A \in \alicealg$, we have
		\begin{align*}
			\wt{\Phi}_{\sspxy}(A) &= \xi^{\alicealg}_{(\omega_\mi, \vr_C), x,y}(A) \\
			\wt{\Phi}_{\sspy}(A) &= \xi_\sspy^{\alicealg}(A)
		\end{align*}
		where $\wt{\Phi}_\sspxy$ is the state defined by
		\begin{equation*}
			\wt{\Phi}_\sspxy(A) = \braket{\wt{\Phi}_\sspxy| A |\wt{\Phi}_\sspxy}
		\end{equation*}
		with $\ket{\wt{\Phi}_\sspxy}$ being defined on~\eqref{eq:def-psitt}. 
	\end{claim}
	
	\begin{proof}
		Recall from~\eqref{eq:xidefcondoni}, we can write the state $\xi_\ssxy^{\bOmega \bX_C \bY \ac \alicealg}$ explicitly as
		\begin{equation}
			\xi^{\bOmega \bX_C \bY \ac \alicealg}_{(\omega_\mi, \vr_C), x,y} = \sum_{\vy: \vy|_{C} = \vy_C,  \vy_i = y} \frac{ \P_{\bOmega \bX_C \bY} ( \omega_{\mi,x}^{A},\vx, \vy_C) \,  }{\P_{\bOmega \blR_C \bY_i } ( \omega_{\mi,x}^{A},\vr_C, y) } \cdot \bra{\omega_{\mi,x}^{A}, \vx_C, \vy, \va_C, \vb_C} \otimes \Xi_{\omega_{\mi,x}^{A},\vx_C, \vy,\va_C,\vb_C}~.   \\
		\end{equation}
		To see that the normalization is correct, we evaluate this state on the identity $\Id_{\alicealg}$ on the state $	\xi^{\alicealg}_{(\omega_\mi, \vr_C), x,y}$ to get
		\begin{align*}
			\xi^{\alicealg}_{(\omega_\mi, \vr_C), x,y}(\Id_{\alicealg}) &= \sum_{\vy: \vy|_{C} = \vy_C,  \vy_i = y} \frac{ \P_{\bOmega \bX_C \bY} ( \omega_{\mi,x}^{A},\vx, \vy_C) \,  }{\P_{\bOmega \blR_C \bY_i } ( \omega_{\mi,x}^{A},\vr_C, y) }  \Xi_{\omega_{\mi,x}^{A},\vx_C, \vy,\va_C,\vb_C}(\Id_{\alicealg})~. \\
			&=\sum_{\vy: \vy|_{C} = \vy_C,  \vy_i = y} \frac{ \P_{\bOmega \bX_C \bY} ( \omega_{\mi,x}^{A},\vx, \vy_C) \,  }{\P_{\bOmega \blR_C \bY_i } ( \omega_{\mi,x}^{A},\vr_C, y) } \cdot \bra{\psi} A^{\omega_{\mi,x}^{A}, (\vx_C, \vy_C)}_{\va_C} \, B^\vy_{\vb_C} \ket{\psi} \\
			&= \frac{ 1 }{\P_{\bOmega \blR_C \bY_i }  ( \omega_{\mi,x}^{A},\vr_C, y) } \sum_{\substack{\vx: \vx|_{C} = \vx_C \\ \vy: \vy|_{C} = \vy_C,  \vy_i = y}} \P_{\bOmega \bX \bY  \bA_C \bB_C} ( \omega_{\mi,x}^{A},\vx, \vy, \va_C, \vb_C)  \\
			&= 1~.
		\end{align*}
		Now we compute the restriction of $\xi_\ssxy^{\alicealg}$ to the subalgebra $\alicealg$. For all $M \in \alicealg$
		\begin{align}
			&\notag \xi_\ssxy^{\alicealg}(M) = \sum_{\vy: \vy|_{C} = \vy_C,  \vy_i = y} \frac{ \P_{\bOmega \bX_C \bY} ( \omega_{\mi,x}^{A},\vx, \vy_C) \,  }{\P_{\bOmega \blR_C \bY_i } ( \omega_{\mi,x}^{A},\vr_C, y) }  \Xi_{\omega_{\mi,x}^{A},\vx_C, \vy,\va_C,\vb_C}(M)~. \\
			&= \sum_{\vy: \vy|_{C} = \vy_C,  \vy_i = y} \frac{ \P_{\bOmega \bX_C \bY} ( \omega_{\mi,x}^{A},\vx, \vy_C) \,  }{\P_{\bOmega \blR_C \bY_i } ( \omega_{\mi,x}^{A},\vr_C, y) }  \cdot \bra{\psi} \left (A^{ \omega_{\mi,x}^{A}, (\vx_C, \vy_C)}_{\va_C} \right)^{1/2} \, M \, \left(A^{\omega_{\mi,x}^{A}, (\vx_C, \vy_C)}_{\va_C} \right )^{1/2} B^\vy_{\vb_C} \, \ket{\psi}   \notag \\
			\notag &= \frac{ \P_{\bOmega \bX_C \bY_C, \bY_i} ( \omega_{\mi,x}^{A},\vx_C, \vy_C, y) \,  }{\P_{\bOmega \blR_C \bY_i } ( \omega_{\mi,x}^{A},\vr_C, y) }  \cdot \\
			& \qquad \bra{\psi} \left (A^{\omega_{\mi,x}^{A}, (\vx_C, \vy_C)}_{\va_C} \right)^{1/2} \, M \, \left(A^{\omega_{\mi,x}^{A}, (\vx_C, \vy_C)}_{\va_C} \right )^{1/2} \left( \sum_{\vy: \vy|_{C} = \vy_C,  \vy_i = y} P_{\bY| \Omega = \omega_{\mi}}(\vy) \cdot  B^\vy_{\vb_C} \right) \, \ket{\psi}   \notag \\
			&= \frac{\bra{\psi} \left (A^{\omega_{\mi,x}^{A}, (\vx_C, \vy_C)}_{\va_C} \right)^{1/2} \, M \, \left(A^{\omega_{\mi,x}^{A}, (\vx_C, \vy_C)}_{\va_C} \right )^{1/2} \cdot B^{(\omega_{\mi}, \vy_C), y}_{\va_C}| \, \ket{\psi}}{\P_{\bA_C \bB_C| \bOmega = \omega_{\mi,x}^{A}, \bX_C = \vx_C, \bY_C = \vy_C, \bY_i = y } ( \va_C, \vb_C) }     \notag \\
			&= \gamma_\sspxy^{-2} \bra{\psi} \left (A^{\omega_{\mi,x}^{A}, (\vx_C, \vy_C)}_{\va_C} \right)^{1/2} \, M \, \left(A^{\omega_{\mi,x}^{A}, (\vx_C, \vy_C)}_{\va_C} \right )^{1/2} \, B^{(\omega_{\mi}, \vy_C), y}_{\va_C} \, \ket{\psi} \notag \\
			\ &= \gamma_\sspxy^{-2}\cdot  \bra{\psi} \left(A^{\omega_{\mi,x}^{A}, (\vx_C, \vy_C)}_{\va_C} \right)^{1/2} \left(B^{(\omega_{\mi}, \vy_C), y}_{\va_C}\right)^{1/2}  \, M \left( \left(A^{\omega_{\mi,x}^{A}, (\vx_C, \vy_C)}_{\va_C} \right)^{1/2} B^{(\omega_{\mi}, \vy_C), y}_{\va_C}\right)^{1/2}  \, \ket{\psi},	\label{eq:xi-purification-1}
		\end{align}
		where line 3 follows from $\bY_i$ being independent from all other $\bY_j$ and $\bOmega_j$ for $j \neq i$ and~\Cref{claim:ufacts}, line 4 follows from the definition of $ B^{(\omega_{\mi}, \vy_C), y}$ from~\eqref{eq:states-operators-2}, line 5 follows from~\Cref{prop:gamma}, and the last line follows from $B^{(\omega_{\mi}, \vy_C), y}_{\va_C}  \in \alicealg'$. We see that the quantity given in~\eqref{eq:xi-purification-1} are precisely the definition to $\wt{\Phi}_{\sspxy}(M)$ given in~\eqref{eq:def-psitt}. 
		
		For the second part of the claim, whenever $x = \bot$, $A^{\omega_{\mi,x}^{A}}_{\va_C}$ are the same as $	A^{(\omega_\mi,\vx_C), x}_{\va_C}$ given in~\eqref{eq:states-operators-2}. Hence the second part of the claim holds by~\eqref{eq:randomremarkParallel}. This concludes the proof of~\Cref{clm:xi-purification}.
	\end{proof}

	By combining~\Cref{clm:xi-purification} and~\Cref{claim:xi-change-y}, we have
	\begin{align*}
		&\Ex_{i \sim [r] \setminus C} \, \Ex_{\bOmega_{-i} \blR_C |  W_C} \, \, \, \Ex_{\bX \bY} \,\, \big \| 	\wt{\Phi}_{\ssxy}- 	\wt{\Phi}_{(\omega_\mi, \vr_C), x,\bot} \big \|_\alicealg^2\; = \; O\big(\sqrt{\eta_{\text{PR}}} \big) \\
		&= \Ex_{i \sim [r] \setminus C} \, \Ex_{\bOmega_{-i} \blR_C |  W_C} \, \, \, \Ex_{\bX \bY} \,\, \big \| \xi^{\alicealg}_{(\omega_\mi, \vr_C), x,y} - \xi^{\alicealg}_{(\omega_\mi, \vr_C), x,\bot} \big \|_\alicealg^2\; = \; O\big(\sqrt{\eta_{\text{PR}}} \big)
	\end{align*}
	By conditioning $\cX = \bot$ on the third expectation (which occurs with probability $\eta_{\text{Anchor}} = \frac{1}{2}$),
	\begin{equation*}
		\Ex_{i \sim [r] \setminus C} \, \Ex_{\bOmega_{-i} \blR_C |  W_C} \, \, \, \Ex_{Y} \,\, \big \| 	\wt{\Phi}_{(\omega_\mi, \vr_C), \bot,y}- 	\wt{\Phi}_{(\omega_\mi, \vr_C), \bot, \bot} \big \|_\alicealg^2\; = \; O\big(\sqrt{\eta_{\text{PR}}} \big) 
	\end{equation*}
	Now, by applying~\Cref{prop:uhlmann}, there exists a collection of Unitary operator $\{V_{(\omega_{\mi}, \vr_C), y}\}_{y \in \cY}$ in $\alicealg'$ such that 
	\begin{align*}
		&\Ex_{i \sim [r] \setminus C} \, \Ex_{\bOmega_{-i} \blR_C |  W_C} \, \, \, \Ex_{Y} \,\,  \bigbra{\wt{\Phi}_\sspy } V_{(\omega_{\mi}, \vr_C), y}  \bigket{\wt{\Phi}_\sspp}  \notag \\
		&\geq 1 - \frac{1}{2} \Ex_{i \sim [r] \setminus C} \, \Ex_{\bOmega_{-i} \blR_C |  W_C} \, \, \, \Ex_{Y} \, \, \Ex_{Y}  \,\,  \big \| \wt{\Phi}_\sspy - \wt{\Phi}_\sspp \big \|_\alicealg \notag \\
		&\geq 1 - O(\eta_{\text{PR}}^{1/4})\;, 
	\end{align*}
	where the second line follows from Jensen's inequality. By translating the above equation to Euclidean distance, and then apply Jensen's inequality, we have
	\begin{align*}
		&\Ex_{i \sim [r] \setminus C} \, \Ex_{\bOmega_{-i} \blR_C |  W_C} \, \, \, \Ex_{Y} \,\,  \big \| \ket{\wt{\Phi}_\sspy } - V_{(\omega_{\mi}, \vr_C), y}  \ket{\wt{\Phi}_\sspp}  	\big \| \\
		& \leq \sqrt{\Ex_{i \sim [r] \setminus C} \, \Ex_{\bOmega_{-i} \blR_C |  W_C} \, \, \, \Ex_{Y} \,\,  \big \| \ket{\wt{\Phi}_\sspy } - V_{(\omega_{\mi}, \vr_C), y}  \ket{\wt{\Phi}_\sspp}  \big \|^2}= O\big(\eta_{\text{PR}}^{1/8}\big)~.
	\end{align*}
	Applying Markov's inequality over the index $i$ establishes~\eqref{eq:vy_bound}. The argument for~\eqref{eq:vxy_bound} proceeds similarly. We start by using \Cref{clm:xi-purification}, which establish that the states $\ket{\wt{\Phi}_\sspxy}$ and $\ket{\wt{\Phi}_\sspxp}$ are purifications of the states $\xi^{\alicealg}_\sspxy$ and $\xi^{\alicealg}_\sspxp$ respectively. Using Uhlmann's Theorem and Claim~\ref{claim:xi-change-y} in a similar way to how we derived~\eqref{eq:vy_bound} we deduce the existence of operators $V_{\ssxy}$ satisfying~\eqref{eq:vxy_bound}. To prove~\eqref{eq:ux_bound} we use the following claim whose proof is analogous to that of \Cref{clm:xi-purification}.
	
	\begin{claim}
		\label{clm:lambda-purification}
		For all $A \in \alicealg'$, we have
		\begin{gather*}
			\wt{\Phi}_{\ssxp}(A) = \lambda_\ssxp(A)
		\end{gather*}
		where $\ket{\wt{\Phi}_\ssxy}$ is the vector state defined in~\eqref{eq:def-psitt}. 
	\end{claim}
	
	Using the above claim, we can use Uhlmann's Theorem in a similar manner as above to deduce the existence of operators $U_{(\omega_{\mi}, \vr_C),x}$ in $\alicealg$ which satisfy~\eqref{eq:ux_bound}. This concludes the proof for \Cref{lem:unitary_bounds}
\end{proof}

\subsubsection{Proof of \Cref{prop:local_operators}}

We end the section with the proof of \Cref{prop:local_operators}. We remark that this subsection follows the structure of~\cite[Section 5.4]{bavarianAnchoredParallelRepetition2021}.

\begin{proof}
	For every $i,(\omega_{\mi}, \vr_C)$, $x$ and $y$ let $U_{(\omega_{\mi}, \vr_C), x}$, $V_{(\omega_{\mi}, \vr_C), y}$ and $V_{\ssxy}$ denote the unitary guarantee by~\Cref{lem:unitary_bounds}. For notational convenience we suppress the dependence on $(i, (\omega_{\mi}, \vr_C))$; thus the operators $U_x, V_y, V_{x,y}$, the states $\ket{\Phi_{x,y}}$, and their normalizations $\gamma_{x,y}$ all implicitly depend on $i$ and $(\omega_{\mi}, \vr_C = (\vx_C, \vy_C, \va_C, \vb_C))$. We also write $\Ex_{\bOmega_{-i} \blR_C| \dummy, \dummy, W_C}$ as shorthand for $\Ex_{\bOmega_{-i} \blR_C| \bX_i = \dummy, \bY_i = \dummy, W_C}$
	
	For a fixed index $i \in [r] \setminus C$, we call the index to be \emph{good} if it satisfies (i) the conclusions of \Cref{lem:unitary_bounds}, (ii) the conclusions of \Cref{lem:norms-are-close}, and (iii) it holds that
	\begin{equation}\label{eq:good-i-cond}
		\left \| \P_{\bOmega_{\mi}, \blR_C | \bX_i = \dummy, \bY_i = \dummy,W_C} - \P_{\bOmega_{-i} \blR_C | W_C} \right \| \,\leq\, O(\eta_{\text{PR}}^{1/4})\;.
	\end{equation}
	By applying the data processing inequality (\Cref{lem:data-processing}) to Item 3 of \Cref{lem:classical_skew} to marginalize over the random variable $\bOmega_i$, and then applying Markov's inequality over the index $i$, we get that~\eqref{eq:good-i-cond} holds with probability at least $1 - O(\eta_{\text{PR}}^{1/4})$ over a uniformly random choice of $i$. This combined with \Cref{lem:unitary_bounds} and \Cref{lem:norms-are-close} implies that an index $i$ is good with probability at least $1 - O(\eta_{\text{PR}}^{1/16})$. 
	For a good index $i$, by combining the bound from~\Cref{lem:unitary_bounds} and~\eqref{eq:good-i-cond}, 
	\begin{align}
		\Ex_{\bOmega_{\mi} \blR_C | \dummy,\dummy,W_C} \,\, \Ex_{\bX} \,\, \,\,\, \big \| \ket{\wt{\Phi}_\xp } - U_{x}  \ket{\wt{\Phi}_\pp}  \big \| &=  O(\eta_{\text{PR}}^{1/16})\;,\label{eq:ux_bound-2}\\
		\Ex_{\bOmega_{\mi} \blR_C | \dummy,\dummy,W_C} \,\, \Ex_{\bY} \,\, \,\,\, \big \| V_y  \ket{\wt{\Phi}_\pp}  - \ket{\wt{\Phi}_\py }  \big \| &=  O(\eta_{\text{PR}}^{1/16})\;,\label{eq:vy_bound-2}\\
		\Ex_{\bOmega_{\mi} \blR_C | \dummy,\dummy,W_C} \,\, \Ex_{\bX \bY} \,\,  \big \|  V_\xy  \ket{\wt{\Phi}_\pxy}  - \ket{\wt{\Phi}_\pxp }  \big \| &=  O(\eta_{\text{PR}}^{1/16}) \label{eq:vxy_bound-2}
	\end{align}
	where we bound $O(\eta_{\text{PR}}^{1/16}) + O(\eta_{\text{PR}}^{1/4}) = O(\eta_{\text{PR}}^{1/16})$. 	he main step of the proof of~\Cref{prop:local_operators} is to combine $U_x$ and $V_y$ together by showing the following claim. We remark that the following claim is an commuting operator value variant of~\cite[Claim 5.19]{bavarianAnchoredParallelRepetition2021}. 
	\begin{claim}\label{claim:good-i-2}
		\begin{align}
			\Ex_{\bOmega_{\mi} \blR_C | \dummy,\dummy,W_C} \,\, \Ex_{\bX \bY} \, \big \| U_{x} \, V_{y} \ket{\wt{\Phi}_{\pp}} - \ket{\wt{\Phi}_{\xy}} \big\|  
			&\leq 	\Ex_{\bOmega_{\mi} \blR_C | \dummy,\dummy,W_C} \, \Ex_{\bX \bY} \, \gamma_\pp^{-1} \, \left \| V_y \ket{{\Phi}_{\dummy,\dummy}} - \ket{{\Phi}_{\dummy,y}} \right \|\label{eq:good-i-2a}\\
			&\qquad + 2\eta_{\text{Anchor}}^{-1/2} \gamma_\pp^{-1} \,\big \|  V_{x,y} \ket{\Phi_{\dummyx,y}} - \ket{{\Phi}_{\dummyx,\dummy}} \big \| \label{eq:good-i-2b}\\
			&\qquad + \gamma_\pp^{-1} \, \left \|  U_x \ket{{\Phi}_{\dummy,\dummy}} -  \ket{{\Phi}_{x,\dummy}} \right \| + O(\eta_{\text{PR}}^{1/8})\;,\label{eq:good-i-2c}
		\end{align}
	where $\gamma_\pp$ is defined in~\eqref{eq:def-gamma}. 
	\end{claim}
	Before proving the claim, it would be useful to work out the following two bounds. Using the definition
	\begin{equation}
		\Ex_{\bOmega_{-i} \blR_C| \dummy, \dummy, W_C} \, \Ex_{\bX \bY} \,  \big \|\ket{\wt{\Phi}_{x,y}} - \gamma_{\dummy,\dummy}^{-1}\ket{\Phi_{x,y}} \big \| 
		= \Ex_{\bOmega_{-i} \blR_C| \dummy, \dummy, W_C}  \,\Ex_{\bX \bY} \, \Big|1 - \frac{\gamma_{x,y}}{\gamma_{\dummy,\dummy}} \Big|  = O(\eta_{\text{PR}}^{1/8})~,\label{eq:local_operators-0}
	\end{equation}
	where the second line is by Jensen's inequality and~\eqref{eq:norms-are-close-s0} in \Cref{lem:norms-are-close}. Similarly,
	\begin{align}
		\Ex_{\bOmega_{-i} \blR_C| \dummy, \dummy, W_C}  \, \Ex_{\bX \bY} \,  \big \|\ket{\wt{\Phi}_{\pxy}} - \gamma_{\dummy,\dummy}^{-1}\ket{\Phi_{\pxy}} \big \| &=\Ex_{\bOmega_{-i} \blR_C| \dummy, \dummy, W_C}  \,\Ex_{\bX \bY} \, \Big|1 - \frac{\gamma_{\pxy}}{\gamma_{\dummy,\dummy}} \Big| = O\big(\eta_{\text{PR}}^{1/8}\big)~,\label{eq:local_operators-0b}
	\end{align}
	by~\eqref{eq:norms-are-close-s1}. We note that in the above division by $\gamma_\pp$ is well-defined because $(\omega_{\mi}, \vr_C)$ is sampled with positive probability from the distribution $\P_{(\omega_{\mi}, \vr_C) | \dummy,\dummy,W_C}$. We are now ready to prove~\Cref{claim:good-i-2}. 
	\begin{proof}
		We start by writing 
		\begin{align}
			&\Ex_{\bOmega_{\mi} \blR_C | \dummy,\dummy,W_C} \, \Ex_{\bX \bY} \, \big \| U_{x} \, V_{y} \ket{\wt{\Phi}_{\pp}} - \ket{\wt{\Phi}_{\xy}} \big\| \notag \\
			&\leq  	\Ex_{\bOmega_{\mi} \blR_C | \dummy,\dummy,W_C} \, \Ex_{\bX \bY} \, \left \| U_{x} \, V_{y} \bigket{\wt{\Phi}_{\pp}} - \frac{\gamma_\xy}{\gamma_\pp} \bigket{\wt{\Phi}_{\xy}} \right\| +  \left \| \frac{\gamma_\xy}{\gamma_\pp} \bigket{\wt{\Phi}_{\xy}} - \bigket{\wt{\Phi}_{\xy}} \right\| \notag \\
			&= 	\Ex_{\bOmega_{\mi} \blR_C | \dummy,\dummy,W_C} \, \Ex_{\bX \bY} \, \gamma_\pp^{-1} \left \| U_{x} \, V_{y} \bigket{{\Phi}_{\pp}} - \bigket{{\Phi}_{\xy}} \right\| +  \left | \frac{\gamma_\xy}{\gamma_\pp} -1  \right| \notag \\
			&\leq 	\Ex_{\bOmega_{\mi} \blR_C | \dummy,\dummy,W_C} \, \Ex_{\bX \bY} \, \gamma_\pp^{-1} \left \| U_{x} \, V_{y} \bigket{{\Phi}_{\pp}} - \bigket{{\Phi}_{\xy}} \right\| + O(\eta_{\text{PR}}^{1/8}) \notag \\
			 &\leq \Ex_{\bOmega_{\mi} \blR_C | \dummy,\dummy,W_C} \, \Ex_{\bX \bY} \, \gamma_\pp^{-1} \left( \left \| U_{x} \, V_{y} \bigket{{\Phi}_{\pp}} - U_{x} \bigket{{\Phi}_{\py}} \right\| + \left \| U_{x} \bigket{{\Phi}_{\py}} - \bigket{{\Phi}_{\xy}} \right\|\right) + O(\eta_{\text{PR}}^{1/8})  \notag \\
			  \label{eq:local_operators-1}&\leq \Ex_{\bOmega_{\mi} \blR_C | \dummy,\dummy,W_C} \, \Ex_{\bX \bY} \, \gamma_\pp^{-1} \left( \left \|  \, V_{y} \bigket{{\Phi}_{\pp}} - \bigket{{\Phi}_{\py}} \right\| + \left \| U_{x} \bigket{{\Phi}_{\py}} - \bigket{{\Phi}_{\xy}} \right\|\right) + O(\eta_{\text{PR}}^{1/8})   \;,
		\end{align}
		where the third line follows from~\eqref{eq:local_operators-0}, and the last line follows from $U_x \in \cU(\alicealg)$. For all $\dummyx \in \cX_{\dummy}$, by~\Cref{lem:inversetrick1} and~\eqref{eq:states-operators-3a} and~\eqref{eq:states-operators-3b}. There exist an operator $(C^{\dummyx})^{\frac{1}{2}}$ and $(D^{\dummyx})^{\frac{1}{2}}$ such that 
		\begin{align*}
		\eta_{\text{Anchor}}^{-\frac{1}{2}} (C^{\dummyx})^{\frac{1}{2}}  (A^{\dummyx})^{\frac{1}{2}}&= (A^{\dummy})^{\frac{1}{2}} \\
		(1-\eta_{\text{Anchor}})^{-\frac{1}{2}}  (D^{\dummyx, x})^{\frac{1}{2}} (A^{\dummyx})^{\frac{1}{2}} &= (A^{x})^{\frac{1}{2}}.
		\end{align*}
		By~\eqref{eq:def-state-y}, we have 
		\begin{align*}
			\ket{\Phi_{\dummy,y}} &= (B^y)^{\frac{1}{2}} (A^{\dummy})^{\frac{1}{2}}  \ket{\psi} =  	\eta_{\text{Anchor}}^{-\frac{1}{2}}   (C^{\dummyx})^{\frac{1}{2}} (B^y)^{\frac{1}{2}}  (A^{\dummyx})^{\frac{1}{2}}  \ket{\psi} = 	(1-\eta_{\text{Anchor}})^{-\frac{1}{2}}   (C^{\dummyx})^{\frac{1}{2}}   \ket{\Phi_{\dummyx,y}} \\
			\ket{\Phi_{x,y}} &=		(1-\eta_{\text{Anchor}})^{-\frac{1}{2}}     (D^{\dummyx, x})^{\frac{1}{2}}   \ket{\Phi_{\dummyx,y}} 
		\end{align*}
	 	Thus, for each $x, y \in \cX^2$,
		\begin{align}
			\left \| U_x \ket{\Phi_{\dummy,y}} - \ket{\Phi_{x,y}} \right \| &=   \big \| \eta_{\text{Anchor}}^{-\frac{1}{2}} U_x  (C^{\dummyx})^{\frac{1}{2}} \ket{\Phi_{\dummyx,y}} -	(1-\eta_{\text{Anchor}})^{-\frac{1}{2}} (D^{\dummyx, x})^{\frac{1}{2}}   \ket{\Phi_{\dummyx,y}} \big \| \notag \\
			&=   \big \|  \eta_{\text{Anchor}}^{-\frac{1}{2}} U_x  (C^{\dummyx})^{\frac{1}{2}}  \, V_{x,y} \ket{\Phi_{\dummyx,y}} - 	(1-\eta_{\text{Anchor}})^{-\frac{1}{2}} (D^{\dummyx, x})^{\frac{1}{2}}    \, V_{x,y} \ket{\Phi_{\dummyx,y}} \big \| \notag \\
			&\leq \eta_{\text{Anchor}}^{-\frac{1}{2}} \big \|  U_x  (C^{\dummyx})^{\frac{1}{2}}  \, V_{x,y} \ket{\Phi_{\dummyx,y}} -  U_x  (C^{\dummyx})^{\frac{1}{2}} \ket{\Phi_{\dummyx, \dummy}} \big \|  \label{eq:22-2a}\\
			& \qquad +   \big \|  \eta_{\text{Anchor}}^{-\frac{1}{2}} U_x  (C^{\dummyx})^{\frac{1}{2}} \ket{\Phi_{\dummyx,\dummy}} - 	(1-\eta_{\text{Anchor}})^{-\frac{1}{2}} (D^{\dummyx, x})^{\frac{1}{2}}     \ket{\Phi_{\dummyx,\dummy}} \big \| \label{eq:22-2b}\\
			& \qquad + 	(1-\eta_{\text{Anchor}})^{-\frac{1}{2}} \big \| (D^{\dummyx, x})^{\frac{1}{2}}  \ket{\Phi_{\dummyx,\dummy}} - (D^{\dummyx, x})^{\frac{1}{2}}   \, V_{x,y} \ket{\Phi_{\dummyx,y}} \big \|~,\label{eq:22-2c}
		\end{align}
		where the second line follows from $V_{x,y} \in \cU(\alicealg')$, and the third line follows from the triangle inequality. We bound each of these three terms as follows. For~\eqref{eq:22-2a}
		\begin{equation*}
			\eta_{\text{Anchor}}^{-\frac{1}{2}} \big \|  U_x  (C^{\dummyx})^{\frac{1}{2}}  \, V_{x,y} \ket{\Phi_{\dummyx,y}} -  U_x  (C^{\dummyx})^{\frac{1}{2}} \ket{\Phi_{\dummyx, \dummy}} \big \| \leq  \|  \, V_{x,y} \ket{\Phi_{\dummyx,y}} -  \ket{\Phi_{\dummyx, \dummy}} \big \|.
		\end{equation*}
		Similarly, for~\eqref{eq:22-2c}, since $	(1-\eta_{\text{Anchor}})^{-\frac{1}{2}}  \leq \eta_{\text{Anchor}}^{-\frac{1}{2}}$ whenever $\eta_{\text{Anchor}} \leq \frac{1}{2}$ ( $\eta_{\text{Anchor}} = \frac{1}{4}$)
		\begin{equation*}
			(1-\eta_{\text{Anchor}})^{-\frac{1}{2}} \big \| (D^{\dummyx, x})^{\frac{1}{2}}  \ket{\Phi_{\dummyx,\dummy}} - (D^{\dummyx, x})^{\frac{1}{2}}   \, V_{x,y} \ket{\Phi_{\dummyx,y}} \big \| \leq \eta_{\text{Anchor}}^{-\frac{1}{2}} \big \|  \ket{\Phi_{\dummyx,\dummy}} -   \, V_{x,y} \ket{\Phi_{\dummyx,y}} \big \|
		\end{equation*}
		Finally, for~\eqref{eq:22-2b}
		\begin{equation*}
	 \big \|  \eta_{\text{Anchor}}^{-\frac{1}{2}} U_x  (C^{\dummyx})^{\frac{1}{2}} \ket{\Phi_{\dummyx,\dummy}} - 	(1-\eta_{\text{Anchor}})^{-\frac{1}{2}} (D^{\dummyx, x})^{\frac{1}{2}}     \ket{\Phi_{\dummyx,\dummy}} \big \| = \|  U_x   \ket{\Phi_{\dummyx,\dummy}} -     \ket{\Phi_{x,\dummy}} \big \| 
		\end{equation*}
		
		Putting the three bounds together, from~~\eqref{eq:22-2a}--\eqref{eq:22-2c} we get
		\begin{align}
			\left \| U_x \ket{\Phi_{\dummy,y}} - \ket{\Phi_{x,y}} \right \| \leq   2 \eta_{\text{anchor}}^{-1/2} \big \|  V_{x,y} \ket{\Phi}_{\dummyx,y} - \ket{{\Phi}_{\dummyx,\dummy}} \big \| + \left \|  U_x \ket{{\Phi}_{\dummy,\dummy}} -  \ket{{\Phi}_{x,\dummy}} \right \| \; , \label{eq:22-3}
		\end{align}
		from which inserting it into~\eqref{eq:local_operators-1} proves the claim. 
	\end{proof}
	
	To conclude the proof \Cref{prop:local_operators} it remains to bound each of the three terms on the right-hand side of Claim~\ref{claim:good-i-2} by $O(\eta_{\text{PR}}^{1/16})$, and then use~\eqref{eq:good-i-cond} to exchange the expectation $ 	\Ex_{\bOmega_{\mi} \blR_C | \dummy,\dummy,W_C} \,\, $ with $ \Ex_{(\bOmega_{\mi}, \blR_C) | W_C}$ by introducing an additive $O(\eta_{\text{PR}}^{1/4})$ error. 
	We start with bounding~\eqref{eq:good-i-2a}:
	\begin{align*}
	\Ex_{\bOmega_{\mi} \blR_C | \dummy,\dummy,W_C} \, \Ex_{\bY} \, \gamma_\pp^{-1} &\,\, \left \|  V_y \ket{{\Phi}_{\dummy,\dummy}} -  \ket{{\Phi}_\py} \right \|  \\
		&=\Ex_{\bOmega_{\mi} \blR_C | \dummy,\dummy,W_C} \, \Ex_{\bY} \,\, \Big \|  V_y \ket{\wt{\Phi}_{\dummy,\dummy}} -  \frac{\gamma_\py}{\gamma_\pp} \ket{\wt{\Phi}_\py} \Big \| \\
		&\leq  \Ex_{\bOmega_{\mi} \blR_C | \dummy,\dummy,W_C} \, \Ex_{\bY} \, \big \|  V_y \ket{\wt{\Phi}_{\dummy,\dummy}} - \ket{\wt{\Phi}_{\py}} \big \| + \Big \| \ket{\wt{\Phi}_{\py}} - \frac{\gamma_{\py}}{\gamma_\pp} \ket{\wt{\Phi}_{\py}}\Big \|\\
		&= O(\eta_{\text{PR}}^{1/16}) + \Ex_{\bOmega_{\mi} \blR_C | \dummy,\dummy,W_C} \, \Ex_{\bY} \, \Big | 1 - \frac{\gamma_{\py}}{\gamma_\pp} \Big | \\
		&= O(\eta_{\text{PR}}^{1/16}) + O(\eta_{\text{PR}}^{1/8}) = O(\eta_{\text{PR}}^{1/16})~,
	\end{align*}
	where the third line uses~\eqref{eq:ux_bound-2} to bound the first term and the last line follows from~\eqref{eq:local_operators-0} and conditioning on $X = \dummy$ (which occurs with $\eta_{Anchor} = \frac{1}{2}$ probability), which occurs with probability $\frac{1}{2}$. We bound~\eqref{eq:good-i-2c} in an analogous fashion. 
	Finally, we bound~\eqref{eq:good-i-2b} as follows:
	\begin{align*}
		&2\eta_{\text{anchor}}^{-1/2}  \, \Ex_{\bOmega_{\mi} \blR_C | \dummy,\dummy,W_C}  \, \Ex_{\bX \bY} \, \, \gamma_\pp^{-1} \,\, \left \|  V_{x,y} \ket{{\Phi}_{\pxy}} -  \ket{{\Phi}_\pxp} \right \| \\
		&= 2\eta_{\text{anchor}}^{-1/2} \,  \Ex_{\bOmega_{\mi} \blR_C | \dummy,\dummy,W_C}  \, \Ex_{\bX \bY} \,\, \left \|  \frac{\gamma_\pxy}{\gamma_\pp} V_{x,y} \ket{\wt{\Phi}_{\pxy}} -  \frac{\gamma_\pxp}{\gamma_\pp} \ket{\wt{\Phi}_\pxp} \right \| \\
		&\leq  2\eta_{\text{anchor}}^{-1/2} \,\Ex_{\bOmega_{\mi} \blR_C | \dummy,\dummy,W_C}  \, \Ex_{\bX \bY} \, \Big \| \frac{\gamma_\pxy}{\gamma_\pp} \ket{\wt{\Phi}_{\pxy}} - \ket{\wt{\Phi}_{\pxy}} \Big \| + \Big \|  V_{x,y} \ket{\wt{\Phi}_{\pxy}} - \ket{\wt{\Phi}_{\pxp}} \Big \|  \\
		& \qquad \qquad \qquad \qquad \qquad \qquad \qquad \qquad + \left \| \ket{\wt{\Phi}_{\pxp}} - \frac{\gamma_{\pxp}}{\gamma_\pp} \ket{\wt{\Phi}_{\pxp}}\right \|\\
		&= 2\eta_{\text{anchor}}^{-1/2} \, \Ex_{\bOmega_{\mi} \blR_C | \dummy,\dummy,W_C}  \, \Ex_{\bX \bY} \, \Big | 1 - \frac{\gamma_{\pxy}}{\gamma_\pp} \Big | + O(\eta_{\text{PR}}^{1/16}/\eta_{\text{anchor}}^{1/2}) + 2\eta_{\text{anchor}}^{-1/2} \, \Ex_{\bOmega_{\mi} \blR_C | \dummy,\dummy,W_C}  \, \Ex_{\bX \bY} \, \Big | 1 - \frac{\gamma_{\pxp}}{\gamma_\pp} \Big | \\
		&= O(\eta_{\text{PR}}^{1/8}) + O(\eta_{\text{PR}}^{1/16}) + O(\eta_{\text{PR}}^{1/8}) = O(\eta_{\text{PR}}^{1/16})~.
	\end{align*}
	The last line follows from $\eta_{\text{anchor}} = \frac{1}{4}$,~\eqref{eq:local_operators-0b} to bound the first term,~\eqref{eq:vxy_bound-2} to bound the second term, and ~\eqref{eq:local_operators-0b} along with conditioning on $Y = \dummy$ to bound the last term.
	
\end{proof}
\newpage
\section{Soundness proofs}

In this appendix, we give a proof for the ``soundness" clause for both~\Cref{prop:QuestionReduction} and~\Cref{prop:AnswerReduction}. As mentioned previously, the proof of the ``soundness" clause for~\Cref{prop:QuestionReduction} follows a similar structure to~\cite[Section 8.4]{jiMIPRE2022a} and we present it in~\Cref{sec:QRsoundness}, and the proof the ``soundness" clause for~\Cref{prop:AnswerReduction} follows a similar structure to~\cite[Section 10.7]{jiMIPRE2022a} and we present it in~\Cref{sec:ARsoundness}. The only notable change is the translation between finite-dimensional strategies to tracially embeddable strategies using the ``translation chat" given in~\Cref{fig:tracialemdtofd}. 

\subsection{Proof for the ``soundness" clause for the question reduction transformation} \label{sec:QRsoundness}
In this subsection, we continue the proof of soundness of~\Cref{thm:Introspectgame} below. We first establish some notations which we use in the proof. Let $(\alicealg, \tau)$ be a tracial von Neumann algebra represented under the standard form $(\chi_{\tau}, \cL^2(\alicealg, \tau), \ket{\tau})$. $\cM_n(\bC) \otimes \alicealg$ is also a tracial von Neumann algebra with the trace $\Tr$. In this case, the standard form for $\cM_n(\bC) \otimes \alicealg$ are represented as $\cM_n(\bC) \otimes \cI_n \otimes \alicealg$ with the tracial state $\ket{\text{ME}_n} \otimes \ket{\tau}$. In this case the opposite map $op$ maps elements from $\left(\cM_n(\bC)\right)_A \otimes \left(\cI_n\right)_B \otimes \alicealg$ to $ \left(\cI_n\right)_A \otimes \left(\cM_n(\bC)\right)_B  \otimes \alicealg$. Hence, when discussing measurement operators $P_{A_1 \alicealg}$ from $\strategy''$, $\left(P_{A_1 \alicealg}\right)^{op}$ are define in $B_1 \alicealg$ (where the $\alicealg$ register are defined within the Hilbert space $\cL^2(\alicealg, \tau)$ and contains measurements from both $\alicealg$ and $\alicealg'$). We sometimes write $P^{op}_{B_1 \alicealg}$ for a measurement operator $P_{A_1 \alicealg}$ to specified the registers. 

For a canonical register subspace $V \subseteq \bF_{2^p}^m$, we define $\bC_V \subseteq (\bC^{2^{p}})^{\otimes m}$ as the subspace span by the basis $\{\ket{v}\}_{v \in V}$, and $\cI_{V} = \sum_{v \in V} \ketbra{v_0, \cdots v_{m-1}}{v_0, \cdots v_{m-1}} \in \cM_{2^{p \cdot m}}(\bC)$. Furthermore, for $A,B \in \bobalg$ for some von Neumann algebra $\bobalg$, we write $[A,B] = A\cdot B - B \cdot A$. For a multi-outcome measurement $\{P_{a,b}\}_{a \in \cA, b \in \cB}$, we denote $P_{a} = \sum_{b \in \cB} P_{a,b}$. For a function $\textbf{f}: \cA \rightarrow \cC$, we also use $P_{[\textbf{f}(a)| b]}$ to denote the data process measurement being applied on $P_a$. For two sets of POVM $\{P_{a_1, b}\}_{a_1 \in\cA , b \in \cB}$, $\{Q_{c, a_2}\}_{c \in \cB, a_2 \in \cA}$, we write $P_{a_1} \approx Q_{a_2 = a_1}$ to emphasize the outcome variables for the measurement outcome which follows the $\approx$ relationship. 

\subsubsection{Preliminary lemmas}
Before continuing with the proof, we begin by recall the following important lemma from~\cite{jiMIPRE2022a}. 
\begin{lemma}[Pauli twirl decomposition, Lemma 8.15 of~\cite{jiMIPRE2022a}] \label{lem:PauliTwirl}
	Let $\cX, \cA$ be two finite sets, $\mu$ be a distribution over $\cX$, $V$ be a canonical subspace of $\bF_{2^p}^m$, and $(\alicealg, \tau)$ be a tracial von Neumann algebra represented in the standard form $(\chi_{\tau}, \cL^2(\alicealg, \tau), \ket{\tau})$. For each $x \in \cX$, let $W_x$ and $V_x$ be two canonical subspace of $V$ such that $W_x \subseteq V_x \subseteq V$, and let $\decidertypedfunction_x: W_x \rightarrow W_x$ be a linear map. 
	
	Consider the state $\ket{\psi} = \ket{\text{ME}_{2^{p}}}^{\otimes m}_{A_1 A_2} \otimes \ket{\text{Aux}}_{\alicealg} \in  \left(\bC^{2^{p }} \otimes \bC^{2^{p }} \right)_{A_1 B_1}^{\otimes m} \otimes \cL^2(\alicealg, \tau)$ where $\ket{\text{Aux}}_{\alicealg}$ is an arbitrary vector. For each $x \in \cX$, let $\{ M^{x}_{t, a}\}_{t \in W_x, a \in \cA}$ be a set of PVM on $\cB(\bC_{V_x}) \otimes \alicealg$. Suppose that the following condition holds for some $\eps > 0$
	\begin{align*}
		&M^{x}_{t} \otimes \cI_{V_x^C} \approx_{\eps} \rho_{[\decidertypedfunction_x | t]}^{p, Z} \otimes \cI_{W_x^C} \otimes \cI_\alicealg\\
		&\left[(M^{x}_{t, a} \otimes \cI_{V_x^C}),  (\rho_{z}^{p, Z} \otimes \cI_{W_x^C} \otimes \cI_\alicealg)\right] \approx_{\eps} 0 \\
		&\left[(M^{x}_{t, a} \otimes \cI_{V_x^C}),  (\rho_{[\decidertypedfunction^{\bot}_x| t^{\bot}]}^{p, X} \otimes \cI_{W_x^C} \otimes \cI_\alicealg) \right] \approx_{\eps} 0, 
	\end{align*}
	where the $\approx$ is defined over the distribution $\mu$ and summing over $t, t^{\bot} \in W_x$ and $a \in \cA$. In the equation above, both $\rho_{[\decidertypedfunction^{\bot}_x| t^{\bot}]}^{p, X}$ and $\rho_{[\decidertypedfunction_x | t]}^{p, Z}$ are in $\cB(\cH_{V_x \setminus W_x}) \otimes \alicealg$. 
	
	Then for each $x \in \cX$ and $t \in W_x$, there exists a set of POVM $\{M_a^{x,t}\}_{a \in \cA}$ acting on $\cB(\bC_{V_x \setminus U_x}) \otimes \alicealg$ such that on average over $x \sim \mu$
	\begin{equation*}
		(M_{t,a}^x \otimes \cI_{V_x^C}) \approx_{O(\poly(\eps))}   \rho_{[\decidertypedfunction_x | r]}^{p, Z} \otimes M_a^{x,t} \otimes \cI_{V_x^C}.
	\end{equation*}
\end{lemma}
We also recall several lemmas from~\cite[Section 8.4.2]{jiMIPRE2022a}, which is useful for decomposing PVM measurements between different Hilbert spaces, and as well as showing commutation relationships between measurements.

\begin{lemma}[Decomposition of measurements over the $\approx$ distance, Lemma 8.16 of~\cite{jiMIPRE2022a}] \label{lem:QRsounddecomp}
	Let $\cA$ and $\cB$ be two finite sets, $\eps>0$, $\ket{\psi_Q} \in \cH_Q$ and $\ket{\psi_A} \in \cH_A$. Furthermore, let $\{ Q_{a} \} \subseteq \cB(\cH_Q)$ be a set of PVM and for all $a \in \cA$, and let $\{ A^a_b\}_{b \in \cB} ,\{ B^a_b\}_{b \in \cB} \subseteq \cB(\cH_A)$ be two sets of POVMs. Then the following are equivalent:
	\begin{itemize}
			\item $(Q_a \otimes A^a_b) \approx_{\eps}(Q_a \otimes B^a_b)$ on the state $\ket{\psi_Q} \otimes \ket{\psi_A}$.
			\item Over the distribution $P(a) =  \braket{\psi_Q| Q_a|\psi_Q}$ and the state $\ket{\psi_A}$, we have $ A^a_b \approx_{\eps} B^a_b$.
	\end{itemize} 
\end{lemma}

\begin{lemma}[Lemma 8.18 of~\cite{jiMIPRE2022a}]  \label{lem:QRsoundnessproofLem3}
	Let $\cX$, $\cY$ and $\cZ$ be three finite sets, and $(\alicealg, \tau)$ be a tracial von Neumann algebra represented in the standard form $(\chi_{\tau}, \cL^2(\alicealg, \tau), \ket{\tau})$. Furthermore, for all $x \in \cX$, let $y \in \cY$, let $\{ A^{y}_{x,z}\} \subseteq \alicealg$ be a set of POVM, $\{B_{x,y,z}\} \subseteq \alicealg'$ be a set of PVM Suppose that
	\begin{equation*}
		\sum_{x,y,z} \braket{\psi| A^y_{x,z} B_{x,y,z} |\psi} \geq 1 -\eps,
	\end{equation*} 
	for some state $\ket{\psi} \in \cL^2(\alicealg, \tau)$ and $\eps > 0$. Then with respect to the state $\ket{\psi}$,  $ B_{x,y,z} \approx_{\eps}  A^y_{x,z} B_{x,y}$. 
\end{lemma}

\begin{lemma}[Approximation relationship implies commutation, Lemma 5.25 of~\cite{jiMIPRE2022a}]  \label{lem:QRsoundnessproofcomm1}
	Let $\alicealg$ be a von Neumann algebra, let $\cX$, $\cA$, $\cB$, and $\cC$ be four finite sets. For every $x \in \cX $, let $\{A^x_{a, b}\}_{a \in \cA, b \in \cB}, \{C^x_{a, c}\}_{a \in \cA, c \in \cC} \subseteq \alicealg$ be two sets of POVM and let $\{B^x_{a,b,c}\}_{a \in \cA, b \in \cB, c \in \cC} \subseteq \alicealg'$ be a set of PVM. Suppose that for some $\delta > 0$
	\begin{equation*}
		A^x_{a,b} \approx_{\delta} B^x_{a,b} \, \qquad \, C^x_{a,c} \approx_{\delta} B^x_{a,c}. 
	\end{equation*}
	Then $[A^x_{a,b},  C^x_{a,c}]  \approx_{\delta} 0$.
\end{lemma}

\begin{lemma}[Decomposition measurements preserves approximate commutation over the $\approx$ distance, Lemma 8.16 of~\cite{jiMIPRE2022a}] \label{lem:QRsoundnessproofcomm2}
	Let $\cA$, $\cB$, and $\cC$ be three finite sets, $\eps>0$, $\ket{\psi_Q} \in \cH_Q$ and $\ket{\psi_A} \in \cH_A$. Furthermore, let $\{ Q_{a} \} \subseteq \cB(\cH_Q)$ be a set of PVM, let $\{ A^a_b\}_{b \in \cB} ,\{ B^a_b\}_{b \in \cB} \subseteq \cB(\cH_A)$ be two sets of POVMs, and let $A_{a,b} = Q_a \otimes A^a_b$, $B_{a,c} = Q_a \otimes A^a_c$. 
	
	Suppose that $[A_{a,b}, B_{a,c}] \approx_{\eps} 0$ with respect to the state $\ket{\psi_Q} \otimes \ket{\psi_A}$. Then
	\begin{equation*}
		[A^a_b, B^a_c] \simeq_{\eps} 0
	\end{equation*}
	where the $\approx$ is defined over the distribution $P(a) = \braket{\psi_Q| Q_a|\psi_Q}$, the state $\ket{\psi_A}$., and summing over $(b,c) \in \cB \times \cC$. 
\end{lemma} 

We remark that although the lemmas in this subsection is originally define for finite dimensional matrices, the proof can be trivially modified for the infinite dimension setting.

\subsubsection{The proof}

We continue the proof from~\Cref{sec:introspectionprotocol}. Recall, given the original game $\cG$, the input distribution $\mu$ is described by a $(k,m,p)$ CL distribution which is defined over two $k$-th level CL function $\decidertypedfunction^{A}, \decidertypedfunction^{B}: \bF_{2^p}^m \rightarrow \bF_{2^p}^m$ over registers $\{V_j\}_{j \in [k]}$, $ \bF_{2^p}^m = V = \bigoplus_{j \in [k]} V_j$. For the introspection transformation $\cG^{\text{Intro}} =(\cX^{\text{Intro}}, \cA^{\text{Intro}}, \mu^{\text{Intro}}, D^{\text{Intro}})$ of $\cG$, we refer the question $x \in \cX^{\text{Intro}}$ as an introspection question if the question label for $x$ corresponds to a vertex which intersects with either a orange or green edge (i.e. all the vertices on the right side of ``(Pauli, X)" and ``(Pauli, Z)"). 

Furthermore recall from~\Cref{sec:introspectionprotocol}, there exist a projective, symmetric strategy 
\begin{equation} \label{eq:QRsoundnessstrategy}
	\strategy'' =  \left(\bC^{2^{p\cdot m}}_{A_1} \otimes \bC^{2^{p \cdot m}}_{B_1} \otimes  \cL^2(\alicealg,\tau), \ket{\psi} = \ket{\text{ME}_{2^p}}^{\otimes m}_{A_1 B_1} \otimes \ket{\text{Aux}}_{\alicealg}, \{ P_{a}^x \} \right),
\end{equation}
such that $\omega(\cG^{\text{Intro}}, \strategy'') > 1 - O(\poly(n)) = \delta_1$, with 
\begin{equation} \label{eq:QRsoundnessapproxgenpauli}
	\rho^{p, W}_{s}  \approx_{O(\poly(k, \eps))} P^{(\text{Gen Pauli, W})}_s
\end{equation}
for $W \in \{X,Z\}$. In this case, $\alicealg$ also includes the extra finite-dimensional registers $A_2 B_2$, and the extra EPR pair guaranteed by~\Cref{thm:soundnessPaulibasis}. Our goal is construct a strategy for $\cG$ which succeed with probability $1 - O(\poly(\eps))$ using the measurement $P^{\text{Intro}, \decidertypedfunction^{A}}$ and $P^{\text{Intro}, \decidertypedfunction^{B}}$. Unless otherwise stated, the $\approx$ and $\simeq$ relationship in this subsection are define over the state $\ket{\psi}$ used to define the strategy $\strategy''$ 

For every $s \in \bF_{2^p}^m$, we partition $s = \sum_{j \in [k]} s_j$ where each $s_j \in V_j$, and we write 
\begin{equation}
	s_{<j} = \sum_{i \in [j]} s_i, \, \qquad \, s_{\geq j} = \sum_{j \leq i <k} s_i.
\end{equation}
Furthermore, for $W \in \{X,Z\}$ and $s_j \in V_j$, we use the notation $\rho^{p, W}_{s_j} =  \sum_{t \in \bF_{2^p}^m| t_j = s_j}\rho^{p, W}_{s_j}$, and we define $\rho^{p, W}_{s_{< j}}$ (resp. $\rho^{p, W}_{s_{\geq j}}$) in a similar manner for $s_{< j} \in V_{< j}$ (resp. $s_{\geq j} \in  V_{\geq j}$). Since $\rho^{p, W}$ is projective, we have $\rho^{p, W}_s = \Pi_{i \in [k]} \rho^{p, W}_{s_i}$ by definition. Since $V_j$ is a canonical basis subspace, 
\begin{equation}
	\rho^{p, W}_s = \rho^{p, W}_{s_0, \cdots, s_{k-1}}  = \bigotimes_{i \in [k]} \rho^{p, W}_{s_i}
\end{equation}
where each $s_i \in V_i$. When decomposing the generalized Pauli measurement in this manner, we often times write it as $\bigotimes_{i \in [k]} \left(\rho^{p, W}_{s_i}\right)_{V_i} \in \bigotimes_{i \in [k]} \bC_{V_i} = \bC^{2^{p\cdot m}}$ to emphasize the underlying Hilbert space. For a canonical basis subspace $V \subseteq \bF_{2^p}^m$, we write as the (normalized) state $\left(\ket{\text{ME}_{2^p}}^{\otimes m}_{A_1 B_1} \right)_V := \sum_{x \in V} \ket{x} \otimes \ket{x}$.

For $j \in [k]$,  $s_{\leq j}, t_{\leq j} \in V_{\leq j}$ and $r_{< j} \in V_{< j}$, and $P \in \{A,B\}$, define
\begin{align*}
	\left(\rho^{p,Z}_{[\decidertypedfunction^P_{\leq j}(s_{\leq j})| t_{\leq j}]}\right)_{V_{\leq j}} &:= \left(\rho^{p,Z}_{[\decidertypedfunction^P_{0,0}\left(s_{\leq j}\right)_0|\left(t_{\leq j}\right)_0]}\right)_{V_0} \otimes \left(\bigotimes_{1 \leq  i \leq j} \left( \rho^{p,Z}_{\left[\decidertypedfunction^P_{i,\left( t_{\leq j} \right)_{<i-1}} \left((s_{\leq j})_{i} \right)|\left(t_{\leq j}\right)_i \right]} \right)_{V_j} \right), \\
	\left(\rho^{p,X}_{[\left(\decidertypedfunction^P\right)^{\bot}_{\leq j, r_{< j}}(s_{\leq j})| t_{\leq j}]}\right)_{V_{\leq j}} &:= \left(\rho^{p,X}_{\left[\left(\decidertypedfunction^P_{0,0}\right)^{\bot}\left(s_{\leq j}\right)_0|\left(t_{\leq j}\right)_0 \right]}\right)_{V_0}  \otimes \left(\bigotimes_{1 \leq  i < j} \left( \rho^{p,X}_{\left[\left(\decidertypedfunction^P_{i,\left( r_{< j} \right)_{<i-1}}\right)^{\bot} \left((s_{\leq j})_{i} \right)|\left(t_{\leq j}\right)_i \right]} \right)_{V_j}\right), 
\end{align*} 
where 
\begin{equation*}
	s_{\leq j} = \sum_{i \in [j]} \left(	s_{\leq j} \right)_i  \in \bigoplus_{i \in [j]} V_i, \, t_{\leq j} = \sum_{i \in [j]} \left(	t_{\leq j} \right)_i \in \bigoplus_{i \in [j]} V_i, \, \text{ and } \, r_{< j} = \sum_{i \in [j-1]} \left(r_{< j} \right)_i \in \bigoplus_{i \in [j-1]} V_i. 
\end{equation*}
By definition, for $s, t \in \bF_{2^p}^m$, $\left(\rho^{p,Z}_{[\decidertypedfunction^P_{< k}(s)| t]}\right) = \left(\rho^{p,Z}_{[\decidertypedfunction^P(s)| t]}\right)$. Base on the synchronicity condition of~\Cref{fig:Introspectionverification}, we have the following claim.
\begin{claim}[Consistency of measurement output] \label{claim:QRsoundsyncmeasurement}
	For all $x \in \cX^{\text{Intro}}$ where $x$ is an introspection question, $P^{x}_a \approx_{O(\poly(k, \delta_1))} (P^{x}_a)^{op}$ over the state $\ket{\text{ME}_{2^{p}}}^{\otimes m}_{A_1 A_2} \otimes \ket{\text{Aux}}_{\alicealg}$. 
\end{claim}
\begin{proof}
	Fix an introspection question $x$, given a question pair sampled from $\mu^{\text{Intro}}$, there is a $O(\frac{1}{k^2})$ probability that the sampled question pair is $(x,x)$. Since $\strategy''$ is a strategy for $\cG^{\text{Intro}}$ which succeed with probability $1 - \delta_1$, this implies that $P^{x}_a \simeq_{O(\poly(k, \delta_1))} (P^{x}_a)^{op}$ by the ``synchronicity" clause given by~\Cref{fig:Introspectionverification}. The claim then follows from~\Cref{lem:closetodistance} to convert between $\simeq$ distance to $\approx$ distance.
\end{proof}

Based on the verification procedure given in~\Cref{fig:Introspectionverification}, we have the following claim about the approximation related to the Pauli $X$ measurement

\begin{claim}[Approximation of strategies related to the Pauli $X$ measurement] \label{claim:QRsoundXmeasure}
	Let $P \in \{A,B\}$ and $0 < j < k$, then 
	\allowdisplaybreaks{
	\begin{align} 
		\label{eq:QRproofX-1} P^{\text{Hide}_0, \, \decidertypedfunction^P}_{t^{\bot}_{\leq 0}, r_{> 0}} &\approx_{O(k, \delta_1)} \left(\rho^{p,X}_{[\left(\decidertypedfunction^P_{0,0}\right)^{\bot}(s_0)| t^{\bot}_{\leq 0}]}\right)_{V_0} \otimes  \left(\rho^{p,X}_{s_{> 0} = r_{> 0}}\right)_{_{V_{>0}}} \otimes \cI_{\alicealg}, 
	\end{align}
	where we partition the measurement outcome for $\rho^{p, X}$ as $\sum_{i \in [k]} s_i \in \bigoplus_{i \in [k]} V_i$ and $s_{> 0} \in  V_{> 0}$ in the above equation. 
}

\end{claim}
\begin{proof}
	 Fix $P \in \{A,B\}$, since the question pair (Gen Pauli, X) \textbf{--} ($\text{Hide}_0$, $\decidertypedfunction^{P}$) is sampled with probability $O(k)$ from the distribution $\mu^{\text{Intro}}$, and $\strategy''$ is a strategy for $\cG^{\text{Intro}}$ which succeed with probability $1 - \delta_1$, combining with~\eqref{eq:QRsoundnessapproxgenpauli}
	\begin{equation*}
		P^{\text{Hide}_0, \decidertypedfunction}_{(t^{\bot}_{\leq 0}, r_{> 0})} \simeq_{O(k, \delta_1)} \left(\rho^{p, X}_{s_0 = t^{\bot}_{\leq 0}, s_{> 0} =  r_{> 0}} \right)^{op},
	\end{equation*}
	where $s = s_0 + s_0^{C} +s_{> 0} = \ker{\decidertypedfunction^{P}_{0,0}} \oplus \ker{\decidertypedfunction^{P}_{0,0}}^{C} \oplus V_{> 0}$ according to the verification procedure from~\Cref{fig:Introspectionverification}. By applying~\Cref{lem:closetodistance}, ~\Cref{claim:QRsoundsyncmeasurement} and the triangle inequality for $\approx$ distance,
	\begin{equation*}
		P^{\text{Hide}_0, \decidertypedfunction}_{(t^{\bot}_{\leq 0}, r_{> 0})} \approx_{O(k, \delta_1)} \rho^{p, X}_{s_0 = t^{\bot}_{\leq 0}, s_{> 0} =  r_{> 0}} \otimes \cI_{\alicealg}.
	\end{equation*}
	Finally, by the definition of $\left(\decidertypedfunction^P_{0,0}\right)^{\bot}$, and $ \rho^{p, X}$ is projective, we obtain~\eqref{eq:QRproofX-1}.
\end{proof}

Similarly, we have the following claim about the approximation about the Pauli $Z$ measurement. 

\begin{claim}[Approximation of strategies related to the Pauli $Z$ measurement] \label{claim:QRsoundZmeasure}
	For every $P \in \{A,B\}$ and $1 \leq j < k$ the following hold
	\begin{align}
		\label{eq:QRproofZ-1} P^{\text{Sample}, \, \decidertypedfunction^P}_{s_{\text{sample}}} &\approx_{O(\poly(k, \delta_1))}  \rho^{p,Z}_{s_{\text{sample}}} \otimes \cI_{\alicealg}, \\
		\label{eq:QRproofZ-2} P^{\text{Intro}, \, \decidertypedfunction^P}_{x_P, a_P}   &\approx_{O(\poly(k, \delta_1))}  \rho^{p,Z}_{[\decidertypedfunction^P(s)| x_P]}  \otimes \cI_{\alicealg}, \\
		\label{eq:QRproofZ-3} P^{\text{Read}, \, \decidertypedfunction^P}_{t_{\text{Read}}^{\text{Line}}} &\approx_{O(\poly(k, \delta_1))}  \rho^{p,Z}_{[\decidertypedfunction^P(s)| t_{\text{Read}}^{\text{Line}}]} \otimes \cI_{\alicealg}, \\
		\label{eq:QRproofZ-4} P^{\text{Hide}_j, \, \decidertypedfunction^P}_{t_{< j}^{\text{Line}}} &\approx_{O(\poly(k, \delta_1))} \left(\rho^{p,Z}_{[\decidertypedfunction^P_{< j-1}(s_{< j})| t_{< j}^{\text{Line}}]}\right)_{V_{< j}} \otimes \cI_{V_{\geq j}} \otimes \cI_{\alicealg}.  
	\end{align}

\end{claim}

\begin{proof}
	Since, the proof for each approximation follows a similar structure as~\eqref{eq:QRproofX-1} from~\Cref{claim:QRsoundXmeasure}, we only give a rough sketch for the proof for each of the equation below.
	\begin{itemize}
		\item \Cref{eq:QRproofZ-1} follows from the verification procedure for (Gen Pauli, $Z$) \textbf{--} $(\text{Sample}, \decidertypedfunction^{P})$ question pair from~\Cref{fig:Introspectionverification}. 
		\item \Cref{eq:QRproofZ-2} follows from the verification procedure for $(\text{Sample}, \decidertypedfunction^{P})$ \textbf{--} $(\text{Intro},  \decidertypedfunction^{P})$ question pair, and applying the triangle inequality of $\approx$ distance to ~\Cref{eq:QRproofZ-1}. 
		\item \Cref{eq:QRproofZ-3} follows from the verification procedure for $(\text{Read},  \decidertypedfunction^{P})$ \textbf{--}  $(\text{Intro},  \decidertypedfunction^{P})$ question pair,  and applying the triangle inequality of $\approx$ distance to ~\Cref{eq:QRproofZ-2}. 
		\item For~\Cref{eq:QRproofZ-4} in the case where $j  = k-1$. The equation follows from the verification procedure for  $(\text{Hide}_{k-1}, \decidertypedfunction^{P})$ \textbf{--}   $(\text{Read},  \decidertypedfunction^{P})$ question pair. 
		\item For~\Cref{eq:QRproofZ-4} in the case where $0<j  < k-1$. The equation follows from an inductive proof using j = k-1 as the base case, and the inductive step follows from the verification procedure for $(\text{Hide}_{i}, \decidertypedfunction^{P})$ \textbf{--}  $(\text{Hide}_{i+1}, \decidertypedfunction^{P})$ and~\Cref{lem:QRsoundnessproofLem3}.
	\end{itemize}
\end{proof}

Base on the commutation with the Pauli-$X$ and Pauli-$Z$ measurement, we conclude the following approximation relationship related to the $P^{\text{Hide}_j, \, \decidertypedfunction^P}$ and $P^{\text{Read}, \, \decidertypedfunction^P}$. Since the proof of the below claim is almost identical to~\cite[Lemma 8.22, Lemma 8.23]{jiMIPRE2022a}, except we use~\Cref{claim:QRsoundsyncmeasurement} to swap the measurement operator to one register.
\begin{claim} \label{claim:structureHideandRead}
	For $P \in \{A,B\}$ and all $j \in [k]$, 
	\begin{align}
		\label{eq:QRHidevariable} P^{\text{Hide}_j, \, \decidertypedfunction^P}_{t^{\text{Line}}_{< j} ,t^{\bot}_{\leq j}, r_{> j}} &\approx_{O(\poly(k, \delta_1))}  \left(\rho^{p,X}_{[\left(\decidertypedfunction^P\right)^{\bot}_{< j, t^{\text{Line}}_{< j}}(s_{\leq j})| t^{\bot}_{\leq j}]} \cdot \rho^{p,Z}_{[\decidertypedfunction^P_{< j-1}(s_{< j})| t_{< j}^{\text{Line}}]}\right)_{V_{\leq j}} \otimes  \left(\rho^{p,X}_{s_{> j} = r_{> j} }\right)_{V_{> j}} \otimes \cI_{\alicealg}, \\
		\label{eq:QRReadvariable}  P^{\text{Read}, \decidertypedfunction^P}_{t_{\text{Read}, P}^{\bot},t_{\text{Read}, P}^{\text{Line}}} &\approx_{O(\poly(k, \delta_1))}  \left(\rho^{p,X}_{\left[\left(\decidertypedfunction^P\right)^{\bot}_{\leq k, t_{\text{Read}, P}^{\text{Line}} }(s)| t_{\text{Read}, P}^{\bot} \right]} \cdot \rho^{p,Z}_{[\decidertypedfunction^P(s)| t_{\text{Read}}^{\text{Line}}]} \right)_{V} \otimes \cI_{\alicealg}. 
	\end{align}
\end{claim}
\begin{proof}
	We start by showing~\Cref{eq:QRHidevariable}, the proof follows by an inductive argument. For the case where $j = 0$, this precisely follows from~\Cref{eq:QRproofX-1}. For the inductive step, fix $1 \leq i < k$ and assume~\Cref{eq:QRHidevariable} holds for all $0 \leq j \leq i$, we wish to show~\Cref{eq:QRHidevariable} for $i+1$. Since the question pair $(\text{Hide}_i, \,\decidertypedfunction^{P})$ \textbf{--}$(\text{Hide}_{i+1}, \,\decidertypedfunction^{P})$ in $\cG^{\text{Intro}}$ are selected with probability $O(\frac{1}{k})$, by the ``Hiding test" verification procedure given in~\Cref{fig:Introspectionverification}
	\begin{equation*}
		\sum_{\substack{\left(v_{\leq i}, v_{\leq i+1}^{\bot}, u_{> i+1}\right) \\ \quad \in  V_{< i+1} \times V_{\leq i+1} \times V_{> i+1}}} \braket{\psi| P^{\text{Hide}_i, \, \decidertypedfunction^P}_{ \left(\substack{t_{< i}^{\text{line}} = v_{< i},\\
					t_{\leq i}^{\bot} = v_{\leq i}^{\bot}, \\
					\left[ \decidertypedfunction^{\bot}_{i+1, v_{\leq i} } (r_{i})| v_{i+1}^{\bot}, \right] \\
					r_{> i+1} =   u_{> i+1}}. \right) } \cdot \left( P^{\text{Hide}_{i+1}\, \decidertypedfunction^P}_{ v_{\leq i}, v_{\leq i+1}^{\bot},u_{> i+1}}  \right)^{op} |\psi} \geq 1 - O(\poly(k, \delta_1))
	\end{equation*}
	where we partition the variable according to the convention that $v_{\leq i} = v_{< i} + v_i \in V_{< i} \oplus V_{i}$ and likewise with the other variables in the sum. We wish to apply~\Cref{lem:QRsoundnessproofLem3} where the ``A" measurement is $P^{\text{Hide}_i, \, \decidertypedfunction^P}$, the ``B" measurement is $P^{\text{Hide}_{i+1}, \, \decidertypedfunction^P}$, the ``x" variable is $v_{< i}$, the ``y" variable is $v_{i}$ and the ``z" variable are $( v_{\leq i+1}^{\bot}, u_{> i+1})$. By this formulation, for $u_{i} \in V_i$, we can rewrite the ``A" measurement as 
	\begin{equation}
		P^{\text{Hide}_i, \, \decidertypedfunction^P,u_{i}}_{v_{< i}, v_{\leq i}^{\bot},\left[ \decidertypedfunction^{\bot}_{i+1, v_{\leq i} } (u_{i})| v_{i+1}^{\bot}, \right],  u_{> i+1}}
	\end{equation}
	using the summation for the measurement outcome above, where in this case, we divide up the output $r_{> i} = r_{i} + r_{> i+1}$ and write it as the third and fourth output of $P^{\text{Hide}_i, \, \decidertypedfunction^P, v_{i+1}}$. By the inductive hypothesis,
	\begin{align}
			\nonumber &P^{\text{Hide}_i, \, \decidertypedfunction^P, v_{i+1}}_{v_{< i}, v_{\leq i}^{\bot},\left[ \decidertypedfunction^{\bot}_{i+1, v_{\leq i} } (u_{i})| v_{i+1}^{\bot}, \right],  u_{> i+1}} \\
			\nonumber &\approx_{O(\poly(k, \delta_1))} \left(\rho^{p,X}_{[\left(\decidertypedfunction^P\right)^{\bot}_{< i, u_{< i}}(s_{\leq i})| u^{\bot}_{\leq i}]} \cdot \rho^{p,Z}_{[\decidertypedfunction^P_{< i-1}(s_{< i})| t_{< i}]}\right)_{V_{\leq i}} \\
			\label{eq:QRHidevariablepf-1}&\qquad \, \otimes \left(\rho^{p,X}_{\left[ \decidertypedfunction^{\bot}_{i+1, v_{\leq i} } (u_{i})| v_{i+1}^{\bot}, \right]}\right)_{V_{i+1}} \otimes \left(\rho^{p,X}_{s_{> i+1} = r_{> i+1} }\right)_{V_{> i +1}} \otimes \cI_{\alicealg}. 
	\end{align}
	Similarly, since $P^{\text{Hide}_{i+1}, \, \decidertypedfunction^P}$ is projective, we can write the corresponding $B_{x,y}$ as $P^{\text{Hide}_{i+1}, \, \decidertypedfunction^P}_{v_{< i}, v_i} = P^{\text{Hide}_{i+1}, \, \decidertypedfunction^P}_{v_{< i+1}}$. Hence
	\allowdisplaybreaks{
	\begin{align*}
		P^{\text{Hide}_{i+1}\, \decidertypedfunction^P}_{ t_{< i+1}^{\text{Line}},t_{\leq i+1}^{\bot}, r_{> i+1} }  &\approx_{O(\poly(k, \delta_1))} \left( P^{\text{Hide}_{i+1}\, \decidertypedfunction^P}_{ t_{< i+1}^{\text{Line}},t_{\leq i+1}^{\bot}, r_{> i+1} } \right)^{op} \\
		&\approx_{O(\poly(k, \delta_1))}  P^{\text{Hide}_i, \, \decidertypedfunction^P,u_{i}}_{t_{< i}^{\text{Line}},t_{\leq i}^{\bot},\left[ \decidertypedfunction^{\bot}_{i+1, t_{< i+1}^{\text{Line}} } (r_i)| t_{i+1}^{\bot}, \right],  r_{> i+1}} \cdot \left(P^{\text{Hide}_{i+1}, \, \decidertypedfunction^P}_{ t_{< i+1}^{\text{Line}}}\right)^{op} \\
		&\approx_{O(\poly(k, \delta_1))}  \left(\rho^{p,X}_{[\left(\decidertypedfunction^P\right)^{\bot}_{< i, u_{< i}}(s_{\leq i})| u^{\bot}_{\leq i}]} \cdot \rho^{p,Z}_{[\decidertypedfunction^P_{< i-1}(s_{< i})| t_{< i}]}\right)_{V_{\leq i}} \\
		&\qquad \, \otimes \left(\rho^{p,X}_{\left[ \decidertypedfunction^{\bot}_{i+1, v_{\leq i} } (u_{i})| v_{i+1}^{\bot}, \right]}\right)_{V_{i+1}} \otimes \left(\rho^{p,X}_{s_{> i+1} = r_{> i+1} }\right)_{V_{> i +1}} \otimes \cI_{\alicealg} \cdot \left(P^{\text{Hide}_{i+1}, \, \decidertypedfunction^P}_{ t_{< i+1}^{\text{Line}}}\right)^{op} \\
		&\approx_{O(\poly(k, \delta_1))}  \left(\rho^{p,X}_{[\left(\decidertypedfunction^P\right)^{\bot}_{< i, u_{< i}}(s_{\leq i})| u^{\bot}_{\leq i}]} \cdot \rho^{p,Z}_{[\decidertypedfunction^P_{< i-1}(s_{< i})| t_{< i}]}\right)_{V_{\leq i}} \cdot \left(\rho^{p,Z}_{[\decidertypedfunction^P_{< i}(s_{< i+1})| t_{< i+1}^{\text{Line}}]}\right)_{V_{\leq i}}  \\
		&\qquad \, \otimes \left(\rho^{p,X}_{\left[ \decidertypedfunction^{\bot}_{i+1, v_{\leq i} } (u_{i})| v_{i+1}^{\bot}, \right]}\right)_{V_{i+1}} \otimes \left(\rho^{p,X}_{s_{> i+1} = r_{> i+1} } \right)_{V_{> i +1}} \otimes \cI_{\alicealg} \\
		&= \left(\rho^{p,X}_{[\left(\decidertypedfunction^P\right)^{\bot}_{< i+1, t^{\text{Line}}_{< i+1}}(s_{\leq i+1})| t^{\bot}_{\leq i+1}]} \cdot \rho^{p,Z}_{[\decidertypedfunction^P_{< i}(s_{< i+1})| t_{< i+1}^{\text{Line}}]}\right)_{V_{\leq i}} \otimes  \left(\rho^{p,X}_{s_{> i} = r_{> i+1} }\right)_{V_{> i}} \otimes \cI_{\alicealg}
	\end{align*}
}
	where the first inequality follows from~\Cref{claim:QRsoundsyncmeasurement} and~\Cref{lem:closetodistance}, the second line follows from~\Cref{lem:QRsoundnessproofLem3} labelled above. Line 3 follows from applying~\Cref{lem:combmeasure} along with~\eqref{eq:QRHidevariablepf-1}. Line 4 follows from applying~\Cref{lem:combmeasure} along with~\eqref{eq:QRproofZ-4} and~\Cref{claim:QRsoundsyncmeasurement}. The last line follows from the definition of $\decidertypedfunction^P$ and $\left(\decidertypedfunction^P\right)^{\bot}$. This shows the case for $i+1$, which show~\eqref{eq:QRHidevariable}. \eqref{eq:QRReadvariable} follows from a similar argument using the question pair $(\text{Hide}_{k-1}, \,\decidertypedfunction^{P})$ \textbf{--} $(\text{Read}, \,\decidertypedfunction^{P})$. 
\end{proof}

Finally, we show that the measurement $P^{\text{Intro}, \decidertypedfunction^P}$ can be decompose as a tensor product of a ``question sampling" PVM using the Generalized Pauli $Z$ measurement and a ``game strategy" PVM. We remark that this is an analogue for~\cite[Lemma 8.24]{jiMIPRE2022a} for the commuting operator model and the proof follows a similar structure. 
\begin{claim} \label{claim:QRconstructingstrategy}
	Fix $P \in \{A,B\}$. For every $j \in [k+1]$ and $\left(x_p\right)_{< j} \in V_{< j}$, there exist a set of POVM $\left\{\left(P^{\text{Intro}, \, \decidertypedfunction^P, \left(x_p\right)_{< j}}_{x_{\geq j}, a_P} \right)_{V_{\geq j} \alicealg} \right\}_{ x_{\geq j} \in V_{\geq j}, a_P \in \cA} \in \cB(\bC_{V_{\geq j}}) \otimes \alicealg$ such that 
	\begin{equation*}
		\left(P^{\text{Intro} \, \decidertypedfunction^P}_{(x_P)_{< j}, (x_P)_{\geq j}, a_P} \right)_{V \alicealg}	\approx_{O(\poly(k, \delta_1)^{1/2^j})} \left(\rho^{p,Z}_{[\decidertypedfunction^P_{< j}(s_{< j})|(x_P)_{< j} ]}\right)_{V_{< j}} \otimes \left(P^{\text{Intro}, \, \decidertypedfunction^P, \left(x_p\right)_{< j}}_{(x_P)_{\geq j}, a_P} \right)_{V_{\geq j} \alicealg}, 
	\end{equation*}
	where the summation are over $((x_P)_{< i}, (x_P)_{\geq i}, a_P) \in V_{< i} \times V_{\geq i} \times \cA$.
\end{claim}

\begin{proof}
	We show this claim via induction. For $k = 0$, the claim trivially follows by setting 
	\begin{equation*}
		P^{\text{Intro} \, \decidertypedfunction^P, 0}_{(x_P)_{< i}, (x_P)_{\geq i}, a_P} = P^{\text{Intro}, \decidertypedfunction^P}_{(x_P)_{< i}+ (x_P)_{\geq i}, a_P}.
	\end{equation*}

	For the inductive step, fix $1 \leq i < k+1$ and assume~\Cref{eq:QRHidevariable} holds for all $0 \leq j \leq i$, we wish to show the claim for $i+1$. For $(x_P)_{< i}$, let $\left(P^{\text{Intro} \, \decidertypedfunction^P}_{(x_P)_{< j}, (x_P)_{\geq j}, a_P}\right)_{V \alicealg}$ be the PVM guaranteed by the inductive hypothesis. Let $x \sim \decidertypedfunction^P_{< i}$ denote the distribution where $s$ is first sampled uniformly randomly from $V_{<i}$, and the first $i$-th levels of $ \decidertypedfunction^P$ are then applied to $s$. The goal is to use~\Cref{lem:PauliTwirl} in order to construct the measurement require for the lemma. To do so, we show that on average over $(x_P)_{< i}  \sim \decidertypedfunction^P_{< i}$, the following set of equations hold:
	\begin{align}
		\label{eq:ARproofdecompIntro-1} &\left(P^{\text{Intro}, \, \decidertypedfunction^P, \left(x_p\right)_{< i}}_{(x_P)_{i}} \right)_{V_{\geq i} \alicealg} \approx_{O(\poly(k, \delta_1)^{1/2^j})} \left(\rho_{\left[\decidertypedfunction_{i, (x_p)_{< i}}^{P} (s_i) | (x_P)_i \right]}^{p, Z}\right)_{V_i} \otimes \cI_{V_{> i}} \otimes \cI_\alicealg \\
		=\label{eq:ARproofdecompIntro-2} &\left[\left(P^{\text{Intro}, \, \decidertypedfunction^P, \left(x_p\right)_{< i}}_{(x_P)_{\geq i}, a_P} \right)_{V_{\geq i} \alicealg},  \left(\rho_{z}^{p, Z}\right)_{V_i} \otimes \cI_{V_i^C} \otimes \cI_\alicealg)\right]\approx_{O(\poly(k, \delta_1)^{1/2^j})} 0 \\
		\label{eq:ARproofdecompIntro-3} &\left[\left(P^{\text{Intro}, \, \decidertypedfunction^P, \left(x_p\right)_{< i}}_{(x_P)_{\geq i}, a_P} \right)_{V_{\geq i} \alicealg} ,  \left(\rho_{\left[\left(\decidertypedfunction_{i, (x_p)_{< i}}^{P}\right)^{\bot} (s_i) | (x_P)^{\bot} \right]}^{p, X}\right)_{V_i} \otimes \cI_{V_{> i}} \otimes \cI_\alicealg) \right] \approx_{O(\poly(k, \delta_1)^{1/2^j})} 0, 
	\end{align}
	 where we partition the output $(x_P)_{\leq i}$ from $P^{\text{Intro}, \, \decidertypedfunction^P, \left(x_p\right)_{< i}}$ as $(x_P)_{< i} + (x_P)_{i} \in V_{< i} \oplus V_{i}$. In this case, the ``x" variable from~\Cref{lem:PauliTwirl} correspond to $(x_p)_{< i}$, the ``t" variable corresponds to $(x_P)_i$ and the ``a" variable corresponds to $((x_P)_{> i}, a_P)$. 
	 
	 For~\eqref{eq:ARproofdecompIntro-1}, by~\Cref{eq:QRproofZ-3} when restricted to the output to $(x_P)_{< i+1}$, 
	 \begin{align*}
	 		\left(P^{\text{Intro} \, \decidertypedfunction^P}_{(x_P)_{< i+1}} \right)_{V \alicealg}	&\approx_{O(\poly(k, \delta_1)^{1/2^j})} \left(\rho^{p,Z}_{[\decidertypedfunction^P_{< i+1}|(x_P)_{< i+1} ]}\right)_{V_{< i+1}}\otimes \cI_{V_{> i}} \otimes \cI_\alicealg,  \\
	 		&=\left(\rho^{p,Z}_{[\decidertypedfunction^P_{< i}|(x_P)_{< i} ]}\right)_{V_{< i}} \otimes \left(\rho_{\left[\decidertypedfunction_{i, (x_p)_{< i}}^{P} | (x_P)_i \right]}^{p, Z}\right)_{V_i} \otimes \cI_{V_{> i}} \otimes \cI_\alicealg
	 \end{align*}
 	where the second equality follows from the definition of $\decidertypedfunction^P$. \eqref{eq:ARproofdecompIntro-1} then follows from~\Cref{lem:QRsoundnessproofcomm1}. 
 	
 	For~\eqref{eq:ARproofdecompIntro-2}. By using the fact that $\strategy''$ succeed with probability $1 - \delta_1$, the fact that the question pair (Gen Pauli, W) \textbf{--} (Sample, $\decidertypedfunction^P$) is sampled with probability $O(k)$ from the distribution $\mu^{\text{Intro}}$, and~\Cref{lem:closetodistance}, ~\eqref{eq:QRsoundnessapproxgenpauli} and the triangle inequality of $\approx$ distance, we have
	\begin{equation}
		\left(P^{\text{Sample} \, \decidertypedfunction^P}_{s_{< i+1}} \right)_{V \alicealg} \approx_{O(\poly(k, \delta_1)^{1/2^j})} \left(\rho^{p, Z}_{s_{< i+1}} \right)_{V_{< i+1}} \otimes \cI_{V_{\geq i+1}} \otimes \cI_{\alicealg}.
	\end{equation}
	By applying the same argument with the (Sample, $\decidertypedfunction^P$) \textbf{--}  (Intro, $\decidertypedfunction^P$) along with the inductive hypothesis
	\begin{equation}
		\left(P^{\text{Sample} \, \decidertypedfunction^P}_{a_P} \right) \approx_{O(\poly(k, \delta_1)^{1/2^j})} \left(\rho^{p,Z}_{[\decidertypedfunction^P_{< i}(s_{< i})|(x_P)_{< i} ]}\right)_{V_{< i}} \otimes \left(P^{\text{Intro}, \, \decidertypedfunction^P, \left(x_p\right)_{< i}}_{a_P} \right)_{V_{\geq i} \alicealg}.
	\end{equation}	
	Hence, by~\Cref{lem:QRsoundnessproofcomm1}
	\begin{equation*}
		[\left(\rho^{p, Z}_{s_{< i+1}} \right)_{V_{< i+1}} \otimes \cI_{V_{\geq i+1}} \otimes \cI_{\alicealg}, \left(\rho^{p,Z}_{[\decidertypedfunction^P_{< i}(s_{< i})|(x_P)_{< i} ]}\right)_{V_{< i}} \otimes \left(P^{\text{Intro}, \, \decidertypedfunction^P, \left(x_p\right)_{< i}}_{a_P} \right)_{V_{\geq i} \alicealg}] \approx_{O(\poly(k, \delta_1)^{1/2^j})} 0.
	\end{equation*}
	Finally, since the underlying state are $\ket{\text{ME}_{2^p}}^{\otimes m}$ on the registers $V$, ~\eqref{eq:ARproofdecompIntro-2} follows from~\Cref{lem:QRsoundnessproofcomm2}. For~\eqref{eq:ARproofdecompIntro-3}, by restricting the output from~\Cref{eq:QRReadvariable}
	\begin{equation*}	
		 P^{\text{Read}, \decidertypedfunction^P}_{\left(t_{\text{Read}, P}^{\bot}\right)_{i}, \left(t_{\text{Read}, P}^{\text{Line}}\right)_{< i}} \approx_{O(\poly(k, \delta_1)^{1/2^j})}  \left(\rho^{p,Z}_{\left[\decidertypedfunction^P_{< i}(s_{< i})| \left(t_{\text{Read}, P}^{\text{Line}}\right)_{< i} \right]}\right)_{V_{< j}}  \otimes \left(\rho_{\left[\left(\decidertypedfunction_{i,  \left(t_{\text{Read}, P}^{\text{Line}}\right)_{< i}}^{P}\right)^{\bot} (s_i)| \left(t_{\text{Read}, P}^{\bot}\right)_{i} \right]}^{p, X}\right)_{V_i}.
	\end{equation*}
	Similarly, by considering the  (Read, $\decidertypedfunction^P$) \textbf{--}  (Intro, $\decidertypedfunction^P$) along with the inductive hypothesis
	\begin{equation}
		\left(P^{\text{Read} \, \decidertypedfunction^P}_{a_P} \right) \approx_{O(\poly(k, \delta_1)^{1/2^j})} \left(\rho^{p,Z}_{[\decidertypedfunction^P_{< i}(s_{< i})|(x_P)_{< i} ]}\right)_{V_{< i}} \otimes \left(P^{\text{Intro}, \, \decidertypedfunction^P, \left(x_p\right)_{< i}}_{a_P} \right)_{V_{\geq i} \alicealg}.
	\end{equation}	
	\eqref{eq:ARproofdecompIntro-3} then follows from~\Cref{lem:QRsoundnessproofcomm1} to obtain the commutation relationship, followed by~\Cref{lem:QRsoundnessproofcomm2} on the register $V_{< i}$.
	
	By~\Cref{lem:PauliTwirl}, by taking the ``t" variable as $(x_P)_{i} \in V_i$ and the ``a" variable as $\left( (x_P)_{> i+1}, a_P \right) \in V_{> i+1} \times \cA $, for every $\left(x_p\right)_{< i} + (x_P)_{i} = \left(x_p\right)_{< i+1} \in V_{i+1}$, there exists a set of POVM measurement $\widehat{\left(P^{\text{Intro}, \, \decidertypedfunction^P, \left(x_p\right)_{< i+1}}_{ (x_P)_{> i+1}, a_P } \right)_{V_{> i} \alicealg}}$ such that, on expectation over $\left(x_p\right)_{< i} \sim \decidertypedfunction^P_{< i}$
	\begin{equation*}
		\left(P^{\text{Intro}, \, \decidertypedfunction^P, \left(x_p\right)_{< i}}_{(x_P)_{\geq i}, a_P} \right)_{V_{\geq i} \alicealg} 	\approx_{O(\poly(\eps))}   \left(\rho_{\left[\decidertypedfunction_{i, (x_p)_{< i}}^{P} | (x_P)_i \right]}^{p, Z}\right)_{V_i} \otimes \widehat{\left(P^{\text{Intro}, \, \decidertypedfunction^P, \left(x_p\right)_{< i+1}}_{ (x_P)_{\geq i+1}, a_P } \right)_{V_{> i} \alicealg}}. 
	\end{equation*}
	Since $\bra{ME_{2^P}}^{\otimes m}_{V_i} \left(\left(\rho^{p,Z}_{[\decidertypedfunction^P_{< i}(s_{< i})|(x_P)_{< i} ]}\right)_{V_{< i}} \otimes \cI_V \right)_{A_1 B_1} \ket{\text{ME}_{2^P}}^{\otimes m}_{V_i} $ precisely describes the distribution $\left(x_p\right)_{< i}\sim \decidertypedfunction^P_{< i}$. By~\Cref{lem:QRsounddecomp}, 
	\begin{align*}
		&\left(\rho^{p,Z}_{[\decidertypedfunction^P_{< i}(s_{< i})|(x_P)_{< i} ]}\right)_{V_{< i}} \otimes \left(P^{\text{Intro}, \, \decidertypedfunction^P, \left(x_p\right)_{< i}}_{(x_P)_{\geq i}, a_P} \right)_{V_{\geq i} \alicealg} \\
		&\approx_{O(\poly(\eps))}   \left(\rho^{p,Z}_{[\decidertypedfunction^P_{< i}(s_{< i})|(x_P)_{< i} ]}\right)_{V_{< i}} \otimes \left(\rho_{\left[\decidertypedfunction_{i, (x_p)_{< i}}^{P} (s_i)| (x_P)_i \right]}^{p, Z}\right)_{V_i} \otimes \widehat{\left(P^{\text{Intro}, \, \decidertypedfunction^P, \left(x_p\right)_{< i+1}}_{ (x_P)_{\geq i+1}, a_P } \right)_{V_{> i} \alicealg}} \\
		&= \left(\rho^{p,Z}_{[\decidertypedfunction^P_{< i+1}(s_{< i+1})|(x_P)_{< i+1} ]}\right)_{V_{< i+1}} \otimes \widehat{\left(P^{\text{Intro}, \, \decidertypedfunction^P, \left(x_p\right)_{< i+1}}_{ (x_P)_{\geq i+1}, a_P } \right)_{V_{> i} \alicealg}},
	\end{align*}
	where the last line follows from the definition of $\decidertypedfunction^P$. Thus, combining with the inductive hypothesis, 
	\begin{equation} 	\label{eq:ARproofdecompIntro-4}
		\left(P^{\text{Intro} \, \decidertypedfunction^P}_{(x_P)_{< i+1}, (x_P)_{\geq i+1}, a_P} \right)_{V \alicealg}	\approx_{O(\poly(k, \delta_1)^{1/2^j})} \left(\rho^{p,Z}_{[\decidertypedfunction^P_{< i+1}(s_{< i+1})|(x_P)_{< i+1} ]}\right)_{V_{< i+1}} \otimes \widehat{\left(P^{\text{Intro}, \, \decidertypedfunction^P, \left(x_p\right)_{< i+1}}_{(x_P)_{\geq i+1}, a_P} \right)_{V_{\geq i+1} \alicealg}}, 
	\end{equation}
	Finally, we wish to replace each of the $\widehat{\left(P^{\text{Intro}, \, \decidertypedfunction^P, \left(x_p\right)_{< i+1}}_{(x_P)_{\geq i+1}, a_P} \right)_{V_{\geq i+1} \alicealg}}$ with a set of PVM via~\Cref{lem:orthogonalizationlemma}. To do so, we show the following lemma
	\begin{lemma}[Approximation of almost projective measurements] \label{lem:approximationofPOVMtoPVM}
		Let $\cA$ be finite sets, and $\alicealg \subseteq \BofH$ be a von Neumann algebra. Let $\{ A_{a} \}_a \subseteq \alicealg$ be a set of PVM and $\{ B_{a} \}_a \subseteq \alicealg$ be a set of POVM such that 
		\begin{equation*}
			A_{a} \approx_{\eps}  B_{a}
		\end{equation*}
		for some state $\ket{\psi} \in \cH$ and some $\eps > 0$. Then $\braket{\psi|B_a^2|\psi} \geq 1 - O(\sqrt{\eps})$.
	\end{lemma}
	
	\begin{proof}
		Since $\{A_a^x\}$ is a set of PVM, by~\Cref{lem:closetodistance}, $A_{a} \approx_{\sqrt{\eps}}  B_{a}$ over $\ket{\psi}$ and $\eps > 0$. By definition, we have $\sum_{a} \braket{\psi| A_a B_a|\psi} \geq 1 - \sqrt{\eps}$, or
		\begin{align*}
			\sqrt{\eps} &\geq 1 -  \sum_{a} \braket{\psi| A_a B_a|\psi} \geq 1 - \sqrt{\sum_a \braket{\psi| A_a^2|\psi}} \cdot  \sqrt{\sum_a \braket{\psi| B_a^2|\psi}} = 1 - \sqrt{\sum_a \braket{\psi| B_a^2|\psi}} 
		\end{align*}
		where the last equality follows from $\{A_a\}_{a \in\cA}$ being a set of PVM. This implies that $ \sum_a \braket{\psi| B_a^2|\psi}  \geq (1 - \sqrt{\eps})^2 = 1 - O(\sqrt{\eps})$, as desired. 
	\end{proof}
	Applying the above lemma to~\eqref{eq:ARproofdecompIntro-4} along with~\Cref{lem:QRsounddecomp}, this implies that on expectation over $(x_P)_{<i+1} \sim \decidertypedfunction^P_{<i+1}$, we have $\widehat{\left(P^{\text{Intro}, \, \decidertypedfunction^P, \left(x_p\right)_{< i+1}}_{(x_P)_{\geq i+1}, a_P} \right)_{V_{\geq i+1} \alicealg}^2 } \approx_{O(\poly(k, \delta_1))^{1/2^{j+1}}} 0$. Hence, by~\Cref{lem:orthogonalizationlemma}, there exist sets of PVM $\left\{\left(P^{\text{Intro}, \, \decidertypedfunction^P, \left(x_p\right)_{< i+1}}_{(x_P)_{\geq i+1}, a_P} \right)_{V_{\geq i+1} \alicealg} \right\}_{{(x_P)_{\geq i+1}} \in V_{\geq i+1}, \, a_P \in  \cA}$ indexed by $ \left(x_p\right)_{< i+1} \in V_{< i+1}$ such that, on expectation over $(x_P)_{<i+1} \sim \decidertypedfunction^P_{<i+1}$, 
	\begin{equation*}
		\left(P^{\text{Intro} \, \decidertypedfunction^P}_{(x_P)_{< i+1}, (x_P)_{\geq i+1}, a_P} \right)_{V \alicealg}  \approx_{O(\poly(k, \delta_1))^{1/2^{j+1}}} \left(P^{\text{Intro}, \, \decidertypedfunction^P, \left(x_p\right)_{< i+1}}_{(x_P)_{\geq i+1}, a_P} \right)_{V_{\geq i+1} \alicealg}.
	\end{equation*}
	By applying~\Cref{lem:QRsounddecomp} again, and the triangle inequality of $\approx$ distance applied to~\eqref{eq:ARproofdecompIntro-4},
	\begin{equation*}
		\left(P^{\text{Intro} \, \decidertypedfunction^P}_{(x_P)_{< i+1}, (x_P)_{\geq i+1}, a_P} \right)_{V \alicealg}  \approx_{O(\poly(k, \delta_1))^{1/2^{j+1}}} \left(\rho^{p,Z}_{[\decidertypedfunction^P_{< i+1}(s_{< i+1})|(x_P)_{< i+1} ]}\right)_{V_{< i+1}} \otimes \left(P^{\text{Intro}, \, \decidertypedfunction^P, \left(x_p\right)_{< i+1}}_{(x_P)_{\geq i+1}, a_P} \right)_{V_{\geq i+1} \alicealg}.
	\end{equation*}
	This shows the claim for $i+1$, thus concluding the proof.

\end{proof} 

We remark that the dependency of $\frac{1}{1/2^j}$ power in the above lemma arises from using~\Cref{lem:orthogonalizationlemma} in order to make force the POVM guaranteed by~\Cref{lem:PauliTwirl} to be PVMs. Since $k$ is assumed to be a constant in this paper and $j \in[k+1]$, this power dependency does not change the result of this paper. 

By~\Cref{claim:QRconstructingstrategy} for the case where $j = k$ and $P = A$, for every $x_P \in \bF_{2^p}^m$, there exist a set of PVM $\left\{\left(P^{\text{Intro}, \, \decidertypedfunction^P, x_A}_{a_A} \right)_{\alicealg} \right\}_{a_A \in \cA} \subseteq \alicealg$ such that

\begin{equation} \label{eq:QRsoundnessAlice}
	\left(P^{\text{Intro} \, \decidertypedfunction^A}_{x_A, a_A} \right)_{A_1 \alicealg} \approx_{O(\poly(k, \delta_1)^{1/2^k})} \left(\rho^{p,Z}_{[\decidertypedfunction^A|(x_A) ]}\right)_{A_1} \otimes \left(P^{\text{Intro}, \, \decidertypedfunction^P, x_A}_{a_A} \right)_{\alicealg}, 
\end{equation}

By the same argument as~\Cref{claim:QRconstructingstrategy} applied to $\left(P^{\text{Intro} \, \decidertypedfunction^P}_{x_P, a_P} \right)_{B_1 \alicealg}^{op} \subseteq \cB(\bC_{V}) \otimes \alicealg'$ and take the case where $j = k$ and $P = B$, for every $x_P \in \bF_{2^p}^m$, there exist a set of PVM $\left\{\left(Q^{\text{Intro}, \, \decidertypedfunction^P, x_B}_{a_B} \right)_{\alicealg} \right\}_{a_B \in \cA} \subseteq \alicealg$ such that
\begin{equation}\label{eq:QRsoundnessBob}
	\left(P^{\text{Intro} \, \decidertypedfunction^B}_{x_B, a_B} \right)_{B_1 \alicealg}^{op} \approx_{O(\poly(k, \delta_1)^{1/2^k})} \left(\rho^{p,Z}_{[\decidertypedfunction^B|(x_B) ]}\right)_{B_1} \otimes \left(Q^{\text{Intro}, \, \decidertypedfunction^P, x_B}_{a_B} \right)_{\alicealg}^{op}. 
\end{equation}

Since the question pair (Intro, $\decidertypedfunction^A$) \textbf{--} (Intro, $\decidertypedfunction^B$) is sampled with probability $O(k)$ from the distribution $\mu^{\text{Intro}}$. This means that whenever the question/answer pair $(x_A, x_B, a_A,a_B)$ are sampled by the measurement $\braket{\psi|	\left(P^{\text{Intro} \, \decidertypedfunction^A}_{x_A, a_A} \right)_{A_1 \alicealg} 	\left(P^{\text{Intro} \, \decidertypedfunction^B}_{x_B, a_B} \right)_{B_1 \alicealg}^{op}|\psi}$, by the ``Introspection of $\cG$" question clause,  from~\Cref{fig:Introspectionverification}, the question/answer pair succeed in $\cG$ (i.e. $D(x_A, x_B, a_A,b_B) =1$) with probability at least $1 - O(\poly(k, \delta_1))$. Combining with~\eqref{eq:QRsoundnessAlice}, ~\eqref{eq:QRsoundnessBob} and the definition of $\ket{\psi}$ from~\Cref{eq:QRsoundnessstrategy}, this shows that the question/answer pair sampled by
\begin{align*}
	 &\left(\bra{\text{ME}_{2^p}}^{\otimes m}\left(\rho^{p,Z}_{[\decidertypedfunction^A|(x_A) ]}\right)_{A_1} \otimes \left(\rho^{p,Z}_{[\decidertypedfunction^B|(x_B) ]}\right)_{B_1} \ket{\text{ME}_{2^p}}^{\otimes m} \right)_{A_1 B_1} \\
	 &\qquad \otimes \left(\braket{\text{Aux}| \left(P^{\text{Intro}, \, \decidertypedfunction^P, x_A}_{a_A} \right) \cdot \left(Q^{\text{Intro}, \, \decidertypedfunction^P, x_B}_{a_B} \right)^{op} \text{Aux}} \right)_{\alicealg}
\end{align*}
succeed on $\cG$ with probability at least $1 - O(\poly(k, \delta_1)) - O(\poly(k, \delta_1)^{1/2^k}) = 1 - O(\poly(k, \delta_1)^{1/2^k})$. To conclude the proof, we see that $(x_A, x_B)$ sampled from the measurement
\begin{equation*}
	\left(\bra{\text{ME}_{2^p}}^{\otimes m}\left(\rho^{p,Z}_{[\decidertypedfunction^A|(x_A) ]}\right)_{A_1} \otimes \left(\rho^{p,Z}_{[\decidertypedfunction^B|(x_B) ]}\right)_{B_1} \ket{\text{ME}_{2^p}}^{\otimes m} \right)_{A_1 B_1}
\end{equation*}
 are equal to $(x_A,x_B) \sim \mu$. This shows that the strategy
\begin{equation}
	\strategy^{\cG} =  \left(\cL^2(\alicealg,\tau), \ket{\text{Aux}}_{\alicealg}, \{ P^{\text{Intro}, \, \decidertypedfunction^P, x_A}_{a_A} \}, \{Q^{\text{Intro}, \, \decidertypedfunction^P, x_B}_{a_B} \} \right),
\end{equation}
satisfies $\omega(\cG, \strategy^{\cG}) > 1 - O(\poly(k, \delta_1)^{1/2^k}) = 1 - O(\poly(\exp{k}, \eps))  $ by the initial definition of $\delta_1$. This shows the ``soundness" clause for~\Cref{thm:Introspectgame}.

\subsection{Proof for the ``soundness" clause for the answer reduction transformation}  \label{sec:ARsoundness}
As mentioned in~\Cref{sec:Answerreductionproof}, this subsection follows a similar structure as~\cite[Section 10.7]{jiMIPRE2022a}

Fix $\alpha, n \in \bN$, let $\cG_n = (\cX_n, \cA_n,\mu_n, D_n)$ be the $n$th game of $\verifiersequence$, and let $\cG_n^{\Answerred} = (\cX_n^{\Answerred}, \cA_n^{\Answerred},\mu_n^{\Answerred}, D_n^{\Answerred})$ and $\cG_n^{\text{Ora}}  = (\cX_n^{\text{Ora}}, \cA_n^{\text{Ora}},\mu_n^{\text{Ora}}, D_n^{\text{Ora}})$ denote the answer reduction transformation given in~\Cref{sec:Answerreductionproof} and~\Cref{sec:Oracularization}. Given a question label $x \in \cX_n^{\Answerred}$, we write $x = (x^{\text{Ora}}, x^{\text{LDL}}, (x^{\text{game}},x^{\text{LDC}}))$ where each elements corresponds to the ``oracularizable question label", ``SLDT question label", question content for the original game $\cG_n$, and  question content for SLDT from~\Cref{fig:Answerredsample} respectively. When we fix certain question labels, we might omit that portion of the question label for simplicity notation. Furthermore, let $(\mathtt{PCPParameter}_{\alpha}, \mathtt{ComputePCP}_{\alpha})$ be the two Turing machine guaranteed by~\Cref{thm:classicalPCPP}, and let $(m^{\text{ans}},m, g, p) = \mathtt{PCPParameter}_{\alpha}(n)$. 

Before we start giving a proof of the ``soundness clause", we first give a first overview on how the proof goes. To show the ``soundness" clause given in~\Cref{eq:ARsoundnessclause}, it is equivalent to show that there exist a polynomial function $\textbf{t}_{\alpha}^{\Answerred}$ with $\textbf{t}_{\alpha}^{\Answerred} = O(\polylog(n), \poly(\eps))$ such that for model $t \in \{*, co\}$
\begin{equation*}
	\omega^{t}(\cG^{\Answerred}_n) > 1 - \eps \Longrightarrow \omega^t(\cG) > 1 -\mathbf{t}^{\Answerred}_{\alpha}(\eps, n). 
\end{equation*}
Hence, assume that $\omega^{t}(\cG^{\Answerred}_n) > 1  - \eps$, and fix strategy in model $t$ such that $\strategy$ succeed at $\cG^{\Answerred}_n$ with probability at most $\eps$. We show the soundness clause by proving the following:
\begin{enumerate}
	\item By using~\Cref{thm:soundnessQLDT}, we first show that there exist a strategy $\strategy^{\text{poly}}$ for $\cG^{\Answerred}_n$ which consist of the provers first performing hidden measurements and sampled $6+m$ low-individual degree polynomials, and then using these polynomials to pass all the low-individual degree polynomial/simultaneous low-individual degree polynomial test for $\cG_n^{\Answerred}$.
	\item Then we use~\Cref{thm:classicalsoundnessPCPP} on the polynomial generated by $\strategy^{\text{poly}}$ to construct a strategy $\strategy^{\text{Ora}}$ for $\cG^{\text{Ora}}_n$ which succeed with probability at least $1- O(\polylog(n), \poly(\eps))$. 
	\item Finally, we conclude the proof by applying~\Cref{lem:oratransformation} to show that there exist a strategy for $\cG_n$  which succeed with probability at least $1- O(\polylog(n), \poly(\eps))$, thus showing the lemma. 
\end{enumerate}
For simplicity of notations, we work with synchronous strategies in order to show point 1 and 2. For a question pair $(x,y) = \left( (x^{\text{Ora}}, x^{\text{LDL}}, (x^{\text{game}},x^{\text{LDC}} )), (y^{\text{Ora}}, y^{\text{LDL}}, (y^{\text{game}},y^{\text{LDC}}))\right)$ for $\cG^{\Answerred}_n$, we observe that the synchronous question pair for $\cG^{\Answerred}_n$ corresponds to the case where $x^{\text{Ora}} = y^{\text{Ora}}$ and $x^{\text{LDL}} = y^{\text{LDL}}$ which occur with constant probability. This implies that the game $\cG^{\Answerred}_n$ is $\frac{1}{c_b}$-balanced for some constant $\frac{1}{c_b}$, and hence any strategy $\strategy$ for $\cG^{\Answerred}_n$ which succeed with probability $1- \eps$ must be $c_b \cdot  \eps$-synchronous. By~\Cref{thm:Rounding}, we have  
\begin{equation}
	\omega^{t}(\cG^{\Answerred}_n) > 1 - \eps  \Longrightarrow	\omega^{t}_s(\cG^{\Answerred}_n) > 1 - \eps - \textbf{s}^{\text{Rounding}}(c_b\eps) = 1 - \delta_1.
\end{equation}
Hence, fix a synchronous strategy $\strategy = (\cL^2(\alicealg, \tau), \ket{\tau}, \{A_a^x \})$ such that $\omega(\cG^{\Answerred}_n, \strategy) > 1 - c_1 \cdot \delta_1$, where $c_1$ is a sufficiently small constant choose later down the proof.

Define the function $\text{eval}_s^n: \Idpoly(p,m^{\text{ans}},p)^{\times n} \rightarrow \bF_{2^p}^{\times n}$ as $\text{eval}_s^n(\textbf{g}_0, \cdots, \textbf{g}_{n-1}) = (\textbf{g}_0(s),\cdots ,\textbf{g}_{n-1}(s))$. We wish to first show the following claim:
\newcommand{\game}{\text{game}}
\begin{claim} \label{claim:ARpolyPVM}
	For all $(x^{\game},y^{\game}) \in \cX^2$, there exist six sets of PVM in $\alicealg'$, 
	\begin{itemize}
		\item $\{G_{\overline{\textbf{g}}}^{(\text{Prover, A}), x^{\game}}\}_{\overline{\textbf{g}} \in \Idpoly(p,m^{\text{ans}},p)}$,
		\item $\{G_{\overline{\textbf{g}}}^{(\text{Prover, B}), y^{\game}}\}_{\overline{\textbf{g}} \in \Idpoly(p,m^{\text{ans}},p)}$,
		\item $\{G_{\overline{\textbf{g}}}^{(\text{Ora})_{o}, (x^{\game}, y^{\game})}\}_{\overline{\textbf{g}} \in \Idpoly(p,m^{\text{ans}},p)}$ for $o \in \{0,1,2\}$,
		\item $\{G_{\textbf{g}_{U_0}, \cdots, \textbf{g}_{U_4}, \textbf{g}_{\Gamma}, \textbf{g}_{B_0} \cdots,\textbf{g}_{B_{m-1}}}^{(\text{Ora}), (x^{\game}, y^{\game})}\}_{\textbf{g}_v \in \Idpoly(p,m,p)}$, 
	\end{itemize}
	such that the following hold: For $s \in \bF_{2^p}^{m^{\text{ans}}}$ and $n \in \bN$,
	\begin{align}
		\label{eq:ARsoundAlicepoly} A^{(\text{Prover, A}), (\text{Point}), x^{\game} ,s}_{u} &\simeq_{O(\poly(\eps, \log(n), \alpha))} G_{[\textbf{eval}_{s}^1 | u]}^{(\text{Prover, A}), x^{\game}}, \\
		\label{eq:ARsoundBobpoly} A^{(\text{Prover, B}), (\text{Point}), y^{\game} ,s}_{u} &\simeq_{O(\poly(\eps, \log(n), \alpha))} G_{[\textbf{eval}_{s}^1 | u]}^{(\text{Prover, B}), y^{\game}}, \\
		\label{eq:ARsoundvarpoly}  A^{(\text{Ora})_{o}, (\text{Point}), ((x^{\game}, y^{\game}) ,s)}_{u} &\simeq_{O(\poly(\eps, \log(n), \alpha))} G_{[\textbf{eval}_{s}^1 | u]}^{(\text{Ora})_{o}, (x^{\game}, y^{\game})}, \quad o \in \{0,1,2\}, \\
		\label{eq:ARsoundAllpoly} A^{(\text{Ora}), (\text{Point}), ((x^{\game}, y^{\game}) ,s)}_{u_0, \cdots, u_4, \gamma, \beta_0, \cdots,\beta_{m-1}} &\simeq_{O(\poly(\eps, \log(n), \alpha))} G_{[\textbf{eval}_s^{6+m}|( u_0, \cdots, u_4, \gamma, \beta_0, \cdots,\beta_{m-1})]}^{\text{Ora}, (x^{\game}, y^{\game})}, 
	\end{align}
	where $\simeq$ for the above four equations is defined over the distribution $(x^{\game}, y^{\game}) \sim \mu_n$ and $s \sim \bF_{2^p}^{m^{\text{ans}}}$ for the first 3 equation, and $s \sim \bF_{2^p}^{m}$ for the last equation and over the state $\ket{\tau}$.
\end{claim}
\begin{proof}
$\cG^{\Answerred}$, when the question set is restricted to the case where the oracularization question label is restricted to ``(Prover A)", and the question content for $\cG_n$ is fixed to some $x^{\game} \in \cX$, is precisely an instance of the $(p,m,p)$-low-individual degree test. Since $\strategy$ succeed with probability $1 - c_1 \cdot \delta_1$, and the probability that both oracularization question label in a question pair to be both ``(Prover A)" being $\frac{1}{9}$. This implies that, on average over $(\mu_n)_X$, the strategy $\strategy$, when restricted to the case where $(x^{\text{Ora}}, y^{\text{Ora}})$ being both ``(Prover A)", succeed with probability at least $1 - 9 \cdot \delta_1$. Hence, by~\Cref{thm:soundnessQLDT}
\begin{equation*}
 	A^{(\text{Prover, A}), (\text{Point}), x^{\game} ,s}_{u} \simeq_{\mathbf{\eta}_{\text{LD}}(p,m,p,9 \cdot \delta_1)} G_{[\textbf{eval}_{s}^1 | u]}^{(\text{Prover, A}), x^{\game}},
\end{equation*}
where $\simeq$ for the above equations is defined over the distribution $(x^{\game}) \sim \mu_n$ and $s \sim \bF_{2^p}^{m^{\text{ans}}}$
By the choice of the parameter $p, m^{\text{ans}} = O(\poly(\alpha, \log(n)))$, this immediately shows~\eqref{eq:ARsoundAlicepoly}. ~\eqref{eq:ARsoundBobpoly} and~\eqref{eq:ARsoundvarpoly} follows from a similar argument (except with the oracularization question label replaced with ``(Prover B)" and ``$(\text{Ora})_{o}$", $i = \{0,1,2\}$ respectively). ~\eqref{eq:ARsoundAllpoly} follows a similar argument with the label ``(Ora)" and using~\Cref{lem:simuQLID}. This completes the proof.
\end{proof}
We wish to modify the output for the PVM associated with the ``(Prover, A)", ``(Prover, B)" and ``$(\text{Ora})_{o}$", $i \in \{0,1,2\}$ as PVM which outputs $\textbf{g} \in \Idpoly(p,m,p)$. For $s \in \bF_{2^p}^m$, partition $s = (s_0, \cdots, s_4, w)$ where $s_i \in \bF_{2^p}^{m^{\text{ans}}}$, $i \in [5] $ and $w \in  \bF_{2^p}^{5+g}$, we make the following post measurement processing to PVMs given in the last lemma as follows
\begin{itemize}
	\item Treat the outputs $\textbf{g}$ from $\{G_{\overline{\textbf{g}}}^{(\text{Prover, A}), x^{\game}}\}$ as $\textbf{g} \in \Idpoly(p,m,p)$ as $\textbf{g}(s) =\overline{\textbf{g}}(s_0)$.
	\item Treat the outputs $\textbf{g}$ from $\{G_{\overline{\textbf{g}}}^{(\text{Prover, B}), y^{\game}}\}$ as $\textbf{g} \in \Idpoly(p,m,p)$ as $\textbf{g}(s) =\overline{\textbf{g}}(s_1)$.
	\item For $o \in \{0,1,2\}$, treat the outputs $\textbf{g}$ from $\{G_{\overline{\textbf{g}}}^{(\text{Ora})_{o}, (x^{\game}, y^{\game})}\}$ as $\textbf{g} \in \Idpoly(p,m,p)$ with $\textbf{g}(s) =\overline{\textbf{g}}(s_{o+2})$.
\end{itemize}
To distinguish the two measurement outputs, we write the output polynomial as $\overline{\textbf{g}}$ if the resulting polynomial output are from $\Idpoly(p,m^{\text{ans}},p)$ and $\textbf{g}$ if the output are from $\Idpoly(p,m,p)$. For $i \in [5]$, define the PVM measurement
\begin{align}
	\label{eq:ARproofOracFull-1} G_{\textbf{g}_{U_i}}^{(\text{Ora}), (x^{\game}, y^{\game}), U_i} &= \sum_{\substack{\textbf{g}_{U_0}, \cdots,\textbf{g}_{U_{i-1}},\textbf{g}_{U_{i+1}}, \cdots,\textbf{g}_{U_4} \\ \textbf{g}_{\Gamma}, \textbf{g}_{B_0},\cdots,  \textbf{g}_{B_{m-1}} \in \Idpoly(p, m, p) } }  G_{\textbf{g}_{U_0}, \cdots,\textbf{g}_{U_4},  \textbf{g}_{\Gamma}, \textbf{g}_{B_0},\cdots,  \textbf{g}_{B_{m-1}}  }^{(\text{Ora}), (x^{\game}, y^{\game})} \\
	\label{eq:ARproofOracFull-2} G_{\textbf{g}_{\Gamma}, \textbf{g}_{B_0} \cdots,\textbf{g}_{B_{m-1}}}^{(\text{Ora}), (x^{\game}, y^{\game}), \text{Full}} &= \sum_{\textbf{g}_{U_0}, \cdots, \textbf{g}_{U_4}  \in \Idpoly(p, m, p) }  G_{\textbf{g}_{U_0}, \cdots,\textbf{g}_{U_4},  \textbf{g}_{\Gamma}, \textbf{g}_{B_0},\cdots,  \textbf{g}_{B_{m-1}}  }^{(\text{Ora}), (x^{\game}, y^{\game})}.
\end{align}
Since $ G^{(\text{Ora}), (x^{\game}, y^{\game})} $ is projective, for all $(x,y) \in \cX_n^2$ and outputs $\textbf{g}_v$, $v \in \{U_0, \cdots, U_4, \Gamma, B_0, \cdots, B_{m-1}\}$
\begin{align} \label{eq:ARproofGOrafull}
	 G^{(\text{Ora}), (x, y} = &G^{(\text{Ora}), (x, y),U_0 } \cdots  G^{(\text{Ora}), (x, y), U_4} G^{(\text{Ora}), (x, y), \text{Full}} G^{(\text{Ora}), (x, y),U_4 } \cdots G^{(\text{Ora}), (x, y), U_0}.
\end{align}
Base on the decision procedure for $\cG^{\Answerred}$, we show the following claim 
\begin{claim} \label{claim:ARconsistpoly}
	On average over $(x,y) \sim \mu$, $s \in \bF_{2^p}^m$ and over the state $\ket{\tau}$
	\begin{align}
		\label{eq:ARconsistAlicepoly} G_{[\textbf{eval}_{s}^1 | u_0]}^{(\text{Prover, A}), x^{\game}}	&\simeq_{O(\poly(\eps, \log(n), \alpha))} 	(G_{[\textbf{eval}_{s}^1 | u_0]}^{(\text{Ora}), (x^{\game}, y^{\game}), U_0})^{op}  \\
		\label{eq:ARconsistBobpoly} G_{[\textbf{eval}_{s}^1 | u_1]}^{(\text{Prover, B}), y^{\game}}	&\simeq_{O(\poly(\eps, \log(n), \alpha))} 	(G_{[\textbf{eval}_{s}^1 | u_1]}^{(\text{Ora}), (x^{\game}, y^{\game}), U_1})^{op}  \\
		\label{eq:ARconsistOracpoly} G_{[\textbf{eval}_{s}^{o+2} | u_{o+2}]}^{(\text{Ora})_{o}, (x^{\game}, y^{\game})}	&\simeq_{O(\poly(\eps, \log(n), \alpha))} 	(	G_{[\textbf{eval}_{s}^{o+2} | u_{o+2}]}^{(\text{Ora}), (x^{\game}, y^{\game}), U_{o+2}})^{op} , \, o \in \{0,1,2\}.
	\end{align} 
\end{claim}

\begin{proof}
	We show the proof for~\eqref{eq:ARconsistAlicepoly} below, the proof for~\eqref{eq:ARconsistBobpoly} and~\eqref{eq:ARconsistOracpoly} follows a similar proof. Consider the question pair $(x,y)$ for $\cG^{\Answerred}$ where $(x^{\text{Ora}}, y^{\text{Ora}}) =$ ((Prover A), (Ora)) and $x^{\text{LDL}} = y^{\text{LDL}} = $ (Point), this occur with constant probability. Since $\strategy$ succeed with probability at least $1-  c_1 \cdot\delta_1$, by point 2 of the ``Prover consistency check" from~\Cref{fig:Answerredver},
	\begin{equation} \label{eq:ARproofB1}
		A^{((\text{Prover, A}), (\text{Point}), (x^{\game} ,s))}_{u_0} \simeq_{O(\poly(\eps, \log(n), \alpha))} A^{(\text{Ora}), (\text{Point}), ((x^{\game}, y^{\game}) ,s)}_{u_0, \cdots, u_4, \gamma, \beta_0, \cdots,\beta_{m-1}} 
	\end{equation}
	over $(x,y) \sim \mu$ and $s \sim \bF_{2^p}^m$ and over the state $\ket{\tau}$. ~\eqref{eq:ARconsistAlicepoly} then follows from~\Cref{lem:closetodistance} point 1 and 2 which translates between $\simeq$ distance to $\approx$ distance, and the triangle inequality for $\approx$ distance applied to~\eqref{eq:ARsoundAlicepoly},~\eqref{eq:ARproofB1}, and~\eqref{eq:ARsoundAllpoly}. 
\end{proof}

For $(x^{\game}, y^{\game}) \in \cX_n$, define the POVM measurement $M^{(x^{\game}, y^{\game})}_{\textbf{g}_{U_0}, \cdots, \textbf{g}_{U_4}, \textbf{g}_{\Gamma}, \textbf{g}_{B_0} \cdots,\textbf{g}_{B_{m-1}}}$ with outcomes $\textbf{g}_v \in \Idpoly(p,m,p)$ as
\begin{align*}
	M^{(x, y)}_{\textbf{g}_{U_0}, \cdots, \textbf{g}_{U_4}, \textbf{g}_{\Gamma}, \textbf{g}_{B_0} \cdots,\textbf{g}_{B_{m-1}}} = &G_{\textbf{g}_{U_0}}^{(\text{A}), x} G_{\textbf{g}_{U_1}}^{(\text{B}), y}  G_{\textbf{g}_{U_2}}^{(\text{Orc})_0, (x,y)}  G_{\textbf{g}_{U_3}}^{(\text{Orc})_1, (x,y)}  G_{\textbf{g}_{U_4}}^{(\text{Orc})_2, (x,y)} 	G_{\textbf{g}_{\Gamma}, \textbf{g}_{B_0} \cdots,\textbf{g}_{B_{m-1}}}^{(\text{Ora}), (x, y), \text{Full}} \cdot \\
	&\quad G_{\textbf{g}_{U_4}}^{(\text{Orc})_2, (x,y)}  G_{\textbf{g}_{U_3}}^{(\text{Orc})_1, (x,y)}  G_{\textbf{g}_{U_2}}^{(\text{Orc})_0, (x,y)} G_{\textbf{g}_{U_1}}^{(\text{B}), y} G_{\textbf{g}_{U_0}}^{(\text{A}), x}, 
\end{align*}
where for $P \in \{A,B\}$, we shorten the label (Prover, P) to (P), and remove the superscript ``\game" in the above equation. We remark that in contrast to $G^{(\text{Ora}), (x^{\game}, y^{\game})}$, the output $\textbf{g}_{U_i}$, $i \in [5]$ from $M^{(x, y)}$ are secretly polynomials in $\Idpoly(p, m^{\text{ans}},m)$ which is consistent with the definition for the 5 polynomials given in~\Cref{thm:classicalPCPP}. 

Furthermore, we define $M^{(x^{\text{game}}, y^{\text{game}}), U_i}$ for $i \in [5]$, and $M^{(x^{\text{game}}, y^{\text{game}}), \text{Full}}$ in a similar manner as~\eqref{eq:ARproofOracFull-1} and~\eqref{eq:ARproofOracFull-2}. By definition
\begin{align*}
	M^{(x^{\text{game}}, y^{\text{game}}), U_0} &=  G^{(\text{Prover, A}), x^{\game}} \\
	M^{(x^{\text{game}}, y^{\text{game}}), U_1} &=  G^{(\text{Prover, B}), y^{\game}} \\
	M^{(x^{\text{game}}, y^{\text{game}}), U_{o+2}} &=  G_{\overline{\textbf{g}}}^{(\text{Ora})_{o}, (x^{\game}, y^{\game})}, o \in \{0,1,2\}.
\end{align*}
For $(x^{\text{game}}, y^{\text{game}}) \in \cX^2$ and output tuple $(\textbf{g}_{U_0}, \cdots, \textbf{g}_{U_4}, \textbf{g}_{\Gamma}, \textbf{g}_{B_0} \cdots,\textbf{g}_{B_{m-1}})$ from $M^{(x^{\text{game}}, y^{\text{game}})}$. We refer to the output as ``good" if for a uniformly random $s =(s_0, \cdots, s_4, b_0, \cdots, b_4, z) \sim \bF_{2^p}^m$, the following occurs with probability over $\frac{1}{2}$:
\begin{itemize}
	\item $\textbf{g}_{\Gamma}(s) = \textbf{g}_{\deciderTM}(s) (\textbf{g}_{U_0}(s) - b_0)  (\textbf{g}_{U_1}(s) - b_1)  (\textbf{g}_{U_2}(s) - b_2)  (\textbf{g}_{U_3}(s) - b_3)  (\textbf{g}_{U_4}(s) - b_4)$
	\item $\textbf{g}_{\Gamma}(s) = \sum_{i \in [m]} \textbf{g}_{B_i}(s) \textbf{zero}(s) $,
\end{itemize}
where $\textbf{g}_{\deciderTM} = \mathtt{ComputePCP}_{\alpha}(\langle \deciderTM \rangle ,n, x^{\text{game}}, y^{\text{game}})$. By~\Cref{thm:classicalsoundnessPCPP}, if the output for $M^{(x^{\text{game}}, y^{\text{game}})}$ is a ``good" output, then there exist $a,b \in \{0,1\}^*$ with $|a|, |b|\leq \log^{\alpha}(n)$ such that $\textbf{g}_{U_0} =  \text{enc}_{\Gamma}(a)$ and $\textbf{g}_{U_1} =  \text{enc}_{\Gamma}(b)$ and $D(x^{\text{game}}, y^{\text{game}},a,b) =1$

We now proof the following claim regarding the measurement $M^{(x^{\text{game}}, y^{\text{game}})}$
\begin{claim} \label{claim:ARgoodoutput}
	On average over $(x^{\text{game}}, y^{\text{game}}) \sim \mu_n$ and the state $\ket{\tau}$
	\begin{equation} \label{eq:ARproofSynccondpoly}
		M^{(x^{\text{game}}, y^{\text{game}})} \simeq_{O(\poly(\eps, \log(n), \alpha))}  (M^{(x^{\text{game}}, y^{\text{game}})})^{op}.
	\end{equation}
	Furthermore, on average over $(x^{\text{game}}, y^{\text{game}})$, the measurement output for $\braket{\tau|M^{(x^{\text{game}}, y^{\text{game}})}|\tau}$ is ``good" with probability at least $1 - O(\poly(\eps, \log(n), \alpha))$
\end{claim}
\begin{proof}
	 We first show that, on average over $(x^{\text{game}}, y^{\text{game}}) \sim \mu_n$ and the state $\ket{\tau}$
	\begin{equation} \label{eq:ARsoundnessgoodoutcome-1}
		M^{(x^{\text{game}}, y^{\text{game}})} \simeq_{O(\poly(\eps, \log(n), \alpha))}  G^{(\text{Ora}), (x^{\game}, y^{\game})}
	\end{equation}
	Since $G^{(\text{Ora}), (x^{\game}, y^{\game})}$ is projective, by the definition of $M^{(x^{\text{game}}, y^{\text{game}})}$, the last $1+m$ measurement outcome for $G^{(\text{Ora}), (x^{\game}, y^{\game})}$ ($ \textbf{g}_{\Gamma}, \textbf{g}_{B_0},\cdots,  \textbf{g}_{B_{m-1}}$) will always be the same as the measurement $M^{(x^{\text{game}}, y^{\text{game}})}$ when the two measurements are made simultaneously (on any state). 
	For $i \in [5]$, by~\Cref{claim:ARconsistpoly} and the Schwartz-Zippel lemma (\Cref{lem:Schwartz_Zipple})
	\begin{align*}
		M^{(x^{\text{game}}, y^{\text{game}}), U_i}_{\textbf{g}_{U_i}}&\simeq_{\delta_2} (G_{\textbf{g}_{U_i}}^{(\text{Ora}), (x^{\game}, y^{\game}), U_i})^{op}
	\end{align*}
	for $\delta_2 = O(\poly(\eps, \log(n), \alpha)) + \frac{m \cdot d}{2^p}$. Hence, by repeatedly applying~\Cref{lem:combmeasure}, the underlying vector state is a tracial state, and using~\Cref{lem:closetodistance} to convert between $\simeq$ distance to $\approx$ distance,
	\allowdisplaybreaks{
		\begin{align*}
			 &G^{(\text{Ora}), (x, y)} = G^{(x, y),U_0 } \cdots  G^{(x, y), U_4} G^{(x, y), \text{Full}} G^{(x, y),U_4 } \cdots G^{ (x, y), U_0} \\
			 &\approx_{\delta_2} (M^{(x, y), U_0})^{op} G^{(x, y),U_0 } \cdots  G^{(x, y), U_4} G^{(x, y), \text{Full}} G^{(x, y),U_4 } \cdots G^{(x, y),U_1 } \\
			 &\cdots \\
			 &\approx_{\delta_2} (M^{(x, y), U_4} \cdots M^{(x, y), U_0})^{op}  G^{ (x, y),U_0 } \cdots  G^{(x, y),U_4 }   G^{(x, y), U_4} G^{(x, y), \text{Full}}  \\
			 &=  (M^{(x, y), \text{Full}} M^{(x, y), U_4} \cdots M^{(x, y), U_0})^{op}  G^{ (x, y),U_0 } \cdots  G^{(x, y),U_4 } \\
			 &\approx_{\delta_2}  (M^{(x, y), U_4} M^{(x, y), \text{Full}} M^{(x, y), U_4} \cdots M^{(x, y), U_0})^{op}  G^{ (x, y),U_0 } \cdots  G^{(x, y),U_3 }  \\
			 &\cdots \\
			 &\approx_{\delta_2}  ( M^{(x, y), U_0} \cdots M^{(x, y), U_4} M^{(x, y), \text{Full}} M^{(x, y), U_4} \cdots M^{(x, y), U_0})^{op}  =  (M^{(x, y)})^{op}
		\end{align*}
	}
	where we remove the superscript ``game" and (\text{Ora}) in the above derivation for clarity. Hence, by using the triangle inequality for $\approx$ distance and~\Cref{lem:closetodistance}, this implies that 
	\begin{equation*}
		G^{(\text{Ora}), (x^{\text{game}}, y^{\text{game}})}  \simeq_{10 \delta_2}(M^{(x^{\text{game}}, y^{\text{game}})})^{op}
	\end{equation*}
	and since the underlying state is a tracial state and $\delta_2 = O(\poly(\eps, \log(n), \alpha))$, this shows~\eqref{eq:ARsoundnessgoodoutcome-1}. Since the underlying state for~\eqref{eq:ARsoundnessgoodoutcome-1} is the tracial state $\ket{\tau}$, we also have
	\begin{equation} \label{eq:ARsoundnessgoodoutcome-2}
		(M^{(x^{\text{game}}, y^{\text{game}})})^{op} \simeq_{O(\poly(\eps, \log(n), \alpha))}  (G^{(\text{Ora}), (x^{\game}, y^{\game})})^{op}
	\end{equation}
	For~\eqref{eq:ARproofSynccondpoly}, since $G^{(\text{Ora}), (x^{\game}, y^{\game})} $ are all projective measurements, 
	\begin{equation} \label{eq:ARsoundnessgoodoutcome-3}
		G^{(\text{Ora}), (x^{\game}, y^{\game})} \simeq_0 (	G^{(\text{Ora}), (x^{\game}, y^{\game})})^{op}
	\end{equation}
	over $(x^{\game}, y^{\game}) \sim \mu_n$ and the tracial state $\ket{\tau}$.~\Cref{eq:ARproofSynccondpoly} then follows from~\Cref{lem:closetodistance} and the triangle inequality of $\approx$ distance being applied to \eqref{eq:ARsoundnessgoodoutcome-1}, \eqref{eq:ARsoundnessgoodoutcome-2} and \eqref{eq:ARsoundnessgoodoutcome-3}. 
	
	For the second part of~\Cref{claim:ARgoodoutput}, by applying the data processing inequality to~\eqref{eq:ARsoundnessgoodoutcome-1},
	\begin{equation} \label{eq:ARsoundnessgoodoutcome-4}
		M^{(x^{\text{game}}, y^{\text{game}})}_{[\textbf{eval}_s^{6+m}|( u_0, \cdots, u_4, \gamma, \beta_0, \cdots,\beta_{m-1})]} \simeq_{O(\poly(\eps, \log(n), \alpha))} A^{(\text{Ora}), (\text{Point}), ((x^{\game}, y^{\game}) ,s)}_{[\textbf{eval}_s^{6+m}|( u_0, \cdots, u_4, \gamma, \beta_0, \cdots,\beta_{m-1})]}.
	\end{equation}
	Hence by applying the triangle inequality for $\approx$ distance and~\Cref{lem:closetodistance} to~\eqref{eq:ARsoundnessgoodoutcome-4} and~\eqref{eq:ARsoundAllpoly}, we obtain
	\begin{equation}~\label{eq:ARsoundnessgoodoutcome-5}
		M^{(x^{\text{game}}, y^{\text{game}})}_{[\textbf{eval}_s^{6+m}|( u_0, \cdots, u_4, \gamma, \beta_0, \cdots,\beta_{m-1})]} \simeq_{O(\poly(\eps, \log(n), \alpha))}  A^{(\text{Ora}), (\text{Point}), ((x^{\game}, y^{\game}) ,s)}_{u_0, \cdots, u_4, \gamma, \beta_0, \cdots,\beta_{m-1}}. 
	\end{equation}
	Recall that $\strategy$ is a synchronous strategy which succeed at $\cG^{\text{AR}}$ with probability at least $1 -  c_1 \cdot \delta_1$. Since the oracularization question label is pick with constant probability, by the ``PCPP proof check" clause of~\Cref{fig:Answerredver}, on expectation over $(x^{\game},y^{\game}) \sim \mu$ and $s = (s_0, \cdots, s_4, b_0, \cdots, b_4, z) \in \bF_{2^p}^m$, the measurement 
	\begin{equation*}
		\braket{\tau|A_{(u_0, \cdots, u_4, \gamma, \beta_0, \cdots, \beta_{m-1} )}^{(\text{Orc}), (\text{Point}), ((x^{\game}, y^{\game}),  s)}|\tau} 
	\end{equation*}	
	outputs the answer which satisfies the properties below with probability $1 - c_2 c_1 \cdot\delta_1$
	\begin{enumerate}
		\item $\gamma = \textbf{g}_{\deciderTM}(s) \cdot (u_1 - b_0) \cdots (u_4 - b_4)$,
		\item $\gamma = \sum_{i \in[m]} \beta_i \cdot \textbf{zero}(s_i)$,
	\end{enumerate}
	where $\textbf{g}_{\deciderTM} = \mathtt{ComputePCP}_{\alpha}-(\langle \deciderTM \rangle ,n, x^{\text{game}}, y^{\text{game}})$. Pick $c_1 \in (0,1)$ used to define $\strategy$ such that $1 - c_2 c_1 \cdot\delta_1 \geq \frac{1}{2}$. Combine this with~\Cref{eq:ARsoundnessgoodoutcome-5}, this shows that on average over $(x^{\game}, y^{\game}) \sim \mu_n$, the probability that $M^{(x^{\text{game}}, y^{\text{game}})}$ gives an output which is ``good" with probability at least $1-= O(\poly(\eps, \log(n), \alpha))$, thus completing the claim for the lemma. 
\end{proof}
Base on the POVM $M^{(x^{\text{game}}, y^{\text{game}})}$, we define a symmetric strategy $\strategy^{\text{Ora}} =(\cL^2(\alicealg, \tau), \ket{\tau}, \{B_a^x \})$ for $\cG^{\text{Ora}}$ as follows: Fixed $(x^{\text{game}}, y^{\text{game}}) \in \cX_n$, the measurement operator $\{B^{\text{(Prover, A)}, x^{\text{game}}}_a\}$ as a data processing measurement as follows:
\begin{itemize}
	\item Perform the measurement $M^{x^{\text{game}}, U_0}_{\textbf{g}_{U_0}}$ and obtain a polynomial $\textbf{g}_{U_0}$.
	\item If there exist an $a \in \cA_n$ such that $|a| \leq \log^{\alpha}(n)$ and $\textbf{g}_{U_0} = \text{enc}$, output $a$. Otherwise, output $0$. 
\end{itemize}
The measurement operator $\{B^{\text{(Prover, B)}, y^{\text{game}}}\}$ is define in a similar manner as $\{B^{\text{(Prover, A)}, x^{\text{game}}}\}$ with the measurement $M^{x^{\text{game}}, U_1}_{\textbf{g}_{U_1}}$. The measurement operator $\{B^{\text{(Orac)}, (x^{\text{game}},y^{\text{game}})}\}$ similarly define as a data processing measurement as follows:
\begin{itemize}
	\item Perform the measurement $M^{(x^{\text{game}}, y^{\text{game}})}$ and obtain the tuple of polynomials $(\textbf{g}_{v})$.
	\item If the given measurement outcome is a ``good" output, by~\Cref{thm:classicalsoundnessPCPP}, there exist $(a,b) \in \cA_n^2$ such that $\textbf{g}_{U_0} = \text{enc}_{\Gamma}(a)$, $\textbf{g}_{U_1} = \text{enc}_{\Gamma}(b)$, output the corresponding $(a,b) \in \cA_n^2$. Otherwise, output $(0,0)$. 
\end{itemize}
We remark that the above strategy is not necessarily synchronous strategy, since $\{M^{(x^{\text{game}}, y^{\text{game}})}\}$ does not necessarily have to be a PVM. We make the following claim about $\strategy^{\text{Ora}}$. 
\begin{claim}
	$\omega(\cG^{\textbf{Orac}}, \strategy^{\text{Ora}}) \geq 1 - O(\poly(\eps, \log(n), \alpha))$. 
\end{claim}
\begin{proof}
	We consider the performance of $\strategy^{\text{Ora}}$ for different question pairs given in~\Cref{fig:Oracularization} below:
	\begin{itemize}
		\item (Prover, P) \textbf{-}  (Prover, P), $P \in \{A,B\}$: Since both $M^{x^{\text{game}}, U_A} = G^{(\text{Prover, A}), x^{\game}} $ and $M^{y^{\text{game}}, U_B} = G^{(\text{Prover, B}), y^{\game}} $ are both projective and $\strategy^{\text{Ora}}$ uses the tracial state as the underlying state. This implies that $\strategy^{\text{Ora}}$ always succeed on this question pair.
		\item 	(Oracularization) \textbf{--} (Oracularization): For the ``consistency" part of this question pair,  $B^{\text{(Orac)}, (x^{\text{game}},y^{\text{game}})}$ is a data processed measurement of $M^{(x^{\text{game}}, y^{\text{game}})}$, which by~\eqref{eq:ARproofSynccondpoly}, are consistent with probability at least $1 - O(\poly(\eps, \log(n), \alpha))$. For the ``proof checking" part of this question pair,, whenever $M^{(x^{\text{game}}, y^{\text{game}})}$ returns a ``good" output when performing the measurement $B^{\text{(Orac)}, (x^{\text{game}},y^{\text{game}})}$, the corresponding output $(a,b) \in \cA_n^2$ always satisfies $D_n(x,y,a,b) = 1$ by~\Cref{thm:classicalsoundnessPCPP}. By~\Cref{claim:ARgoodoutput}, this occurs with $1 - O(\poly(\eps, \log(n), \alpha))$. Combining these two facts, this implies that $\strategy^{\text{Ora}}$ succeed on this question pair with probability at least $1 - O(\poly(\eps, \log(n), \alpha))$.  
		\item (Oracularization) \textbf{--} (Prover, P), $P \in \{A,B\}$: Restricted to the case when the provers receiving the question label `` (Oracularization)" and obtain a ``good" outcome from the measurement of $M^{(x^{\text{game}}, y^{\text{game}})}$, by the definition of $M$, the `` (Prover, P)" prover would receive the same polynomial from his/her measurement output, and hence output a consistent answer label as the `` (Oracularization)" prover. Since a ``good" outcome occurs with probability  $1 - O(\poly(\eps, \log(n), \alpha))$, this implies that $\strategy^{\text{Ora}}$ succeed with probability on this question pair with probability at least $1 - O(\poly(\eps, \log(n), \alpha))$.  
	\end{itemize}
	By averaging out the probability given above, we see that $\omega(\cG^{\textbf{Orac}}, \strategy^{\text{Ora}}) > 1 - O(\poly(\eps, \log(n), \alpha))$, completing the proof of the claim. 
\end{proof}
This shows that for model $t \in \{*,co\}$, $\omega^t(\cG) > 1- \eps$ implies that $\omega^t(\cG^{\text{Ora}}) > 1- O(\poly(\eps, \log(n), \alpha))$. The proof of~\Cref{prop:AnswerReduction} then follows from the ``soundness" clause of~\Cref{lem:oratransformation}.

\newpage
\begin{CJK*}{UTF8}{gbsn}
\printglossary[style={indexgroup},title={Nomenclature}]
\end{CJK*}
\end{document}